\documentclass[11pt,english]{article}
\usepackage{color}
\usepackage[T1]{fontenc}
\usepackage[latin9]{inputenc}
\usepackage[a4paper]{geometry}
\usepackage{textcomp}
\usepackage{amsthm}
\usepackage{amsmath}
\usepackage{amssymb}
\usepackage{esint}
\usepackage{bbm}
\usepackage{hyperref}
\usepackage{babel}
\usepackage{amscd}
\usepackage{amsfonts}
\usepackage{amsthm}
\usepackage{mathrsfs}
\usepackage{dsfont}
\usepackage{mathtools}
\usepackage{enumerate}
\usepackage{bm}
\usepackage{bbm}
\usepackage{authblk}
\usepackage{cleveref}

\usepackage[english=american]{csquotes}

\usepackage{cite}

\numberwithin{equation}{section}

\usepackage{etoolbox}
\DeclarePairedDelimiter\abs\lvert\rvert
\DeclarePairedDelimiterX\norm[1]\lVert\rVert{
  \ifblank{#1}{\,\cdot\,}{#1}
}

\DeclarePairedDelimiterX\BS[1]{[}{]}{
  
  #1
}
\DeclarePairedDelimiterX\PS[1]{(}{)}{
  
  #1
}
\newcommand\Aver{\mathbb{E}\BS}

\DeclarePairedDelimiterX\Set[1]\{\}{%
  #1}

\newtheorem{thm}{Theorem}[section]
\newtheorem{lemma}[thm]{Lemma}

\theoremstyle{definition}
\newtheorem{definition}[thm]{Definition}

\newtheorem{setting}[thm]{Setting}
\newtheorem{notation}[thm]{Notation}

\theoremstyle{remark}
\newtheorem{remark}[thm]{Remark}
\newtheorem{observation}[thm]{Observation}

\newcommand{\R}{\mathbb{R}}
\newcommand{\C}{\mathbb{C}}
\newcommand{\N}{\mathbb{N}}

\newcommand{\Pro}{\mathbb{P}}

\DeclareMathOperator{\OpW}{Op_\hbar^\mathrm{w}}

\DeclareMathOperator{\im}{Im}

\DeclareMathOperator{\Vol}{Vol}

\DeclareMathOperator{\sgn}{sgn}

\begin{document}

\title{Schr\"odinger evolution in a low-density random potential: annealed convergence to the linear Boltzmann equation for general semiclassical Wigner measures\footnote{Research supported by EPSRC grant EP/S024948/1}}

\author{S{\o}ren Mikkelsen}

   \affil{\small{University of Bath, Department of Mathematical Sciences\\ Bath
BA2 7AY, United Kingdom} }

\date{\today}

\maketitle

\begin{abstract}
We consider solutions of the time-dependent Schr\"odinger equation for a potential localised at the points of a Poisson process. We prove convergence of the phase-space distribution in the annealed Boltzmann-Grad limit to a semiclassical Wigner (or defect) measure and show that it is a solution of the linear Boltzmann equation. Our results hold for a large class of square-integrable initial data associated to Wigner measures, including Langragian states, WKB states and coherent states. This extends important previous work by Eng and Erd{\H o}s.
\end{abstract}

\tableofcontents

\section{Introduction}

For over a century, the {\em linear Boltzmann equation} has served as a central model for macroscopic particle transport in low-density matter \cite{Lorentz}. There are now several rigorous derivations of this integro-differential equation, starting from a single-particle Schr\"odinger equation with a suitably scaled random potential. These include the ground-breaking studies by Spohn \cite{MR471824} and by Erd{\H o}s and Yau \cite{MR1744001}, which concern the weak-coupling limit with potential rescaled by a multiplicative constant but otherwise fixed. We will here investigate the  scaling considered in Eng and Erd{\H o}s' important paper \cite{MR2156632}, which corresponds to the {\em low-density} or {\em Boltzmann-Grad limit} as in Lorentz' original work \cite{Lorentz} in the context of classical dynamics, cf. also \cite{MR496299,PhysRev.185.308,MR725107}.

Eng and Erd{\H o}s investigate the Schr\"odinger equation in dimension $d\geq 3$ restricted to a box with sides $L$, for a potential that is the superposition of $N$ identical potentials localised at uniformly distributed points in the box. The relevant limits are $N,L\to\infty$ so that the particle density $\rho=N/L^d$ is fixed, followed by the small wave length limit $\hbar\to 0$ where the localisation length of each potential is taken to be of the same order. The latter assumption leads to a Boltzmann equation whose collision kernel is given by the full quantum mechanical $T$-matrix, rather than the first-order Born approximation as in the case of the weak-coupling limit. The convergence to the linear Boltzmann equation is established for Husimi functions of the solution of the Schr\"odinger equation, whose initial condition is a specific Lagrangian state with sufficiently  smooth weight and linear phase.

The main novelty of the results in the present work is two fold. Firstly we will allow a significantly more general class of $L^2$ initial states than \cite{MR2156632}, and formulate the convergence in terms of semiclassical Wigner measures instead of Husimi functions. This larger class will even include singular initial states. In the case of the weak-coupling limit, where the potential is given by a random field, the extension to such general initial conditions was given by Breteaux in \cite{MR3330163}.  

Secondly we remove the box and work with the Schr\"odinger equation in $\R^d$, assuming that the potential is localised at the points of a Poisson process of intensity $\rho$. The advantage of this setup over the box is that it allows not only a discussion of the annealed limit (convergence on average over the scatterer configuration) but also the quenched limit (almost sure convergence) which we will discuss in a forthcoming paper. 
In the classical dynamics setting, this is precisely the distinction between the works of Gallavotti  \cite{PhysRev.185.308}, Spohn \cite{MR496299} (annealed) on one hand and Boldrighini, Bunimovich and Sinai  \cite{MR725107}  (quenched) on the other. As in the work of Eng and Erd{\H o}s we will restrict our attention to the case $d=3$. However most of our estimates hold for all $d\geq3$.
In Section~\ref{sec_comparision} we will give an outline of the paper and a comparison to previous works on this type of problem.

It should be noted that, although the linear Boltzmann equation breaks down in the case of periodic potentials, the same low-density scaling limit as in the random setting leads to new limit processes in both the classical and quantum setting. For details see \cite{MR2726643, MR2811599,marklof2021kinetic} and \cite{MR4026550,MR4292484}, respectively. However if damping is introduced into the microscopic equations convergence to the linear Boltzmann equation holds even for the periodic case. For details see \cite{Griffin_damped}.

\subsection{Schr\"odinger equation in a low density potential}
Our setup is as follows. We consider the Hamilton operator 
\begin{equation}\label{Hami:def}
	H_{\hbar,\lambda} = -\frac{\hbar^2}{2}\Delta + \lambda W_r(x),
\end{equation}
where $\hbar>0$ is a scaling parameter (Planck's  constant measured in suitable units), $\Delta $ is the standard Laplacian, $\lambda\geq 0$ is the coupling constant, and the potential $W_r(x)$ is the superposition of single site potentials $V_{r,x_j}$ located at the points $x_j$  of a locally finite set $\mathcal{X}$,
\begin{equation*}
	W_r(x)=  \sum_{x_j\in\mathcal{X}} V_{r,x_j}(x) .
\end{equation*}
More precisely,
\begin{equation}\label{def_potential_int}
V_{r,z}(x)=  V\bigg(\frac{x-r^{1-1/d} z}{r}\bigg)
\end{equation}
is the $r$-scaled single-site potential $V$ localised at the point $r^{1-1/d} z$, with $V$ a non-negative function in the Schwartz class $\mathcal{S}(\R^d)$. 

Throughout this paper the point set $\mathcal{X}$ is assumed to be a realisation of a Poisson point process with constant intensity $\rho$. We will implicitly assume we have a sufficiently rich probability space $(\Omega,\Sigma,\Pro)$ on which the point process is defined. The scaling of the scatterer locations $x_j$ by $r^{1-1/d} $ relative to the potential scaling by $r$ ensures that the mean free path length remains constant as $r\to 0$. This scaling corresponds precisely to the Boltzmann-Grad limit. 
In this setting, for every given $\hbar\in(0,\hbar_0]$, the Hamiltonian $H_{\hbar,\lambda}$ is defined as the unique self-adjoint Friedrichs extension of the associated form $\Pro$-almost surely (see Section \ref{sec:prelim} for details). 

The $T$-operator $T(E)$ for the quantum mechanical scattering in the single-site potential  $V$ is defined as the limit
\begin{equation}\label{def_t_operator}
	 T(E)=\lim_{\gamma\to 0_+} T^\gamma(E), \qquad    T^\gamma(E)=\lambda V + \lambda^2 V \frac{1}{E-(-\frac{1}{2}\Delta+\lambda V)+i\gamma} V ,
\end{equation}
with $V$ denoting the multiplication operator by the function $V$.
An important quantity is the kernel of $T(E)$ in momentum representation with $E=\frac{1}{2}p^2$, which we denote by $\hat{T}(p,q)$. This ``$T$-matrix'' is related to the scattering matrix $S(p,q)$ by the relation
\begin{equation*}
	S(p,q) = \delta(p-q) - 2\pi i \,\delta(\tfrac{1}{2} p^2- \tfrac{1}{2} q^2)\, \hat{T}(p,q) .
\end{equation*} 
The unitarity of $S$ implies the ``optical theorem''
 \begin{equation}\label{EQ:Optical_theorem}
	\im \hat{T}(q,q) = \pi \int_{\R^d} \delta( \tfrac{1}{2} p^2-\tfrac{1}{2}q^2)| \hat{T}(p,q)|^2 \,dp.
\end{equation} 

As in the work of Eng and Erd{\H o}s \cite{MR2156632}, we will assume that $r$ is of the same order as $\hbar$ as this will allow us to observe the full quantum mechanical scattering process in the limit. To simplify notation we set in the following (without loss of generality) $\hbar=r$.

We consider solutions $\varphi(t,x)$ of the Schr\"odinger equation
 \begin{equation}\label{intial_val_pro}
\begin{cases}
	i\hbar \partial_t \varphi(t,x) = H_{\hbar,\lambda} \varphi(t,x)
	\\
	\varphi(0,x) = \varphi_0(x)
\end{cases}
\end{equation}
with initial data $\varphi_0\in L^2(\R^d)$.
The solution of this initial value problem can be written as 
\begin{equation}\label{Sch:sol}
	\varphi(t,x) =  U_{\hbar,\lambda}(t) \varphi_0(x), \qquad U_{\hbar,\lambda}(t) = e^{-it \hbar^{-1}H_{\hbar,\lambda}} .
\end{equation}
A key point of our work is to choose $\hbar$-dependent initial states $\varphi_0$, whose phase space distribution converges to a limit measure as $\hbar\to 0$. We refer to the measures that can arise in this way as {\em semiclassical Wigner measures}.

\subsection{Semiclassical Wigner measures}

Denote by $\OpW(a)$ the Weyl-quantised $\hbar$-pseudo-differential operator corresponding to the classical observable, or symbol, $a\in\mathcal{S}(\R^{2d})$. That is,
\begin{equation}\label{WeylPDO}
	\OpW(a)\psi(x) = \frac{1}{(2\pi\hbar)^d} \int_{\R^{2d}} e^{i\hbar^{-1}\langle x-y,p \rangle} a(\tfrac{x+y}{2},p) \psi(y) \, dydp,
\end{equation}
for $\psi\in\mathcal{S}(\R^d)$; cf.~\cite{MR897108,MR2952218}. When the symbol $a$ is assumed to be in $\mathcal{S}(\R^{2d})$ we have from \cite[Theorem~4.21]{MR2952218} that the operator $\OpW(a)$ extends to a bounded operator on $L^2(\R^d)$. In particular, the operator norm is uniformly bounded in $\hbar$, i.e.
\begin{equation*}
	\sup_{\hbar\in(0,\hbar_0]}\norm{\OpW(a)}_{\mathrm{op}} < \infty.
\end{equation*}

A {\em semiclassical Wigner measure} $\mu$ (also known as {\em defect measure}) associated to a family of states $\{\psi_\hbar\}_{\hbar}$ is a Radon measure on $\R^{2d}$ such that, for all $a\in\mathcal{S}(\R^{2d})$,
  \begin{equation}\label{uniqueW}
  	\langle \OpW(a) \psi_{\hbar}, \psi_{\hbar} \rangle \to \int_{\R^{2d}} a \, d\mu 
  \end{equation}
 along some subsequence of $\hbar\to 0$. Note that the standard assumption in the literature is usually ``for all $a\in C_0^\infty(\R)$'', but we will here need the stronger ``for all $a\in\mathcal{S}(\R^{2d})$''. The Wigner measure may in general depend on the subsequence and not be unique. We will ensure uniqueness in the following by considering suitable subsequences $\{\psi_\hbar\}_{\hbar\in I}$. We say a semiclassical family is {\em uniform} in $\mathcal{H}_\hbar^{n}(\R^d)$ if
\begin{equation*}
\sup_{\hbar\in I} \| \psi_\hbar \|_{\mathcal{H}_\hbar^{n}(\R^d)}<\infty,
\end{equation*}
 where $n\in\N$ and the space $\mathcal{H}_\hbar^{n}(\R^d)$ is the semiclassical Sobolev space and is given by  
 \begin{equation*}
 	\mathcal{H}_\hbar^{n}(\R^d) = \Big\{f\in L^2(\R^d)\, \Big| \, \sup_{|\alpha|\leq n} \int_{\R^d} |(-i\hbar\partial_x)^{\alpha} f(x)|^2 \,dx <\infty \Big\}.
 \end{equation*}
The following table shows examples of uniform semiclassical families $\{\psi_\hbar\}_{\hbar\in(0,1]}$ and their (unique) Wigner measures $\mu$; cf. \cite[pp.~102-104]{MR2952218}: 

\begin{center}
\begin{tabular}{ l | c | c   }
\hline
Type & $\psi_\hbar(x)$ & $d\mu(x,p)$ \\
\hline
Lagrangian & $w(x)\exp(i\hbar^{-1}\langle x,p_0\rangle)$ & $|w(x)|^2  \delta(p-p_0) \,dx\, dp$  \\
Lagrangian & $\hbar^{-d/2} w(\hbar^{-1} (x-x_0))$ & $\delta(x-x_0)  |\hat w(p)|^2 \,dx\, dp$  \\
WKB & $w(x)\exp(i\hbar^{-1}S(x)) $ & $|w(x)|^2 \delta(p-\partial_x  S(x)) \,dx\, dp$ \\
Coherent & $(\pi\hbar)^{-d/4} \exp(i\hbar^{-1}\langle x-x_0,p_0 +i2(x-x_0)\rangle)$ & $\delta(x-x_0)  \delta(p-p_0) \,dx\, dp$ \\
\hline
\end{tabular}
\end{center}
For the Lagrangian states we assume $w\in \mathcal{H}_1^{5d+5}(\R^d)$; for the WKB state we assume in addition that $w$ and $S$ are smooth and $w$ has compact support.
\subsection{The linear Boltzmann equation}

The  linear Boltzmann equation with initial data $f_0$ is given by
 \begin{equation}\label{intial_val_pro_boltzmanneqq}
\begin{cases}\displaystyle
	 \partial_t f(t,x,p) +\langle p, \nabla_x f(t,x,p) \rangle = \int_{\R^d} [\Sigma(p,q)f(t,x,q) - \Sigma(q,p) f(t,x,p)] \, dq
	\\
	f(0,x,p) = f_0(x,p),
\end{cases}
\end{equation}
where $\Sigma(p,q)$ is the collision kernel of the underlying microscopic single-site scattering process. The physical interpretation of $f(t,x,p)$ is that of a particle density in position $x\in\R^d$ and momentum $p\in\R^d$ at time $t$ of a gas of non-interacting particles in a low-density medium. 

The solution of the linear Boltzmann equation \eqref{intial_val_pro_boltzmanneqq} can be expressed as the collision series
 \begin{equation}\label{collision_series_expansion}
 	 f(t,x,p) = \sum_{n=0}^\infty  f^{(n)}(t,x,p),
 \end{equation}
where
\begin{equation}
	\begin{aligned}
	f^{(n)}(t,x,q_0)
	=   \int_{[0,t]_{\leq}^n}  \int_{\R^{nd}}   & \prod_{i=1}^{n} \Sigma(q_i,q_{i-1})  
	  \prod_{i=1}^{n+1}  e^{-(s_{i-1}-s_i) \Sigma_{\text{\rm tot}}(q)}  
	  \\  
	  &\times f_0(x- tq_0 -\sum_{i=1}^{n}s_{i}  (q_{i}-q_{i-1}) ,q_{n}) \,  d\boldsymbol{q}_{1,n}d\boldsymbol{s}_{n,1},
	  \end{aligned}
\end{equation}
with the convention $s_0=t$, $s_{n+1}=0$, the notation
\begin{equation}
	[0,t]_{\leq}^n = \{ s_n \leq \cdots \leq s_1 \,|\, s_i \in [0,t] \},
\end{equation}
and the total scattering cross section 
\begin{equation}
	\Sigma_{\text{\rm tot}}(q) = \int_{\R^d} \Sigma(p,q) \, dp .
\end{equation}
The expansion \eqref{collision_series_expansion} defines a semigroup $\{\mathcal{L}_t\}_{t\geq 0}$ of linear operators $\mathcal{L}_t :L^1(\R^{2d})\to L^1(\R^{2d})$ so that $f(t,x,p)=\mathcal{L}_t f_0(x,p)$. An important property of $\mathcal{L}_t$ is that it preserves mass, i.e., $\| \mathcal{L}_t f\|_1=\| f\|_1$ for every measurable $f\geq 0$. We define the adjoint $\mathcal{L}_t^*$ by
\begin{equation}
\int a(x,p) f(t,x,p) \,dx\, dp = \int [\mathcal{L}_t^* a](x,p) f_0(x,p)\, dx\, dp .
\end{equation}
Given a Borel measure $\mu_0$ on $\R^{2d}$, we define the measure $\mu_t$ by
\begin{equation}
\int a(x,p) \, d\mu_t(x,p) = \int [\mathcal{L}_t^* a](x,p) \, d\mu_0(x,p) .
\end{equation}
We will say {\em the family of measures $\{\mu_t\}_{t\geq 0}$ is a solution of the linear Boltzmann equation with initial data $\mu_0$.}

\subsection{Main theorems}

Define the norm
\begin{equation*}
	 \norm{ f}_{p,q,n} = \max\bigg( \sup_{\beta,|\epsilon|\leq n} \norm{\langle x \rangle^{\beta}\partial^\epsilon f}_{L^p(\R^d)},  \sup_{\beta,|\epsilon|\leq n} \norm{\langle x \rangle^{\beta}\partial^\epsilon f}_{L^q(\R^d)} \bigg).
\end{equation*}

{\em The standing assumption throughout this paper is that $\mathcal{X}$ is a realisation of a Poisson point process with intensity $\rho$, the single-site potential $V \in\mathcal{S}(\R^d)$ is non-negative, and for a constant $C_d$ only depending on the dimension the coupling constant satisfies 
\begin{equation} \label{assump_1}
\lambda  < \frac{1}{2   C_d \norm{ \hat{V}}_{1,\infty,5d+5}}.
\end{equation}
}

The probability and expectation of the Poisson process will be denoted by $\mathbb{P}$ and $\mathbb{E}$, respectively.
Under our assumptions, the Hamiltonian \eqref{Hami:def} is well defined as the unique self-adjoint Friedrichs extension of the associated form $\Pro$-almost surely (see Section \ref{sec:integrability} for details). 

\begin{thm}\label{main_thm1}
Let $d=3$ and suppose $\{\varphi_\hbar\}_{\hbar\in I}$ is a uniform semiclassical family in $\mathcal{H}_\hbar^{5d+5}(\R^3)$ with Wigner measure $\mu_0$,
and let $\varphi_{\hbar}(t)=U_{\hbar,\lambda}(t) \varphi_{\hbar}$.
 Then, for any $t>0$, $a\in\mathcal{S}(\R^{2d})$, we have that
  \begin{equation}
\mathbb{E}	\langle \OpW(a)  \varphi_{\hbar}(t), \varphi_{\hbar}(t) \rangle \to \int_{\R^{2d}} a(x,p) \, d\mu_t(x,p) , 
  \end{equation}
for $\hbar\to 0$ in $I$, where $\mu_t$ is a solution of the linear Boltzmann equation with initial data $\mu_0$ and collision kernel 
\begin{equation}\label{Sigma:def}
	\Sigma(p,q) = (2\pi)^{d+1} \rho\,  \delta( \tfrac{1}{2} p^2-\tfrac{1}{2}q^2)| \hat{T}(p,q)|^2 .
\end{equation}
\end{thm}

We restate Theorem \ref{main_thm1} in the Heisenberg picture, where one considers the time evolution of observables rather than wave functions. The Heisenberg evolution of an initial observable $A_\hbar=\OpW(a)$ is given by
\begin{equation}
A_\hbar(t)= U_{\hbar,\lambda}(t)\, A_\hbar\, U_{\hbar,\lambda}(-t) . 
\end{equation}

\begin{thm}\label{main_thm2}
Let $d=3$ and suppose $a\in\mathcal{S}(\R^{2d})$ with $A_\hbar(t)$ beeing the Heisenberg evolution of $A_\hbar=\OpW(a)$.
Then, for any $t>0$, and every uniform semiclassical family $\{\varphi_\hbar\}_{\hbar\in I}$ in $\mathcal{H}_\hbar^{5d+5}(\R^3)$ with Wigner measure $\mu_0$, we have that
  \begin{equation}
 \mathbb{E}	\langle A_\hbar(t) \varphi_{\hbar}, \varphi_{\hbar} \rangle \to  \int_{\R^{2d}} f(t,x,p) \, d\mu_0(x,p) , 
  \end{equation}
for $\hbar\to 0$ in $I$, where $f(t,x,p)$ is a solution of the linear Boltzmann equation with initial data $f(0,x,p)=a(x,p)$ and collision kernel \eqref{Sigma:def}.
\end{thm}

By unitarity,
\begin{equation}
\langle A_\hbar(t) \varphi_{\hbar}, \varphi_{\hbar} \rangle = \langle \OpW(a)  \varphi_{\hbar}(-t), \varphi_{\hbar}(-t) \rangle,
\end{equation}
so Theorem \ref{main_thm2} considers the reverse time direction of Theorem \ref{main_thm1} in the quantum evolution. Likewise, the limit evolutions in Theorems \ref{main_thm1} and Theorem \ref{main_thm2} are governed by $\mathcal{L}_t^* a$ and $\mathcal{L}_t a$, respectively. Thus the two theorems are in fact equivalent, and we will only give the proof of one of them: Theorem \ref{main_thm2}.

A natural question would be to ask if the ``smoothness'' of the semiclassical family is really needed for the result to be true. In the next theorem we will see that if the semiclassical family is  ``well approximated'' this assumption on smoothness can be relaxed. What we precisely mean by  ``well approximated''  will be made precise in the statement of the theorem.   
Firstly recall that a sequence of finite Borel measures $\{\mu_n\}_{n\in\N}$ on $\R^d$ converges weakly to a finite Borel measure $\mu$ on $\R^d$ if for all bounded continuous functions  $f$ it holds that
\begin{equation}
	\int f \,d\mu_n \rightarrow \int f \,d\mu \qquad\text{as $n\rightarrow \infty$.}
\end{equation}
With this we are now ready to state the generalisation.
\begin{thm}\label{Main_thm_gene}
Suppose $d=3$ and let $\{\varphi_\hbar\}_{\hbar\in I}$ uniform semiclassical family in $L^2(\R^d)$ with Wigner measure $\mu_0$.  Assume that for every $n\in\N$ there exists an uniform semiclassical family $\{\varphi_\hbar^n\}_{\hbar\in I}$ in $\mathcal{H}^{5d+5}_\hbar(\R^d)$ with Wigner measure $\mu_0^n$ such that
\begin{equation}
	\lim_{n\rightarrow \infty} \sup_{\hbar\in I} \norm{\varphi_\hbar-\varphi_\hbar^n}_{L^2(\R^d)} = 0
\end{equation}
and the sequence $\{\mu_0^n\}$ converges weakly to $\mu_0$. Then the assertion of Theorem~\ref{main_thm1} and Theorem~\ref{main_thm2} holds for the semiclassical family $\{\varphi_\hbar\}_{\hbar\in I}$.
\end{thm}   
A proof of this theorem will be given in Section~\ref{sec:proof_main_gen}. In this section we also give two examples of such ``well approximated'' semiclassical families. That such semiclassical families exists solidify the intuition that in this set up solutions to the linear Boltzmann equation are the universal limits no matter the initial conditions. 
\subsection{Outline of the paper and comparison to previous works}\label{sec_comparision} 
The remainder of the paper develops the various pieces needed for the proof of  Theorem~\ref{main_thm2}. From here on we will mostly be working in $\R^d$ for $d\geq3$ since most of our results are valid in all dimensions greater than or equal to $3$.  In Section~\ref{sec:prelim}  we firstly fix the choice of Fourier transform and calculate the momentum representation of two operators, which we will use frequently in the following sections. Next we state and prove some integral estimates, which is fundamental in this approach. Lastly we prove that our Hamiltonian \eqref{Hami:def} is well defined under the stated assumptions on the potential, for almost every realisation of the Poisson process and some results on expectation of random phases.

The methods used in the proof of Theorem~\ref{main_thm2} are based on Duhamel expansions of the Schr\"odinger evolution in terms of multiple collision integrals. 
 We will not do a ``full'' expansion but truncate the expansion for certain collision patterns. This type of truncation was introduced by Erd{\H o}s and Yau in their study of the weak-coupling limit \cite{MR1744001}. For the expansions we will follow the methodology of Eng and Erd{\H o}s \cite{MR2156632}, but we need to expand to higher order in recollisions. We need these higher order terms due to our weakened assumptions on our initial states and to obtain improved error bounds. The specifics of the expansion and how we truncate it are given in Section~\ref{sec:duhamel_exp}. The main result in Section~\ref{sec:duhamel_exp} is Proposition~\ref{duhamel_expansion_lemma}, which provides the explicit expansion derived from the Duhamel principle. The section is structured such that we first define the functions and operators that arise in the expansion and establish some of their key properties. Then  Proposition~\ref{duhamel_expansion_lemma} is stated and proven. The expansion we consider have a different form than the expansion in \cite{MR2156632}. The two forms are equivalent and we can pass from one to the other by a change of variables in the time variables. This change of variables is given in Observation~\ref{obs_form_I_op_kernel}. We state the expansion in this form as it is convenient for determining the limit, but we use the other expression for the a priori estimates. 

A large part of this work is to estimate the different terms arising from the expansion. This is the content of  \Cref{Sec:tech_est_0_recol,Sec:tech_est_1_recol,Sec:tech_est_2_recol,Sec:tech_est_truncated_recol_1,Sec:tech_est_truncated_recol_2}.
Due to the presence of the additional terms in the Duhamel expansion and the general initial states we can not directly apply the estimates from \cite{MR2156632}. Instead we have modified the techniques and arguments used in \cite{MR2156632} to obtain estimates for the different terms in the expansions. The main point of the modification is that we a priori need less derivatives of the initial states. Secondly we also obtain improved bounds from the modifications. One of these improvements is that for some estimates we get a factor $k$ compared to a factor $k!$ in  \cite{MR2156632}. 
The modification is explained in more details in Remark~\ref{remark_mod_eng_erdos_1}. Moreover, we will in \Cref{Sec:tech_est_truncated_recol_1,Sec:tech_est_truncated_recol_2} use a different technique compared to \cite{MR2156632} for estimating some of the terms in the recollision error. 

Since we are working in the framework of semiclassical/Wigner measures we need different techniques to prove the actual convergence of the main terms in the expansion compared to \cite{MR2156632}. Especially we can not use that we are in a box of finite volume when we do regulations/renormalisations of the main terms. However, we will also use some of the techniques from \cite{MR2156632} used to obtain the estimates on the terms in the Duhamel expansion from the previous sections. This is the content of \Cref{Sec:reg_main_op,Sec:con_error_term,Sec:limit_main_terms}.  

 In Section~\ref{Sec:reg_main_op} we prove two lemmas where we ``regularise'' the main terms in the expansion. The first one can be thought of as going to complex energies for the $T$-operator; the second is to remove some of the $\hbar$-dependence from the main terms. It is in this section where we start working with the form of the operators used in the statement of Proposition~\ref{duhamel_expansion_lemma}. In \cite{MR2156632} they also need to regularise or renormalise as they call it. In both cases this is done to get the main terms of the expansion into a form for which convergence can be proven.   
 
 In Section~\ref{Sec:con_error_term} we prove that the error terms do indeed go to zero in the semiclasical limit. The proof of the first Lemma is done by collecting the previous established norm bounds and verifying that under suitable choice of two numbers the error term convergence to zero. This is the same as in \cite{MR2156632}. In the second Lemma we prove that the ``main'' terms with recollisions is of lower order and vanish in the semicalssical limit. Here we again use the modified techniques from  \cite{MR2156632} and some new techniques to handle certain terms.    

Section~\ref{Sec:limit_main_terms}  is devoted to working out semiclassical limit for the main terms in the expansion. This is done in a number of different steps including a further regularisation. In Section~\ref{sec:proof_main} we collect all previous results into a proof of Theorem~\ref{main_thm2}.

\subsection*{Acknowledgement}
I would like to thank Professor Jens Marklof for numerous discussions on this subject and for introducing me to it. I would also like to thank Professor L\'aszl\'o Erd{\H o}s for valuable feedback on a previous version of this paper.

\section{Preliminaries}\label{sec:prelim}

\subsection{Fourier transform and momentum representation}\label{momentum_rep_apx}

We will use the following normalisation of the semiclassical Fourier transform and its inverse,
\begin{equation*}
	\begin{aligned}
	\mathcal{F}_\hbar\varphi(p) &= \int_{\R^d} e^{-i\hbar^{-1}\langle x,p\rangle}\varphi(x) \, dx
	\\
	\mathcal{F}_\hbar^{-1}\psi(x) &= \frac{1}{(2\pi\hbar)^d} \int_{\R^d} e^{i\hbar^{-1}\langle x,p\rangle}\psi(p) \, dp,
	\end{aligned}
\end{equation*} 
for $\varphi$ and $\psi \in L^2(\R^d)$. In the case $\hbar=1$ we will use the notation $\hat{\varphi}=\mathcal{F}\varphi=\mathcal{F}_1\varphi$. In the following we will also drop the subscript on the integrals if it is just ``the right power'' of $\R$. For some specific cases we will keep it though like in the statement of Lemma~\ref{LE:Fourier_trans_resolvent}, which is only valid in $3$-dimensions. 

For an operator $A$ we define its momentum representation by $\hat{A} =\mathcal{F}_\hbar A \mathcal{F}_\hbar^{-1} $. We will here give three examples of kernels in momentum representation. We will later in the analysis be using all three. We have in the sense of distributions that the kernels in momentum representation are
\begin{equation*}
	\begin{aligned}
	\mathcal{F}_\hbar W_\hbar \mathcal{F}_\hbar^{-1}(p,q) = \frac{1}{(2\pi)^{d}} \sum_{x_j\in\mathcal{X}} e^{-i\hbar^{-1/d}\langle x_j,p-q \rangle} \hat{V}(p-q),
	\end{aligned}
\end{equation*}
\begin{equation*}
	\begin{aligned}
	\mathcal{F}_\hbar e^{-it\hbar^{-1}(-\frac{\hbar^2}{2}\Delta)} \mathcal{F}_\hbar^{-1}(p,q) 
	=  \delta(p-q) e^{-it\hbar^{-1}\frac{1}{2}q^2}
	\end{aligned}
\end{equation*}
and 
\begin{equation*}
	\begin{aligned}
	K_{\hat{A}_\hbar}(p,q)
	 =\frac{1}{(2\pi\hbar)^{d}} \mathcal{F}_\hbar [a(\cdot,\tfrac{p+q}{2})](p-q).
	\end{aligned}
\end{equation*}
Lastly we have the following Lemma for which a proof can be found in e.g. \cite{MR3969938}.
\begin{lemma}\label{LE:Fourier_trans_resolvent}
Let $z\in\C$ be given such that $\im(z)\neq0$. Then we have
 \begin{equation}
 	\frac{1}{(2\pi)^3} \int_{\R^3} \frac{e^{i\langle x,p\rangle}}{\frac{1}{2} p^2 - z } \,dp = \frac{e^{i|x|\sqrt{2z}}}{2\pi |x|},
 \end{equation}
 where it is the branch of the square root giving a positive imaginary part. 
\end{lemma}
 \subsection{Integral estimates}
In the following the set of natural numbers (excluding zero) is denoted $\N$, and we also set $\N_0 = \{0\}\cup\N$. For a collection of vectors $x_l, x_{l+1}, \dots,x_j$ and scalars $s_l,\dots,s_j$ we define.
\begin{equation*}
	 \boldsymbol{x}_{l,j} = (x_l,  \dots,x_j),\qquad \boldsymbol{x}_{l,j}^{+}=\sum_{i=l}^j x_i ,\qquad \boldsymbol{sx}_{l,j}^{+}=\sum_{i=l}^j s_ix_i .
\end{equation*}
Furthermore, will use the convention that for $m<n$ that
\begin{equation*}
	\prod_{i=n}^m a_i =1 \qquad\text{and}\qquad \sum_{i=n}^m a_i =0.
\end{equation*}
The following lemma is a version of the classical stationary phase argument. We have here included a proof since we will be using the same type of argumentation in later proofs.
   \begin{lemma}\label{app_quadratic_integral_tech_est}
  Let $n\in\N$  and $f\in\mathcal{S}(\R^{nd})$. Suppose $Q(\boldsymbol{t})$ is an $nd\times nd$ diagonal matrix with block form 
  \begin{equation*}
  	Q(\boldsymbol{t}) = \begin{pmatrix}
	t_1 I_d & 0 & 0 &0
	\\
	0 & t_2 I_d &0&0
	\\
	\vdots & 0 & \ddots &0
	\\
	0 &\cdots & 0 &t_n I_d
	\end{pmatrix},
  \end{equation*}
  with $t_i\in\R$ and $I_d$ is the $d\times d$ identity matrix and finally let $\boldsymbol{y}\in\R^{nd}$. Then   
  \begin{equation*}
  	\begin{aligned}
  	\MoveEqLeft \Big|\int f(x_1,\dots,x_n) e^{i\frac{1}{2}\langle Q(\boldsymbol{t}) \boldsymbol{x},\boldsymbol{x}\rangle} e^{i\langle \boldsymbol{x},\boldsymbol{y}\rangle}  \,d\boldsymbol{x} \Big| 
	\\
	\leq {}& \tilde{C}_d^{n}\prod_{i=1}^n \frac{1}{\max(1, | t_i |)^\frac{d}{2}} \sup_{|\alpha_1|,\dots,|\alpha_n|\leq d+1} \norm{  \partial_{x_1}^{\alpha_1}\cdots \partial_{x_n}^{\alpha_n} f(x_1,\dots,x_n)}_{L^1(\R^{nd})},
	\end{aligned}
  \end{equation*}
  where the constant is only dependent on the dimension $d$ and is given by
  \begin{equation*}
   \tilde{C}_d= \left(  \frac{\sqrt{\pi}  d (d+1)^{(d+1)} \Gamma(\frac{d}{2})}{\Gamma(\frac{d}{2}+1) \Gamma(\frac{d+1}{2})} \right).
  \end{equation*}
  \end{lemma}
  \begin{proof}
  We start by dividing our integral into different cases depending on the absolute value of the $t$'s. We divide into the cases where $|t_i|\geq1$ and  $|t_i|<1$. This gives us $2^n$ different cases to consider and we denote each corresponding set by $B_l$. For each case the indices for which $|t_i|\geq1$ will be denoted by $J_1$ and the remaining will be in $J_2$. With this in place we write the integral as
    \begin{equation}\label{App_quadratic_integral_tech_est_1}
  	\int f(x_1,\dots,x_n) e^{i\frac{1}{2}\langle Q(\boldsymbol{t}) \boldsymbol{x},\boldsymbol{x}\rangle} e^{i\langle \boldsymbol{x},\boldsymbol{y}\rangle}  \,d\boldsymbol{x} = \sum_{l=1}^{2^n} \int \boldsymbol{1}_{B_l}(\boldsymbol{t}) f(x_1,\dots,x_n) e^{i\frac{1}{2}\langle Q(\boldsymbol{t}) \boldsymbol{x},\boldsymbol{x}\rangle}e^{i\langle \boldsymbol{x},\boldsymbol{y}\rangle}  \,d\boldsymbol{x}.
  \end{equation}
  We will decompose $Q(\boldsymbol{t})$ into two different matrices depending on which indices is in $J_1$ or $J_2$. By assumption we have that
  \begin{equation}\label{App_decom_Q_mat}
  	\begin{aligned}
  	 \langle Q(\boldsymbol{t}) \boldsymbol{x},\boldsymbol{x}\rangle 
	= \sum_{i=1}^n t_i^d x_i^2 
	= \sum_{i\in J_1}t_i^d x_i^2   + \sum_{i\in J_2}t_i^d x_i^2  
	=\langle Q_1(\boldsymbol{t}) \boldsymbol{x}_{J_1},\boldsymbol{x}_{J_1}\rangle+\langle Q_2(\boldsymbol{t}) \boldsymbol{x}_{J_2},\boldsymbol{x}_{J_2}\rangle ,
 	 \end{aligned}
  \end{equation}
  where $Q_1(\boldsymbol{t}) $ and $Q_2(\boldsymbol{t}) $ are diagonal square matrices. In particular we have that
  \begin{equation*}
  	|\det[Q_1(\boldsymbol{t})] |= \prod_{i\in J_1} |t_i|^d \geq1.
  \end{equation*}
  Hence using decomposition \eqref{App_decom_Q_mat} and writing the function $e^{i\frac{1}{2}\langle Q_1(\boldsymbol{t}) \boldsymbol{x},\boldsymbol{x}\rangle}$ as the inverse Fourier transform of its Fourier transform we get that
      \begin{equation*}
      \begin{aligned}
  	\int \boldsymbol{1}_{B_l}(\boldsymbol{t}) f(x_1,\dots,x_n) e^{i\frac{1}{2}\langle Q(\boldsymbol{t}) \boldsymbol{x},\boldsymbol{x}\rangle} \,d\boldsymbol{x}
	 ={}&\prod_{i\in J_1} \frac{1}{(2\pi t_i)^\frac{d}{2}} e^{\frac{i\pi}{4}\sgn(Q_1(\boldsymbol{t}))} \int \boldsymbol{1}_{B_l}(\boldsymbol{t}) f(x_1,\dots,x_n) e^{-i\frac{1}{2}\langle Q_1(\boldsymbol{t})^{-1} \boldsymbol{p}_{J_1},\boldsymbol{p}_{J_1}\rangle}
	 \\
	 &
	 \times e^{i\frac{1}{2}\langle Q_2(\boldsymbol{t}) \boldsymbol{x}_{J_2},\boldsymbol{x}_{J_2}\rangle}e^{i\langle \boldsymbol{p}_{J_1}+\boldsymbol{y}_{J_1} ,\boldsymbol{x}_{J_1}\rangle} e^{i\langle \boldsymbol{x}_{J_2},\boldsymbol{y}_{J_2}\rangle} \,d\boldsymbol{p}d\boldsymbol{x} .
	\end{aligned}
  \end{equation*}
  To ensure integrability in the $p$ variables we define the differential operator
  \begin{equation*}
	 L =\prod_{i\in J_1} \frac{1-i\langle p_i  , \nabla_{x_i}\rangle}{1+|p_i |^2}.
\end{equation*}
We observe that $L e^{i\langle \boldsymbol{p}_{J_1} ,\boldsymbol{x}_{J_1}\rangle}=e^{i\langle \boldsymbol{p}_{J_1} ,\boldsymbol{x}_{J_1}\rangle}$.  Using this, we insert $L$, $d+1$ times, and integrate by parts. This gives us the  estimate
 \begin{equation}\label{App_quadratic_integral_tech_est_2}
      \begin{aligned}
  	\MoveEqLeft  \Big| \int \boldsymbol{1}_{B_l}(\boldsymbol{t}) f(x_1,\dots,x_n) e^{i\frac{1}{2}\langle Q(\boldsymbol{t}) \boldsymbol{x},\boldsymbol{x}\rangle} e^{i\langle \boldsymbol{x},\boldsymbol{y}\rangle} \,d\boldsymbol{x} \Big|
	 \leq \sum_{|\alpha_1|,\dots,|\alpha_{j_1}|\leq d+1} \boldsymbol{1}_{B_l}(\boldsymbol{t}) \prod_{i\in J_1} \frac{C(\alpha_i)}{(2\pi | t_i |)^\frac{d}{2}}  
	 \\
	 &
	 \times  \int | \partial_{\boldsymbol{x}_{J_1}}^{\boldsymbol{\alpha}_{J_1}} f(x_1,\dots,x_n)| \prod_{i\in J_1} \frac{|p_i ^{\alpha_i} |}{(1+|p_i |^2)^{d+1}} \,d\boldsymbol{p}_{J_1}d\boldsymbol{x},
	\end{aligned}
  \end{equation}
  where we have used the notation  $C(\alpha)= \binom{d+1}{(\alpha,d+1-\abs{\alpha})}$. We note that
  \begin{equation}\label{App_quadratic_integral_tech_est_3}
	\begin{aligned}
	 \int \frac{|p_i^{\alpha_{i}}|}{(1+|p_i|^2)^{d+1}} \, dp 
	 \leq \int \frac{1}{(1+|p_i|^2)^{\frac{d+1}{2}}} \, dp_i
	 = \frac{\pi^\frac{d}{2} d}{\Gamma(\frac{d}{2}+1)} \int_{0}^\infty \frac{r^{d-1}}{(1+r^2)^{\frac{d+1}{2}}} \,dr =  \frac{\pi^\frac{d+1}{2} d \Gamma(\frac{d}{2})}{2\Gamma(\frac{d}{2}+1) \Gamma(\frac{d+1}{2})}.
  	\end{aligned}
\end{equation}
  Combining \eqref{App_quadratic_integral_tech_est_2} and \eqref{App_quadratic_integral_tech_est_3} we obtain the estimate
  \begin{equation}\label{App_quadratic_integral_tech_est_4}
      \begin{aligned}
  	\MoveEqLeft  \Big| \int \boldsymbol{1}_{B_l}(\boldsymbol{t}) f(x_1,\dots,x_n) e^{i\frac{1}{2}\langle Q(\boldsymbol{t}) \boldsymbol{x},\boldsymbol{x}\rangle} e^{i\langle \boldsymbol{x},\boldsymbol{y}\rangle}\,d\boldsymbol{x}  \Big|
	 \leq \left(  \frac{\sqrt{\pi} d \Gamma(\frac{d}{2})}{2\Gamma(\frac{d}{2}+1) \Gamma(\frac{d+1}{2})} \right)^{j_1}
	 \\
	 &
	 \times \sup_{|\alpha_1|,\dots,|\alpha_n|\leq d+1} \norm{  \partial_{x_1}^{\alpha_1}\cdots \partial_{x_n}^{\alpha_n} f(x_1,\dots,x_n)}_{L^1(\R^{nd})}  \sum_{|\alpha_1|,\dots,|\alpha_{j_1}|\leq d+1} \boldsymbol{1}_{B_l}(\boldsymbol{t}) \prod_{i\in J_1} \frac{C(\alpha_i)}{ | t_i |^\frac{d}{2}}.  
	\end{aligned}
  \end{equation}
To obtain the desired estimate we first note for all $l$ that 
  \begin{equation}\label{App_quadratic_integral_tech_est_5}
      \begin{aligned}
  	\MoveEqLeft   \boldsymbol{1}_{B_l}(\boldsymbol{t}) \prod_{i\in J_1} \frac{1}{ | t_i |^\frac{d}{2}} =  \prod_{i=1}^n \frac{1}{\max(1, | t_i |)^\frac{d}{2}}. 
	\end{aligned}
  \end{equation}
Secondly we note that
 \begin{equation}\label{App_quadratic_integral_tech_est_6}
	\begin{aligned}
	 \sum_{|\alpha_{1}|,\dots,|\alpha_{j_1}|\leq d+1}  \prod_{i=1}^{j_1}   \binom{d+1}{(\alpha_{i},d+1-\abs{\alpha_{i}})} 
	=   \prod_{i=1}^{j_1}  \sum_{|\alpha_{i}| \leq d+1}   \binom{d+1}{(\alpha_{i},d+1-\abs{\alpha_{i}})} 
	=(d+1)^{(d+1)j_1}.
	\end{aligned}
\end{equation}
  Combing \eqref{App_quadratic_integral_tech_est_4}, \eqref{App_quadratic_integral_tech_est_5} and \eqref{App_quadratic_integral_tech_est_6} we get the estimate
    \begin{equation*}
      \begin{aligned}
  	\MoveEqLeft  \Big| \int \boldsymbol{1}_{B_l}(\boldsymbol{t}) f(x_1,\dots,x_n) e^{i\frac{1}{2}\langle Q(\boldsymbol{t}) \boldsymbol{x},\boldsymbol{x}\rangle} e^{i\langle \boldsymbol{x},\boldsymbol{y}\rangle} \,d\boldsymbol{x} \Big|
	\\
	& \leq \left(  \frac{\tilde{C}_d}{2}\right)^{n}
	   \prod_{i=1}^n \frac{1}{\max(1, | t_i |)^\frac{d}{2}} \sup_{|\alpha_1|,\dots,|\alpha_n|\leq d+1} \norm{  \partial_{x_1}^{\alpha_1}\cdots \partial_{x_n}^{\alpha_n} f(x_1,\dots,x_n)}_{L^1(\R^{nd})}.  
	\end{aligned}
  \end{equation*}
  which combined with \eqref{App_quadratic_integral_tech_est_1} gives us the stated estimate.
    \end{proof}
\begin{lemma}\label{LE:gamma_reg_time_int}
	Let $\gamma\in(0,1]$ be given. Then the following estimate holds
	\begin{equation*}
		\int_0^\infty \frac{1-e^{-\gamma t}}{\max(1,t^{\frac{3}{2}})} \,dt \leq 5 \gamma^{\frac{1}{3}}.
	\end{equation*}
\end{lemma}
\begin{proof}
We start by setting $T=\gamma^{-\frac{2}{3}}$ and write the integral as
\begin{equation}\label{EQ:gamma_reg_time_int_1} 
	\int_0^\infty \frac{1-e^{-\gamma t}}{\max(1,t^{\frac{3}{2}})} \,dt  = \int_0^T \frac{1-e^{-\gamma t}}{\max(1,t^{\frac{3}{2}})} \,dt  + \int_T^\infty \frac{1-e^{-\gamma t}}{\max(1,t^{\frac{3}{2}})} \,dt. 
\end{equation}
For the first integral we use the estimate $1-e^{-\gamma t} \leq \gamma t$ and obtain that
\begin{equation}\label{EQ:gamma_reg_time_int_2} 
	 \int_0^T \frac{1-e^{-\gamma t}}{\max(1,t^{\frac{3}{2}})} \,dt  \leq  \gamma T  \int_0^T \frac{1}{\max(1,t^{\frac{3}{2}})}\,dt \leq  \gamma T  \int_0^\infty \frac{1}{\max(1,t^{\frac{3}{2}})}\,dt = 3 \gamma T = 3 \gamma^{\frac{1}{3}}.
\end{equation}
For the second integral we use the trivial bound $1-e^{-\gamma t} \leq 1$ and that $T\geq1$ and obtain that
\begin{equation}\label{EQ:gamma_reg_time_int_3} 
	\int_T^\infty \frac{1-e^{-\gamma t}}{\max(1,t^{\frac{3}{2}})} \,dt  \leq \int_T^\infty \frac{1}{t^{\frac{3}{2}}} \,dt  =2 T^{-\frac{1}{2}} = 2 \gamma^{\frac{1}{3}}.
\end{equation}
The stated estimate is obtained from combining \eqref{EQ:gamma_reg_time_int_1}, \eqref{EQ:gamma_reg_time_int_2} and \eqref{EQ:gamma_reg_time_int_3}. 
\end{proof}
We will in our analysis also need the following integral estimates contained in the next two lemmas. The proofs of both lemmas are given in the Appendix.
\begin{lemma}\label{LE:int_posistion}
Let $y\in\R^d$ be given. Then the following estimates holds
 \begin{equation}
 	\int_{\R^d}  \frac{1}{\langle x \rangle^{d+1} |x-y|} \, dx \leq \frac{C}{\langle y \rangle^{\frac{1}{2}}}.
 \end{equation}
\end{lemma}
\begin{lemma}\label{LE:resolvent_int_est}
Let $\zeta\in(0,1/2)$ be given. Then the following estimates holds
 \begin{equation}\label{LE:resolvent_int_est_EQ_1}
 	\sup_{p\in\R^d}\int_\R  \frac{1}{\langle \nu \rangle |\frac{1}{2} p^2+\nu +i\zeta|} \, d\nu \leq C| \log(\zeta)|.
 \end{equation}
 \begin{equation}\label{LE:resolvent_int_est_EQ_2}
 	\sup_{\nu\in\R}\int_{\R^d}  \frac{\langle \nu\rangle }{\langle p \rangle^{d+1} |\frac{1}{2} p^2+\nu +i\zeta|} \, dp \leq C| \log(\zeta)|.
 \end{equation}
 \begin{equation}\label{LE:resolvent_int_est_EQ_3}
	\sup_{q\in\R^d,\nu\in\R}  \int_{\R^d}  \frac{1}{| p- q |}  \frac{\langle \nu\rangle}{ \langle p \rangle^{d+1}|\frac{1}{2} p^2+\nu+i\zeta|}
	\, dp \leq C |\log(\zeta)|.
\end{equation}
  \begin{equation}\label{LE:resolvent_int_est_EQ_4}
 	\sup_{q\in\R^d,\nu\in\R}\int_{\R^d}  \frac{1 }{\langle p-q \rangle^{d+1} |\frac{1}{2} p^2+\nu +i\zeta|} \, dp \leq C| \log(\zeta)|.
 \end{equation}
 Where all constants are independent of $\zeta$, but they depend on the dimension.
 \end{lemma}
The last integral estimate we will need is contained in the following Lemma and a proof of this Lemma can be found in \cite[Proposition 2.3]{MR2156632}. For sake of completness a proof can also be found in the appendix.
\begin{lemma}\label{LE:est_res_combined}
Let $\zeta\in(0,1/2)$ be given. Then the following estimate holds
 \begin{equation}
 	\sup_{\nu,\tilde{\nu}\in\R}\int_{\R^d}   \frac{\langle \nu\rangle \langle \tilde{\nu} \rangle }{\langle p \rangle^{d+1} \langle p-q \rangle^{d+1} |\frac{1}{2} p^2+\nu +i\zeta|  |\frac{1}{2} (p-q)^2+\tilde{\nu} +i\zeta|} \, dp \leq \frac{ C\log(\zeta)^2}{|q| + \zeta},
 \end{equation}
 where the constant is independent of $\zeta$ and $q$.
\end{lemma}  
\begin{remark}
For applications of this Lemma we will use it in the following form
 \begin{equation}
 	\sup_{\nu,\tilde{\nu}\in\R}\int_{\R^d}   \frac{\langle \nu\rangle \langle \tilde{\nu} \rangle }{\langle p \rangle^{d+1} \langle p-q \rangle^{d+1} |\frac{1}{2} p^2+\nu +i\zeta|  |\frac{1}{2} (p-q)^2+\tilde{\nu} +i\zeta|} \, dp \leq \frac{ C\log(\zeta)^2}{|q|},
 \end{equation}
where this bound is only true almost everywhere with respect to the Lebesgue measure.
\end{remark}
\subsection{Almost sure local integrability of the potential and expectation of random phases} \label{sec:integrability}
We are considering an operator of the form $H_{\hbar,\lambda}=-\hbar^2\Delta +\lambda W_\hbar$, where the potential is of the form 
\begin{equation*}
	W_\hbar(x) =\sum_{x_j\in\mathcal{X}} V(\hbar^{-1}x-\hbar^{-1/d}x_j),
\end{equation*}
where $\mathcal{X}$ is a Possion point process with intensity measure $\rho dx$ and $V$ is a measurable function. A first natural question is weather this operator is well defined and self-adjoint. One approach is to establish that  $W_\hbar$ is in $L^1_{loc}(\R^d)$ $\Pro$-almost surely and is bounded from below. Then we can define the operator $\Pro$-almost surely as a Friedrichs extension of the quadratic form given by 
\begin{equation*}
  \mathfrak{h}[f,g] = \int \hbar^2\sum_{i=1}^d \partial_{x_i}f(x) \overline{\partial_{x_i}g(x)}  + \lambda W_\hbar(x)f(x)\overline{g(x)}\;dx, \qquad f,g \in \mathcal{D}(\mathfrak{h}),
\end{equation*}
where
\begin{equation*}
  \mathcal{D}(\mathfrak{h}) = \left\{ f\in L^2(\R^d) | \int_{\R^d} \abs{p}^2 \abs{\hat{f}(p)}^2 \;dp<\infty \text{ and } \int \abs{W_\hbar(x)}\abs{f(x)}^2 \;dx <\infty \right\}.
\end{equation*}
In fact this Friedrichs extension will  be self-adjoint and unique; cf. \cite{MR0493420}.

For a start assume $V(\hbar^{-1}(x-\hbar^{1-1/d}\cdot))$ to be in $L^1(\R^d,\rho dx)$, uniformly for all $x \in \R^d$. By Campbell's theorem we have that
\begin{equation*}
	\Aver{|W_\hbar(x)|} \leq \Aver{\sum_{x_j\in\mathcal{X}} |V(\hbar^{-1}x-\hbar^{-1/d}x_j)|  } = \rho \int |V(\hbar^{-1}x-\hbar^{-1/d}y)| dy<\infty .
\end{equation*}
Then by the Markov inequality we get for $\varepsilon>0$ that
\begin{equation*}
	\Pro(|W_\hbar(x)| > \varepsilon^{-1}) \leq \varepsilon \Aver{|W_\hbar(x)|}.
\end{equation*}
This gives us that the probability of the total potential  being pointwise finite is $1$ as long as the single site potential is integrable. This implies that the potential $\Pro$-almost surely will be finite on any countable collection of points. However, in the case of stationary Poisson point process (as we consider) the potential will not be in $L^\infty(\R^d)$ almost surely. To see this let $\{U_n\}_{n\in\N}$ be a collection of disjoint sets with volume one such that $\R^d=\cup_{n\in\N} U_n$. For a $k \in \N$ consider the  probability
\begin{equation*}
	\Pro( \exists n \in \N \, : \#(\mathcal{X}\cap U_n)>k) = 1- \Pro(  \#(\mathcal{X}\cap U_n)\leq k  \, \forall n \in \N).
\end{equation*}
Since we assume that the sets have volume $1$ we get for all $n \in \N$ that
\begin{equation*}
	\Pro( \#(\mathcal{X}\cap U_n)\leq k) =e^{-\rho} \sum_{j=0}^k \frac{\rho^j}{j!} <1.
\end{equation*}
Moreover, we assumed the sets are disjoint, hence we get by independence that
\begin{equation*}
	\Pro( \#(\mathcal{X}\cap U_n)\leq k \,  \forall n \in \N) = \lim_{l\rightarrow \infty} \big(e^{-\rho} \sum_{j=0}^k \frac{\rho^j}{j!} \big)^l =0.
\end{equation*}
Hence for every $k \in \N$ we have that
\begin{equation*}
	\Pro( \exists n \in \N \, : \#(\mathcal{X}\cap U_n)>k) = 1.
\end{equation*}
This will imply that we cannot $\Pro$-almost surely have a uniformly bound on the potential. This is also the reason we assume our single site potentials is nonnegative, since otherwise we cannot verify almost surely that the full potential $W_\hbar$ will be bounded from below. We will in the following lemma prove that the full potential will be in $L^1_{loc}(\R^d)$ almost surely. 
\begin{lemma}\label{Properties_Wr}
Let $\mathcal{X}$ be a Poisson point process with intensity $\rho$ and let $V \in \mathcal{S}(\R^d)$. For $\hbar \in (0,1]$ define the function $W_\hbar(x)$ by
\begin{equation}
	W_\hbar(x)= \sum_{x_j\in\mathcal{X}} V(\hbar^{-1}x-\hbar^{-1/d}x_j).
\end{equation}
Then for given $\hbar$ we have that $W_\hbar(x) \in L^1_{loc}(\R^d)$ $\Pro$-almost surely.
\end{lemma}
\begin{proof}
First we observer by use of Campbell's theorem and for any $\hbar>0$ that
\begin{equation}
	\Aver{\sum_{x_j\in\mathcal{X}} \langle  \hbar^{-1/d}x_j\rangle^{-d-1}} =\hbar  \int \langle x\rangle^{-d-1}\rho \,dx <\infty,
\end{equation}
where we have used the notation $\langle x\rangle = (1+|x|^2)^{1/2}$. This implies that 
\begin{equation*}
\sum_{x_j\in\mathcal{X}} \langle  \hbar^{-1/d}x_j\rangle^{-d-1} < \infty \qquad\text{$\Pro$-a.s.}
\end{equation*}
To prove that $W_\hbar(x)\in L^1_{loc}(\R^d)$ $\Pro$-almost surely, it is sufficient to prove that 
\begin{equation}
	\int_{B(0,R)} |W_\hbar(x)| \, dx <\infty \qquad\text{$\Pro$-a.s.}
\end{equation}
for every $R>0$, where the $\Pro$-null set is independent of $R$. Let $R>0$ be given. For $x \in B(0,\hbar^{-1}R)$ we have that
\begin{equation}
	\begin{aligned}
	\langle \hbar^{-1/d} x_j \rangle &=( 1+ |\hbar^{-1/d}x_j-x+x|^2)^{\frac{1}{2}} 
	\\
	&\leq \sqrt{2}(\langle \hbar^{-1/d}x_j-x  \rangle  +\hbar^{-1}R) \leq \sqrt{2}(1+\hbar^{-1}R)\langle \hbar^{-1/d}x_j-x \rangle.
	\end{aligned}
\end{equation}
This implies that
\begin{equation}
 	\frac{\langle \hbar^{-1/d} x_j \rangle}{\sqrt{2}(1+\hbar^{-1} R)} \leq \langle \hbar^{-1/d}x_j-x \rangle. 
\end{equation}
With this bound we have that
\begin{equation*}
	\begin{aligned}
 	\int_{B(0,R)} |W_\hbar(x)| \, dx  
	 \leq {}& \hbar^{d} \sum_{x_j \in \mathcal{X}}   \int_{B(0,\hbar^{-1}R)} |V(x-\hbar^{-1/d}x_j)| \, dx
	 \\
	 \leq {}& \sqrt{2}(1+\hbar^{-1} R)^{d+1} C \Vol(B(0,R))\sum_{x_j \in \mathcal{X} }  \frac{1}{\langle \hbar^{-1/d} x_j \rangle^{d+1}}
	 <\infty \quad\text{$\Pro$-a.s.}
	 \end{aligned}
\end{equation*}
where we have that the $\Pro$-null set is independent of $R$. This concludes the proof.  
\end{proof}   
Our next lemma is a result on the average of a random phase. To state the Lemma we will need some additional notation to keep track of the combinatorics involved in this problem. We will by  $\mathcal{S}_n$ denote the set of permutations of $n$ elements and we let
\begin{equation} \label{orderset_def}
	\mathcal{A}(k,n) = \{ \boldsymbol{\sigma} \in \{1,\dots,k\}^n \, |\, \sigma_1<\sigma_2<\cdots<\sigma_n \}.
\end{equation}
We can now state the lemma.
 \begin{lemma}\label{LE:Exp_ran_phases}
  Let $k_1,k_2\in\N$ such that $k_1\geq k_2$, $n\in\{0,\dots,k_2\}$, $\sigma^1\in\mathcal{A}(k_1,n)$ and $\sigma^2\in\mathcal{A}(k_2,n)$. Suppose $f\in L^2(\R^{(k_1+1)d})$, $g\in L^2(\R^{(k_2+1)d})$ and let $\kappa$ be in $\mathcal{S}_n$. Then
  \begin{equation}\label{Exp_ran_phases_eq}
  	\begin{aligned}
  	\MoveEqLeft \mathbb{E} \Big[ \sum_{(\boldsymbol{x},\tilde{\boldsymbol{x}})\in \mathcal{X}_{\neq}^{k_1+k_2}}  \prod_{i=1}^n \rho^{-1}\delta(x_{\sigma_i^1}- \tilde{x}_{\sigma_{\kappa(i)}^2})  \int \prod_{m=1}^{k_1} e^{ -i   \langle  \hbar^{-1/d}x_m,p_{m}-p_{m-1} \rangle} \prod_{m=1}^{k_2} e^{ i   \langle  \hbar^{-1/d}\tilde{x}_m,q_{m}-q_{m-1} \rangle}  f(\boldsymbol{p})  g(\boldsymbol{q}) \, d\boldsymbol{p}  d\boldsymbol{q}  \Big]
	\\
	&=\left(\rho\hbar(2\pi)^d\right)^{k_1+k_2-n}  \int
	\Lambda_n(\boldsymbol{p},\boldsymbol{q},\sigma^1,\sigma^2,\kappa)   f(\boldsymbol{p})  g(\boldsymbol{q}) \, d\boldsymbol{p}   d\boldsymbol{q} ,
	\end{aligned}
  \end{equation}
  where the function $\Lambda_n(\boldsymbol{p},\boldsymbol{q},\sigma^1,\sigma^2,\kappa)$ is given by
  \begin{equation*}
  	\begin{aligned}
  	\MoveEqLeft \Lambda_n(\boldsymbol{p},\boldsymbol{q},\sigma^1,\sigma^2,\kappa) 
	\\
	&= \prod_{i=1}^n \delta(p_{\sigma^1_{i-1}}- q_{\sigma^2_{\kappa(i)-1}} - p_{\sigma^1_{n}}+ q_{\sigma^2_{n}} - l^{\kappa}_{i}(\boldsymbol{q}_{\sigma^2}))   
	 \prod_{i=1}^{n+1} \prod_{m=\sigma^1_{i-1}+1}^{\sigma^1_{i}-1} \delta(p_m-p_{\sigma^1_{i-1}})  \prod_{m=\sigma^2_{i-1}+1}^{\sigma^2_{i}-1} \delta(q_m-q_{\sigma^1_{i-1}}),
	\end{aligned}
  \end{equation*}
  where we have used the notation $l^{\kappa}_{i}(\boldsymbol{q}_{\sigma^2})= \sum_{j=i+1}^{n+1} ( q_{\sigma^2_{\kappa(j)-1}}-q_{\sigma^2_{\kappa(j-1)}} )$.
    \end{lemma}
    \begin{remark}\label{RE:LE:Exp_ran_phases}
    Assume we are in the setting of Lemma~\ref{LE:Exp_ran_phases}. For any $\kappa\in\mathcal{S}_n$ we have that
    \begin{equation*}
    q_{\sigma^2_{\kappa(1)-1}} + l^{\kappa}_{1}(\boldsymbol{q}_{\sigma^2}) = q_0.
    \end{equation*}
    In the case where $\kappa=\mathrm{id}$ we have that $l^{\kappa}_{i}(\boldsymbol{q}_{\sigma^2}) =0$ for all $i$. For all other $\kappa$, there is at least one $i$ such that $l^{\kappa}_{i}(\boldsymbol{q}_{\sigma^2}) \neq 0$.
    
    We will later also need a slightly different version of the Lemma, where the function $\Lambda_n(\boldsymbol{p},\boldsymbol{q},\sigma^1,\sigma^2,\kappa)$ will be given by
      \begin{equation*}
  	\begin{aligned}
  	\MoveEqLeft \Lambda_n(\boldsymbol{p},\boldsymbol{q},\sigma^1,\sigma^2,\kappa) 
	\\
	&= \prod_{i=1}^n \delta(p_{\sigma^1_{i}}- q_{\sigma^2_{\kappa(i)}} - p_{\sigma^1_{1}}+ q_{\sigma^2_{1}} -\tilde{l}^{\kappa}_{i}(\boldsymbol{q}))   
	 \prod_{i=1}^{n+1} \prod_{m=\sigma^1_{i-1}+1}^{\sigma^1_{i}-1} \delta(p_m-p_{\sigma^1_{i}})  \prod_{m=\sigma^2_{i-1}+1}^{\sigma^2_{i}-1} \delta(q_m-q_{\sigma^1_{i}}),
	\end{aligned}
  \end{equation*}
    where we have used the notation $\tilde{l}^{\kappa}_{i}(\boldsymbol{q})= \sum_{j=1}^{i} ( q_{{\kappa(j-1)}}-q_{{\kappa(j)-1}} )$. To obtain this other statement the final change of variables in the following proof needs to be change to $x_n\mapsto x_n$ and $x_i\mapsto x_i-x_{i+1}$ for all $i\in\{1,\dots,n-1\}$. 
    \end{remark}
  \begin{proof}
  	We will in the proof use the notation $\mathcal{T}(\mathcal{X},k_1,k_2)$ for the integral we take expectation over in the lefthand side of \eqref{Exp_ran_phases_eq}. Using Campbell's theorem we get that the average is given by
\begin{equation*}
  	\begin{aligned}
  	\MoveEqLeft \Aver{\mathcal{T}(\mathcal{X},k_1,k_2)} 
	\\
	= {}& \rho^{k_1+k_2-n}  \int   \prod_{i=1}^n \delta(x_{\sigma_i^1}- \tilde{x}_{\sigma_{\tau(i)}^2}) \prod_{m=1}^{k_1} e^{ -i   \langle  \hbar^{-1/d}x_m,p_{m}-p_{m-1} \rangle} 
	 \prod_{m=1}^{k_2} e^{ i   \langle  \hbar^{-1/d}\tilde{x}_m,q_{m}-q_{m-1} \rangle}  f(\boldsymbol{p})  g(\boldsymbol{q}) \, d\boldsymbol{x} d\boldsymbol{\tilde{x}}d\boldsymbol{p}   d\boldsymbol{q}.  
	\end{aligned}
  \end{equation*}
	We can write the phase functions just involving $\boldsymbol{x}$ in the following way
	\begin{equation*}
		\begin{aligned}
	 \prod_{m=1}^{k_1} e^{ -i   \langle  \hbar^{-1/d}x_m,p_{m}-p_{m-1} \rangle} = \prod_{i=1}^{n+1}  e^{ -i  \hbar^{-1/d}  \langle  x_{\sigma^1_i},p_{\sigma^1_i} \rangle}e^{ i  \hbar^{-1/d}  \langle  x_{\sigma^1_{i-1}+1},p_{\sigma^1_{i-1}} \rangle} \prod_{m=\sigma^1_{i-1}+1}^{\sigma^1_{i}-1} e^{ -i  \hbar^{-1/d}  \langle  x_m-x_{m+1},p_{m} \rangle},
	\end{aligned}
	\end{equation*}
	where we have used the conventions that $\sigma^1_0=0$ and $\sigma^1_{n+1}=k_1$. In the cases where $\sigma^1_n=k_1$ the last factor in the product over $i$ is identical 1. The phases functions involving $\boldsymbol{\tilde{x}}$ can be written in an analogous way.
	Next  for all $i\in\{1,\dots,n+1\}$ we do the following change of variables
 \begin{equation*}
 \begin{aligned}
 	&x_{\sigma^1_i}\mapsto x_{\sigma^1_i} \quad\text{and}\quad  x_m\mapsto x_m - x_{m+1} \quad\text{for all $m\in\{ \sigma^1_{{i}-1}+1,\dots, {\sigma^1_{i}-1} \}$} 
	\\ 
	&\tilde{x}_{\sigma^2_i}\mapsto \tilde{x}_{\sigma^2_i} \quad\text{and}\quad  \tilde{x}_m\mapsto \tilde{x}_m -\tilde{ x}_{m+1} \quad\text{for all $m\in\{ \sigma^2_{{i}-1}+1,\dots, {\sigma^2_{i}-1} \}$} .
	\end{aligned}
 \end{equation*}
After this change of variables we evaluate all integrals in $\boldsymbol{\tilde{x}}$ and all $\boldsymbol{x}$ except $x_{\sigma_i^1}$ for all $i\in\{1,\dots,n\}$. This gives us the following expression
\begin{equation*}
  	\begin{aligned}
  	\Aver{\mathcal{T}(\mathcal{X},k_1,k_2)} 
	= \frac{(\hbar\rho(2\pi)^d)^{k_1+k_2-n}}{(2\pi)^{dn}}  \int  & \prod_{i=1}^n  e^{ -i   \langle  x_{i},p_{\sigma^1_i} -p_{\sigma^1_{i-1}}- q_{\sigma^2_{\kappa(i)}} +q_{\sigma^2_{\kappa(i)-1}} \rangle} \prod_{i=1}^{n+1} \Big\{ \prod_{m=\sigma^1_{i-1}+1}^{\sigma^1_{i}-1} \delta(p_m-p_{\sigma^1_{i-1}})
	\\
	&\times  \prod_{m=\sigma^2_{i-1}+1}^{\sigma^2_{i}-1} \delta(q_m-q_{\sigma^1_{i-1}})\Big\}  f(\boldsymbol{p})  g(\boldsymbol{q}) \, d\boldsymbol{x}  d\boldsymbol{p}   d\boldsymbol{q} , 
	\end{aligned}
  \end{equation*}
where we have done a relabelling and rescaling of the remaining $x$ variables. To obtain the stated form in the Lemma we do the change of variables $x_1\mapsto x_1$ and $x_i\mapsto x_i - x_{i-1}$ for $i  \in \{2,\dots,n\}$.  This yields
\begin{equation*}
  	\begin{aligned}
  	\MoveEqLeft \Aver{\mathcal{T}(\mathcal{X},k_1,k_2)} 
	= \frac{(\hbar\rho(2\pi)^d)^{k_1+k_2-n}}{(2\pi)^{dn}}  \int   \prod_{i=1}^n  e^{ i   \langle  x_{i},p_{\sigma^1_{i-1}}- q_{\sigma^2_{\kappa(i)-1}} - p_{\sigma^1_{n}}+ q_{\sigma^2_{n}} - l^{\kappa}_{i}(\boldsymbol{q}_{\sigma^2}) \rangle} 
	\\
	&\times\prod_{i=1}^{n+1} \prod_{m=\sigma^1_{i-1}+1}^{\sigma^1_{i}-1} \delta(p_m-p_{\sigma^1_{i-1}})  \prod_{m=\sigma^2_{i-1}+1}^{\sigma^2_{i}-1} \delta(q_m-q_{\sigma^1_{i-1}})  f(\boldsymbol{p})  g(\boldsymbol{q}) \, d\boldsymbol{x}  d\boldsymbol{p}  d\boldsymbol{q}, 
	\end{aligned}
  \end{equation*}
where we have used the notation $l^{\kappa}_{i}(\boldsymbol{q}_{\sigma^2})= \sum_{j=i+1}^{n+1} ( q_{\sigma^2_{\kappa(j)-1}}-q_{\sigma^2_{\kappa(j-1)}} )$. To obtain the result we evaluate the remaining integrals in $x$. This concludes the proof.  
  \end{proof}
Before the next Lemma we will need the following notation. We set
\begin{equation*}
\mathcal{X}_{\neq}^k = \{ (x_1,\ldots,x_k)\in \mathcal{X}^k :  \text{$x_i\neq x_j$ for all $i\neq j$}\}.   
\end{equation*} 
\begin{lemma}\label{LE:crossing_dom_ladder}
Let $k\in\N$ and $f$ a functions such that $f\in L^2(\R^{(k+1)d})$.  Then 
\begin{equation*}
	\begin{aligned}
	 \mathbb{E}\Big[ \big\lVert \sum_{\boldsymbol{x}\in \mathcal{X}_{\neq}^{k}}  f(\boldsymbol{x},y)\big\rVert_{L^2(\R^d_y)}^2 \Big] \leq 2 \sum_{n=0}^{k} \binom{k}{n} n!  \sum_{\sigma\in\mathcal{A}(k,n)}  \mathbb{E}\Big[\sum_{(\boldsymbol{x},\boldsymbol{\tilde{x}})\in \mathcal{X}_{\neq}^{2k}}  \prod_{i=1}^n \frac{\delta(x_{\sigma_i}- \tilde{x}_{\sigma_{i}}) }{\rho} \int_{\R^d} f(\boldsymbol{x},y)\overline{f(\boldsymbol{\tilde{x}},y)} \, dy \Big]. 
	\end{aligned}
\end{equation*}
\end{lemma}
\begin{proof}
From writing out the $L^2$-norm we get that
\begin{equation}\label{EQ:crossing_dom_ladder_1}
	\big\lVert \sum_{\boldsymbol{x}\in \mathcal{X}_{\neq}^{k}}  f(\boldsymbol{x},y)\big\rVert_{L^2(\R^d_y)}^2 = \sum_{n=0}^k \sum_{\sigma^1,\sigma^2\in\mathcal{A}(k,n)} \sum_{\kappa\in\mathcal{S}_n} \sum_{(\boldsymbol{x},\boldsymbol{\tilde{x}})\in \mathcal{X}_{\neq}^{2k}} \prod_{i=1}^n \frac{\delta(x_{\sigma_i}- \tilde{x}_{\sigma_{\kappa(i)}}) }{\rho} \int_{\R^d} f(\boldsymbol{x},y)\overline{f(\boldsymbol{\tilde{x}},y)} \, dy,
\end{equation}
where we are summing over all possible ways we can pair the two collision series. By linearity of the average we will consider a fixed $n$, $\sigma^1$, $\sigma^2$ and $\kappa$. Moreover, we will denote by $\mathcal{X}_{\sigma^i}$ the points of the point process who's index is given by the values of $\sigma^i$. Independence of the points of the point process gives us that
\begin{equation}\label{EQ:crossing_dom_ladder_2}
	\begin{aligned}
	\MoveEqLeft \mathbb{E}\Big[\sum_{(\boldsymbol{x},\boldsymbol{\tilde{x}})\in \mathcal{X}_{\neq}^{2k}} \prod_{i=1}^n \frac{\delta(x_{\sigma_i^1}- \tilde{x}_{\sigma_{\kappa(i)}^2}) }{\rho} \int_{\R^d} f(\boldsymbol{x},y)\overline{f(\boldsymbol{\tilde{x}},y)} \, dy \Big] 
	\\
	&\leq   \mathbb{E}_{\mathcal{X}_\sigma} \Big[  \int_{\R^d} \big| \mathbb{E}_{\mathcal{X}^{\neq}_{k}\setminus\mathcal{X}_{\sigma^1}} [ f(\boldsymbol{x},y)]    \mathbb{E}_{\mathcal{X}^{\neq}_{k}\setminus\mathcal{X}_{\sigma^2}} [\overline{f(\boldsymbol{\tilde{x}},y)} ] \big| \,dy \Big],
	\end{aligned}
\end{equation}	
where we have used that the left hand side is a positive number and moved the absolute value under the integrals to obtain the inequality. Using that for $a,b\geq0$ we have the inequality $ab\leq a^2+b^2$ we obtain that
\begin{equation}\label{EQ:crossing_dom_ladder_3}
	\begin{aligned}	
	\MoveEqLeft  \mathbb{E}_{\mathcal{X}_\sigma} \Big[  \int_{\R^d} \big| \mathbb{E}_{\mathcal{X}^{\neq}_{k}\setminus\mathcal{X}_{\sigma^1}} [ f(\boldsymbol{x},y)]    \mathbb{E}_{\mathcal{X}^{\neq}_{k}\setminus\mathcal{X}_{\sigma^2}} [\overline{f(\boldsymbol{\tilde{x}},y)} ] \big| \,dy \Big]
	\\
	\leq{}&  \mathbb{E}_{\mathcal{X}_\sigma} \big[  \int_{\R^d} \big| \mathbb{E}_{\mathcal{X}^{\neq}_{k-n}\setminus\mathcal{X}_{\sigma^1}} \big[ f(\boldsymbol{x},y)  \big]  \big|^2 +  \big| \mathbb{E}_{\mathcal{X}^{\neq}_{k-n}\setminus\mathcal{X}_{\sigma^2}}\big[\overline{ f(\boldsymbol{x},y)}  \big] \big|^2\,dy \big].
	\end{aligned}
\end{equation}
Using independence of the points in the point process again we get that
\begin{equation}\label{EQ:crossing_dom_ladder_4}
	\begin{aligned}	
	\MoveEqLeft  \mathbb{E}_{\mathcal{X}_\sigma} \big[  \int_{\R^d} \big| \mathbb{E}_{\mathcal{X}^{\neq}_{k-n}\setminus\mathcal{X}_{\sigma^1}} \big[ f(\boldsymbol{x},y)  \big]  \big|^2 +  \big| \mathbb{E}_{\mathcal{X}^{\neq}_{k-n}\setminus\mathcal{X}_{\sigma^2}}\big[\overline{ f(\boldsymbol{x},y))}  \big] \big|^2\,dy \big]
	\\
	&= \sum_{j=1}^2  \mathbb{E}\Big[\sum_{(\boldsymbol{x},\boldsymbol{\tilde{x}})\in \mathcal{X}_{\neq}^{2k}} \prod_{i=1}^n \frac{\delta(x_{\sigma_i^j}- \tilde{x}_{\sigma_{i}^j}) }{\rho} \int_{\R^d} f(\boldsymbol{x},y)\overline{f(\boldsymbol{\tilde{x}},y)} \, dy \Big]. 
	\end{aligned}
\end{equation}
From combining \eqref{EQ:crossing_dom_ladder_2}, \eqref{EQ:crossing_dom_ladder_3} and  \eqref{EQ:crossing_dom_ladder_4} we obtain the estimate
\begin{equation}\label{EQ:crossing_dom_ladder_5}
	\begin{aligned}
	\MoveEqLeft \mathbb{E}\Big[\sum_{(\boldsymbol{x},\boldsymbol{\tilde{x}})\in \mathcal{X}_{\neq}^{2k}} \prod_{i=1}^n \frac{\delta(x_{\sigma_i^1}- \tilde{x}_{\sigma_{\kappa(i)}^2}) }{\rho} \int_{\R^d} f(\boldsymbol{x},y)\overline{f(\boldsymbol{\tilde{x}},y)} \, dy \Big] 
	\\
	&\leq   \sum_{j=1}^2  \mathbb{E}\Big[\sum_{(\boldsymbol{x},\boldsymbol{\tilde{x}})\in \mathcal{X}_{\neq}^{2k}}\prod_{i=1}^n \frac{\delta(x_{\sigma_i^j}- \tilde{x}_{\sigma_{i}^j}) }{\rho} \int_{\R^d} f(\boldsymbol{x},y)\overline{f(\boldsymbol{\tilde{x}},y)} \, dy \Big].
	\end{aligned}
\end{equation}	
Using that the number of elements in $\mathcal{S}_n$ is $n!$ and the number of elements in $\mathcal{A}(k,n)$ is $\binom{k}{n}$ combined with \eqref{EQ:crossing_dom_ladder_1} and \eqref{EQ:crossing_dom_ladder_5} we obtain the stated estimate.
\end{proof}
\section{Duhamel expansion}\label{sec:duhamel_exp}
We start this section by introducing the heuristics of the Duhamel expansion and then introduce functions and operators that will appear naturally in our expansion and establish basic properties of these.
We will then perform the Duhamel expansion. The convergence of the expansion will rely on some technical estimates that will be proved in the following sections. The convergence of the expansion in the limit will be proven as a part of the proof for our main theorem. The difficulties in obtaining the convergence in the limit is that we need uniform bounds in $\hbar$ or bounds that goes to zero as $\hbar$ goes to zero. If we only needed the convergence for $\hbar$ fixed, it would be ``easier'' to obtain the convergence but the error terms would divergence as $\hbar$ goes to zero.

 The expansions builds on the fundamental theorem of calculus which gives us the Duhamel formula 
\begin{equation*}
	\begin{aligned}
	U_{\hbar,\lambda}(-t)U_{\hbar,0}(t) -I
	&=   \frac{i\lambda}{\hbar}  \int_0^t U_{\hbar,\lambda}(-t_1) W_\hbar U_{\hbar,0}(-t_1)   \, d\boldsymbol{t}_{1}
	\\
	&=   \frac{i\lambda}{\hbar}  \int_0^t U_{\hbar,\lambda}(-t_1) U_{\hbar,0}(-t_1)   W_\hbar^{t_1}   \, d\boldsymbol{t}_{1},
	\end{aligned}
\end{equation*}
where we have used the notation $ W_\hbar^{t_1}  =  U_{\hbar,0}(t_1)  W_\hbar U_{\hbar,0}(-t_1)$. We will obtain our expansion by iterating the above procedure. We separate terms into diagonal and off-diagonal forms, e.g.
\begin{equation*}
	W_\hbar^{t_2} W_\hbar^{t_1} =  \sum_{x\in\mathcal{X}} V_{\hbar,x}^{t_2}  \sum_{x\in\mathcal{X}} V_{\hbar,x}^{t_1} =  \sum_{x\in\mathcal{X}} V_{\hbar,x}^{t_2}V_{\hbar,x}^{t_1} +  \sum_{(x_1,x_2)\in\mathcal{X}_{\neq}^2} V_{\hbar,x_2}^{t_2} V_{\hbar,x_1}^{t_1}.
\end{equation*}
In each step of our iteration we will preform this type of sorting, which produces all possible ``collision series'' up to that point. This is simply a list over the potentials we encounter, where we allow for repetitions. When consecutive  potentials are located at the same point we call it ``internal scattering''. We will in our expansion consider these as one collision event and when the next potential is located at another point  as the next collision with a number of internal scatterings. We will denote the number of internal scatterings at the $i$'th collision by $\alpha_i$. For us a recollision will be if a potential at some location appears again non consecutively.     

When we do these expansions we will continuously keep track of the collision series. If certain collision terms appear we will stop the expansion. To be precise we will stop when the third recollision appears in a collision series. But if we just do this for the full operator $U_{\hbar,\lambda}(-t)$ we will not be able to control the error terms. So we will combine this with a time division argument. That is we write the operator $U_{\hbar,\lambda}(-t)$ as a product of the following form  
\begin{equation*}
	U_{\hbar,\lambda}(-t) = \prod_{\tau=1}^{\tau_0}U_{\hbar,\lambda}(-\tfrac{t}{\tau_0}), 
\end{equation*}
then we will do the expansion for each operator in the product. This will produce some main terms and some error terms.  Heuristically we will get that
\begin{equation*}
	U_{\hbar,\lambda}(-\tfrac{t}{\tau_0}) = \sum \mathcal{I} + \mathcal{R}_1, 
\end{equation*}
where $I$ is the sum over our main term and $\mathcal{R}_1$ is the reminder. The explicit expressions for these terms will be given below. For the ``next'' $U_{\hbar,\lambda}(-\tfrac{t}{\tau_0})$ we only expand if we have the main terms to the right. So we obtain that 
\begin{equation*}
	U_{\hbar,\lambda}(-\tfrac{2t}{\tau_0}) = \sum \mathcal{I} +\mathcal{R}_2+  U_{\hbar,\lambda}(-\tfrac{1}{\tau_0}) \mathcal{R}_1. 
\end{equation*}
Iterating this process we obtain an expression of the form 
\begin{equation*}
	U_{\hbar,\lambda}(-t) = \sum \mathcal{I} +\sum_{\tau=1}^{\tau_0}  U_{\hbar,\lambda}(-\tfrac{t(\tau_0-\tau}{\tau_0}) \mathcal{R}_\tau.
\end{equation*}
The reminder/error terms $\mathcal{R}_\tau$ will not only contain terms where we have observed a third recollision but also terms, where we have seen more than $k_0$ different collision centers in the time interval $[0,\frac{t}{\tau_0}]$.  We will later choose the numbers $k_0$ and $\tau_0$ to be $\hbar$ depended such that we can control the reminder/error terms. 

As stated above we need some notation to  keep track of these collision series and expansions. We will in the following subsection define the functions and operators, which we will use to keep track of the expansions and prove some basic properties of these.
\subsection{Definitions and properties of functions and operators for Duhamel expansion} 
We start by defining a function that will appear in the kernel for the operator associated to the internal scattering events.
\begin{definition}\label{def_reg_operator}
For $m \in \N$, $s \in [0,\infty]$, $\gamma\geq0$ and $V \in \mathcal{S}(\R^d)$ we define
\begin{equation*}
	\begin{aligned}
	 \MoveEqLeft \Psi_{m}^\gamma(p_m,p_0,s ; V)
	\\
	&= \frac{1}{(2\pi)^{dm}}   \int_{\R_{+}^{m-1}}  \int  \boldsymbol{1}_{[0,s]}\big(  \boldsymbol{t}_{1,m-1}^{+}\big) 
  \hat{V}(p_m-p_{m-1})   
	 \prod_{i=1}^{m-1}  e^{i t_i \frac{1}{2} (p_{i}^2-p_{0}^2+i\gamma)}\hat{V}(p_{i}-p_{i-1})\, d\boldsymbol{p}_{1,m-1}  d\boldsymbol{t}.
	\end{aligned}
\end{equation*}
\end{definition}
The  parameter $\gamma$ should be viewed as a regularisation parameter, which we will need later. We will in the following also use the convention that $ \Psi_{m}^0= \Psi_{m}$. In later estimates it will be convenient to write the function  $\Psi_{m}^\gamma(p_m,p_0,s ; V)$ as
 \begin{equation} \label{obs_con_form_psi}
	\begin{aligned}
	 \Psi_{m}^\gamma(p_m,p_0,s ; V)
	={}&    \int_{\R_{+}^{m-1}}  \int  \boldsymbol{1}_{[0,s]}\big(  \boldsymbol{t}_{1,m-1}^{+}\big) 
  \hat{\mathcal{V}}_{m}(p_m,p_0,\boldsymbol{p}_{1,m-1})   
  	e^{- \frac{1}{2} \gamma\boldsymbol{t}_{1,m-1}^{+}}
	    \prod_{i=1}^{m-1} e^{i \frac{1}{2} t_{i}(p_i^2-p_0^2)} \, d\boldsymbol{p}_{1,m-1}  d\boldsymbol{t},
	\end{aligned}
\end{equation}
where the functions  $\hat{\mathcal{V}}_{m}(p_m,p_0,\boldsymbol{p}_{1,m-1})$ is defined by
\begin{equation}\label{obs_con_form_psi_f}
	\hat{\mathcal{V}}_{m}(p_{m},p_0,\boldsymbol{p}_{1,m-1}) = \frac{1}{(2\pi)^{dm}}  \begin{cases}
	   \hat{V}(p_{m}-p_{0}) & \text{if $m=1$,} 
	 \\
	  \hat{V}(p_{m}-p_{m-1}) \hat{V}(p_{1} -p_{0} ) \prod_{i=2}^{m-1}  \hat{V}(p_{i}-p_{i-1}) & \text{if $m>1$.} 
	  \end{cases}
\end{equation}
 \begin{lemma}\label{psi_m_gamma_est}
For $V \in \mathcal{S}(\R^d)$, $m \in \N$, $\hbar \in (0,\hbar_0]$, $\gamma\geq0$  and $0\leq s$, we have
%
\begin{equation*}
	 \int |  \Psi_{m}^\gamma(p_m,p_0,s ; V)| \,dp_{0/m} \leq   \frac{C_d^{m-1}}{(2\pi)^d}  \norm{  \hat{V}}_{1,\infty,2d+2}^m,
\end{equation*}
\begin{equation*}
	\sup_{p_m,p_0\in\R^d} |  \Psi_{m}^\gamma(p_m,p_0,s ; V)|  \leq   \frac{C_d^{m-1}}{(2\pi)^d}   \norm{  \hat{V}}_{1,\infty,2d+2}^m,
\end{equation*}
where $\,dp_{0/m}$ means either the integral with respect to $p_0$ or with respect to $p_m$. The bound is uniform for $\hbar \in [0,\hbar_0]$ and $\gamma\geq0$.  
\end{lemma}
 \begin{proof}
We will prove the bound for the case where we integrate with respect to $p_m$. The other case is proven analogously. We have by the definition of $ \Psi_{m}^\gamma$ (Definition~\ref{def_reg_operator}) that
\begin{equation*}
	\begin{aligned}
	 \int | \Psi_{m}^\gamma(p_m,p_0,s ; V)| \,dp_m
	\leq  \frac{1}{(2\pi)^{dm}} \int  \int_{\R_{+}^{m-1}} \Big|  \int \prod_{i=1}^{m}  \hat{V}(p_{i}-p_{i-1})   
 \prod_{i=1}^{m-1}  e^{i t_i \frac{1}{2} p_{i}^2}\, d\boldsymbol{p}_{1,m-1} \Big| d\boldsymbol{t} dp_m.
	\end{aligned}
\end{equation*}
We note here the integral over the $p$'s are of the same form as in Lemma~\ref{app_quadratic_integral_tech_est}. Hence we get the estimate
\begin{equation}\label{bound_psi_gam_1}
	\begin{aligned}
	\int_{\R^d} | \Psi_{m}^\gamma(p_m,p_0,s ; V)| \,dp_m
	\leq \frac{(d-2)^{m-1}}{d^{m-1} (2\pi)^d}   C_d^{m-1} \norm{ \hat{V}}_{1,\infty,2d+2}^m \int_{\R_{+}^{m-1}}   \prod_{i=1}^{m-1} \frac{1}{\max(1, | t_i |)^\frac{d}{2}}  d\boldsymbol{t}_{m-1,1}.
	\end{aligned}
\end{equation}
Moreover we have that
\begin{equation}\label{bound_psi_gam_2}
	\begin{aligned}
	 \int_{\R_{+}^{m-1}} \prod_{i=1}^{m-1} \frac{1}{\max(1, | t_i |)^\frac{d}{2}}  d\boldsymbol{t}_{m-1,1} = \left( \frac{d}{d-2} \right)^{m-1}.
  	\end{aligned}
\end{equation}
By combining \eqref{bound_psi_gam_1} and   \eqref{bound_psi_gam_2} we obtain the first estimate. The second estimate is proven analogously. This concludes the proof.
\end{proof}
 \begin{remark}\label{born_series_convergence_remark}
 Recall that the $T$-operator $T^\gamma(E)$ is defined as the $\gamma\to 0_+$ limit of 
\begin{equation*}
	T^\gamma(E) = \lambda V +\lambda^2 V \frac{1}{E-(-\frac{1}{2}\Delta+\lambda V)+i\gamma} V.
\end{equation*}
From the resolvent formalism we obtain the formal series expansion given by
\begin{equation*}
	T^\gamma(E) = \lambda V \sum_{n=0}^\infty \big[ \lambda  \frac{1}{E-(-\frac{1}{2}\Delta)+i\gamma} V \big]^n.
\end{equation*}
We are interested in the kernel of the above expression in momentum representation and we denote this kernel by $\hat{T}^\gamma(p,p_0) $, where we have set $E=\frac{1}{2}p_0^2$. When we do the same for the kernels in the formal expansion we observe that the kernel of the $n$'th operator is given by
\begin{equation*}
	i^{n-1} \Psi_{n}^\gamma(p,p_0,\infty ; V).
\end{equation*}
From this we get the formal expression for the kernel of the operator in momentum representation is given by
\begin{equation}\label{born_series}
	\hat{T}^\gamma(p,p_0) =-i  \sum_{n=1}^\infty (i\lambda)^{n} \Psi_{n}^\gamma(p,p_0,\infty ; V).
\end{equation}
Lemma~\ref{psi_m_gamma_est} give us, under the condition $\lambda C_d  \norm{  \hat{V}}_{1,\infty,2d+2} <1$,  that the series in \eqref{born_series} converges in the strong operator topology,  uniformly for all $\gamma\geq 0$.
 \end{remark}
The next definition is of the operators which describes the internal collisions.
 \begin{definition}\label{def_reg_operator_2}
For $m \in \N$, $s_1 \geq s_2 \geq 0$, $\hbar>0$, $z\in\R^d$ and $V \in \mathcal{S}(\R^d)$ we define $\Theta_{m}(s_{1},{s}_{2},z;V,\hbar)$ as the operator 
\begin{equation*}
	\begin{aligned}
 \frac{1}{\hbar} \int_{\R_{+}^{m-1}} \boldsymbol{1}_{[0,\hbar^{-1}(s_{1}-s_{2})]}\big(\boldsymbol{t}_{1,m-1}^+\big) U_{\hbar,0}(-s_{2}) V_{\hbar,z} \Big\{ \prod_{i=1}^{m-1} U_{\hbar,0}(-t_{i}) V_{\hbar,z}\Big\}  U_{\hbar,0}(\boldsymbol{t}_{1,m-1}^+ +s_2)
	   \, d\boldsymbol{t}_{1,m-1},
	\end{aligned}
\end{equation*}
where $V_{\hbar,z}$ is defined in \eqref{def_potential_int}. 
\end{definition}
\begin{remark}\label{Re:Kernel_of_theta}
Note that the kernel of $\Theta_{m}(s_{1},{s}_{2},z;V,\hbar)$ in momentum representation is given by
\begin{equation}\label{EQ:ef_reg_operator_2}
	\begin{aligned}
	 (p_{m},p_{0}) \mapsto{} \frac{1}{\hbar} e^{-i\hbar^{-1/d}\langle z,p_{m}-p_{0} \rangle} 
  e^{i s_{2}\hbar^{-1} \frac{1}{2} (p_{m}^2-p_{0}^2)} \Psi_{m}(p_{m},p_{0},\hbar^{-1}(s_{1}-s_{2}); V),
  	\end{aligned}
\end{equation}
where the functions $\Psi_{m}$ is defined in Definition~\ref{def_reg_operator}. Using this observation, Lemma~\ref{psi_m_gamma_est} and the Schur test we obtain that
\begin{equation*}
	\norm{\Theta_{m}(s_{1},{s}_{2},z;V,\hbar)}_{\mathrm{op}} \leq \hbar^{-1} C_d^{m-1} \sup_{|\beta|\leq2d+2} \norm{\partial^\beta \hat{V}}_{L^1(\R^d)}^{m},
\end{equation*}
where $\norm{}_{\mathrm{op}}$, is the operator norm. We will in the following for $\gamma\geq0$ use the notation $\Theta_{m}^\gamma(s_{1},{s}_{2},z;V,\hbar)$ for operators where the kernels have the same form as in \eqref{EQ:ef_reg_operator_2} with $ \Psi_{m}$ replaced by $ \Psi_{m}^\gamma$.
\end{remark} 
We will now start to define the operators associated to  fully expanded and truncated terms. We start with terms without recollisions. 
 \begin{definition}\label{functions_for_exp_def}
Let $H_{\hbar,\lambda} = -\frac{\hbar^2}{2}\Delta + \lambda W_\hbar(x)$, where  $W_\hbar$ and $\lambda$  satisfy Assumption~\eqref{assump_1}, $\hbar \in (0,\hbar_0]$ and  let $ U_{\hbar,\lambda}(t)=e^{-it\hbar^{-1}H_{\hbar,\lambda}}$.
For any $k \in \N$, $\gamma\geq0$ and $\boldsymbol{x} \in \mathcal{X}_{\neq}^k$ we define the operators
\begin{equation*}
	\begin{aligned}
	\mathcal{I}_{0,0}^\gamma(k,\boldsymbol{x},t;\hbar) = \sum_{\alpha\in \N^k} (i\lambda)^{|\alpha|}  \int_{[0,t]_{\leq}^k}   \prod_{m=1}^k \Theta_{\alpha_m}^\gamma(s_{m-1},{s}_{m},x_m;V,\hbar)\, d\boldsymbol{s}_{k,1}U_{\hbar,0}(-t).
	  \end{aligned}
\end{equation*}
We will use the convention $\mathcal{I}_{0,0}^0(k,\boldsymbol{x},t;\hbar)=\mathcal{I}_{0,0}(k,\boldsymbol{x},t;\hbar)$. For a number $k_0\in\N$ and $\boldsymbol{x} \in \mathcal{X}_{\neq}^{k_0+1}$ we define the operator 
\begin{equation*}
	\begin{aligned}
	\mathcal{E}^{k_0}_{0,0}(\boldsymbol{x},t;\hbar) =  \sum_{\alpha\in \N^{k_0} \times\{1\}}  (i\lambda)^{|\alpha|} \int_{0}^{t}  U_{\hbar,\lambda}(-s_{k_0+1}) U_{\hbar,0}(s_{k_0+1})  \tilde{\mathcal{E}}^{k_0}(\boldsymbol{x},s_{k_0+1},\alpha,t;\hbar) \, ds_{k_0+1},
	  \end{aligned}
\end{equation*}
where
\begin{equation*}
	\begin{aligned}
	\tilde{\mathcal{E}}^{k_0}_{0,0}(\boldsymbol{x},s_{k_0+1},\alpha,t;\hbar) =   \int_{[0,t]_{\leq}^{k_0}}   \boldsymbol{1}_{[s_{k_0+1},t]}(s_{k_0})   \prod_{m=1}^{k_0+1} \Theta_{\alpha_m}(s_{m-1},{s}_{m},x_m;V,\hbar)\, d\boldsymbol{s}_{k_0,1}U_{\hbar,0}(-t).
	  \end{aligned}
\end{equation*}
\end{definition}
\begin{remark}
Note that the operators just defined is indeed bounded operators. To see this use Remark~\ref{Re:Kernel_of_theta} and Assumption~\eqref{assump_1} to obtain the the sum congresses absolutely in the space of bounded operators. The same will be true for the operators defined in Definition~\ref{functions_for_exp_rec_def} and Definition~\ref{def_recol_error}. 
\end{remark}
Before we define the operators associated fully expanded and truncated terms, where we have recollisions, we need some notation to keep track of in which position we have the recollision. We do this by defining the following classes of maps. Firstly the class describing one recollision. 
\begin{definition}
	For any $k\geq3$ in $\N$ we let $\mathcal{Q}_{k,1,1}$ be the set of maps $\iota:\{1,\dots,k\}\mapsto\{1,\dots,k-1\}$ such that
	\begin{equation*}
		\iota(m) =
		\begin{cases}
		m & \text{if $m< m_2$}
		\\
		m_1 &\text{if $m= m_2$}
		\\
		m-1 & \text{if $m> m_2$},
		\end{cases}
	\end{equation*}
	where $m_1\in\{1,\dots,k-2\}$ and $m_2\in\{m_1+2,\dots,k\}$. Moreover, to each map $\iota$ we associate $\iota^{*}$ defined by
	\begin{equation*}
		\iota^{*}(m) =
		\begin{cases}
		m & \text{if $m< m_2$}
		\\
		m + 1 &\text{if $m\geq m_2$}.
		\end{cases}
	\end{equation*}
\end{definition}
For two recollisions we can either recollide two times at the same position or we can have two recollision with different potentials. We will distinguish these two scenarios by defining two diffrent classes of maps. The fist class is when we recollide twice at the same scatterer
\begin{definition}
	For any $k\geq5$ in $\N$ we let $\mathcal{Q}_{k,2,1}$  be the set of maps  $\iota:\{1,\dots,k\}\mapsto\{1,\dots,k-2\}$ such that  	%
	\begin{equation*}
		\iota(m) =
		\begin{cases}
		m & \text{if $m< m_2$}
		\\
		m_1 &\text{if $m= m_2$}
		\\
		m-1 & \text{if $ m_2<m<m_3$}
		\\
		m_1 &\text{if $m= m_3$}
		\\
		m-2 & \text{if $ m>m_3$},
		\end{cases}
	\end{equation*}
	where $m_1\in\{1,\dots,k-4\}$, $m_2\in\{m_1+2,\dots,k-2\}$ and $m_3\in\{m_2+2,\dots,k\}$. Moreover, to each map $\iota$ we associate $\iota^{*}$ defined by
	\begin{equation*}
		\iota^{*}(m) =
		\begin{cases}
		m & \text{if $m< m_2$}
		\\
		m + 1 &\text{if $ m_2 \leq m< m_3$}
		\\
		m + 2 &\text{if $ m \geq m_3$}.
		\end{cases}
	\end{equation*}
\end{definition}
\begin{definition}
	For any $k\geq4$ in $\N$ we let $\mathcal{Q}_{k,2,2}$ be the set of maps $\iota:\{1,\dots,k\}\mapsto\{1,\dots,k-2\}$ such that 
	\begin{equation*}
		\iota(m) =
		\begin{cases}
		m & \text{if $m< m_{1,2}$}
		\\
		m_{1,1} &\text{if $m= m_{1,2}$}
		\\
		m-1 & \text{if $ m_{1,2}<m<m_{2,2}$}
		\\
		m_{2,1} &\text{if $m= m_{2,2}$}
		\\
		m-2 & \text{if $ m>m_{2,2}$},
		\end{cases}
	\end{equation*}
	where $m_{1,1}\in\{1,\dots,k-2\}$, $m_{1,2}\in\{m_{1,1}+2,\dots,k-2\}$, $m_{2,1}\in\{1,\dots,k-2\}\setminus\{m_{1,1},m_{1,2} \}$ and $m_{2,2}\in\{\min(m_{2,1}+2,m_{1,2}+1),\dots,k\}$.  Moreover, to each map $\iota$ we associate $\iota^{*}$ defined by
	\begin{equation*}
		\iota^{*}(m) =
		\begin{cases}
		m & \text{if $m< m_{1,2}$}
		\\
		m + 1 &\text{if $ m_{1,2} \leq m< m_{2,2}$}
		\\
		m + 2 &\text{if $ m \geq m_{2,2}$}.
		\end{cases}
	\end{equation*}
\end{definition}
\begin{remark}\label{Remark_grow_of_Q}
In the following analysis we will need to understand how many maps are in $Q_{k,i,j}$ for $(i,j)\in\{(0,0),(1,1),(2,1),(2,2)\}$ and each $k$, where $Q_{k,0,0}=\{\mathrm{id}\}$. By counting the maps in each set we get that
\begin{equation*}
	\#Q_{k,i,j} = 
	\begin{cases}
	1 & \text{$(i,j)=(0,0)$ and $k\geq1$}
	\\
	\frac{(k-2)(k-1)}{2} & \text{$(i,j)=(1,1)$ and $k\geq3$}
	\\
	\frac{(k-4)(k-3)(k-2)}{6} & \text{$(i,j)=(2,1)$ and $k\geq5$}
	\\
	\frac{3k^4-30k^2+117k^2-210k+144}{24} & \text{$(i,j)=(2,2)$ and $k\geq4$},
	\end{cases}
\end{equation*}
where $\#Q_{k,i,j} $ is notation for the number of elements in the set $Q_{k,i,j}$.
\end{remark}
For the operators associated with fully expanded and truncated terms with recollisions we have the following definition: 
 \begin{definition}\label{functions_for_exp_rec_def}
Let $H_{\hbar,\lambda} = -\frac{\hbar^2}{2}\Delta + \lambda W_\hbar(x)$, where  $W_\hbar$ and $\lambda$  satisfy Assumption~\eqref{assump_1}, $\hbar \in (0,\hbar_0]$ and  let $ U_{\hbar,\lambda}(t)=e^{-it\hbar^{-1}H_{\hbar,\lambda}}$. For $(i,j)\in\{(1,1),(2,1),(2,2)\}$,  $k\geq k_{ij} $ in $\N$, $\boldsymbol{x} \in \mathcal{X}_{\neq}^{k-i}$	 and $\iota\in\mathcal{Q}_{k,i,j}$ we define the operators
\begin{equation*}
	\begin{aligned}
	\mathcal{I}_{i,j}(k,\boldsymbol{x},\iota,t;\hbar) =  \sum_{\alpha\in \N^k} (i\lambda)^{|\alpha|} \int_{[0,t]_{\leq}^k}   \prod_{m=1}^k \Theta_{\alpha_m}(s_{m-1},{s}_{m},x_{\iota(m)};V,\hbar)\, d\boldsymbol{s}_{k,1}U_{\hbar,0}(-t),
	  \end{aligned}
\end{equation*}
where $k_{11}=3$, $k_{21}=5$ and $k_{22}=4$. For a number $(i,j)\in\{(1,1),(2,1),(2,2)\}$, $k_0\geq5$ in $\N$, $\boldsymbol{x} \in \mathcal{X}_{\neq}^{k_0+1-i}$ and $\iota\in\mathcal{Q}_{k_0+1,i,j}$ we define the operator 
\begin{equation*}
	\begin{aligned}
	\mathcal{E}_{i,j}^{k_0}(\boldsymbol{x},\iota,t;\hbar) = \sum_{\alpha\in \N^{k_0} \times\{1\}}(i\lambda)^{|\alpha|}    \int_{0}^{t}  U_{\hbar,\lambda}(-s_{k_0}) U_{\hbar,0}(s_{k_0})  \tilde{\mathcal{E}}_{i,j}^{k_0}(\boldsymbol{x},s_{k_0},\alpha,\iota,t;\hbar) \, ds_{k_0},
	  \end{aligned}
\end{equation*}
where
\begin{equation*}
	\begin{aligned}
	\ \tilde{\mathcal{E}}_{i,j}^{k_0}(\boldsymbol{x},s_{k_0},\alpha,\iota,t;\hbar) =   \int_{[0,t]_{\leq}^{k_0}}   \boldsymbol{1}_{[s_{k_0},t]}(s_{k_0-1})   \prod_{m=1}^{k_0+1} \Theta_{\alpha_m}(s_{m-1},{s}_{m},x_{\iota(m)};V,\hbar)\, d\boldsymbol{s}_{k_0-1,1}U_{\hbar,0}(-t).
	  \end{aligned}
\end{equation*}
\end{definition}
Next definition is of the truncated terms where we have encountered a third recollision.
\begin{definition}\label{def_recol_error}
Let $H_{\hbar,\lambda} = -\frac{\hbar^2}{2}\Delta + \lambda W_\hbar(x)$, where  $W_\hbar$ and $\lambda$  satisfy Assumption~\eqref{assump_1}, $\hbar \in (0,\hbar_0]$ and  let $ U_{\hbar,\lambda}(t)=e^{-it\hbar^{-1}H_{\hbar,\lambda}}$. For $j\in\{1,2\}$, any $k,k_1 \in \N_0$, $\boldsymbol{x} \in \mathcal{X}_{\neq}^{k_1+k-2}$ and $\iota\in\mathcal{Q}_{k_1+k,2,j}$, where we assume $k_1+k\geq6-j$ we define the operators
\begin{equation*}
	\begin{aligned}
	\MoveEqLeft \mathcal{E}_{2,j}^{\mathrm{rec}}(k,k_1,\boldsymbol{x},t;\hbar) 
	\\
	={}&  \hbar^{-1} \sum_{\alpha\in \N^k}(i\lambda)^{|\alpha|+1}    \int_{0}^t   U_{\hbar,\lambda}(-s_{k+1}) U_{\hbar,0}(s_{k+1})\sum_{l=1,\iota(l)\neq k_1+k}^{k_1+k-2}   \tilde{\mathcal{E}}_{2,j}^{\mathrm{rec}}(s_{k+1},k,\boldsymbol{x},\alpha,l,t,\iota;\hbar) \,ds_{k+1},
	\end{aligned}
\end{equation*}
where
\begin{equation*}
	\begin{aligned}
	\MoveEqLeft \tilde{\mathcal{E}}_{2,j}^{\mathrm{rec}}(s_{k+1},k,\boldsymbol{x},\alpha,l,t,\iota;\hbar)
	\\
	={}&    \int_{[0,t]_{\leq}^{k}} \boldsymbol{1}_{[s_{k+1},t]}(s_k)   V^{s_{k+1}}_{\hbar, x_{l}} \prod_{m=1}^k \Theta_{\alpha_m}(s_{m-1},{s}_{m},x_{\iota(m+k_1)};V,\hbar)  \, d\boldsymbol{s}_{k,1}U_{\hbar,0}(-t). 
	\end{aligned}
\end{equation*}
\end{definition}
As described above we will not only preform one Duhamel expansion but do a time division and for each part we will preform the duhamel expansion. These last two definitions will define the different error operators that will appear when combining our stopping rules and the time division. 
\begin{definition}\label{def_remainder_k_0}
Assume we are in the same setting as Definition~\ref{functions_for_exp_def} and Definition~\ref{functions_for_exp_rec_def}. For numbers $\tau_0\in\N$, $k_0\geq5$ and $(i,j)\in\{(0,0),(1,1),(2,1),(2,2)\}$ we define the operator
\begin{equation*}
	\begin{aligned}
	\mathcal{R}_{1,i,j}^{k_0}(\hbar) =
	  \sum_{\iota\in\mathcal{Q}_{k_0+1,i,j}}\sum_{\boldsymbol{x}_1\in \mathcal{X}_{\neq}^{k_0+1-i}}  \mathcal{E}_{i,j}^{k_0}(\boldsymbol{x}_1,\iota,\tfrac{t}{\tau_0};\hbar) .
	 \end{aligned}
\end{equation*}
For $\tau\in\{2,\dots,\tau_0\}$ we define the operator
\begin{equation*}
	\begin{aligned}
\MoveEqLeft \mathcal{R}_{\tau,i,j}^{k_0}(N;\hbar) 
	 =  \sum_{\iota\in\mathcal{Q}_{k_0+1,i,j}}   \sum_{\boldsymbol{x}\in \mathcal{X}_{\neq}^{k_0+1-i}}  \mathcal{E}_{i,j}^{k_0}(\boldsymbol{x},\iota,\tfrac{t}{\tau_0};\hbar) U_{\hbar,0}(-\tfrac{\tau-1}{\tau_0}t;\hbar) 
	 \\
	 &+ \sum_{k_2=1}^{k_0} \sum_{\iota\in\mathcal{Q}_{k_0+k_2+1,i,j}}   \sum_{(\boldsymbol{x}_1,\boldsymbol{x}_2)\in \mathcal{X}_{\neq}^{k_0+k_2+1-i}}  \mathcal{E}_{i,j}^{k_0}(\boldsymbol{x}_1,\iota,\tfrac{t}{\tau_0};\hbar) \mathcal{I}_{i,j}(k_2,\boldsymbol{x}_2,\iota,\tfrac{\tau-1}{\tau_0}t;\hbar) 
	 \\
	 &+  \sum_{k_2=1}^{k_0} \sum_{\iota\in\mathcal{Q}_{k_0+k_2+1,i,j}}   \sum_{(\boldsymbol{x}_1,\boldsymbol{x}_2)\in \mathcal{X}_{\neq}^{k_0+k_2+1-i}}  \mathcal{E}_{i,j}^{k_0+1}((x_{2,k_2},\boldsymbol{x}_1),\iota,\tfrac{t}{\tau_0};\hbar) \mathcal{I}_{i,j}(k_2,\boldsymbol{x}_2,\iota,\tfrac{\tau-1}{\tau_0}t;\hbar) 
	 \\
	 &+ \sum_{k_1=1}^{k_0} \sum_{k_2=k_0-k_1+1}^{k_0} \sum_{\iota\in\mathcal{Q}_{k_1+k_2,i,j}} 
	  \sum_{(\boldsymbol{x}_1,\boldsymbol{x}_2)\in \mathcal{X}_{\neq}^{k_1+k_2-i}}  \mathcal{I}_{i,j}(k_1,\boldsymbol{x}_1,\iota,\tfrac{t}{\tau_0};\hbar)  \mathcal{I}_{i,j}(k_2,\boldsymbol{x}_2,\iota,\tfrac{\tau-1}{\tau_0}t;\hbar)
	  \\
	  &+  \sum_{k_1=1}^{k_0} \sum_{k_2=k_0-k_1+1}^{k_0} \sum_{\iota\in\mathcal{Q}_{k_1+k_2,i,j}} 
	  \sum_{(\boldsymbol{x}_1,\boldsymbol{x}_2)\in \mathcal{X}_{\neq}^{k_1+k_2-i}}  \mathcal{I}_{i,j}(k_1+1,(x_{2,k_2},\boldsymbol{x}_1),\iota,\tfrac{t}{\tau_0};\hbar)  \mathcal{I}_{i,j}(k_2,\boldsymbol{x}_2,\iota,\tfrac{\tau-1}{\tau_0}t;\hbar), 
	 \end{aligned}
\end{equation*}
where we use the convention that $\mathcal{Q}_{n,0,0}=\{\mathrm{id}\}$ for all $n\in\N$.
\end{definition}
\begin{definition}\label{def_recol_reminder}
Assume we are in the same setting as Definition~\ref{def_recol_error}. For numbers $\tau_0\in\N$ and $k_0\geq 5$ in $\N$  we define the operator
\begin{equation*}
	\mathcal{R}_1^{\mathrm{rec}}(k_0;\hbar)  = \sum_{j=1}^2
	  \sum_{k_1=6-j}^{k_0} \sum_{\iota\in\mathcal{Q}_{k_1},2,j}   \sum_{\boldsymbol{x}_1\in \mathcal{X}_{\neq}^{k_1-2}}  \mathcal{E}_{2,j}^{\mathrm{rec}}(k_1,0,\boldsymbol{x}_1,\iota,\tfrac{t}{\tau_0};\hbar).
\end{equation*}
For $\tau\in\{2,\dots,\tau_0\}$ we define the operators
\begin{equation*}
	\begin{aligned}
	 \mathcal{R}_\tau^{\mathrm{rec}}(k_0;\hbar) ={}&\sum_{j=1}^2 \sum_{k_1=1}^{k_0} \sum_{k_2=0}^{k_0} \sum_{\iota\in\mathcal{Q}_{k_1+k_2,2,j}}  \boldsymbol{1}_{\{k_1+k_2\geq6-j\}} 
	 \\
	 &\times  \sum_{(\boldsymbol{x}_1,\boldsymbol{x}_2)\in \mathcal{X}_{\neq}^{k_1+k_2-2}}  \mathcal{E}_{2,j}^{\mathrm{rec}}(k_1,k_2,\boldsymbol{x}_1,\iota,\tfrac{t}{\tau_0};\hbar) \mathcal{I}_{2,j}^{\mathrm{rec}}(k_2,\boldsymbol{x}_2,\iota,\tfrac{\tau-1}{\tau_0}t;\hbar)
	 \\
	 &+ \sum_{j=1}^2 \sum_{k_1=1}^{k_0} \sum_{k_2=0}^{k_0} \sum_{\iota\in\mathcal{Q}_{k_1+k_2,2,j}}  \boldsymbol{1}_{\{k_1+k_2\geq6-j\}} 
	 \\
	 &\times  \sum_{(\boldsymbol{x}_1,\boldsymbol{x}_2)\in \mathcal{X}_{\neq}^{k_1+k_2-2}}  \mathcal{E}_{2,j}^{\mathrm{rec}}(k_1+1,k_2,(x_{2,k_2},\boldsymbol{x}_1),\iota,\tfrac{t}{\tau_0};\hbar) \mathcal{I}_{2,j}^{\mathrm{rec}}(k_2,\boldsymbol{x}_2,\iota,\tfrac{\tau-1}{\tau_0}t;\hbar).
	\end{aligned}
\end{equation*}
\end{definition}
\subsection{The Duhamel expansion}
\begin{lemma}\label{duhamel_expansion_lemma} 
Let $H_{\hbar,\lambda} = -\frac{\hbar^2}{2}\Delta + \lambda W_\hbar(x)$, where $\hbar \in (0,\hbar_0]$,  $W_\hbar$ and $\lambda$  satisfy Assumption~\eqref{assump_1} and  let $ U_{\hbar,\lambda}(t)=e^{-it\hbar^{-1}H_{\hbar,\lambda}}$. Then for any $\tau_0\in\N$ and $k_0\in\N$ we have that
\begin{equation*}
	 U_{\hbar,\lambda}(-t) 
	= U_{\hbar,0}(-t) +  \sum_{i=0}^2 \sum_{j=\min(1,i)}^i   \sum_{k=k_{ij}}^{k_0} \sum_{\iota\in\mathcal{Q}_{k,i,j}}   \sum_{\boldsymbol{x}\in\mathcal{X}_{\neq}^{k-i}}      \mathcal{I}_{i,j}(k,\boldsymbol{x},t;\hbar) + \mathcal{R}(\tau_0,k_0:\hbar),
\end{equation*}
where $k_{00}=1$, $k_{11}=3$, $k_{21}=5$, $k_{22}=4$ and
\begin{equation*}
	\begin{aligned}
	 \mathcal{R}(\tau_0,k_0:\hbar) = \sum_{\tau=1}^{\tau_0} U_{\hbar,\lambda}(-\tfrac{\tau_0-\tau}{\tau_0}t) \big[  \mathcal{R}_\tau^{\mathrm{rec}}(k_0;\hbar)+ \sum_{i=0}^2 \sum_{j=\min(1,i)}^i \mathcal{R}_{\tau,i,j}^{k_0}(N;\hbar)  \big]. 
	\end{aligned}
\end{equation*}
The operators $\mathcal{I}_{i,j}(k,\boldsymbol{x},\alpha,t;\hbar) $ are defined in Definition~\ref{functions_for_exp_def} and \ref{functions_for_exp_rec_def}. The remainder operators are defined in Definition~\ref{def_remainder_k_0} and \ref{def_recol_reminder} respectively.
\end{lemma}
The proof of this Lemma is essentially a carefully done Duhamel expansion, where we at all time keep track of the different collision series. However we do combine this with a time division argument. This combination is the driving force between the slightly more complicated error terms appearing. 
\begin{proof}
From the self adjointness of $H_{\hbar,\lambda} = -\frac{\hbar^2}{2}\Delta + \lambda W_\hbar(x)$ we have that  $U_{\hbar,\lambda}(-t)$ is a semigroup parametrised by $t$. Hence we have the following decomposition
\begin{equation}\label{EQ:Duhamel_proof_eq_0}
	U_{\hbar,\lambda}(-t) = \prod_{\tau=1}^{\tau_0} U_{\hbar,\lambda}(-\tfrac{t}{\tau_0}) = U_{\hbar,\lambda}(-\tfrac{t(\tau_0-1)}{\tau_0}) U_{\hbar,\lambda}(-\tfrac{t}{\tau_0}).
\end{equation}
In the following we let $t_0=\frac{t}{\tau_0}$. From preforming a carefull Duhamel expansion as described in the start of this section we obtain the formal exspansion
\begin{equation}\label{EQ:Duhamel_proof_eq_1}
	\begin{aligned}
	U_{\hbar,\lambda}(-t_0) = U_{\hbar,0}(-t_0) +  \sum_{i=0}^2 \sum_{j=\min(1,i)}^i   \sum_{k=k_{ij}}^{k_0} \sum_{\iota\in\mathcal{Q}_{k,i,j}}   \sum_{\boldsymbol{x}\in\mathcal{X}_{\neq}^{k-i}} \mathcal{T}(\iota,k,\boldsymbol{x},t_0;\hbar)+ \mathcal{D}_1^1(t_0) + \mathcal{D}_2^1(t_0),
	\end{aligned}
\end{equation}
where the different terms $\mathcal{T}(\iota,k,\boldsymbol{x},t_0;\hbar)$, $\mathcal{D}_1^1(t_0)$ and $\mathcal{D}_2^1(t_0)$ are given by
\begin{equation}\label{EQ:Duhamel_proof_eq_4.5}
\mathcal{T}(\iota,k,\boldsymbol{x},t_0;\hbar) =  \sum_{\alpha\in\N^k}   \left( \frac{i\lambda}{\hbar}\right)^{|\alpha|} 
	 \int_{[0,t_0]_{\leq}^{|\alpha|}}   
	   \prod_{m=1}^k \prod_{i=\beta_{m-1} +1}^{\beta_{m}}  V^{t_i}_{\hbar,x_{\iota(m)}}  \, d\boldsymbol{t}_{|\alpha|,1} U_{\hbar,0}(-t_0).
\end{equation}
\begin{equation}\label{EQ:Duhamel_proof_eq_1.2}
	\begin{aligned}
	    \mathcal{D}_1^1(t_0) ={}&  \sum_{i=0}^2 \sum_{j=\min(1,i)}^i    \sum_{\iota\in\mathcal{Q}_{k_0+1,i,j}}   \sum_{\boldsymbol{x}\in\mathcal{X}_{\neq}^{k_0-i+1}}  
	   \\
	   &\times \sum_{\alpha\in\N^{k_0}\times\{1\}}   \left( \frac{i\lambda}{\hbar}\right)^{|\alpha|} 
	 \int_{[0,t_0]_{\leq}^{|\alpha|}}   U_{\hbar,\lambda}(-t_{|\alpha|}) U_{\hbar,0}(t_{|\alpha|})  
	   \prod_{m=1}^{k_0} \prod_{i=\beta_{m-1} +1}^{\beta_{m}}  V^{t_i}_{\hbar,x_{\iota(m)}}  \, d\boldsymbol{t}_{|\alpha|,1}U_{\hbar,0}(-t_0).
	   \end{aligned}
\end{equation}
and
\begin{equation}\label{EQ:Duhamel_proof_eq_1.3}
	\begin{aligned}
	   \mathcal{D}_2^1(t_0) ={}&  \sum_{j=1}^2 \sum_{k=6-j}^{k_0} \sum_{\iota\in\mathcal{Q}_{k,2,j}}   \sum_{\boldsymbol{x}\in \mathcal{X}_{\neq}^{k-2}} \sum_{i=1,i\neq \iota(k)}^{k-2} 
	  \\
	  &\times
	  \sum_{\alpha\in\N^k}   \left( \frac{i\lambda}{\hbar}\right)^{|\alpha|+1} 
	 \int_{[0,t_0]_{\leq}^{|\alpha|}}     U_{\hbar,\lambda}(-t_{|\alpha|}) U_{\hbar,0}(t_{|\alpha|})  V^{t_{|\alpha|}}_{\hbar,x_i} 
	   \prod_{m=1}^k \prod_{i=\beta_{m-1} +1}^{\beta_{m}}  V^{t_i}_{\hbar,x_{\iota(m)}}  \, d\boldsymbol{t}_{|\alpha|,1}U_{\hbar,0}(-t_0).
	\end{aligned}
\end{equation}
Where we in the expression for the terms $\mathcal{T}(\iota,k,\boldsymbol{x},t_0;\hbar)$, $\mathcal{D}_1^1(t_0)$ and $\mathcal{D}_2^1(t_0)$ have used the notation introduced in this section, $k$ is the number of different scatterers we see, the entries of $\alpha$ record the number of internal collisions for each scatterer, and finally the sum over $\boldsymbol{x}$ is the sum over all possible configurations with different potentials. To compress the above expression we have introduced the notation $\beta \in \N_0^{k+1}$ to be given by $\beta_0=0$ and $\beta_i = \sum_{j=1}^{i} \alpha_i$. 

It is here where we should have introduced an additional error that truncated the sum over $\alpha\in\N^k$.  We have chosen not to do this to simplify notation and presentation. Doing this more carefull expansion do not give rise to any technical difficulties. 

For the two terms $\mathcal{D}_1^1(t_0)$ and $\mathcal{D}_2^1(t_0)$ we make the following change of variables $t_{\beta_m} \mapsto s_{m,\beta_m} $ and $ t_i-t_{i+1}\mapsto s_{m,i}$  for $i\in\{\beta_{m-1}+1,\dots,\beta_m-1\}$ $m\in\{1,\dots,k\}$. With this change of variables we have for the term $ \mathcal{D}_1^1(t_0)$ 
\begin{equation}\label{EQ:Duhamel_proof_eq_2}
	\begin{aligned}
	 \mathcal{D}_1^1(t_0)  ={}&  \sum_{\alpha\in\N^{k_0}\times\{1\}}   \left( \frac{i\lambda}{\hbar}\right)^{|\alpha|}  \int_{[0,t_0]^{|\alpha|}}  \boldsymbol{1}_{[0,t]^k_{\leq}}(s_{k,\beta_k},\dots,s_{1,\beta_1}) U_{\hbar,\lambda}(-s_{k,\beta_k}) U_{\hbar,0}(s_{k,\beta_k}) 
	\\
	&\phantom{\sum_{\alpha\in\N^{k_0}\times\{1\}}  }{} 
	\times \Big[\prod_{m=1}^k \boldsymbol{1}_{[0,s_{m-1,\beta_{m-1}}-s_{m,\beta_{m}}]} \big(\sum_{i=\beta_{m-1}+1}^{\beta_m-1}s_{m,i}\big) \Big\{ \prod_{i=\beta_{m-1} +1}^{\beta_{m}} U_{\hbar,0}(-s_{m,\beta_i})) V_{\hbar,x_m}\Big\}   
	  \\
	  &
	  \phantom{\sum_{\alpha\in\N^{k_0}\times\{1\}}  }{} \times \Big\{\prod_{i=\beta_{m-1} +1}^{\beta_{m}}U_{\hbar,0}(s_{m,i})) \Big\}  \Big]   \, \prod_{m=1}^kd\boldsymbol{s}_{m,(\beta_{m-1}+1,\beta_m)}U_{\hbar,0}(-t_0)
	  \\
	  ={}& \mathcal{E}^{k_0}_{0,0}(\boldsymbol{x},t_0;\hbar), 
	\end{aligned}
\end{equation}
where we have used Definition~\ref{def_reg_operator_2} and Definition~\ref{functions_for_exp_def} to recognise different terms. For $\mathcal{D}_2^1(t_0)$  we have that 
\begin{equation}\label{EQ:Duhamel_proof_eq_3}
	\begin{aligned}
	   \mathcal{D}_2^1(t_0) ={}& \sum_{i=1,i\neq \iota(k)}^{k-2}  \sum_{\alpha\in\N^k}   \left( \frac{i\lambda}{\hbar}\right)^{|\alpha|+1}  \int_{[0,t_0]^{|\alpha|}}  \boldsymbol{1}_{[0,t]^k_{\leq}}(s_{k,\beta_k},\dots,s_{1,\beta_1}) U_{\hbar,\lambda}(-s_{k,\beta_k}) U_{\hbar,0}(s_{k,\beta_k})  V^{t_{|\alpha|}}_{\hbar,x_i} 
	\\
	&\phantom{\sum_{\alpha\in\N^{k_0}\times\{1\}}  }{} 
	\times \Big[\prod_{m=1}^k \boldsymbol{1}_{[0,s_{m-1,\beta_{m-1}}-s_{m,\beta_{m}}]} \big(\sum_{i=\beta_{m-1}+1}^{\beta_m-1}s_{m,i}\big) \Big\{ \prod_{i=\beta_{m-1} +1}^{\beta_{m}} U_{\hbar,0}(-s_{m,\beta_i})) V_{\hbar,x_m}\Big\}   
	  \\
	  &
	  \phantom{\sum_{\alpha\in\N^{k_0}\times\{1\}}  }{} \times \Big\{\prod_{i=\beta_{m-1} +1}^{\beta_{m}}U_{\hbar,0}(s_{m,i})) \Big\}  \Big]   \, \prod_{m=1}^kd\boldsymbol{s}_{m,(\beta_{m-1}+1,\beta_m)}U_{\hbar,0}(-t_0)
	  \\
	  ={}&\mathcal{E}_{2,j}^{\mathrm{rec}}(k,0,\boldsymbol{x},t_0;\hbar), 
	\end{aligned}
\end{equation}
where we have used Definition~\ref{def_reg_operator_2} and Definition~\ref{def_recol_error} to recognise different terms. Combining \eqref{EQ:Duhamel_proof_eq_0}--\eqref{EQ:Duhamel_proof_eq_3} we obtain that
\begin{equation*}
	\begin{aligned}
	U_{\hbar,\lambda}(-t) = {}& U_{\hbar,\lambda}(-\tfrac{t(\tau_0-1)}{\tau_0})\Big[ U_{\hbar,0}(-t_0)+  \sum_{i=0}^2 \sum_{j=\min(1,i)}^i   \sum_{k=k_{ij}}^{k_0} \sum_{\iota\in\mathcal{Q}_{k,i,j}}   \sum_{\boldsymbol{x}\in\mathcal{X}_{\neq}^{k-i}} \mathcal{T}(\iota,k,\boldsymbol{x},t_0;\hbar) \Big]
	\\
	&+  U_{\hbar,\lambda}(-\tfrac{t(\tau_0-1)}{\tau_0})\big[  \mathcal{R}_1^{\mathrm{rec}}(k_0;\hbar)+ \sum_{i=0}^2 \sum_{j=\min(1,i)}^i \mathcal{R}_{1,i,j}^{k_0}(N;\hbar)  \big] ,
	\end{aligned}
\end{equation*}
where we have used Definition~\ref{def_remainder_k_0} and Definition~\ref{def_recol_reminder} to recognise terms.
The next step is to iterate the argument we have just done. Firstly we write $U_{\hbar,\lambda}(-\tfrac{t(\tau_0-1)}{\tau_0})$ as $U_{\hbar,\lambda}(-\tfrac{t(\tau_0-2)}{\tau_0})U_{\hbar,\lambda}(-\tfrac{t}{\tau_0})$ and then we do the same Duhamel expansion of $U_{\hbar,\lambda}(-\tfrac{t}{\tau_0})$ as we have already done. However this time we also need to keep the collision series to the right in mind. Especially we have to be carefully if the point we ended the first expansion with is the one we start with again. Making this Duhamel expansion we obtain that
\begin{equation}\label{EQ:Duhamel_proof_eq_5}
	\begin{aligned}
	\MoveEqLeft U_{\hbar,\lambda}(-t_0)\Big[ U_{\hbar,0}(-t_0)+  \sum_{i=0}^2 \sum_{j=\min(1,i)}^i   \sum_{k=k_{ij}}^{k_0} \sum_{\iota\in\mathcal{Q}_{k,i,j}}   \sum_{\boldsymbol{x}\in\mathcal{X}_{\neq}^{k-i}} \mathcal{T}(\iota,k,\boldsymbol{x},t_0;\hbar) \Big]	
	\\
	&=   U_{\hbar,0}(-2t_0)+ \mathcal{D}_1^2(t_0) + \mathcal{D}_2^2(t_0) +\mathcal{D}_3^2(t_0) + \mathcal{D}_4^2(t_0),
	\end{aligned}
\end{equation}
where the terms $\mathcal{D}_1^2(t_0)$, $\mathcal{D}_2^2(t_0)$, $\mathcal{D}_3^2(t_0)$ and $\mathcal{D}_4^2(t_0)$ are given by
\begin{equation*}
	\begin{aligned}
	\mathcal{D}_1^2(t_0) ={} &\sum_{i=0}^2 \sum_{j=\min(1,i)}^i  \sum_{k=k_{ij}}^{k_0}  \sum_{\iota\in\mathcal{Q}_{k,i,j}}   \sum_{(\boldsymbol{x}^1,\boldsymbol{x}^2)\in\mathcal{X}_{\neq}^{k-i}}   \sum_{k_1+k_2=k}\sum_{\alpha\in\N^{k_1}} \sum_{\tilde{\alpha}\in\N_0} \left( \frac{i\lambda}{\hbar}\right)^{|\alpha| +\tilde{\alpha}} 
	\\
	&\times \int_{[0,t_0]_{\leq}^{|\alpha| + \tilde{\alpha}}}   
	   \prod_{m=1}^{k_1} \prod_{i=\beta_{m-1} +1}^{\beta_{m}}  V^{t_{i+\tilde{\alpha}}}_{\hbar,x_{\iota(m)}}   \prod_{i=1}^{\tilde{\alpha}}  V^{t_{i}}_{\hbar,x^2_{\iota(k_2)}}  \, d\boldsymbol{t}_{|\alpha|+\tilde{\alpha},1}U_{\hbar,0}(-t_0)  \mathcal{T}(\iota,k_2,\boldsymbol{x}^2,t_0;\hbar),
	   \end{aligned}
\end{equation*}
\begin{equation*}
	\begin{aligned}
	  \mathcal{D}_2^2(t_0)={} &  \sum_{i=0}^2 \sum_{j=\min(1,i)}^i  \sum_{k_1=1}^{k_0} \sum_{k_2=k_0-k_1+1}^{k_0} \sum_{\iota\in\mathcal{Q}_{k_1+k_2,i,j}} 
	  \sum_{(\boldsymbol{x}_1,\boldsymbol{x}_2)\in \mathcal{X}_{\neq}^{k_1+k_2-i}} \sum_{\alpha\in\N^{k_1}} \sum_{\tilde{\alpha}\in\N_0} \left( \frac{i\lambda}{\hbar}\right)^{|\alpha| +\tilde{\alpha}} 
	  \\
	&\times \int_{[0,t_0]_{\leq}^{|\alpha| + \tilde{\alpha}}}   
	   \prod_{m=1}^{k_1} \prod_{i=\beta_{m-1} +1}^{\beta_{m}}  V^{t_{i+\tilde{\alpha}}}_{\hbar,x_{\iota(m)}}   \prod_{i=1}^{\tilde{\alpha}}  V^{t_{i}}_{\hbar,x^2_{\iota(k_2)}}  \, d\boldsymbol{t}_{|\alpha|+\tilde{\alpha},1}U_{\hbar,0}(-t_0)  \mathcal{T}(\iota,k_2,\boldsymbol{x}^2,t_0;\hbar),
	 \end{aligned}
\end{equation*}
\begin{equation*}
	\begin{aligned}
	\mathcal{D}_3^2(t_0)={}&  \sum_{i=0}^2 \sum_{j=\min(1,i)}^i   \sum_{k_2=0}^{k_0}   \sum_{\iota\in\mathcal{Q}_{k_0+1+k_2,i,j}}   \sum_{(\boldsymbol{x}^1,\boldsymbol{x}^2)\in\mathcal{X}_{\neq}^{k_0+k_2-i+1}}  \sum_{\alpha\in\N^{k_0}\times\{1\}}  \sum_{\tilde{\alpha}\in\N_0} \left( \frac{i\lambda}{\hbar}\right)^{|\alpha| +\tilde{\alpha}}  \int_{[0,t_0]_{\leq}^{|\alpha|+\tilde{\alpha}}}   U_{\hbar,\lambda}(-t_{|\alpha|})
	\\
	&\times  U_{\hbar,0}(t_{|\alpha|})  
	   \prod_{m=1}^{k_0}  \prod_{i=\beta_{m-1} +1}^{\beta_{m}}  V^{t_{i+\tilde{\alpha}}}_{\hbar,x_{\iota(m)}}   \prod_{i=1}^{\tilde{\alpha}}  V^{t_{i}}_{\hbar,x^2_{\iota(k_2)}}  \, d\boldsymbol{t}_{|\alpha|+\tilde{\alpha},1}U_{\hbar,0}(-t_0)  \mathcal{T}(\iota,k_2,\boldsymbol{x}^2,t_0;\hbar),
	\end{aligned}
\end{equation*}
and
\begin{equation*}
	\begin{aligned}
	  &\mathcal{D}_4^2(t_0)=  \sum_{j=1}^2 \sum_{k_1=1}^{k_0}  \sum_{k_2=0}^{k_0} \boldsymbol{1}_{\{k_1+k_2\geq 6-j\}}  \sum_{\iota\in\mathcal{Q}_{k_1+k_2,2,j}}   \sum_{(\boldsymbol{x}^1,\boldsymbol{x}^2)\in \mathcal{X}_{\neq}^{k_1+k_2-2}} \sum_{i=1,i\neq \iota(k)}^{k_1+k_2-2}   \sum_{\alpha\in\N^{k_1}}  \sum_{\tilde{\alpha}\in\N_0} \left( \frac{i\lambda}{\hbar}\right)^{|\alpha| +\tilde{\alpha}+1} 
	  \\
	  &\times
	 \int_{[0,t_0]_{\leq}^{|\alpha|+\tilde{\alpha}}}     U_{\hbar,\lambda}(-t_{|\alpha|}) U_{\hbar,0}(t_{|\alpha|})  V^{t_{|\alpha|}}_{\hbar,x_i} 
	     \prod_{i=\beta_{m-1} +1}^{\beta_{m}}  V^{t_{i+\tilde{\alpha}}}_{\hbar,x_{\iota(m)}}   \prod_{i=1}^{\tilde{\alpha}}  V^{t_{i}}_{\hbar,x^2_{\iota(k_2)}}  \, d\boldsymbol{t}_{|\alpha|+\tilde{\alpha},1}U_{\hbar,0}(-t_0) \mathcal{T}(\iota,k_2,\boldsymbol{x}^2,t_0;\hbar),
	\end{aligned}
\end{equation*}
where we have sorted the fully expanded terms in those such that $k_1+k_2\leq k_0 $ and $k_1+k_2> k_0$.  For the term $\mathcal{D}_1^2(t_0)$ we make the change of variables $t_i\mapsto t_i+t_0$ for all $i$ in the integrals in the function $\mathcal{T}$. This gives us
\begin{equation}\label{EQ:Duhamel_proof_eq_6}
	\begin{aligned}
	\mathcal{D}_1^2(t_0)	 ={}& \sum_{k_1+k_2=k}\sum_{\alpha\in\N^{k_1}} \sum_{\tilde{\alpha}\in\N_0} \left( \frac{i\lambda}{\hbar}\right)^{|\alpha| +\tilde{\alpha}} 
	 \int_{[0,t_0]_{\leq}^{|\alpha| + \tilde{\alpha}}} \prod_{m=1}^{k_1} \prod_{i=\beta_{m-1} +1}^{\beta_{m}}  V^{t_{i+\tilde{\alpha}}}_{\hbar,x_{\iota(m)}}   
	 \prod_{i=1}^{\tilde{\alpha}}  V^{t_{i}}_{\hbar,x^2_{\iota(k_2)}}\, d\boldsymbol{t}_{|\alpha|+\tilde{\alpha},1}
	 \\
	 &\times \sum_{\alpha\in\N^{k_2}}   \left( \frac{i\lambda}{\hbar}\right)^{|\alpha|} 
	 \int_{[t_0,2t_0]_{\leq}^{|\alpha|}}   
	   \prod_{m=1}^{k_2} \prod_{i=\beta_{m-1} +1}^{\beta_{m}}  V^{t_i}_{\hbar,x^2_{\iota(m)}}  \, d\boldsymbol{t}_{|\alpha|,1}U_{\hbar,0}(-2t_0) 
	   \\
	   ={}&  \sum_{\alpha\in\N^{k}}  \left( \frac{i\lambda}{\hbar}\right)^{|\alpha|}  \int_{[0,2t_0]_{\leq}^{|\alpha|}}   
	   \prod_{m=1}^{k_2} \prod_{i=\beta_{m-1} +1}^{\beta_{m}}  V^{t_i}_{\hbar,x^2_{\iota(m)}}  \, d\boldsymbol{t}_{|\alpha|,1}U_{\hbar,0}(-2t_0) =\mathcal{T}(\iota,k,(\boldsymbol{x}^1,\boldsymbol{x}^2),2t_0;\hbar),
	   \end{aligned}
\end{equation}
where we have sorted the terms in a particular way before writing the product of the integrals as one integral and again we have that  $\mathcal{T}(\iota,k,(\boldsymbol{x}^1,\boldsymbol{x}^2),2t_0;\hbar)$ is given by \eqref{EQ:Duhamel_proof_eq_4.5}. 

For the last three terms $\mathcal{D}_2^2(t_0)$, $\mathcal{D}_3^2(t_0)$ and $\mathcal{D}_4^2(t_0)$ we make the following type of change of variables $t_{\beta_m} \mapsto s_{m,\beta_m} $ and $ t_i-t_{i+1}\mapsto s_{m,i}$  for $i\in\{\beta_{m-1}+1,\dots,\beta_m-1\}$ $m\in\{1,\dots,k\}$, where the indices will be diffrent if $\tilde{\alpha}\neq0$. We do this for both the integrals in $\mathcal{T}$ and the ones obtained from the recent Duhamel expansion. Then by recognising different terms using the definitions made above, \eqref{EQ:Duhamel_proof_eq_5} and \eqref{EQ:Duhamel_proof_eq_6} we obtain that
\begin{equation*}
	\begin{aligned}
	U_{\hbar,\lambda}(-t) = {}& U_{\hbar,\lambda}(-\tfrac{t(\tau_0-2)}{\tau_0})\Big[ U_{\hbar,0}(-2t_0)+  \sum_{i=0}^2 \sum_{j=\min(1,i)}^i   \sum_{k=k_{ij}}^{k_0} \sum_{\iota\in\mathcal{Q}_{k,i,j}}   \sum_{\boldsymbol{x}\in\mathcal{X}_{\neq}^{k-i}} \mathcal{T}(\iota,k,\boldsymbol{x},2t_0;\hbar)U_{\hbar,0}(-2t_0) \Big]
	\\
	&+  \sum_{\tau=1}^{2} U_{\hbar,\lambda}(-\tfrac{t(\tau_0-\tau)}{\tau_0})\big[  \mathcal{R}_\tau^{\mathrm{rec}}(k_0;\hbar)+ \sum_{i=0}^2 \sum_{j=\min(1,i)}^i \mathcal{R}_{\tau,i,j}^{k_0}(N;\hbar)  \big]. 
	\end{aligned}
\end{equation*}
By iterating the above arguments we obtain the following equality
\begin{equation*}
	\begin{aligned}
	U_{\hbar,\lambda}(-t) = {}&  U_{\hbar,0}(-t)+  \sum_{i=0}^2 \sum_{j=\min(1,i)}^i   \sum_{k=k_{ij}}^{k_0} \sum_{\iota\in\mathcal{Q}_{k,i,j}}   \sum_{\boldsymbol{x}\in\mathcal{X}_{\neq}^{k-i}} \mathcal{T}(\iota,k,\boldsymbol{x},t;\hbar)
	\\
	&+  \sum_{\tau=1}^{\tau_0} U_{\hbar,\lambda}(-\tfrac{t(\tau_0-\tau)}{\tau_0})\big[  \mathcal{R}_\tau^{\mathrm{rec}}(k_0;\hbar)+ \sum_{i=0}^2 \sum_{j=\min(1,i)}^i \mathcal{R}_{\tau,i,j}^{k_0}(N;\hbar)  \big]. 
	\end{aligned}
\end{equation*}
To obtain the stated form we again make the change of variables $t_{\beta_m} \mapsto s_{m,\beta_m} $ and $ t_i-t_{i+1}\mapsto s_{m,i}$  for $i\in\{\beta_{m-1}+1,\dots,\beta_m-1\}$ $m\in\{1,\dots,k\}$ for the integrals in $\mathcal{T}(\iota,k,\boldsymbol{x},t;\hbar)U_{\hbar,0}(-t)$. This concludes the proof.
\end{proof}
\section{Terms without recollisions}\label{Sec:tech_est_0_recol}
We will in this section give the first of a number of technical estimates which gives that the Duhamel expansion do converge in appropriate sense. We will here focus on the terms where we have not observed a recollision.
\subsection{Estimates for fully expanded terms}
 \begin{lemma}\label{expansion_aver_bound_Mainterm}
Assume we are in the setting of Definition~\ref{functions_for_exp_def} and let $\varphi \in \mathcal{H}^{2d+2}_\hbar(\R^d)$ then
\begin{equation*}
	\begin{aligned}
	\MoveEqLeft \mathbb{E}\Big[\big\lVert \sum_{\boldsymbol{x}\in\mathcal{X}_{\neq}^{k}}\mathcal{I}_{0,0}(k,\boldsymbol{x},t;\hbar) \varphi\big\rVert_{L^2(\R^d)}^2\Big] \leq   \frac{ C^k}{k!} \norm{\varphi}_{L^2(\R^d)}^2  
	+ \hbar k C^k |\log(\tfrac{\hbar}{t})|^{k+3}  \norm{\varphi}_{\mathcal{H}^{2d+2}_\hbar(\R^d)}^2,
	\end{aligned}
\end{equation*}
where the constant $C$ depends on $\rho$, $t$, the single site potential $V$ and the coupling constant $\lambda$. In particular, $ \sum_{\boldsymbol{x}\in\mathcal{X}_{\neq}^{k}}\mathcal{I}_{0,0}(k,\boldsymbol{x},t;\hbar) \varphi\in L^2(\R^d)$ $\Pro$-almost surely. 
\end{lemma}
\begin{observation}\label{obs_form_I_op_kernel}
The function $\mathcal{I}_{0,0}(k,\boldsymbol{x},t;\hbar) \varphi(x)$ can be expressed in two different ways firstly by Definition~\ref{def_reg_operator_2} and Definition~\ref{functions_for_exp_def}  we have an explicit formula for the kernel of the operator $\mathcal{I}_{0,0}(k,\boldsymbol{x},t;\hbar)$. Hence we get that
\begin{equation*}
	\begin{aligned}
 	\MoveEqLeft \mathcal{I}_{0,0}(k,\boldsymbol{x},t;\hbar) \varphi(x) 
	=  \frac{1}{(2\pi\hbar)^d\hbar^{k}}  \sum_{\alpha\in \N^k} (i\lambda)^{|\alpha|}   \int_{[0,t]_{\leq}^k} \int_{\R^{(k+1)d}}
	 e^{ i   \langle  \hbar^{-1}x,p_{k} \rangle} \Big\{\prod_{m=1}^k  e^{ -i   \langle  \hbar^{-1/d}x_m,p_{m}-p_{m-1} \rangle}  
	\\
	&\times \Psi_{\alpha_m}(p_{m},p_{m-1},\hbar^{-1}(s_{m-1}-s_{m}); V)  e^{i  s_{m} \frac{1}{2}\hbar^{-1} (p_{m}^2-p_{m-1}^2)}  \, \Big\}  e^{i  t \frac{1}{2}\hbar^{-1} p_0^2} \hat{\varphi}(\tfrac{p_0}{\hbar}) \,d\boldsymbol{p}_{0,k} d\boldsymbol{s}_{k,1}.
	\end{aligned}
\end{equation*}
Secondly we can consider the original expression for the operator $ \mathcal{I}_{0,0}(k,\boldsymbol{x},t;\hbar)$ obtained from the Duhamel expansion before any change of variables in time. Instead of preforming the change of variables in time as in the proof of the Lemma we can preform the change of variables $t_n\mapsto t_n$ and $t_i\mapsto \hbar^{-1}(t_i-t_{i+1})$ for all $i\in\{1,\dots,n-1\}$. From this change of variables and a relabelling of the the variables we obtain the form
\begin{equation*}
	\begin{aligned}
 	\MoveEqLeft \mathcal{I}_{0,0}(k,\boldsymbol{x},t;\hbar) \varphi(x) 
	\\
	={}&  \frac{1}{(2\pi\hbar)^d\hbar^k}  \sum_{\alpha\in \N^k} (i\lambda)^{|\alpha|}   \int_{\R_{+}^j} \int_{\R^{(j+1)d}} \boldsymbol{1}_{[0,t]}( \boldsymbol{s}_{1,k}^{+}+ \hbar\boldsymbol{t}_{1,a_k}^{+})
	 e^{ i   \langle  \hbar^{-1}x,p_{k} \rangle}  \prod_{i=1}^{a_k} e^{i  t_{i} \frac{1}{2} \eta_{i}^2} \Big\{\prod_{m=1}^k  e^{ -i   \langle  \hbar^{-1/d}x_m,p_{m}-p_{m-1} \rangle}  
	\\
	&\times e^{i  s_{m}\hbar^{-1} \frac{1}{2} p_{m}^2}   \hat{\mathcal{V}}_{\alpha_m}(p_m,p_{m-1},\boldsymbol{\eta} )  \Big\}    e^{i  \hbar^{-1}(t- \boldsymbol{s}_{1,k}^{+}- \hbar\boldsymbol{t}_{1,a_k}^{+}) \frac{1}{2} p_0^2} \hat{\varphi}(\tfrac{p_0}{\hbar}) \, d\boldsymbol{\eta}_{1,a_k}d\boldsymbol{p}_{0,k}   d\boldsymbol{t}_{1,a_k} d\boldsymbol{s}_{1,k},
	\end{aligned}
\end{equation*}
where the functions $\hat{\mathcal{V}}_{\alpha_m}$ is defined in \eqref{obs_con_form_psi_f} and the numbers $a_m$ is given by
\begin{equation*}
a_m = \sum_{i=1}^m \alpha_i -m.
\end{equation*}
 The depends on the variables $\boldsymbol{\eta}$ of functions $\hat{\mathcal{V}}_{\alpha_m}$ is only $(\eta_{a_{m-1}+1},\dots,\eta_{a_m})$. This other way of expressing the function $\mathcal{I}_{0,0}(k,\boldsymbol{x},t;\hbar) \varphi(x)$ can also be used when we act with the other operators defined in the previous section. Observe that by preforming the right change of variables we can easily switch between the to ways of writing the operator.
\end{observation}
\begin{proof}
When we  take the $L^2$-norm we will get two sums over the set $\mathcal{X}_{\neq}^k$. In order to control this double sum we will divide into cases depending on how many points $\boldsymbol{x}$ and $\boldsymbol{\tilde{x}}$ have in common and where they are placed. This yields
\begin{equation}\label{EQ:exp_aver_b_Mainterm_1}
	\big\lVert \sum_{\boldsymbol{x}\in\mathcal{X}_{\neq}^{k}}\mathcal{I}_{0,0}(k,\boldsymbol{x},t;\hbar) \varphi\big\rVert_{L^2(\R^d)}^2
	=\sum_{n=0}^k \sum_{\sigma^1,\sigma^2\in\mathcal{A}(k,n)} \sum_{\kappa\in\mathcal{S}_n}\sum_{\alpha,\tilde{\alpha}\in\N^k} (i\lambda)^{|\alpha|}(-i\lambda)^{|\tilde{\alpha}|}\mathcal{T}(n,\alpha,\tilde{\alpha},\sigma^1,\sigma^2,\kappa),
\end{equation}
where the numbers $\mathcal{T}(n,\alpha,\tilde{\alpha},\sigma^1,\sigma^2,\kappa)$ are given by
\begin{equation}\label{proof-exp_aver_bound_main2}
	\mathcal{T}(n,\alpha,\tilde{\alpha},\sigma^1,\sigma^2,\kappa)
	= \sum_{(\boldsymbol{x},\boldsymbol{\tilde{x}})\in \mathcal{X}_{\neq}^{2k}}  \prod_{i=1}^n \frac{\delta(x_{\sigma_i^1}- \tilde{x}_{\sigma_{\kappa(i)^2}})}{\rho}  \int_{\R^d}  \mathcal{I}(k,\boldsymbol{x},\alpha,t;\hbar) \varphi(x)\overline{ \mathcal{I}(k,\boldsymbol{\tilde{x}},\tilde{\alpha},t;\hbar) \varphi(x)}  \,dx,
\end{equation}
where $ \mathcal{I}(k,\boldsymbol{x},\alpha,t;\hbar) \varphi$ is $\mathcal{I}_{0,0}(k,\boldsymbol{x},t;\hbar) \varphi$ for a fixed $\alpha$.
We will divide into the two different cases $\kappa=\mathrm{id}$ and $\kappa\neq\mathrm{id}$. For both cases we will use the second expression for the kernel given in Observation~\ref{obs_form_I_op_kernel}.  We start with the first case. Here we have that
\begin{equation*}
	\begin{aligned}
	\MoveEqLeft \mathcal{T}(n,\alpha,\tilde{\alpha},\sigma^1,\sigma^2,\mathrm{id})
	= \sum_{(\boldsymbol{x},\boldsymbol{\tilde{x}})\in \mathcal{X}_{\neq}^{2k}}  \frac{1}{(2\pi\hbar)^{d}\hbar^{2k}}\prod_{i=1}^n \rho^{-1} \delta(x_{\sigma_i^1}- \tilde{x}_{\sigma_{i}^2})   \int_{\R_{+}^{2j}} \int  \delta(p_k-q_k)   
	\\
	&\times \boldsymbol{1}_{[0,t]}( \boldsymbol{s}_{1,k}^{+}+\hbar \boldsymbol{t}_{1,a_k}^{+})  \boldsymbol{1}_{[0,t]}( \boldsymbol{\tilde{s}}_{1,k}^{+}+\hbar \boldsymbol{\tilde{t}}_{1,\tilde{a}_k}^{+})  \prod_{i=1}^{a_k} e^{i \frac{1}{2}t_{i}  \eta_{i}^2}  \prod_{i=1}^{\tilde{a}_k} e^{-i \frac{1}{2} \tilde{t}_{i}  \xi_{i}^2} \Big\{  \prod_{m=1}^k  e^{ -i   \langle  \hbar^{-1/d}x_m,p_{m}-p_{m-1} \rangle}  
	\\
	&\times  e^{i \hbar^{-1} s_{m} \frac{1}{2} p_{m}^2}  \hat{\mathcal{V}}_{\alpha_m}(p_m,p_{m-1},\boldsymbol{\eta} )      e^{ i   \langle  \hbar^{-1/d}\tilde{x}_m,q_{m}-q_{m-1} \rangle}  e^{-i  \hbar^{-1}\tilde{s}_{m} \frac{1}{2} q_{m}^2}  \overline{\hat{\mathcal{V}}_{\tilde{\alpha}_m}(q_m,q_{m-1},\boldsymbol{\xi} ) } \Big\} 
	\\
	&\times   e^{i \hbar^{-1}  (t- \boldsymbol{s}_{1,k}^{+}- \hbar\boldsymbol{t}_{1,a_k}^{+}) \frac{1}{2} p_0^2} e^{-i  \hbar^{-1}(t- \boldsymbol{\tilde{s}}_{1,k}^{+}- \hbar\boldsymbol{\tilde{t}}_{1,\tilde{a}_k}^{+}) \frac{1}{2} q_0^2} 
	 \hat{\varphi}(\tfrac{p_0}{\hbar})  \hat{\varphi}(\tfrac{q_0}{\hbar})   \,d\boldsymbol{\eta}d\boldsymbol{p} \boldsymbol{\xi}d\boldsymbol{q}   d\boldsymbol{t} d\boldsymbol{s}  d\boldsymbol{\tilde{t}} d\boldsymbol{\tilde{s}},
	\end{aligned}
\end{equation*}
where we have evaluated the integral in $x$. We observe that this has the form as the expressions considered in Lemma~\ref{LE:Exp_ran_phases} and applying this Lemma yields
\begin{equation*}
	\begin{aligned}
	\mathbb{E}&[ \mathcal{T}(n,\alpha,\tilde{\alpha},\sigma^1,\sigma^2,\mathrm{id})]
	= \frac{\left(\rho(2\pi)^d \right)^{2k-n}  }{(2\pi\hbar)^d \hbar^n}     \int_{\R_{+}^{|\alpha|+|\tilde{\alpha}|}}  \int  \delta(p_k-q_k) \Lambda_n(\boldsymbol{p},\boldsymbol{q},\sigma^1,\sigma^2,\mathrm{id}) \prod_{i=1}^{\tilde{a}_k} e^{-i \frac{1}{2} \tilde{t}_{i}  \xi_{i}^2}
	\\
	&\times  \boldsymbol{1}_{[0,t]}( \boldsymbol{s}_{1,k}^{+}+\hbar \boldsymbol{t}_{1,a_k}^{+})  \boldsymbol{1}_{[0,t]}( \boldsymbol{\tilde{s}}_{1,k}^{+}+\hbar \boldsymbol{\tilde{t}}_{1,\tilde{a}_k}^{+}) \prod_{i=1}^{a_k} e^{i \frac{1}{2}t_{i}  \eta_{i}^2}  \prod_{m=1}^k  \Big\{  
	 e^{i \hbar^{-1} s_{m} \frac{1}{2} p_{m}^2}  \hat{\mathcal{V}}_{\alpha_m}(p_m,p_{m-1},\boldsymbol{\eta} )   e^{-i \hbar^{-1} \tilde{s}_{m} \frac{1}{2} q_{m}^2}  
	 \\
	 &\times   \overline{\hat{\mathcal{V}}_{\tilde{\alpha}_m}(q_m,q_{m-1},\boldsymbol{\xi} ) } \Big\} 
	  e^{i  \hbar^{-1}(t- \boldsymbol{s}_{1,k}^{+}-\hbar \boldsymbol{t}_{1,a_k}^{+}) \frac{1}{2} p_0^2} e^{-i  \hbar^{-1}(t- \boldsymbol{\tilde{s}}_{1,k}^{+}-\hbar \boldsymbol{\tilde{t}}_{1,\tilde{a}_k}^{+}) \frac{1}{2} q_0^2} 
	 \hat{\varphi}(\tfrac{p_0}{\hbar})  \hat{\varphi}(\tfrac{q_0}{\hbar})   \,d\boldsymbol{\eta}d\boldsymbol{p} \boldsymbol{\xi}d\boldsymbol{q}   d\boldsymbol{t} d\boldsymbol{s}  d\boldsymbol{\tilde{t}} d\boldsymbol{\tilde{s}}.
	\end{aligned}
\end{equation*}
Recalling that the function $\Lambda_n(\boldsymbol{p},\boldsymbol{q},\sigma^1,\sigma^2,\mathrm{id}) $ is given by the expression
  \begin{equation*}
  	\begin{aligned}
  	 \Lambda_n(\boldsymbol{p},\boldsymbol{q},\sigma^1,\sigma^2,\mathrm{id}) 
	=  \prod_{i=1}^n \delta(p_{\sigma^1_{i-1}}- q_{\sigma^2_{i-1}} - p_{\sigma^1_{n}}+ q_{\sigma^2_{n}} )   
	 \prod_{i=1}^{n+1} \prod_{m=\sigma^1_{i-1}+1}^{\sigma^1_{i}-1} \delta(p_m-p_{\sigma^1_{i-1}})  \prod_{m=\sigma^2_{i-1}+1}^{\sigma^2_{i}-1} \delta(q_m-q_{\sigma^1_{i-1}}).
	\end{aligned}
  \end{equation*}
We get by evaluating all integrals in $\boldsymbol{q}$ and $\boldsymbol{p}$ except $p_{\sigma_i^1}$ for all $i$  the bound 
\begin{equation}\label{proof-exp_aver_bound_main3}
	\begin{aligned}
	\MoveEqLeft |\Aver{\mathcal{T}(n,\alpha,\tilde{\alpha},\sigma^1,\sigma^2,\mathrm{id})}|
	\leq \frac{((2\pi)^d \rho)^{2k-n} }{(2\pi\hbar)^{d} \hbar^n} \int_{\R_{+}^{|\alpha|+|\tilde{\alpha}|}}  \boldsymbol{1}_{[0,t]}( \boldsymbol{s}_{1,k}^{+}+\hbar \boldsymbol{t}_{1,a_k}^{+})  \boldsymbol{1}_{[0,t]}( \boldsymbol{\tilde{s}}_{1,k}^{+}+\hbar \boldsymbol{\tilde{t}}_{1,\tilde{a}_k}^{+})
	\\
	&\times   \Big| \int e^{i  \frac{1}{2} \langle Q(\boldsymbol{t}) \boldsymbol{\eta}, \boldsymbol{\eta}\rangle }   e^{-i  \frac{1}{2} \langle Q(\boldsymbol{\tilde{t}}) \boldsymbol{\xi}, \boldsymbol{\xi}\rangle }    \prod_{i=1}^{n}          e^{i  (\boldsymbol{s}_{\sigma^1_i,\sigma^1_{i+1}-1}^{+} - \boldsymbol{\tilde{s}}_{\sigma_i^2,\sigma_{i+1}^2-1}^{+} - \hbar\boldsymbol{t}_{a_{\sigma^1_i}+1,a_{\sigma_{i+1}^1}}^{+} +\hbar\boldsymbol{\tilde{t}}_{\tilde{a}_{\sigma^2_i}+1,\tilde{a}_{\sigma^2_{i+1}}}^{+}  ) \frac{1}{2}\hbar^{-1} p_{i}^2} 
	\\
	&\times    \mathcal{G}(\boldsymbol{p},\boldsymbol{\eta},\boldsymbol{\xi},\sigma^1,\sigma^2,\alpha,\tilde{\alpha})
	    e^{-i  (\boldsymbol{s}_{\sigma^1_1,k}^{+}- \boldsymbol{\tilde{s}}_{\sigma_1^2,k}^{+} -\hbar\boldsymbol{t}_{1,a_{\sigma_1^1}}^{+} +\hbar\boldsymbol{\tilde{t}}_{1,\tilde{a}_{\sigma^2_1}}^{+}) \frac{1}{2}\hbar^{-1} p_{0}^2}    |\hat{\varphi}(\tfrac{p_0}{\hbar}) |^2  \,d\boldsymbol{p} d\boldsymbol{\eta}d\boldsymbol{\xi} \Big|  d\boldsymbol{t}  d\boldsymbol{\tilde{t}} d\boldsymbol{\tilde{s}} d\boldsymbol{s},
	\end{aligned}
\end{equation}
where the numbers $a_m$, $\tilde{a}_m$ are the ones associated to $\alpha$ and $\tilde{\alpha}$ respectively as introduced above. The function $\mathcal{G}(\boldsymbol{p},\boldsymbol{\eta},\boldsymbol{\xi},\sigma^1,\sigma^2,\alpha,\tilde{\alpha})$ is given by
\begin{equation}\label{EQ_def_G_4.1}
	\begin{aligned}
	\MoveEqLeft \mathcal{G}(\boldsymbol{p},\boldsymbol{\eta},\boldsymbol{\xi},\sigma^1,\sigma^2,\alpha,\tilde{\alpha}) 
	\\
	={}&  \prod_{i=1}^n  \hat{\mathcal{V}}_{\alpha_{\sigma^1_i}}(p_i,p_{i-1},\boldsymbol{\eta}) \overline{\hat{\mathcal{V}}_{\tilde{\alpha}_{\sigma^1_i}}(p_i,p_{i-1},\boldsymbol{\xi})}  
	 \prod_{i=0}^n  \prod_{m=\sigma^1_i +1}^{\sigma^1_{i+1} -1} \hat{\mathcal{V}}_{\alpha_m}(p_i,p_i,\boldsymbol{\eta})   \prod_{m=\sigma^2_i +1}^{\sigma^2_{i+1} -1} \overline{\hat{\mathcal{V}}_{\tilde{\alpha}_m}(p_i,p_i,\boldsymbol{\xi})}.
	\end{aligned}
\end{equation}
The matrix $Q(\boldsymbol{t}) $ is given by
  \begin{equation*}
  	Q(\boldsymbol{t}) = \begin{pmatrix}
	t_1 I_d & 0 & 0 &0
	\\
	0 & t_2 I_d &0&0
	\\
	\vdots & 0 & \ddots &0
	\\
	0 &\cdots & 0 &t_{a_{k}} I_d
	\end{pmatrix}.
  \end{equation*}
We will now divide into different cases depending on the absolute value of $t$'s and  $\tilde{t}$'s. This is the same type of division done in the proof of Lemma~\ref{app_quadratic_integral_tech_est}. Again we divide into the cases where $|t_i|\geq1$ and  $|t_i|<1$ and similar for $\tilde{t}$. This gives us $2^{a_{k}+\tilde{a}_k}$ different cases to consider and we denote each corresponding set by $B_l$. For each case the indices for which $|t_i|\geq1$ will be denoted by $J_1$ and the remaining will be in $J_2$ the indices for which $|\tilde{t}_i|\geq1$ will be denoted by $\tilde{J}_1$ and the remaining will be in $\tilde{J}_2$. As in the proof of Lemma~\ref{app_quadratic_integral_tech_est} we will write the parts of the quadratic exponential with indices in $J_1$ and $\tilde{J}_1$ as a  inverse Fourier transform of its Fourier transform. To ensure integrability we define certain differential operators and preform integration by parts. After this division and preforming this integration by parts we obtain the estimate
   \begin{equation}\label{EQ:exp_aver_b_Mainterm_3}
   	\begin{aligned}
	\MoveEqLeft  |\Aver{\mathcal{T}(n,\alpha,\tilde{\alpha},\sigma^1,\sigma^2,\mathrm{id})}|
	 \leq \frac{((2\pi)^d\hbar \rho)^{2k-n}}{(2\pi\hbar)^{d} }  \sum_{l=1}^{2^{a_{k}+\tilde{a}_k}}   \sum_{|\epsilon_1|,\dots,|\epsilon_{j_1}|\leq d+1} \sum_{|\nu_1|,\dots,|\nu_{\tilde{j}_1}|\leq d+1}    \int_{\R_{+}^{|\alpha|+|\tilde{\alpha}|}} \boldsymbol{1}_{B_l}(\boldsymbol{t},\boldsymbol{\tilde{t}})
	\\
	&\times  \boldsymbol{1}_{[0,t]}( \boldsymbol{s}_{1,k}^{+}+\hbar \boldsymbol{t}_{1,a_k}^{+})  \boldsymbol{1}_{[0,t]}( \boldsymbol{\tilde{s}}_{1,k}^{+}+\hbar \boldsymbol{\tilde{t}}_{1,\tilde{a}_k}^{+}) 
	  \prod_{i\in J_1} \frac{C(\epsilon_i)}{(2\pi t_i))^\frac{d}{2} }
	   \prod_{i\in \tilde{J}_1} \frac{C(\nu_i)   }{(2\pi \tilde{t}_i)^\frac{d}{2}}  |\mathcal{I}_n(\boldsymbol{t},\boldsymbol{\tilde{t}}) |  d\boldsymbol{t}  d\boldsymbol{\tilde{t}} d\boldsymbol{\tilde{s}}d\boldsymbol{s},
	\end{aligned}
\end{equation}
 where we have used the notation  $C(\nu)= \binom{d+1}{(\nu,d+1-\abs{\nu})}$ and
 \begin{equation*}
	\begin{aligned}
	\MoveEqLeft \mathcal{I}_n(\boldsymbol{t},\boldsymbol{\tilde{t}}) 
	=    \int   
	  e^{-i  \frac{1}{2} \langle Q_{k_1,J_1}^{-1}(\boldsymbol{t})  \boldsymbol{x}_{J_1} -2\boldsymbol{\eta}_{J_1}, \boldsymbol{x}_{J_1}\rangle }  
	 e^{i  \frac{1}{2} \langle Q_{k_1,J_2}(\boldsymbol{t}) \boldsymbol{\eta}_{J_2}, \boldsymbol{\eta}_{J_2}\rangle }    e^{i  \frac{1}{2} \langle Q_{k_2,\tilde{J}_1}^{-1}(\boldsymbol{\tilde{t}}) \boldsymbol{y}_{\tilde{J}_1} +2 \boldsymbol{\xi}_{\tilde{J}_1}, \boldsymbol{y}_{\tilde{J}_1}\rangle } e^{i  \frac{1}{2} \langle Q_{k_2,\tilde{J}_2}(\boldsymbol{\tilde{t}}) \boldsymbol{\xi}_{\tilde{J}_2}, \boldsymbol{\xi}_{\tilde{J}_2}\rangle } 
	 \\
	 &\times  \prod_{i\in \tilde{J}_1} \frac{|y_i ^{\nu_i} |}{(1+|y_i |^2)^{d+1}}   
	  \prod_{i\in J_1} \frac{|x_i ^{\epsilon_i} |}{(1+|x_i |^2)^{d+1}}    \prod_{i=1}^{n}         e^{i  (\boldsymbol{s}_{\sigma^1_i,\sigma^1_{i+1}-1}^{+} - \boldsymbol{\tilde{s}}_{\sigma_i^2,\sigma_{i+1}^2-1}^{+} - \hbar\boldsymbol{t}_{a_{\sigma^1_i}+1,a_{\sigma_{i+1}^1}}^{+} +\hbar\boldsymbol{\tilde{t}}_{\tilde{a}_{\sigma^2_i}+1,\tilde{a}_{\sigma^2_{i+1}}}^{+}  ) \frac{1}{2}\hbar^{-1} p_{i}^2} 
	\\
	&\times  \big[\partial_{\boldsymbol{\xi}_{\tilde{J}_1}}^{\boldsymbol{\nu}} \partial_{\boldsymbol{\eta}_{J_1}}^{\boldsymbol{\epsilon}}  \mathcal{G}(\boldsymbol{p},\boldsymbol{\eta},\boldsymbol{\xi},\sigma^1,\sigma^2,\alpha,\tilde{\alpha}) \big]
	e^{-i  (\boldsymbol{s}_{\sigma^1_1,k}^{+}- \boldsymbol{\tilde{s}}_{\sigma_1^2,k}^{+} -\hbar\boldsymbol{t}_{1,a_{\sigma_1^1}}^{+} +\hbar\boldsymbol{\tilde{t}}_{1,\tilde{a}_{\sigma^2_1}}^{+}) \frac{1}{2}\hbar^{-1} p_{0}^2}   |\hat{\varphi}(\tfrac{p_0}{\hbar}) |^2  \, d\boldsymbol{x}d\boldsymbol{y}d\boldsymbol{p} d\boldsymbol{\eta} d\boldsymbol{\xi}.
	\end{aligned}
\end{equation*}
We will first estimate $|\mathcal{I}_n(\boldsymbol{t},\boldsymbol{\tilde{t}})|$. To do this we start by for all $i\in\{1,\dots,n\}$ preforming the following change of variables
 \begin{equation*}
 \begin{aligned}
 	&\eta_j \mapsto \eta_j-p_i \quad\text{for all $j\in\{a_{\sigma^1_{i}-1}+1,\dots,a_{\sigma^1_{i+1}-1} \}$} 
	\\ 
	&\xi_j \mapsto \xi_j-p_i \quad\text{for all $j\in\{\tilde{a}_{\sigma^2_{i}-1}+1,\dots,\tilde{a}_{\sigma^2_{i+1}-1} \}$}.
	\end{aligned}
 \end{equation*}
 This yields
  \begin{equation*}
	\begin{aligned}
	\MoveEqLeft  \mathcal{I}_n(\boldsymbol{t},\boldsymbol{\tilde{t}}) =   \int
	  e^{-i  \frac{1}{2} \langle Q_{k_1,J_1}^{-1}(\boldsymbol{t})  \boldsymbol{x}_{J_1} -2\boldsymbol{\eta}_{J_1}, \boldsymbol{x}_{J_1}\rangle }  
	 e^{i  \frac{1}{2} \langle Q_{k_2,\tilde{J}_1}^{-1}(\boldsymbol{\tilde{t}}) \boldsymbol{y}_{\tilde{J}_1} +2 \boldsymbol{\xi}_{\tilde{J}_1}, \boldsymbol{y}_{\tilde{J}_1}\rangle } e^{i  \frac{1}{2} \langle Q_{k_1,J_2}(\boldsymbol{t}) \boldsymbol{\eta}_{J_2}, \boldsymbol{\eta}_{J_2}\rangle }   e^{i  \frac{1}{2} \langle Q_{k_2,\tilde{J}_2}(\boldsymbol{\tilde{t}}) \boldsymbol{\xi}_{\tilde{J}_2}, \boldsymbol{\xi}_{\tilde{J}_2}\rangle }  
	 \\
	 &\times \prod_{i\in J_1} \frac{|x_i ^{\epsilon_i} |}{(1+|x_i |^2)^{d+1}}
	 \prod_{i\in \tilde{J}_1} \frac{|y_i ^{\nu_i} |}{(1+|y_i |^2)^{d+1}}     \prod_{i=1}^{n}          e^{i  (\boldsymbol{s}_{\sigma^1_i,\sigma^1_{i+1}-1}^{+} - \boldsymbol{\tilde{s}}_{\sigma_i^2,\sigma_{i+1}^2-1}^{+} +\hbar l^1_i(\boldsymbol{t}, \boldsymbol{\tilde{t}})  ) \frac{1}{2}\hbar^{-1} p_{i}^2} e^{i\langle p_i, l^2_i(\boldsymbol{t}, \boldsymbol{\tilde{t}},\boldsymbol{x},\boldsymbol{y},\boldsymbol{\eta},\boldsymbol{\xi}) \rangle} 
	\\
	&\times \big[\partial_{\boldsymbol{\xi}_{\tilde{J}_1}}^{\boldsymbol{\nu}} \partial_{\boldsymbol{\eta}_{J_1}}^{\boldsymbol{\epsilon}}  \mathcal{G}(\boldsymbol{p},\boldsymbol{\eta},\boldsymbol{\xi},\sigma^1,\sigma^2,\alpha,\tilde{\alpha}) \big]
	  e^{-i  (\boldsymbol{s}_{\sigma^1_1,k}^{+}- \boldsymbol{\tilde{s}}_{\sigma_1^2,k}^{+}  +\hbar l^1_0(\boldsymbol{t}, \boldsymbol{\tilde{t}}) ) \frac{1}{2}\hbar^{-1} p_{0}^2}   |\hat{\varphi}(\tfrac{p_0}{\hbar}) |^2  \, d\boldsymbol{x}d\boldsymbol{y}d\boldsymbol{p} d\boldsymbol{\eta}d\boldsymbol{\xi},
	\end{aligned}
\end{equation*}
 where $l_{i}^j$ is a linear function for all $i$ and $j$. By applying Lemma~\ref{app_quadratic_integral_tech_est} in the variables $p_1,\dots,p_n$ we get the estimate
  \begin{equation}\label{EQ:exp_aver_b_Mainterm_4}
	\begin{aligned}
	& |\mathcal{I}_n(\boldsymbol{t},\boldsymbol{\tilde{t}})| 
	 \leq  \tilde{C}_d^{n}    \int\prod_{i\in \tilde{J}_1} \frac{|y_i ^{\nu_i} |}{(1+|y_i |^2)^{d+1}}   \, d\boldsymbol{y}
	\prod_{i=1}^{n} \frac{1}{\max(1,\hbar^{-1} | \boldsymbol{s}_{\sigma^1_i,\sigma^1_{i+1}-1}^{+} - \boldsymbol{\tilde{s}}_{\sigma_i^2,\sigma_{i+1}^2-1}^{+} +\hbar l^1_i(\boldsymbol{t}, \boldsymbol{\tilde{t}})) |)^\frac{d}{2}} 
	  \\
	&\times  \int  \prod_{i\in J_1} \frac{|x_i ^{\epsilon_i} |}{(1+|x_i |^2)^{d+1}}\,d\boldsymbol{x} \sup_{|\gamma_1|,\dots,|\gamma_n|\leq d+1}  \int \big| \partial_{p_1}^{\gamma_1}\cdots \partial_{p_n}^{\gamma_n} \partial_{\boldsymbol{\xi}_{\tilde{J}_1}}^{\boldsymbol{\nu}} \partial_{\boldsymbol{\eta}_{J_1}}^{\boldsymbol{\epsilon}}  \mathcal{G}(\boldsymbol{p},\boldsymbol{\eta},\boldsymbol{\xi},\sigma^1,\sigma^2,\alpha,\tilde{\alpha})      |\hat{\varphi}(\tfrac{p_0}{\hbar}) |^2 \big|  \, d\boldsymbol{p} d\boldsymbol{\eta}d\boldsymbol{\xi}.
	\end{aligned}
\end{equation} 
We will now estimate each pice of this expression separately. Firstly we have that
  \begin{equation}\label{EQ:exp_aver_b_Mainterm_5}
	\begin{aligned}
	\MoveEqLeft   \int_{\R^{2k}_{+}}\prod_{i=1}^{n} \frac{ \boldsymbol{1}_{[0,t]}( \boldsymbol{s}_{1,k}^{+}+\hbar \boldsymbol{t}_{1,a_k}^{+})  \boldsymbol{1}_{[0,t]}( \boldsymbol{\tilde{s}}_{1,k}^{+}+\hbar \boldsymbol{\tilde{t}}_{1,\tilde{a}_k}^{+}) }{\max(1,\hbar^{-1} | \boldsymbol{s}_{\sigma^1_i,\sigma^1_{i+1}-1}^{+} - \boldsymbol{\tilde{s}}_{\sigma_i^2,\sigma_{i+1}^2-1}^{+} +\hbar l^1_i(\boldsymbol{t}, \boldsymbol{\tilde{t}})) |)^\frac{d}{2}}   \, d\boldsymbol{\tilde{s}} d\boldsymbol{s}
	\\
	&\leq \frac{\hbar^n t^{2k-n}}{k!(k-n)!} \left( \int_{\R} \frac{1}{\max(1, | s |)^\frac{d}{2}} \,ds \right)^n = \left(\frac{2d}{d-2}\right)^n  \frac{\hbar^n t^{2k-n}}{k!(k-n)!}.
	\end{aligned}
\end{equation} 
By arguing as in the proof of Lemma~\ref{app_quadratic_integral_tech_est} we get that 
   \begin{equation}\label{EQ:exp_aver_b_Mainterm_6}
   	\begin{aligned}
	\MoveEqLeft  \sum_{l=1}^{4^{a_{k}}}   \sum_{|\epsilon_1|,\dots,|\epsilon_{j_1}|\leq d+1} \sum_{|\nu_1|,\dots,|\nu_{\tilde{j}_1}|\leq d+1}  \int_{\R^{a_k+\tilde{a}_k}_{+}}
	  \boldsymbol{1}_{B_l}(\boldsymbol{t},\boldsymbol{\tilde{t}}) \prod_{i\in J_1} \frac{C(\epsilon_i)}{(2\pi t_i))^\frac{d}{2} }
	   \prod_{i\in \tilde{J}_1} \frac{C(\nu_i)   }{(2\pi \tilde{t}_i)^\frac{d}{2}}  d\boldsymbol{t}  d\boldsymbol{\tilde{t}}   
	   \\
	   &\times \int    \prod_{i\in J_1} \frac{|x_i ^{\epsilon_i} |}{(1+|x_i |^2)^{d+1}}\,d\boldsymbol{x} \int \prod_{i\in \tilde{J}_1} \frac{|y_i ^{\nu_i} |}{(1+|y_i |^2)^{d+1}}   \, d\boldsymbol{y}
	   \\
	   \leq{}&  \tilde{C}_d^{a_k+\tilde{a}_k}\left(\frac{d}{d-2}\right)^{a_k+\tilde{a}_k}.
	\end{aligned}
\end{equation}
Lastly using the definition of the functions $\mathcal{G}(\boldsymbol{p},\boldsymbol{\eta},\boldsymbol{\xi},\sigma^1,\sigma^2,\alpha,\tilde{\alpha})$, equation~\eqref{EQ_def_G_4.1}, and our assumptions on the single site potential $V$ we get that
 \begin{equation}\label{EQ:exp_aver_b_Mainterm_7}
	\begin{aligned}
	\MoveEqLeft   \sup_{|\gamma_1|,\dots,|\gamma_n|\leq d+1}  \int  \big| \partial_{p_1}^{\gamma_1}\cdots \partial_{p_n}^{\gamma_n}    \partial_{\boldsymbol{\xi}_{\tilde{J}_1}}^{\boldsymbol{\nu}} \partial_{\boldsymbol{\eta}_{J_1}}^{\boldsymbol{\epsilon}}  \mathcal{G}(\boldsymbol{p},\boldsymbol{\eta},\boldsymbol{\xi},\sigma^1,\sigma^2,\alpha,\tilde{\alpha})         |\hat{\varphi}(\tfrac{p_0}{\hbar}) |^2 \big|  \, d\boldsymbol{p}d\boldsymbol{\eta}d\boldsymbol{\xi}
	 \\
	 \leq {}& 2^{(d+1)(j-k)} 4^{(d+1)n} (2\pi\hbar)^d \norm{\varphi}_{L^2(\R^d)}^2 \left(\frac{\norm{\hat{V}}_{1,\infty,3d+3}}{(2\pi)^d} \right)^{|\alpha| + |\tilde{\alpha}|} ,
	\end{aligned}
\end{equation}
where the numbers $2^{(d+1)(j-k)} 4^{(d+1)n}$ comes from multiple application of the product rule. From combining \cref{EQ:exp_aver_b_Mainterm_3,EQ:exp_aver_b_Mainterm_4,EQ:exp_aver_b_Mainterm_5,EQ:exp_aver_b_Mainterm_6,EQ:exp_aver_b_Mainterm_7} we obtain the estimate
   \begin{equation}\label{EQ:exp_aver_b_Mainterm_8}
   	\begin{aligned}
	| \Aver{\mathcal{T}(n,\alpha,\tilde{\alpha},\sigma^1,\sigma^2,\mathrm{id})}|
	 \leq  C_d^{a_k+\tilde{a}_k+n}  \norm{\varphi}_{L^2(\R^d)}^2 \left(\norm{\hat{V}}_{1,\infty,3d+3}\right)^{|\alpha|+|\tilde{\alpha}|}  \frac{ (\rho t)^{2k-n}}{k!(k-n)!},
	\end{aligned}
\end{equation}
where the constant $C_d$ is only depending on the dimension.  Next we turn to the case where $\kappa\neq\mathrm{id}$. For this case there exists a smallest index $i$ such that $\kappa(i)\neq i$. We will in the following denote this index by $i^{*}$.  Again we use the second expression for the kernel given in Observation~\ref{obs_form_I_op_kernel}, but this time we insert the function $f(\boldsymbol{s},\boldsymbol{\tilde{s}})$ defined by
\begin{equation*}
	f(\boldsymbol{s},\boldsymbol{\tilde{s}}) = \boldsymbol{1}_{[0,\hbar^{-1}t]}\big( \sum_{i=1,i\neq i^{*}}^k s_i \big) \boldsymbol{1}_{[0,\hbar^{-1}t]}\big( \sum_{i=1,i \notin \sigma^2}^k \tilde{s}_i \big).
\end{equation*}  
We then get that
\begin{equation*}
	\begin{aligned}
	\MoveEqLeft \mathcal{T}(n,\alpha,\tilde{\alpha},\sigma^1,\sigma^2,\kappa)
	= \sum_{(\boldsymbol{x},\boldsymbol{\tilde{x}})\in \mathcal{X}_{\neq}^{2k}}  \frac{1}{(2\pi\hbar)^{d}}\prod_{i=1}^n \rho^{-1} \delta(x_{\sigma_i^1}- \tilde{x}_{\sigma_{\kappa(i)}^2})   \int_{\R_{+}^{|\alpha|+|\tilde{\alpha}|}}    \int \delta(p_k-q_k)    \boldsymbol{1}_{[0,\hbar^{-1}t]}( \boldsymbol{s}_{1,k}^{+}+ \boldsymbol{t}_{1,a_k}^{+})
	\\
	\times&  \boldsymbol{1}_{[0,\hbar^{-1}t]}( \boldsymbol{\tilde{s}}_{1,k}^{+}+ \boldsymbol{\tilde{t}}_{1,\tilde{a}_k}^{+}) f(\boldsymbol{s},\boldsymbol{\tilde{s}}) \prod_{i=1}^{a_k} e^{i \frac{1}{2}t_{i}  \eta_{i}^2}  \prod_{i=1}^{\tilde{a}_k} e^{-i \frac{1}{2} \tilde{t}_{i}  \xi_{i}^2} \Big\{  \prod_{m=1}^k  e^{ -i   \langle  \hbar^{-1/d}x_m,p_{m}-p_{m-1} \rangle}  
	 e^{i  s_{m} \frac{1}{2} p_{m}^2}  \hat{\mathcal{V}}_{\alpha_m}(p_m,p_{m-1},\boldsymbol{\eta} )   
	 \\
	 \times&   e^{ i   \langle  \hbar^{-1/d}\tilde{x}_m,q_{m}-q_{m-1} \rangle}  e^{-i  \tilde{s}_{m} \frac{1}{2} q_{m}^2}  \overline{\hat{\mathcal{V}}_{\tilde{\alpha}_m}(q_m,q_{m-1},\boldsymbol{\xi} ) } \Big\}  
	  e^{i  (\hbar^{-1}t- \boldsymbol{s}_{1,k}^{+}- \boldsymbol{t}_{1,a_k}^{+}) \frac{1}{2} p_0^2} e^{-i  (\hbar^{-1}t- \boldsymbol{\tilde{s}}_{1,k}^{+}- \boldsymbol{\tilde{t}}_{1,\tilde{a}_k}^{+}) \frac{1}{2} q_0^2} 
	 \\
	 \times& \hat{\varphi}(\tfrac{p_0}{\hbar})  \hat{\varphi}(\tfrac{q_0}{\hbar})   \,d\boldsymbol{\eta}d\boldsymbol{p} \boldsymbol{\xi}d\boldsymbol{q}   d\boldsymbol{t} d\boldsymbol{s}  d\boldsymbol{\tilde{t}} d\boldsymbol{\tilde{s}},
	\end{aligned}
\end{equation*}
where we also have made a rescaling in some of the time integrals compared to before. We note the we can freely insert $f$ due to the two characteristic functions already there. We now introduce the variables $s_0$ and $\tilde{s}_0$ and get that
\begin{equation*}
	\begin{aligned}
	\MoveEqLeft \mathcal{T}(n,\alpha,\tilde{\alpha},\sigma^1,\sigma^2,\kappa)
	=  \sum_{(\boldsymbol{x},\boldsymbol{\tilde{x}})\in \mathcal{X}_{\neq}^{2k}} \frac{1}{(2\pi\hbar)^{d}} \prod_{i=1}^n \rho^{-1} \delta(x_{\sigma_i^1}- \tilde{x}_{\sigma_{\kappa(i)}^2})   \int_{\R_{+}^{|\alpha|+|\tilde{\alpha}|+2}} \int  \delta(p_k-q_k) \delta(\tfrac{t}{\hbar}- \boldsymbol{s}_{0,k}^{+}- \boldsymbol{t}_{1,a_k}^{+}) 
	\\
	\times &  \delta(\tfrac{t}{\hbar}- \boldsymbol{\tilde{s}}_{0,k}^{+}- \boldsymbol{\tilde{t}}_{1,\tilde{a}_k}^{+})  \prod_{i=1}^{a_k} e^{i \frac{1}{2} t_{i}  \eta_{i}^2 }  \prod_{i=1}^{\tilde{a}_k} e^{-i \frac{1}{2} \tilde{t}_{i}  \xi_{i}^2}  \prod_{m=1}^k  \Big\{e^{ -i   \langle  \hbar^{-1/d}x_m,p_{m}-p_{m-1} \rangle}  
	 e^{i  s_{m} \frac{1}{2} p_{m}^2}    \hat{\mathcal{V}}_{\alpha_m}(p_m,p_{m-1},\boldsymbol{\eta} ) 
	 \\
	 \times &    e^{ i   \langle  \hbar^{-1/d}\tilde{x}_m,q_{m}-q_{m-1} \rangle} e^{-i  \tilde{s}_{m} \frac{1}{2} q_{m}^2}   \overline{\hat{\mathcal{V}}_{\tilde{\alpha}_m}(q_m,q_{m-1},\boldsymbol{\xi} ) } \Big\}    e^{i  s_0 \frac{1}{2} p_0^2}e^{-i \tilde{s}_0 \frac{1}{2} q_0^2} 
	 \hat{\varphi}(\tfrac{p_0}{\hbar})  \hat{\varphi}(\tfrac{q_0}{\hbar})   f(\boldsymbol{s},\boldsymbol{\tilde{s}})  \,d\boldsymbol{\eta}d\boldsymbol{p}\boldsymbol{\xi}d\boldsymbol{q}   d\boldsymbol{t} d\boldsymbol{s}  d\boldsymbol{\tilde{t}}d\boldsymbol{\tilde{s}}.
	\end{aligned}
\end{equation*}
 We define the functions $\zeta(t)$ for $t\geq0$ to be
 \begin{equation*}
 	\zeta(t) = \frac{1}{\max(1,t)}.
 \end{equation*}
  Using this function we insert the two functions $e^{(\hbar^{-1}t- \boldsymbol{s}_{0,k}^{+}- \boldsymbol{t}_{1,a_k}^{+})\zeta(\hbar^{-1}t)}$ and e$^{(\hbar^{-1}t- \boldsymbol{\tilde{s}}_{0,k}^{+}- \boldsymbol{\tilde{t}}_{1,\tilde{a}_k}^{+})\zeta(\hbar^{-1}t)}$ as they are identical $1$ on our domain of integration. We will in the following use the convention $\zeta(\hbar^{-1}t)=\zeta$. So from inserting the functions we obtain
\begin{equation*}
	\begin{aligned}
	\MoveEqLeft \mathcal{T}(n,\alpha,\tilde{\alpha},\sigma^1,\sigma^2,\kappa)
	= \sum_{(\boldsymbol{x},\boldsymbol{\tilde{x}})\in \mathcal{X}_{\neq}^{2k}} \frac{e^{2\hbar^{-1}t\zeta}}{(2\pi)^2(2\pi\hbar)^{d}}\prod_{i=1}^n \rho^{-1} \delta(x_{\sigma_i^1}- \tilde{x}_{\sigma_{\kappa(i)}^2})   \int_{\R_{+}^{2(j+1)}} f(\boldsymbol{s},\boldsymbol{\tilde{s}}) 
	 \int \delta(p_k-q_k)  
	 \\
	 \times& \prod_{i=1}^{a_k} e^{i t_{i} ( \frac{1}{2}   \eta_{i}^2+\nu +i\zeta )}   \prod_{i=1}^{\tilde{a}_k} e^{- i  \tilde{t}_{i}( \frac{1}{2}   \xi_{i}^2-\tilde{\nu} - i\zeta )}
	 \prod_{m=1}^k \Big\{  e^{i  s_{m} ( \frac{1}{2} p_{m}^2+\nu+i\zeta)}
	 e^{ -i   \langle  \hbar^{-1/d}x_m,p_{m}-p_{m-1} \rangle}   \hat{\mathcal{V}}_{\alpha_m}(p_m,p_{m-1},\boldsymbol{\eta} ) 
	\\
	\times& e^{ i   \langle  \hbar^{-1/d}\tilde{x}_m,q_{m}-q_{m-1} \rangle} e^{-i  \tilde{s}_{m} (\frac{1}{2} q_{m}^2-\tilde{\nu}-i\zeta)}   \overline{\hat{\mathcal{V}}_{\tilde{\alpha}_m}(q_m,q_{m-1},\boldsymbol{\xi} ) } \Big\}  e^{i  s_0 (\frac{1}{2} p_0^2+\nu+i\zeta)} e^{-i \tilde{s}_0 (\frac{1}{2} q_0^2-\tilde{\nu}-i\zeta)}    e^{-i\hbar^{-1}t\tilde{\nu}}  e^{-i\hbar^{-1}t\nu}   
	\\
	\times&  
	 \hat{\varphi}(\tfrac{p_0}{\hbar})  \hat{\varphi}(\tfrac{q_0}{\hbar})   \,d\boldsymbol{\eta}d\boldsymbol{p} \boldsymbol{\xi}d\boldsymbol{q} d\nu d\tilde{\nu}   d\boldsymbol{t} d\boldsymbol{s}  d\boldsymbol{\tilde{t}} d\boldsymbol{\tilde{s}},
	\end{aligned}
\end{equation*}
where we have also written the two delta functions $\delta(\tfrac{t}{\hbar}- \boldsymbol{s}_{0,k}^{+}- \boldsymbol{t}_{1,a_k}^{+})$ and $\delta(\tfrac{t}{\hbar}- \boldsymbol{\tilde{s}}_{0,k}^{+}- \boldsymbol{\tilde{t}}_{1,a_k}^{+})$ as Fourier transforms of $1$. We observe that this has the form as the expressions considered in Lemma~\ref{LE:Exp_ran_phases} and applying this Lemma yields
\begin{equation}
	\begin{aligned}
	\MoveEqLeft \Aver{ \mathcal{T}(n,\alpha,\tilde{\alpha},\sigma^1,\sigma^2,\kappa)}
	=   \frac{ \left(\rho\hbar(2\pi)^d\right)^{2k-n}e^{2\hbar^{-1}t\zeta}}{(2\pi)^2(2\pi\hbar)^{d}}   \int_{\R_{+}^{2(j+1)}} f(\boldsymbol{s},\boldsymbol{\tilde{s}}) \int \delta(p_k-q_k)  \Lambda_n(\boldsymbol{p},\boldsymbol{q},\sigma^1,\sigma^2,\kappa) e^{-i\hbar^{-1}t(\nu+\tilde{\nu})} 
	\\
	\times&   
	 \prod_{i=1}^{a_k} e^{i t_{i} ( \frac{1}{2}   \eta_{i}^2+\nu +i\zeta )}   \prod_{i=1}^{\tilde{a}_k}e^{- i  \tilde{t}_{i}( \frac{1}{2}   \xi_{i}^2-\tilde{\nu} - i\zeta )}
	\prod_{m=1}^k  \Big\{   \hat{\mathcal{V}}_{\alpha_m}(p_m,p_{m-1},\boldsymbol{\eta} ) 
	   e^{i  s_{m} ( \frac{1}{2} p_{m}^2+\nu+i\zeta)}  e^{-i  \tilde{s}_{m} (\frac{1}{2} q_{m}^2-\tilde{\nu}-i\zeta)}  
	   \\
	   \times& \overline{\hat{\mathcal{V}}_{\tilde{\alpha}_m}(q_m,q_{m-1},\boldsymbol{\xi} ) } \Big\}    e^{i  s_0 (\frac{1}{2} p_0^2+\nu+i\zeta)}e^{-i \tilde{s}_0 (\frac{1}{2} q_0^2-\tilde{\nu}-i\zeta)} 
	 \hat{\varphi}(\tfrac{p_0}{\hbar})  \hat{\varphi}(\tfrac{q_0}{\hbar})   \,d\boldsymbol{\eta}d\boldsymbol{p}\boldsymbol{\xi}d\boldsymbol{q} d\nu d\tilde{\nu}   d\boldsymbol{t} d\boldsymbol{s}  d\boldsymbol{\tilde{t}}d\boldsymbol{\tilde{s}}.
	\end{aligned}
\end{equation}
We will again divide into different cases depending on the absolute value of $t$'s and  $\tilde{t}$'s as we did above. We do the same type of argument and obtain the estimate
   \begin{equation}\label{EQ:exp_aver_b_Mainterm_9}
   	\begin{aligned}
	\Aver{ \mathcal{T}(n,\alpha,\tilde{\alpha},\sigma^1,\sigma^2,\kappa)}
	 \leq{}&  \frac{ \left(\rho\hbar(2\pi)^d\right)^{2k-n} e^{2\hbar^{-1}t\zeta}}{(2\pi)^2(2\pi\hbar)^{d}}  \sum_{l=1}^{2^{a_{k}+\tilde{a}_k}}   \sum_{|\epsilon_1|,\dots,|\epsilon_{j_1}|\leq d+1} \sum_{|\nu_1|,\dots,|\nu_{\tilde{j}_1}|\leq d+1} 
	 \\
	 &\times   \int_{\R^{a_k+\tilde{a}_k}_{+}}
	  \boldsymbol{1}_{B_l}(\boldsymbol{t},\boldsymbol{\tilde{t}}) \prod_{i\in J_1} \frac{C(\epsilon_i)}{(2\pi t_i)^\frac{d}{2} }
	  \prod_{i\in \tilde{J}_1} \frac{C(\nu_i)   }{(2\pi \tilde{t}_i)^\frac{d}{2}}  |\mathcal{I}_n(\boldsymbol{t},\boldsymbol{\tilde{t}}) |  d\boldsymbol{t} d\boldsymbol{\tilde{t}} ,
	\end{aligned}
\end{equation}
where
\begin{equation*}
	\begin{aligned}
	\MoveEqLeft \mathcal{I}_n(\boldsymbol{t},\boldsymbol{\tilde{t}})
	=     \int_{\R_{+}^{2(k+1)}} f(\boldsymbol{s},\boldsymbol{\tilde{s}}) \int  \delta(p_k-q_k)   e^{-i\hbar^{-1}t(\nu+\tilde{\nu})} 
	 \Lambda_n(\boldsymbol{p},\boldsymbol{q},\sigma^1,\sigma^2,\kappa) \prod_{m=0}^k e^{i  s_{m} ( \frac{1}{2} p_{m}^2+\nu+i\zeta)} e^{-i  \tilde{s}_{m} (\frac{1}{2} q_{m}^2-\tilde{\nu}-i\zeta)} 
	\\
	\times&   \prod_{i\in \tilde{J}_1} \frac{e^{i\langle y_i, \xi_i  \rangle}  e^{-i  \tilde{t}_{i}^{-1} \frac{1}{2} y_{i}^2}e^{i  \tilde{t}_{i} (\tilde{\nu}+i\zeta)} y_i ^{\nu_i} }{(1+|y_i |^2)^{d+1}}  \prod_{i\in J_1} \frac{e^{i\langle x_i, \eta_i  \rangle}e^{i  t_{i}^{-1}  \frac{1}{2} x_{i}^2 } e^{i  t_{i} (\nu+i\zeta)} x_i ^{\epsilon_i} }{(1+|x_i |^2)^{d+1}} \prod_{i\in J_2} e^{i  t_{i} ( \frac{1}{2} \eta_{i}^2 +\nu+i\zeta)}   \prod_{i\in \tilde{J}_2}e^{-i  \tilde{t}_{i} (\frac{1}{2} \xi_{i}^2-\tilde{\nu}-i\zeta)}  
	  \\
	  \times& 
	   \Big\{  \partial_{\boldsymbol{\xi}_{\tilde{J}_1}}^{\boldsymbol{\nu}} \partial_{\boldsymbol{\eta}_{J_1}}^{\boldsymbol{\epsilon}}  \prod_{m=1}^k    \hat{\mathcal{V}}_{\alpha_m}(p_m,p_{m-1},\boldsymbol{\eta} )  \overline{\hat{\mathcal{V}}_{\tilde{\alpha}_m}(q_m,q_{m-1},\boldsymbol{\xi} ) } \Big\}  \hat{\varphi}(\tfrac{p_0}{\hbar})  \hat{\varphi}(\tfrac{q_0}{\hbar})  
	  \, d\boldsymbol{x} d\boldsymbol{y}  d\boldsymbol{\eta}d\boldsymbol{p} \boldsymbol{\xi}d\boldsymbol{q}d\nu d\tilde{\nu}    d\boldsymbol{s}   d\boldsymbol{\tilde{s}}.
	\end{aligned}
\end{equation*}
Recall that $i^{*}$ was the first index such that $\kappa(i^{*})\neq i^{*}$ and that the function $\Lambda_n(\boldsymbol{p},\boldsymbol{q},\sigma^1,\sigma^2,\kappa)$ is given by
  \begin{equation*}
  	\begin{aligned}
  	\MoveEqLeft\Lambda_n(\boldsymbol{p},\boldsymbol{q},\sigma^1,\sigma^2,\kappa)   
	\\
	 = {}& \prod_{i=1}^n \delta(p_{\sigma^1_{i-1}}- q_{\sigma^2_{\kappa(i)-1}} - p_{\sigma^1_{n}}+ q_{\sigma^2_{n}} - l^{\kappa}_{i}(\boldsymbol{q}_{\sigma^2}))\prod_{i=1}^{n+1} \prod_{m=\sigma^1_{i-1}+1}^{\sigma^1_{i}-1} \delta(p_m-p_{\sigma^1_{i-1}})  \prod_{m=\sigma^2_{i-1}+1}^{\sigma^2_{i}-1} \delta(q_m-q_{\sigma^1_{i-1}}).
	\end{aligned}
  \end{equation*}
We now do the integrals in $s_0$, $s_{\sigma^1_{i^{*}}}$, $\tilde{s}_{\sigma^2_{i}}$ for all $i\in\{0,\dots,n\}$ all $p$ and all $q$ for which the index is not in $\sigma^2$. Note that the integrals we do in $s$ and $\tilde{s}$ variables are exactly the variables the $f(\boldsymbol{s},\boldsymbol{\tilde{s}})$ does not depend on. After evaluating these integrals and moving the absolute value in under the integrals we get the estimate
\begin{equation}\label{EQ:exp_aver_b_Mainterm_10}
	\begin{aligned}
	 \big| \mathcal{I}_n(\boldsymbol{t},\boldsymbol{\tilde{t}})\big|
	\leq{}& \int  \prod_{i\in J_1} \frac{| x_i ^{\epsilon_i}| }{(1+|x_i |^2)^{d+1}} \,d\boldsymbol{x} \int\prod_{i\in \tilde{J}_1} \frac{| y_i ^{\nu_i}| }{(1+|y_i |^2)^{d+1}}  \,d\boldsymbol{y}      \int_{\R_{+}^{2k-n-1}} f(\boldsymbol{s},\boldsymbol{\tilde{s}})  d\boldsymbol{s}   d\boldsymbol{\tilde{s}}
	\\
	&\times   \int |\hat{\varphi}(\tfrac{q_0}{\hbar}) |^2   \frac{1}{|\frac{1}{2} q_0^2-\tilde{\nu}-i\zeta|}\frac{1}{|\frac{1}{2} q_0^2+\nu +i\zeta|}
	  \frac{1}{| \frac{1}{2} (q_{\kappa(i^{*})}+q_{i^{*}-1}-q_{\kappa(i^{*})-1})^2 +\nu+i\zeta|}  
	  \\
	  &\times    \prod_{m=1}^n \frac{1}{ |\frac{1}{2} q_{m}^2-\tilde{\nu}-i\zeta|}
	 \Big| \partial_{\boldsymbol{\xi}_{\tilde{J}_1}}^{\boldsymbol{\nu}} \partial_{\boldsymbol{\eta}_{J_1}}^{\boldsymbol{\epsilon}}   \mathcal{G}(\boldsymbol{q},\boldsymbol{\eta},\boldsymbol{\xi},\sigma^1,\sigma^2,\alpha,\tilde{\alpha})  \Big| 
	  \,d\boldsymbol{\eta}d\boldsymbol{\xi}d\boldsymbol{q} d\nu d\tilde{\nu},
	\end{aligned}
\end{equation}
where we have used that $l_{i^{*}}^\kappa(\boldsymbol{q}) = q_{i^{*}-1}-q_{\kappa(i^{*})-1}$ and $q_{\kappa(1)-1}+l_{1}^\kappa(\boldsymbol{q}) = q_{0}$ and the notation
\begin{equation*}
	\begin{aligned}
	 \mathcal{G}(\boldsymbol{q},\boldsymbol{\eta},\boldsymbol{\xi},\sigma^1,\sigma^2,\alpha,\tilde{\alpha}) ={}&   \prod_{i=1}^n \hat{\mathcal{V}}_{\alpha_{\sigma^1_i}}(q_{\kappa(i)}+l_{i}^\kappa(\boldsymbol{q}),q_{\kappa(i)-1}+l_{i}^\kappa(\boldsymbol{q}),\boldsymbol{\eta})  \overline{\hat{\mathcal{V}}_{\tilde{\alpha}_{\sigma^1_i}}(q_i,q_{i-1},\boldsymbol{\xi})}  
	\\
	&\times  \prod_{i=0}^n  \prod_{m=\sigma^1_i +1}^{\sigma^1_{i+1} -1} \hat{\mathcal{V}}_{\alpha_m}(q_{\kappa(i)}+l_{i}^\kappa(\boldsymbol{q}),q_{\kappa(i)}+l_{i}^\kappa(\boldsymbol{q}),\boldsymbol{\eta})  \prod_{m=\sigma^2_i +1}^{\sigma^2_{i+1} -1} \overline{\hat{\mathcal{V}}_{\tilde{\alpha}_m}(q_i,q_i,\boldsymbol{\xi})}.   
	\end{aligned}
\end{equation*}
Again we wil consider each of the integrals separately. The integrals  over $\boldsymbol{x}$ and $\boldsymbol{y}$ we can estimate with \eqref{EQ:exp_aver_b_Mainterm_6}. For the integrals over $\boldsymbol{s}$ and $\boldsymbol{\tilde{s}}$ we have that
\begin{equation}\label{EQ:exp_aver_b_Mainterm_11}
	\begin{aligned}
	 \int_{\R_{+}^{2k-n-1}} f(\boldsymbol{s},\boldsymbol{\tilde{s}})  d\boldsymbol{s}  d\boldsymbol{\tilde{s}}
	=  \int_{\R_{+}^{2k-n-1}} \boldsymbol{1}_{[0,\hbar^{-1}t]}\big(  \boldsymbol{s}_{1,k-1}^{+} \big) \boldsymbol{1}_{[0,\hbar^{-1}t]}\big( \boldsymbol{\tilde{s}}_{1,k-n}^{+} \big) d\boldsymbol{s}   d\boldsymbol{\tilde{s}} 
	 = \frac{(\hbar^{-1}t)^{2k-n-1}}{(k-1)!(k-n)!}.
	\end{aligned}
\end{equation}
In the estimates of the remaining integrals we will introduce the term  
  \begin{equation*}
  	\frac{\langle q_0\rangle^{4d+4}}{\langle q_0\rangle^{4d+4}} \frac{\langle q_{\kappa(i^{*})}+q_{i^{*}-1}-q_{\kappa(i^{*})-1}\rangle^{d+1}}{\langle q_{\kappa(i^{*})}+q_{i^{*}-1}-q_{\kappa(i^{*})-1}\rangle^{d+1}} \prod_{i=1}^n \frac{\langle q_{i}-q_{i-1}\rangle^{5d+5}}{\langle  q_{i}-q_{i-1}\rangle^{5d+5}}. 
  \end{equation*}
  Firstly we have that
  \begin{equation}\label{EQ:exp_aver_b_Mainterm_12}
	\begin{aligned}
	\MoveEqLeft  \sup_{\boldsymbol{q}}   \int  \langle q_{{\kappa(i^{*})}-1}-q_{\kappa(i^{*})-1}\rangle^{d+1} \prod_{i=1}^n \langle q_{i}-q_{i-1}\rangle^{5d+5}  	 \Big| \partial_{\boldsymbol{\xi}_{\tilde{J}_1}}^{\boldsymbol{\nu}} \partial_{\boldsymbol{\eta}_{J_1}}^{\boldsymbol{\epsilon}} \mathcal{G}(\boldsymbol{q},\boldsymbol{\eta},\boldsymbol{\xi},\sigma^1,\sigma^2,\alpha,\tilde{\alpha}) \Big| \,d\boldsymbol{\eta}d\boldsymbol{\xi}
	 \\
	 &\leq C^{a_k+\tilde{a}_k}  \norm{\hat{V}}_{1,\infty,3d+3}^{|\alpha|+|\tilde{\alpha}|} .
	\end{aligned}
\end{equation}
  Next using that 
    \begin{equation*}
  	\frac{\langle q_{\kappa(i^{*})}+q_{i^{*}-1}-q_{\kappa(i^{*})-1}\rangle^{d+1} \langle q_{\kappa(i^{*})}\rangle^{d+1} \langle q_{\kappa(i^{*})+1}\rangle^{d+1}\langle q_{\kappa(i^{*})-1}\rangle^{d+1}}{\langle q_{\kappa(i^{*})}-q_{\kappa(i^{*})-1}\rangle^{d+1}\langle q_0\rangle^{4d+4}}  \prod_{i=1}^n \frac{1}{\langle  q_{i}-q_{i-1}\rangle^{4d+4}} \leq C.
  \end{equation*}
  We only have to consider the integral
  \begin{equation}\label{EQ:exp_aver_b_Mainterm_13}
	\begin{aligned}
	\MoveEqLeft   \int  \frac{1}{\langle\tilde{\nu}\rangle|\frac{1}{2} q_0^2-\tilde{\nu}-i\zeta|}\frac{1}{\langle{\nu}\rangle|\frac{1}{2} q_0^2+\nu +i\zeta|}   \frac{ \langle q_{\kappa(i^{*})}+q_{i^{*}-1}-q_{\kappa(i^{*})-1}\rangle^{-d-1}}{| \frac{1}{2} (q_{\kappa(i^{*})}+q_{i^{*}-1}-q_{\kappa(i^{*})-1})^2 +\nu+i\zeta|} 
	  \\
	  &\times  \frac{\langle\tilde{\nu}\rangle \langle{\nu}\rangle}{\langle q_{\kappa(i^{*})-1}\rangle^{d+1}\langle q_{\kappa(i^{*})}\rangle^{d+1} \langle q_{\kappa(i^{*})+1}\rangle^{d+1}}   \prod_{m=1}^n \frac{\langle  q_{m}-q_{m-1}\rangle^{-d-1}}{ |\frac{1}{2} q_{m}^2-\tilde{\nu}-i\zeta|}
	\, d\boldsymbol{q} d\nu d\tilde{\nu}.
	\end{aligned}
\end{equation}
  Applying the estimate $\langle q_{\kappa(i^{*})  } - q_{\kappa(i^{*}) -1 } \rangle^{-d-1} \langle q_{\kappa(i^{*}) +1 } - q_{\kappa(i^{*}) } \rangle^{-d-1}\leq1$ we can use Lemma~\ref{LE:est_res_combined} with $q_{\kappa(i^{*})}$ as ``$p$'' and $ q_{\kappa(i^{*})-1}- q_{i^{*}-1}$ as ``$q$''.  This yields the estimate
    \begin{equation}\label{EQ:exp_aver_b_Mainterm_15}
	\begin{aligned}
	\MoveEqLeft   \int_{\R^{nd+2}}    \frac{1}{\langle\tilde{\nu}\rangle|\frac{1}{2} q_0^2-\tilde{\nu}-i\zeta|}\frac{1}{\langle{\nu}\rangle|\frac{1}{2} q_0^2+\nu +i\zeta|}   \frac{ \langle q_{\kappa(i^{*})}+q_{i^{*}-1}-q_{\kappa(i^{*})-1}\rangle^{-d-1}}{| \frac{1}{2} (q_{\kappa(i^{*})}+q_{i^{*}-1}-q_{\kappa(i^{*})-1})^2 +\nu+i\zeta|} 
	  \\
	  &\times  \frac{\langle\tilde{\nu}\rangle \langle{\nu}\rangle}{\langle q_{\kappa(i^{*})-1}\rangle^{d+1} \langle q_{\kappa(i^{*})}\rangle^{d+1} \langle q_{\kappa(i^{*})+1}\rangle^{d+1}}   \prod_{m=1}^n \frac{\langle  q_{m}-q_{m-1}\rangle^{-d-1}}{ |\frac{1}{2} q_{m}^2-\tilde{\nu}-i\zeta|}
	\, d\boldsymbol{q} d\nu d\tilde{\nu}
	\\
	\leq{}& C\log(\zeta)^2\int  \frac{1}{\langle\tilde{\nu}\rangle|\frac{1}{2} q_0^2-\tilde{\nu}-i\zeta|}\frac{1}{\langle{\nu}\rangle|\frac{1}{2} q_0^2+\nu +i\zeta|} \frac{\langle q_{\kappa(i^{*})-1}\rangle^{-d-1}}{| q_{\kappa(i^{*})-1}- q_{i^{*}-1}|} \frac{ \langle q_{\kappa(i^{*})+1}\rangle^{-d-1}}{|\frac{1}{2} q_{\kappa(i^{*})+1}^2-\tilde{\nu}-i\zeta|} 
	  \\
	  &\times   \prod_{m=1}^{\kappa(i^{*})-1} \frac{\langle  q_{m}-q_{m-1}\rangle^{-d-1}}{ |\frac{1}{2} q_{m}^2-\tilde{\nu}-i\zeta|} \prod_{m=\kappa(i^{*})+2}^n \frac{\langle  q_{m}-q_{m-1}\rangle^{-d-1}}{ |\frac{1}{2} q_{m}^2-\tilde{\nu}-i\zeta|}
	\, d\boldsymbol{q}_{1,\kappa(i^{*})-1} d\boldsymbol{q}_{\kappa(i^{*})+1,n} d\nu d\tilde{\nu}.
	\end{aligned}
\end{equation}
  Next we do the integral in $q_{\kappa(i^{*})-1}$.
  By using the estimate $\langle  q_{\kappa(i^{*})-1}-q_{\kappa(i^{*})-2}\rangle^{-d-1}\leq 1$ and Lemma~\ref{LE:resolvent_int_est} we get that
      \begin{equation}\label{EQ:exp_aver_b_Mainterm_15.1}
	\begin{aligned}
	  \int  \frac{\langle q_{\kappa(i^{*})-1}\rangle^{-d-1}}{| q_{\kappa(i^{*})-1}- q_{i^{*}-1}|}  \frac{\langle  q_{\kappa(i^{*})-1}-q_{\kappa(i^{*})-2}\rangle^{-d-1}}{ |\frac{1}{2} q_{\kappa(i^{*})-1}^2-\tilde{\nu}-i\zeta|}
	\, dq_{\kappa(i^{*})-1}
	\leq  C|\log(\zeta)|.
	\end{aligned}
\end{equation}
   For the remaining integrals we apply Lemma~\ref{LE:resolvent_int_est} repeatedly. This combined with \eqref{EQ:exp_aver_b_Mainterm_15} and \eqref{EQ:exp_aver_b_Mainterm_15.1} gives us the estimate    
      \begin{equation}\label{EQ:exp_aver_b_Mainterm_16}
	\begin{aligned}
	\MoveEqLeft   \int    \frac{1}{\langle\tilde{\nu}\rangle|\frac{1}{2} q_0^2-\tilde{\nu}-i\zeta|}\frac{1}{\langle{\nu}\rangle|\frac{1}{2} q_0^2+\nu +i\zeta|}   \frac{ \langle q_{\kappa(i^{*})}+q_{i^{*}-1}-q_{\kappa(i^{*})-1}\rangle^{-d-1}}{| \frac{1}{2} (q_{\kappa(i^{*})}+q_{i^{*}-1}-q_{\kappa(i^{*})-1})^2 +\nu+i\zeta|} 
	  \\
	  &\times  \frac{\langle\tilde{\nu}\rangle \langle{\nu}\rangle}{\langle q_{\kappa(i^{*})-1}\rangle^{d+1} \langle q_{\kappa(i^{*})}\rangle^{d+1} \langle q_{\kappa(i^{*})+1}\rangle^{d+1}}   \prod_{m=1}^n \frac{\langle  q_{m}-q_{m-1}\rangle^{-d-1}}{ |\frac{1}{2} q_{m}^2-\tilde{\nu}-i\zeta|}
	\, d\boldsymbol{q}_{1,n} d\nu d\tilde{\nu}
	\leq C^n |\log(\zeta)|^{n+3}.
	\end{aligned}
\end{equation}
Moreover, we have that
\begin{equation}\label{EQ:exp_aver_b_Mainterm_17}
	\int \langle q_0 \rangle^{4d+4} |\hat{\varphi}(\tfrac{q_0}{\hbar}) |^2 \, dq_0 \leq \hbar^d C \norm{\varphi}_{\mathcal{H}^{2d+2}_\hbar(\R^d)}^2.
\end{equation}
From combining \cref{EQ:exp_aver_b_Mainterm_6,EQ:exp_aver_b_Mainterm_10,EQ:exp_aver_b_Mainterm_9,EQ:exp_aver_b_Mainterm_17,EQ:exp_aver_b_Mainterm_16,EQ:exp_aver_b_Mainterm_12,EQ:exp_aver_b_Mainterm_11} we obtain the estimate
   \begin{equation}\label{EQ:exp_aver_b_Mainterm_18}
   	\begin{aligned}
	\Aver{\mathcal{T}(n,\alpha,\tilde{\alpha},\sigma^1,\sigma^2,\kappa)}  \leq   C_d^{a_k+\tilde{a_k}+n}    \left(\norm{\hat{V}}_{1,\infty,3d+3}\right)^{|\alpha|+|\tilde{\alpha}|}     \frac{\rho(\rho t)^{2k-n-1} \hbar |\log(\zeta)|^{n+3}}{(k-1)!(k-n)!} \norm{\varphi}_{\mathcal{H}^{2d+2}_\hbar(\R^d)}^2,
	\end{aligned}
\end{equation}
where the constant $C_d$ only depends on the dimension. Moreover, we have used that by definition of $\zeta$ we have the estimate $e^{2\hbar^{-1}t\zeta}\leq~C$, where the constant is independent of $t$ and $\hbar$. Combining \eqref{EQ:exp_aver_b_Mainterm_1}, \eqref{EQ:exp_aver_b_Mainterm_8} and \eqref{EQ:exp_aver_b_Mainterm_18} we get the estimate
\begin{equation*}
	\begin{aligned}
	\mathbb{E}\Big[\big\lVert \sum_{\boldsymbol{x}\in\mathcal{X}_{\neq}^{k}}\mathcal{I}_{0,0}(k,\boldsymbol{x},t;\hbar) \varphi\big\rVert_{L^2(\R^d)}^2\Big]
	\leq{}& \sum_{\alpha,\tilde{\alpha}\in\N^k}  C_d^{|\alpha|+|\tilde{\alpha}|-k} \left(\lambda\norm{\hat{V}}_{1,\infty,3d+3}\right)^{|\alpha|+|\tilde{\alpha}|}   \sum_{n=0}^k  \Big[   \frac{ (\rho t)^{2k-n}}{k!(k-n)!} \norm{\varphi}_{L^2(\R^d)}^2 
	\\
	&
	+ \hbar |\log(\zeta)|^{n+3} \frac{\rho(\rho t)^{2k-n-1} n!}{(k-1)!(k-n)!} \norm{\varphi}_{\mathcal{H}^{2d+2}_\hbar(\R^d)}^2 \Big],
	\end{aligned}
\end{equation*}
where we have used that the number of elements in $\mathcal{A}(k,n)$ is bounded by $2^k$ and that the number of elements in $\mathcal{S}_n$ is $n!$. Using our assumptions on the potential and the coupling constant we get that 
\begin{equation*}
	\begin{aligned}
	\MoveEqLeft \mathbb{E}\Big[\big\lVert \sum_{\boldsymbol{x}\in\mathcal{X}_{\neq}^{k}}\mathcal{I}_{0,0}(k,\boldsymbol{x},t;\hbar) \varphi\big\rVert_{L^2(\R^d)}^2\Big]
	\\
	\leq{}& C^k  \sum_{n=0}^k  \Big[   \frac{ (\rho t)^{2k-n}}{k!(k-n)!} \norm{\varphi}_{L^2(\R^d)}^2  
	+ \hbar |\log(\zeta)|^{n+3} \frac{\rho(\rho t)^{2k-n-1} n!}{(k-1)!(k-n)!} \norm{\varphi}_{\mathcal{H}^{2d+2}_\hbar(\R^d)}^2 \Big].
	\end{aligned}
\end{equation*}
Using that $\frac{n!}{(k-1)!}\leq k$, $\zeta \geq \frac{\hbar}{t} $ and evaluating the sum in $n$ we obtain the estimate
\begin{equation*}
	\begin{aligned}
	\mathbb{E}\Big[\big\lVert \sum_{\boldsymbol{x}\in\mathcal{X}_{\neq}^{k}}\mathcal{I}_{0,0}(k,\boldsymbol{x},t;\hbar) \varphi\big\rVert_{L^2(\R^d)}^2\Big]
	\leq   \frac{ C^{k}}{k!} \norm{\varphi}_{L^2(\R^d)}^2  + \hbar k C^k |\log(\tfrac{\hbar}{t})|^{k+3} \norm{\varphi}_{\mathcal{H}^{2d+2}_\hbar(\R^d)}^2,
	\end{aligned}
\end{equation*}
where $C$ depends on $\rho$, $t$, the single site potential $V$ and the coupling constant $\lambda$. This concludes the proof.
\end{proof}
\begin{remark}\label{remark_mod_eng_erdos_1}
As mentioned in the introduction the methodology of the proof just given follows the methodology of Eng and Erd{\H o}s \cite{MR2156632} but is not identical. Firstly we write the expressions different and use other notation. But the main difference in the proofs is that we  use a simultaneous  Fourier based argument to ensure that all integrals in the internal time variables is finite. That is in the notation in the above proof the $t$ and $\tilde{t}$ variables. For these variable a successive type of the same argument is used in \cite{MR2156632}. 

In the cases, where $\kappa\neq\mathrm{id}$ we get from this method the factor $k$ instead of $k!$ as in \cite{MR2156632} in the second term in our error estimate.  

For the case where $\kappa=\mathrm{id}$ we connect the techniques used to obtain the necessary integrability in the $s$ and $\tilde{s}$ variables with the integrability argument for the $t$ and $\tilde{t}$ variables. This ensures that we here do not need any derivatives of our initial data in this case. Moreover, from making this coupling we get the combinatorial factor of size $4^{n(d+1)}$ in most cases. This combinatorial factor comes from our use of integration by parts as each $p_i$ for $i\in\{1,\dots,n\}$  is in four places when we have connected the two arguments. If we do not make this connection the the combinatorial factor will be of the size
\begin{equation*}
	4^{n(d+1)}\prod_{i=1}^n (\sigma_{i+1}^1- \sigma_{i}^1)^{d+1}  (\sigma_{i+1}^2- \sigma_{i}^2)^{d+1},
\end{equation*}
for most cases. Hence for small values of $n$ this will make a difference. Since we use uniform estimate over all cases we have obtained a smaller constant from connecting the arguments. 

One disadvantage from making this ``coupling'' between the estimates is that we then need to preform the Duhamel expansion to higher order in recollisions. 
This is also the reason that we later will use a different argument to prove that the recollision error indeed converges to zero in the limit. 
\end{remark}
\begin{lemma}\label{expansion_aver_bound_Mainterm_2}
  Assume we are in the setting of Definition~\ref{def_remainder_k_0}. Let $\varphi \in \mathcal{H}^{3d+3}_\hbar(\R^d)$ and $\tau_0\in\N$. Then for any
  $\tau\in\{2,\dots,\tau_0\}$
    \begin{equation*}
	\begin{aligned}
	\MoveEqLeft \mathbb{E}\Big[ \big\lVert \sum_{k_1=1}^{k_0} \sum_{k_2=k_0-k_1+1}^{k_0} 
	  \sum_{(\boldsymbol{x}_1,\boldsymbol{x}_2)\in \mathcal{X}_{\neq}^{k_1+k_2}}  \mathcal{I}_{0,0}(k_1,\boldsymbol{x}_1,\iota,\tfrac{t}{\tau_0};\hbar)  \mathcal{I}_{0,0}(k_2,\boldsymbol{x}_2,\iota,\tfrac{\tau-1}{\tau_0}t;\hbar)\varphi\big\rVert_{L^2(\R^d)}^2 \Big]
	  \\
	  &\leq \frac{k_0^3}{\tau_0} \frac{C^{k_0}}{ (k_0+1)!} \norm{\varphi}_{L^2(\R^d)}^2
	 +  \frac{k_0^5\hbar}{\tau_0} C^{k_0} |\log(\tfrac{\hbar}{t})|^{2k_0+4}   \norm{\varphi}_{\mathcal{H}^{3d+3}_\hbar(\R^d)}^2,
	 \end{aligned}
\end{equation*}
where the constant $C$ depends on $\rho$, $t$, the single site potential $V$ and the coupling constant $\lambda$. In particular we have that the function is in  $L^2(\R^d)$ $\Pro$-almost surely. 
  \end{lemma}
  \begin{proof}
  Firstly we observe that
  \begin{equation}\label{EQ:Aver_full_ex_0,0_dob1}
	\begin{aligned}
	\MoveEqLeft  \big\lVert \sum_{k_1=1}^{k_0} \sum_{k_2=k_0-k_1+1}^{k_0} 
	  \sum_{(\boldsymbol{x}_1,\boldsymbol{x}_2)\in \mathcal{X}_{\neq}^{k_1+k_2}}  \mathcal{I}_{0,0}(k_1,\boldsymbol{x}_1,\iota,\tfrac{t}{\tau_0};\hbar)  \mathcal{I}_{0,0}(k_2,\boldsymbol{x}_2,\iota,\tfrac{\tau-1}{\tau_0}t;\hbar)\varphi\big\rVert_{L^2(\R^d)}^2 
	  \\
	  &\leq k_0^2  \sum_{k_1=1}^{k_0} \sum_{k_2=k_0-k_1+1}^{k_0} 
	   \big\lVert \sum_{(\boldsymbol{x}_1,\boldsymbol{x}_2)\in \mathcal{X}_{\neq}^{k_1+k_2}}  \mathcal{I}_{0,0}(k_1,\boldsymbol{x}_1,\iota,\tfrac{t}{\tau_0};\hbar)  \mathcal{I}_{0,0}(k_2,\boldsymbol{x}_2,\iota,\tfrac{\tau-1}{\tau_0}t;\hbar)\varphi\big\rVert_{L^2(\R^d)}^2 
	   \\
	   &
	   = k_0^2 \sum_{k=k_0+1}^{2k_0}  \sum_{k_1+k_2=k} 
	   \big\lVert \sum_{(\boldsymbol{x}_1,\boldsymbol{x}_2)\in \mathcal{X}_{\neq}^{k}}  \mathcal{I}_{0,0}(k_1,\boldsymbol{x}_1,\iota,\tfrac{t}{\tau_0};\hbar)  \mathcal{I}_{0,0}(k_2,\boldsymbol{x}_2,\iota,\tfrac{\tau-1}{\tau_0}t;\hbar)\varphi\big\rVert_{L^2(\R^d)}^2.
	 \end{aligned}
\end{equation}
  We now fix $k=k_1+k_2$.
As in the proof of Lemma~\ref{expansion_aver_bound_Mainterm} we get two sums over the set $\mathcal{X}_{\neq}^{k}$ when we take the $L^2$-norm. Again we will divide this double sum into cases depending on how many points they have in common. However this time we also need to keep track of how these points is chosen in relation to the two operators. With this in mind we have that
\begin{equation}\label{EQ:exp_aver_b_er_mainterm_1}
	\begin{aligned}
	\MoveEqLeft  \big\lVert \sum_{(\boldsymbol{x}_1,\boldsymbol{x}_2)\in \mathcal{X}_{\neq}^{k}}  \mathcal{I}_{0,0}(k_1,\boldsymbol{x}_1,\iota,\tfrac{t}{\tau_0};\hbar)  \mathcal{I}_{0,0}(k_2,\boldsymbol{x}_2,\iota,\tfrac{\tau-1}{\tau_0}t;\hbar)\varphi\big\rVert_{L^2(\R^d)}^2 = \sum_{\alpha,\tilde{\alpha}\in\N^k} (i\lambda)^\alpha(-i\lambda)^{\tilde{\alpha}}
	\\
	&\times \sum_{n=0}^k \sum_{\substack{n_1+n_2=n \\ n_1\leq k_1,n_2\leq k_2}} \sum_{\sigma^1\in\tilde{\mathcal{A}}(k_1,n_1,k_2,n_2)} \sum_{\substack{\tilde{n}_1+\tilde{n}_2=n \\ \tilde{n}_1\leq k_1,\tilde{n}_2\leq k_2}} \sum_{\sigma^2\in\tilde{\mathcal{A}}(k_1,\tilde{n}_1,k_2,\tilde{n}_2)}  \sum_{\kappa\in\mathcal{S}_n} \mathcal{T}(n,\sigma^1,\sigma^2,\alpha,\tilde{\alpha},\kappa),
	\end{aligned}
\end{equation}
where the sets $\tilde{\mathcal{A}}(k_1,n_1,k_2,n_2)$ are defined by
\begin{equation} \label{orderset_def_2}
\begin{aligned}
	 \tilde{\mathcal{A}}(k_1,n_1,k_2,n_2) = \{ \boldsymbol{\sigma} \in \{1,\dots,k_1+k_2\}^n \, |\, \sigma_1<\cdots<\sigma_{n_2}\leq k_2 < \sigma_{n_2+1}<\cdots <\sigma_{n_1+n_2} \}.
	\end{aligned}
\end{equation}
The numbers $\mathcal{T}(n,\sigma^1,\sigma^2,\alpha,\tilde{\alpha},\kappa)$ are given by  
\begin{equation*}
	\begin{aligned}
	\MoveEqLeft  \mathcal{T}(n,\sigma^1,\sigma^2,\alpha,\tilde{\alpha},\kappa)
	= \sum_{(\boldsymbol{x},\boldsymbol{\tilde{x}})\in \mathcal{X}_{\neq}^{2k}}  \prod_{i=1}^n \rho^{-1}\delta(x_{\sigma_i^1}- \tilde{x}_{\sigma_{\kappa(i)^2}}) 
	\\
	\times&  \int \mathcal{I}(k_1,\boldsymbol{x}_1,\alpha,\iota,\tfrac{t}{\tau_0};\hbar)  \mathcal{I}(k_2,\boldsymbol{x}_2,\alpha,\iota,\tfrac{\tau-1}{\tau_0}t;\hbar)\varphi(x)\overline{ \mathcal{I}(k_1,\boldsymbol{\tilde{x}}_1,\tilde{\alpha}\iota,\tfrac{t}{\tau_0};\hbar)  \mathcal{I}(k_2,\boldsymbol{\tilde{x}}_2,\tilde{\alpha},\iota,\tfrac{\tau-1}{\tau_0}t;\hbar)\varphi (x)}  \,dx,
	\end{aligned}
\end{equation*}
where again $\mathcal{I}(k_1,\boldsymbol{x}_1,\alpha,\iota,\tfrac{t}{\tau_0};\hbar)$ are defined as $ \mathcal{I}_{0,0}(k_1,\boldsymbol{x}_1,\iota,\tfrac{t}{\tau_0};\hbar)$ for $\alpha$ fixed.  As in the proof of Lemma~\ref{expansion_aver_bound_Mainterm} we consider the case where $\kappa=\mathrm{id}$ and $\kappa\neq\mathrm{id}$ separately. The proof for the first case is analogous to the proof of the same case in Lemma~\ref{expansion_aver_bound_Mainterm}. To see this first note that from the definition of the operators $\mathcal{I}$ we have that
    \begin{equation*}
    	\begin{aligned}
  	 \mathcal{I}(k_1,\boldsymbol{x}_1,\alpha,\iota,\tfrac{t}{\tau_0};\hbar)  \mathcal{I}(k_2,\boldsymbol{x}_2,\alpha,\iota,\tfrac{\tau-1}{\tau_0}t;\hbar)
	 ={}& \int_{[0,\frac{t}{\tau_0}]_{\leq}^{k_1}}   \prod_{m=k_2+1}^{k} \Theta_{\alpha_m}(s_{m-1},{s}_{m},x_m;V,\hbar)\, d\boldsymbol{s}U_{\hbar,0}(-\tfrac{t}{\tau_0})
	\\
	&\times \int_{[0,\frac{\tau-1}{\tau_0}t]_{\leq}^{k_2}}   \prod_{m=1}^{k_2} \Theta_{\alpha_m}(s_{m-1},{s}_{m},x_m;V,\hbar)\, d\boldsymbol{s}U_{\hbar,0}(-\tfrac{\tau-1}{\tau_0}t)
	 \\
	 ={}&    \int_{[\frac{t}{\tau_0},\frac{\tau}{\tau_0}t]_{\leq}^{k_2}} \int_{[0,\frac{t}{\tau_0}]_{\leq}^{k_1}} \prod_{m=1}^{k} \Theta_{\alpha_m}(s_{m-1},{s}_{m},x_m;V,\hbar)\, d\boldsymbol{s}U_{\hbar,0}(-\tfrac{\tau}{\tau_0}t),
	 \end{aligned}
  \end{equation*}
  where we for the last equality we have done the change of variables $s_m \mapsto s_m + \frac{t}{\tau} $ for all $m\in\{1,\dots,k_2\}$. We see that this is almost the same form as the original operator $\mathcal{I}$ expect the split time integral. Hence the proofs will be almost the same with  the difference is in the estimate done in \eqref{EQ:exp_aver_b_Mainterm_5}. For this current case this estimate will be
    \begin{equation*}
	\begin{aligned}
	\MoveEqLeft     \int_{\R^{2k}_{+}} \prod_{i=1}^n 
	\frac{ \boldsymbol{1}_{[0,\frac{t}{\tau_0}]}( \boldsymbol{s}_{1,k_1}^{+}+\hbar \boldsymbol{t}_{1,a_{k_1}}^{+})  
	\boldsymbol{1}_{[0,\frac{t}{\tau_0}]}( \boldsymbol{\tilde{s}}_{1,k_1}^{+}+\hbar \boldsymbol{\tilde{t}}_{1,\tilde{a}_{k_1}}^{+}) }
	 {\max(1,\hbar^{-1} | \boldsymbol{s}_{\sigma^1_i,\sigma^1_{i+1}-1}^{+} - \boldsymbol{\tilde{s}}_{\sigma_i^2,\sigma_{i+1}^2-1}^{+} +\hbar l^1_i(\boldsymbol{t}, \boldsymbol{\tilde{t}})) |)^\frac{d}{2}}  
	 \\
	 &\times \boldsymbol{1}_{[\frac{t}{\tau_0},\frac{\tau t}{\tau_0}]}( \boldsymbol{s}_{k_1+1,k}^{+}+\hbar \boldsymbol{t}_{a_{k_1}+1,a_k}^{+})  
	\boldsymbol{1}_{[\frac{t}{\tau_0},\frac{\tau t}{\tau_0}]}( \boldsymbol{\tilde{s}}_{k_1+1,k}^{+}+\hbar \boldsymbol{\tilde{t}}_{\tilde{a}_{k_1+1},\tilde{a}_k}^{+}) \, d\boldsymbol{\tilde{s}} d\boldsymbol{s}
	\\
	\leq{}& \frac{\hbar^{n} \left(\frac{\tau-1}{\tau_0}t\right)^{2k_2-n_2}  \left(\frac{t}{\tau_0}\right)^{2k_1-n_1}}{k_1!(k_1-n_1)! k_2!(k_2-n_2)!} \left( \int_{\R} \frac{1}{\max(1, | s |)^\frac{d}{2}} \,ds \right)^n 
	\\
	\leq{}& \left(\frac{2d}{d-2}\right)^n  \frac{\hbar^{n} t^{2k-n} \tau_0^{n_1-2k_1}}{k_1!(k_1-n_1)! k_2!(k_2-n_2)!},
	\end{aligned}
\end{equation*} 
where we have used the estimate $\frac{\tau-1}{\tau_0}\leq1$. With this modification one arrives at the estimate
     \begin{equation}\label{EQ:exp_aver_b_er_mainterm_2}
   	\begin{aligned}
	 \Aver{ \mathcal{T}(n,\sigma^1,\sigma^2,\alpha,\tilde{\alpha},\mathrm{id})}
	 \leq  C_d^{a_k+\tilde{a}_k+n}  \norm{\varphi}_{L^2(\R^d)}^2 \left(\norm{\hat{V}}_{1,\infty,3d+3}\right)^{|\alpha|+|\tilde{\alpha}|}  \frac{(\rho t)^{2k-n} \tau_0^{n_2-2k_2}}{k_1!(k_1-n_1)! k_2!(k_2-n_2)!}.
	\end{aligned}
\end{equation}
We now turn to the case when $\kappa\neq\mathrm{id}$. The proof will be analogous to that of the same case in Lemma~\ref{expansion_aver_bound_Mainterm}. Again we will here also denote the smallest $i$ such that $\kappa(i)\neq i$ by $i^{*}$ and we will use the second expressions for the kernels of the operators $\mathcal{I}$ from Observation~\ref{obs_form_I_op_kernel}. Using these we get that
\begin{equation*}
	\begin{aligned}
	\MoveEqLeft \mathcal{T}(n,\sigma^1,\sigma^2,\alpha,\tilde{\alpha},\kappa)
	=  \sum_{(\boldsymbol{x},\boldsymbol{\tilde{x}})\in \mathcal{X}_{\neq}^{2k}}  \frac{1}{(2\pi\hbar)^{d}}\prod_{i=1}^n \rho^{-1} \delta(x_{\sigma_i^1}- \tilde{x}_{\sigma_{\kappa(i)}^2})   \int_{\R_{+}^{|\alpha|+|\tilde{\alpha}|}} \int  \delta(p_k-q_k) 
	\\
	&\times \boldsymbol{1}_{[0,(\tau_1\hbar)^{-1}t]}( \boldsymbol{s}_{1,k_2}^{+}+ \boldsymbol{t}_{1,a_{k_2}}^{+}) \boldsymbol{1}_{[0,(\tau_1\hbar)^{-1}t]}( \boldsymbol{\tilde{s}}_{1,k_2}^{+}+ \boldsymbol{\tilde{t}}_{1,\tilde{a}_{k_2}}^{+})  \boldsymbol{1}_{[0,(\tau_0\hbar)^{-1}t]}( \boldsymbol{s}_{k_2+1,k}^{+}+ \boldsymbol{t}_{a_{k_2}+1,a_{k}}^{+}) 
	\\
	&\times     \boldsymbol{1}_{[0,(\tau_0\hbar)^{-1}t]}( \boldsymbol{\tilde{s}}_{k_2+1,k}^{+}+ \boldsymbol{\tilde{t}}_{\tilde{a}_{k_2}+1,\tilde{a}_{k}}^{+})  e^{i  ((\tau_0\hbar)^{-1}t- \boldsymbol{s}_{k_2+1,k}^{+}- \boldsymbol{t}_{a_{k_2}+1,a_k}^{+}) \frac{1}{2} p_{k_2}^2} e^{-i  ((\tau_0\hbar)^{-1}t- \boldsymbol{\tilde{s}}_{k_2+1,k}^{+}- \boldsymbol{\tilde{t}}_{\tilde{a}_{k_2}+1,\tilde{a}_k}^{+}) \frac{1}{2} q_{k_2}^2} 
	\\
	&\times  \prod_{i=1}^{a_k} e^{i \frac{1}{2} t_{i}  \eta_{i}^2 } \prod_{i=1}^{\tilde{a}_k} e^{-i \frac{1}{2}  \tilde{t}_{i}  \xi_{i}^2}  \Big\{ \prod_{m=1}^k  e^{ -i   \langle  \hbar^{-1/d}x_m,p_{m}-p_{m-1} \rangle}   e^{i  s_{m} \frac{1}{2} p_{m}^2}  \hat{\mathcal{V}}_{\alpha_m}(p_m,p_{m-1},\boldsymbol{\eta} )
	 e^{-i  \tilde{s}_{m} \frac{1}{2} q_{m}^2}  
	 \\
	 &\times e^{ i   \langle  \hbar^{-1/d}\tilde{x}_m,q_{m}-q_{m-1} \rangle} \overline{\hat{\mathcal{V}}_{\tilde{\alpha}_m}(q_m,q_{m-1},\boldsymbol{\xi} ) } \Big\}   
	  e^{i  ((\tau_1\hbar)^{-1}t- \boldsymbol{s}_{1,k_2}^{+}- \boldsymbol{t}_{1,a_{k_2}}^{+}) \frac{1}{2} p_{0}^2} e^{-i  ((\tau_1\hbar)^{-1}t- \boldsymbol{\tilde{s}}_{1,k_2}^{+}- \boldsymbol{\tilde{t}}_{1,\tilde{a}_{k_2}}^{+}) \frac{1}{2} q_{0}^2} 
	\\
	&\times   
	 \hat{\varphi}(\tfrac{p_0}{\hbar})  \hat{\varphi}(\tfrac{q_0}{\hbar})   \,d\boldsymbol{\eta}d\boldsymbol{p} \boldsymbol{\xi}d\boldsymbol{q}  d\boldsymbol{t} d\boldsymbol{s}  d\boldsymbol{\tilde{t}} d\boldsymbol{\tilde{s}},
	\end{aligned}
\end{equation*}
where we have used the notation $\tau_1^{-1}=\frac{\tau-1}{\tau_0}$. The numbers $a_m$ and  $\tilde{a}_m$ are the numbers associated to $\alpha$ and $\tilde{\alpha}$ respectively. Next we introduce three new variables $\tilde{s}_0$, $\tilde{s}_{k+1}$ and depending on $i^{*}$ we introduce $s_0$ or $s_{k+1}$. If $i^*\leq k_2$ we introduce $s_0$ and $s_{k+1}$ otherwise. We will assume $i^*\leq k_2$. The other case is analogous with some change of indices. By introducing these variables we get the following expression
\begin{equation*}
	\begin{aligned}
	\MoveEqLeft \mathcal{T}(n,\sigma^1,\sigma^2,\alpha,\tilde{\alpha},\kappa)
	= \sum_{(\boldsymbol{x},\boldsymbol{\tilde{x}})\in \mathcal{X}_{\neq}^{2k}} \frac{1}{(2\pi\hbar)^{d}}\prod_{i=1}^n \rho^{-1} \delta(x_{\sigma_i^1}- \tilde{x}_{\sigma_{\kappa(i)}^2})   \int_{\R_{+}^{|\alpha|+|\tilde{\alpha}|+3}} \int \delta(p_k-q_k) 
	\\
	&\times \delta(\tfrac{t}{\tau_1\hbar}- \boldsymbol{s}_{0,k_2}^{+}- \boldsymbol{t}_{1,a_{k_2}}^{+})\delta(\tfrac{t}{\tau_1\hbar}- \boldsymbol{\tilde{s}}_{0,k_2}^{+}- \boldsymbol{\tilde{t}}_{1,\tilde{a}_{k_2}}^{+})  \boldsymbol{1}_{[0,(\tau_0\hbar)^{-1}t]}( \boldsymbol{s}_{k_2+1,k}^{+}+ \boldsymbol{t}_{a_{k_2}+1,a_{k}}^{+}) 
	\\
	&\times     \delta(\tfrac{t}{\tau_0\hbar}- \boldsymbol{\tilde{s}}_{k_2+1,k+1}^{+}-\boldsymbol{\tilde{t}}_{\tilde{a}_{k_2}+1,\tilde{a}_{k}}^{+})  e^{i  ((\tau_0\hbar)^{-1}t- \boldsymbol{s}_{k_2+1,k}^{+}- \boldsymbol{t}_{a_{k_2}+1,a_k}^{+}) \frac{1}{2} p_{k_2}^2} e^{-i  \tilde{s}_{k+1} \frac{1}{2} q_{k_2}^2} 
	\\
	&\times  \prod_{i=1}^{a_k} e^{i \frac{1}{2} t_{i}  \eta_{i}^2 } \prod_{i=1}^{\tilde{a}_k} e^{-i \frac{1}{2}  \tilde{t}_{i}  \xi_{i}^2}   \Big\{ \prod_{m=1}^k  e^{ -i   \langle  \hbar^{-1/d}x_m,p_{m}-p_{m-1} \rangle}   e^{i  s_{m} \frac{1}{2} p_{m}^2}  \hat{\mathcal{V}}_{\alpha_m}(p_m,p_{m-1},\boldsymbol{\eta} )
	 e^{-i  \tilde{s}_{m} \frac{1}{2} q_{m}^2}  
	 \\
	 &\times e^{ i   \langle  \hbar^{-1/d}\tilde{x}_m,q_{m}-q_{m-1} \rangle} \overline{\hat{\mathcal{V}}_{\tilde{\alpha}_m}(q_m,q_{m-1},\boldsymbol{\xi} ) } \Big\}   
	  e^{i  s_0 \frac{1}{2} p_{0}^2} e^{-i \tilde{s}_0 \frac{1}{2} q_{0}^2} f(\boldsymbol{s},\boldsymbol{\tilde{s}})  
	 \hat{\varphi}(\tfrac{p_0}{\hbar})  \hat{\varphi}(\tfrac{q_0}{\hbar})   \,d\boldsymbol{\eta}d\boldsymbol{p} \boldsymbol{\xi}d\boldsymbol{q}   d\boldsymbol{t} d\boldsymbol{s}  d\boldsymbol{\tilde{t}} d\boldsymbol{\tilde{s}},
	\end{aligned}
\end{equation*}
where we had inserted the function $f(\boldsymbol{s},\boldsymbol{\tilde{s}})$ before introducing the new variables. The   function $f(\boldsymbol{s},\boldsymbol{\tilde{s}})$ is defined by
\begin{equation*}
	\begin{aligned}
	f(\boldsymbol{s},\boldsymbol{\tilde{s}}) = {}& \boldsymbol{1}_{[0,(\tau_1\hbar)^{-1}t]}\big( \sum_{i=1,i\neq i^{*}}^{k_2} s_i \big) \boldsymbol{1}_{[0,(\tau_1\hbar)^{-1}t]}\big( \sum_{i=1,i \notin \sigma^2}^{k_2} \tilde{s}_i \big) 
	\boldsymbol{1}_{[0,(\tau_0\hbar)^{-1}t]}\big( \sum_{i=k_2+1}^k s_i \big) \boldsymbol{1}_{[0,(\tau_0\hbar)^{-1}t]}\big( \sum_{i=k_2+1,i \notin \sigma^2}^{k} \tilde{s}_i \big).
	\end{aligned}
\end{equation*}
We define the functions $\zeta(t)$ for $t\geq0$ to be
 \begin{equation*}
 	\zeta(t) = \frac{1}{\max(1,t)}.
 \end{equation*}
 We will set $\zeta_i = \zeta((\tau_i\hbar)^{-1}t) $ for $i=0,1$. Using these numbers function we insert the three functions $e^{(\tfrac{t}{\tau_1\hbar}- \boldsymbol{s}_{0,k_2}^{+}- \boldsymbol{t}_{1,a_{k_2}}^{+})\zeta_1}$,   $e^{(\tfrac{t}{\tau_1\hbar}- \boldsymbol{\tilde{s}}_{0,k_2}^{+}- \boldsymbol{\tilde{t}}_{1,\tilde{a}_{k_2}}^{+})\zeta_1}$ and e$^{(\tfrac{t}{\tau_0\hbar}- \boldsymbol{\tilde{s}}_{k_2+1,k+1}^{+}-\boldsymbol{\tilde{t}}_{\tilde{a}_{k_2}+1,\tilde{a}_{k}}^{+})\zeta_0}$ as they are identical $1$ on our domain of integration.
\begin{equation*}
	\begin{aligned}
	\MoveEqLeft \mathcal{T}(n,\sigma^1,\sigma^2,\alpha,\tilde{\alpha},\kappa)
	= \sum_{(\boldsymbol{x},\boldsymbol{\tilde{x}})\in \mathcal{X}_{\neq}^{2k}} \frac{e^{2\frac{t}{\tau_1\hbar}\zeta_1 + \frac{t}{\tau_0\hbar}\zeta_0 }}{(2\pi)^3(2\pi\hbar)^{d}}\prod_{i=1}^n \rho^{-1} \delta(x_{\sigma_i^1}- \tilde{x}_{\sigma_{\kappa(i)}^2})   \int_{\R_{+}^{|\alpha| + |\tilde{\alpha}|+3}} \int \delta(p_k-q_k) 
	\\
	&\times e^{-i\tilde{\nu}_1(\tfrac{t}{\tau_1\hbar}-  \boldsymbol{\tilde{t}}_{1,\tilde{a}_{k_2}}^{+})}  \boldsymbol{1}_{[0,(\tau_0\hbar)^{-1}t]}( \boldsymbol{s}_{k_2+1,k}^{+}+ \boldsymbol{t}_{a_{k_2}+1,a_{k}}^{+}) 
	e^{- (\boldsymbol{t}_{1,a_{k_2}}^{+} +\boldsymbol{\tilde{t}}_{1,\tilde{a}_{k_2}}^{+} )\zeta_1 -\boldsymbol{\tilde{t}}_{\tilde{a}_{k_2}+1,\tilde{a}_{k}}^{+}\zeta_0}
	\\
	&\times e^{-i\nu(\tfrac{t}{\tau_1\hbar}- \boldsymbol{t}_{1,a_{k_2}}^{+})}     e^{-i\tilde{\nu}_0(\tfrac{t}{\tau_0\hbar}- \boldsymbol{\tilde{t}}_{\tilde{a}_{k_2}+1,\tilde{a}_{k}}^{+})}  e^{i  ((\tau_0\hbar)^{-1}t- \boldsymbol{s}_{k_2+1,k}^{+}- \boldsymbol{t}_{a_{k_2}+1,a_k}^{+}) \frac{1}{2} p_{k_2}^2} e^{-i  \tilde{s}_{k+1} (\frac{1}{2} q_{k_2}^2-\tilde{\nu}_0-i\zeta_0)} 
	\\
	&\times \prod_{i=1}^{a_k} e^{i \frac{1}{2} t_{i}  \eta_{i}^2 } \prod_{i=1}^{\tilde{a}_k} e^{-i \frac{1}{2} \tilde{t}_{i}  \xi_{i}^2}    \prod_{m=0}^{k_2} e^{i  s_{m}( \frac{1}{2} p_{m}^2+\nu + i\zeta_1)}  e^{-i  \tilde{s}_{m}( \frac{1}{2} q_{m}^2-\tilde{\nu}_1-i\zeta_1)}  \prod_{m=k_2+1}^{k}e^{i  s_{m} \frac{1}{2} p_{m}^2}  e^{-i  \tilde{s}_{m}( \frac{1}{2} q_{m}^2-\tilde{\nu_0}-i\zeta_0)}  
 	\\
	&\times   \Big\{ \prod_{m=1}^k  e^{ -i   \langle  \hbar^{-1/d}x_m,p_{m}-p_{m-1} \rangle}     \hat{\mathcal{V}}_{\alpha_m}(p_m,p_{m-1},\boldsymbol{\eta} )
	 e^{ i   \langle  \hbar^{-1/d}\tilde{x}_m,q_{m}-q_{m-1} \rangle} \overline{\hat{\mathcal{V}}_{\tilde{\alpha}_m}(q_m,q_{m-1},\boldsymbol{\xi} ) } \Big\}   
	\\
	&\times f(\boldsymbol{s},\boldsymbol{\tilde{s}})  
	 \hat{\varphi}(\tfrac{p_0}{\hbar})  \hat{\varphi}(\tfrac{q_0}{\hbar})   \, d\nu d\tilde{\nu}_1 d\tilde{\nu}_0 d\boldsymbol{\eta}d\boldsymbol{p} \boldsymbol{\xi}d\boldsymbol{q}  d\boldsymbol{t} d\boldsymbol{s} d\boldsymbol{\tilde{t}} d\boldsymbol{\tilde{s}},
	\end{aligned}
\end{equation*}
where we have also written the three delta functions as Fourier transforms of $1$. From here the proof will proceed analogous to the proof of the similar case in Lemma~\ref{expansion_aver_bound_Mainterm}. We first apply Lemma~\ref{LE:Exp_ran_phases}, then we do the argument where we divide into cases depending on the size of the $t$'s and $\tilde{t}$'s. After this we evaluate the integrals in $s_0$, $s_{\sigma^1_{i^{*}}}$, all $\tilde{s}$ with an index in $\sigma^2$, all $p$ and all $q$ that does not have an index in $\sigma^2$. After these operations we obtain the estimate
   \begin{equation}\label{EQ:exp_aver_b_er_mainterm_2.1}
   	\begin{aligned}
	 \Aver{ \mathcal{T}(n,\sigma^1,\sigma^2,\alpha,\tilde{\alpha},\kappa)}
	 \leq {}& C   \frac{ \left(\rho\hbar(2\pi)^d\right)^{2k-n}}{(2\pi)^3(2\pi\hbar)^{d}}  \sum_{l=1}^{2^{a_{k}+\tilde{a}_{k}}}   \sum_{|\epsilon_1|,\dots,|\epsilon_{j_1}|\leq d+1} \sum_{|\nu_1|,\dots,|\nu_{\tilde{j}_1}|\leq d+1}   \int_{\R^{a_{k}+\tilde{a}_{k}}_{+}}
	  \boldsymbol{1}_{B_l}(\boldsymbol{t},\boldsymbol{\tilde{t}}) 
	  \\
	  &\times \prod_{i\in J_1} \frac{C(\epsilon_i)}{(2\pi t_i))^\frac{d}{2} } \prod_{i\in \tilde{J}_1} \frac{C(\nu_i)   }{(2\pi \tilde{t}_i)^\frac{d}{2}}  |\mathcal{I}_n(\boldsymbol{t},\boldsymbol{\tilde{t}}) |  d\boldsymbol{t} d\boldsymbol{\tilde{t}} ,
	\end{aligned}
\end{equation}
where we have used that $e^{2\frac{t}{\tau_1\hbar}\zeta_1 + \frac{t}{\tau_0\hbar}\zeta_0 }\leq C$ and the numbers $\mathcal{I}_n(\boldsymbol{t},\boldsymbol{\tilde{t}}) $ can be estimated by
\begin{equation}\label{EQ:exp_aver_b_er_mainterm_3}
	\begin{aligned}
	\MoveEqLeft \big| \mathcal{I}_n(\boldsymbol{t},\boldsymbol{\tilde{t}})\big|
	\leq \int  \prod_{i\in J_1} \frac{| x_i ^{\epsilon_i}| }{(1+|x_i |^2)^{d+1}} \,d\boldsymbol{x} \int \prod_{i\in \tilde{J}_1} \frac{| y_i ^{\nu_i}| }{(1+|y_i |^2)^{d+1}}  \,d\boldsymbol{y}    \int_{\R_{+}^{2k-n-1}} f(\boldsymbol{s},\boldsymbol{\tilde{s}})  d\boldsymbol{s}  d\boldsymbol{\tilde{s}}
	\\
	&\times      \int |\hat{\varphi}(\tfrac{q_0}{\hbar}) |^2   \frac{1}{|\frac{1}{2} q_0^2-\tilde{\nu}_1-i\zeta_1|}\frac{1}{|\frac{1}{2} q_0^2+\nu +i\zeta_1|}  \frac{1}{|\frac{1}{2} q_{k_1}^2-\tilde{\nu}_2-i\zeta_0|}
	  \\
	  &\times    \frac{1}{| \frac{1}{2} (q_{\kappa(i^{*})}+q_{i^{*}-1}-q_{\kappa(i^{*})-1})^2 +\nu+i\zeta|}    \prod_{m=1}^{n_2} \frac{1}{ |\frac{1}{2} q_{m}^2-\tilde{\nu}_1-i\zeta_1|}  \prod_{m=n_2+1}^n \frac{1}{ |\frac{1}{2} q_{m}^2-\tilde{\nu}_2-i\zeta_0|} 
	\\
	&\times 
	\Big| \partial_{\boldsymbol{\xi}_{\tilde{J}_1}}^{\boldsymbol{\nu}} \partial_{\boldsymbol{\eta}_{J_1}}^{\boldsymbol{\epsilon}} \mathcal{G}(\boldsymbol{q},\boldsymbol{\eta},\boldsymbol{\xi},\sigma^1,\sigma^2,\alpha,\tilde{\alpha}) \Big| 
	  \,d\boldsymbol{\eta}d\boldsymbol{\xi}d\boldsymbol{q} d\nu d\tilde{\nu}_1d\tilde{\nu}_0,
	\end{aligned}
\end{equation}
where
\begin{equation*}
	\begin{aligned}
	 \mathcal{G}(\boldsymbol{q},\boldsymbol{\eta},\boldsymbol{\xi},\sigma^1,\sigma^2,\alpha,\tilde{\alpha}) ={}&  \prod_{i=1}^n  \hat{\mathcal{V}}_{\alpha_{\sigma^1_i}}(q_{\kappa(i)}+l_{i}^\kappa(\boldsymbol{q}),q_{\kappa(i)-1}+l_{i}^\kappa(\boldsymbol{q}),\boldsymbol{\eta})  \overline{\hat{\mathcal{V}}_{\tilde{\alpha}_{\sigma^1_i}}(q_i,q_{i-1},\boldsymbol{\xi})}   
	\\
	&\times \prod_{i=0}^n  \prod_{m=\sigma^2_i +1}^{\sigma^2_{i+1} -1} \overline{\hat{\mathcal{V}}_{\tilde{\alpha}_m}(q_i,q_i,\boldsymbol{\xi})}
	  \prod_{m=\sigma^1_i +1}^{\sigma^1_{i+1} -1} \hat{\mathcal{V}}_{\alpha_m}(q_{\kappa(i)}+l_{i}^\kappa(\boldsymbol{q}),q_{\kappa(i)}+l_{i}^\kappa(\boldsymbol{q}),\boldsymbol{\eta})  .
	\end{aligned}
\end{equation*}
Using the definition of the function $f$ we get that
\begin{equation*}
	\begin{aligned}
	     \int_{\R_{+}^{2k-n-1}} f(\boldsymbol{s},\boldsymbol{\tilde{s}})  d\boldsymbol{s}_{1,k-1}   d\boldsymbol{\tilde{s}}_{1,k-n} = \frac{(\hbar^{-1} t)^{2k-n-1} \tau_0^{n_1-2k_1}}{(k_1-1)!(k_1-n_1)! k_2!(k_2-n_2)!}  .
	\end{aligned}
\end{equation*}
Since we in general could have that $i^{*}> k_2$ we will instead use the estimate
\begin{equation}\label{EQ:exp_aver_b_er_mainterm_4}
	\begin{aligned}
	     \int_{\R_{+}^{2k-n-1}} f(\boldsymbol{s},\boldsymbol{\tilde{s}})  d\boldsymbol{s}_{1,k-1}   d\boldsymbol{\tilde{s}}_{1,k-n} \leq \frac{(\hbar^{-1} t)^{2k-n-1} \tau_0^{n_1+1-2k_1}}{(k_1-1)!(k_1-n_1)! (k_2-1)!(k_2-n_2)!}  ,
	\end{aligned}
\end{equation}
as it cover both cases. In order to estimate the integrals in $q$, $\eta$ and $\xi$ we insert the fraction
  \begin{equation*}
  	\frac{\langle q_0\rangle^{5d+5}}{\langle q_0\rangle^{5d+5}} \frac{\langle q_{\kappa(i^{*})}+q_{i^{*}-1}-q_{\kappa(i^{*})-1}\rangle^{d+1}}{\langle q_{\kappa(i^{*})}+q_{i^{*}-1}-q_{\kappa(i^{*})-1}\rangle^{d+1}} \prod_{i=1}^n \frac{\langle q_{i}-q_{i-1}\rangle^{6d+6}}{\langle  q_{i}-q_{i-1}\rangle^{6d+6}}. 
  \end{equation*}
Note here that we have additional $d+1$ of $\langle q_0\rangle$ and  $\langle  q_{i}-q_{i-1}\rangle$. This is due to the additional $\tilde{\nu}_2$ compared to the the proof of Lemma~\ref{expansion_aver_bound_Mainterm}. We have that
  \begin{equation}\label{EQ:exp_aver_b_er_mainterm_5}
	\begin{aligned}
	\MoveEqLeft  \sup_{\boldsymbol{q}}   \int  \langle q_{{\kappa(i^{*})}}-q_{\kappa(i^{*})-1}\rangle^{d+1} \prod_{i=1}^n \langle q_{i}-q_{i-1}\rangle^{6d+6}  	 \Big| \partial_{\boldsymbol{\xi}_{\tilde{J}_1}}^{\boldsymbol{\nu}} \partial_{\boldsymbol{\eta}_{J_1}}^{\boldsymbol{\epsilon}} \mathcal{G}(\boldsymbol{q},\boldsymbol{\eta},\boldsymbol{\xi},\sigma^1,\sigma^2,\alpha,\tilde{\alpha})  \Big| \,d\boldsymbol{\eta}d\boldsymbol{\xi}
	 \\
	 \leq {}& C^{a_k+\tilde{a}_k}  \left(\norm{\hat{V}}_{1,\infty,4d+4}\right)^{|\alpha|+|\tilde{\alpha}|} .
	\end{aligned}
\end{equation}
Next we have that 
    \begin{equation*}
  	\frac{\langle q_{\kappa(i^{*})}+q_{i^{*}-1}-q_{\kappa(i^{*})-1}\rangle^{d+1} \langle q_{k'}\rangle^{d+1} \langle q_{\kappa(i^{*})-1}\rangle^{d+1} \langle q_{\kappa(i^{*})}\rangle^{d+1} \langle q_{\kappa(i^{*})+1}\rangle^{d+1}}{\langle q_{\kappa(i^{*})}-q_{\kappa(i^{*})-1}\rangle^{d+1}\langle q_0\rangle^{5d+5}}  \prod_{i=1}^n \frac{1}{\langle  q_{i}-q_{i-1}\rangle^{5d+5}} \leq C,
  \end{equation*}
  where the index $k'$ is chosen to be $1$ if $\kappa(i^{*})> k_1$ and $k_1+1$ if  $\kappa(i^{*}) \leq k_1$. We only have to consider the integral
  \begin{equation*}
	\begin{aligned}
	\MoveEqLeft   \int  \frac{1}{\langle\tilde{\nu}_1\rangle|\frac{1}{2} q_0^2-\tilde{\nu}_1-i\zeta_1|}\frac{1}{\langle{\nu}\rangle|\frac{1}{2} q_0^2+\nu +i\zeta_1|}
	 \frac{1}{\langle\tilde{\nu}_2\rangle|\frac{1}{2} q_{k_1}^2-\tilde{\nu}_2-i\zeta_0|}
	   \\
	   &\times  \frac{\langle{\nu}\rangle \langle q_{\kappa(i^{*})}+q_{i^{*}-1}-q_{\kappa(i^{*})-1}\rangle^{-d-1}}{| \frac{1}{2} (q_{\kappa(i^{*})}+q_{i^{*}-1}-q_{\kappa(i^{*})-1})^2 +\nu+i\zeta_1|}  \frac{\langle\tilde{\nu}_1\rangle  \langle\tilde{\nu}_2\rangle }{ \langle q_{k'}\rangle^{d+1} \langle q_{\kappa(i^{*})-1}\rangle^{d+1} \langle q_{\kappa(i^{*})}\rangle^{d+1} \langle q_{\kappa(i^{*})+1}\rangle^{d+1}}
	  \\
	  &\times    \prod_{m=1}^{n_2} \frac{\langle  q_{m}-q_{m-1}\rangle^{-d-1}}{ |\frac{1}{2} q_{m}^2-\tilde{\nu}_1-i\zeta_1|} \prod_{m=n_2+1}^n \frac{\langle  q_{m}-q_{m-1}\rangle^{-d-1}}{ |\frac{1}{2} q_{m}^2-\tilde{\nu}_2-i\zeta_0|}
	\, d\boldsymbol{q}_{1,n} d\nu d\tilde{\nu}_1d\tilde{\nu}_2.
	\end{aligned}
\end{equation*}
  Applying the estimate $\langle q_{\kappa(i^{*}) -1 } - q_{\kappa(i^{*}) -2 } \rangle^{-d-1}\langle q_{\kappa(i^{*})  } - q_{\kappa(i^{*}) -1 } \rangle^{-d-1} \langle q_{\kappa(i^{*}) +1 } - q_{\kappa(i^{*}) } \rangle^{-d-1}\leq1$ we can use Lemma~\ref{LE:est_res_combined} with $q_{\kappa(i^{*})}$ as ``$p$'' and $ q_{\kappa(i^{*})-1}- q_{i^{*}-1}$ as ``$q$''. After applying this lemma we integrate $q_{\kappa(i^{*})-1}$ by using Lemma~\ref{LE:resolvent_int_est}. For all the remaining integrals we also apply Lemma~\ref{LE:resolvent_int_est} repeatedly. This gives us the estimate    
\begin{equation}\label{EQ:exp_aver_b_er_mainterm_6}
	\begin{aligned}
	\MoveEqLeft   \int   \frac{1}{\langle\tilde{\nu}_1\rangle|\frac{1}{2} q_0^2-\tilde{\nu}_1-i\zeta_1|}\frac{1}{\langle{\nu}\rangle|\frac{1}{2} q_0^2+\nu +i\zeta_1|}
	 \frac{1}{\langle\tilde{\nu}_2\rangle|\frac{1}{2} q_{k_1}^2-\tilde{\nu}_2-i\zeta_0|}
	   \\
	   &\times  \frac{\langle{\nu}\rangle \langle q_{\kappa(i^{*})}+q_{i^{*}-1}-q_{\kappa(i^{*})-1}\rangle^{-d-1}}{| \frac{1}{2} (q_{\kappa(i^{*})}+q_{i^{*}-1}-q_{\kappa(i^{*})-1})^2 +\nu+i\zeta_1|}  \frac{\langle\tilde{\nu}_1\rangle  \langle\tilde{\nu}_2\rangle }{ \langle q_{k'}\rangle^{d+1}\langle q_{\kappa(i^{*})}\rangle^{d+1} \langle q_{\kappa(i^{*})+1}\rangle^{d+1}}
	  \\
	  &\times    \prod_{m=1}^{n_2} \frac{\langle  q_{m}-q_{m-1}\rangle^{-d-1}}{ |\frac{1}{2} q_{m}^2-\tilde{\nu}_1-i\zeta_1|} \prod_{m=n_2+1}^n \frac{\langle  q_{m}-q_{m-1}\rangle^{-d-1}}{ |\frac{1}{2} q_{m}^2-\tilde{\nu}_2-i\zeta_0|}
	\, d\boldsymbol{q} d\nu d\tilde{\nu}_1d\tilde{\nu}_2 
	\\
	\leq{}& C^n |\log(\zeta_1)|^{n_2+3} |\log(\zeta_0)|^{n_1+1}.
	\end{aligned}
\end{equation}
Moreover we have that
\begin{equation}\label{EQ:exp_aver_b_er_mainterm_7}
	\int \langle q_0 \rangle^{5d+5} |\hat{\varphi}(\tfrac{q_0}{\hbar}) |^2 \, dq_0 \leq \hbar^d C  \norm{\varphi}_{\mathcal{H}^{3d+3}_\hbar(\R^d)}^2.
\end{equation}
Combing the estimates in \cref{EQ:exp_aver_b_er_mainterm_2.1,EQ:exp_aver_b_er_mainterm_3,EQ:exp_aver_b_er_mainterm_4,EQ:exp_aver_b_er_mainterm_5,EQ:exp_aver_b_er_mainterm_6,EQ:exp_aver_b_er_mainterm_7} we get the estimate
   \begin{equation}\label{EQ:exp_aver_b_er_mainterm_8}
   	\begin{aligned}
	\MoveEqLeft \Aver{ \mathcal{T}(n,\sigma^1,\sigma^2,\alpha,\tilde{\alpha},\kappa)}
	\\
	& \leq   C_d^{a_k+\tilde{a}_k}    \norm{\hat{V}}_{1,\infty,4d+4}^{|\alpha|+|\tilde{\alpha}|}  \tau_0^{n_1+1-2k_1}   \frac{\rho(\rho t)^{2k-n-1} \hbar |\log(\zeta_1)|^{n_2+3} |\log(\zeta_0)|^{n_1+1}}{(k_1-1)!(k_1-n_1)! (k_2-1)!(k_2-n_2)!}  \norm{\varphi}_{\mathcal{H}^{3d+3}_\hbar(\R^d)}^2.
	\end{aligned}
\end{equation}
We observe that
\begin{equation*}
	\begin{aligned}
	&\sum_{\substack{n_1+n_2=n \\ n_1\leq k_1,n_2\leq k_2}} \sum_{\sigma^1\in\tilde{\mathcal{A}}(k_1,n_1,k_2,n_2)} 1\leq 2^k
	\\
	&\sum_{\substack{n_1+n_2=n \\ n_1\leq k_1,n_2\leq k_2}} \sum_{\sigma^1\in\tilde{\mathcal{A}}(k_1,n_1,k_2,n_2)} \frac{1}{(k_1-n_1)! (k_2-n_2)!}\leq \frac{2^k}{(k_1+k_2-n)!}.
	\end{aligned}
\end{equation*}
Using these observations in combination with \cref{EQ:exp_aver_b_er_mainterm_1,EQ:exp_aver_b_er_mainterm_2,EQ:exp_aver_b_er_mainterm_8} we get that
\begin{equation}
	\begin{aligned}
	\MoveEqLeft \big\lVert \sum_{(\boldsymbol{x}_1,\boldsymbol{x}_2)\in \mathcal{X}_{\neq}^{k}}  \mathcal{I}_{0,0}(k_1,\boldsymbol{x}_1,\iota,\tfrac{t}{\tau_0};\hbar)  \mathcal{I}_{0,0}(k_2,\boldsymbol{x}_2,\iota,\tfrac{\tau-1}{\tau_0}t;\hbar)\varphi\big\rVert_{L^2(\R^d)}^2 
	\\
	\leq{}& \sum_{\alpha,\tilde{\alpha}\in\N^k}  (\lambda C_d    \norm{\hat{V}}_{1,\infty,4d+4})^{|\alpha|+|\tilde{\alpha}|}  \sum_{n=0}^k \Big[   \frac{(\rho t)^{2k-n}}{\tau_0^{k_1} k_1! k_2! (k-n)! } \norm{\varphi}_{L^2(\R^d)}^2
	\\
	&  +  \hbar |\log(\tfrac{\hbar\tau_0}{t})|^{n+4} \frac{\rho(\rho t)^{2k-n-1} n!}{\tau_0^{k_1} (k_1-1)! (k_2-1)!(k-n)!}  \norm{\varphi}_{\mathcal{H}^{3d+3}_\hbar(\R^d)}^2 \Big].
	\end{aligned}
\end{equation}
Combining this with \eqref{EQ:Aver_full_ex_0,0_dob1} and using our assumptions on the single site potential $V$ and the coupling constant $\lambda$ we obtain that
  \begin{equation}
	\begin{aligned}
	\MoveEqLeft \mathbb{E}\Big[ \big\lVert \sum_{k_1=1}^{k_0} \sum_{k_2=k_0-k_1+1}^{k_0} 
	  \sum_{(\boldsymbol{x}_1,\boldsymbol{x}_2)\in \mathcal{X}_{\neq}^{k_1+k_2}}  \mathcal{I}_{0,0}(k_1,\boldsymbol{x}_1,\iota,\tfrac{t}{\tau_0};\hbar)  \mathcal{I}_{0,0}(k_2,\boldsymbol{x}_2,\iota,\tfrac{\tau-1}{\tau_0}t;\hbar)\varphi\big\rVert_{L^2(\R^d)}^2 \Big]
	  \\
	   &
	   \leq \frac{k_0^2}{\tau_0} \sum_{k=k_0+1}^{2k_0}  C^k   \sum_{n=0}^k \Big[   \frac{(\rho t)^{2k-n}}{ k! (k-n)! } \norm{\varphi}_{L^2(\R^d)}^2
	 +  \hbar |\log(\tfrac{\hbar\tau_0}{t})|^{n+4} \frac{\rho(\rho t)^{2k-n-1} n!}{ (k-2)! (k-n)!}  \norm{\varphi}_{\mathcal{H}^{3d+3}_\hbar(\R^d)}^2 \Big],
	 \end{aligned}
\end{equation}
where we have used that 
\begin{equation}
	\sum_{k=k_0+1}^{2k_0}  \sum_{k_1+k_2=k}  \frac{1}{k_1! k_2 !} +  \frac{1}{(k_1-1)! (k_2-1) !} = \sum_{k=k_0+1}^{2k_0}  \frac{2^k}{k!} + \frac{2^{k-2}}{(k-2)!},
\end{equation}
since we assume $k_0>1$. Using that $\frac{n!}{(k-2)!}\leq k^2$, evaluating the sum over $n$, maximising each term in the sum over $k$ by plugging in $k_0+1$ or $2k_0$ and then estimating the sum over $k$ by $k_0$ (the number of terms in the sum) we obtain the desired estimate. 
  \end{proof}
What remains to estimate is the terms we get when we start the second expansion with the same point we ended with. Here we have the following result
\begin{lemma}\label{expansion_aver_bound_Mainterm_3}
  Assume we are in the setting of Definition~\ref{def_remainder_k_0}. Let $\varphi \in \mathcal{H}^{3d+3}_\hbar(\R^d)$. Then for any
  $\tau\in\{2,\dots,\tau_0\}$
    \begin{equation*}
	\begin{aligned}
	\MoveEqLeft \mathbb{E}\Big[ \big\lVert \sum_{k_1=1}^{k_0} \sum_{k_2=k_0-k_1+1}^{k_0} 
	  \sum_{(\boldsymbol{x}_1,\boldsymbol{x}_2)\in \mathcal{X}_{\neq}^{k_1+k_2}}  \mathcal{I}_{0,0}(k_1+1,(x_{2,k_2},\boldsymbol{x}_1),\iota,\tfrac{t}{\tau_0};\hbar)  \mathcal{I}_{0,0}(k_2,\boldsymbol{x}_2,\iota,\tfrac{\tau-1}{\tau_0}t;\hbar) \varphi\big\rVert_{L^2(\R^d)}^2 \Big]
	  \\
	   &
	   \leq  \frac{k_0^3}{\tau_0} \frac{C^{k_0}}{ (k_0+1)!} \norm{\varphi}_{L^2(\R^d)}^2
	 +  \frac{k_0^5\hbar}{\tau_0} C^{k_0} |\log(\tfrac{\hbar\tau_0}{t})|^{2k_0+4}   \norm{\varphi}_{\mathcal{H}^{3d+3}_\hbar(\R^d)}^2,
	 \end{aligned}
\end{equation*}
where the constant $C$ depends on $\rho$, $t$, the single site potential $V$ and the coupling constant $\lambda$. In particular we have that the function is in  $L^2(\R^d)$ $\Pro$-almost surely. 
  \end{lemma}
\begin{proof}
As In the proof of Lemma~\ref{expansion_aver_bound_Mainterm_2} we have that
  \begin{equation}\label{EQ:Aver_full_ex_same_0,0_dob1}
	\begin{aligned}
	\MoveEqLeft  \big\lVert \sum_{k_1=1}^{k_0} \sum_{k_2=k_0-k_1+1}^{k_0} 
	  \sum_{(\boldsymbol{x}_1,\boldsymbol{x}_2)\in \mathcal{X}_{\neq}^{k_1+k_2}}  \mathcal{I}_{0,0}(k_1+1,(x_{2,k_2},\boldsymbol{x}_1),\iota,\tfrac{t}{\tau_0};\hbar)  \mathcal{I}_{0,0}(k_2,\boldsymbol{x}_2,\iota,\tfrac{\tau-1}{\tau_0}t;\hbar)\varphi\big\rVert_{L^2(\R^d)}^2 
	  \\
	  &
	   \leq k_0^2 \sum_{k=k_0+1}^{2k_0}  \sum_{k_1+k_2=k} 
	   \big\lVert \sum_{(\boldsymbol{x}_1,\boldsymbol{x}_2)\in \mathcal{X}_{\neq}^{k}}  \mathcal{I}_{0,0}(k_1+1,(x_{2,k_2},\boldsymbol{x}_1),\iota,\tfrac{t}{\tau_0};\hbar)  \mathcal{I}_{0,0}(k_2,\boldsymbol{x}_2,\iota,\tfrac{\tau-1}{\tau_0}t;\hbar)\varphi\big\rVert_{L^2(\R^d)}^2.
	 \end{aligned}
\end{equation}
We now fix $k=k_1+k_2$. As in the other proofs we get two sums over the set $\mathcal{X}_{\neq}^{k}$ when we take the $L^2$-norm. Again we will divide this double sum into cases depending on how many points they have in common. As before we also need to keep track of how these points is chosen in relation to the two operators. With this in mind we have that
\begin{equation}\label{EQ:Aver_full_ex_same_0,0_dob2}
	\begin{aligned}
	\MoveEqLeft \big\lVert \sum_{(\boldsymbol{x}_1,\boldsymbol{x}_2)\in \mathcal{X}_{\neq}^{k}}  \mathcal{I}_{0,0}(k_1+1,(x_{2,k_2},\boldsymbol{x}_1),\iota,\tfrac{t}{\tau_0};\hbar)  \mathcal{I}_{0,0}(k_2,\boldsymbol{x}_2,\iota,\tfrac{\tau-1}{\tau_0}t;\hbar)\varphi\big\rVert_{L^2(\R^d)}^2 = \sum_{\alpha,\tilde{\alpha}\in\N^{k+1}} (i\lambda)^\alpha(-i\lambda)^{\tilde{\alpha}}
	\\
	&\times \sum_{n=0}^k \sum_{\substack{n_1+n_2=n \\ n_1\leq k_1+1,n_2\leq k_2-1}} \sum_{\substack{\tilde{n}_1+\tilde{n}_2=n \\ \tilde{n}_1\leq k_1+1,\tilde{n}_2\leq k_2-1}} \sum_{\substack{\sigma^1\in\tilde{\mathcal{A}}(k_1+1,n_1,k_2-1,n_2)\\  \sigma^2\in\tilde{\mathcal{A}}(k_1+1,\tilde{n}_1,k_2-1,\tilde{n}_2)}}  \sum_{\kappa\in\mathcal{S}_n} \mathcal{T}(n,\sigma^1,\sigma^2,\alpha,\tilde{\alpha},\kappa),
	\end{aligned}
\end{equation}
where the sets $\tilde{\mathcal{A}}(k_1,n_1,k_2,n_2)$ is defined in\eqref{orderset_def_2}.
The numbers $\mathcal{T}(n,\sigma^1,\sigma^2,\kappa)$ are given by  
\begin{equation*}
	\begin{aligned}
	\MoveEqLeft  \mathcal{T}(n,\sigma^1,\sigma^2,\alpha,\tilde{\alpha},\kappa)
	= \sum_{(\boldsymbol{x},\boldsymbol{\tilde{x}})\in \mathcal{X}_{\neq}^{2k}} \prod_{i=1}^n \rho^{-1}\delta(x_{\sigma_i^1}- \tilde{x}_{\sigma_{\kappa(i)^2}}) \int \mathcal{I}(k_1+1,(x_{2,k_2},\boldsymbol{x}_1),\alpha,\iota,\tfrac{t}{\tau_0};\hbar) 
	\\
	\times&   \mathcal{I}(k_2,\boldsymbol{x}_2,\alpha,\iota,\tfrac{\tau-1}{\tau_0}t;\hbar)\varphi(x)\overline{ \mathcal{I}(k_1+1,(\tilde{x}_{2,k_2},\boldsymbol{\tilde{x}}_1),\tilde{\alpha},\iota,\tfrac{t}{\tau_0};\hbar)  \mathcal{I}(k_2,\boldsymbol{\tilde{x}}_2,\tilde{\alpha},\iota,\tfrac{\tau-1}{\tau_0}t;\hbar)\varphi (x)}  \,dx,
	\end{aligned}
\end{equation*}
where again $\mathcal{I}(k_1+1,(x_{2,k_2},\boldsymbol{x}_1),\alpha,\iota,\tfrac{t}{\tau_0};\hbar)$ are defined as $ \mathcal{I}_{0,0}(k_1+1,(x_{2,k_2},\boldsymbol{x}_1),\iota,\tfrac{t}{\tau_0};\hbar)$ for $\alpha$ fixed. Again we consider the case where $\kappa=\mathrm{id}$ and $\kappa\neq\mathrm{id}$ separately. The proof for the first case is again almost analogous to the proof of the same case in Lemma~\ref{expansion_aver_bound_Mainterm}. To see this first note that from the definition of the operators $\mathcal{I}$ we have that
 \begin{equation*}
    	\begin{aligned}
  	\MoveEqLeft \mathcal{I}(k_1+1,(x_{2,k_2},\boldsymbol{x}_1),\tilde{\alpha}\iota,\tfrac{t}{\tau_0};\hbar)  \mathcal{I}(k_2,\boldsymbol{\tilde{x}}_2,\tilde{\alpha},\iota,\tfrac{\tau-1}{\tau_0}t;\hbar)
	\\
	 ={}& \int_{[0,\frac{t}{\tau_0}]_{\leq}^{k_1+1}}   \prod_{m=k_2+2}^{k+1} \Theta_{\alpha_m}(s_{m-1},{s}_{m},x_{m-1};V,\hbar)  \Theta_{\alpha_{k_2+1}}(s_{k_2},{s}_{k_2+1},x_{k_2};V,\hbar)\, d\boldsymbol{s}U_{\hbar,0}(-\tfrac{t}{\tau_0})
	\\
	&\times \int_{[0,\frac{\tau-1}{\tau_0}t]_{\leq}^{k_2}}   \prod_{m=1}^{k_2} \Theta_{\alpha_m}(s_{m-1},{s}_{m},x_m;V,\hbar)\, d\boldsymbol{s}U_{\hbar,0}(-\tfrac{\tau-1}{\tau_0}t)
	 \\
	 ={}&    \int_{[\frac{t}{\tau_0},\frac{\tau}{\tau_0}t]_{\leq}^{k_2}} \int_{[0,\frac{t}{\tau_0}]_{\leq}^{k_1+1}} \prod_{m=k_2+2}^{k+1} \Theta_{\alpha_m}(s_{m-1},{s}_{m},x_{m-1};V,\hbar) \Theta_{\alpha_{k_2+1}}(s_{k_2},{s}_{k_2+1},x_{k_2};V,\hbar)
	 \\&\times  \Theta_{\alpha_{k_2}}(s_{k_2-1},{s}_{k_2},x_{k_2};V,\hbar) \prod_{m=1}^{k_2-1} \Theta_{\alpha_m}(s_{m-1},{s}_{m},x_m;V,\hbar)\, d\boldsymbol{s}U_{\hbar,0}(-\tfrac{\tau}{\tau_0}t),
	 \end{aligned}
  \end{equation*}
 where we for the last equality have done the change of variables $s_m \mapsto s_m + \frac{t}{\tau} $ for all $m\in\{1,\dots,k_2\}$. We see that this is almost the same form as the original operator $\mathcal{I}$ expect the split time integral and we have two operators depending on the same position next to each other. %
 Since all ``dependence'' of $s_{k_2}$ is contained in the kernel of $\Theta_{\alpha_{k_2+1}}(s_{k_2},{s}_{k_2+1},x_{k_2};V,\hbar)\Theta_{\alpha_{k_2}}(s_{k_2-1},{s}_{k_2},x_{k_2};V,\hbar)$ we make the change of variables $s_{k_2}\mapsto s_{k_2} -s_{k_2+1}$ this yields
 \begin{equation*}
    	\begin{aligned}
  	\MoveEqLeft \int_{[\frac{t}{\tau_0},\frac{\tau}{\tau_0}t]_{\leq}^{k_2}} \int_{[0,\frac{t}{\tau_0}]_{\leq}^{k_1+1}} \prod_{m=k_2+2}^{k+1} \Theta_{\alpha_m}(s_{m-1},{s}_{m},x_{m-1};V,\hbar) \Theta_{\alpha_{k_2+1}}(s_{k_2},{s}_{k_2+1},x_{k_2};V,\hbar)
	 \\&\times  \Theta_{\alpha_{k_2}}(s_{k_2-1},{s}_{k_2},x_{k_2};V,\hbar) \prod_{m=1}^{k_2-1} \Theta_{\alpha_m}(s_{m-1},{s}_{m},x_m;V,\hbar)\, d\boldsymbol{s}U_{\hbar,0}(-\tfrac{\tau}{\tau_0}t)
	 \\
	 ={}& \int_{\R_{+}^{k+1}} \boldsymbol{1}_{[\frac{t}{\tau_0},\frac{\tau}{\tau_0}t]_{\leq}}(s_{k_2} + s_{k_2+1}, s_{k_2-1},\dots,s_1) \boldsymbol{1}_{[0,\frac{t}{\tau_0}]_{\leq}}( s_{k+1}, s_{k},\dots,s_{k_2+1}) 
	 \\
	 &\times \prod_{m=k_2+2}^{k+1} \Theta_{\alpha_m}(s_{m-1},{s}_{m},x_{m-1};V,\hbar)
	 \tilde{\Theta}_{\alpha_{k_2+1}+\alpha_{k_2}}(s_{k_2-1},s_{k_2},{s}_{k_2+1},x_{k_2};V,\hbar)
	 \\
	 &\times  \prod_{m=1}^{k_2-1} \Theta_{\alpha_m}(s_{m-1},{s}_{m},x_m;V,\hbar)\,  d\boldsymbol{s}U_{\hbar,0}(-\tfrac{\tau}{\tau_0}t),
	 \end{aligned}
  \end{equation*}
  where the kernel of $\tilde{\Theta}_{\alpha_{k_2+1}+\alpha_{k_2}}(s_{k_2-1},s_{k_2},{s}_{k_2+1},x_{k_2};V,\hbar)$ is given by
   \begin{equation*}
	\begin{aligned}
	 (p_{k_{2}+1},p_{k_{2}-1}) \mapsto{}& \frac{1}{\hbar^2} e^{-i\hbar^{-1/d}\langle x_{k_2},p_{k_{2}+1}-p_{k_{2}-1}\rangle} 
  	e^{i s_{k_2+1}\hbar^{-1} \frac{1}{2} (p_{k_2+1}^2-p_{k_2-1}^2)}   \int  e^{i s_{k_2}\hbar^{-1} \frac{1}{2} (p_{k_2}^2-p_{k_2-1}^2)}
 	 \\
 	 &\times  \Psi_{\alpha_{k_2+1}}(p_{k_2+1},p_{k_2},\hbar^{-1}s_{k_2}; V)  \Psi_{\alpha_{k_2}}(p_{k_2},p_{k_2-1},\hbar^{-1}(s_{k_2-1}-s_{k_2+1}-s_{k_2}); V) \, dp_{k_2}.
  	\end{aligned}
\end{equation*}
  With this form we observe that this is the same type of integrals as we considered in the proof of Lemma~\ref{expansion_aver_bound_Mainterm}, where one of the operators is slightly different. 
  To be precise the difference is in that we in this proof treats $p_{k_2}$ as an ``$\eta$''-variable and thereby also will treat $s_{k_2}$ as a ``$t$''-variable. 
 Hence with an argument analogous to that used in the proof of Lemma~\ref{expansion_aver_bound_Mainterm} for the case where $\kappa=\mathrm{id}$ we obtain that
     \begin{equation}\label{EQ:Aver_full_ex_same_0,0_dob3}
   	\begin{aligned}
	 \Aver{ \mathcal{T}(n,\sigma^1,\sigma^2,\alpha,\tilde{\alpha},\mathrm{id})}
	 \leq   \frac{ C_d^{a_k+\tilde{a}_k+n}  \norm{\varphi}_{L^2(\R^d)}^2 \norm{\hat{V}}_{1,\infty,3d+3}^{|\alpha|+|\tilde{\alpha}|}  (\rho t)^{2k-n} \tau_0^{n_2-2k_2}}{(k_1+1)!(k_1+1-n_1)! (k_2-1)!(k_2-1-n_2)!}.
	\end{aligned}
\end{equation}
  Since we have treaded $s_{k_2}$ as a ``$t$'' variable in the notation used in the proof of Lemma~\ref{expansion_aver_bound_Mainterm} the main difference in the arguments is in the estimate obtained in \eqref{EQ:exp_aver_b_Mainterm_5} in the proof of Lemma~\ref{expansion_aver_bound_Mainterm} for this case we use the estimate
    \begin{equation*}
	\begin{aligned}
	\MoveEqLeft     \int_{\R^{2k+2}_{+}} \prod_{i=1}^n 
	\frac{ \boldsymbol{1}_{[0,\frac{t}{\tau_0}]}( \boldsymbol{s}_{1,k_1+1}^{+}+\hbar \boldsymbol{t}_{1,a_{k_1+1}}^{+})  
	\boldsymbol{1}_{[0,\frac{t}{\tau_0}]}( \boldsymbol{\tilde{s}}_{1,k_1+1}^{+}+\hbar \boldsymbol{\tilde{t}}_{1,\tilde{a}_{k_1+1}}^{+}) }
	 {\max(1,\hbar^{-1} | \boldsymbol{s}_{\sigma^1_i,\sigma^1_{i+1}-1}^{+} - \boldsymbol{\tilde{s}}_{\sigma_i^2,\sigma_{i+1}^2-1}^{+} +\hbar l^1_i(\boldsymbol{t}, \boldsymbol{\tilde{t}})) |)^\frac{d}{2}}  
	 \\
	 &\times\frac{ \boldsymbol{1}_{[\frac{t}{\tau_0},\frac{\tau t}{\tau_0}]}( \boldsymbol{s}_{k_1+2,k}^{+}+\hbar \boldsymbol{t}_{a_{k_1}+2,a_k}^{+})  
	\boldsymbol{1}_{[\frac{t}{\tau_0},\frac{\tau t}{\tau_0}]}( \boldsymbol{\tilde{s}}_{k_1+2,k}^{+}+\hbar \boldsymbol{\tilde{t}}_{\tilde{a}_{k_1+2},\tilde{a}_k}^{+}) } 
	{\max(1, \hbar^{-1}| s_{k_2}|)^\frac{d}{2}\max(1, \hbar^{-1}| \tilde{s}_{k_2}|)^\frac{d}{2}} \, d\boldsymbol{\tilde{s}} d\boldsymbol{s}
	\\
	\leq{}& \frac{\hbar^{n+2} \left(\frac{\tau-1}{\tau_0}t\right)^{2k_2-2-n_2}  \left(\frac{t}{\tau_0}\right)^{2k_1 +2-n_1}}{(k_1+1)!(k_1+1-n_1)! (k_2-1)!(k_2-1-n_2)!} \left( \int_{\R} \frac{1}{\max(1, | s |)^\frac{d}{2}} \,ds \right)^{n+2} 
	\\
	\leq{}& \left(\frac{2d}{d-2}\right)^{n+2}  \frac{\hbar^{n+2} t^{2k-n} \tau_0^{n_1-2k_1-2}}{(k_1+1)!(k_1+1-n_1)! (k_2-1)!(k_2-1-n_2)!}.
	\end{aligned}
\end{equation*} 
We now turn to the case when $\kappa\neq\mathrm{id}$. Again this part of the proof will be analogous to that of the same case in Lemma~\ref{expansion_aver_bound_Mainterm}. Again we will here also denote the smallest $i$ such that $\kappa(i)\neq i$ by $i^{*}$ and we will use the second expressions for the kernels of the operators $\mathcal{I}$ from Observation~\ref{obs_form_I_op_kernel}. Using these we get that
\begin{equation*}
	\begin{aligned}
	\MoveEqLeft \mathcal{T}(n,\sigma^1,\sigma^2,\alpha,\tilde{\alpha},\kappa)
	= \sum_{(\boldsymbol{x},\boldsymbol{\tilde{x}})\in \mathcal{X}_{\neq}^{2k}} \frac{1}{(2\pi\hbar)^{d}}\prod_{i=1}^n \rho^{-1} \delta(x_{\sigma_i^1}- \tilde{x}_{\sigma_{\kappa(i)}^2})   \int_{\R_{+}^{|\alpha|+|\tilde{\alpha}|}} \int  \delta(p_k-q_k) 
	\\
	&\times \boldsymbol{1}_{[0,(\tau_1\hbar)^{-1}t]}( \boldsymbol{s}_{1,k_2}^{+}+ \boldsymbol{t}_{1,a_{k_2}}^{+}) \boldsymbol{1}_{[0,(\tau_1\hbar)^{-1}t]}( \boldsymbol{\tilde{s}}_{1,k_2}^{+}+ \boldsymbol{\tilde{t}}_{1,\tilde{a}_{k_2}}^{+})  \boldsymbol{1}_{[0,(\tau_0\hbar)^{-1}t]}( \boldsymbol{s}_{k_2+1,k+1}^{+}+ \boldsymbol{t}_{a_{k_2}+1,a_{k}}^{+}) 
	\\
	&\times     \boldsymbol{1}_{[0,(\tau_0\hbar)^{-1}t]}( \boldsymbol{\tilde{s}}_{k_2+1,k+1}^{+}+ \boldsymbol{\tilde{t}}_{\tilde{a}_{k_2}+1,\tilde{a}_{k}}^{+})  e^{i  ((\tau_0\hbar)^{-1}t- \boldsymbol{s}_{k_2+1,k+1}^{+}- \boldsymbol{t}_{a_{k_2}+1,a_k}^{+}) \frac{1}{2} p_{k_2}^2} 
	 \prod_{i=1}^{a_k} e^{i \frac{1}{2} t_{i}  \eta_{i}^2 } \prod_{i=1}^{\tilde{a}_k} e^{-i \frac{1}{2}  \tilde{t}_{i}  \xi_{i}^2} 
	\\
	&\times e^{-i  ((\tau_0\hbar)^{-1}t- \boldsymbol{\tilde{s}}_{k_2+1,k+1}^{+}- \boldsymbol{\tilde{t}}_{\tilde{a}_{k_2}+1,\tilde{a}_k}^{+}) \frac{1}{2} q_{k_2}^2} 
	  \Big\{ \prod_{m=1}^{k_2-1} e^{ -i   \langle  \hbar^{-1/d}x_m,p_{m}-p_{m-1} \rangle} e^{ i   \langle  \hbar^{-1/d}\tilde{x}_m,q_{m}-q_{m-1} \rangle} \Big\}
	  \\
	  &\times e^{ -i   \langle  \hbar^{-1/d}x_{k_2},p_{k_2+1}-p_{k_2-1} \rangle} e^{ i   \langle  \hbar^{-1/d}\tilde{x}_{k_2},q_{k_2+1}-q_{k_2-1} \rangle} \Big\{ \prod_{m=k_2+1}^{k+1} e^{ -i   \langle  \hbar^{-1/d}x_{m-1},p_{m}-p_{m-1} \rangle} 
	  \\
	  &\times e^{ i   \langle  \hbar^{-1/d}\tilde{x}_m,q_{m}-q_{m-1} \rangle} \Big\} \Big\{ \prod_{m=1}^k  e^{ -i   \langle  \hbar^{-1/d}x_m,p_{m}-p_{m-1} \rangle}   e^{i  s_{m} \frac{1}{2} p_{m}^2}  \hat{\mathcal{V}}_{\alpha_m}(p_m,p_{m-1},\boldsymbol{\eta} )
	 e^{-i  \tilde{s}_{m} \frac{1}{2} q_{m}^2}  
	 \\
	 &\times e^{ i   \langle  \hbar^{-1/d}\tilde{x}_m,q_{m}-q_{m-1} \rangle} \overline{\hat{\mathcal{V}}_{\tilde{\alpha}_m}(q_m,q_{m-1},\boldsymbol{\xi} ) } \Big\}   
	  e^{i  ((\tau_1\hbar)^{-1}t- \boldsymbol{s}_{1,k_2}^{+}- \boldsymbol{t}_{1,a_{k_2}}^{+}) \frac{1}{2} p_{0}^2} e^{-i  ((\tau_1\hbar)^{-1}t- \boldsymbol{\tilde{s}}_{1,k_2}^{+}- \boldsymbol{\tilde{t}}_{1,\tilde{a}_{k_2}}^{+}) \frac{1}{2} q_{0}^2} 
	\\
	&\times   
	 \hat{\varphi}(\tfrac{p_0}{\hbar})  \hat{\varphi}(\tfrac{q_0}{\hbar})   \,d\boldsymbol{\eta}d\boldsymbol{p} \boldsymbol{\xi}d\boldsymbol{q}  d\boldsymbol{t} d\boldsymbol{s}  d\boldsymbol{\tilde{t}} d\boldsymbol{\tilde{s}}.
	\end{aligned}
\end{equation*}
We observe that this is the same form as the same in case in the proof of Lemma~\ref{expansion_aver_bound_Mainterm_2} with more complicated relations between the indices. Keeping this in mind the argument is analogous to that made in the proof of Lemma~\ref{expansion_aver_bound_Mainterm_2}. We then obtain the estimate
  \begin{equation}\label{EQ:Aver_full_ex_same_0,0_dob4}
   	\begin{aligned}
	\MoveEqLeft \Aver{ \mathcal{T}(n,\sigma^1,\sigma^2,\alpha,\tilde{\alpha},\kappa)}
	\\
	& \leq   C_d^{a_k+\tilde{a}_k}    \norm{\hat{V}}_{1,\infty,4d+4}^{|\alpha|+|\tilde{\alpha}|}  \tau_0^{n_1+2-2k_1}   \frac{\rho(\rho t)^{2k-n-1} \hbar |\log(\zeta_1)|^{n_2+3} |\log(\zeta_0)|^{n_1+1}}{k_1!(k_1+1-n_1)! (k_2-2)!(k_2-1-n_2)!}  \norm{\varphi}_{\mathcal{H}^{3d+3}_\hbar(\R^d)}^2.
	\end{aligned}
\end{equation}
From combining \eqref{EQ:Aver_full_ex_same_0,0_dob1},  \eqref{EQ:Aver_full_ex_same_0,0_dob2},  \eqref{EQ:Aver_full_ex_same_0,0_dob3},  \eqref{EQ:Aver_full_ex_same_0,0_dob4} and arguing as in the proof of Lemma~\ref{expansion_aver_bound_Mainterm_2} we obtain  the desired estimate.
This concludes the proof.
\end{proof}
\subsection{Estimates for truncated terms}
We now turn to estimating the truncated terms without recollisions. Our first estimate is given in the following lemma.
\begin{lemma}\label{LE:Truncated_Without_re_bound_Mainterm_1}
Assume we are in the setting of Definition~\ref{functions_for_exp_def} and let $\varphi \in \mathcal{H}^{2d+2}_\hbar(\R^d)$ then
\begin{equation*}
	\begin{aligned}
	 \mathbb{E}\Big[\big\lVert \sum_{\boldsymbol{x}\in\mathcal{X}_{\neq}^{k_0+1}}\mathcal{E}^{k_0}_{0,0}(\boldsymbol{x},\tfrac{t}{\tau_0};\hbar) \varphi\big\rVert_{L^2(\R^d)}^2\Big]
	\leq \frac{C^{k_0}}{\hbar k_0!} \norm{\varphi}_{L^2(\R^d)}^2 +  \frac{k_0}{\tau_0^{k_0-1}\hbar}|\log(\tfrac{\hbar}{t})|^{k_0+3} C^{k_0}  \norm{\varphi}_{\mathcal{H}^{2d+2}_\hbar(\R^d)}^2,
	\end{aligned}
\end{equation*}
where the constant $C$ depends on $\rho$, $t$, the single site potential $V$ and the coupling constant $\lambda$. In particular, $  \sum_{\boldsymbol{x}\in\mathcal{X}_{\neq}^{k_0+1}}\mathcal{E}^{k_0}_{0,0}(\boldsymbol{x},\tfrac{t}{\tau_0};\hbar) \varphi \in L^2(\R^d)$ $\Pro$-almost surely. 
\end{lemma}
\begin{proof}
From definition of the operator $\mathcal{E}^{k_0}_{0,0}(\boldsymbol{x},t;\hbar)$ we get that
\begin{equation}\label{EQ:Truncated_Without_re_bound_Mainterm_1_1}
	\begin{aligned}
	\MoveEqLeft \big\lVert \sum_{\boldsymbol{x}\in\mathcal{X}_{\neq}^{k_0+1}} \mathcal{E}^{k_0}_{0,0}(\boldsymbol{x},\tfrac{t}{\tau_0}:\hbar) \varphi\big\rVert_{L^2(\R^d)}^2 
	\\
	&\leq \frac{t}{\tau_0}  \int_{0}^{\frac{t}{\tau_0}} \big\lVert  \sum_{\boldsymbol{x}\in\mathcal{X}_{\neq}^{k_0+1}}  \sum_{\alpha\in \N^{k_0} \times\{1\}}  (i\lambda)^{|\alpha|} \tilde{\mathcal{E}}^{k_0}_{0,0}(\boldsymbol{x},s_{k_0+1},\alpha,\tfrac{t}{\tau_0};\hbar) \varphi \big\rVert_{L^2(\R^d)}^2  \, ds_{k_0+1},
	\end{aligned}
\end{equation}
where we have used the unitarity of $U_{\hbar,\lambda}(-s_{k_0+1})$ and $U_{\hbar,0}(s_{k_0+1})$, Jensen's inequality and the operator $ \tilde{\mathcal{E}}^{k_0}_{0,0}(\boldsymbol{x},s_{k_0+1},\alpha,\frac{t}{\tau_0};\hbar)$ is given by
\begin{equation*}
	\begin{aligned}
	\tilde{\mathcal{E}}^{k_0}_{0,0}(\boldsymbol{x},s_{k_0+1},\alpha,\tfrac{t}{\tau_0};\hbar) =   \int_{[0,\frac{t}{\tau_0}]_{\leq}^{k_0}}   \boldsymbol{1}_{[s_{k_0+1},t]}(s_{k_0})   \prod_{m=1}^{k_0+1} \Theta_{\alpha_m}(s_{m-1},{s}_{m},x_m;V,\hbar)\, d\boldsymbol{s}_{k_0,1}U_{\hbar,0}(-t).
	  \end{aligned}
\end{equation*}
We observe that the right hand side in \eqref{EQ:Truncated_Without_re_bound_Mainterm_1_1} is essentially the same expression as the one considered in the proof of Lemma~\ref{expansion_aver_bound_Mainterm}, where the ``last'' time integral is outside the $L^2$-norm. From arguing as in the proof of Lemma~\ref{expansion_aver_bound_Mainterm} we obtain that
\begin{equation}\label{EQ:Truncated_Without_re_bound_Mainterm_1_2}
	\begin{aligned}
	\MoveEqLeft \int_{0}^{\frac{t}{\tau_0}} \mathbb{E} \Big[ \big\lVert  \sum_{\boldsymbol{x}\in\mathcal{X}_{\neq}^{k_0+1}}  \sum_{\alpha\in \N^{k_0} \times\{1\}}  (i\lambda)^{|\alpha|} \tilde{\mathcal{E}}^{k_0}_{0,0}(\boldsymbol{x},s_{k_0+1},\alpha,\tfrac{t}{\tau_0};\hbar) \varphi \big\rVert_{L^2(\R^d)}^2\Big]  \, ds_{k_0+1}
	\\
	&\leq  \frac{C^{k_0}}{\hbar k_0!} \norm{\varphi}_{L^2(\R^d)}^2 +  \frac{k_0}{\tau_0^{k_0-1}\hbar} |\log(\tfrac{\hbar}{t})|^{k_0+3} C^{k_0}  \norm{\varphi}_{\mathcal{H}^{2d+2}_\hbar(\R^d)}^2,
	\end{aligned}
\end{equation}
where we for the estimate corresponding to in the estimate in \eqref{EQ:exp_aver_b_Mainterm_11} from the proof of Lemma~\ref{expansion_aver_bound_Mainterm} for this case will be integratig $k_0-1$ variables all in belonging to the interval $[0,\frac{t}{\tau_0}]$. By combining \eqref{EQ:Truncated_Without_re_bound_Mainterm_1_1} and \eqref{EQ:Truncated_Without_re_bound_Mainterm_1_2}  we obtain the state estimate and this concludes the proof.
\end{proof}
\begin{lemma}\label{LE:Truncated_Without_re_bound_Mainterm_2}
  Assume we are in the setting of Definition~\ref{def_remainder_k_0}. Let $\varphi \in \mathcal{H}^{3d+3}_\hbar(\R^d)$. Then for any
  $\tau\in\{2,\dots,\tau_0\}$
    \begin{equation*}
	\begin{aligned}
	\MoveEqLeft \mathbb{E}\Big[ \big\lVert \sum_{k_2=1}^{k_0}   \sum_{(\boldsymbol{x}_1,\boldsymbol{x}_2)\in \mathcal{X}_{\neq}^{k_0+k_2+1}}  \mathcal{E}_{0,0}^{k_0}(\boldsymbol{x}_1,\iota,\tfrac{t}{\tau_0};\hbar) \mathcal{I}_{0,0}(k_2,\boldsymbol{x}_2,\iota,\tfrac{\tau-1}{\tau_0}t;\hbar) \varphi\big\rVert_{L^2(\R^d)}^2 \Big]
	  \\
	   &
	   \leq \frac{C^{k_0} k_0^3}{\tau_0^{k_0-2} \hbar k_0!} \norm{\varphi}_{L^2(\R^d)}^2 + \frac{ k_0^6}{\tau_0^{k_0-2}\hbar} |\log(\tfrac{\hbar}{t})|^{2k_0+3} C^{k_0}  \norm{\varphi}_{\mathcal{H}^{3d+3}_\hbar(\R^d)}^2,
	 \end{aligned}
\end{equation*}
and
   \begin{equation*}
	\begin{aligned}
	\MoveEqLeft \mathbb{E}\Big[ \big\lVert \sum_{k_2=1}^{k_0}   \sum_{(\boldsymbol{x}_1,\boldsymbol{x}_2)\in \mathcal{X}_{\neq}^{k_0+k_2+1}}  \mathcal{E}_{0,0}^{k_0}(\boldsymbol{x}_1,\iota,\tfrac{t}{\tau_0};\hbar) \mathcal{I}_{0,0}(k_2,\boldsymbol{x}_2,\iota,\tfrac{\tau-1}{\tau_0}t;\hbar) \varphi\big\rVert_{L^2(\R^d)}^2 \Big]
	  \\
	   &
	   \leq \frac{C^{k_0} k_0^3}{\tau_0^{k_0-2} \hbar k_0!} \norm{\varphi}_{L^2(\R^d)}^2 + \frac{ k_0^6}{\tau_0^{k_0-2}\hbar} |\log(\tfrac{\hbar}{t})|^{2k_0+3} C^{k_0}  \norm{\varphi}_{\mathcal{H}^{3d+3}_\hbar(\R^d)}^2,
	 \end{aligned}
\end{equation*}
where the constant $C$ depends on the single site potential $V$ and the coupling constant $\lambda$. In particular we have that the function is in  $L^2(\R^d)$ $\Pro$-almost surely. 
\end{lemma}
\begin{proof}
Firstly by arguing as in the proof of Lemma~\ref{LE:Truncated_Without_re_bound_Mainterm_1} we obtain an expression similar to the one estimated in Lemma~\ref{expansion_aver_bound_Mainterm_2}, but with the last time integral outside of the $L^2$-norm. Then using an argument analogous to that used in the proof of Lemma~\ref{expansion_aver_bound_Mainterm_2} or Lemma~\ref{expansion_aver_bound_Mainterm_3}  we obtain an estimate that put together with the estimate obtained by arguing as  in the proof of Lemma~\ref{LE:Truncated_Without_re_bound_Mainterm_1} gives us the stated estimate.
\end{proof}
\section{Terms with one recollisions}\label{Sec:tech_est_1_recol}
We will in this section continue with the technical estimates which gives that the Duhamel expansion do converge in appropriate sense. We will here focus on the terms where we have observed a single recollision.
\subsection{Estimates for fully expanded terms}
\begin{lemma}\label{expansion_aver_bound_Mainterm_rec}
Assume we are in the setting of Definition~\ref{functions_for_exp_rec_def}. Let  $\iota\in\mathcal{Q}_{k,1,1}$ and let $\varphi \in \mathcal{H}^{2d+2}_\hbar(\R^d)$. Then 
\begin{equation*}
	\begin{aligned}
	\MoveEqLeft \mathbb{E}\Big[\big\lVert \sum_{\boldsymbol{x}\in\mathcal{X}_{\neq}^{k-1}}\mathcal{I}_{1,1}(k,\boldsymbol{x},\iota,t;\hbar) \varphi\big\rVert_{L^2(\R^d)}^2 \Big] \leq   \frac{ C^{k}}{(k-1)!}  \norm{\varphi}_{L^{2}(\R^d)}^2 +  \hbar k |\log(\zeta)|^{k+4} C^k  \norm{\varphi}_{\mathcal{H}^{2d+2}_\hbar(\R^d)}^2,
	\end{aligned}
\end{equation*}
where the constant $C$ depends on the single site potential $V$ and the coupling constant $\lambda$. In particular we have that the function is in  $L^2(\R^d)$ $\Pro$-almost surely. 
\end{lemma}
\begin{proof}
As in the proof of Lemma~\ref{expansion_aver_bound_Mainterm} we get two sums over the set $\mathcal{X}_{\neq}^{k-1}$, when we take the $L^2$-norm. Again we divide into different cases depending on how many points $x_m$ and $\tilde{x}_m$ they have in common and where they are placed. This yields
\begin{equation}\label{EQ:exp_aver_b_Mainterm_rec_1}
	\big\lVert \sum_{\boldsymbol{x}\in\mathcal{X}_{\neq}^{k-1}} \mathcal{I}_{1,1}(k,\boldsymbol{x},\iota,t;\hbar) \varphi\big\rVert_{L^2(\R^d)}^2
	=\sum_{\alpha,\tilde{\alpha}\in\N^k}\sum_{n=0}^{k-1} \sum_{\sigma^1,\sigma^2\in\mathcal{A}(k-1,n)} \sum_{\kappa\in\mathcal{S}_n} (i\lambda)^{|\alpha|}(-i\lambda)^{|\tilde{\alpha}|} \mathcal{T}(n,\sigma^1,\sigma^2,\alpha,\tilde{\alpha},\kappa),
\end{equation}
where the numbers $\mathcal{T}(n,\sigma^1,\sigma^2,\alpha,\tilde{\alpha},\kappa)$ are given by
\begin{equation*}
	\begin{aligned}
	\MoveEqLeft \mathcal{T}(n,\sigma^1,\sigma^2,\alpha,\tilde{\alpha},\kappa)
	\\
	={}& \sum_{\boldsymbol{x},\tilde{\boldsymbol{x}}\in \mathcal{X}_{\neq}^{k-1}}    \prod_{i=1}^n \rho^{-1}\delta(x_{\sigma_i^1}- \tilde{x}_{\sigma_{\kappa(i)^2}})  \int \mathcal{I}_{1,1}(k,\boldsymbol{x},\alpha,\iota,t;\hbar) \varphi(x)\overline{ \mathcal{I}_{1,1}(k,\boldsymbol{\tilde{x}},\tilde{\alpha},\iota,t;\hbar) \varphi(x)}  \,dx,
	\end{aligned}
\end{equation*}
where the operators $\mathcal{I}_{1,1}(k,\boldsymbol{x},\alpha,\iota,t;\hbar)$ is defined as $\mathcal{I}_{1,1}(k,\boldsymbol{x},\iota,t;\hbar)$ for a fixed $\alpha$. We will again divide into the two different cases depending on if $\kappa=\mathrm{id}$ or $\kappa\neq\mathrm{id}$. We start with the case where $\kappa=\mathrm{id}$. For this case we divide further into the two cases depending on the relation between the map $\iota$ and the two vectors $\sigma^1$ and $\sigma^2$. First case is if there exist $i_1$ and $i_2$ in $\{1,n+1\}$ such that  following condition is satisfied 
\begin{equation}\label{EQ:exp_aver_b_Mainterm_rec_1.5}
	\sigma_{i_1-1}^1\leq m_1<m_2 < \sigma_{i_1}^1 \quad\text{and}\quad \sigma_{i_2-1}^2 \leq m_1< m_2 < \sigma_{i_2}^2.
\end{equation}
We will in this case use the following expression for the function $\mathcal{I}_{1,1}(k,\boldsymbol{x},\alpha,\iota,t;\hbar)\varphi$ 
\begin{equation*}
	\begin{aligned}
 	\MoveEqLeft\mathcal{I}_{1,1}(k,\boldsymbol{x},\alpha,\iota,t;\hbar) \varphi(x) 
	=  \frac{1}{(2\pi\hbar)^d\hbar^k} \int_{\R_{+}^{|\alpha|}} \int \boldsymbol{1}_{[0,t]}( \boldsymbol{s}_{1,k}^{+}+ \hbar\boldsymbol{t}_{1,a_k}^{+})
	 e^{ i   \langle  \hbar^{-1}x,p_{k} \rangle}  \prod_{i=1}^{a_k} e^{i  t_{i} \frac{1}{2} \eta_{i}^2} \Big\{\prod_{m=1}^k e^{i  s_{m} \hbar^{-1} \frac{1}{2} p_{m}^2}
	\\
	&\times   e^{ -i   \langle  \hbar^{-1/d}x_{\iota(m)},p_{m}-p_{m-1} \rangle}     \hat{\mathcal{V}}_{\alpha_m}(p_m,p_{m-1},\boldsymbol{\eta} )  \Big\}    e^{i  \hbar^{-1}(t- \boldsymbol{s}_{1,k}^{+}- \hbar\boldsymbol{t}_{1,a_k}^{+}) \frac{1}{2} p_0^2} \hat{\varphi}(\tfrac{p_0}{\hbar}) \, d\boldsymbol{\eta} d\boldsymbol{p}   d\boldsymbol{t}d\boldsymbol{s}.
	\end{aligned}
\end{equation*}
We start by making the change of variables $p_{m_2-1}\mapsto p_{m_2-1} - p_{m_2}$ and $p_{m}\mapsto p_{m} - p_{m_2-1}$ for all $m\in\{m_1,\dots,m_2-2\}$, where $m_1$ and $m_2$ are the numbers associated to the map $\iota$. Moreover we also do a relabelling by changing $p_{m_2-1}$ into $\tilde{p}$ and $p_m$ into $p_{m-1}$ for all $m\in\{m_2,\dots,k\}$.This yields the form
\begin{equation}\label{EQ:exp_aver_b_Mainterm_rec_2}
	\begin{aligned}
 	\MoveEqLeft\mathcal{I}_{1,1}(k,\boldsymbol{x},\alpha,\iota,t;\hbar) \varphi(x) 
	=  \frac{1}{(2\pi\hbar)^d\hbar^k} \int_{\R_{+}^{|\alpha|}} \int \boldsymbol{1}_{[0,t]}( \boldsymbol{s}_{1,k}^{+}+\hbar \boldsymbol{t}_{1,a_k}^{+})
	 e^{ i   \langle  \hbar^{-1}x,p_{k-1} \rangle}  e^{i  s_{m_2} \hbar^{-1} \frac{1}{2} p_{m_2-1}^2} 
	 \\
	 &\times  \hat{\mathcal{V}}_{\alpha_{m_2}}(p_{m_2},\tilde{p}+p_{m_2},\boldsymbol{\eta} )  \prod_{i=1}^{a_k} e^{i  t_{i} \frac{1}{2} \eta_{i}^2}  
	  \prod_{m=1}^{k-1}  e^{ -i   \langle  \hbar^{-1/d}x_{m},p_{m}-p_{m-1}\rangle} 
	 \Big\{\prod_{m=1}^{k-1}  
	   e^{i  s_{\iota^{*}(m)} \hbar^{-1} \frac{1}{2} (p_{m}+\pi_{m}^1(\tilde{p}))^2} 
	   \\
	   &\times   \hat{\mathcal{V}}_{\alpha_{\iota^{*}(m)}}(p_m+\pi_{m}^1(\tilde{p}),p_{m-1}+\pi_{m}^2(\tilde{p}),\boldsymbol{\eta} )  \Big\}   e^{i  \hbar^{-1}(t- \boldsymbol{s}_{1,k}^{+}- \hbar\boldsymbol{t}_{1,a_k}^{+}) \frac{1}{2} p_0^2} \hat{\varphi}(\tfrac{p_0}{\hbar}) \, d\boldsymbol{\eta}d\boldsymbol{p}d\tilde{p}   d\boldsymbol{t} d\boldsymbol{s},
	\end{aligned}
\end{equation}
where $\iota^{*}$ is the function associated to $\iota$ and
\begin{equation}\label{EQ:exp_aver_b_Mainterm_rec_3}
	\pi_{m}^1(\tilde{p} ) = \begin{cases}
	\tilde{p} & \text{if $m\in\{ m_1,\dots,m_2-1  \}$}
	\\
	0 &\text{otherwise}.
	\end{cases}
	\quad\text{and}\quad 
	\pi_{m}^2(\tilde{p} ) = \begin{cases}
	\tilde{p} & \text{if $m\in\{ m_1+1,\dots,m_2-1  \}$}
	\\
	0 &\text{otherwise}.
	\end{cases}
\end{equation}
With this expression we can write down $\mathcal{T} (n,\sigma^1,\sigma^2,\alpha,\tilde{\alpha},\mathrm{id})$ and use Lemma~\ref{LE:Exp_ran_phases} to obtain that
\begin{equation}\label{EQ:exp_aver_b_Mainterm_rec_3.5}
	\begin{aligned}
	\MoveEqLeft \Aver{\mathcal{T}(n,\sigma^1,\sigma^2,\alpha,\tilde{\alpha},\mathrm{id}) }=  \frac{(\rho(2\pi)^d)^{2k-n-2}}{(2\pi\hbar)^d\hbar^{n+2}} \int_{\R_{+}^{|\alpha|+|\tilde{\alpha}|}} \int \boldsymbol{1}_{[0,t]}( \boldsymbol{s}_{1,k}^{+}+ \hbar\boldsymbol{t}_{1,a_k}^{+})  \boldsymbol{1}_{[0,t]}( \boldsymbol{\tilde{s}}_{1,k}^{+}+ \hbar\boldsymbol{\tilde{t}}_{1,\tilde{a}_k}^{+})
	 \\
	 &\times 
	  e^{i \hbar^{-1}  s_{m_2} \frac{1}{2} p_{i_1-1}^2}   
	  e^{-i \hbar^{-1}  \tilde{s}_{m_2} \frac{1}{2} p_{i_2-1}^2}   
	   \prod_{i=1}^{a_k} e^{i  t_{i} \frac{1}{2} \eta_{i}^2}  
	    \prod_{i=1}^{\tilde{a}_k} e^{-i  \tilde{t}_{i} \frac{1}{2} \xi_{i}^2}  
	    \prod_{i=1}^{n+1} \big\{ \prod_{m=\sigma_{i}^1+1}^{\sigma_{i}^1-1}  e^{i \hbar^{-1}  s_{\iota^{*}(m)} \frac{1}{2} (p_{i-1}+\pi_m^1(\tilde{p}))^2}   
	 \\
	 &\times 
	  \prod_{m=\sigma_{i-1}^2+1}^{\sigma_{i}^2-1}  e^{-i \hbar^{-1}  \tilde{s}_{\iota^{*}(m)} \frac{1}{2} (p_{i-1}+\pi_m^1(\tilde{q}))^2}   \big\}
	  \prod_{i=1}^{n} e^{i \hbar^{-1}  s_{\iota^{*}(\sigma_i^1)} \frac{1}{2} (p_{i}+\pi_{\sigma_i^1}^1(\tilde{p}))^2} 
	  e^{-i \hbar^{-1}  \tilde{s}_{\iota^{*}(\sigma_i^2)} \frac{1}{2} (p_{i}+\pi_{\sigma_i^2}^1(\tilde{q}))^2}  
	\\
	&\times e^{i  \hbar^{-1}(\boldsymbol{\tilde{s}}_{1,k}^{+}+  \hbar \boldsymbol{\tilde{t}}_{1,a_k}^{+} - \boldsymbol{s}_{1,k}^{+}- \hbar \boldsymbol{t}_{1,a_k}^{+}) \frac{1}{2} p_0^2}   \mathcal{G}(\tilde{p},\tilde{q},\boldsymbol{p},\boldsymbol{\eta},\boldsymbol{\xi},\sigma^1,\sigma^2,\alpha,\tilde{\alpha})  |\hat{\varphi}(\tfrac{p_0}{\hbar})|^2 \, d\boldsymbol{\eta}d\boldsymbol{\xi}d\boldsymbol{p} d\tilde{p}d\tilde{q}   d\boldsymbol{t} d\boldsymbol{s}d\boldsymbol{\tilde{t}} d\boldsymbol{\tilde{s}},
	\end{aligned}
\end{equation}
where we have used our assumption on the two indices $i_1$ and $i_2$ given in \eqref{EQ:exp_aver_b_Mainterm_rec_1.5} and evaluated all integrals involving delta functions. Moreover, we have used the notation
\begin{equation*}
	\begin{aligned}
	\mathcal{G}(\tilde{p},\tilde{q},\boldsymbol{p},\boldsymbol{\eta},\boldsymbol{\xi},\sigma^1,\sigma^2,\alpha,&\tilde{\alpha}) 
	= \hat{\mathcal{V}}_{\alpha_{m_2}}(p_{i_1-1},\tilde{p} +p_{i_1-1},\boldsymbol{\eta} )  \prod_{i=1}^{n+1} \prod_{m=\sigma_{i-1}^1+1}^{\sigma_{i}^1-1}  \hat{\mathcal{V}}_{\alpha_{\iota^{*}(m)}}(p_{i-1}+\pi_m^1(\tilde{p}),p_{i-1}+\pi_m^2(\tilde{p}),\boldsymbol{\eta} ) 
	\\
	&\times \overline{\hat{\mathcal{V}}_{\alpha_{m_2}}(p_{i_2-1},\tilde{q} +p_{i_2-1},\boldsymbol{\xi} )} 
	\prod_{i=1}^{n+1} 
	    \prod_{m=\sigma_{i-1}^2+1}^{\sigma_{i}^2-1}  \overline{\hat{\mathcal{V}}_{\alpha_{\iota^{*}}(m)}(p_{i-1}+\pi_m^1(\tilde{q}),p_{i-1}+\pi_m^2(\tilde{q}),\boldsymbol{\xi} )} 
	  \\
	 &\times \prod_{i=1}^{n}  \hat{\mathcal{V}}_{\alpha_{\iota^{*}}(\sigma_i^1)}(p_{i} + \pi_{\sigma_i^1}^1(\tilde{p}) ,p_{i-1} + \pi_{\sigma_i^1}^2(\tilde{p}),\boldsymbol{\eta} )   
	  \overline{ \hat{\mathcal{V}}_{\alpha_{\iota^{*}}(\sigma_i^2)}(p_{i}+\pi_{\sigma_i^2}^1(\tilde{q}),p_{i-1}+\pi_{\sigma_i^2}^1(\tilde{q}),\boldsymbol{\xi} ) }.   
	\end{aligned}
\end{equation*}
We note that this form is almost identical to the form of the function considered in the proof of Lemma~\ref{expansion_aver_bound_Mainterm}. The main difference is the dependence on the variables $\tilde{p}$ and $\tilde{q}$. 

We note that $\pi_m^1(\boldsymbol{\tilde{p}})\neq0$ for $ \sigma_{i_1-1}^1\leq m < \sigma_{i_1}^1$ and $\pi_m^1(\boldsymbol{\tilde{q}})\neq0$ for $ \sigma_{i_1-1}^2\leq m < \sigma_{i_1}^2$. Hence we preform the following changes of variables $\tilde{p} \mapsto \tilde{p}  + p_{i_1-1}$ and  $\tilde{q} \mapsto \tilde{q}_1 + p_{i_2-1}$. This change of variables ensures that in all quadratic phases we only have one variable. Hence we can argue as in the proof of Lemma~\ref{expansion_aver_bound_Mainterm} to ensure integrability in the $t$'s and $\tilde{t}$'s and case by case use Lemma~\ref{app_quadratic_integral_tech_est} directly for the all variables $p_1,\dots,p_n,\tilde{p},\tilde{q}$. Preforming these arguments we arrive at the bound
   \begin{equation}\label{EQ:exp_aver_b_Mainterm_rec_4}
   	\begin{aligned}
	  |\Aver{\mathcal{T}(n,\sigma^1,\sigma^2,\alpha,\tilde{\alpha},\mathrm{id})}|
	 \leq  C_d^{a_k +\tilde{a}_k+n}  \norm{\varphi}_{L^2(\R^d)}^2 \norm{\hat{V}}_{1,\infty,3d+3}^{|\alpha|+|\tilde{\alpha}|}  \frac{ (\rho t)^{2k-2-n}}{(k-1)!(k-1-n)!}.
	\end{aligned}
\end{equation}
We now turn to the case where such to indices does not exists. For this case we can without loss of generality assume that there exists $i_1$ such that $m_1<\sigma^1_{i_1}<m_2$. We will here use the same expression as above for the function $\mathcal{I}_{1,1}(k,\boldsymbol{x},\alpha,\iota,t;\hbar)\varphi$ but with a rescaling of the time variables $s_1,\dots,s_k$. Hence with the same notation (the functions $\pi^1_m$ and $\pi^2_m$) as above we consider the expression
\begin{equation}\label{EQ:exp_aver_b_Mainterm_rec_5}
	\begin{aligned}
	\MoveEqLeft \Aver{\mathcal{T}(n,\sigma^1,\sigma^2,\alpha,\tilde{\alpha},\mathrm{id}) }=  \frac{(\rho(2\pi)^d\hbar)^{2k-n-2}}{(2\pi\hbar)^d} \int_{\R_{+}^{|\alpha|+|\tilde{\alpha}|}} \int \boldsymbol{1}_{[0,\hbar^{-1}t]}( \boldsymbol{s}_{1,k}^{+}+ \boldsymbol{t}_{1,a_k}^{+})  \boldsymbol{1}_{[0,\hbar^{-1}t]}( \boldsymbol{\tilde{s}}_{1,k}^{+}+ \boldsymbol{\tilde{t}}_{1,\tilde{a}_k}^{+})
	 \\
	 &\times 
	  e^{i   s_{m_2} \frac{1}{2} p_{i_1-1}^2}   
	  e^{-i   \tilde{s}_{m_2} \frac{1}{2} p_{i_2-1}^2}   
	   \prod_{i=1}^{a_k} e^{i  t_{i} \frac{1}{2} \eta_{i}^2}  
	    \prod_{i=1}^{\tilde{a}_k} e^{-i  \tilde{t}_{i} \frac{1}{2} \xi_{i}^2}  
	    \prod_{i=1}^{n+1} \big\{ \prod_{m=\sigma_{i}^1+1}^{\sigma_{i}^1-1}  e^{i   s_{\iota^{*}(m)} \frac{1}{2} (p_{i-1}+\pi_m^1(\tilde{p}))^2}   
	 \\
	 &\times 
	  \prod_{m=\sigma_{i-1}^2+1}^{\sigma_{i}^2-1}  e^{-i \hbar^{-1}  \tilde{s}_{\iota^{*}(m)} \frac{1}{2} (p_{i-1}+\pi_m^1(\tilde{q}))^2}   \big\}
	  \prod_{i=1}^{n} e^{i  s_{\iota^{*}(\sigma_i^1)} \frac{1}{2} (p_{i}+\pi_{\sigma_i^1}^1(\tilde{p}))^2} 
	  e^{-i  \tilde{s}_{\iota^{*}(\sigma_i^2)} \frac{1}{2} (p_{i}+\pi_{\sigma_i^2}^1(\tilde{q}))^2}  
	\\
	&\times e^{i  (\boldsymbol{\tilde{s}}_{1,k}^{+}+  \boldsymbol{\tilde{t}}_{1,a_k}^{+} - \boldsymbol{s}_{1,k}^{+}-  \boldsymbol{t}_{1,a_k}^{+}) \frac{1}{2} p_0^2}   \mathcal{G}(\tilde{p},\tilde{q},\boldsymbol{p},\boldsymbol{\eta},\boldsymbol{\xi},\sigma^1,\sigma^2,\alpha,\tilde{\alpha})  |\hat{\varphi}(\tfrac{p_0}{\hbar})|^2 \, d\boldsymbol{\eta}d\boldsymbol{\xi}d\boldsymbol{p} d\tilde{p}d\tilde{q}   d\boldsymbol{t} d\boldsymbol{s}d\boldsymbol{\tilde{t}} d\boldsymbol{\tilde{s}},
	\end{aligned}
\end{equation}
where we in this case have that $\mathcal{G}$ is given by 
\begin{equation*}
	\begin{aligned}
	\MoveEqLeft \mathcal{G}(\tilde{p},\tilde{q},\boldsymbol{q},\boldsymbol{\eta},\boldsymbol{\xi},\sigma^1,\sigma^2,\alpha,\tilde{\alpha}) = \hat{\mathcal{V}}_{\alpha_{m_2}}(q_{\sigma_{i(m_2)}^1},\tilde{p}+q_{\sigma_{i(m_2)}^1},\boldsymbol{\eta} ) \overline{\hat{\mathcal{V}}_{\tilde{\alpha}_{m_2}}(q_{\sigma_{\tilde{i}(m_2)}^2},\tilde{q}+q_{\sigma_{\tilde{i}(m_2)}^2},\boldsymbol{\xi} )}
	\\
	&\times \prod_{i=1}^n  \overline{\hat{\mathcal{V}}_{\tilde{\alpha}_{\iota^{*}(\sigma^2_i)}}(q_{\sigma^2_i}+\pi_{m}^1(\tilde{q}),q_{\sigma_{i-1}^2}+\pi_{m}^2(\tilde{q}),\boldsymbol{\xi})}  \hat{\mathcal{V}}_{\alpha_{\iota^{*}(\sigma^1_i)}}(q_{\sigma_i^1}+\pi_{m}^1(\tilde{p}),q_{\sigma_{i-1}^1}+\pi_{m}^2(\tilde{p}),\boldsymbol{\eta}) 
	\\
	&\times  \prod_{i=0}^n 
	 \prod_{m=\sigma^2_i +1}^{\sigma^2_{i+1} -1} \overline{\hat{\mathcal{V}}_{\tilde{\alpha}_{\iota^{*}(m)}}(q_{\sigma^2_i}+\pi_{m}^1(\tilde{q}),q_{\sigma^2_i}+\pi_{m}^2(\tilde{q}),\boldsymbol{\xi})} \prod_{m=\sigma^1_i +1}^{\sigma^1_{i+1} -1} \hat{\mathcal{V}}_{\alpha_{\iota^{*}(m)}}(q_{\sigma_i^1}+\pi_{m}^1(\tilde{p}),q_{\sigma_i^1}+\pi_{m}^2(\tilde{p}),\boldsymbol{\eta}),   
	\end{aligned}
\end{equation*}
where $i(m_2)$ is the largest index such that $m_2\geq\sigma^1_{i(m_2)}$ and  $\tilde{i}(m_2)$ is the largest index such that $m_2\geq\sigma^2_{\tilde{i}(m_2)}$.
From this form the proof is similar to the case of $\kappa\neq \mathrm{id}$ in the proof of Lemma~\ref{expansion_aver_bound_Mainterm}. We first insert a function $f(\boldsymbol{s},\boldsymbol{\tilde{s}})$ of some of the time variables as in the other proof. Again it will have the form that the sum over some of the $s$ and $\tilde{s}$ variables is bounded by $\hbar^{-1}t$. For $\boldsymbol{s}$ it will not depend on $s_{m_1}$ and $s_{m_2-1}$. For $\tilde{s}$ it will not depend on $\tilde{s}_{m_2-1}$ and $\tilde{s}_m$ for $m\in \iota^{*}(\sigma^2)$. For the cases, where $m_2-1\in\iota^{*}(\sigma^2)$ it will also not depend on $m_2$. Hence for $\tilde{s}$ there are always $n+1$ variables it does not depend on. The function can be written as
\begin{equation*}
	f(\boldsymbol{s},\boldsymbol{\tilde{s}}) = \boldsymbol{1}_{[0,\hbar^{-1}t]}\big( \sum_{i\in\mathcal{B}} s_i \big) \boldsymbol{1}_{[0,\hbar^{-1}t]}\big( \sum_{i\in\tilde{\mathcal{B}}} s_i \big),
\end{equation*}
where $\mathcal{B}$ and $\tilde{\mathcal{B}}$ are index set as described above. As in the proof of Lemma~\ref{expansion_aver_bound_Mainterm} we introduce the two variables $s_0$ and $\tilde{s}_0$ and thereby also two delta functions. We again use the function $\zeta$ and insert the two functions  $e^{(\hbar^{-1}t- \boldsymbol{s}_{0,k}^{+}- \boldsymbol{t}_{1,a_k}^{+})\zeta(\hbar^{-1}t)}$ and $e^{(\hbar^{-1}t- \boldsymbol{\tilde{s}}_{0,k}^{+}- \boldsymbol{\tilde{t}}_{1,\tilde{a}_k}^{+})\zeta(\hbar^{-1}t)}$, write the delta functions as Fourier transforms of $1$. Then we use Lemma~\ref{LE:Exp_ran_phases} and the standard argument to obtain integrability in the $t$ and $\tilde{t}$ variables and we evaluate all integrals in $\tilde{s}$ and $s$ the the function $f(\boldsymbol{s},\boldsymbol{\tilde{s}})$ does not depend on. We first assume that we are in the case where $m_2-1\in\iota^{*}(\sigma^2)$. Preforming these operations we obtain the estimate
   \begin{equation}\label{EQ:exp_aver_b_Mainterm_rec_8}
   	\begin{aligned}
	 |\Aver{\mathcal{T} (n,\sigma^1,\sigma^2,\alpha,\tilde{\alpha},\mathrm{id})}|
	 \leq {}&    \frac{ \left(\rho\hbar(2\pi)^d\right)^{2k-2-n}e^{2\hbar^{-1}t\zeta}}{(2\pi)^2(2\pi\hbar)^{d}}  \sum_{l=1}^{4^{a_{k}}}   \sum_{|\epsilon_1|,\dots,|\epsilon_{j_1}|\leq d+1} \sum_{|\nu_1|,\dots,|\nu_{\tilde{j}_1}|\leq d+1}   \int_{\R^{a_k+\tilde{a_k}}_{+}}
	  \boldsymbol{1}_{B_l}(\boldsymbol{t},\boldsymbol{\tilde{t}}) 
	  \\
	  &\times \prod_{i\in J_1} \frac{C(\epsilon_i)}{(2\pi t_i)^\frac{d}{2} } \prod_{i\in \tilde{J}_1} \frac{C(\nu_i)   }{(2\pi \tilde{t}_i)^\frac{d}{2}}  |\mathcal{I}_n(\boldsymbol{t},\boldsymbol{\tilde{t}}) |  d\boldsymbol{t}_{1,a_k}  d\boldsymbol{\tilde{t}}_{1,\tilde{a}_k} ,
	\end{aligned}
\end{equation}
where
\begin{equation}\label{EQ:exp_aver_b_Mainterm_rec_9}
	\begin{aligned}
	 \big| \mathcal{I}_n(\boldsymbol{t},\boldsymbol{\tilde{t}})\big|
	\leq{}& \int  \prod_{i\in J_1} \frac{| x_i ^{\epsilon_i}| }{(1+|x_i |^2)^{d+1}} \,d\boldsymbol{x} \int \prod_{i\in \tilde{J}_1} \frac{| y_i ^{\nu_i}| }{(1+|y_i |^2)^{d+1}}  \,d\boldsymbol{y}    \int_{\R_{+}^{2k-n-3}} f(\boldsymbol{s},\boldsymbol{\tilde{s}})  d\boldsymbol{s}   d\boldsymbol{\tilde{s}}  \int    \frac{ |\hat{\varphi}(\tfrac{q_0}{\hbar}) |^2}{|\frac{1}{2} q_0^2-\tilde{\nu}-i\zeta|}
	\\
	&\times    \frac{1}{|\frac{1}{2} q_0^2+\nu +i\zeta|} \frac{1}{| \frac{1}{2} (\tilde{p} +q_{l_1})^2 +\nu+i\zeta|}   \frac{1}{| \frac{1}{2} (\tilde{p} +q_{l_2})^2 +\nu+i\zeta|}  \frac{1}{ |\frac{1}{2} q_{\sigma^2_{m^*}}^2-\tilde{\nu}-i\zeta|} 
	  \\
	  &\times    \prod_{m=1}^n \frac{1}{ |\frac{1}{2} (q_{\sigma^2_m}+\pi_{\sigma^2_m}(\tilde{q}))^2-\tilde{\nu}-i\zeta|}
	 \Big| \partial_{\boldsymbol{\xi}_{\tilde{J}_1}}^{\boldsymbol{\nu}} \partial_{\boldsymbol{\eta}_{J_1}}^{\boldsymbol{\epsilon}}  \mathcal{G}(\tilde{p},\tilde{q},\boldsymbol{q},\boldsymbol{\eta},\boldsymbol{\xi},\sigma^1,\sigma^2,\alpha,\tilde{\alpha}) \Big| 
	  \,d\boldsymbol{\eta}d\boldsymbol{\xi}d\boldsymbol{q} d\tilde{p}d\tilde{q} d\nu d\tilde{\nu}.
	\end{aligned}
\end{equation}
where the indices $l_1<l_2$ and they depend on $\sigma^1$. That the indices $l_1$ and $l_2$ are different is due to the assumption that there exsist $i_1$ such that $m_1<\sigma^1_{i_1}<m_2$. The index $m^{*}$ is chosen such that $\sigma_{m^{*}}=m_2-1$. Again we need to estimate these integrals. Firstly we have that  
\begin{equation}\label{EQ:exp_aver_b_Mainterm_rec_10}
	\begin{aligned}
	     \int_{\R_{+}^{2k-n-3}} f(\boldsymbol{s},\boldsymbol{\tilde{s}})  d\boldsymbol{s}  d\boldsymbol{\tilde{s}} \leq \frac{(\hbar^{-1} t)^{2k-n-3} }{(k-2)!(k-n-1)!} . 
	\end{aligned}
\end{equation}
In order to estimate the integrals in $\tilde{p}$, $\tilde{q}$, $q$, $\eta$ and $\xi$ we insert the fraction
  \begin{equation*}
  	\frac{\langle q_0\rangle^{4d+4}}{\langle q_0\rangle^{4d+4}} \frac{\langle\tilde{q} \rangle^{5d+5}\langle\tilde{p} \rangle^{2d+2}}{\langle\tilde{q} \rangle^{5d+5}\langle \tilde{p} \rangle^{2d+2}} \prod_{i=1}^n \frac{\langle q_{\sigma_i^2}+\pi_{\sigma_i^2}^1(\tilde{q})-q_{\sigma_{i-1}^2} - \pi_{\sigma_i^2}^2(\tilde{q})\rangle^{5d+5}}{\langle q_{\sigma_i^2}+\pi_{\sigma_i^2}^1(\tilde{q})-q_{\sigma_{i-1}^2} - \pi_{\sigma_i^2}^2(\tilde{q})\rangle^{5d+5}}. 
  \end{equation*}
Firstly we have that
  \begin{equation}\label{EQ:exp_aver_b_Mainterm_rec_11}
	\begin{aligned}
	 \sup_{\tilde{p},\tilde{q},\boldsymbol{q}}   \int  &\langle\tilde{q} \rangle^{5d+5}\langle\tilde{p} \rangle^{2d+2}\prod_{i=1}^n \langle q_{\sigma_i^2}+\pi_{\sigma_i}^1(\tilde{q})-q_{\sigma_{i-1}^2} - \pi_{\sigma_i}^2(\tilde{q})\rangle^{5d+5} \Big| \partial_{\boldsymbol{\xi}_{\tilde{J}_1}}^{\boldsymbol{\nu}} \partial_{\boldsymbol{\eta}_{J_1}}^{\boldsymbol{\epsilon}}   \mathcal{G}(\tilde{p},\tilde{q},\boldsymbol{q},\boldsymbol{\eta},\boldsymbol{\xi},\sigma^1,\sigma^2,\alpha,\tilde{\alpha})  \Big|  \,d\boldsymbol{\eta}d\boldsymbol{\xi}	
	 \\
	 &  
	 \leq  C^{a_k+\tilde{a}_k}  \left(\norm{\hat{V}}_{1,\infty,5d+5}\right)^{|\alpha|+|\tilde{\alpha}|} .
	\end{aligned}
\end{equation}
Using that
    \begin{equation}\label{EQ:exp_aver_b_Mainterm_rec_11.5}
  	\frac{\langle\tilde{p} +q_{l_2}\rangle^{d+1}\langle \tilde{p} +q_{\sigma^2_{m^*}}\rangle^{d+1} \langle q_{l_2}\rangle^{d+1}\langle q_{\sigma^2_{m^*}}\rangle^{d+1} }{\langle\tilde{p} \rangle^{2d+2}\langle q_0\rangle^{4d+4}\langle \tilde{q}\rangle^{4d+4}}  \prod_{i=1}^n \frac{1}{\langle q_{\sigma_i^2}+\pi_{\sigma_i^2}^1(\tilde{q})-q_{\sigma_{i-1}^2} - \pi_{\sigma_i^2}^2(\tilde{q})\rangle^{4d+4}} \leq C,
  \end{equation}
what remain to estimate is the integral
\begin{equation}\label{EQ:exp_aver_b_Mainterm_rec_12}
	\begin{aligned}
	\MoveEqLeft   \int    \frac{\langle\nu\rangle^{-1}}{|\frac{1}{2} q_0^2-\tilde{\nu}-i\zeta|}\frac{\langle\tilde{\nu}\rangle^{-1}}{|\frac{1}{2} q_0^2+\nu +i\zeta|} \frac{ \langle\nu\rangle \langle\tilde{p} +q_{l_2}\rangle^{-d-1}}{| \frac{1}{2} (\tilde{p} +q_{l_1})^2 +\nu+i\zeta|} \frac{\langle q_{l_2}\rangle^{-d-1}\langle\tilde{q} \rangle^{-d-1}}{ |\frac{1}{2} q_{\sigma^2_{m^*}}^2-\tilde{\nu}-i\zeta|} 
	  \\
	  &\times   \frac{\langle\tilde{\nu}\rangle \langle q_{\sigma^2_{m^*}}\rangle^{-d-1} \langle \tilde{p} +q_{l_2}\rangle^{-d-1}}{| \frac{1}{2} (\tilde{p} +q_{l_2})^2 +\nu+i\zeta|}    \prod_{m=1}^n \frac{\langle q_{\sigma_i^2}+\pi_{\sigma_i}^1(\tilde{q})-q_{\sigma_{i-1}^2} - \pi_{\sigma_i}^2(\tilde{q})\rangle^{-d-1}}{ |\frac{1}{2} (q_{\sigma^2_m}+\pi_{\sigma^2_m}(\tilde{q}))^2-\tilde{\nu}-i\zeta|}
	  \,d\boldsymbol{q} d\tilde{p}d\tilde{q} d\nu d\tilde{\nu}.
	\end{aligned}
\end{equation}
To estimate this integral we first use Lemma~\ref{LE:est_res_combined} to do the integral in $\tilde{p}$. This gives us two powers of $|\log(\zeta)|$ and the factor $|q_{l_2}-q_{l_1}|^{-1}$. We then use the bound $\leq1$ and do the $q$ integrals starting from $q_n$ and skipping the integral in $q_{\sigma^2_{m^*}}$ using Lemma~\ref{LE:resolvent_int_est}. We then to the integral in $\tilde{q}$ $q_{\sigma^2_{m^*}}$ , $\nu$ and $\tilde{\nu}$ in that order using Lemma~\ref{LE:resolvent_int_est}. This gives us the estimate
\begin{equation}\label{EQ:exp_aver_b_Mainterm_rec_13}
	\begin{aligned}
	\MoveEqLeft   \int    \frac{\langle\nu\rangle^{-1}}{|\frac{1}{2} q_0^2-\tilde{\nu}-i\zeta|}\frac{\langle\tilde{\nu}\rangle^{-1}}{|\frac{1}{2} q_0^2+\nu +i\zeta|} \frac{ \langle\nu\rangle \langle\tilde{p} +q_{l_2}\rangle^{-d-1}}{| \frac{1}{2} (\tilde{p} +q_{l_1})^2 +\nu+i\zeta|} \frac{\langle q_{l_2}\rangle^{-d-1}\langle\tilde{q} \rangle^{-d-1}}{ |\frac{1}{2} q_{\sigma^2_{m^*}}^2-\tilde{\nu}-i\zeta|} 
	  \\
	  &\times   \frac{\langle\tilde{\nu}\rangle \langle q_{\sigma^2_{m^*}}\rangle^{-d-1} \langle \tilde{p} +q_{l_2}\rangle^{-d-1}}{| \frac{1}{2} (\tilde{p} +q_{l_2})^2 +\nu+i\zeta|}    \prod_{m=1}^n \frac{\langle q_{\sigma_i^2}+\pi_{\sigma_i}^1(\tilde{q})-q_{\sigma_{i-1}^2} - \pi_{\sigma_i}^2(\tilde{q})\rangle^{-d-1}}{ |\frac{1}{2} (q_{\sigma^2_m}+\pi_{\sigma^2_m}(\tilde{q}))^2-\tilde{\nu}-i\zeta|}
	  \,d\boldsymbol{q} d\tilde{p}d\tilde{q} d\nu d\tilde{\nu}
	  \leq C |\log(\zeta)|^{n+4}.
	\end{aligned}
\end{equation}
Combining the estimates in \cref{EQ:exp_aver_b_Mainterm_rec_8,EQ:exp_aver_b_Mainterm_rec_9,EQ:exp_aver_b_Mainterm_rec_10,EQ:exp_aver_b_Mainterm_rec_11,EQ:exp_aver_b_Mainterm_rec_11.5,EQ:exp_aver_b_Mainterm_rec_12,EQ:exp_aver_b_Mainterm_rec_13} we arrive at
   \begin{equation}\label{EQ:exp_aver_b_Mainterm_rec_14}
   	\begin{aligned}
	| \Aver{\mathcal{T}(n,\sigma^1,\sigma^2,\alpha,\tilde{\alpha},\mathrm{id})}|
	 \leq {}&   C_d^{a_k+\tilde{a}_k}  \left(\norm{\hat{V}}_{1,\infty,5d+5}\right)^{|\alpha|+|\tilde{\alpha}|} 
	  \hbar |\log(\zeta)|^{n+4} \frac{\rho(\rho t)^{2k-n-3} }{(k-2)!(k-n-1)!} \norm{\varphi}_{\mathcal{H}^{2d+2}_\hbar(\R^d)}^2.
	\end{aligned}
\end{equation}
The case where $m_2-1\notin\iota^{*}(\sigma^2)$ is estimated with an analogous argument and the same estimate as in \eqref{EQ:exp_aver_b_Mainterm_rec_14} is obtained.

We now turn to the case where $\kappa\neq\mathrm{id}$. As in the proof of Lemma~\ref{expansion_aver_bound_Mainterm} we will by $i^{*}$ denote the first number such that $\kappa(i)\neq i$.  For this case we also use the expression for $\mathcal{T}(n,\sigma^1,\sigma^2,\alpha,\tilde{\alpha},\kappa)$ given in \eqref{EQ:exp_aver_b_Mainterm_rec_5}, where the difference is that in the delta function we now have $\tilde{x}_{\sigma_{\kappa(i)}^2}$ instead of $\tilde{x}_{\sigma_{i}^2}$. Again we start by introducing the a function $f(\boldsymbol{s},\boldsymbol{\tilde{s}})$ of some of the time variables as above. Again it will have the form that the sum over some of the $s$ and $\tilde{s}$ variables is bounded by $\hbar^{-1}t$. For $\boldsymbol{s}$ it will not depend on $s_{\iota^{*}(i^{*})}$ and $s_{m_2-1}$. If $\iota^{*}(i^{*})=m_2-1$ it will not depend on $s_{m_2}$ and $s_{m_2-1}$. For $\tilde{s}$ it will not depend on $\tilde{s}_{m_2-1}$ and $\tilde{s}_m$ for $m\in \iota^{*}(\sigma^2)$. For the cases, where $m_2-1\in\iota^{*}(\sigma^2)$ it will also not depend on $s_{m_2}$. From here the argument is analogous to the argument we have just done combined with the arguments in the proof of Lemma~\ref{expansion_aver_bound_Mainterm} for the case, where $\kappa\neq\mathrm{id}$. Redoing the argument we get the bound 
   \begin{equation}\label{EQ:exp_aver_b_Mainterm_rec_15}
   	\begin{aligned}
	| \Aver{\mathcal{T}(n,\sigma^1,\sigma^2,\alpha,\tilde{\alpha},\kappa)}|
	 \leq   C_d^{a_k+\tilde{a}_k}  \norm{\hat{V}}_{1,\infty,5d+5}^{|\alpha|+|\tilde{\alpha}|} 
	  +\hbar |\log(\zeta)|^{n+4} \frac{\rho(\rho t)^{2k-n-3}}{(k-2)!(k-n-1)!} \norm{\varphi}_{\mathcal{H}^{2d+2}_\hbar(\R^d)}^2.
	\end{aligned}
\end{equation}
Combining \Cref{EQ:exp_aver_b_Mainterm_rec_1,EQ:exp_aver_b_Mainterm_rec_4,EQ:exp_aver_b_Mainterm_rec_14,EQ:exp_aver_b_Mainterm_rec_15}, using our assumptions on the coupling constant and the potential and argue as in the previous proofs we obtain the estimate
\begin{equation*}
	\begin{aligned}
	\MoveEqLeft \mathbb{E}\Big[\big\lVert \sum_{\boldsymbol{x}\in\mathcal{X}_{\neq}^{k-1}}\mathcal{I}_{1,1}(k,\boldsymbol{x},\iota,t;\hbar) \varphi\big\rVert_{L^2(\R^d)}^2 \Big] \leq \frac{ C^{k}}{(k-1)!}  \norm{\varphi}_{L^{2}(\R^d)}^2 +  \hbar k |\log(\zeta)|^{k+4} C^k  \norm{\varphi}_{\mathcal{H}^{2d+2}_\hbar(\R^d)}^2,
	\end{aligned}
\end{equation*}
where the constant $C$ depends on $\rho$, $t$, the potential and the coupling constant. This is the desired estimate and this concludes the proof.
\end{proof}
\begin{lemma}\label{expansion_aver_bound_Mainterm_rec_2}
  Assume we are in the setting of Definition~\ref{def_remainder_k_0}. Let $\varphi \in \mathcal{H}^{3d+3}_\hbar(\R^d)$. Then for any
  $\tau\in\{2,\dots,\tau_0\}$
    \begin{equation*}
	\begin{aligned}
	\MoveEqLeft \mathbb{E}\Big[ \big\lVert \sum_{k_1=1}^{k_0} \sum_{k_2=k_0-k_1+1}^{k_0} \sum_{\iota\in\mathcal{Q}_{k_1+k_2,1,1}} 
	  \sum_{(\boldsymbol{x}_1,\boldsymbol{x}_2)\in \mathcal{X}_{\neq}^{k_1+k_2-1}}  \mathcal{I}_{1,1}(k_1,\boldsymbol{x}_1,\iota,\tfrac{t}{\tau_0};\hbar)  \mathcal{I}_{1,1}(k_2,\boldsymbol{x}_2,\iota,\tfrac{\tau-1}{\tau_0}t;\hbar) \varphi\big\rVert_{L^2(\R^d)}^2 \Big]
	  \\
	  &
	  \leq  \frac{k_0^9 C^{k_0}}{\tau_0 (k_0-1)!} \norm{\varphi}_{L^{2}(\R^d)}^2 +   \hbar  \frac{k_0^{12}}{\tau_0} C^{k_0} |\log(\tfrac{\hbar}{t})|^{k_0+8} \norm{\varphi}_{\mathcal{H}^{3d+3}_\hbar(\R^d)}^2 ,
	 \end{aligned}
\end{equation*}
where the constant $C$ depends on $\rho$, $t$, the single site potential $V$ and the coupling constant $\lambda$. In particular we have that the function is in  $L^2(\R^d)$ $\Pro$-almost surely. 
  \end{lemma}
\begin{proof}
We start by observing that
  \begin{equation}\label{EQ:Aver_full_ex_1,1_dob1_1}
	\begin{aligned}
	\MoveEqLeft  \big\lVert\sum_{k_1=1}^{k_0} \sum_{k_2=k_0-k_1+1}^{k_0} \sum_{\iota\in\mathcal{Q}_{k_1+k_2,1,1}} 
	  \sum_{(\boldsymbol{x}_1,\boldsymbol{x}_2)\in \mathcal{X}_{\neq}^{k_1+k_2-1}}  \mathcal{I}_{1,1}(k_1,\boldsymbol{x}_1,\iota,\tfrac{t}{\tau_0};\hbar)  \mathcal{I}_{1,1}(k_2,\boldsymbol{x}_2,\iota,\tfrac{\tau-1}{\tau_0}t;\hbar) \varphi\big\rVert_{L^2(\R^d)}^2 
	  \\
	   &
	   \leq Ck_0^6 \sum_{k=k_0+1}^{2k_0}  \sum_{k_1+k_2=k} 
	   \sum_{\iota\in\mathcal{Q}_{k,1,1}} 
	   \big\lVert  \sum_{(\boldsymbol{x}_1,\boldsymbol{x}_2)\in \mathcal{X}_{\neq}^{k-1}}  \mathcal{I}_{1,1}(k_1,\boldsymbol{x}_1,\iota,\tfrac{t}{\tau_0};\hbar)  \mathcal{I}_{1,1}(k_2,\boldsymbol{x}_2,\iota,\tfrac{\tau-1}{\tau_0}t;\hbar) \varphi\big\rVert_{L^2(\R^d)}^2,
	 \end{aligned}
\end{equation}
where we have used that the number of elements in $\mathcal{Q}_{k,1,1}$ is bounded by a constant times $k^2$ for all $k$. This follows from Remark~\ref{Remark_grow_of_Q}. As in the proof of Lemma~\ref{expansion_aver_bound_Mainterm_2} we fix a $k=k_1+k_2$  and expand the $L^2$-norm and get that
\begin{equation*}
	\begin{aligned}
	\MoveEqLeft  \big\lVert \sum_{(\boldsymbol{x}_1,\boldsymbol{x}_2)\in \mathcal{X}_{\neq}^{k-1}}  \mathcal{I}_{1,1}(k_1,\boldsymbol{x}_1,\iota,\tfrac{t}{\tau_0};\hbar)  \mathcal{I}_{1,1}(k_2,\boldsymbol{x}_2,\iota,\tfrac{\tau-1}{\tau_0}t;\hbar) \varphi \big\rVert_{L^2(\R^d)}^2 = \sum_{\alpha,\tilde{\alpha}\in\N^k} (i\lambda)^\alpha(-i\lambda)^{\tilde{\alpha}}
	\\
	&\times \sum_{n=0}^{k-1} \sum_{\substack{n_1+n_2=n \\ n_1\leq k_1,n_2\leq k_2}} \sum_{\sigma^1\in\tilde{\mathcal{A}}(k_1,n_1,k_2,n_2)} \sum_{\substack{\tilde{n}_1+\tilde{n}_2=n \\ \tilde{n}_1\leq k_1,\tilde{n}_2\leq k_2}} \sum_{\sigma^2\in\tilde{\mathcal{A}}(k_1,\tilde{n}_1,k_2,\tilde{n}_2)}  \sum_{\kappa\in\mathcal{S}_n} \mathcal{T}(n,\sigma^1,\sigma^2,\alpha,\tilde{\alpha},\kappa),
	\end{aligned}
\end{equation*}
the sets $\tilde{\mathcal{A}}(k_1,n_1,k_2,n_2)$ is defined in \eqref{orderset_def_2}.
The numbers $\mathcal{T}(n,\sigma^1,\sigma^2,\alpha,\tilde{\alpha},\kappa)$ are given by  
\begin{equation*}
	\begin{aligned}
	\MoveEqLeft  \mathcal{T}(n,\sigma^1,\sigma^2,\alpha,\tilde{\alpha},\kappa)
	= \sum_{(\boldsymbol{x},\boldsymbol{\tilde{x}})\in \mathcal{X}_{\neq}^{2k}}  \prod_{i=1}^n \rho^{-1}\delta(x_{\sigma_i^1}- \tilde{x}_{\sigma_{\kappa(i)^2}}) \int_{\R^d} \mathcal{I}(k_1,\boldsymbol{x}_1,\alpha,\iota,\tfrac{t}{\tau_0};\hbar) 
	\\
	\times&   \mathcal{I}(k_2,\boldsymbol{x}_2,\alpha,\iota,\tfrac{\tau-1}{\tau_0}t;\hbar)\varphi(x)\overline{ \mathcal{I}(k_1,\boldsymbol{\tilde{x}}_1,\tilde{\alpha},\iota,\tfrac{t}{\tau_0};\hbar)  \mathcal{I}(k_2,\boldsymbol{\tilde{x}}_2,\tilde{\alpha},\iota,\tfrac{\tau-1}{\tau_0}t;\hbar)\varphi (x)}  \,dx,
	\end{aligned}
\end{equation*}
where again $\mathcal{I}(k_1,\boldsymbol{x}_1,\alpha,\iota,\tfrac{t}{\tau_0};\hbar)$ are defined as $ \mathcal{I}_{1,1}(k_1,\boldsymbol{x}_1,\iota,\tfrac{t}{\tau_0};\hbar)$ for $\alpha$ fixed. Again we distinguish the two cases $\kappa=\mathrm{id}$ and $\kappa\neq\mathrm{id}$. In the first case we will spilt into a further number of cases depending on the relations between $\sigma^1$, $\sigma^2$ and the map $\iota$ as in the proof of Lemma~\ref{expansion_aver_bound_Mainterm_rec}. To estimate each of these cases we make an argument that is a combination of the argument made in the proofs of Lemma~\ref{expansion_aver_bound_Mainterm_2} and Lemma~\ref{expansion_aver_bound_Mainterm_rec}. Making this argument we obtain the estimate
\begin{equation}\label{EQ:Aver_full_ex_1,1_dob1_3}
	\begin{aligned}
	\MoveEqLeft \mathbb{E}\Big[\big\lVert  \sum_{(\boldsymbol{x}_1,\boldsymbol{x}_2)\in \mathcal{X}_{\neq}^{k-1}}  \mathcal{I}_{1,1}(k_1,\boldsymbol{x}_1,\iota,\tfrac{t}{\tau_0};\hbar)  \mathcal{I}_{1,1}(k_2,\boldsymbol{x}_2,\iota,\tfrac{\tau-1}{\tau_0}t;\hbar) \varphi\big\rVert_{L^2(\R^d)}^2 \Big]
	\\
	&\leq C^{2k}\sum_{n=0}^{k-1} \Big[  \frac{ (\rho t)^{2k-2-n}  \norm{\varphi}_{L^{2}(\R^d)}^2 }{(k_1-1)!(k_2-1)!(k-1-n)!} +  \frac{ \hbar |\log(\tfrac{\hbar\tau_0}{t})|^{n+8} \rho(\rho t)^{2k-n-3}n!   \norm{\varphi}_{\mathcal{H}^{3d+3}_\hbar(\R^d)}^2 }{\tau_0^{k_1} (k_1-2)! (k_2-2)!(k-1-n)!}     \Big].
	\end{aligned}
\end{equation}
From combining \eqref{EQ:Aver_full_ex_1,1_dob1_1}, \eqref{EQ:Aver_full_ex_1,1_dob1_3} and arguing as in the proof of Lemma~\ref{expansion_aver_bound_Mainterm_2} we obtain the desired bound. 
This concludes the proof.
\end{proof}
\begin{lemma}\label{expansion_aver_bound_Mainterm_rec_3}
  Assume we are in the setting of Definition~\ref{def_remainder_k_0}. Let $\varphi \in \mathcal{H}^{3d+3}_\hbar(\R^d)$. Then for any
  $\tau\in\{2,\dots,\tau_0\}$
    \begin{equation*}
	\begin{aligned}
	\MoveEqLeft \mathbb{E}\Big[ \big\lVert \sum_{k_1=1}^{k_0} \sum_{k_2=k_0-k_1+1}^{k_0} \sum_{\iota\in\mathcal{Q}_{k_1+k_2,1,1}} 
	  \\
	  &\times \sum_{(\boldsymbol{x}_1,\boldsymbol{x}_2)\in \mathcal{X}_{\neq}^{k_1+k_2-1}}   \mathcal{I}_{1,1}(k_1+1,(x_{2,k_2},\boldsymbol{x}_1),\iota,\tfrac{t}{\tau_0};\hbar)  \mathcal{I}_{1,1}(k_2,\boldsymbol{x}_2,\iota,\tfrac{\tau-1}{\tau_0}t;\hbar)  \varphi\big\rVert_{L^2(\R^d)}^2 \Big]
	  \\
	   \leq{}& \frac{k_0^9 C^{k_0}}{\tau_0 (k_0-1)!} \norm{\varphi}_{L^{2}(\R^d)}^2 +   \hbar  \frac{k_0^{12}}{\tau_0} C^{k_0} |\log(\tfrac{\hbar}{t})|^{k_0+8} \norm{\varphi}_{\mathcal{H}^{3d+3}_\hbar(\R^d)}^2, 
	 \end{aligned}
\end{equation*}
where the constant $C$ depends on the single site potential $V$ and the coupling constant $\lambda$. In particular we have that the function is in  $L^2(\R^d)$ $\Pro$-almost surely. 
  \end{lemma}
\begin{proof}
The estimate is obtained by an argument analogous to that of Lemma~\ref{expansion_aver_bound_Mainterm_3} combined with the arguments used in the proof Lemma~\ref{expansion_aver_bound_Mainterm_rec_3}. 
\end{proof}
\subsection{Estimates for truncated terms}
\begin{lemma}\label{expansion_aver_bound_Mainterm_rec_trun}
Assume we are in the setting of Definition~\ref{functions_for_exp_rec_def}. Let  $\iota\in\mathcal{Q}_{k,1,1}$ and let $\varphi \in \mathcal{H}^{3d+3}_\hbar(\R^d)$. Then 
\begin{equation*}
	\begin{aligned}
	\MoveEqLeft \mathbb{E}\Big[\big\lVert \sum_{\iota\in\mathcal{Q}_{k_0+1,1,1}} \sum_{\boldsymbol{x}\in\mathcal{X}_{\neq}^{k_0}} \mathcal{E}_{1,1}^{k_0}(\boldsymbol{x},\iota,\tfrac{t}{\tau_0};\hbar)\varphi\big\rVert_{L^2(\R^d)}^2 \Big]
	\\
	&\leq  \frac{k_0^6 C^{k_0}}{ \tau_0^{k_0-1}\hbar (k_0-1)!} \norm{\varphi}_{L^{2}(\R^d)}^2 +   \frac{k_0^7 C^{k_0}}{ \tau_0^{k_0-2}}  |\log(\tfrac{\hbar}{t})|^{k_0+5} \norm{\varphi}_{\mathcal{H}^{2d+2}_\hbar(\R^d)}^2,
	\end{aligned}
\end{equation*}
where the constant $C$ depends on $\rho$, $t$, the single site potential $V$ and the coupling constant $\lambda$. In particular we have that the function is in  $L^2(\R^d)$ $\Pro$-almost surely. 
\end{lemma}
\begin{proof}
The estimates follows from combining the arguments in the proofs of Lemma~\ref{LE:Truncated_Without_re_bound_Mainterm_1} and Lemma~\ref{expansion_aver_bound_Mainterm_rec}. Moreover we also use that the number of maps in $\mathcal{Q}_{k,1,1}$ grows quadratic in $k$.
\end{proof}
The proof of the following Lemma also follow from combining arguments made in earlier proofs. Hence we will not write it down explicitly. 
\begin{lemma}\label{expansion_aver_bound_Mainterm_rec_trun_2}
  Assume we are in the setting of Definition~\ref{def_remainder_k_0}. Let $\varphi \in \mathcal{H}^{3d+3}_\hbar(\R^d)$. Then for any
  $\tau\in\{2,\dots,\tau_0\}$
    \begin{equation*}
	\begin{aligned}
	\MoveEqLeft \mathbb{E}\Big[ \big\lVert \sum_{k_2=1}^{k_0}  \sum_{\iota\in\mathcal{Q}_{k_0+k_2+1,1,1}}   \sum_{(\boldsymbol{x}_1,\boldsymbol{x}_2)\in \mathcal{X}_{\neq}^{k_0+k_2}}  \mathcal{E}_{1,1}^{k_0}(\boldsymbol{x}_1,\iota,\tfrac{t}{\tau_0};\hbar) \mathcal{I}_{1,1}(k_2,\boldsymbol{x}_2,\iota,\tfrac{\tau-1}{\tau_0}t;\hbar) \varphi\big\rVert_{L^2(\R^d)}^2 \Big]
	\\
	& \leq \frac{k_0^8 C^{k_0}}{ \tau_0^{k_0-1}\hbar (k_0-1)!} \norm{\varphi}_{L^{2}(\R^d)}^2 +   \frac{k_0^{11} C^{k_0}}{ \tau_0^{k_0-2}}  |\log(\tfrac{\hbar}{t})|^{k_0+8} \norm{\varphi}_{\mathcal{H}^{3d+3}_\hbar(\R^d)}^2,
	 \end{aligned}
\end{equation*}
and
    \begin{equation*}
	\begin{aligned}
	\MoveEqLeft \mathbb{E}\Big[ \big\lVert  \sum_{k_2=1}^{k_0}  \sum_{\iota\in\mathcal{Q}_{k_0+k_2+1,1,1}}   \sum_{(\boldsymbol{x}_1,\boldsymbol{x}_2)\in \mathcal{X}_{\neq}^{k_0+k_2}}  \mathcal{E}_{1,1}^{k_0+1}((x_{2,k_2},\boldsymbol{x}_1),\iota,\tfrac{t}{\tau_0};\hbar) \mathcal{I}_{1,1}(k_2,\boldsymbol{x}_2,\iota,\tfrac{\tau-1}{\tau_0}t;\hbar)  \varphi\big\rVert_{L^2(\R^d)}^2 \Big]
	  \\
	   &
	    \leq \frac{k_0^8 C^{k_0}}{ \tau_0^{k_0-1}\hbar (k_0-1)!} \norm{\varphi}_{L^{2}(\R^d)}^2 +   \frac{k_0^{11} C^{k_0}}{ \tau_0^{k_0-2}}  |\log(\tfrac{\hbar}{t})|^{k_0+8} \norm{\varphi}_{\mathcal{H}^{3d+3}_\hbar(\R^d)}^2,
	 \end{aligned}
\end{equation*}
where the constant $C$ depends on the single site potential $V$ and the coupling constant $\lambda$. In particular we have that the function is in  $L^2(\R^d)$ $\Pro$-almost surely. 
\end{lemma}
\section{Terms with two recollisions}\label{Sec:tech_est_2_recol}
We will in this section continue with the technical estimates which gives that the Duhamel expansion do converge in appropriate sense. We will here focus on the terms where we have observed two recollisions.
\subsection{Estimates for fully expanded terms}
\begin{lemma}\label{expansion_aver_bound_Mainterm_rec_21}
Assume we are in the setting of Definition~\ref{functions_for_exp_rec_def}. Let  $\iota\in\mathcal{Q}_{k,2,1}$ and let $\varphi \in \mathcal{H}^{4d+4}_\hbar(\R^d)$. Then 
\begin{equation}
	\begin{aligned}
	\MoveEqLeft \mathbb{E}\Big[\big\lVert \sum_{\boldsymbol{x}\in\mathcal{X}_{\neq}^{k-2}}\mathcal{I}_{2,1}(k,\boldsymbol{x},\iota,t;\hbar) \varphi\big\rVert_{L^2(\R^d)}^2 \Big] \leq  \frac{ C^{k_0}}{  (k_0-2)!} \norm{\varphi}_{L^{2}(\R^d)}^2 +   k_0 C^{k_0}\hbar  |\log(\tfrac{\hbar}{t})|^{k_0+7} \norm{\varphi}_{\mathcal{H}^{4d+4}_\hbar(\R^d)}^2,
	\end{aligned}
\end{equation}
where the constant $C$ depends on the single site potential $V$ and the coupling constant $\lambda$. In particular we have that the function is in  $L^2(\R^d)$ $\Pro$-almost surely. 
\end{lemma}
\begin{proof}
As in the proof of Lemma~\ref{expansion_aver_bound_Mainterm} we get two sums over the set $\mathcal{X}_{\neq}^{k-2}$, when we take the $L^2$-norm. Again we divide into different cases depending on how many points $x_m$ and $\tilde{x}_m$ they have in common and where they are placed. This yields
\begin{equation}\label{EQ:exp_aver_bound_Mainterm_rec_21_1}
	\big\lVert \sum_{\boldsymbol{x}\in\mathcal{X}_{\neq}^{k-2}} \mathcal{I}_{2,1}(k,\boldsymbol{x},\iota,t;\hbar) \varphi\big\rVert_{L^2(\R^d)}^2
	=\sum_{\alpha,\tilde{\alpha}\in\N^k}\sum_{n=0}^{k-2} \sum_{\sigma^1,\sigma^2\in\mathcal{A}(k-2,n)} \sum_{\kappa\in\mathcal{S}_n} (i\lambda)^{|\alpha|}(-i\lambda)^{|\tilde{\alpha}|} \mathcal{T}(n,\sigma^1,\sigma^2,\alpha,\tilde{\alpha},\kappa),
\end{equation}
where the numbers $\mathcal{T}(n,\sigma^1,\sigma^2,\alpha,\tilde{\alpha},\kappa)$ are given by
\begin{equation*}
	\begin{aligned}
	\MoveEqLeft \mathcal{T}(n,\sigma^1,\sigma^2,\alpha,\tilde{\alpha},\kappa) = \sum_{\boldsymbol{x},\tilde{\boldsymbol{x}}\in \mathcal{X}_{\neq}^{k-2}}    \prod_{i=1}^n \rho^{-1}\delta(x_{\sigma_i^1}- \tilde{x}_{\sigma_{\kappa(i)^2}})
	   \int\mathcal{I}(k,\boldsymbol{x},\alpha,\iota,t;\hbar)\varphi(x)\overline{\mathcal{I}(k,\boldsymbol{\tilde{x}},\tilde{\alpha},\iota,t;\hbar) \varphi(x)}  \,dx,
	\end{aligned}
\end{equation*}
where the operators $\mathcal{I}(k,\boldsymbol{x},\alpha,\iota,t;\hbar)$ is defined as $ \mathcal{I}_{2,1}(k,\boldsymbol{x},\iota,t;\hbar)$ just for a fixed $\alpha$. As in the previous proofs we will split the estimate into a number of cases. We will in all cases use the following form for the function we consider 
\begin{equation*}
	\begin{aligned}
 	\MoveEqLeft \mathcal{I}(k,\boldsymbol{x},\alpha,\iota,t;\hbar)\varphi(x) 
	=  \frac{1}{(2\pi\hbar)^d\hbar^{k}} \int_{\R_{+}^{|\alpha|}} \int \boldsymbol{1}_{[0,t]}( \boldsymbol{s}_{1,k}^{+}+ \hbar\boldsymbol{t}_{1,a_k}^{+})
	 e^{ i   \langle  \hbar^{-1}x,p_{k} \rangle} \prod_{i=1}^{a_k} e^{i  t_{i} \frac{1}{2} \eta_{i}^2}   \Big\{\prod_{m=1}^k  e^{i \hbar^{-1}  s_{m} \frac{1}{2} p_{m}^2} 
	 \\
	 &\times    
	 e^{ -i   \langle  \hbar^{-1/d}x_{\iota(m)},p_{m}-p_{m-1} \rangle}  \hat{\mathcal{V}}_{\alpha_m}(p_m,p_{m-1},\boldsymbol{\eta} )  \Big\}    e^{i  \hbar^{-1}(t- \boldsymbol{s}_{1,k}^{+}- \hbar \boldsymbol{t}_{1,a_k}^{+}) \frac{1}{2} p_0^2} \hat{\varphi}(\tfrac{p_0}{\hbar}) \, d\boldsymbol{\eta}d\boldsymbol{p}   d\boldsymbol{t} d\boldsymbol{s}.
	\end{aligned}
\end{equation*}
Recall that since $\iota\in\mathcal{Q}_{k,2,1}$ we have the numbers $m_1$, $m_2$ and $m_3$ associated to the map $\iota$ and the ``adjoint map'' $\iota^{*}$. We start by making the change of variables $p_{m_2-1}\mapsto p_{m_2-1} - p_{m_2}$ and $p_{m}\mapsto p_{m} - p_{m_2-1}$ for all $m\in\{m_1,\dots,m_2-2\}$, where $m_1$ and $m_2$ are the numbers associated to the map $\iota$. 
Moreover we also do the change of variables $p_{m_3-1}\mapsto p_{m_3-1}-p_{m_3}$ and $p_m\mapsto p_m-p_{m_3-1}$ for all $m\in\{m_1,\dots,m_3-2\}\setminus\{m_2-1\}$. This yields after a relabelling  
\begin{equation}\label{EQ:exp_aver_bound_Mainterm_rec_21_3}
	\begin{aligned}
 	& \mathcal{I} (k, \boldsymbol{x},\alpha,\iota,t;\hbar)\varphi(x) 
	=  \frac{1}{(2\pi\hbar)^d\hbar^{k}} \int_{\R_{+}^{|\alpha|}} \int \boldsymbol{1}_{[0,t]}( \boldsymbol{s}_{1,k}^{+}+ \hbar\boldsymbol{t}_{1,a_k}^{+})
	 e^{ i   \langle  \hbar^{-1}x,p_{k-2} \rangle} \prod_{i=1}^{a_k} e^{i  t_{i} \frac{1}{2} \eta_{i}^2}   
	 \\
	 &\times e^{i \hbar^{-1}  s_{m_2} \frac{1}{2} (p_{m_2-1}+\tilde{p}_2)^2}   \hat{\mathcal{V}}_{\alpha_{m_2}}(p_{m_2-1},\tilde{p}_1 +p_{m_2-1},\boldsymbol{\eta} ) e^{i \hbar^{-1}  s_{m_3} \frac{1}{2} p_{m_3-2}^2}   \hat{\mathcal{V}}_{\alpha_{m_3}}(p_{m_3-1},\tilde{p}_2 + p_{m_3-1},\boldsymbol{\eta} ) 
	 \\
	 &\times \Big\{\prod_{m=1}^{k-2}  e^{ -i   \langle  \hbar^{-1/d}x_{m},p_{m}-p_{m-1} \rangle}
	 e^{i \hbar^{-1}  s_{\iota^{*}(m)} \frac{1}{2} (p_{m}+\pi_m^1(\boldsymbol{\tilde{p}}))^2}   \hat{\mathcal{V}}_{\alpha_\iota^{*}(m)}(p_m+\pi_m^1(\boldsymbol{\tilde{p}}),p_{m-1}+\pi_m^2(\boldsymbol{\tilde{p}}),\boldsymbol{\eta} )  \Big\}    
	 \\
	 &\times e^{i  \hbar^{-1}(t- \boldsymbol{s}_{1,k}^{+}- \hbar \boldsymbol{t}_{1,a_k}^{+}) \frac{1}{2} p_0^2} \hat{\varphi}(\tfrac{p_0}{\hbar}) \, d\boldsymbol{\eta}d\boldsymbol{p}   d\boldsymbol{t} d\boldsymbol{s},
	\end{aligned}
\end{equation}
where we have used the following definition of the functions 
\begin{equation*}
\begin{aligned}
	\pi_{m}^1(\tilde{p} ) = \begin{cases}
	\tilde{p}_1 +\tilde{p}_2 & \text{if $m\in\{ m_1,\dots,m_2-1  \}$}
	\\
	\tilde{p}_2 & \text{if $m\in\{ m_2,\dots,m_3-2  \}$}
	\\
	0 &\text{otherwise}.
	\end{cases}
	\quad\quad
	\pi_{m}^2(\tilde{p} ) = \begin{cases}
	\tilde{p}_1 + \tilde{p}_2 & \text{if $m\in\{ m_1+1,\dots,m_2-1  \}$}
	\\
	 \tilde{p}_2 & \text{if $m\in\{ m_2,\dots,m_3-2  \}$}
	\\
	0 &\text{otherwise}.
	\end{cases}
	\end{aligned}
\end{equation*}
How we will estimate the different terms will depend on $\kappa$, and the relation between the two $\sigma$'s and $\kappa$. Firstly we split into the case if $\kappa\neq\mathrm{id}$ or $\kappa=\mathrm{id}$. We start with the case where $\kappa=\mathrm{id}$,  here we again divide into two different cases. First case is if there exist $i_1$ and $i_2$ in $\{1,n+1\}$ such that  following condition is satisfied 
\begin{equation*}
	\sigma_{i_1-1}^1\leq m_1<m_3 -2< \sigma_{i_1}^1 \quad\text{and}\quad \sigma_{i_2-1}^2 \leq m_1< m_3-2 < \sigma_{i_2}^2.
\end{equation*} 
Second case is the converse. Assume such two indices exist. For this case we get after applying Lemma~\ref{LE:Exp_ran_phases} and evaluating all integrals involving delta functions that
\begin{equation*}
	\begin{aligned}
	\MoveEqLeft \Aver{\mathcal{T}(n,\sigma^1,\sigma^2,\alpha,\tilde{\alpha},\mathrm{id}) }=  \frac{(\rho(2\pi)^d)^{2k-n-4}}{(2\pi\hbar)^d\hbar^{n+4}} \int_{\R_{+}^{|\alpha|+|\tilde{\alpha}|}} \int \boldsymbol{1}_{[0,t]}( \boldsymbol{s}_{1,k}^{+}+ \hbar\boldsymbol{t}_{1,a_k}^{+})  \boldsymbol{1}_{[0,t]}( \boldsymbol{\tilde{s}}_{1,k}^{+}+ \hbar\boldsymbol{\tilde{t}}_{1,\tilde{a}_k}^{+})
	 \\
	 &\times 
	  e^{i \hbar^{-1}  s_{m_2} \frac{1}{2} (p_{i_1-1}+\tilde{p}_2)^2}    e^{i \hbar^{-1}  s_{m_3} \frac{1}{2} p_{i_1-1}^2}  
	  e^{-i \hbar^{-1}  \tilde{s}_{m_2} \frac{1}{2} (p_{i_2-1}+\tilde{q}_2)^2}  e^{-i \hbar^{-1}  \tilde{s}_{m_3} \frac{1}{2} p_{i_2-1}^2}  
	   \prod_{i=1}^{a_k} e^{i  t_{i} \frac{1}{2} \eta_{i}^2}  
	    \prod_{i=1}^{\tilde{a}_k} e^{-i  \tilde{t}_{i} \frac{1}{2} \xi_{i}^2}  
	 \\
	 &\times \prod_{i=1}^{n+1} \prod_{m=\sigma_{i}^1+1}^{\sigma_{i}^1-1}  e^{i \hbar^{-1}  s_{\iota^{*}(m)} \frac{1}{2} (p_{i-1}+\pi_m^1(\boldsymbol{\tilde{p}}))^2}   
	  \prod_{m=\sigma_{i-1}^2+1}^{\sigma_{i}^2-1}  e^{-i \hbar^{-1}  \tilde{s}_{\iota^{*}(m)} \frac{1}{2} (p_{i-1}+\pi_m^1(\boldsymbol{\tilde{q}}))^2}   
	 \\
	 &\times 
	  \prod_{i=1}^{n} \Big\{ e^{i \hbar^{-1}  s_{\iota^{*}(\sigma_i^1)} \frac{1}{2} (p_{i}+\pi_{\sigma_i^1}^1(\boldsymbol{\tilde{p}}))^2} 
	  e^{-i \hbar^{-1}  \tilde{s}_{\iota^{*}(\sigma_i^2)} \frac{1}{2} (p_{i}+\pi_{\sigma_i^2}^1(\boldsymbol{\tilde{q}}))^2}  \Big\}
	e^{i  \hbar^{-1}(\boldsymbol{\tilde{s}}_{1,k}^{+}+  \hbar \boldsymbol{\tilde{t}}_{1,a_k}^{+} - \boldsymbol{s}_{1,k}^{+}- \hbar \boldsymbol{t}_{1,a_k}^{+}) \frac{1}{2} p_0^2}  
	\\&\times \mathcal{G}(\boldsymbol{\tilde{p}},\boldsymbol{\tilde{q}},\boldsymbol{p},\boldsymbol{\eta},\boldsymbol{\xi},\sigma^1,\sigma^2,\alpha,\tilde{\alpha})  |\hat{\varphi}(\tfrac{p_0}{\hbar})|^2 \, d\boldsymbol{\eta}d\boldsymbol{\xi}d\boldsymbol{p} d\tilde{p}d\tilde{q}   d\boldsymbol{t} d\boldsymbol{s}d\boldsymbol{\tilde{t}} d\boldsymbol{\tilde{s}},
	\end{aligned}
\end{equation*}
where we have used the notation
\begin{equation*}
	\begin{aligned}
	\MoveEqLeft \mathcal{G}(\boldsymbol{\tilde{p}},\boldsymbol{\tilde{q}},\boldsymbol{p},\boldsymbol{\eta},\boldsymbol{\xi},\sigma^1,\sigma^2,\alpha,\tilde{\alpha}) = \hat{\mathcal{V}}_{\alpha_{m_2}}(p_{i_1-1},\tilde{p}_1 +p_{i_1-1},\boldsymbol{\eta} )  \overline{\hat{\mathcal{V}}_{\alpha_{m_2}}(p_{i_2-1},\tilde{q}_1 +p_{i_2-1},\boldsymbol{\xi} )}
	\\
	& \times    \hat{\mathcal{V}}_{\alpha_{m_3}}(p_{i_1-1},\tilde{p}_2 + p_{i_1-1},\boldsymbol{\eta} )   
	 \Big\{ \prod_{i=1}^{n+1} \prod_{m=\sigma_{i-1}^1+1}^{\sigma_{i}^1-1}  \hat{\mathcal{V}}_{\alpha_{\iota^{*}(m)}}(p_{i-1}+\pi_m^1(\boldsymbol{\tilde{p}}),p_{i-1}+\pi_m^2(\boldsymbol{\tilde{p}}),\boldsymbol{\eta} ) 
	   \\
	   &\times \prod_{m=\sigma_{i-1}^2+1}^{\sigma_{i}^2-1}  \overline{\hat{\mathcal{V}}_{\alpha_{\iota^{*}}(m)}(p_{i-1}+\pi_m^1(\boldsymbol{\tilde{q}}),p_{i-1}+\pi_m^2(\boldsymbol{\tilde{q}}),\boldsymbol{\xi} )} \Big\} \overline{ \hat{\mathcal{V}}_{\alpha_{m_3}}(p_{i_2-1},\tilde{q}_2 + p_{i_2-1},\boldsymbol{\xi} ) }
	  \\
	 &\times \prod_{i=1}^{n}  \hat{\mathcal{V}}_{\alpha_{\iota^{*}}(\sigma_i^1)}(p_{i} + \pi_{\sigma_i^1}^1(\boldsymbol{\tilde{p}}) ,p_{i-1} + \pi_{\sigma_i^1}^2(\boldsymbol{\tilde{p}}),\boldsymbol{\eta} )   
	  \overline{ \hat{\mathcal{V}}_{\alpha_{\iota^{*}}(\sigma_i^2)}(p_{i}+\pi_{\sigma_i^2}^1(\boldsymbol{\tilde{q}}),p_{i-1}+\pi_{\sigma_i^2}^1(\boldsymbol{\tilde{q}}),\boldsymbol{\xi} ) }.   
	\end{aligned}
\end{equation*}
We note that $\pi_m^1(\boldsymbol{\tilde{p}})\neq0$ and $\pi_m^2(\boldsymbol{\tilde{p}})\neq0$ for $ \sigma_{i_1-1}^1\leq m < \sigma_{i_1}^1$ and $\pi_m^1(\boldsymbol{\tilde{q}})\neq0$ and $\pi_m^2(\boldsymbol{\tilde{q}})\neq0$ for $ \sigma_{i_1-1}^2\leq m < \sigma_{i_1}^2$. Hence we preform the following changes of variables $\tilde{p}_1 \mapsto \tilde{p}_1 + \tilde{p}_2 + p_{i_1-1}$, $\tilde{p}_2 \mapsto \tilde{p}_2 + p_{i_1-1}$, $\tilde{q}_1 \mapsto \tilde{q}_1 + \tilde{q}_2 + p_{i_2-1}$ and $\tilde{q}_2 \mapsto \tilde{q}_2 + p_{i_2-1}$. This change of variables ensures that in all quadratic phases we only have one variable. Hence we can argue as in the proof of Lemma~\ref{expansion_aver_bound_Mainterm} to ensure integrability in the $t$'s and $\tilde{t}$'s and case by case use Lemma~\ref{app_quadratic_integral_tech_est} directly for the all variables $p_1,\dots,p_n,\tilde{p}_1,\tilde{p}_2,\tilde{q}_1,\tilde{q}_2$. Preforming these arguments we arrive at the bound
   \begin{equation}\label{EQ:exp_aver_bound_Mainterm_rec_21_5}
   	\begin{aligned}
	 \Aver{\mathcal{T}(n,\sigma^1,\sigma^2,\alpha,\tilde{\alpha},\mathrm{id})}
	 \leq  C_d^{a_k +\tilde{a}_k+n+4}  \norm{\varphi}_{L^2(\R^d)}^2\norm{\hat{V}}_{1,\infty,3d+3}^{|\alpha|+|\tilde{\alpha}|}  \frac{ (\rho t)^{2k-n-4}}{(k-2)!(k-n-2)!}.
	\end{aligned}
\end{equation}
Now we turn to the case where such two indices does not exists. Here we have that there exists an $i^{*}$ such that
\begin{equation}
	 m_1< \sigma_{i^{*}}^1\leq m_3 -2  \quad\text{or}\quad  m_1< \sigma_{i^{*}}^2\leq m_3-2 .
\end{equation} 
For this case we will use the ``$\nu$''-representation of the operators. By again first introducing a suitable chosen $f(\boldsymbol{s},\boldsymbol{\tilde{s}})$ inserting this function and introduce the additional variable $s_0$. We introduce the function $\zeta$ again and insert the functions $e^{(\hbar^{-1}t- \boldsymbol{s}_{0,k}^{+}- \boldsymbol{t}_{1,a_k}^{+})\zeta(\hbar^{-1}t)}$ and $e^{(\hbar^{-1}t- \boldsymbol{\tilde{s}}_{0,k}^{+}- \boldsymbol{\tilde{t}}_{1,a_k}^{+})\zeta(\hbar^{-1}t)}$. Then we do an analogous argument to that of Lemma~\ref{expansion_aver_bound_Mainterm_rec} to obtain the bound 
   \begin{equation}\label{EQ:exp_aver_bound_Mainterm_rec_21_6}
   	\begin{aligned}
	 \Aver{\mathcal{T}(n,\sigma^1,\sigma^2,\alpha,\tilde{\alpha},\mathrm{id})}
	 \leq {}&   C_d^{a_k+\tilde{a}_k+n} \norm{\hat{V}}_{1,\infty,5d+5}^{|\alpha|+|\tilde{\alpha}|} 
	  + \hbar |\log(\zeta)|^{n+7} \frac{\rho(\rho t)^{2k-n-5} }{(k-3)!(k-n-2)!} \norm{\varphi}_{\mathcal{H}^{4d+4}_\hbar(\R^d)}^2.
	\end{aligned}
\end{equation}
What remains is the case where $\kappa\neq\mathrm{id}$. For this case we need just to keep track of the first $i$ such that $\kappa(i)\neq i$ and use an analogous argument to that used in the proof of Lemma~\ref{expansion_aver_bound_Mainterm} for the case, where $\kappa\neq\mathrm{id}$. This yields the bound
 \begin{equation}\label{EQ:exp_aver_bound_Mainterm_rec_21_7}
   	\begin{aligned}
	 \Aver{\mathcal{T}(n,\sigma^1,\sigma^2,\alpha,\tilde{\alpha},\kappa)}
	 \leq {}&   C_d^{a_k+\tilde{a}_k+n} \norm{\hat{V}}_{1,\infty,5d+5}^{|\alpha|+|\tilde{\alpha}|} 
	  + \hbar |\log(\zeta)|^{n+7} \frac{\rho(\rho t)^{2k-n-5} }{(k-3)!(k-n-2)!} \norm{\varphi}_{\mathcal{H}^{4d+4}_\hbar(\R^d)}^2.
	\end{aligned}
\end{equation}
By combining \eqref{EQ:exp_aver_bound_Mainterm_rec_21_1}, \eqref{EQ:exp_aver_bound_Mainterm_rec_21_5}, \eqref{EQ:exp_aver_bound_Mainterm_rec_21_6},  \eqref{EQ:exp_aver_bound_Mainterm_rec_21_7} and arguing as in the proofs for the previous estimates  we get the estimate
\begin{equation}
	\begin{aligned}
	\mathbb{E}\Big[\big\lVert \sum_{\boldsymbol{x}\in\mathcal{X}_{\neq}^{k-2}}\mathcal{I}_{2,1}(k,\boldsymbol{x},\iota,t;\hbar) \varphi\big\rVert_{L^2(\R^d)}^2 \Big] \leq  \frac{ C^{k_0}}{  (k_0-2)!} \norm{\varphi}_{L^{2}(\R^d)}^2 +   k_0 C^{k_0}\hbar  |\log(\tfrac{\hbar}{t})|^{k_0+7} \norm{\varphi}_{\mathcal{H}^{4d+4}_\hbar(\R^d)}^2.
	\end{aligned}
\end{equation}
This concludes the proof.
\end{proof}
\begin{lemma}\label{expansion_aver_bound_Mainterm_rec_22}
Assume we are in the setting of Definition~\ref{functions_for_exp_rec_def}. Let  $\iota\in\mathcal{Q}_{k,2,2}$ and let $\varphi \in \mathcal{H}^{4d+4}_\hbar(\R^d)$. Then 
\begin{equation}
	\begin{aligned}
	\MoveEqLeft \mathbb{E}\Big[\big\lVert \sum_{\boldsymbol{x}\in\mathcal{X}_{\neq}^{k-2}}\mathcal{I}_{2,2}(k,\boldsymbol{x},\iota,t;\hbar) \varphi\big\rVert_{L^2(\R^d)}^2 \Big] \leq \frac{ C^{k_0}}{  (k_0-2)!} \norm{\varphi}_{L^{2}(\R^d)}^2 +   k_0 C^{k_0}\hbar  |\log(\tfrac{\hbar}{t})|^{k_0+7} \norm{\varphi}_{\mathcal{H}^{4d+4}_\hbar(\R^d)}^2,
	\end{aligned}
\end{equation}
where the constant $C$ depends on $\rho$, $t$, the single site potential $V$ and the coupling constant $\lambda$. In particular we have that the function is in  $L^2(\R^d)$ $\Pro$-almost surely. 
\end{lemma}
\begin{proof}
The proof is analogous to that of Lemma~\ref{expansion_aver_bound_Mainterm_rec_21}. The same form for the function $\mathcal{I}_{2,2}^{\mathrm{rec}}(k,\boldsymbol{x},\alpha,\iota,t;\hbar) \varphi$ is used here. However, we here do the change of variables variables $p_{m_{1,2}-1}\mapsto p_{m_{1,2}-1} - p_{m_{1,2}}$ and $p_{m}\mapsto p_{m} - p_{m_{1,2}-1}$ for all $m\in\{m_{1,2},\dots,m_{1,2}-2\}$, where $m_{1,1}$ and $m_{1,2}$ are numbers associated to the map $\iota$. 
Moreover we also do the change of variables $p_{m_{2,2}-1}\mapsto p_{m_{2,2}-1}-p_{m_{2,2}}$ and $p_m\mapsto p_m-p_{m_{2,2}-1}$ for all $m\in\{m_{2,1},\dots,m_{2,2}-2\}\setminus\{m_{1,2}-1\}$, where  again $m_{2,1}$ and $m_{2,2}$ are numbers associated to the map $\iota$. After this change of variable and a relabelling the proof proceeds as the proof of Lemma~\ref{expansion_aver_bound_Mainterm_rec_21} by dividing into different cases depending on the relations between the pairing relations and the numbers associated to $\iota$.  
\end{proof}
We will omit the proofs of the following Lemmas. Both proofs are done by combining arguments from previous proofs.
\begin{lemma}\label{expansion_aver_bound_Mainterm_rec_21_2}
  Assume we are in the setting of Definition~\ref{def_remainder_k_0}. Let $\varphi \in \mathcal{H}^{5d+5}_\hbar(\R^d)$. Then for any
  $\tau\in\{2,\dots,\tau_0\}$
    \begin{equation*}
	\begin{aligned}
	\MoveEqLeft \mathbb{E}\Big[ \big\lVert \sum_{k_1=1}^{k_0} \sum_{k_2=k_0-k_1+1}^{k_0} \sum_{\iota\in\mathcal{Q}_{k_1+k_2,2,1}} 
	  \sum_{(\boldsymbol{x}_1,\boldsymbol{x}_2)\in \mathcal{X}_{\neq}^{k_1+k_2-2}}   \mathcal{I}_{2,1}(k_1,\boldsymbol{x}_1,\iota,\tfrac{t}{\tau_0};\hbar)  \mathcal{I}_{2,1}(k_2,\boldsymbol{x}_2,\iota,\tfrac{\tau-1}{\tau_0}t;\hbar)  \varphi\big\rVert_{L^2(\R^d)}^2 \Big]
	  \\
	  &
	   \leq \frac{k_0^{12} C^{k_0}}{ \tau_0 (k_0-2)!} \norm{\varphi}_{L^{2}(\R^d)}^2 +  \frac{ k_0^{13} C^{k_0}\hbar }{\tau_0} |\log(\tfrac{\hbar}{t})|^{k_0+9} \norm{\varphi}_{\mathcal{H}^{5d+5}_\hbar(\R^d)}^2,
	 \end{aligned}
\end{equation*}
and
  \begin{equation*}
	\begin{aligned}
	\MoveEqLeft \mathbb{E}\Big[ \big\lVert \sum_{k_1=1}^{k_0} \sum_{k_2=k_0-k_1+1}^{k_0} \sum_{\iota\in\mathcal{Q}_{k_1+k_2,2,1}} 
	 \\
	 &\times  \sum_{(\boldsymbol{x}_1,\boldsymbol{x}_2)\in \mathcal{X}_{\neq}^{k_1+k_2-2}}  \mathcal{I}_{2,1}(k_1+1,(x_{2,k_2},\boldsymbol{x}_1),\iota,\tfrac{t}{\tau_0};\hbar)  \mathcal{I}_{2,1}(k_2,\boldsymbol{x}_2,\iota,\tfrac{\tau-1}{\tau_0}t;\hbar)   \varphi\big\rVert_{L^2(\R^d)}^2 \Big]
	  \\
	   \leq{}&  \frac{k_0^{12} C^{k_0}}{ \tau_0 (k_0-2)!} \norm{\varphi}_{L^{2}(\R^d)}^2 +  \frac{ k_0^{13} C^{k_0}\hbar}{\tau_0}  |\log(\tfrac{\hbar}{t})|^{k_0+9} \norm{\varphi}_{\mathcal{H}^{5d+5}_\hbar(\R^d)}^2,
	 \end{aligned}
\end{equation*}
where the constant $C$ depends on $\rho$, $t$, the single site potential $V$ and the coupling constant $\lambda$. In particular we have that the function is in  $L^2(\R^d)$ $\Pro$-almost surely. 
\end{lemma}
\begin{lemma}\label{expansion_aver_bound_Mainterm_rec_22_2}
  Assume we are in the setting of Definition~\ref{def_remainder_k_0}. Let $\varphi \in \mathcal{H}^{5d+5}_\hbar(\R^d)$. Then for any
  $\tau\in\{2,\dots,\tau_0\}$
    \begin{equation*}
	\begin{aligned}
	\MoveEqLeft \mathbb{E}\Big[ \big\lVert \sum_{k_1=1}^{k_0} \sum_{k_2=k_0-k_1+1}^{k_0} \sum_{\iota\in\mathcal{Q}_{k_1+k_2,2,2}} 
	  \sum_{(\boldsymbol{x}_1,\boldsymbol{x}_2)\in \mathcal{X}_{\neq}^{k_1+k_2-2}}   \mathcal{I}_{2,2}(k_1,\boldsymbol{x}_1,\iota,\tfrac{t}{\tau_0};\hbar)  \mathcal{I}_{2,2}(k_2,\boldsymbol{x}_2,\iota,\tfrac{\tau-1}{\tau_0}t;\hbar)  \varphi\big\rVert_{L^2(\R^d)}^2 \Big]
	  \\
	  &
	  \leq  \frac{k_0^{15} C^{k_0}}{ \tau_0 (k_0-2)!} \norm{\varphi}_{L^{2}(\R^d)}^2 +  \frac{ k_0^{16} C^{k_0}\hbar}{\tau_0}  |\log(\tfrac{\hbar}{t})|^{k_0+10} \norm{\varphi}_{\mathcal{H}^{5d+5}_\hbar(\R^d)}^2,
	 \end{aligned}
\end{equation*}
and
  \begin{equation*}
	\begin{aligned}
	\MoveEqLeft  \mathbb{E}\Big[ \big\lVert \sum_{k_1=1}^{k_0} \sum_{k_2=k_0-k_1+1}^{k_0} \sum_{\iota\in\mathcal{Q}_{k_1+k_2,2,2}} 
	  \\
	  &\times \sum_{(\boldsymbol{x}_1,\boldsymbol{x}_2)\in \mathcal{X}_{\neq}^{k_1+k_2-2}}  \mathcal{I}_{2,2}(k_1+1,(x_{2,k_2},\boldsymbol{x}_1),\iota,\tfrac{t}{\tau_0};\hbar)  \mathcal{I}_{2,1}(k_2,\boldsymbol{x}_2,\iota,\tfrac{\tau-1}{\tau_0}t;\hbar)   \varphi\big\rVert_{L^2(\R^d)}^2 \Big]
	  \\
	   \leq{}&  \frac{k_0^{15} C^{k_0}}{ \tau_0 (k_0-2)!} \norm{\varphi}_{L^{2}(\R^d)}^2 +  \frac{ k_0^{16} C^{k_0}\hbar}{\tau_0}  |\log(\tfrac{\hbar}{t})|^{k_0+10} \norm{\varphi}_{\mathcal{H}^{5d+5}_\hbar(\R^d)}^2,
	 \end{aligned}
\end{equation*}
where the constant $C$ depends on the single site potential $V$ and the coupling constant $\lambda$. In particular we have that the function is in  $L^2(\R^d)$ $\Pro$-almost surely. 
\end{lemma}%
\subsection{Estimates for truncated terms}
\begin{lemma}\label{expansion_aver_bound_Mainterm_rec_trun2j}
Assume we are in the setting of Definition~\ref{functions_for_exp_rec_def}. Let $j\in\{1,2\}$ and let $\varphi \in \mathcal{H}^{4d+4}_\hbar(\R^d)$. Then 
\begin{equation*}
	\begin{aligned}
	\MoveEqLeft \mathbb{E}\Big[\big\lVert \sum_{\iota\in\mathcal{Q}_{k_0+1,2,j}} \sum_{\boldsymbol{x}\in\mathcal{X}_{\neq}^{k_0-1}} \mathcal{E}_{2,j}^{k_0}(\boldsymbol{x},\iota,t;\hbar)\varphi\big\rVert_{L^2(\R^d)}^2 \Big]
	\\
	&\leq \frac{k_0^{13} C^{k_0}}{ \tau_0^{k_0-2}\hbar (k_0-2)!} \norm{\varphi}_{L^{2}(\R^d)}^2 +   \frac{k_0^{16} C^{k_0}}{ \tau_0^{k_0-4}} \hbar  |\log(\tfrac{\hbar}{t})|^{k_0+7} \norm{\varphi}_{\mathcal{H}^{4d+4}_\hbar(\R^d)}^2,
	\end{aligned}
\end{equation*}
where the constant $C$ depends on the single site potential $V$ and the coupling constant $\lambda$. In particular we have that the function is in  $L^2(\R^d)$ $\Pro$-almost surely. 
\end{lemma}
\begin{proof}
The proof is done by combining the arguments of the previous proofs. We do remark that in the case $j=1$ we can get $k_0^9$ instead of $k_0^{12}$ but we have choosen to give the general estimate.
\end{proof}
The proof of the following Lemma is also done by combining arguments from previous proofs. Here we have also chosen to give a general estimate for both types of configurations. For $j=1$ we could have had $k_0^{11}$ and not $k_0^{14}$. 
\begin{lemma}\label{expansion_aver_bound_Mainterm_rec_trun_j_2}
  Assume we are in the setting of Definition~\ref{def_remainder_k_0}. Let $j\in\{1,2\}$ and let $\varphi \in \mathcal{H}^{5d+5}_\hbar(\R^d)$. Then for any
  $\tau\in\{2,\dots,\tau_0\}$
    \begin{equation*}
	\begin{aligned}
	\MoveEqLeft \mathbb{E}\Big[ \big\lVert \sum_{k_2=1}^{k_0}  \sum_{\iota\in\mathcal{Q}_{k_0+k_2+1,2,j}}   \sum_{(\boldsymbol{x}_1,\boldsymbol{x}_2)\in \mathcal{X}_{\neq}^{k_0+k_2-1}}  \mathcal{E}_{2,j}^{k_0}(\boldsymbol{x}_1,\iota,\tfrac{t}{\tau_0};\hbar) \mathcal{I}_{2,j}(k_2,\boldsymbol{x}_2,\iota,\tfrac{\tau-1}{\tau_0}t;\hbar) \varphi\big\rVert_{L^2(\R^d)}^2 \Big]
	  \\
	  &\leq \frac{k_0^{15} C^{k_0}}{ \tau_0^{k_0-2}\hbar (k_0-2)!} \norm{\varphi}_{L^{2}(\R^d)}^2 +   \frac{k_0^{17} C^{k_0}}{ \tau_0^{k_0-4}} \hbar  |\log(\tfrac{\hbar}{t})|^{k_0+7} \norm{\varphi}_{\mathcal{H}^{5d+5}_\hbar(\R^d)}^2,
	 \end{aligned}
\end{equation*}
and
    \begin{equation*}
	\begin{aligned}
	\MoveEqLeft \mathbb{E}\Big[ \big\lVert  \sum_{k_2=1}^{k_0}  \sum_{\iota\in\mathcal{Q}_{k_0+k_2+1,2,j}}   \sum_{(\boldsymbol{x}_1,\boldsymbol{x}_2)\in \mathcal{X}_{\neq}^{k_0+k_2-1}}  \mathcal{E}_{2,j}^{k_0+1}((x_{2,k_2},\boldsymbol{x}_1),\iota,\tfrac{t}{\tau_0};\hbar) \mathcal{I}_{2,j}(k_2,\boldsymbol{x}_2,\iota,\tfrac{\tau-1}{\tau_0}t;\hbar)  \varphi\big\rVert_{L^2(\R^d)}^2 \Big]
	  \\
	   &
	   \leq \frac{k_0^{15} C^{k_0}}{ \tau_0^{k_0-2}\hbar (k_0-2)!} \norm{\varphi}_{L^{2}(\R^d)}^2 +   \frac{k_0^{17} C^{k_0}}{ \tau_0^{k_0-4}} \hbar  |\log(\tfrac{\hbar}{t})|^{k_0+7} \norm{\varphi}_{\mathcal{H}^{5d+5}_\hbar(\R^d)}^2,
	 \end{aligned}
\end{equation*}
where the constant $C$ depends on the single site potential $V$ and the coupling constant $\lambda$. In particular we have that the function is in  $L^2(\R^d)$ $\Pro$-almost surely. 
\end{lemma}
\section{Recollision error terms for $\mathcal{Q}_{k,2,1}$}\label{Sec:tech_est_truncated_recol_1}
We will in this section prove estimates for some of the terms in the ``recollision''-error. The techniques used will mostly be the same as in the previous section. We will however in this case  need to use some different techniques when we estimate the terms, where $\kappa$ is the identity. 
\begin{lemma}\label{expansion_aver_bound_Mainterm_rec_trun21_1}
Assume we are in the setting of Definition~\ref{def_recol_reminder}. Let $\varphi \in \mathcal{H}^{5d+5}_\hbar(\R^d)$. Then 
\begin{equation*}
	\begin{aligned}
	\mathbb{E}\Big[\big\lVert \sum_{k_1=5}^{k_0} \sum_{\iota\in\mathcal{Q}_{k_1,2,1}}   \sum_{\boldsymbol{x}_1\in \mathcal{X}_{\neq}^{k_1-2}}  \mathcal{E}_{2,1}^{\mathrm{rec}}(k_1,0,\boldsymbol{x}_1,\iota,t;\hbar)\varphi\big\rVert_{L^2(\R^d)}^2 \Big]
	\leq k_0^{16} C^{2k_0}  |\log(\tfrac{\hbar}{\tau_0})|^{k_0+7}  \norm{\varphi}_{\mathcal{H}^{5d+5}_\hbar(\R^d)}^2  \hbar,
	\end{aligned}
\end{equation*}
where the constant $C$ depends on the single site potential $V$ and the coupling constant $\lambda$. In particular we have that the function is in  $L^2(\R^d)$ $\Pro$-almost surely. 
\end{lemma}
\begin{proof}
From the definition of the operator $\mathcal{E}_{2,1}^{\mathrm{rec}}(k_1,0,\boldsymbol{x}_1,\iota,t;\hbar)$ (Definition~\ref{def_recol_error}) we get that
\begin{equation}\label{EQ:expansion_aver_bound_3.recol_21_0}
	\begin{aligned}
		\MoveEqLeft \big\lVert \sum_{k=5}^{k_0} \sum_{\iota\in\mathcal{Q}_{k},2,1}   \sum_{\boldsymbol{x}\in \mathcal{X}_{\neq}^{k-2}}  \mathcal{E}_{2,1}^{\mathrm{rec}}(k,0,\boldsymbol{x},\iota,t;\hbar)\varphi\big\rVert_{L^2(\R^d)}^2
		\\
		\leq{}& \frac{Ck_0^8 t}{\hbar^2}  \sum_{k=5}^{k_0}   \sum_{\iota\in\mathcal{Q}_{k},2,1} \sum_{l=1, \iota(l)\neq k}^{k-2}  \int_{0}^{t}  \big\lVert \sum_{\alpha\in \N^k}(i\lambda)^{|\alpha|+1}     \sum_{\boldsymbol{x}\in \mathcal{X}_{\neq}^{k-2}}  \tilde{\mathcal{E}}_{2,j}^{\mathrm{rec}}(s_{k+1},k,\boldsymbol{x},\alpha,l,t,\iota;\hbar)\varphi\big\rVert_{L^2(\R^d)}^2 \,ds_{k+1},
	\end{aligned}
\end{equation}
where we have that
\begin{equation*}
	\begin{aligned}
	\tilde{\mathcal{E}}_{2,1}^{\mathrm{rec}}(s_{k+1},k,\boldsymbol{x},\alpha,l,t,\iota;\hbar)
	=    \int_{[0,t]_{\leq}^{k}} \boldsymbol{1}_{[s_{k+1},t]}(s_k)   V^{s_{k+1}}_{\hbar, x_{l}} \prod_{m=1}^k \Theta_{\alpha_m}(s_{m-1},{s}_{m},x_{\iota(m+k_1)};V,\hbar)  \, d\boldsymbol{s}_{k,1}U_{\hbar,0}(-t). 
	\end{aligned}
\end{equation*}
As in the previous proofs we get two sums over the set $\mathcal{X}_{\neq}^{k-2}$, when we take the $L^2$-norm. Again we divide into different cases depending on how many points $x_m$ and $\tilde{x}_m$ they have in common and where they are placed. This yields
\begin{equation}\label{EQ:expansion_aver_bound_3.recol_21_1}
	\begin{aligned}
	\MoveEqLeft \int_0^t \big\lVert \sum_{\alpha\in \N^k}(i\lambda)^{|\alpha|+1}     \sum_{\boldsymbol{x}\in \mathcal{X}_{\neq}^{k-2}}  \tilde{\mathcal{E}}_{2,j}^{\mathrm{rec}}(s_{k+1},k,\boldsymbol{x},\alpha,l,t,\iota;\hbar)\varphi\big\rVert_{L^2(\R^d)}^2\,ds_{k+1}
	\\
	&= \sum_{\alpha,\tilde{\alpha}\in \N^k} \sum_{n=0}^{k-2} \sum_{\sigma^1,\sigma^2\in\mathcal{A}(k-2,n)} \sum_{\kappa\in\mathcal{S}_n} \mathcal{T}(n,\sigma^1,\sigma^2,\alpha,\tilde{\alpha},\kappa),
	\end{aligned}
\end{equation}
where the numbers $\mathcal{T}(n,\sigma^1,\sigma^2,\alpha,\tilde{\alpha},\kappa)$ are given by
\begin{equation*}
	\begin{aligned}
	\MoveEqLeft \mathcal{T}(n,\sigma^1,\sigma^2,\alpha,\tilde{\alpha},\kappa) = \sum_{\boldsymbol{x},\tilde{\boldsymbol{x}}\in \mathcal{X}_{\neq}^{k-2}}    \prod_{i=1}^n \rho^{-1}\delta(x_{\sigma_i^1}- \tilde{x}_{\sigma_{\kappa(i)^2}})
	\\
	&\times  \int_0^t \int \tilde{\mathcal{E}}_{2,1}^{\mathrm{rec}}(s_{k+1},k,\boldsymbol{x},\alpha,l,t,\iota;\hbar)\varphi(x)\overline{\tilde{\mathcal{E}}_{2,1}^{\mathrm{rec}}(s_{k+1},k,\boldsymbol{\tilde{x}},\tilde{\alpha},l,t,\iota;\hbar) \varphi(x)}  \,dx ds_{k+1}.
	\end{aligned}
\end{equation*}
As in the previous proofs we will split the estimate into a number of cases. However we will in all cases use the following form for the function we consider 
\begin{equation*}
	\begin{aligned}
 	\MoveEqLeft \tilde{\mathcal{E}}_{2,1}^{\mathrm{rec}}(s_{k+1},k,\boldsymbol{x},\alpha,l,t,\iota;\hbar)\varphi(x) 
	=  \frac{1}{(2\pi\hbar)^d\hbar^{k}} \int_{\R_{+}^{|\alpha|}} \int \boldsymbol{1}_{[0,t]}( \boldsymbol{s}_{1,k+1}^{+}+ \hbar\boldsymbol{t}_{1,a_k}^{+})
	 e^{ i   \langle  \hbar^{-1}x,p_{k+1} \rangle} 
	  e^{ -i   \langle  \hbar^{-1/d}x_{l},p_{k+1}-p_{k} \rangle}  
	  \\
	  &\times e^{i  s_{k+1} \frac{1}{2}\hbar^{-1} p_{k+1}^2}  \hat{V}(p_{k+1}-p_k)  
	  \prod_{i=1}^{a_k} e^{i  t_{i} \frac{1}{2} \eta_{i}^2} \prod_{m=1}^k  e^{ -i   \langle  \hbar^{-1/d}x_{\iota(m)},p_{m}-p_{m-1} \rangle}  
	 e^{i \hbar^{-1}  s_{m} \frac{1}{2} p_{m}^2}   \hat{\mathcal{V}}_{\alpha_m}(p_m,p_{m-1},\boldsymbol{\eta} )      
	 \\
	 &\times e^{i  \hbar^{-1}(t- \boldsymbol{s}_{1,k+1}^{+}- \hbar \boldsymbol{t}_{1,a_k}^{+}) \frac{1}{2} p_0^2} \hat{\varphi}(\tfrac{p_0}{\hbar}) \, d\boldsymbol{\eta}d\boldsymbol{p}   d\boldsymbol{t} d\boldsymbol{s}.
	\end{aligned}
\end{equation*}
Recall that since $\iota\in\mathcal{Q}_{k,2,1}$ we have the numbers $m_1$, $m_2$ and $m_3$ associated to the map $\iota$ and the ``adjoint map'' $\iota^{*}$.
We start by preforming the change of variables $p_{m_2-1}\mapsto p_{m_2-1} - p_{m_2}$ and $p_{m}\mapsto p_{m} - p_{m_2-1}$ for all $m\in\{m_1,\dots,m_2-2\}$ and  the change of variables $p_{m_3-1}\mapsto p_{m_3-1}-p_{m_3}$ and $p_m\mapsto p_m-p_{m_3-1}$ for all $m\in\{m_1,\dots,m_3-2\}\setminus\{m_2-1\}$. Finally we preform the change of variables $p_k\mapsto p_k -p_{k+1}$ and $p_m\mapsto p_m-p_{k}$ for all $m\in\{l,\dots,k-1\}\setminus\{m_2-1,m_3-1\}$   This yields after a relabelling  
\begin{equation*}
	\begin{aligned}
 	\MoveEqLeft \tilde{\mathcal{E}}_{2,1}^{\mathrm{rec}}(s_{k+1},k,\boldsymbol{x},\alpha,l,t,\iota;\hbar)\varphi(x) 
	\\
	={}&  \frac{1}{(2\pi\hbar)^d\hbar^{k}} \int_{\R_{+}^{|\alpha|}} \int \boldsymbol{1}_{[0,t]}( \boldsymbol{s}_{1,k+1}^{+}+ \hbar\boldsymbol{t}_{1,a_k}^{+})
	 e^{ i   \langle  \hbar^{-1}x,p_{k-2} \rangle} 
	 e^{i  s_{k+1} \frac{1}{2}\hbar^{-1} p_{k-2}^2}   
	 e^{i \hbar^{-1}  s_{m_3} \frac{1}{2} (p_{m_3-2}+\tilde{p}_3 \boldsymbol{1}_{\{l<m_3\}})^2} \hat{V}(-\tilde{p}_3)  
	 \\
	 &\times  e^{i \hbar^{-1}  s_{m_2} \frac{1}{2} (p_{m_2-1}+\tilde{p}_2+\tilde{p}_3\boldsymbol{1}_{\{l<m_2\}})^2}  \hat{\mathcal{V}}_{\alpha_{m_2}}(p_{m_2-1}+\tilde{p}_2+\tilde{p}_3\boldsymbol{1}_{\{l<m_2\}},\tilde{p}_1 +p_{m_2-1}+\tilde{p}_2+\tilde{p}_3\boldsymbol{1}_{\{l<m_2\}},\boldsymbol{\eta} )
	 \\
	 &\times \prod_{i=1}^{a_k} e^{i  t_{i} \frac{1}{2} \eta_{i}^2}   \prod_{m=1}^{k-2}  e^{ -i   \langle  \hbar^{-1/d}x_{m},p_{m}-p_{m-1} \rangle}
	 e^{i \hbar^{-1}  s_{\iota^{*}(m)} \frac{1}{2} (p_{m}+\pi_m^1(\boldsymbol{\tilde{p}}))^2}   \hat{\mathcal{V}}_{\alpha_\iota^{*}(m)}(p_m+\pi_m^1(\boldsymbol{\tilde{p}}),p_{m-1}+\pi_m^2(\boldsymbol{\tilde{p}}),\boldsymbol{\eta} )       
	 \\
	 &\times \hat{\mathcal{V}}_{\alpha_{m_3}}(p_{m_3-2}+\tilde{p}_3\boldsymbol{1}_{\{l<m_3\}},\tilde{p}_2 + p_{m_3-2}+\tilde{p}_3\boldsymbol{1}_{\{l<m_3\}},\boldsymbol{\eta} ) 
	e^{i  \hbar^{-1}(t- \boldsymbol{s}_{1,k+1}^{+}- \hbar \boldsymbol{t}_{1,a_k}^{+}) \frac{1}{2} p_0^2} \hat{\varphi}(\tfrac{p_0}{\hbar}) \, d\boldsymbol{\eta}d\boldsymbol{\tilde{p}}d\boldsymbol{p}   d\boldsymbol{t} d\boldsymbol{s},
	\end{aligned}
\end{equation*}
where we have used the following definition of the functions 
\begin{equation*}
\begin{aligned}
	&\pi_{m}^1(\boldsymbol{\tilde{p}} ) = \tilde{p}_1 \boldsymbol{1}_{\{m_1,\dots,m_2-1\}}(m) + \tilde{p}_2 \boldsymbol{1}_{\{m_1,\dots,m_3-2\}}(m) + \tilde{p}_3 \boldsymbol{1}_{\{\iota^{*[-1]}(l),\dots,k-2\}}(m)
	\\
	&\pi_{m}^2(\boldsymbol{\tilde{p}} ) = \tilde{p}_1 \boldsymbol{1}_{\{m_1+1,\dots,m_2-1\}}(m) + \tilde{p}_2 \boldsymbol{1}_{\{m_1+1,\dots,m_3-2\}}(m) + \tilde{p}_3 \boldsymbol{1}_{\{\iota^{*[-1]}(l+1),\dots,k-2\}}(m),
	\end{aligned}
\end{equation*}
where $\iota^{*[-1]}$ is the inverse of the map $\iota^{*}$ with the convention that $\iota^{*[-1]}(m_1)=m_1$. How we will estimate the different terms will depend on $\kappa$, and the relation between the two $\sigma$'s and $\kappa$. Firstly we split into the case if $\kappa\neq\mathrm{id}$ or $\kappa=\mathrm{id}$. We start with the case where $\kappa=\mathrm{id}$ here we will further split into the four cases $l> m_3$, $m_3>l>m_1 $ with $l\neq m_2$, $l=m_1$ and $l<m_1$.
We start with the case $l> m_3$. Observe then we need $k\geq8$ for this to be possible.  Here we split into different cases depending on $\sigma^1$ and $\sigma^2$. Firstly we assume that $\sigma^1_n,\sigma^2_n\leq l-2$\footnote{It is indeed $l-2$ here since we assume $l>m_2$ and hence have that if $\sigma^1_n,\sigma^2_n\leq l-2$ then will $\iota^{*}(\sigma^1_n),\iota^{*}(\sigma^2_n)\leq l$.} and that there exists  two indices $i_1$ and $i_2$ in $\{1,\dots,n+1\}$ such that
\begin{equation*}
	\sigma_{i_1-1}^1 \leq m_1<m_3-2 < \sigma_{i_1}^1 \quad\text{and}\quad \sigma_{i_2-1}^2 \leq m_1<m_3-2 < \sigma_{i_2}^2.
\end{equation*}
Note that if such indices exists in the case $k=8$, then $n$ can be at most $2$. In general we have that $n\leq l-m_3+m_1$ otherwise such indices can not exists.  For this case we get after applying Lemma~\ref{LE:Exp_ran_phases} and evaluating all integrals involving delta functions that
\begin{equation}\label{EQ:expansion_aver_bound_3.recol_21_2}
	\begin{aligned}
	\MoveEqLeft\Aver{ \mathcal{T}(n,\sigma^1,\sigma^2,\alpha,\tilde{\alpha},\kappa)} = \frac{(\rho(2\pi)^{d})^{2k-n-4}}{(2\pi\hbar)^d\hbar^{n+4}} \int_{\R_{+}^{|\alpha|+|\tilde{\alpha}|+1}} \int \boldsymbol{1}_{[0,t]}( \boldsymbol{s}_{1,k+1}^{+}+ \hbar\boldsymbol{t}_{1,a_k}^{+}) \boldsymbol{1}_{[0,t]}( \boldsymbol{\tilde{s}}_{1,k}^{+}+s_{k+1}+ \hbar\boldsymbol{\tilde{t}}_{1,a_k}^{+}) 
	\\
	&\times 
	  e^{i \hbar^{-1}  s_{m_2} \frac{1}{2} (p_{i_1-1}+\tilde{p}_2)^2} e^{i \hbar^{-1}  s_{m_3} \frac{1}{2} p_{i_1-1}^2}  
	  e^{-i \hbar^{-1}  \tilde{s}_{m_2} \frac{1}{2} (p_{i_2-1}+\tilde{q}_2)^2}  e^{-i \hbar^{-1}  \tilde{s}_{m_3} \frac{1}{2} p_{i_2-1}^2}
	    \prod_{i=1}^{a_k} e^{i  t_{i} \frac{1}{2} \eta_{i}^2}  \prod_{i=1}^{\tilde{a}_k} e^{-i  \tilde{t}_{i} \frac{1}{2} \xi_{i}^2}
	  \\
	  &\times     \prod_{m=\sigma_n^1+1}^{k-1} e^{i \hbar^{-1}  s_{\iota^{*}(m)} \frac{1}{2} (p_{n}+\pi_m^1(\boldsymbol{\tilde{p}}) )^2}    \prod_{m=\sigma_n^2+1}^{k-1} e^{-i \hbar^{-1}  \tilde{s}_{\iota^{*}(m)} \frac{1}{2} (p_{n}+\pi_m^1(\boldsymbol{\tilde{q}}) )^2} 
	  \prod_{i=1}^n \Big\{ e^{i \hbar^{-1}  s_{\iota^{*}(\sigma_{i}^1)} \frac{1}{2} (p_{i}+\pi_{\sigma_{i}^1}^1(\boldsymbol{\tilde{p}}) )^2}
	  \\
	  &\times   
	 e^{-i \hbar^{-1}  \tilde{s}_{\iota^{*}(\sigma_{i}^2)} \frac{1}{2} (p_{i}+\pi_{\sigma_{i}^2}^1(\boldsymbol{\tilde{q}}) )^2} 
	 \prod_{m=\sigma_{i-1}^1+1}^{\sigma_{i}^1-1} e^{i \hbar^{-1}  s_{\iota^{*}(m)} \frac{1}{2} (p_{i-1}+\pi_m^1(\boldsymbol{\tilde{p}}) )^2}   
	 \prod_{m=\sigma_{i-1}^2+1}^{\sigma_{i}^2-1} e^{-i \hbar^{-1}  s_{\iota^{*}(m)} \frac{1}{2} (p_{i-1}+\pi_m^1(\boldsymbol{\tilde{q}}) )^2}  \Big\}  
	  \\
	  &\times \mathcal{G}(\boldsymbol{\tilde{p}},\boldsymbol{\tilde{q}},\boldsymbol{q},\boldsymbol{\eta},\boldsymbol{\xi},\sigma^1,\sigma^2,\alpha,\tilde{\alpha}) 
	e^{i  \hbar^{-1}(\boldsymbol{\tilde{s}}_{1,k}^{+}- \boldsymbol{s}_{1,k}^{+}+ \hbar( \boldsymbol{\tilde{t}}_{1,\tilde{a}_k}^{+}- \boldsymbol{t}_{1,a_k}^{+})) \frac{1}{2} p_0^2} |\hat{\varphi}(\tfrac{p_0}{\hbar})|^2
	  \, d\boldsymbol{\eta}d\boldsymbol{\tilde{p}}d\boldsymbol{\xi}d\boldsymbol{\tilde{q}}d\boldsymbol{p}   d\boldsymbol{t} d\boldsymbol{s}, 
	\end{aligned}
\end{equation}
where the function $ \mathcal{G}(\boldsymbol{\tilde{p}},\boldsymbol{\tilde{q}},\boldsymbol{q},\boldsymbol{\eta},\boldsymbol{\xi},\sigma^1,\sigma^2,\alpha,\tilde{\alpha}) $ is defined by
\begin{equation}\label{EQ:def_G_recol_21}
	\begin{aligned}
	\MoveEqLeft \mathcal{G}(\boldsymbol{\tilde{p}},\boldsymbol{\tilde{q}},\boldsymbol{q},\boldsymbol{\eta},\boldsymbol{\xi},\sigma^1,\sigma^2,\alpha,\tilde{\alpha}) = \hat{V}(-\tilde{p}_3)  \hat{\mathcal{V}}_{\alpha_{m_2}}(p_{i_1-1}+\tilde{p}_2,\tilde{p}_1 +p_{i_1-1}+\tilde{p}_2,\boldsymbol{\eta} ) \hat{\mathcal{V}}_{\alpha_{m_3}}(p_{i_1-1},\tilde{p}_2 + p_{i_1-1},\boldsymbol{\eta} ) 
	 \\
	  &\times \overline{ \hat{V}(-\tilde{q}_3)}  \overline{\hat{\mathcal{V}}_{\alpha_{m_2}}(p_{i_2-1}+\tilde{q}_2,\tilde{q}_1 +p_{i_2-1}+\tilde{q}_2,\boldsymbol{\xi} )} \overline{\hat{\mathcal{V}}_{\alpha_{m_3}}(p_{i_2-1},\tilde{q}_2 + p_{i_2-1},\boldsymbol{\xi} ) }
	  \\
	  &\times \prod_{m=\sigma_n^1+1}^{k-1}   \hat{\mathcal{V}}_{\alpha_{\iota^{*}(m)}}(p_n+\pi_m^1(\boldsymbol{\tilde{p}}),p_{n}+\pi_m^2(\boldsymbol{\tilde{p}}),\boldsymbol{\eta} )  \prod_{m=\sigma_n^2+1}^{k-1} \overline{  \hat{\mathcal{V}}_{\alpha_{\iota^{*}(m)}}(p_n+\pi_m^1(\boldsymbol{\tilde{q}}),p_{n}+\pi_m^2(\boldsymbol{\tilde{q}}),\boldsymbol{\xi} ) }
	 \\
	 &\times \prod_{i=1}^n \Big\{  \hat{\mathcal{V}}_{\alpha_{\iota^{*}(\sigma_{i}^1)}}(p_{i}+\pi_{\sigma_{i}^1}^1(\boldsymbol{\tilde{p}}),p_{i-1}+\pi_{\sigma_{i}^1}^2(\boldsymbol{\tilde{p}}),\boldsymbol{\eta} ) 
	 \prod_{m=\sigma_{i-1}^1+1}^{\sigma_{i}^1-1}   \hat{\mathcal{V}}_{\alpha_{\iota^{*}(m)}}(p_{i-1}+\pi_m^1(\boldsymbol{\tilde{p}}),p_{i-1}+\pi_m^2(\boldsymbol{\tilde{p}}),\boldsymbol{\eta} ) 
	  \\
	  &\times   \overline{  \hat{\mathcal{V}}_{\alpha_{\iota^{*}(\sigma_{i}^2)}}(p_{i}+\pi_{\sigma_{i}^2}^1(\boldsymbol{\tilde{q}}),p_{i-1}+\pi_{\sigma_{i}^2}^2(\boldsymbol{\tilde{q}}),\boldsymbol{\xi} ) }
	  \prod_{m=\sigma_{i-1}^2+1}^{\sigma_{i}^2-1}   \overline{\hat{\mathcal{V}}_{\alpha_{\iota^{*}(m)}}(p_{i-1}+\pi_m^1(\boldsymbol{\tilde{q}}),p_{i-1}+\pi_m^2(\boldsymbol{\tilde{q}}),\boldsymbol{\xi} ) }  \Big\}.
	\end{aligned}
\end{equation}
We now preform the change of variables $\tilde{p}_1\mapsto p_{i_1-1} +\tilde{p}_1+\tilde{p}_2$, $\tilde{p}_2\mapsto  p_{i_1-1} +\tilde{p}_2$, $\tilde{p}_3\mapsto  p_{n} +\tilde{p}_3$, $\tilde{q}_1\mapsto p_{i_2-1} +\tilde{q}_1+\tilde{q}_2$, $\tilde{q}_2\mapsto  p_{i_2-1} +\tilde{q}_2$, $\tilde{q}_3\mapsto  p_{n} +\tilde{q}_3$. After this change of variables the argument is analogous to the case of $\kappa=\mathrm{id}$ in the proof of Lemma~\ref{expansion_aver_bound_Mainterm}, where we use the integration by parts trick. Here we will do it in all $\boldsymbol{\eta}$, $\boldsymbol{\xi}$, $p_1,\dots,p_n$, $\tilde{p}_1,\tilde{p}_2,\tilde{p}_3,\tilde{q}_1,\tilde{q}_2$ and $\tilde{q}_3$ if  either $\sigma_n^1<l-2$ or $\sigma_n^2<l-2$. If $\sigma_n^1=\sigma_n^2=l-2$, we cannot use the argument in the variable $p_n$.  The difference in the two arguments will be in how we estimate the time integrals. The time integrals will in this case be
 \begin{equation}\label{EQ:time_int_est_3_recol_1}
	\begin{aligned}
	 & \int_{\R_{+}^{2k+1}}  \tfrac{\max(1,\hbar^{-1} |s_{m_3}+\sum_{m=\sigma_{i_1-1}^1}^{m_1-1} s_{\iota^{*}(m)}+\sum_{m=m_2-2}^{\sigma_{i_1}^1-1} s_{\iota^{*}(m)}-\sum_{m=\sigma_{i_1-1}^2}^{\sigma_{i_1}^2-1} \tilde{s}_{\iota^{*}(m)}  + \hbar l^1_{i_1-1}(\boldsymbol{t}, \boldsymbol{\tilde{t}}) |)^{-\frac{d}{2}} }
	{ \max(1,\hbar^{-1} |\sum_{m=\sigma_{i_2-1}^1}^{\sigma_{i_2}^1-1} s_{\iota^{*}(m)} - \tilde{s}_{m_3}-\sum_{m=\sigma_{i_2-1}^2}^{m_1-1} \tilde{s}_{\iota^{*}(m)}-\sum_{m=m_2-2}^{\sigma_{i_2}^2-1} \tilde{s}_{\iota^{*}(m)}  + \hbar l^1_{i_2-1}(\boldsymbol{t}, \boldsymbol{\tilde{t}}) |)^{\frac{d}{2}} }
	 \\
	 &\times \prod_{\substack{i=1\\ i\neq i_1-1,i_2-1}}^{n-1}\tfrac{\boldsymbol{1}_{[0,t]}( \boldsymbol{s}_{1,k+1}^{+}+ \hbar\boldsymbol{t}_{1,a_k}^{+}) \boldsymbol{1}_{[0,t]}( \boldsymbol{\tilde{s}}_{1,k}^{+}+s_{k+1}+ \hbar\boldsymbol{\tilde{t}}_{1,a_k}^{+})}{ \max(1,\hbar^{-1} |\sum_{m=\sigma_{i}^1}^{\sigma_{i+1}^1-1} s_{\iota^{*}(m)} - \sum_{m=\sigma_{i}^2}^{\sigma_{i+1}^2-1} \tilde{s}_{\iota^{*}(m)}  + \hbar l^1_i(\boldsymbol{t}, \boldsymbol{\tilde{t}}) |)^{\frac{d}{2}} } \tfrac{ 1}{ \max(1, \hbar^{-1}| \sum_{m=\sigma_n^1}^{l-3} s_{\iota^{*}(m)}- \sum_{m=\sigma_n^2}^{l-3} \tilde{s}_{\iota^{*}(m)} +l^1_n(\boldsymbol{t}, \boldsymbol{\tilde{t}}) |)^{\frac{d}{2}}}
	\\
	&\times \tfrac{1}{[\max(1, \hbar^{-1}| \sum_{m=m_1}^{m_2-1} s_{\iota^{*}(m)}+\tilde{l}^1_{1,p}(\boldsymbol{t}) |) \max(1, \hbar^{-1}| s_{m_2}+ \sum_{m=m_2}^{m_3-2} s_{\iota^{*}(m)}+\tilde{l}^1_{2,p}(\boldsymbol{t}) |) \max(1, \hbar^{-1}| \sum_{m=l-2}^{k-2} s_{\iota^{*}(m)}+\tilde{l}^1_{3,p}(\boldsymbol{t}) |)]^{\frac{d}{2}}} 
	\\
	&\times \tfrac{1}{ [\max(1, \hbar^{-1}| \sum_{m=m_1}^{m_2-1} \tilde{s}_{\iota^{*}(m)}+\tilde{l}^1_{1,q}(\boldsymbol{\tilde{t}}) |)
	\max(1, \hbar^{-1}| \tilde{s}_{m_2}+ \sum_{m=m_2}^{m_3-2} \tilde{s}_{\iota^{*}(m)}+\tilde{l}^1_{2,q}(\boldsymbol{\tilde{t}}) |)
	 \max(1, \hbar^{-1}| \sum_{m=l-2}^{k-2} \tilde{s}_{\iota^{*}(m)}+\tilde{l}^1_{3,q}(\boldsymbol{\tilde{t}}) |)]^{\frac{d}{2}}}  
	 \, d\boldsymbol{\tilde{s}}d\boldsymbol{s},
	\end{aligned}
\end{equation} 
where all functions of $t$ are linear functions. They are determined by which $t_i$'s and $\tilde{t}_i$ are larger than one or smaller. Moreover, we have included the integral over $s_{k+1}$.  In particular we have that all functions with a tilde are a sum of either $t$'s or $\tilde{t}$'s hence they are positive and we get something larger by dropping them. Note the in the case, where $\sigma_n^1=\sigma_n^2=l-2$ we have that the term with the sums ranging over $m=\sigma_n^j$ to $l-3$ will not be there. Due to this we will also just drop this term in the following estimate. Doing this and using the definition of the map $\iota^{*}$ we get that the integrals can be bounded by
 \begin{equation*}
	\begin{aligned}
	\MoveEqLeft  \int_{\R_{+}^{2k+1}}   \tfrac{\boldsymbol{1}_{[0,t]}( \boldsymbol{s}_{1,k+1}^{+}+ \hbar\boldsymbol{t}_{1,a_k}^{+}) \boldsymbol{1}_{[0,t]}( \boldsymbol{\tilde{s}}_{1,k}^{+}+s_{k+1}+ \hbar\boldsymbol{\tilde{t}}_{1,a_k}^{+})}{\max(1,\hbar^{-1} |\sum_{m=\sigma_{i_2-1}^1}^{\sigma_{i_2}^1-1} s_{\iota^{*}(m)} - \tilde{s}_{m_3}-\sum_{m=\sigma_{i_2-1}^2}^{m_1-1} \tilde{s}_{\iota^{*}(m)}-\sum_{m=m_2-2}^{\sigma_{i_2}^2-1} \tilde{s}_{\iota^{*}(m)}  + \hbar l^1_{i_2-1}(\boldsymbol{t}, \boldsymbol{\tilde{t}}) |)^\frac{d}{2}} 
	\\
	&\times \tfrac{1}{\max(1,\hbar^{-1} |s_{m_3}+\sum_{m=\sigma_{i_1-1}^1}^{m_1-1} s_{\iota^{*}(m)}+\sum_{m=m_2-2}^{\sigma_{i_1}^1-1} s_{\iota^{*}(m)}-\sum_{m=\sigma_{i_1-1}^2}^{\sigma_{i_1}^2-1} \tilde{s}_{\iota^{*}(m)}  + \hbar l^1_{i_1-1}(\boldsymbol{t}, \boldsymbol{\tilde{t}}) |)^\frac{d}{2}} 
	\\
	&\times  \prod_{\substack{i=1\\ i\neq i_1-1,i_2-1}}^{n-1} \tfrac{1}{\max(1,\hbar^{-1} |\sum_{m=\sigma_{i}^1}^{\sigma_{i+1}^1-1} s_{\iota^{*}(m)} - \sum_{m=\sigma_{i}^2}^{\sigma_{i+1}^2-1} \tilde{s}_{\iota^{*}(m)}  + \hbar l^1_i(\boldsymbol{t}, \boldsymbol{\tilde{t}}) |)^\frac{d}{2}} 
	\\
	&\times \tfrac{1}{\max(1, \hbar^{-1}| \sum_{m=m_1}^{m_2-1} s_{m} |)^\frac{d}{2}}  \tfrac{1}{\max(1, \hbar^{-1}| \sum_{m=m_2}^{m_3-1} s_{m} |)^\frac{d}{2}}  \tfrac{1}{\max(1, \hbar^{-1}| \sum_{m=l}^{k} s_{m} |)^\frac{d}{2}} 
	\\
	&\times \tfrac{1}{\max(1, \hbar^{-1}| \sum_{m=m_1}^{m_2-1} \tilde{s}_{m} |)^\frac{d}{2}} 
	 \tfrac{1}{\max(1, \hbar^{-1}| \sum_{m=m_2}^{m_3-1} \tilde{s}_{m} |)^\frac{d}{2}}  \tfrac{1}{\max(1, \hbar^{-1}| \sum_{m=l}^{k} \tilde{s}_{m} |)^\frac{d}{2}}  
	 \, d\boldsymbol{\tilde{s}} d\boldsymbol{s}.
	\end{aligned}
\end{equation*} 
It is important to observe that each $s$ and $\tilde{s}$ only appears once. We note that for the characteristic functions we have the estimate
\begin{equation*}
	\begin{aligned}
	\MoveEqLeft \boldsymbol{1}_{[0,t]}( \boldsymbol{s}_{1,k+1}^{+}+ \hbar\boldsymbol{t}_{1,a_k}^{+}) \boldsymbol{1}_{[0,t]}( \boldsymbol{\tilde{s}}_{1,k}^{+}+s_{k+1}+ \hbar\boldsymbol{\tilde{t}}_{1,a_k}^{+}) 
	\leq \boldsymbol{1}_{[0,t]}( \boldsymbol{s}_{1,m_1-1}^{+} +\boldsymbol{s}_{m_3,l-1}^{+})\boldsymbol{1}_{[0,t]}( \boldsymbol{s}_{l,k}^{+})
	\\
	&\times \boldsymbol{1}_{[0,t]}( \boldsymbol{s}_{m_1,m_2-1}^{+}+s_k) \boldsymbol{1}_{[0,t]}( \boldsymbol{s}_{m_2,m_3-1}^{+} +s_k) \boldsymbol{1}_{[0,t]}( \boldsymbol{\tilde{s}}_{1,m_1-1}^{+} +\boldsymbol{\tilde{s}}_{m_3,l-1}^{+} - g(\tilde{s}) +s_{k+1})
	\\
	&\times \boldsymbol{1}_{[0,t]}( \boldsymbol{\tilde{s}}_{l,k}^{+})\boldsymbol{1}_{[0,t]}( \boldsymbol{\tilde{s}}_{m_1,m_2-1}^{+}+\tilde{s}_k) \boldsymbol{1}_{[0,t]}( \boldsymbol{\tilde{s}}_{m_2,m_3-1}^{+} +\tilde{s}_k), 
	\end{aligned}	
\end{equation*}
where $g(\tilde{s}) = \sum_{i=1, i\neq i_1-1,i_2-1}^{n-1} \tilde{s}_{\iota^{*}(\sigma_i^2)} +\tilde{s}_{m_3} + \tilde{s}_{\iota^{*}(\sigma_{i_1-1}^2)}\boldsymbol{1}_{i_1\neq i_2} $.
We will further more make the following change of variables 
\begin{equation*}
	\begin{aligned}
	s_m &\mapsto \sum_{j=m}^k s_j \quad\text{and}\quad \tilde{s}_m \mapsto \sum_{j=m}^k \tilde{s}_j \qquad\text{for all $m\in\{l,\dots,k\}$},
	\\
	s_m &\mapsto \sum_{j=m}^{m_3-1} s_j + s_k \quad\text{and}\quad \tilde{s}_m \mapsto \sum_{j=m}^{m_3-1} \tilde{s}_j + \tilde{s}_k \qquad\text{for all $m\in\{m_2,\dots,m_3-1\}$},
	\\
	s_m &\mapsto \sum_{j=m}^{m_2-1} s_j + s_k \quad\text{and}\quad \tilde{s}_m \mapsto \sum_{j=m}^{m_2-1} \tilde{s}_j + \tilde{s}_k  \qquad\text{for all $m\in\{m_1,\dots,m_2-1\}$}.
	\end{aligned}
\end{equation*}
Using the estimate for the characteristic function and evaluating all integrals except the ones we have just made a change of variables we obtain that our initial integral is bounded by the expression
 \begin{equation*}
	\begin{aligned}
	\MoveEqLeft  C^{n-1}  \frac{\hbar^{n-1} t^{2(l-m_3+m_1)-n}}{(l-m_3+m_1-1)!(l-m_3+m_1-n+1)!} \Big[ \int_{\R_{+}} \int_{\R_{+}^{k-l}} \frac{\boldsymbol{1}_{[0,t]_{\leq}^{k-l+1}}( \boldsymbol{s}_{k,l})}{\max(1, \hbar^{-1}| s_{l} |)^\frac{d}{2}} \,d\boldsymbol{s}_{l,k-1}
	\\
	&\times \int_{\R_{+}^{m_2-m_1}} \frac{\boldsymbol{1}_{[0,t]_{\leq}^{m_2-m_1+1}}(s_k, \boldsymbol{s}_{m_2-1,m_1})}{\max(1, \hbar^{-1}| s_{m_1}-s_k |)^\frac{d}{2}} \,d\boldsymbol{s}_{m_1,m_2-1} \int_{\R_{+}^{m_3-m_2}} \frac{\boldsymbol{1}_{[0,t]_{\leq}^{m_3-m_2}}(s_k, \boldsymbol{s}_{m_3-1,m_2})}{\max(1, \hbar^{-1}|  s_{m_2} - s_k |)^\frac{d}{2}} \,d\boldsymbol{s}_{m_2,m_3-1} \, ds_k \Big]^2.
	\end{aligned}
\end{equation*} 
Evaluating the integral in $s_l$, $s_{m_1}$ and $s_{m_2}$ we get the bounds 
 \begin{equation*}
	\begin{aligned}
	 \int_{\R_{+}^{k-l}} \frac{\boldsymbol{1}_{[0,t]_{\leq}^{k-l+1}}( \boldsymbol{s}_{k,l})}{\max(1, \hbar^{-1}| s_{l} |)^\frac{d}{2}} \,d\boldsymbol{s}_{l,k-1}
	\leq \frac{4\hbar}{d-2}  \int_{\R_{+}^{k-l-1}} \frac{\boldsymbol{1}_{[0,t]_{\leq}^{k-l-1}}( \boldsymbol{s}_{k-1,l+1})}{(1+ \hbar^{-1} s_{k} )^\frac{d-2}{2}} \,d\boldsymbol{s}_{l+1,k-1},
	\end{aligned}
\end{equation*} 
 \begin{equation*}
	\begin{aligned}
	\MoveEqLeft \int_{\R_{+}^{m_2-m_1}} \frac{\boldsymbol{1}_{[0,t]_{\leq}^{m_2-m_1+1}}(s_k, \boldsymbol{s}_{m_2-1,m_1})}{\max(1, \hbar^{-1}| s_{m_1}-s_k |)^\frac{d}{2}} \,d\boldsymbol{s}_{m_1,m_2-1}
	\\
	&\leq \frac{4\hbar}{d-2}  \int_{\R_{+}^{m_2-m_1-1}} \frac{\boldsymbol{1}_{[0,t]_{\leq}^{m_2-m_1-1}}( \boldsymbol{s}_{m_2-1,m_1+1})}{(1+ \hbar^{-1}|s_{m_1+1}- s_{k}| )^\frac{d-2}{2}} \,d\boldsymbol{s}_{m_1+1,m_2-1},
	\end{aligned}
\end{equation*} 
and
 \begin{equation*}
	\begin{aligned}
	\MoveEqLeft \int_{\R_{+}^{m_3-m_2}} \frac{\boldsymbol{1}_{[0,t]_{\leq}^{m_3-m_2+1}}(s_k, \boldsymbol{s}_{m_3-1,m_2})}{\max(1, \hbar^{-1}| s_{m_2}-s_k |)^\frac{d}{2}} \,d\boldsymbol{s}_{m_2,m_3-1}
	\\
	&\leq \frac{4\hbar}{d-2}  \int_{\R_{+}^{m_3-m_2-1}} \frac{\boldsymbol{1}_{[0,t]_{\leq}^{m_3-m_2-1}}( \boldsymbol{s}_{m_3-1,m_2+1})}{(1+ \hbar^{-1}|s_{m_2+1}- s_{k}| )^\frac{d-2}{2}} \,d\boldsymbol{s}_{m_2+1,m_3-1}.
	\end{aligned}
\end{equation*}
Applying these estimates and evaluating the integral in $s_k$, which gives an additional power of $\hbar$, and then  the remaining  integrals we obtain that the initial time integral in \eqref{EQ:time_int_est_3_recol_1} is bounded by
 \begin{equation*}
	\begin{aligned}
	\MoveEqLeft C^n  \frac{\hbar^{n+7} t^{2k-6-n}}{(l-m_3+m_1-1)!(l-m_3+m_1-n+1)!((m_3-m_2-1)! (m_2-m_1-1)!(k-l-1)!)^2} 
	\\
	&\leq C^n \frac{\hbar^{n+7} t^{2k-6-n}}{(k-4)!(k-n-2)!}. 
	\end{aligned}
\end{equation*} 
Using this estimate and the methods used in the proof of Lemma~\ref{expansion_aver_bound_Mainterm} we can obtain the estimate
   \begin{equation}\label{EQ:expansion_aver_bound_3.recol_21_3}
   	\begin{aligned}
	 \Aver{\mathcal{T}(n,\sigma^1,\sigma^2,\alpha,\tilde{\alpha},\mathrm{id})}
	 \leq  C_d^{a_k+\tilde{a}_k+n}  \norm{\varphi}_{L^2(\R^d)}^2 \norm{\hat{V}}_{1,\infty,3d+3}^{|\alpha|+|\tilde{\alpha}|+2}  \frac{\hbar^3 \rho^{2} (\rho t)^{2(k-3)-n} }{(k-4)!(k-n-2)!}.
	\end{aligned}
\end{equation}
We now turn to the case where $\sigma^1_n > l-2$ or $\sigma^2_n >l-2$ or there exists $i'$ such that $m_1<\sigma_{i'}^1<m_3$ or here exists $i'$ such that $m_1<\sigma_{i'}^2<m_3$.  This is the case where we have at least one genuine recollision happening between the internal recollisions. Again we have the expression from \eqref{EQ:expansion_aver_bound_3.recol_21_2}, but we have done a change of variables so we get the following expression
\begin{equation*}
	\begin{aligned}
	\MoveEqLeft\Aver{ \mathcal{T}(n,\sigma^1,\sigma^2,\alpha,\tilde{\alpha},\mathrm{id})} = \frac{(\rho\hbar(2\pi)^{d})^{2k-n-4}}{(2\pi\hbar)^d} \int_{\R_{+}^{|\alpha|+|\tilde{\alpha}|+1}} \int \boldsymbol{1}_{[0,t]}( \hbar\boldsymbol{s}_{1,k}^{+}+s_{k+1}+ \hbar\boldsymbol{t}_{1,a_k}^{+})  \prod_{i=1}^{a_k} e^{i  t_{i} \frac{1}{2} \eta_{i}^2}   \prod_{i=1}^{\tilde{a}_k} e^{-i  \tilde{t}_{i} \frac{1}{2} \xi_{i}^2}
	 \\
	 &\times   \boldsymbol{1}_{[0,t]}( \hbar\boldsymbol{\tilde{s}}_{1,k}^{+}+s_{k+1}+ \hbar\boldsymbol{\tilde{t}}_{1,\tilde{a}_k}^{+}) 
	 e^{i   s_{m_2} \frac{1}{2} (p_{i_1-1}+\tilde{p}_2)^2} e^{i   s_{m_3} \frac{1}{2} p_{i_1-1}^2}  e^{-i  \tilde{s}_{m_2} \frac{1}{2} (p_{i_2-1}+\tilde{q}_2)^2}   e^{-i  \tilde{s}_{m_3} \frac{1}{2} p_{i_2-1}^2} 
	  \\
	  &\times   \prod_{m=\sigma_n^1+1}^{k-1} e^{i   s_{\iota^{*}(m)} \frac{1}{2} (p_{n}+\pi_m^1(\boldsymbol{\tilde{p}}) )^2}   \prod_{m=\sigma_n^2+1}^{k-1} e^{-i  \tilde{s}_{\iota^{*}(m)} \frac{1}{2} (p_{n}+\pi_m^1(\boldsymbol{\tilde{q}}) )^2} 
	  \prod_{i=1}^n \Big\{ e^{i  s_{\iota^{*}(\sigma_{i}^1)} \frac{1}{2} (p_{i}+\pi_{\sigma_{i}^1}^1(\boldsymbol{\tilde{p}}) )^2}   
	 \\
	 &\times e^{-i   \tilde{s}_{\iota^{*}(\sigma_{i}^2)} \frac{1}{2} (p_{i}+\pi_{\sigma_{i}^2}^1(\boldsymbol{\tilde{q}}) )^2} 
	   \prod_{m=\sigma_{i-1}^1+1}^{\sigma_{i}^1-1} e^{i   s_{\iota^{*}(m)} \frac{1}{2} (p_{i-1}+\pi_m^1(\boldsymbol{\tilde{p}}) )^2}   
	   \prod_{m=\sigma_{i-1}^2+1}^{\sigma_{i}^2-1} e^{-i  s_{\iota^{*}(m)} \frac{1}{2} (p_{i-1}+\pi_m^1(\boldsymbol{\tilde{q}}) )^2}   \Big\}
	 \\
	 &\times e^{i  (\boldsymbol{\tilde{s}}_{1,k}^{+}- \boldsymbol{s}_{1,k}^{+}+ \boldsymbol{\tilde{t}}_{1,\tilde{a}_k}^{+}- \boldsymbol{t}_{1,a_k}^{+}) \frac{1}{2} p_0^2} 
	 \mathcal{G}(\boldsymbol{\tilde{p}},\boldsymbol{\tilde{q}},\boldsymbol{q},\boldsymbol{\eta},\boldsymbol{\xi},\sigma^1,\sigma^2,\alpha,\tilde{\alpha})  
	 |\hat{\varphi}(\tfrac{p_0}{\hbar})|^2
	  \, d\boldsymbol{\eta}d\boldsymbol{\tilde{p}}d\boldsymbol{\xi}d\boldsymbol{\tilde{q}}d\boldsymbol{p}   d\boldsymbol{t} d\boldsymbol{s}, 
	\end{aligned}
\end{equation*}
where $\mathcal{G}$ is defined in \eqref{EQ:def_G_recol_21}.
For this case we will use the ``$\nu$-representation'' of the kernel. Hence we again we insert a suitable $f(\boldsymbol{s},\boldsymbol{\tilde{s}})$, where the dependence on $\boldsymbol{s},\boldsymbol{\tilde{s}}$ is determined by the relation between the $\sigma$'s and $m_1$, $m_2$, $m_3$ and $l$. It will depend on $k-4$ of the $s$ variables and $k-n-2$ of the $\tilde{s}$ variables, where we count $s_{k+1}$ as a $\tilde{s}$ variable. From here the argument is analogous that of Lemma~\ref{expansion_aver_bound_Mainterm} and one can obtain the estimate
\begin{equation}\label{EQ:expansion_aver_bound_3.recol_21_4}
   	\begin{aligned}
	| \Aver{\mathcal{T}(n,\sigma^1,\sigma^2,\alpha,\tilde{\alpha},\mathrm{id})}|
	 \leq   C_d^{a_k+\tilde{a}_k+n}   \norm{\hat{V}}_{1,\infty,5d+5}^{|\alpha|+|\tilde{\alpha}|+2}     \frac{\rho^2(\rho t)^{2k-n-6} \hbar^3 |\log(\zeta)|^{n+7}}{(k-4)!(k-n-2)!} \norm{\varphi}_{\mathcal{H}^{5d+5}_\hbar(\R^d)}^2.
	\end{aligned}
\end{equation}
In the special case where $n=1$ and the numbers $\sigma_1^1$ and $\sigma_1^2$ satisfies that $\sigma_1^2=l-2$, $\sigma^1_1 > l-2$ or $m_1<\sigma_{1}^1<m_3$ one has to combine the methods from above with the ``$\nu$-representation''. For the notation used above one has to use method based on Fourier transform for the $\tilde{q}$ variables and for the other ``side'' (the $\tilde{p}$ variables) one uses the method of the ``$\nu$-representation''.

We now turn to the case where $m_3>l>m_1$ with $l\neq m_2$. For this case we can use the  ``$\nu$-representation''  of the kernel and argue as above to obtain the estimate
 \begin{equation}\label{EQ:expansion_aver_bound_3.recol_21_5}
   	\begin{aligned}
	|\Aver{\mathcal{T}(n,\sigma^1,\sigma^2,\alpha,\tilde{\alpha},\mathrm{id})}|
	 \leq   C_d^{a_k+\tilde{a}_k+n}   \norm{\hat{V}}_{1,\infty,5d+5}^{|\alpha|+|\tilde{\alpha}|+2}    \frac{\rho^2(\rho t)^{2k-n-6} \hbar^3 |\log(\zeta)|^{n+7}}{(k-4)!(k-n-2)!} \norm{\varphi}_{\mathcal{H}^{5d+5}_\hbar(\R^d)}^2
	\end{aligned}
\end{equation}
for all $\sigma^1$ and $\sigma^2$. For the cases where $m_1\geq l$ we can argue as above and obtain that if $\sigma_n^1,\sigma_n^2<l$ the estimate
   \begin{equation}\label{EQ:expansion_aver_bound_3.recol_21_6}
   	\begin{aligned}
	| \Aver{\mathcal{T}(n,\sigma^1,\sigma^2,\alpha,\tilde{\alpha},\mathrm{id})}|
	 \leq C_d^{a_k+\tilde{a}_k+n}   \norm{\hat{V}}_{1,\infty,5d+5}^{|\alpha|+|\tilde{\alpha}|+2}    \norm{\varphi}_{L^2(\R^d)}^2   \frac{\hbar^3 \rho^{2} (\rho t)^{2(k-3)-n} }{(k-4)!(k-n-2)!}.
	\end{aligned}
\end{equation}
If this is not the case we can use the  ``$\nu$-representation''  of the kernel and obtain the estimate
 \begin{equation}\label{EQ:expansion_aver_bound_3.recol_21_7}
   	\begin{aligned}
	|\Aver{\mathcal{T}(n,\sigma^1,\sigma^2,\alpha,\tilde{\alpha},\mathrm{id})}|
	 \leq   C_d^{a_k+\tilde{a}_k+n}   \norm{\hat{V}}_{1,\infty,5d+5}^{|\alpha|+|\tilde{\alpha}|+2}   \frac{\rho^2(\rho t)^{2k-n-6} \hbar^3 |\log(\zeta)|^{n+7}}{(k-4)!(k-n-2)!} \norm{\varphi}_{\mathcal{H}^{5d+5}_\hbar(\R^d)}^2.
	\end{aligned}
\end{equation}
We now turn to the case where $\kappa\neq\mathrm{id}$. The estimates for all of these cases will be done by using the ``$\nu$-representation''  of the kernel. For this case we will also have to divide into different cases but for all cases we will obtain the estimate
 \begin{equation}\label{EQ:expansion_aver_bound_3.recol_21_8}
   	\begin{aligned}
	|\Aver{\mathcal{T}(n,\sigma^1,\sigma^2,\alpha,\tilde{\alpha},\kappa)}|
	 \leq   C_d^{a_k+\tilde{a}_k+n}   \norm{\hat{V}}_{1,\infty,5d+5}^{|\alpha|+|\tilde{\alpha}|+2}   \frac{\rho^2(\rho t)^{2k-n-6} \hbar^3 |\log(\zeta)|^{n+7}}{(k-4)!(k-n-2)!} \norm{\varphi}_{\mathcal{H}^{5d+5}_\hbar(\R^d)}^2.
	\end{aligned}
\end{equation}
This estimate is obtained analogous to the previous estimates. Combining our estimates in \cref{EQ:expansion_aver_bound_3.recol_21_0,EQ:expansion_aver_bound_3.recol_21_3,EQ:expansion_aver_bound_3.recol_21_4,EQ:expansion_aver_bound_3.recol_21_5,EQ:expansion_aver_bound_3.recol_21_7,EQ:expansion_aver_bound_3.recol_21_6,EQ:expansion_aver_bound_3.recol_21_8} and arguing as in the previous proofs we get that
\begin{equation*}
	\begin{aligned}
		\MoveEqLeft \mathbb{E}\Big[ \big\lVert \sum_{k=5}^{k_0} \sum_{\iota\in\mathcal{Q}_{k},2,1}   \sum_{\boldsymbol{x}\in \mathcal{X}_{\neq}^{k-2}}  \mathcal{E}_{2,1}^{\mathrm{rec}}(k,0,\boldsymbol{x},\iota,\tfrac{t}{\tau_0};\hbar)\varphi\big\rVert_{L^2(\R^d)}^2\Big]
		\\
		\leq{}& k_0^{12} \hbar   \norm{\varphi}_{\mathcal{H}^{5d+5}_\hbar(\R^d)}^2 \sum_{k=5}^{k_0} C^{2k} \sum_{n=0}^{k-2}\Big[  \frac{ \rho (\rho t)^{2k-5-n} }{(k-4)!(k-n-2)!}+  \frac{\rho(\rho t)^{2k-n-5}  |\log(\zeta)|^{n+7}n!}{(k-4)!(k-n-2)!} \Big]
		\\
		\leq{}& k_0^{16} C^{2k_0} |\log(\tfrac{\hbar}{t})|^{k_0+7}  \norm{\varphi}_{\mathcal{H}^{5d+5}_\hbar(\R^d)}^2  \hbar.
	\end{aligned}
\end{equation*}
This concludes the proof.
\end{proof}
\begin{lemma}\label{expansion_aver_bound_Mainterm_rec_trun21_2}
Assume we are in the setting of Definition~\ref{def_recol_reminder} and let $\varphi \in \mathcal{H}^{3d+3}_\hbar(\R^d)$. Then 
\begin{equation*}
	\begin{aligned}
	\MoveEqLeft \mathbb{E}\Big[\big\lVert  \sum_{k_1=1}^{k_0} \sum_{k_2=1}^{k_0} \sum_{\iota\in\mathcal{Q}_{k_1+k_2,2,1}}  \boldsymbol{1}_{\{k_1+k_2\geq5\}} 
	  \sum_{(\boldsymbol{x}_1,\boldsymbol{x}_2)\in \mathcal{X}_{\neq}^{k_1+k_2-2}}  \mathcal{E}_{2,1}^{\mathrm{rec}}(k_1,k_2,\boldsymbol{x}_1,\iota,\tfrac{t}{\tau_0};\hbar) \mathcal{I}_{2,1}^{\mathrm{rec}}(k_2,\boldsymbol{x}_2,\iota,\tfrac{\tau-1}{\tau_0}t;\hbar)\varphi\big\rVert_{L^2(\R^d)}^2 \Big]
	\\
	&\leq k_0^{18} C^{2k_0}  |\log(\tfrac{\hbar}{\tau_0})|^{k_0+9}  \norm{\varphi}_{\mathcal{H}^{3d+3}_\hbar(\R^d)}^2  \hbar,
	\end{aligned}
\end{equation*}
and
\begin{equation*}
	\begin{aligned}
	\MoveEqLeft \mathbb{E}\Big[\big\lVert  \sum_{k_1=1}^{k_0} \sum_{k_2=1}^{k_0} \sum_{\iota\in\mathcal{Q}_{k_1+k_2,2,1}}  \boldsymbol{1}_{\{k_1+k_2\geq5\}} 
	\\
	&\times 
	  \sum_{(\boldsymbol{x}_1,\boldsymbol{x}_2)\in \mathcal{X}_{\neq}^{k_1+k_2-2}} \mathcal{E}_{2,j}^{\mathrm{rec}}(k_1+1,k_2,(x_{2,k_2},\boldsymbol{x}_1),\iota,\tfrac{t}{\tau_0};\hbar) \mathcal{I}_{2,1}^{\mathrm{rec}}(k_2,\boldsymbol{x}_2,\iota,\tfrac{\tau-1}{\tau_0}t;\hbar)\varphi\big\rVert_{L^2(\R^d)}^2 \Big]
	\\
	&\leq k_0^{18} C^{2k_0}  |\log(\tfrac{\hbar}{\tau_0})|^{k_0+9}  \norm{\varphi}_{\mathcal{H}^{3d+3}_\hbar(\R^d)}^2  \hbar,
	\end{aligned}
\end{equation*}
where the constant $C$ depends on the single site potential $V$ and the coupling constant $\lambda$. In particular we have that the function is in  $L^2(\R^d)$ $\Pro$-almost surely. 
\end{lemma}
\begin{proof}
We observe that
\begin{equation*}
	\begin{aligned}
	\MoveEqLeft \big\lVert  \sum_{k_1=1}^{k_0} \sum_{k_2=1}^{k_0} \sum_{\iota\in\mathcal{Q}_{k_1+k_2,2,1}}  \boldsymbol{1}_{\{k_1+k_2\geq5\}} 
	  \sum_{(\boldsymbol{x}_1,\boldsymbol{x}_2)\in \mathcal{X}_{\neq}^{k_1+k_2-2}}  \mathcal{E}_{2,1}^{\mathrm{rec}}(k_1,k_2,\boldsymbol{x}_1,\iota,\tfrac{t}{\tau_0};\hbar) \mathcal{I}_{2,1}^{\mathrm{rec}}(k_2,\boldsymbol{x}_2,\iota,\tfrac{\tau-1}{\tau_0}t;\hbar)\varphi\big\rVert_{L^2(\R^d)}^2 
	\\
	\leq{}& C k_0^{10} \sum_{k=5}^{2k_0} \sum_{\iota\in\mathcal{Q}_{k,2,1}}   \sum_{k_1+k_2=k}    \boldsymbol{1}_{\{1\leq k_1,k_2\leq k_0\}}   
	\\
	&\times  \big\lVert  \sum_{(\boldsymbol{x}_1,\boldsymbol{x}_2)\in \mathcal{X}_{\neq}^{k-2}}  \mathcal{E}_{2,1}^{\mathrm{rec}}(k_1,k_2,\boldsymbol{x}_1,\iota,\tfrac{t}{\tau_0};\hbar) \mathcal{I}_{2,1}^{\mathrm{rec}}(k_2,\boldsymbol{x}_2,\iota,\tfrac{\tau-1}{\tau_0}t;\hbar)\varphi\big\rVert_{L^2(\R^d)}^2 .
	\end{aligned}
\end{equation*}
Moreover, we have for fixed $k,k_1,k_2$ and $\iota$ that
\begin{equation*}\label{EQ:expansion_aver_bound_3.recol_21_1_1}
	\begin{aligned}
	\MoveEqLeft   \big\lVert  \sum_{(\boldsymbol{x}_1,\boldsymbol{x}_2)\in \mathcal{X}_{\neq}^{k-2}}  \mathcal{E}_{2,1}^{\mathrm{rec}}(k_1,k_2,\boldsymbol{x}_1,\iota,\tfrac{t}{\tau_0};\hbar) \mathcal{I}_{2,1}^{\mathrm{rec}}(k_2,\boldsymbol{x}_2,\iota,\tfrac{\tau-1}{\tau_0}t;\hbar)\varphi\big\rVert_{L^2(\R^d)}^2 
	\\
	&\leq \frac{t}{\hbar^2}  \sum_{l=1, \iota(l)\neq k}^{k-2}   \int_{0}^{t} \big\lVert \sum_{\alpha\in \N^{k_2}}(i\lambda)^{|\alpha|+1}    
	\\
	&\phantom{\leq \frac{t}{\hbar^2}  \sum_{l=1, \iota(l)\neq k}^{k-2}   \int_{0}^{t} \big\lVert \sum} {}
	\sum_{(\boldsymbol{x}_1,\boldsymbol{x}_2)\in \mathcal{X}_{\neq}^{k-2}}   \tilde{\mathcal{E}}_{2,j}^{\mathrm{rec}}(s_{k+1},k,\boldsymbol{x}_2,\alpha,l,t,\iota;\hbar)  \mathcal{I}_{2,1}^{\mathrm{rec}}(k_2,\boldsymbol{x}_2,\iota,\tfrac{\tau-1}{\tau_0}t;\hbar) \varphi\big\rVert_{L^2(\R^d)}^2 \,ds_{k+1},
	\end{aligned}
\end{equation*}
To estimate the average of these terms we combine the arguments in the proof of Lemma~\ref{expansion_aver_bound_Mainterm_2} and Lemma~\ref{expansion_aver_bound_Mainterm_rec_trun21_1}.  From this we obtain the bound
\begin{equation*}\label{EQ:expansion_aver_bound_3.recol_21_1_2}
	\begin{aligned}
	\MoveEqLeft   \sum_{k_1+k_2=k}    \boldsymbol{1}_{\{1\leq k_1,k_2\leq k_0\}}   \mathbb{E}\Big[  \big\lVert  \sum_{(\boldsymbol{x}_1,\boldsymbol{x}_2)\in \mathcal{X}_{\neq}^{k-2}}  \mathcal{E}_{2,1}^{\mathrm{rec}}(k_1,k_2,\boldsymbol{x}_1,\iota,\tfrac{t}{\tau_0};\hbar) \mathcal{I}_{2,1}^{\mathrm{rec}}(k_2,\boldsymbol{x}_2,\iota,\tfrac{\tau-1}{\tau_0}t;\hbar)\varphi\big\rVert_{L^2(\R^d)}^2  \Big]
	\\
	\leq {}& \hbar   \norm{\varphi}_{\mathcal{H}^{2d+2}_\hbar(\R^d)}^2  C^{2k} \sum_{n=0}^{k-2}\Big[  \frac{ \rho (\rho t)^{2k-5-n} }{(k-4)!(k-n-2)!}+  \frac{\rho(\rho t)^{2k-n-5}  |\log(\zeta)|^{n+9}n!}{(k-4)!(k-n-2)!} \Big].
	\end{aligned}
\end{equation*}
Finally by arguing as in the proof of Lemma~\ref{expansion_aver_bound_Mainterm_rec_trun21_1} we obtain the desired estimate. To obtain the second estimate an analogous argument is used combined with the arguments in the proof of Lemma~\ref{expansion_aver_bound_Mainterm_3}.
\end{proof}
\section{Recollision error terms for $\mathcal{Q}_{k,2,2}$}\label{Sec:tech_est_truncated_recol_2}
For this case we will have to distinguish between three types of  configurations. Recall that for each $\iota\in \mathcal{Q}_{k,2,2}$ we have the four numbers $m_{1,1},m_{1,2},m_{2,1}$ and $m_{2,2}$, where the first two numbers keep track of the first internal recollision and the two last ones the second. The first type of configurations is the configurations such that
\begin{equation}\label{Def_Q_K_2_2_3}
	\mathcal{Q}_{k,2,2}^{3}=\{\iota\in \mathcal{Q}_{k,2,2}\, | \,  m_{1,1}= k-3, m_{2,1}=k-2 \}.
\end{equation}
These are the configurations, where we will have to estimate two terms born series. Note that for all $k\geq4$ the set $\mathcal{Q}_{k,2,2}^{3}$ will only contain one map/configuration. The second type of configurations is
\begin{equation}\label{Def_Q_K_2_2_2}
	\mathcal{Q}_{k,2,2}^{2}=\{\iota\in \mathcal{Q}_{k,2,2}\setminus \mathcal{Q}_{k,2,2}^{3} \, | \,  m_{1,1}< m_{2,1}<m_{1,2} \}.
\end{equation}
The third type of configurations are the remaing configurations
\begin{equation}\label{Def_Q_K_2_2_1}
	\mathcal{Q}_{k,2,2}^{1}=\mathcal{Q}_{k,2,2}\setminus (\mathcal{Q}_{k,2,2}^{2}\cup \mathcal{Q}_{k,2,2}^{3}).
\end{equation}
We will treat each of these cases with slightly different methods. The estimate for configurations in $\mathcal{Q}_{k,2,2}^{1}$ will be analogous to those used for configurations in $\mathcal{Q}_{k,2,1}$. The argument for configurations in $\mathcal{Q}_{k,2,2}^{2}$ will be based only on the resolvent method. For certain configurations in $\mathcal{Q}_{k,2,2}^{3}$ we will need to continue the expansion. This will depend on the ``new'' recollision we stopped at.
\subsection{Estimates for $\iota\in\mathcal{Q}_{k,2,2}^{1}$}
\begin{lemma}\label{expansion_aver_bound_Mainterm_rec_trun22_1}
Assume we are in the setting of Definition~\ref{def_recol_reminder}. Let $\varphi \in \mathcal{H}^{2d+2}_\hbar(\R^d)$. Then 
\begin{equation*}
	\begin{aligned}
	\mathbb{E}\Big[\big\lVert \sum_{k_1=4}^{k_0} \sum_{\iota\in\mathcal{Q}_{k,2,2}^{1}}   \sum_{\boldsymbol{x}_1\in \mathcal{X}_{\neq}^{k_1-2}}  \mathcal{E}_{2,2}^{\mathrm{rec}}(k_1,0,\boldsymbol{x}_1,\iota,\tfrac{t}{\tau_0};\hbar)\varphi\big\rVert_{L^2(\R^d)}^2 \Big]
	\leq k_0^{19} C^{2k_0}  |\log(\tfrac{\hbar}{\tau_0})|^{k_0+7}  \norm{\varphi}_{\mathcal{H}^{2d+2}_\hbar(\R^d)}^2  \hbar,
	\end{aligned}
\end{equation*}
where the constant $C$ depends on the single site potential $V$ and the coupling constant $\lambda$. In particular we have that the function is in  $L^2(\R^d)$ $\Pro$-almost surely. 
\end{lemma}
\begin{proof}
We will only sketch the proof as it is analogous to that used for $\iota\in \mathcal{Q}_{k,2,1}$. We split into the two different cases $m_{2,1}>m_{1,2}$ and $m_{2,1}<m_{1,1}$. We then write up the $L^2$-norm as in the previous proofs and get these long sums over different configurations. For the functions we obtain this we preform the change of  variables $p_{m_{1,2}-1}\mapsto p_{m_{1,2}-1} - p_{m_{1,2}}$ and $p_{m}\mapsto p_{m} - p_{m_{1,2}-1}$ for all $m\in\{m_{1,2},\dots,m_{1,2}-2\}$, and  the change of variables $p_{m_{2,2}-1}\mapsto p_{m_{2,2}-1}-p_{m_{2,2}}$ and $p_m\mapsto p_m-p_{m_{2,2}-1}$ for all $m\in\{m_{2,1},\dots,m_{2,2}-2\}\setminus\{m_{1,2}-1\}$, where  again  $m_{1,1}$, $m_{1,2}$, $m_{2,1}$ and $m_{2,2}$ are the numbers associated to the map $\iota$.  Finally we preform the change of variables $p_k\mapsto p_k -p_{k+1}$ and $p_m\mapsto p_m-p_{k}$ for all $m\in\{l,\dots,k-1\}\setminus\{m_{2,2}-1,m_{1,2}-1\}$. This we do in both cases.   

After this we then write up the the averages of the inner products using Lemma~\ref{LE:Exp_ran_phases}. For each of these terms we now need to estimate these using either the Fourier method or the resolvent method.  In general for every case where we have a crossing coming from either $l$ being between $m_{1,1}$ and $m_{1,2}$ or  between $m_{2,1}$ and $m_{2,2}$ or from one of the index es $\sigma$ being between $m_{1,1}$ and $m_{1,2}$ or  between $m_{2,1}$ and $m_{2,2}$ or $\sigma_n>l$ we will use the resolvent method. When such a crossing is not present we will use the Fourier method. 
\end{proof}
The proof of the following Lemma follows as the other proofs of similar lemmas done in the previous sections.
\begin{lemma}\label{expansion_aver_bound_Mainterm_rec_trun22_2}
Assume we are in the setting of Definition~\ref{def_recol_reminder}. Let $\varphi \in \mathcal{H}^{3d+3}_\hbar(\R^d)$. Then 
\begin{equation*}
	\begin{aligned}
	\MoveEqLeft \mathbb{E}\Big[\big\lVert  \sum_{k_1=1}^{k_0} \sum_{k_2=1}^{k_0} \sum_{\iota\in\mathcal{Q}_{k_1+k_2,2,2}^{1}}  \boldsymbol{1}_{\{k_1+k_2\geq4\}} 
	\\
	&\times 
	  \sum_{(\boldsymbol{x}_1,\boldsymbol{x}_2)\in \mathcal{X}_{\neq}^{k_1+k_2-2}}  \mathcal{E}_{2,1}^{\mathrm{rec}}(k_1,k_2,\boldsymbol{x}_1,\iota,\tfrac{t}{\tau_0};\hbar) \mathcal{I}_{2,1}^{\mathrm{rec}}(k_2,\boldsymbol{x}_2,\iota,\tfrac{\tau-1}{\tau_0}t;\hbar)\varphi\big\rVert_{L^2(\R^d)}^2 \Big]
	\\
	&\leq k_0^{21} C^{2k_0}  |\log(\tfrac{\hbar}{\tau_0})|^{k_0+9}  \norm{\varphi}_{\mathcal{H}^{3d+3}_\hbar(\R^d)}^2  \hbar,
	\end{aligned}
\end{equation*}
and
\begin{equation*}
	\begin{aligned}
	\MoveEqLeft \mathbb{E}\Big[\big\lVert  \sum_{k_1=1}^{k_0} \sum_{k_2=1}^{k_0} \sum_{\iota\in\mathcal{Q}_{k_1+k_2,2,2}^{1}}  \boldsymbol{1}_{\{k_1+k_2\geq4\}} 
	\\
	&\times 
	  \sum_{(\boldsymbol{x}_1,\boldsymbol{x}_2)\in \mathcal{X}_{\neq}^{k_1+k_2-2}} \mathcal{E}_{2,j}^{\mathrm{rec}}(k_1+1,k_2,(x_{2,k_2},\boldsymbol{x}_1),\iota,\tfrac{t}{\tau_0};\hbar) \mathcal{I}_{2,j}^{\mathrm{rec}}(k_2,\boldsymbol{x}_2,\iota,\tfrac{\tau-1}{\tau_0}t;\hbar)\varphi\big\rVert_{L^2(\R^d)}^2 \Big]
	\\
	&\leq k_0^{21} C^{2k_0}  |\log(\tfrac{\hbar}{\tau_0})|^{k_0+9}  \norm{\varphi}_{\mathcal{H}^{3d+3}_\hbar(\R^d)}^2  \hbar,
	\end{aligned}
\end{equation*}
where the constant $C$ depends on the single site potential $V$ and the coupling constant $\lambda$. In particular we have that the function is in  $L^2(\R^d)$ $\Pro$-almost surely. 
\end{lemma}
\subsection{Estimates for $\iota\in\mathcal{Q}_{k,2,2}^{2}$}
\begin{lemma}\label{expansion_aver_bound_Mainterm_rec_trun22_3}
Assume we are in the setting of Definition~\ref{def_recol_reminder}. Let $\varphi \in \mathcal{H}^{2d+2}_\hbar(\R^d)$. Then 
\begin{equation*}
	\begin{aligned}
	\mathbb{E}\Big[\big\lVert \sum_{k_1=4}^{k_0} \sum_{\iota\in\mathcal{Q}_{k,2,2}^{2}}   \sum_{\boldsymbol{x}_1\in \mathcal{X}_{\neq}^{k_1-2}}  \mathcal{E}_{2,2}^{\mathrm{rec}}(k_1,0,\boldsymbol{x}_1,\iota,\tfrac{t}{\tau_0};\hbar)\varphi\big\rVert_{L^2(\R^d)}^2 \Big]
	\leq k_0^{19} C^{2k_0}  |\log(\tfrac{\hbar}{\tau_0})|^{k_0+7}  \norm{\varphi}_{\mathcal{H}^{2d+2}_\hbar(\R^d)}^2  \hbar,
	\end{aligned}
\end{equation*}
where the constant $C$ depends on the single site potential $V$ and the coupling constant $\lambda$. In particular we have that the function is in  $L^2(\R^d)$ $\Pro$-almost surely. 
\end{lemma}
\begin{proof}
To obtain these estimates we again write up the $L^2$-norm as in the previous proofs and get these long sums over different configurations. For the functions we obtain this we preform the change of  variables $p_{m_{1,2}-1}\mapsto p_{m_{1,2}-1} - p_{m_{1,2}}$ and $p_{m}\mapsto p_{m} - p_{m_{1,2}-1}$ for all $m\in\{m_{1,2},\dots,m_{1,2}-2\}$, and  the change of variables $p_{m_{2,2}-1}\mapsto p_{m_{2,2}-1}-p_{m_{2,2}}$ and $p_m\mapsto p_m-p_{m_{2,2}-1}$ for all $m\in\{m_{2,1},\dots,m_{2,2}-2\}\setminus\{m_{1,2}-1\}$, where  again  $m_{1,1}$, $m_{1,2}$, $m_{2,1}$ and $m_{2,2}$ are the numbers associated to the map $\iota$.  Finally we preform the change of variables $p_k\mapsto p_k -p_{k+1}$ and $p_m\mapsto p_m-p_{k}$ for all $m\in\{l,\dots,k-1\}\setminus\{m_{2,2}-1,m_{1,2}-1\}$. After this we then write up the the averages of the inner products using Lemma~\ref{LE:Exp_ran_phases}. For each of these terms we now need to estimate these using the resolvent method. We always use the resolvent method to estimate these since we always have a crossing present. 
\end{proof}
\begin{lemma}\label{expansion_aver_bound_Mainterm_rec_trun22_4}
Assume we are in the setting of Definition~\ref{def_recol_reminder}. Let $\varphi \in \mathcal{H}^{3d+3}_\hbar(\R^d)$. Then 
\begin{equation*}
	\begin{aligned}
	\MoveEqLeft \mathbb{E}\Big[\big\lVert  \sum_{k_1=1}^{k_0} \sum_{k_2=1}^{k_0} \sum_{\iota\in\mathcal{Q}_{k_1+k_2,2,2}^{2}}  \boldsymbol{1}_{\{k_1+k_2\geq4\}} 
	\\
	&\times 
	  \sum_{(\boldsymbol{x}_1,\boldsymbol{x}_2)\in \mathcal{X}_{\neq}^{k_1+k_2-2}}  \mathcal{E}_{2,1}^{\mathrm{rec}}(k_1,k_2,\boldsymbol{x}_1,\iota,\tfrac{t}{\tau_0};\hbar) \mathcal{I}_{2,1}^{\mathrm{rec}}(k_2,\boldsymbol{x}_2,\iota,\tfrac{\tau-1}{\tau_0}t;\hbar)\varphi\big\rVert_{L^2(\R^d)}^2 \Big]
	\\
	&\leq k_0^{21} C^{2k_0}  |\log(\tfrac{\hbar}{\tau_0})|^{k_0+9}  \norm{\varphi}_{\mathcal{H}^{3d+3}_\hbar(\R^d)}^2  \hbar,
	\end{aligned}
\end{equation*}
and
\begin{equation*}
	\begin{aligned}
	\MoveEqLeft \mathbb{E}\Big[\big\lVert  \sum_{k_1=1}^{k_0} \sum_{k_2=1}^{k_0} \sum_{\iota\in\mathcal{Q}_{k_1+k_2,2,2}^{2}}  \boldsymbol{1}_{\{k_1+k_2\geq4\}} 
	\\
	&\times 
	  \sum_{(\boldsymbol{x}_1,\boldsymbol{x}_2)\in \mathcal{X}_{\neq}^{k_1+k_2-2}} \mathcal{E}_{2,j}^{\mathrm{rec}}(k_1+1,k_2,(x_{2,k_2},\boldsymbol{x}_1),\iota,\tfrac{t}{\tau_0};\hbar) \mathcal{I}_{2,j}^{\mathrm{rec}}(k_2,\boldsymbol{x}_2,\iota,\tfrac{\tau-1}{\tau_0}t;\hbar)\varphi\big\rVert_{L^2(\R^d)}^2 \Big]
	\\
	&\leq k_0^{21} C^{2k_0}  |\log(\tfrac{\hbar}{\tau_0})|^{k_0+9}  \norm{\varphi}_{\mathcal{H}^{3d+3}_\hbar(\R^d)}^2  \hbar,
	\end{aligned}
\end{equation*}
where the constant $C$ depends on the single site potential $V$ and the coupling constant $\lambda$. In particular we have that the function is in  $L^2(\R^d)$ $\Pro$-almost surely. 
\end{lemma}
\subsection{Estimates for $\iota\in\mathcal{Q}_{k,2,2}^{3}$}
For this case we will have to restart our expansion as the collision almost behave as a single collision happening within two potentials. It is only to estimate these terms that we are restricting the dimension to be three. Hence in this section we always assume $d=3$.
\begin{lemma}\label{expansion_aver_bound_Mainterm_rec_trun22_5}
Assume we are in the setting of Definition~\ref{def_recol_reminder}. Let $\varphi \in \mathcal{H}^{5d+5}_\hbar(\R^d)$ and for all $k$ let $\iota\in\mathcal{Q}_{k,2,2}^{3}$. Then 
\begin{equation*}
	\begin{aligned}
	\MoveEqLeft \mathbb{E}\Big[\big\lVert \sum_{k=4}^{k_0}   \sum_{\boldsymbol{x}_1\in \mathcal{X}_{\neq}^{k-2}}  \mathcal{E}_{2,2}^{\mathrm{rec}}(k,0,\boldsymbol{x}_1,\iota,\tfrac{t}{\tau_0};\hbar)\varphi\big\rVert_{L^2(\R^d)}^2 \Big]
	\\
	&\leq C \hbar k_0^{10} C^{k_0}   |\log(\tfrac{\hbar}{t})|^{k_0+9} \norm{\varphi}_{\mathcal{H}^{5d+5}_\hbar(\R^d)}^2 +C \tau_0^{-3}  \norm{\varphi}_{\mathcal{H}^{5d+5}_\hbar(\R^d)}^2 k_0^{10} C^{k_0} |\log(\tfrac{\hbar}{t})|^{k_0+20} ,
	\end{aligned}
\end{equation*}
where the constants $C$ depends on the single site potential $V$ and the coupling constant $\lambda$. In particular we have that the function is in  $L^2(\R^d)$ $\Pro$-almost surely. 
\end{lemma}
\begin{proof}
From the definition of the operator $ \mathcal{E}_{2,2}^{\mathrm{rec}}(k,0,\boldsymbol{x}_1,\iota,\tfrac{t}{\tau_0};\hbar)$ we have that
\begin{equation}\label{EQ:expansion_aver_bound_Mainterm_rec_trun22_5_1}
	\begin{aligned}
	&\big\lVert \sum_{k=4}^{k_0}   \sum_{\boldsymbol{x}_1\in \mathcal{X}_{\neq}^{k-2}}  \mathcal{E}_{2,2}^{\mathrm{rec}}(k,0,\boldsymbol{x}_1,\iota,\tfrac{t}{\tau_0};\hbar)\varphi\big\rVert_{L^2(\R^d)}^2
	\\
	&\leq\frac{2}{\hbar^2}  \big\lVert \sum_{k=5}^{k_0}   \sum_{\boldsymbol{x}_1\in \mathcal{X}_{\neq}^{k-2}}  \sum_{\alpha\in \N^k}(i\lambda)^{|\alpha|+1}    \int_{0}^t   U_{\hbar,\lambda}(-s_{k+1}) U_{\hbar,0}(s_{k+1})\sum_{l=1}^{k-4}   \tilde{\mathcal{E}}_{2,2}^{\mathrm{rec}}(s_{k+1},k,\boldsymbol{x},\alpha,l,t,\iota;\hbar) \varphi\,ds_{k+1} \big\rVert_{L^2(\R^d)}^2
	\\
	&\phantom{=}{}+ 2  \big\lVert \sum_{k=4}^{k_0}   \sum_{\boldsymbol{x}_1\in \mathcal{X}_{\neq}^{k-2}}  \sum_{\alpha\in \N^k}\frac{(i\lambda)^{|\alpha|+1} }{\hbar}   \int_{0}^t   U_{\hbar,\lambda}(-s_{k+1}) U_{\hbar,0}(s_{k+1})   \tilde{\mathcal{E}}_{2,2}^{\mathrm{rec}}(s_{k+1},k,\boldsymbol{x},\alpha,k-3,t,\iota;\hbar) \varphi\,ds_{k+1} \big\rVert_{L^2(\R^d)}^2,
	\end{aligned}
\end{equation}
where 
\begin{equation*}
	\begin{aligned}
	 \tilde{\mathcal{E}}_{2,2}^{\mathrm{rec}}(s_{k+1},k,\boldsymbol{x},\alpha,l,t,\iota;\hbar)
	=   \int_{[0,t]_{\leq}^{k}} \boldsymbol{1}_{[s_{k+1},t]}(s_k)   V^{s_{k+1}}_{\hbar, x_{l}} \prod_{m=1}^k \Theta_{\alpha_m}(s_{m-1},{s}_{m},x_{\iota(m+k_1)};V,\hbar)  \, d\boldsymbol{s}_{k,1}U_{\hbar,0}(-t). 
	\end{aligned}
\end{equation*}
We start by estimating the average of the first $L^2$-norm. From applying Lemma~\ref{LE:crossing_dom_ladder} we obtain that
\begin{equation}\label{EQ:expansion_aver_bound_Mainterm_rec_trun22_5_2}
	\begin{aligned}
	\MoveEqLeft\mathbb{E}\Big[ \big\lVert \sum_{k=5}^{k_0}   \sum_{\boldsymbol{x}_1\in \mathcal{X}_{\neq}^{k-2}}  \sum_{\alpha\in \N^k}(i\lambda)^{|\alpha|+1}    \int_{0}^t   U_{\hbar,\lambda}(-s_{k+1}) U_{\hbar,0}(s_{k+1})\sum_{l=1}^{k-4}   \tilde{\mathcal{E}}_{2,2}^{\mathrm{rec}}(s_{k+1},k,\boldsymbol{x},\alpha,l,t,\iota;\hbar) \varphi\,ds_{k+1} \big\rVert_{L^2(\R^d)}^2\Big]
	\\
	\leq{}&  2k_0^4\sum_{k=5}^{k_0}  \sum_{l=1}^{k-4}   \sum_{\alpha,\tilde{\alpha}\in \N^k}\lambda^{|\alpha|+|\tilde{\alpha}|+2}   \sum_{n=0}^{k-2} \binom{k-2}{n} n!  \sum_{\sigma\in\mathcal{A}(k-2,n)} \big|\Aver{ \mathcal{T}(k,\sigma,\alpha,\tilde{\alpha},l;\hbar)}\big|,
	\end{aligned}
\end{equation}
where we have used the notation
\begin{equation*}
	\begin{aligned}
	 \mathcal{T}(k,\sigma,\alpha,\tilde{\alpha},l;\hbar) = {}&\sum_{(\boldsymbol{x},\tilde{\boldsymbol{x}})\in \mathcal{X}_{\neq}^{2k-4}}    \prod_{i=1}^n \frac{\delta(x_{\sigma_i}- \tilde{x}_{\sigma_{i}})}{\rho}
	\\
	&\times  \int_0^t \int  \tilde{\mathcal{E}}_{2,2}^{\mathrm{rec}}(s_{k+1},k,\boldsymbol{x},\alpha,l,t,\iota;\hbar)\varphi(x)\overline{ \tilde{\mathcal{E}}_{2,2}^{\mathrm{rec}}(s_{k+1},k,\boldsymbol{\tilde{x}},\tilde{\alpha},l,t,\iota;\hbar) \varphi(x)}  \,dx ds_{k+1}.
	\end{aligned}
\end{equation*}
Before we evaluate the averages $\Aver{ \mathcal{T}(k,\sigma,\alpha,\tilde{\alpha},l;\hbar)}$ we will consider the expressions for the functions. By definition we have that
\begin{equation*}
	\begin{aligned}
 	\MoveEqLeft \tilde{\mathcal{E}}_{2,2}^{\mathrm{rec}}(s_{k+1},k,\boldsymbol{x},\alpha,l,t,\iota;\hbar)\varphi(x)
	=  \frac{1}{(2\pi\hbar)^d\hbar^{k}} \int_{\R_{+}^{|\alpha|}} \int \boldsymbol{1}_{[0,t]}( \boldsymbol{s}_{1,k+1}^{+}+ \hbar\boldsymbol{t}_{1,a_k}^{+})
	 e^{ i   \langle  \hbar^{-1}x,p_{k+1} \rangle} 
	 \\
	 &\times  e^{ -i   \langle  \hbar^{-1/d}x_{l},p_{k+1}-p_{k} \rangle}  e^{i  s_{k+1} \frac{1}{2}\hbar^{-1} p_{k+1}^2}  \hat{V}(p_{k+1}-p_k)  \prod_{i=1}^{a_k} e^{i  t_{i} \frac{1}{2} \eta_{i}^2} \Big\{\prod_{m=1}^k  e^{ -i   \langle  \hbar^{-1/d}x_{\iota(m)},p_{m}-p_{m-1} \rangle}  
	\\
	&\times e^{i \hbar^{-1}  s_{m} \frac{1}{2} p_{m}^2}   \hat{\mathcal{V}}_{\alpha_m}(p_m,p_{m-1},\boldsymbol{\eta} )  \Big\}    e^{i  \hbar^{-1}(t- \boldsymbol{s}_{1,k+1}^{+}- \hbar \boldsymbol{t}_{1,a_k}^{+}) \frac{1}{2} p_0^2} \hat{\varphi}(\tfrac{p_0}{\hbar}) \, d\boldsymbol{\eta}d\boldsymbol{p}   d\boldsymbol{t} d\boldsymbol{s}.
	\end{aligned}
\end{equation*}
Firstly we preform the following change of variables in the mentioned order $p_{k-2}\mapsto p_{k-2}-p_{k-1}$, $p_{k-3}\mapsto p_{k-3}-p_{k-2}$, $p_k\mapsto p_k-p_{k+1}$ and $p_m\mapsto p_m - p_{k}$ for all $m\in\{l,\dots,k-3\}$ and a relabelling this gives us the expression
\begin{equation*}
	\begin{aligned}
 	\MoveEqLeft \tilde{\mathcal{E}}_{2,2}^{\mathrm{rec}}(s_{k+1},k,\boldsymbol{x},\alpha,l,t,\iota;\hbar)\varphi(x)
	=  \frac{1}{(2\pi\hbar)^d\hbar^{k}} \int_{\R_{+}^{|\alpha|}} \int_{\R^{(|\alpha|+2)d}} \boldsymbol{1}_{[0,t]}( \boldsymbol{s}_{1,k+1}^{+}+ \hbar\boldsymbol{t}_{1,a_k}^{+})
	 e^{ i   \langle  \hbar^{-1}x,p_{k-2} \rangle} 
	 \\
	 &\times  e^{i  s_{k+1} \frac{1}{2}\hbar^{-1} p_{k-2}^2}\hat{V}(-\tilde{p}_3) e^{i \hbar^{-1}  s_{k} \frac{1}{2} (p_{k-2}+\tilde{p}_3)^2}   \hat{\mathcal{V}}_{\alpha_k}(p_{k-2}+\tilde{p}_3,\tilde{p}_{2},\boldsymbol{\eta} ) e^{i \hbar^{-1}  s_{k-1} \frac{1}{2} \tilde{p}_{2}^2}   \hat{\mathcal{V}}_{\alpha_{k-1}}(\tilde{p}_{2},\tilde{p}_{2}+\tilde{p}_1,\boldsymbol{\eta} )
	 \\
	 &\times e^{i \hbar^{-1}  s_{k-2} \frac{1}{2} (\tilde{p}_{1} +\tilde{p}_2 )^2}   \hat{\mathcal{V}}_{\alpha_{k-2}}(\tilde{p}_1+\tilde{p}_2,\tilde{p}_1 + p_{k-3}+\tilde{p}_3,\boldsymbol{\eta} )
	 e^{i \hbar^{-1}  s_{k-3} \frac{1}{2} (\tilde{p}_{1} +\tilde{p}_3+p_{k-3} )^2}  
	 \\
	 &\times  \hat{\mathcal{V}}_{\alpha_{k-3}}(\tilde{p}_1+\tilde{p}_3+p_{k-3},p_{k-4}+\tilde{p}_3,\boldsymbol{\eta} )
	   \prod_{i=1}^{a_k} e^{i  t_{i} \frac{1}{2} \eta_{i}^2} \prod_{m=1}^{k-2}  e^{ -i   \langle  \hbar^{-1/d}x_{m},p_{m}-p_{m-1} \rangle}   \prod_{m=1}^{k-4} \Big\{e^{i \hbar^{-1}  s_{m} \frac{1}{2} (p_{m}+\pi_{m}(\tilde{p}_3 ))^2} 
	 \\
	 &\times   
	   \hat{\mathcal{V}}_{\alpha_m}(p_m+\pi_{m}(\tilde{p}_3 ),p_{m-1}+\pi_{m-1}(\tilde{p}_3 ),\boldsymbol{\eta} ) \Big\}   e^{i  \hbar^{-1}(t- \boldsymbol{s}_{1,k+1}^{+}- \hbar \boldsymbol{t}_{1,a_k}^{+}) \frac{1}{2} p_0^2} \hat{\varphi}(\tfrac{p_0}{\hbar}) \, d\boldsymbol{\eta}d\boldsymbol{p}   d\boldsymbol{t} d\boldsymbol{s},
	\end{aligned}
\end{equation*}
where we have used the notation
\begin{equation*}
	\pi_{m}(\tilde{p}_3 ) = \tilde{p}_3 \boldsymbol{1}_{\{l,\dots,k-4\}}(m).
\end{equation*}
This is the expression for the function we will be working with for this case. 
To estimate these averages we can argue as in the proofs of Lemma~\ref{LE:Exp_ran_phases} and Lemma~\ref{expansion_aver_bound_Mainterm} and obtain that
\begin{equation*}
	\begin{aligned}
 	\mathbb{E}&\big[  \mathcal{T}(k,\sigma,\alpha,\tilde{\alpha},l;\hbar) \big] =   \frac{(2\pi)^{(2k-2-n)d}(\rho\hbar)^{2k-4-n}}{(2\pi\hbar)^d}  \int_{\R_{+}^{|\alpha|+|\tilde{\alpha}|+2}} f(\boldsymbol{s},\boldsymbol{\tilde{s}}) \int \Lambda_n(\boldsymbol{p},\boldsymbol{q},\sigma)    \hat{\varphi}(\tfrac{p_0}{\hbar}) \overline{\hat{\varphi}(\tfrac{q_0}{\hbar})} 
	\\
	\times & \delta(\tfrac{t}{\hbar}- \boldsymbol{s}_{0,k}^{+}-\tilde{s}_{k+1}- \boldsymbol{t}_{1,a_{k+\beta}}^{+})  \delta(\tfrac{t}{\hbar}- \boldsymbol{\tilde{s}}_{1,k+1}^{+}- \boldsymbol{\tilde{t}}_{1,a_{k+\beta}}^{+}) \prod_{i=1}^{a_{k+\beta}} e^{i  t_{i} \frac{1}{2} \eta_{i}^2} \prod_{i=1}^{\tilde{a}_{k+\beta}} e^{-i  \tilde{t}_{i} \frac{1}{2} \xi_{i}^2} 
	\\
	\times& 
	e^{i \hbar^{-1}  s_{k} \frac{1}{2} (p_{k-2}+\tilde{p}_3)^2} e^{-i \hbar^{-1}  \tilde{s}_{k} \frac{1}{2} (q_{k-2}+\tilde{q}_3)^2}
	 e^{i \hbar^{-1}  s_{k-1} \frac{1}{2} \tilde{p}_{2}^2}  e^{-i \hbar^{-1}  \tilde{s}_{k-1} \frac{1}{2} \tilde{q}_{2}^2}  
	  e^{i \hbar^{-1}  s_{k-2} \frac{1}{2} (\tilde{p}_{1} +\tilde{p}_2 )^2}   e^{-i \hbar^{-1}  \tilde{s}_{k-2} \frac{1}{2} (\tilde{q}_{1} +\tilde{q}_2 )^2} 
	\\
	\times&  
	e^{i \hbar^{-1}  s_{k-3} \frac{1}{2} (\tilde{p}_{1} +\tilde{p}_3+p_{k-3} )^2}   e^{-i \hbar^{-1}  \tilde{s}_{k-3} \frac{1}{2} (\tilde{q}_{1} +\tilde{q}_3+q_{k-3} )^2}  
	 \prod_{m=0}^{k-4} e^{i \hbar^{-1}  s_{m} \frac{1}{2} (p_{m}+\pi_{m}(\tilde{p}_3 ))^2}  e^{-i \hbar^{-1}  \tilde{s}_{m} \frac{1}{2} (q_{m}+\pi_{m}(\tilde{q}_3 ))^2} 
	\\
	\times&  \mathcal{G}(\boldsymbol{p}, \boldsymbol{\tilde{p}},\boldsymbol{\tilde{q}},\boldsymbol{\eta}, \boldsymbol{\xi},\sigma ) d\boldsymbol{\eta} d\boldsymbol{\xi}d\boldsymbol{p}d\boldsymbol{q} d\boldsymbol{\tilde{p}}  d\boldsymbol{\tilde{q}}  d\boldsymbol{t} d\boldsymbol{\tilde{t}}  d\boldsymbol{s} d\boldsymbol{\tilde{s}},
	\end{aligned}
\end{equation*}
where the function $f$ is given by
\begin{equation*}
	f(\boldsymbol{s},\boldsymbol{\tilde{s}}) =
	\begin{cases}
	 \boldsymbol{1}_{[0,\hbar^{-1}t]}\big( \sum_{i=1}^{k-4} s_i  \big) \boldsymbol{1}_{[0,\hbar^{-1}t]}\big( \sum_{i=1,i \notin \sigma}^{k-4} \tilde{s}_i + \tilde{s}_{k+1} \big) &\text{if $\sigma_n\leq l$}
	 \\
	 \boldsymbol{1}_{[0,\hbar^{-1}t]}\big( \sum_{i=1,i\neq l}^{k-4} s_i   \big) \boldsymbol{1}_{[0,\hbar^{-1}t]}\big( \sum_{i=1,i \notin \sigma}^{k-3} \tilde{s}_i  +\tilde{s}_{k+1} \big) &\text{if $\sigma_n> l$}.
	 \end{cases}
\end{equation*}
The function $\mathcal{G}$ is given by
\begin{equation*}
	\begin{aligned}
	\MoveEqLeft \mathcal{G}(\boldsymbol{p}, \boldsymbol{\tilde{p}},\boldsymbol{\tilde{q}},\boldsymbol{\eta}, \boldsymbol{\xi},\sigma ) =  
	\hat{V}(-\tilde{p}_3) \overline{\hat{V}(-\tilde{q}_3)}  \hat{\mathcal{V}}_{\alpha_k}(p_{k-2}+\tilde{p}_3,\tilde{p}_{2},\boldsymbol{\eta} ) \overline{ \hat{\mathcal{V}}_{\alpha_k}(q_{k-2}+\tilde{q}_3,\tilde{q}_{2},\boldsymbol{\eta} )}    \hat{\mathcal{V}}_{\alpha_{k-1}}(\tilde{p}_{2},\tilde{p}_{2}+\tilde{p}_1,\boldsymbol{\eta} )
	\\
	&\times \overline{\hat{\mathcal{V}}_{\alpha_{k-1}}(\tilde{q}_{2},\tilde{q}_{2}+\tilde{q}_1,\boldsymbol{\eta} )}   \hat{\mathcal{V}}_{\alpha_{k-2}}(\tilde{p}_1+\tilde{p}_2,\tilde{p}_1 + p_{k-3}+\tilde{p}_3,\boldsymbol{\eta} ) \overline{  \hat{\mathcal{V}}_{\alpha_{k-2}}(\tilde{q}_1+\tilde{q}_2,\tilde{q}_1 + q_{k-3}+\tilde{q}_3,\boldsymbol{\xi} )}
	 \\
	 &\times  \hat{\mathcal{V}}_{\alpha_{k-3}}(\tilde{p}_1+\tilde{p}_3+p_{k-3},p_{k-4}+\tilde{p}_3,\boldsymbol{\eta} ) \overline{ \hat{\mathcal{V}}_{\tilde{\alpha}_{k-3}}(\tilde{q}_1+\tilde{q}_3+q_{k-3},q_{k-4}+\tilde{q}_3,\boldsymbol{\xi} )}
	    \\
	    &\times \prod_{m=1}^{k-4}  
	   \hat{\mathcal{V}}_{\alpha_m}(p_m+\pi_{m}(\tilde{p}_3 ),p_{m-1}+\pi_{m-1}(\tilde{p}_3 ),\boldsymbol{\eta} ) \overline{ \hat{\mathcal{V}}_{\tilde{\alpha}_m}(q_m+\pi_{m}(\tilde{q}_3 ),q_{m-1}+\pi_{m-1}(\tilde{q}_3 ),\boldsymbol{\xi} ) }.
	\end{aligned}
\end{equation*}
  The function $\Lambda_n(\boldsymbol{p},\boldsymbol{q},\sigma)$ is given by
  \begin{equation*}
  	\Lambda_n(\boldsymbol{p},\boldsymbol{q},\sigma) = \delta(p_{k-2}-q_{k-2})  \prod_{i=1}^n \delta(p_{\sigma_{i-1}}- q_{\sigma{i-1}} - p_{\sigma_n}+ q_{\sigma_n} ))   
	\prod_{i=1}^{n+1} \prod_{m=\sigma_{i-1}+1}^{\sigma_{i}-1} \delta(p_m-p_{\sigma_{i-1}}) \delta(q_m-q_{\sigma^1_{i-1}}).
  \end{equation*}
After having obtained this form the proof is analogous to the proofs previously done. That is we introduce the function $\zeta$, write our delta functions as the Fourier transform of $1$. Next we integrate in all $s$ and $\tilde{s}$ that the function $f$ does not depend on. Moreover, we also do the usual argument by dividing into different cases depending on the size of $t$ and $\tilde{t}$ and do integration by parts for all for all cases, where the variables is ``large''. 
 \begin{equation}\label{EQ:expansion_aver_bound_Mainterm_rec_trun22_5_6}
   	\begin{aligned}
	\mathbb{E}\big[  \mathcal{T}(k,\sigma,\alpha,\tilde{\alpha},l;\hbar) \big] 
	 \leq   C_d^{|\alpha|+|\tilde{\alpha}|}   \norm{\hat{V}}_{1,\infty,6d+6}^{|\alpha|+|\tilde{\alpha}|+2}   \frac{\rho^2(\rho t)^{2k-n-6} \hbar^3 |\log(\zeta)|^{n+9}}{(k-5)!(k-n-1)!} \norm{\varphi}_{\mathcal{H}^{5d+5}_\hbar(\R^d)}^2.
	\end{aligned}
\end{equation}
By combining \eqref{EQ:expansion_aver_bound_Mainterm_rec_trun22_5_2} and \eqref{EQ:expansion_aver_bound_Mainterm_rec_trun22_5_6} we obtain that
\begin{equation}\label{EQ:expansion_aver_bound_Mainterm_rec_trun22_5_7}
	\begin{aligned}
	\MoveEqLeft\mathbb{E}\Big[ \big\lVert \sum_{k=5}^{k_0}   \sum_{\boldsymbol{x}_1\in \mathcal{X}_{\neq}^{k-2}}  \sum_{\alpha\in \N^k}(i\lambda)^{|\alpha|+1}    \int_{0}^t   U_{\hbar,\lambda}(-s_{k+1}) U_{\hbar,0}(s_{k+1})\sum_{l=1}^{k-4}   \tilde{\mathcal{E}}_{2,2}^{\mathrm{rec}}(s_{k+1},k,\boldsymbol{x},\alpha,l,t,\iota;\hbar) \varphi\,ds_{k+1} \big\rVert_{L^2(\R^d)}^2\Big]
	\\
	\leq{}& C \hbar^3 k_0^{10} C^{k_0}   |\log(\tfrac{\hbar}{t})|^{k_0+9} \norm{\varphi}_{\mathcal{H}^{5d+5}_\hbar(\R^d)}^2.
	\end{aligned}
\end{equation}
We now turn to the second term in \eqref{EQ:expansion_aver_bound_Mainterm_rec_trun22_5_1}. We will for this case have to continue the Duhamel expansion. We will stop if we obtain a new genuine collision center or if we see a new recollision that is not in either $x_{k-2}$ or $x_{k-3}$. If we continue to see the pattern where we jump between $x_{k-2}$ and $x_{k-3}$ we continue the expansion. This yields the following formal expression
\begin{equation*}
	\begin{aligned}
	 \MoveEqLeft \sum_{k=4}^{k_0}   \sum_{\boldsymbol{x}_1\in \mathcal{X}_{\neq}^{k-2}}  \sum_{\alpha\in \N^k} \frac{(i\lambda)^{|\alpha|+1} }{\hbar}   \int_{0}^t   U_{\hbar,\lambda}(-s_{k+1}) U_{\hbar,0}(s_{k+1})   \tilde{\mathcal{E}}_{2,2}^{\mathrm{rec}}(s_{k+1},k,\boldsymbol{x},\alpha,k-3,t,\iota;\hbar) \varphi\,ds_{k+1}
	 \\
	&= \sum_{\beta=1}^\infty \Big[  \sum_{k=4}^{k_0} \mathcal{I}^{\mathrm{dob}}_1(\beta,k;\hbar) +  \sum_{k=5}^{k_0}  \mathcal{I}^{\mathrm{dob}}_2(\beta,k;\hbar) 
	+\sum_{k=4}^{k_0}   \mathcal{I}^{\mathrm{dob}}_3(\beta,k;\hbar)  \Big]
	\end{aligned}
\end{equation*}
where the operators $\mathcal{I}^{\mathrm{dob}}_{j}(\beta,k;\hbar) $ are given by
\begin{equation*}
	\begin{aligned}
	 \mathcal{I}^{\mathrm{dob}}_1(\beta,k;\hbar)
	= \sum_{\boldsymbol{x}\in\mathcal{X}_{\neq}^{k-2}} \sum_{\alpha^1\in\N^{k_1}}   \int_{[0,t]^{\beta+k}}
	 \tilde{\Theta}_{\beta}^{\mathrm{dob}}(\boldsymbol{s}_{\beta+k,k};V,\hbar) 
	\prod_{m=1}^{k_1}\Theta_{\alpha_m}(s_{m-1},{s}_{m},x_{\iota(m)};V,\hbar)  \, d\boldsymbol{s}_{k,1} U_{\hbar,0}(-t),
	\end{aligned}
\end{equation*}
\begin{equation*}
	\begin{aligned}	
	\mathcal{I}^{\mathrm{dob}}_2(\beta,k;\hbar)=\frac{i\lambda}{\hbar} &\sum_{l=1}^{k-4} \sum_{\boldsymbol{x}\in\mathcal{X}_{\neq}^{k-2}}  \sum_{\alpha^1\in\N^{k_1}}   \int_{[0,t]^{\beta+k+1}} U_{\hbar,\lambda}(-s_{k_{\beta+k+1}}) U_{\hbar,0}(s_{k_{\beta+k+1}})
	 V^{s_{k+\beta+1}}_{\hbar, x_{l}}  
	  \\
	  &\times \tilde{\Theta}_{\beta}^{\mathrm{dob}}(\boldsymbol{s}_{\beta+k,k};V,\hbar) \prod_{m=1}^{k}\Theta_{\alpha_m^1}(s_{m-1},{s}_{m},x_{\iota(m)};V,\hbar)  \, d\boldsymbol{s}_{k,1} U_{\hbar,0}(-t),
	 \end{aligned}
\end{equation*}
and
\begin{equation*}
	\begin{aligned}	
	\mathcal{I}^{\mathrm{dob}}_3(\beta,k;\hbar)=\frac{i\lambda}{\hbar} & \sum_{\boldsymbol{x}\in\mathcal{X}_{\neq}^{k-1}}  \sum_{\alpha^1\in\N^{k_1}}   \int_{[0,t]^{\beta+k+1}} U_{\hbar,\lambda}(-s_{k_{\beta+k+1}}) U_{\hbar,0}(s_{k_{\beta+k+1}})
	 V^{s_{k+\beta+1}}_{\hbar, x_{k-1}}  
	  \\
	  &\times \tilde{\Theta}_{\beta}^{\mathrm{dob}}(\boldsymbol{s}_{\beta+k,k};V,\hbar) \prod_{m=1}^{k}\Theta_{\alpha_m^1}(s_{m-1},{s}_{m},x_{\iota(m)};V,\hbar)  \, d\boldsymbol{s}_{k,1} U_{\hbar,0}(-t).
	 \end{aligned}
\end{equation*}
Moreover, we have her introduced the operators $\tilde{\Theta}_{\beta}^{\mathrm{dob}}(\boldsymbol{s}_{\beta+k,k};V,\hbar)$, which is defined by
\begin{equation*}
	\begin{aligned}
  	\tilde{\Theta}_{\beta}^{\mathrm{dob}}(\boldsymbol{s}_{\beta+m,k};V,\hbar)  
	= \sum_{\tilde{\alpha}\in \N^\beta} \prod_{m=1}^{\beta}\Theta_{\tilde{\alpha}_m}(s_{m-1+k},{s}_{m+k},x_{\tilde{\iota}(m)};V,\hbar) ,
	\end{aligned}
\end{equation*}
where 
\begin{equation*}
	\tilde{\iota}(m)=
	\begin{cases}
	k-3 & \text{for $m$ odd} \\
	k-2 & \text{for $m$ even}.
	\end{cases}
\end{equation*}
With this further expansion we obtain the estimate
\begin{equation*}
	\begin{aligned}
	\MoveEqLeft  \mathbb{E}\Big[\big\lVert\sum_{k=4}^{k_0}   \sum_{\boldsymbol{x}_1\in \mathcal{X}_{\neq}^{k-2}}  \sum_{\alpha\in \N^k} \frac{(i\lambda)^{|\alpha|+1} }{\hbar}   \int_{0}^t   U_{\hbar,\lambda}(-s_{k+1}) U_{\hbar,0}(s_{k+1})   \tilde{\mathcal{E}}_{2,2}^{\mathrm{rec}}(s_{k+1},k,\boldsymbol{x},\alpha,k-3,t,\iota;\hbar) \varphi\,ds_{k+1} \varphi \big\rVert_{L^2(\R^d)}^2 \Big]
	\\
	&\leq \sum_{\beta=1}^\infty \beta^2 \sum_{j=1}^3   \mathbb{E}\Big[\big\lVert  \sum_{k=k_j}^{k_0} \mathcal{I}^{\mathrm{dob}}_j(\beta,k;\hbar) \varphi \big\rVert_{L^2(\R^d)}^2 \Big] ,
	\end{aligned}
\end{equation*}
where $k_1=k_3=4$ and $k_2=5$.
By applying Lemma~\ref{expansion_aver_bound_Mainterm_rec_trun22_special_case_1}, Lemma~\ref{expansion_aver_bound_Mainterm_rec_trun22_special_case_2}, and  Lemma~\ref{expansion_aver_bound_Mainterm_rec_trun22_special_case_3}  we obtain that
\begin{equation}\label{EQ:expansion_aver_bound_Mainterm_rec_trun22_5_8}
	\begin{aligned}
	\MoveEqLeft  \mathbb{E}\Big[\big\lVert\sum_{k=4}^{k_0}   \sum_{\boldsymbol{x}_1\in \mathcal{X}_{\neq}^{k-2}}  \sum_{\alpha\in \N^k} \frac{(i\lambda)^{|\alpha|+1} }{\hbar}   \int_{0}^t   U_{\hbar,\lambda}(-s_{k+1}) U_{\hbar,0}(s_{k+1})   \tilde{\mathcal{E}}_{2,2}^{\mathrm{rec}}(s_{k+1},k,\boldsymbol{x},\alpha,k-3,t,\iota;\hbar) \varphi\,ds_{k+1} \varphi \big\rVert_{L^2(\R^d)}^2 \Big]
	\\
	\leq{}& \sum_{\beta=1}^\infty \beta^2 \Big\{    \tilde{\lambda}^\beta  \norm{\varphi}_{\mathcal{H}^{4d+4}_\hbar(\R^d)}^2 k_0^5 C^{k_0} \hbar  |\log(\tfrac{\hbar}{t})|^{k_0+19} + \hbar \tilde{\lambda}^\beta  \norm{\varphi}_{\mathcal{H}^{5d+5}_\hbar(\R^d)}^2 k_0^{10} C^{k_0} \hbar  |\log(\tfrac{\hbar}{t})|^{k_0+19} 
	\\
	&+   \tau_0^{-3} \tilde{\lambda}^\beta  \norm{\varphi}_{\mathcal{H}^{5d+5}_\hbar(\R^d)}^2 k_0^{5} C^{k_0}   |\log(\tfrac{\hbar}{t})|^{k_0+20}\Big\} 
	\\
	\leq{}& C \tau_0^{-3}  \norm{\varphi}_{\mathcal{H}^{5d+5}_\hbar(\R^d)}^2 k_0^{10} C^{k_0}  |\log(\tfrac{\hbar}{t})|^{k_0+20}.
	\end{aligned}
\end{equation}
From combining the estimates in \cref{EQ:expansion_aver_bound_Mainterm_rec_trun22_5_7,EQ:expansion_aver_bound_Mainterm_rec_trun22_5_8,EQ:expansion_aver_bound_Mainterm_rec_trun22_5_1} we obtain the desired estimate.
\end{proof}
 \begin{lemma}\label{expansion_aver_bound_Mainterm_rec_trun22_special_case_1}
Assume we are in the setting of Lemma~\ref{expansion_aver_bound_Mainterm_rec_trun22_5}. Then for any $\beta\in\N$ we have that
\begin{equation*}
	\begin{aligned}
	 \mathbb{E}\Big[\big\lVert  \sum_{k=4}^{k_0}  \mathcal{I}^{\mathrm{dob}}_1(\beta,k;\hbar) \varphi \big\rVert_{L^2(\R^d)}^2 \Big]
	\leq \tilde{\lambda}^\beta  \norm{\varphi}_{\mathcal{H}^{4d+4}_\hbar(\R^d)}^2 k_0^5 C^{k_0} \hbar  |\log(\tfrac{\hbar}{t})|^{k_0+19},
	\end{aligned}
\end{equation*}
where $\tilde{\lambda}<1$.  
\end{lemma}
\begin{proof}
Firstly we note that
\begin{equation}\label{EQ:expansion_aver_bound_3.recol_22_3.1_0}
\big\lVert  \sum_{k=4}^{k_0} \mathcal{I}^{\mathrm{dob}}_1(\beta,k;\hbar) \varphi \big\rVert_{L^2(\R^d)}^2 \leq k_0^2\sum_{k=4}^{k_0}  \big\lVert \mathcal{I}^{\mathrm{dob}}_1(\beta,k;\hbar) \varphi \big\rVert_{L^2(\R^d)}^2.
\end{equation}
Using the definition of the operator $\mathcal{I}^{\mathrm{dob}}_1(\beta,k;\hbar) $ and applying Lemma~\ref{LE:crossing_dom_ladder} we obtain that 
\begin{equation}\label{EQ:expansion_aver_bound_3.recol_22_3.1_2.1}
	\begin{aligned}
	\mathbb{E}\Big[ \big\lVert \mathcal{I}^{\mathrm{dob}}_1(\beta,k;\hbar) \varphi \big\rVert_{L^2(\R^d)}^2\Big]
	 \leq 2 \sum_{n=0}^{k-2} \binom{k-2}{n} n! \sum_{\alpha,\tilde{\alpha}\in\N^{k+\beta}} \lambda^{|\alpha|+|\tilde{\alpha}|} \sum_{\sigma\in\mathcal{A}(k-2,n)} \big|\Aver{ \mathcal{T}(\beta,k,\sigma,\alpha,\tilde{\alpha};\hbar)}\big|,
	\end{aligned}
\end{equation}
where
\begin{equation*}
	\begin{aligned}
	\MoveEqLeft \mathcal{T}(\beta,k,\sigma,\alpha,\tilde{\alpha};\hbar) \\
	&= 
	 \sum_{(\boldsymbol{x},\boldsymbol{\tilde{x}})\in \mathcal{X}_{\neq}^{2k-4}} \prod_{i=1}^n \frac{\delta(x_{\sigma_i}- \tilde{x}_{\sigma_{i}})}{\rho}    \int_{\R^d} \mathcal{I}^{\mathrm{dob}}_1(\beta,k;\alpha,\boldsymbol{x},\hbar) \varphi (x)\overline{ \mathcal{I}^{\mathrm{dob}}_1(\beta,k;\tilde{\alpha},\boldsymbol{\tilde{x}},\hbar) \varphi (x)}  \,dx. 
	\end{aligned}
\end{equation*}
Before we proceed we will divide into the two cases $\beta\leq 6$ and $\beta>6$. We will start with the first case. From the definition of the operator we have that
\begin{equation*}
	\begin{aligned}
 	\MoveEqLeft  \mathcal{I}^{\mathrm{dob}}_1(\beta,k,\boldsymbol{x};\hbar) \varphi (x)
	=  \frac{1}{(2\pi\hbar)^d} \sum_{\alpha\in \N^{k+\beta}} (i\lambda)^{|\alpha|} \int_{\R_{+}^{|\alpha|}} \int_{\R^{(|\alpha|+1)d}} \boldsymbol{1}_{[0,\hbar^{-1}t]}( \boldsymbol{s}_{1,k+\beta}^{+}+ \boldsymbol{t}_{1,a_{k+\beta}}^{+})
	 e^{ i   \langle  \hbar^{-1}x,p_{k+\beta} \rangle}   \prod_{i=1}^{a_{k+\beta}} e^{i  t_{i} \frac{1}{2} \eta_{i}^2} 
	 \\
	 \times&    
	\prod_{m=-3}^\beta e^{ -i   \langle  \hbar^{-1/d}x_{\tilde{\iota}(k)}, p_{k+m}-p_{k+m-1} \rangle}  
	\prod_{m=1}^{k-4}  e^{ -i   \langle  \hbar^{-1/d}x_{m},p_{m}-p_{m-1} \rangle}      
	 \prod_{m=1}^{k+\beta}  
	e^{i  s_{m} \frac{1}{2} p_{m}^2}   \hat{\mathcal{V}}_{\alpha_m}(p_m,p_{m-1},\boldsymbol{\eta} ) 
	\\
	\times & e^{i ( \hbar^{-1}t- \boldsymbol{s}_{1,k+\beta}^{+}- \boldsymbol{t}_{1,a_{k+\beta}}^{+}) \frac{1}{2} p_0^2} \hat{\varphi}(\tfrac{p_0}{\hbar}) \, d\boldsymbol{\eta}d\boldsymbol{p}  d\boldsymbol{t} d\boldsymbol{s},
	\end{aligned}
\end{equation*}
Before taking the average we will preform a series  of change of variables  we will start with the change of variables $p_{k-2}\mapsto p_{k-2}-p_{k-1}$ and $p_{k-3}\mapsto p_{k-3}-p_{k-2}$. To motivate the next change of variables we remark that under this change of variables we get that
\begin{equation*}
	\begin{aligned}
	\prod_{m=-3}^1 e^{ -i   \langle  \hbar^{-1/d}x_{\tilde{\iota}(k)}, p_{k+m}-p_{k+m-1} \rangle} 
	&\mapsto  e^{ -i   \langle  \hbar^{-1/d}x_{k-2}, p_{k}-p_{k-1} \rangle}  e^{ -i   \langle  \hbar^{-1/d}x_{k-2}, p_{k-1}-p_{k-3} \rangle}  e^{ -i   \langle  \hbar^{-1/d}x_{k-3}, p_{k-3}-p_{k-4} \rangle}
	\\
	& = e^{ -i   \langle  \hbar^{-1/d}x_{k-2}, p_{k}-p_{k-3} \rangle}   e^{ -i   \langle  \hbar^{-1/d}x_{k-3}, p_{k-3}-p_{k-4} \rangle}.
	\end{aligned}
\end{equation*}
With this observation in mind we proceed with the change of variables 
\begin{equation*}
p_{k+m}\mapsto p_{k+m}-p_{k+m+1} \quad\text{and}\quad  p_{k-3}\mapsto p_{k-3}-p_{k+m}\quad  \text{for all }
\begin{cases}
	m\in\{0\} &\text{if $\beta\in\{1,2\}$}
	\\
	m\in\{0,2\} &\text{if $\beta\in\{3,4\}$}
	\\
	m\in\{0,2,4\} &\text{if $\beta\in\{5,6\}$}.
\end{cases}
\end{equation*}
After this change of variables and a relabelling we obtain that
\begin{equation*}
	\begin{aligned}
 	\MoveEqLeft  \mathcal{I}^{\mathrm{dob}}_1(\beta,k,\boldsymbol{x};\hbar) \varphi (x)
	=  \frac{1}{(2\pi\hbar)^d} \sum_{\alpha\in \N^{k+\beta}} (i\lambda)^{|\alpha|} \int_{\R_{+}^{|\alpha|}} \int \boldsymbol{1}_{[0,\hbar^{-1}t]}( \boldsymbol{s}_{1,k+\beta}^{+}+ \boldsymbol{t}_{1,a_{k+\beta}}^{+})
	 e^{ i   \langle  \hbar^{-1}x,p_{k-2} \rangle}   \prod_{i=1}^{a_{k+\beta}} e^{i  t_{i} \frac{1}{2} \eta_{i}^2} 
	\\
	\times & \tilde{\mathcal{G}}_\beta (p_{k-2},\boldsymbol{\tilde{p}},\boldsymbol{\eta}) \mathcal{P}_\beta (p_{k-2},\boldsymbol{\tilde{p}},\boldsymbol{s})  e^{i  s_{k-2} \frac{1}{2} (\tilde{p}_{-2}+ \tilde{p}_{-1})^2}   \hat{\mathcal{V}}_{\alpha_{k-2}}(\tilde{p}_{-2}+ \tilde{p}_{-1},p_{k-3}+ \pi_{\beta}(\boldsymbol{\tilde{p}}),\boldsymbol{\eta} )   e^{i  s_{k-3} \frac{1}{2} (p_{k-3}+ \pi_{\beta}(\boldsymbol{\tilde{p}}))^2}
	\\
	\times&   \hat{\mathcal{V}}_{\alpha_{k-3}}(p_{k-3}+ \pi_{\beta}(\boldsymbol{\tilde{p}}),p_{k-4},\boldsymbol{\eta} ) 
	 \prod_{m=1}^{k-2}  e^{ -i   \langle  \hbar^{-1/d}x_{m},p_{m}-p_{m-1} \rangle}   \prod_{m=1}^{k-4}  
	e^{i  s_{m} \frac{1}{2} p_{m}^2}   \hat{\mathcal{V}}_{\alpha_m}(p_m,p_{m-1},\boldsymbol{\eta} ) 
	 \\
	 \times& e^{i ( \hbar^{-1}t- \boldsymbol{s}_{1,k+\beta}^{+}- \boldsymbol{t}_{1,a_{k+\beta}}^{+}) \frac{1}{2} p_0^2} \hat{\varphi}(\tfrac{p_0}{\hbar}) \, d\boldsymbol{\eta}d\boldsymbol{p}   d\boldsymbol{t} d\boldsymbol{s},
	\end{aligned}
\end{equation*}
where
\begin{equation*}
	\pi_{\beta}(\boldsymbol{\tilde{p}}) =
	\begin{cases}
	\tilde{p}_{-2} + \tilde{p}_0 &\text{if $\beta\in\{1,2\}$}
	\\
	\tilde{p}_{-2} + \tilde{p}_0 +\tilde{p}_{2} &\text{if $\beta\in\{3,4\}$}
	\\
	\tilde{p}_{-2} + \tilde{p}_0 +\tilde{p}_{2} +\tilde{p}_{4}&\text{if $\beta\in\{5,6\}$}.
	\end{cases}
\end{equation*}
For $\beta=1$ we have that
\begin{equation*}
	\begin{aligned}
	\tilde{\mathcal{G}}_1 (p_{k-2},\boldsymbol{\tilde{p}},\boldsymbol{\eta}) ={}&  
	\hat{\mathcal{V}}_{\alpha_{k-1}}(\tilde{p}_{-1},\tilde{p}_{-1}+\tilde{p}_{-2},\boldsymbol{\eta} )
	\hat{\mathcal{V}}_{\alpha_{k}}(\tilde{p}_{0}+p_{k-2},\tilde{p}_{-1},\boldsymbol{\eta} )
	\hat{\mathcal{V}}_{\alpha_{k+1}}(p_{k-2},\tilde{p}_{0}+p_{k-2},\boldsymbol{\eta} )
	\\
	\mathcal{P}_1 (p_{k-2},\boldsymbol{\tilde{p}},\boldsymbol{s}) = {}& e^{i  s_{k-1} \frac{1}{2} \tilde{p}_{-1}^2} 
	e^{i  s_{k} \frac{1}{2}( \tilde{p}_{0}+p_{k-2})^2}  e^{i  s_{k+1} \frac{1}{2}p_{k-2}^2}
	\end{aligned}
\end{equation*}
and for $\beta\geq2$ we have that
\begin{equation*}
	\begin{aligned}
	\tilde{\mathcal{G}}_\beta (p_{k-2},\boldsymbol{\tilde{p}},\boldsymbol{s},\boldsymbol{\eta}) &=
	\begin{cases}  
	 \hat{\mathcal{V}}_{\alpha_{k+\beta}}(p_{k-2},\tilde{p}_{\beta-1},\boldsymbol{\eta} )\tilde{\mathcal{G}}_{\beta-1} (\tilde{p}_{\beta-2},\boldsymbol{\tilde{p}},\boldsymbol{\eta}) & \text{$\beta$ even}
	\\
	\hat{\mathcal{V}}_{\alpha_{k+\beta}}(p_{k-2}, p_{k-2} + \tilde{p}_{\beta-1},\boldsymbol{\eta} )\tilde{\mathcal{G}}_{\beta-1} (p_{k-2}+ \tilde{p}_{\beta-1},\boldsymbol{\tilde{p}},\boldsymbol{\eta}) & \text{$\beta$ odd}
	\end{cases}
	\\
	\mathcal{P}_\beta (p_{k-2},\boldsymbol{\tilde{p}},\boldsymbol{s}) &=
	\begin{cases}  
	e^{i  s_{k+\beta} \frac{1}{2}p_{k-2}^2} \mathcal{P}_{\beta-1} (\tilde{p}_{\beta-2},\boldsymbol{\tilde{p}},\boldsymbol{s}) & \text{$\beta$ even}
	\\
	e^{i  s_{k+\beta} \frac{1}{2}p_{k-2}^2} \mathcal{P}_{\beta-1} (p_{k-2}+ \tilde{p}_{\beta-1},\boldsymbol{\tilde{p}},\boldsymbol{s}) & \text{$\beta$ odd}.
	\end{cases}
	\end{aligned}
\end{equation*}
This is the expression for the function we will be working with for this case. 
To estimate these averages we can argue as in the proofs of Lemma~\ref{LE:Exp_ran_phases} and Lemma~\ref{expansion_aver_bound_Mainterm} and obtain that
\begin{equation*}
	\begin{aligned}
 	\mathbb{E}&\big[ \mathcal{T}(\beta,k,\sigma,\alpha,\tilde{\alpha};\hbar)\big] =   \frac{(2\pi)^{(2k-2-n)d}(\rho\hbar)^{2k-4-n}}{(2\pi\hbar)^d}  \int_{\R_{+}^{|\alpha|+|\tilde{\alpha}|}} f(\boldsymbol{s},\boldsymbol{\tilde{s}}) \int \Lambda_n(\boldsymbol{p},\boldsymbol{q},\sigma)    \hat{\varphi}(\tfrac{p_0}{\hbar}) \overline{\hat{\varphi}(\tfrac{q_0}{\hbar})} 
	\\
	\times & \delta(\tfrac{t}{\hbar}- \boldsymbol{s}_{0,k+\beta}^{+}- \boldsymbol{t}_{1,a_{k+\beta}}^{+})  \delta(\tfrac{t}{\hbar}- \boldsymbol{\tilde{s}}_{1,k+\beta}^{+}- \boldsymbol{\tilde{t}}_{1,a_{k+\beta}}^{+})  \mathcal{P}_\beta (p_{k-2},\boldsymbol{\tilde{p}},\boldsymbol{s})   \mathcal{P}_\beta (q_{k-2},\boldsymbol{\tilde{q}},\boldsymbol{\tilde{s}})  \prod_{i=1}^{a_{k+\beta}} e^{i  t_{i} \frac{1}{2} \eta_{i}^2} \prod_{i=1}^{\tilde{a}_{k+\beta}} e^{-i  \tilde{t}_{i} \frac{1}{2} \xi_{i}^2} 
	\\
	\times & e^{i  s_{k-2} \frac{1}{2} (\tilde{p}_{-2}+ \tilde{p}_{-1})^2}  e^{-i  \tilde{s}_{k-2} \frac{1}{2} (\tilde{q}_{-2}+ \tilde{q}_{-1})^2}     e^{i  s_{k-3} \frac{1}{2} (p_{k-3}+ \pi_{\beta}(\boldsymbol{\tilde{p}}))^2} e^{-i  \tilde{s}_{k-3} \frac{1}{2} (q_{k-3}+ \pi_{\beta}(\boldsymbol{\tilde{q}}))^2}
	\\
	\times &  
	\prod_{m=0}^{k_1-4}  e^{i  s_{m} \frac{1}{2} p_{m}^2} e^{-i  \tilde{s}_{m} \frac{1}{2} q_{m}^2}  \mathcal{G}(\boldsymbol{p}, \boldsymbol{\tilde{p}},\boldsymbol{\tilde{p}},\boldsymbol{\eta}, \boldsymbol{\xi},\sigma ) d\boldsymbol{\eta} d\boldsymbol{\xi}d\boldsymbol{p}d\boldsymbol{q} d\boldsymbol{\tilde{p}}  d\boldsymbol{\tilde{q}}  d\boldsymbol{t} d\boldsymbol{\tilde{t}}  d\boldsymbol{s} d\boldsymbol{\tilde{s}},
	\end{aligned}
\end{equation*}
where the function $f$ is given by
\begin{equation}\label{EQ:expansion_aver_bound_3.recol_22_3.1_3.1}
	f(\boldsymbol{s},\boldsymbol{\tilde{s}}) = 
	\begin{cases}
	\boldsymbol{1}_{[0,\hbar^{-1}t]}\big( \sum_{i=1}^{k-4} s_i + s_{k+\beta} \big) \boldsymbol{1}_{[0,\hbar^{-1}t]}\big( \sum_{i=1,i \notin \sigma}^{k-3} \tilde{s}_i +  \tilde{s}_{k+\beta}  \big) & \text{if $k-2\neq\sigma_n$}
	\\
	\boldsymbol{1}_{[0,\hbar^{-1}t]}\big( \sum_{i=1}^{k-4} s_i + s_{k+\beta} \big) \boldsymbol{1}_{[0,\hbar^{-1}t]}\big( \sum_{i=1,i \notin \sigma}^{k-3} \tilde{s}_i   \big) & \text{if $k-2=\sigma_n$}.
	\end{cases}
\end{equation}
The function $\mathcal{G}$ is given by
\begin{equation*}
	\begin{aligned}
	\MoveEqLeft \mathcal{G}(\boldsymbol{p}, \boldsymbol{\tilde{p}},\boldsymbol{\tilde{p}},\boldsymbol{\eta}, \boldsymbol{\xi},\sigma ) =  
	   \tilde{\mathcal{G}}_\beta (p_{k-2},\boldsymbol{\tilde{p}},\boldsymbol{\eta}) \overline{  \tilde{\mathcal{G}}_\beta (q_{k-2},\boldsymbol{\tilde{q}},\boldsymbol{\xi}) }  
	     \hat{\mathcal{V}}_{\alpha_{k-2}}(\tilde{p}_{-2}+ \tilde{p}_{-1},p_{k-3}+ \pi_{\beta}(\boldsymbol{\tilde{p}}),\boldsymbol{\eta} )
	     \\
	     \times& \overline{  \hat{\mathcal{V}}_{\tilde{\alpha}_{k-2}}(\tilde{q}_{-2}+ \tilde{q}_{-1},q_{k-3}+ \pi_{\beta}(\boldsymbol{\tilde{q}}),\boldsymbol{\xi} )}
	   \hat{\mathcal{V}}_{\alpha_{k-3}}(p_{k-3}+ \pi_{\beta}(\boldsymbol{\tilde{p}}),p_{k-4},\boldsymbol{\eta} )  \overline{ \hat{\mathcal{V}}_{\tilde{\alpha}_{k-3}}(q_{k-3}+ \pi_{\beta}(\boldsymbol{\tilde{q}}),q_{k-4},\boldsymbol{\xi} ) }
	   \\
	   \times&\prod_{m=1}^{k-4}    \hat{\mathcal{V}}_{\alpha_m}(p_m,p_{m-1},\boldsymbol{\eta} ) \overline{ \hat{\mathcal{V}}_{\tilde{\alpha}_m}(q_m,q_{m-1},\boldsymbol{\xi} ) }.
	\end{aligned}
\end{equation*}
  The function $\Lambda_n(\boldsymbol{p},\boldsymbol{q},\sigma)$ is given by
  \begin{equation*}
  	\Lambda_n(\boldsymbol{p},\boldsymbol{q},\sigma) = \delta(p_{k-2}-q_{k-2})  \prod_{i=1}^n \delta(p_{\sigma_{i-1}}- q_{\sigma{i-1}} - p_{\sigma_n}+ q_{\sigma_n} ))   
	\prod_{i=1}^{n+1} \prod_{m=\sigma_{i-1}+1}^{\sigma_{i}-1} \delta(p_m-p_{\sigma_{i-1}}) \delta(q_m-q_{\sigma^1_{i-1}}).
  \end{equation*}
We again use the function $\zeta$ as in the previous proofs and write the delta functions as Fourier transforms of the constant function $1$. 
Next we integrate in all $s$ and $\tilde{s}$ that the function $f$ does not depend on. Moreover, we also do the usual argument by dividing into different cases depending on the size of $t$ and $\tilde{t}$ and do integration by parts for all for all cases, where the variables is ``large''.  To estimate the bound we now have obtained we argue as in the previous proofs. This yields the bound
\begin{equation}\label{EQ:expansion_aver_bound_3.recol_22_3.1_4}
	\begin{aligned}
 	\MoveEqLeft  
	|\mathbb{E}\big[  \mathcal{T}(k,\sigma,\alpha,\tilde{\alpha},l;\hbar) \big] |
	 \leq   C_d^{|\alpha|+|\tilde{\alpha}|}   \norm{\hat{V}}_{1,\infty,4d+4}^{|\alpha|+|\tilde{\alpha}|}   \frac{\rho(\rho t)^{2k-n-5} \hbar |\log(\zeta)|^{n+19}}{(k-3)!(k-n-2)!} \norm{\varphi}_{\mathcal{H}^{4d+4}_\hbar(\R^d)}^2.
	\end{aligned}
\end{equation}
Combining \eqref{EQ:expansion_aver_bound_3.recol_22_3.1_2.1} and \eqref{EQ:expansion_aver_bound_3.recol_22_3.1_4} we obtain that
\begin{equation}\label{EQ:expansion_aver_bound_3.recol_22_3.1_5}
	\begin{aligned}
	\MoveEqLeft \mathbb{E}\Big[ \big\lVert \sum_{\boldsymbol{x}\in\mathcal{X}_{\neq}^{k-2}}\mathcal{I}^{\mathrm{dob}}_1(\beta,k;\hbar) \varphi \big\rVert_{L^2(\R^d)}^2\Big]
	 \leq 4\hbar  \norm{\varphi}_{\mathcal{H}^{4d+4}_\hbar(\R^d)}^2  \sum_{\alpha,\tilde{\alpha}\in\N^{k+\beta}} (C_d\lambda\norm{\hat{V}}_{1,\infty,4d+4})^{|\alpha|+|\tilde{\alpha}|}  
	 \\
	 &\times     \sum_{n=0}^{k-2} \binom{k-2}{n} n!  \sum_{\sigma\in\mathcal{A}(k-2,n)} \frac{ \rho (\rho t)^{2k-n-5} }{ (k-3)!(k-2-n)!}   |\log(\tfrac{\hbar}{t})|^{n+19}
	 \\
	 \leq{}& \tilde{\lambda}^\beta  \norm{\varphi}_{\mathcal{H}^{4d+4}_\hbar(\R^d)}^2  k C^k \hbar  |\log(\tfrac{\hbar}{t})|^{k+19},
	\end{aligned}
\end{equation}
where 
\begin{equation}\label{EQ_tilde_lambda}
	\tilde{\lambda} = \frac{C_d\lambda\norm{\hat{V}}_{1,\infty,4d+4}}{1-C_d\lambda\norm{\hat{V}}_{1,\infty,4d+4}}.
\end{equation}
That $\tilde{\lambda}$ is strictly less than one follows form our assumptions. Combining \eqref{EQ:expansion_aver_bound_3.recol_22_3.1_0} and \eqref{EQ:expansion_aver_bound_3.recol_22_3.1_5} we obtain the desired estimate. We now turn to the case where $\beta>6$. Here we will for notational convenience consider $\beta$ even and at the end say how to modify the argument for $\beta$ odd. Here we will use the following expression for the function $ \mathcal{I}^{\mathrm{dob}}_1(\beta,k_1,k_2,\boldsymbol{x};\hbar) \varphi (x)$:
\begin{equation}\label{EQ:expansion_aver_bound_3.recol_22_3.1_form}
	\begin{aligned}
 	\MoveEqLeft  \mathcal{I}^{\mathrm{dob}}_1(\beta,k,\boldsymbol{x};\hbar) \varphi (x)
	=  \frac{1}{(2\pi\hbar)^d} \sum_{\alpha\in \N^{k+\beta}} (i\lambda)^{|\alpha|} \int_{\R_{+}^{|\alpha|}} \int\boldsymbol{1}_{[0,\hbar^{-1}t]}( \boldsymbol{s}_{1,k+\beta}^{+}+ \boldsymbol{t}_{1,a_{k+\beta}}^{+})
	 e^{ i   \langle  \hbar^{-1}x,p_{k+\beta} \rangle} 
	 \\
	 &\times    
	e^{ -i   \langle  \hbar^{-1/d}(x_{k-2}-x_{k-3}),\sum_{m=0}^{\beta/2+1} p_{2m+k-2}-p_{2m+k-3} \rangle} e^{ -i   \langle  \hbar^{-1/d}x_{k-3},p_{\beta+k}-p_{k-4} \rangle}  \prod_{m=1}^{k-4}  e^{ -i   \langle  \hbar^{-1/d}x_{m},p_{m}-p_{m-1} \rangle}
	\\
	&\times 
	 \prod_{i=1}^{a_{k+\beta}} e^{i  t_{i} \frac{1}{2} \eta_{i}^2}       
	  \Big\{\prod_{m=1}^{k+\beta}  
	e^{i  s_{m} \frac{1}{2} p_{m}^2}   \hat{\mathcal{V}}_{\alpha_m}(p_m,p_{m-1},\boldsymbol{\eta} )  \Big\}e^{i ( \hbar^{-1}t- \boldsymbol{s}_{1,k+\beta}^{+}- \boldsymbol{t}_{1,a_{k+\beta}}^{+}) \frac{1}{2} p_0^2} \hat{\varphi}(\tfrac{p_0}{\hbar}) \, d\boldsymbol{\eta}_{1,a_k}d\boldsymbol{p}_{0,k}   d\boldsymbol{t}_{1,a_k} d\boldsymbol{s}_{1,k},
	\end{aligned}
\end{equation}
where we have also chosen a specific way of writing our phase function. To estimate the averages we will distinguish between three different cases. These are  $k-2\notin\sigma$, $k-3,k-2\in\sigma$ and finally the case $k-2\in\sigma$ and $k-3\notin\sigma$. We start with the first case $k-2\notin\sigma$. We will assume $k-3\in\sigma$ precisely we will assume $\sigma_{n}=k-3$. This is only done to simplify notation and the other case is estimated completely analogous. After taking the average we star by preforming the change of variables $x_{k-2}\mapsto x_{k-2}-x_{k-3}$ and   $\tilde{x}_{k-2}\mapsto \tilde{x}_{k-2}-\tilde{x}_{k-3}$. Then  by arguing as in the proof of Lemma~\ref{LE:Exp_ran_phases} and Lemma~\ref{expansion_aver_bound_Mainterm} and introducing the function $\zeta$ as in the previous proofs and write the delta functions as Fourier transforms of the constant function $1$ we have that    
\begin{equation}\label{EQ:expansion_aver_bound_3.recol_22_3.1_10}
	\begin{aligned}
 	&
	\Aver{ \mathcal{T}(\beta,k,\sigma,\alpha,\tilde{\alpha};\hbar)}
	=  \frac{(2\pi)^{(2k-2-n)d}(\rho\hbar)^{2k-4-n} e^{2\hbar^{-1}t\zeta}}{(2\pi)^2(2\pi\hbar)^d} \int_{\R_{+}^{|\alpha|+|\tilde{\alpha}|+2}} f(\boldsymbol{s},\boldsymbol{\tilde{s}}) \int e^{-i\hbar^{-1}t\tilde{\nu}} e^{-i\hbar^{-1}t\nu}  
	 \\
	 &\times   e^{ -i   \langle  x, p_{n}-\tilde{p}_{\beta-1} +\sum_{m=0}^{\beta/2} \tilde{p}_{2m-2}-\tilde{p}_{2m-3} \rangle} e^{ i   \langle  \tilde{x}, p_{n}-\tilde{q}_{\beta-1} +\sum_{m=0}^{\beta/2} \tilde{q}_{2m-2}-\tilde{q}_{2m-3} \rangle} \prod_{i=1}^{a_{k+\beta}} e^{i t_{i} ( \frac{1}{2}   \eta_{i}^2+\nu +i\zeta )}  \prod_{i=1}^{\tilde{a}_{k+\beta}} e^{- i  \tilde{t}_{i}( \frac{1}{2}   \xi_{i}^2-\tilde{\nu} - i\zeta )} 
	\\
	&\times  
	 e^{i  s_{k_1+\beta} (\frac{1}{2} p_{n}^2 +\nu+i\zeta)}   e^{-i  \tilde{s}_{k_1+\beta} (\frac{1}{2} p_{n}^2-\tilde{\nu}-i\zeta)} 
	 \Big\{\prod_{m=-3}^{\beta-1}  
	e^{i  s_{m+k}( \frac{1}{2} \tilde{p}_{m}^2+\nu +i\zeta)} e^{-i  \tilde{s}_{m+k}( \frac{1}{2} \tilde{q}_{m}^2-\tilde{\nu} - i\zeta )}   
	  \Big\}  
	  \\
	&  \times  \prod_{i=0}^{n-1} \prod_{m=\sigma_i}^{\sigma_{i+1}-1} e^{i  s_{m}(\frac{1}{2} p_{i}^2+\nu+i\zeta)} e^{-i\tilde{s}_{m}( \frac{1}{2} p_{i}^2-\tilde{\nu}-i\zeta)}
  	   \mathcal{G}(\boldsymbol{p}, \boldsymbol{\tilde{p}},\boldsymbol{\tilde{q}},\boldsymbol{\eta}, \boldsymbol{\xi},\sigma ) 
	 |\hat{\varphi}(\tfrac{p_0}{\hbar})|^2 \, d\boldsymbol{\eta}d\boldsymbol{\xi} dx d\tilde{x} d\boldsymbol{p} d\boldsymbol{\tilde{p}}d\boldsymbol{\tilde{q}} d\nu d\tilde{\nu}  d\boldsymbol{t}  d\boldsymbol{\tilde{t}}  d\boldsymbol{\tilde{s}}d\boldsymbol{s}.
	\end{aligned}
\end{equation}
where the function $f$ is given by
\begin{equation}\label{EQ:expansion_aver_bound_3.recol_22_3.1_10.1}
	f(\boldsymbol{s},\boldsymbol{\tilde{s}}) = 
	\boldsymbol{1}_{[0,\hbar^{-1}t]}\big( \sum_{i=1}^{k-4} s_i + s_{k+\beta} \big) \boldsymbol{1}_{[0,\hbar^{-1}t]}\big( \sum_{i=1,i \notin \sigma}^{k-4} \tilde{s}_i +  \tilde{s}_{k+\beta}  \big)
\end{equation}
and the function $\mathcal{G}$ is given by
\begin{equation}\label{EQ:expansion_aver_bound_3.recol_22_3.1_10.2}
	\begin{aligned}
	\mathcal{G}(\boldsymbol{p}, \boldsymbol{\tilde{p}},\boldsymbol{\tilde{q}},\boldsymbol{\eta}, \boldsymbol{\xi},\sigma )
	= {}&
	  \hat{\mathcal{V}}_{\alpha_{\beta+k_1}}(p_{n}, \tilde{p}_{\beta-1},\boldsymbol{\eta} )  
	 \overline{\hat{\mathcal{V}}_{\tilde{\alpha}_{\beta+k_1}}(p_{n}, \tilde{q}_{\beta-1},\boldsymbol{\xi} )} 
	\prod_{m=-2}^{\beta-1}  
	 \hat{\mathcal{V}}_{\alpha_{m+k_1}}(\tilde{p}_m,\tilde{p}_{m-1},\boldsymbol{\eta} ) \overline{ \hat{\mathcal{V}}_{\tilde{\alpha}_{m+k_1}}(\tilde{q}_m,\tilde{q}_{m-1},\boldsymbol{\xi} )} 
	  \\
	  &\times  \hat{\mathcal{V}}_{\alpha_{k_1-3}}(\tilde{p}_{-3},p_{n-1},\boldsymbol{\eta} )   
	 \overline{\hat{\mathcal{V}}_{\tilde{\alpha}_{k_1-3}}(\tilde{q}_{-3},p_{n-1},\boldsymbol{\eta} )  }
	  \prod_{m=1}^{\sigma_1-1}  
	  \hat{\mathcal{V}}_{\alpha_m}(p_0,p_{0},\boldsymbol{\eta} )  \overline{  \hat{\mathcal{V}}_{\tilde{\alpha}_m}(p_0,p_{0},\boldsymbol{\xi} )}
	  \\
	  &\times  \prod_{i=1}^{n_1-1} \hat{\mathcal{V}}_{\alpha_{\sigma_i}}(p_i,p_{i-1},\boldsymbol{\eta} )
	 \overline{ \hat{\mathcal{V}}_{\tilde{\alpha}_{\sigma_i}}(p_i,p_{i-1},\boldsymbol{\xi} )}   
	  \prod_{m=\sigma_i+1}^{\sigma_{i+1}-1}  
	   \hat{\mathcal{V}}_{\alpha_m}(p_i,p_{i},\boldsymbol{\eta} )\overline{ \hat{\mathcal{V}}_{\tilde{\alpha}_m}(p_i,p_{i},\boldsymbol{\xi} )}  .  
	\end{aligned}
\end{equation}
We now evaluate the integrals in $s_0$, $\tilde{s}_0$, $\tilde{s}_{\sigma_i}$ for $i\in\{1,\dots,n-1\}$  and $s_m$,$\tilde{s}_m$ for $m\in\{k-3,\dots,k+\beta-1\}$. From here we do the usual argument by dividing into different cases depending on the size of $t$ and $\tilde{t}$ and do integration by parts for all for all cases, where the varables is ``large''. For the variables $\tilde{p}_m$, $\tilde{q}_m$ with $m\in\{-2,\dots,\beta-1\}$ we use Lemma~\ref{LE:Fourier_trans_resolvent} and integration by parts. How we use Lemma~\ref{LE:Fourier_trans_resolvent} and integration by parts is captured in the following calculation 
\begin{equation}\label{EQ:expansion_aver_bound_3.recol_22_3.1_11}
	\begin{aligned}
 	\MoveEqLeft  
	  \Big|\int_{\R^{2\beta+4}} \mathcal{G}(\boldsymbol{p}, \boldsymbol{\tilde{p}},\boldsymbol{\tilde{q}},\boldsymbol{\eta}, \boldsymbol{\xi},\sigma ) 
	  \frac{ e^{ i   \langle  x,\tilde{p}_{-3} \rangle}}{\frac{1}{2} \tilde{p}_{-3}^2+\nu +i\zeta}\frac{ e^{ -i   \langle  \tilde{x},\tilde{q}_{-3} \rangle} }{\frac{1}{2} \tilde{q}_{-3}^2-\tilde{\nu} - i\zeta}
	   \prod_{m=-2}^{\beta-1}   \frac{  e^{ -i   \langle  x,(-1)^m\tilde{p}_m \rangle }}{\frac{1}{2} \tilde{p}_{m}^2+\nu +i\zeta} \frac{e^{ i   \langle  \tilde{x},(-1)^m\tilde{q}_m \rangle} }{\frac{1}{2} \tilde{q}_{m}^2-\tilde{\nu} - i\zeta}
	\, d\boldsymbol{\tilde{p}} d\boldsymbol{\tilde{q}} \Big|
	\\
	={}&\frac{1}{(2\pi)^{2\beta+4}} \Big| \int_{\R^{4\beta+8}} \mathcal{G}(\boldsymbol{p}, \boldsymbol{\tilde{p}},\boldsymbol{\tilde{q}},\boldsymbol{\eta}, \boldsymbol{\xi},\sigma )    \prod_{m=-2}^{\beta-1}   \frac{  e^{ -i   \langle  z_m ,\tilde{p}_m \rangle } e^{i|z_m|\sqrt{2(-\nu-i\zeta)} e^{- i   \langle \tilde{z}_m,\tilde{q}_m \rangle}e^{i|\tilde{z}_m|\sqrt{2(\tilde{\nu}+i\zeta)}} } }{|\tilde{z}_m-	(-1)^m\tilde{x}|  |z_m- (-1)^mx|} 
	\\
	&\times \frac{ e^{ i   \langle  x,\tilde{p}_{-3} \rangle}} {\frac{1}{2} \tilde{p}_{-3}^2+\nu +i\zeta}\frac{ e^{ -i   \langle  \tilde{x},\tilde{q}_{-3} \rangle} }{\frac{1}{2} \tilde{q}_{-3}^2-\tilde{\nu} - i\zeta}
	\, d\boldsymbol{z} d\boldsymbol{\tilde{z}} d\boldsymbol{\tilde{p}} d\boldsymbol{\tilde{q}} \Big|
	\\
	\leq&{}  \sum_{|\beta_{-2}|,\dots,|\beta_{k+\beta-1}|\leq d+1} \sum_{|\epsilon_{-2}|,\dots,|\epsilon_{k+\beta-1}|\leq d+1} \frac{ \prod_{i=-2}^{k+\beta-1} C(\beta_i)C(\epsilon_i) }{(2\pi)^{2\beta+4}}  \int_{\R^{4\beta+8}} \frac{ 1} {|\frac{1}{2} \tilde{p}_{-3}^2+\nu +i\zeta|}  \frac{ 1}{|\frac{1}{2} \tilde{q}_{-3}^2-\tilde{\nu} - i\zeta|}
	\\
	&\times   \prod_{m=-2}^{\beta-1}   \frac{  1}{ \langle \tilde{z}_m \rangle^{d+1} |\tilde{z}_m-	(-1)^m\tilde{x}| \langle z_m \rangle^{d+1}  |z_m- (-1)^mx|} 
	 \big| \partial_{\boldsymbol{\tilde{p}}}^{\boldsymbol{\beta}} \partial_{\boldsymbol{\tilde{q}}}^{\boldsymbol{\epsilon}} 
	\mathcal{G}(\boldsymbol{p}, \boldsymbol{\tilde{p}},\boldsymbol{\tilde{q}},\boldsymbol{\eta}, \boldsymbol{\xi},\sigma ) \big| \, d\boldsymbol{z} d\boldsymbol{\tilde{z}} d\boldsymbol{\tilde{p}} d\boldsymbol{\tilde{q}} .
	\end{aligned}
\end{equation}
Combining these arguments we get that
\begin{equation}\label{EQ:expansion_aver_bound_3.recol_22_3.1_12}
	\begin{aligned}
 	\MoveEqLeft  
	\big| \Aver{ \mathcal{T}(\beta,k,\sigma,\alpha,\tilde{\alpha};\hbar)}\big|
	 \leq  \frac{C^{|\alpha|+|\tilde{\alpha}|} (\rho\hbar)^{2k-4-n} }{(2\pi\hbar)^d}   \int_{\R_{+}^{2k-n-6}} f(\boldsymbol{s},\boldsymbol{\tilde{s}})  d\boldsymbol{s}_{1,k-3}   d\boldsymbol{\tilde{s}}_{1,k-n-3}    \sup_{\boldsymbol{\epsilon},\boldsymbol{\gamma},\boldsymbol{\delta},\boldsymbol{\varepsilon}}  \int  |\hat{\varphi}(\tfrac{p_0}{\hbar}) |^2
	  \\
	  &\times     \frac{ \langle p_0\rangle^{2(d+1)}}{\langle \tilde{\nu}\rangle |\frac{1}{2} p_0^2-\tilde{\nu}-i\zeta|}\frac{1}{\langle \nu\rangle |\frac{1}{2} p_0^2+\nu +i\zeta|}
	  \frac{ \langle\nu\rangle \langle \tilde{p}_{-3}\rangle^{-d-1} } {|\frac{1}{2} \tilde{p}_{-3}^2+\nu +i\zeta|}\frac{\langle\tilde{\nu}\rangle \langle \tilde{q}_{-3}\rangle^{-d-1} }{|\frac{1}{2} \tilde{q}_{-3}^2-\tilde{\nu} - i\zeta|}  \prod_{m=1}^{n-1} \frac{\langle p_m-p_{m-1} \rangle^{-d-1}}{ |\frac{1}{2} p_{m}^2+\nu+i\zeta|} 
	  \\
	  &\times
	   \Big| \partial_{\boldsymbol{\eta}}^{\boldsymbol{\epsilon}} \partial_{\boldsymbol{\xi}}^{\boldsymbol{\gamma}}  \partial_{\boldsymbol{\tilde{p}}}^{\boldsymbol{\delta}} \partial_{\boldsymbol{\tilde{q}}}^{\boldsymbol{\varepsilon}}\mathcal{G}(\boldsymbol{p}, \boldsymbol{\tilde{p}},\boldsymbol{\tilde{q}},\boldsymbol{\eta}, \boldsymbol{\xi},\sigma ) \Big|
	    \prod_{m=-2}^{\beta-1}   \frac{  1}{ \langle \tilde{z}_m \rangle^{d+1} |\tilde{z}_m-	(-1)^m\tilde{x}| \langle z_m \rangle^{d+1}  |z_m- (-1)^mx|} 
	 \\
	&\times 
	 \langle \tilde{q}_{-3} -p_{n-1} \rangle^{d+1} \langle \tilde{p}_{-3} -p_{n-1} \rangle^{d+1} \prod_{\substack{m=1\\m\neq n_1} }^{n} \langle p_m-p_{m-1}\rangle^{3(d+1)}
	\, dx d\tilde{x}d\boldsymbol{z} d\boldsymbol{\tilde{z}} d\boldsymbol{\tilde{p}} d\boldsymbol{\tilde{q}} d\boldsymbol{p}d\boldsymbol{\eta}d\boldsymbol{\xi},
	\end{aligned}
\end{equation}
where $C$ is a constant only depending on the dimension and we have used the notation
\begin{equation}\label{EQ:expansion_aver_bound_3.recol_22_3.1_sup_notation}
	 \sup_{\boldsymbol{\epsilon},\boldsymbol{\gamma},\boldsymbol{\delta},\boldsymbol{\varepsilon}} \coloneqq \sup_{|\epsilon_1|,\dots,|\epsilon_{a_{k+\beta}}|\leq d+1} \sup_{|\gamma_1|,\dots,|\gamma_{a_{k+\beta}}|\leq d+1} \sup_{|\delta_1|,\dots,|\delta_{\beta+2}|\leq d+1}  \sup_{|\varepsilon_1|,\dots,|\varepsilon_{\beta+2}|\leq d+1}.
\end{equation}
Moreover, we have used the inequality
\begin{equation*}
	\begin{aligned}
	1 &= \frac{\langle \tilde{\nu} \rangle \langle \nu \rangle \langle \tilde{p}_{0}\rangle^{2d+2} \langle \tilde{q}_{-3} -p_{n-1} \rangle^{d+1} \langle \tilde{p}_{-3} -p_{n-1} \rangle^{d+1} \prod_{m=1}^{n-1} \langle p_m-p_{m-1}\rangle^{3(d+1)}  }{\langle \tilde{\nu} \rangle \langle \nu \rangle \langle \tilde{p}_{0}\rangle^{2d+2} \langle \tilde{q}_{-3} -p_{n-1} \rangle^{d+1} \langle \tilde{p}_{-3} -p_{n-1} \rangle^{d+1} \prod_{m=1}^{n} \langle p_m-p_{m-1}\rangle^{3(d+1)}}
	\\
	&\leq C \frac{\langle \tilde{\nu} \rangle \langle \nu \rangle\langle \tilde{p}_{0}\rangle^{2d+2} \langle \tilde{q}_{-3} -p_{n-1} \rangle^{d+1} \langle \tilde{p}_{-3} -p_{n-1} \rangle^{d+1} \prod_{m=1}^{n-1} \langle p_m-p_{m-1}\rangle^{3(d+1)}  }{\langle \tilde{\nu} \rangle \langle \nu \rangle \langle \tilde{q}_{-3} \rangle^{d+1} \langle \tilde{p}_{-3}\rangle^{d+1} \prod_{m=1}^{n-1} \langle p_m-p_{m-1}\rangle^{d+1}}.
	\end{aligned}
\end{equation*}
To estimate this we note that
  \begin{equation}\label{EQ:expansion_aver_bound_3.recol_22_3.1_13}
	\begin{aligned}
	\MoveEqLeft  \sup_{\boldsymbol{p}}   \int \frac{ \Big| \partial_{\boldsymbol{\eta}}^{\boldsymbol{\epsilon}} \partial_{\boldsymbol{\xi}}^{\boldsymbol{\gamma}}  \partial_{\boldsymbol{\tilde{p}}}^{\boldsymbol{\delta}} \partial_{\boldsymbol{\tilde{q}}}^{\boldsymbol{\varepsilon}}\mathcal{G}(\boldsymbol{p}, \boldsymbol{\tilde{p}},\boldsymbol{\tilde{q}},\boldsymbol{\eta}, \boldsymbol{\xi},\sigma ) \Big|}{
	( \langle \tilde{q}_{-3} -p_{n-1} \rangle \langle \tilde{p}_{-3} -p_{n-1} \rangle)^{-d-1}} \prod_{m=1}^{n-1} \langle p_m-p_{m-1}\rangle^{3(d+1)}
	\, d\boldsymbol{\tilde{p}} d\boldsymbol{\tilde{q}} d\boldsymbol{\eta}d\boldsymbol{\xi} dp_n
	 \leq  (C  \norm{\hat{V}}_{1,\infty,3d+3})^{|\alpha|+|\tilde{\alpha}|}. 
	\end{aligned}
\end{equation}
Moreover, we have that
\begin{equation}\label{EQ:expansion_aver_bound_3.recol_22_3.1_14}
	\int_{\R_{+}^{2k-n-6}} f(\boldsymbol{s},\boldsymbol{\tilde{s}})  d\boldsymbol{s}_{1,k-3}   d\boldsymbol{\tilde{s}}_{1,k-n-3}  = \frac{t^{2k-n-6}}{\hbar^{2k-n-6}(k-3)!(k-3-n)!}.
\end{equation}
By applying Lemma~\ref{LE:int_posistion} and Lemma~\ref{LE:resolvent_int_est} we get that
\begin{equation}\label{EQ:expansion_aver_bound_3.recol_22_3.1_15}
	\begin{aligned}
 	\MoveEqLeft    \int  |\hat{\varphi}(\tfrac{p_0}{\hbar}) |^2   \frac{ \langle p_0\rangle^{2(d+1)}}{\langle \tilde{\nu}\rangle |\frac{1}{2} p_0^2-\tilde{\nu}-i\zeta|}\frac{1}{\langle \nu\rangle |\frac{1}{2} p_0^2+\nu +i\zeta|}
	  \frac{ \langle\nu\rangle \langle \tilde{p}_{-3}\rangle^{-d-1} } {|\frac{1}{2} \tilde{p}_{-3}^2+\nu +i\zeta|}\frac{\langle\tilde{\nu}\rangle \langle \tilde{q}_{-3}\rangle^{-d-1} }{|\frac{1}{2} \tilde{q}_{-3}^2-\tilde{\nu} - i\zeta|} \prod_{m=1}^{n-1} \frac{\langle p_m-p_{m-1} \rangle^{-d-1}}{ |\frac{1}{2} p_{m}^2+\nu+i\zeta|} 
	  \\
	  &\times 
	    \prod_{m=-2}^{\beta-1}   \frac{  1}{ \langle \tilde{z}_m \rangle^{d+1} |\tilde{z}_m-	(-1)^m\tilde{x}| \langle z_m \rangle^{d+1}  |z_m- (-1)^mx|} 
	\, dx d\tilde{x}d\boldsymbol{z} d\boldsymbol{\tilde{z}}  d\boldsymbol{p} d\tilde{p}_{-3} d\tilde{q}_{-3}
	\\
	&\leq   (2\pi\hbar)^dC |\log(\tfrac{\hbar}{t})|^{n+3}  \norm{\varphi}_{\mathcal{H}^{d+1}_\hbar(\R^d)}^2 \int \frac{1}{\langle x \rangle^{\frac{\beta+2}{2}}\langle \tilde{x} \rangle^{\frac{\beta+2}{2}}}  \, dx d\tilde{x}
	\\
	&\leq (2\pi\hbar)^d C |\log(\tfrac{\hbar}{t})|^{n+3}  \norm{\varphi}_{\mathcal{H}^{d+1}_\hbar(\R^d)}^2, 
	\end{aligned}
\end{equation}
where we have used that $\beta>6$. Combining the estimates in  \cref{EQ:expansion_aver_bound_3.recol_22_3.1_12,EQ:expansion_aver_bound_3.recol_22_3.1_13,EQ:expansion_aver_bound_3.recol_22_3.1_14,EQ:expansion_aver_bound_3.recol_22_3.1_15}  we obtain that
\begin{equation}\label{EQ:expansion_aver_bound_3.recol_22_3.1_16}
	\begin{aligned}
 	\MoveEqLeft  
	\big| \Aver{ \mathcal{T}(\beta,k,\sigma,\alpha,\tilde{\alpha};\hbar)}\big|
	 \leq \hbar^2    \frac{C^{|\alpha|+|\tilde{\alpha}|} \rho^2 (\rho t)^{2k-n-6} }{ (k-3)!(k-3-n)!}  
	  |\log(\tfrac{\hbar}{t})|^{n+3}  \norm{\varphi}_{\mathcal{H}^{d+1}_\hbar(\R^d)}^2 \norm{\hat{V}}_{1,\infty,3d+3}^{|\alpha|+|\tilde{\alpha}|}.
	\end{aligned}
\end{equation}
Note that for this case we have assumed that $k-2\notin \sigma$, hence $n\leq k-3$ implying that we are not taking the factorial of a negative number.

We now turn to the case where $k-2\in\sigma$ and $k-3\notin\sigma$. For this case we preform the change of variables $x_{k-3}\mapsto x_{k-3}-x_{k-2}$ and   $\tilde{x}_{k-3}\mapsto \tilde{x}_{k-3}-\tilde{x}_{k-2}$ after taking the average. Other than that one obtains with an analogous argument to the one above the estimate
\begin{equation}\label{EQ:expansion_aver_bound_3.recol_22_3.1_17}
	\begin{aligned}
 	\MoveEqLeft  
	\big| \Aver{ \mathcal{T}(\beta,k,\sigma,\alpha,\tilde{\alpha};\hbar)}\big|
	 \leq \hbar^2    \frac{C^{|\alpha|+|\tilde{\alpha}|} \rho^2 (\rho t)^{2k-n-6} }{ (k-3)!(k-3-n)!}  
	  |\log(\tfrac{\hbar}{t})|^{n+3}  \norm{\varphi}_{\mathcal{H}^{d+1}_\hbar(\R^d)}^2 \norm{\hat{V}}_{1,\infty,3d+3}^{|\alpha|+|\tilde{\alpha}|}.
	\end{aligned}
\end{equation}
Moreover, we again in this case have that $n\leq k-3$. 
We now turn to the case where $k-2,k-3\in\sigma$. For this case we also preform the change of variables $x_{k-2}\mapsto x_{k-2}-x_{k-3}$ and   $\tilde{x}_{k-2}\mapsto \tilde{x}_{k-2}-\tilde{x}_{k-3}$ after taking the average. However this time when we argue as in the proofs of  Lemma~\ref{LE:Exp_ran_phases} and Lemma~\ref{expansion_aver_bound_Mainterm}  we obtain that 
\begin{equation*}
	\begin{aligned}
 	&
	\Aver{ \mathcal{T}(\beta,k,\sigma,\alpha,\tilde{\alpha};\hbar)}
	=  \frac{(2\pi)^{(2k-2-n)d}(\rho\hbar)^{2k-4-n}}{(2\pi\hbar)^d} \int_{\R_{+}^{|\alpha|+|\tilde{\alpha}|+2}} f(\boldsymbol{s},\boldsymbol{\tilde{s}}) \int  |\hat{\varphi}(\tfrac{p_0}{\hbar})|^2  \delta(\tfrac{t}{\hbar}- \boldsymbol{s}_{0,k+\beta}^{+}- \boldsymbol{t}_{1,a_{k+\beta}}^{+})
	 \\
	& \times   \delta(\tfrac{t}{\hbar}- \boldsymbol{\tilde{s}}_{1,k+\beta}^{+}- \boldsymbol{\tilde{t}}_{1,a_{k+\beta}}^{+}) e^{ -i   \langle  x, \tilde{q}_{\beta-1}-\tilde{p}_{\beta-1} +\sum_{m=0}^{\beta/2} \tilde{p}_{2m-2}-\tilde{p}_{2m-3} - \tilde{q}_{2m-2}+\tilde{q}_{2m-3} \rangle} \prod_{i=1}^{a_{k+\beta}} e^{i  t_{i} \frac{1}{2} \eta_{i}^2} \prod_{i=1}^{\tilde{a}_{k+\beta}}  e^{-i  \tilde{t}_{i} \frac{1}{2} \xi_{i}^2}  
	\\
	&\times 
	     e^{i  s_{k-3} \frac{1}{2} \tilde{p}_{-3}^2}  e^{-i  \tilde{s}_{k-3} \frac{1}{2} \tilde{q}_{-3}^2}    
	 e^{i  (s_{k+\beta}-\tilde{s}_{k+\beta}) \frac{1}{2} p_{n-1}^2}  
	  \\
	  &\times \mathcal{G}(\boldsymbol{p}, \boldsymbol{\tilde{p}},\boldsymbol{\tilde{q}},\boldsymbol{\eta}, \boldsymbol{\xi},\sigma )   \prod_{m=-2}^{\beta-1}  
	e^{i  s_{m+k} \frac{1}{2} \tilde{p}_{m}^2} e^{-i  \tilde{s}_{m+k} \frac{1}{2} \tilde{q}_{m}^2}     \prod_{i=0}^{n-2} \prod_{m=\sigma_i}^{\sigma_{i+1}-1} e^{i  (s_{m}-\tilde{s}_{m}) \frac{1}{2} p_{i}^2}
	  \, d\boldsymbol{\eta} d\boldsymbol{\xi} dx d\tilde{x} d\boldsymbol{p} d\boldsymbol{\tilde{p}} d\boldsymbol{\tilde{q}}   d\boldsymbol{t}  d\boldsymbol{\tilde{t}}  d\boldsymbol{\tilde{s}}d\boldsymbol{s},
	\end{aligned}
\end{equation*}
where the function $f$ still is given by \eqref{EQ:expansion_aver_bound_3.recol_22_3.1_3.1}
and the function $\mathcal{G}$ is given by
\begin{equation*}
	\begin{aligned}
	\mathcal{G}&(\boldsymbol{p}, \boldsymbol{\tilde{p}},\boldsymbol{\tilde{q}},\boldsymbol{\eta}, \boldsymbol{\xi},\sigma )
	= 
	  \hat{\mathcal{V}}_{\alpha_{\beta+k}}(p_{n-1}, \tilde{p}_{\beta-1},\boldsymbol{\eta} )  
	 \overline{\hat{\mathcal{V}}_{\tilde{\alpha}_{\beta+k}}(p_{n-1}, \tilde{q}_{\beta-1},\boldsymbol{\xi} )} 
	  \hat{\mathcal{V}}_{\alpha_{k_1-3}}(\tilde{p}_{-3},p_{n_1-1},\boldsymbol{\eta} )   
	  \overline{\hat{\mathcal{V}}_{\tilde{\alpha}_{k_1-3}}(\tilde{q}_{-3},p_{n_1-1},\boldsymbol{\eta} )  }
	  \\
	  &\times \prod_{m=-2}^{\beta-1}  
	 \hat{\mathcal{V}}_{\alpha_{m+k_1}}(\tilde{p}_m,\tilde{p}_{m-1},\boldsymbol{\eta} ) \overline{ \hat{\mathcal{V}}_{\tilde{\alpha}_{m+k_1}}(\tilde{q}_m,\tilde{q}_{m-1},\boldsymbol{\xi} )} 
	     \prod_{m=1}^{\sigma_1-1}  
	  \hat{\mathcal{V}}_{\alpha_m}(p_0,p_{0},\boldsymbol{\eta} )  \overline{  \hat{\mathcal{V}}_{\tilde{\alpha}_m}(p_0,p_{0},\boldsymbol{\xi} )}
	  \\
	  &\times 
	   \prod_{i=1}^{n-2} \hat{\mathcal{V}}_{\alpha_{\sigma_i}}(p_i,p_{i-1},\boldsymbol{\eta} )
	 \overline{ \hat{\mathcal{V}}_{\tilde{\alpha}_{\sigma_i}}(p_i,p_{i-1},\boldsymbol{\xi} )}   
	  \prod_{m=\sigma_i+1}^{\sigma_{i+1}-1}  
	   \hat{\mathcal{V}}_{\alpha_m}(p_i,p_{i},\boldsymbol{\eta} )\overline{ \hat{\mathcal{V}}_{\tilde{\alpha}_m}(p_i,p_{i},\boldsymbol{\xi} )}.  
	\end{aligned}
\end{equation*}
We observe here that effectively we now only have $n-1$ collisions and not $n$. This illustrates the intuition that in this case we treat the collisions in the position $x_{k-2}$ and $x_{k-3}$ as a single collision.  By arguing as above we obtain that if $2\leq n\leq k-3$ the estimate
\begin{equation}\label{EQ:expansion_aver_bound_3.recol_22_3.1_18}
	\begin{aligned}
 	\MoveEqLeft  
	\big| \Aver{ \mathcal{T}(\beta,k,\sigma,\alpha,\tilde{\alpha};\hbar)}\big|
	 \leq \hbar^2    \frac{C^{|\alpha|+|\tilde{\alpha}|} \rho^2 (\rho t)^{2k-n-6} }{ (k-3)!(k-3-n)!}  
	  |\log(\tfrac{\hbar}{t})|^{n+3}  \norm{\varphi}_{\mathcal{H}^{d+1}_\hbar(\R^d)}^2 \norm{\hat{V}}_{1,\infty,3d+3}^{|\alpha|+|\tilde{\alpha}|}.
	\end{aligned}
\end{equation}
Note that this case is empty if $k=4$. Here we can only have that $n=2$ if  $k-2,k-3\in\sigma$. This is covered in the next case. In the case $n=k-2$ we get the estimate
\begin{equation}\label{EQ:expansion_aver_bound_3.recol_22_3.1_19}
	\begin{aligned}
 	\MoveEqLeft  
	\big| \Aver{ \mathcal{T}(\beta,k,\sigma,\alpha,\tilde{\alpha};\hbar)}\big|
	 \leq \hbar    \frac{C^{|\alpha|+|\tilde{\alpha}|} \rho (\rho t)^{k-3} }{ (k-3)!}  
	  |\log(\tfrac{\hbar}{t})|^{n+3}  \norm{\varphi}_{\mathcal{H}^{d+1}_\hbar(\R^d)}^2 \norm{\hat{V}}_{1,\infty,3d+3}^{|\alpha|+|\tilde{\alpha}|}.
	\end{aligned}
\end{equation}
That this case is different is the reason we have to do the full double Born series expansion. Combining the estimates in  \cref{EQ:expansion_aver_bound_3.recol_22_3.1_2.1,EQ:expansion_aver_bound_3.recol_22_3.1_16,EQ:expansion_aver_bound_3.recol_22_3.1_17,EQ:expansion_aver_bound_3.recol_22_3.1_18,EQ:expansion_aver_bound_3.recol_22_3.1_19} we obtain that
\begin{equation}\label{EQ:expansion_aver_bound_3.recol_22_3.1_20}
	\begin{aligned}
	\MoveEqLeft \mathbb{E}\Big[ \big\lVert \sum_{\boldsymbol{x}\in\mathcal{X}_{\neq}^{k-2}}\mathcal{I}^{\mathrm{dob}}_1(\beta,k;\hbar) \varphi \big\rVert_{L^2(\R^d)}^2\Big]
	 \leq 2\hbar  \norm{\varphi}_{\mathcal{H}^{d+1}_\hbar(\R^d)}^2  \sum_{\alpha,\tilde{\alpha}\in\N^{k+\beta}} (C_d\lambda\norm{\hat{V}}_{1,\infty,3d+3})^{|\alpha|+|\tilde{\alpha}|}  
	 \\
	 &\times \Big[  k   \rho (\rho t)^{k-3}  |\log(\tfrac{\hbar}{t})|^{k+1}+ 2 \sum_{n=0}^{k-3} \binom{k-2}{n} n!  \sum_{\sigma\in\mathcal{A}(k-2,n)} \frac{ \rho^2 (\rho t)^{2k-n-6} }{ (k-3)!(k-3-n)!}   |\log(\tfrac{\hbar}{t})|^{n+3}\Big]
	 \\
	 \leq{}& \tilde{\lambda}^\beta  \norm{\varphi}_{\mathcal{H}^{d+1}_\hbar(\R^d)}^2 t^{k-3} k C^k \hbar  |\log(\tfrac{\hbar}{t})|^{k+1},
	\end{aligned}
\end{equation}
where  $\tilde{\lambda}$ is given by \eqref{EQ_tilde_lambda}. Combining \eqref{EQ:expansion_aver_bound_3.recol_22_3.1_2.1} and \eqref{EQ:expansion_aver_bound_3.recol_22_3.1_20} we obtain a bound better than stated in the Lemma. We have chosen to state the Lemma with a bound that covers both small and large values of $\beta$. For $\beta$ odd, the difference is that instead of \eqref{EQ:expansion_aver_bound_3.recol_22_3.1_form} we now have that
\begin{equation*}
	\begin{aligned}
 	\MoveEqLeft  \mathcal{I}^{\mathrm{dob}}_1(\beta,k,\boldsymbol{x};\hbar) \varphi (x)
	=  \frac{1}{(2\pi\hbar)^d} \sum_{\alpha\in \N^{k+\beta}} (i\lambda)^{|\alpha|} \int_{\R_{+}^{|\alpha|}} \int\boldsymbol{1}_{[0,\hbar^{-1}t]}( \boldsymbol{s}_{1,k+\beta}^{+}+ \boldsymbol{t}_{1,a_{k+\beta}}^{+})
	 e^{ i   \langle  \hbar^{-1}x,p_{k+\beta} \rangle} 
	 \\
	 &\times    
	e^{ -i   \langle  \hbar^{-1/d}(x_{k-2}-x_{k-3}),\sum_{m=0}^{(\beta+1)/2} p_{2m+k-2}-p_{2m+k-3} \rangle} e^{ -i   \langle  \hbar^{-1/d}x_{k-3},p_{\beta+k}-p_{k-4} \rangle}  
	\prod_{m=1}^{k-4}  e^{ -i   \langle  \hbar^{-1/d}x_{m},p_{m}-p_{m-1} \rangle}      
	\\
	&\times  \prod_{i=1}^{a_{k+\beta}} e^{i  t_{i} \frac{1}{2} \eta_{i}^2}  \prod_{m=1}^{k+\beta}  \Big\{
	e^{i  s_{m} \frac{1}{2} p_{m}^2}   \hat{\mathcal{V}}_{\alpha_m}(p_m,p_{m-1},\boldsymbol{\eta} )  \Big\}e^{i ( \hbar^{-1}t- \boldsymbol{s}_{1,k+\beta}^{+}- \boldsymbol{t}_{1,a_{k+\beta}}^{+}) \frac{1}{2} p_0^2} \hat{\varphi}(\tfrac{p_0}{\hbar}) \, d\boldsymbol{\eta}d\boldsymbol{p}  d\boldsymbol{t} d\boldsymbol{s},
	\end{aligned}
\end{equation*}
Form here the proofs are analogous and we obtain the same bound as in \eqref{EQ:expansion_aver_bound_3.recol_22_3.1_20}.
\end{proof}
 \begin{lemma}\label{expansion_aver_bound_Mainterm_rec_trun22_special_case_2}
Assume we are in the setting of Lemma~\ref{expansion_aver_bound_Mainterm_rec_trun22_5}. Then for any $\beta\in\N$ we have that
\begin{equation*}
	\begin{aligned}
	 \mathbb{E}\Big[\big\lVert  \sum_{k=5}^{k_0}   \mathcal{I}^{\mathrm{dob}}_2(\beta,k;\hbar)  \varphi\big\rVert_{L^2(\R^d)}^2 \Big]
	\leq  \hbar \tilde{\lambda}^\beta  \norm{\varphi}_{\mathcal{H}^{5d+5}_\hbar(\R^d)}^2 k_0^{10} C^{k_0} \hbar  |\log(\tfrac{\hbar}{t})|^{k_0+19},
	\end{aligned}
\end{equation*}
where $\tilde{\lambda}<1$.  
\end{lemma}
\begin{proof}
The proof is a combination of the proof of Lemma~\ref{expansion_aver_bound_Mainterm_rec_trun22_special_case_1} and the first part of the proof of Lemma~\ref{expansion_aver_bound_Mainterm_rec_trun22_5}. Firstly we note that
\begin{equation}\label{EQ:expansion_aver_bound_3.recol_22_3.2_0}
\big\lVert  \sum_{k=5}^{k_0} \mathcal{I}^{\mathrm{dob}}_2(\beta,k;\hbar) \varphi \big\rVert_{L^2(\R^d)}^2 \leq \frac{k_0^4 t \lambda^2}{\hbar^2} \sum_{k=5}^{k_0} \sum_{l=1}^{k-4}  \int_{0}^t  \big\lVert \tilde{\mathcal{I}}^{\mathrm{dob}}_2(\beta,k,l,s_{k+\beta+1};\hbar) \varphi \big\rVert_{L^2(\R^d)}^2 \, ds_{k+\beta+1},
\end{equation}
where 
\begin{equation*}
	\begin{aligned}	
	 \tilde{\mathcal{I}}^{\mathrm{dob}}_2(\beta,k,l,s_{k+\beta+1};\hbar)
	= {} &   \sum_{\boldsymbol{x}\in\mathcal{X}_{\neq}^{k-2}}  \sum_{\alpha^1\in\N^{k_1}}   \int_{[0,t]^{\beta+k}} \boldsymbol{1}_{[s_{k+\beta+1},t]}(s_{k+\beta})
	 V^{s_{k+\beta+1}}_{\hbar, x_{l}}  
	  \tilde{\Theta}_{\beta}^{\mathrm{dob}}(\boldsymbol{s}_{\beta+k,k};V,\hbar)
	  \\
	  &\times  \prod_{m=1}^{k}\Theta_{\alpha_m^1}(s_{m-1},{s}_{m},x_{\iota(m)};V,\hbar)   U_{\hbar,0}(-t)   \, d\boldsymbol{s}_{k,1}.
	 \end{aligned}
\end{equation*}
Applying Lemma~\ref{LE:crossing_dom_ladder} we obtain that 
\begin{equation}\label{EQ:expansion_aver_bound_3.recol_22_3.2_1}
	\begin{aligned}
	\MoveEqLeft \mathbb{E}\Big[ \int_{0}^t  \big\lVert \tilde{\mathcal{I}}^{\mathrm{dob}}_2(\beta,k,l,s_{k+\beta+1};\hbar) \varphi \big\rVert_{L^2(\R^d)}^2 \, ds_{k+\beta+1}\Big]
	\\
	 &\leq 2 \sum_{n=0}^{k-2} \binom{k-2}{n} n! \sum_{\alpha,\tilde{\alpha}\in\N^{k+\beta}} \lambda^{|\alpha|+|\tilde{\alpha}|} \sum_{\sigma\in\mathcal{A}(k-2,n)} \big|\Aver{ \mathcal{T}(\beta,k,\sigma,\alpha,\tilde{\alpha},l;\hbar)}\big| ,
	\end{aligned}
\end{equation}
where
\begin{equation*}
	\begin{aligned}
	\mathcal{T}(\beta,k,\sigma,\alpha,\tilde{\alpha},l;\hbar)
	= {}&
	 \sum_{(\boldsymbol{x},\boldsymbol{\tilde{x}})\in \mathcal{X}_{\neq}^{2k-4}} \prod_{i=1}^n \frac{\delta(x_{\sigma_i}- \tilde{x}_{\sigma_{i}})}{\rho}  
	 \\
	 &\times \int_{0}^t  \int_{\R^d} \mathcal{I}^{\mathrm{dob}}_2(\beta,k,s_{k+\beta+1};\alpha,\boldsymbol{x},\hbar) \varphi (x) \overline{ \mathcal{I}^{\mathrm{dob}}_2(\beta,k,s_{k+\beta+1};\tilde{\alpha},\boldsymbol{\tilde{x}},\hbar) \varphi (x)}  \,dx ds_{k+\beta+1}. 
	\end{aligned}
\end{equation*}
Before we proceed we will divide into the two cases $\beta\leq 6$ and $\beta>6$. We will start with the first case. From the definition of the operator we have that
\begin{equation*}
	\begin{aligned}
 	\MoveEqLeft   \mathcal{I}^{\mathrm{dob}}_2(\beta,k,s_{k+\beta+1};\alpha,\boldsymbol{x},\hbar) \varphi (x)
	=  \frac{1}{(2\pi\hbar)^d} \sum_{\alpha\in \N^{k+\beta}} (i\lambda)^{|\alpha|} \int_{\R_{+}^{|\alpha|}}  \boldsymbol{1}_{[0,\hbar^{-1}t]}( \boldsymbol{s}_{1,k+\beta}^{+}+\hbar^{-1}s_{k+\beta+1}+ \boldsymbol{t}_{1,a_{k+\beta}}^{+})
	 \\
	 \times& \int  e^{ i   \langle  \hbar^{-1}x,p_{k+\beta} \rangle}  
	   e^{ -i   \langle  \hbar^{-1/d}x_{l},p_{k+\beta+1}-p_{k+\beta} \rangle}      
	e^{i \hbar^{-1} s_{k+\beta+1} \frac{1}{2} p_{k+\beta+1}^2}   \hat{V}(p_{k+\beta+1}-p_{k+\beta}) 
	 \prod_{i=1}^{a_{k+\beta}} e^{i  t_{i} \frac{1}{2} \eta_{i}^2}   
	 \\
	 \times&    
	 \prod_{m=-3}^\beta e^{ -i   \langle  \hbar^{-1/d}x_{\tilde{\iota}(k)}, p_{k+m}-p_{k+m-1} \rangle}
	\prod_{m=1}^{k-4}  e^{ -i   \langle  \hbar^{-1/d}x_{m},p_{m}-p_{m-1} \rangle}      
	 \prod_{m=1}^{k+\beta}  
	e^{i  s_{m} \frac{1}{2} p_{m}^2}   \hat{\mathcal{V}}_{\alpha_m}(p_m,p_{m-1},\boldsymbol{\eta} ) 
	\\
	\times & e^{i ( \hbar^{-1}t- \boldsymbol{s}_{1,k+\beta}^{+}- \boldsymbol{t}_{1,a_{k+\beta}}^{+}) \frac{1}{2} p_0^2} \hat{\varphi}(\tfrac{p_0}{\hbar}) \, d\boldsymbol{\eta}d\boldsymbol{p}   d\boldsymbol{t} d\boldsymbol{s}.
	\end{aligned}
\end{equation*}
As in the proof of Lemma~\ref{expansion_aver_bound_Mainterm_rec_trun22_special_case_1} we preform the change of variables 
\begin{equation*}
p_{k+m}\mapsto p_{k+m}-p_{k+m+1} \quad\text{and}\quad  p_{k-3}\mapsto p_{k-3}-p_{k+m}\quad  \text{for all }
\begin{cases}
	m\in\{-2,0\} &\text{if $\beta\in\{1,2\}$}
	\\
	m\in\{-2,0,2\} &\text{if $\beta\in\{3,4\}$}
	\\
	m\in\{-2,0,2,4\} &\text{if $\beta\in\{5,6\}$}.
\end{cases}
\end{equation*}
After this change of variables and a relabelling we obtain that
\begin{equation*}
	\begin{aligned}
 	\MoveEqLeft   \mathcal{I}^{\mathrm{dob}}_2(\beta,k,s_{k+\beta+1};\alpha,\boldsymbol{x},\hbar) \varphi (x)
	=  \frac{1}{(2\pi\hbar)^d} \sum_{\alpha\in \N^{k+\beta}} (i\lambda)^{|\alpha|} \int_{\R_{+}^{|\alpha|}} \boldsymbol{1}_{[0,\hbar^{-1}t]}( \boldsymbol{s}_{1,k+\beta}^{+}+\hbar^{-1}s_{k+\beta+1}+ \boldsymbol{t}_{1,a_{k+\beta}}^{+})
	 \\
	 &\times  \int  e^{ i   \langle  \hbar^{-1}x,p_{k-1} \rangle}  
	  e^{ -i   \langle  \hbar^{-1/d}x_{l},p_{k-1}-p_{k-2} \rangle}      
	e^{i \hbar^{-1} s_{k+\beta+1} \frac{1}{2} p_{k-1}^2}  \hat{V}(p_{k-1}-p_{k-2}) 
	\tilde{\mathcal{G}}_\beta (p_{k-2},\boldsymbol{\tilde{p}},\boldsymbol{\eta})  \prod_{i=1}^{a_{k+\beta}} e^{i  t_{i} \frac{1}{2} \eta_{i}^2}  
	\\
	&\times  \mathcal{P}_\beta (p_{k-2},\boldsymbol{\tilde{p}},\boldsymbol{s})  e^{i  s_{k-2} \frac{1}{2} (\tilde{p}_{-2}+ \tilde{p}_{-1})^2}   \hat{\mathcal{V}}_{\alpha_{k-2}}(\tilde{p}_{-2}+ \tilde{p}_{-1},p_{k-3}+ \pi_{\beta}(\boldsymbol{\tilde{p}}),\boldsymbol{\eta} )   e^{i  s_{k-3} \frac{1}{2} (p_{k-3}+ \pi_{\beta}(\boldsymbol{\tilde{p}}))^2}
	\\
	&\times   \hat{\mathcal{V}}_{\alpha_{k-3}}(p_{k-3}+ \pi_{\beta}(\boldsymbol{\tilde{p}}),p_{k-4},\boldsymbol{\eta} ) 
	 \prod_{m=1}^{k-2}  e^{ -i   \langle  \hbar^{-1/d}x_{m},p_{m}-p_{m-1} \rangle}   \prod_{m=1}^{k-4}  
	e^{i  s_{m} \frac{1}{2} p_{m}^2}   \hat{\mathcal{V}}_{\alpha_m}(p_m,p_{m-1},\boldsymbol{\eta} ) 
	 \\
	& \times e^{i ( \hbar^{-1}t- \boldsymbol{s}_{1,k+\beta}^{+}- \boldsymbol{t}_{1,a_{k+\beta}}^{+}) \frac{1}{2} p_0^2} \hat{\varphi}(\tfrac{p_0}{\hbar}) \, d\boldsymbol{\eta}d\boldsymbol{p}d\boldsymbol{\tilde{p}}   d\boldsymbol{t} d\boldsymbol{s},
	\end{aligned}
\end{equation*}
where
\begin{equation*}
	\pi_{\beta}(\boldsymbol{\tilde{p}}) =
	\begin{cases}
	\tilde{p}_{-2} + \tilde{p}_0 &\text{if $\beta\in\{1,2\}$}
	\\
	\tilde{p}_{-2} + \tilde{p}_0 +\tilde{p}_{2} &\text{if $\beta\in\{3,4\}$}
	\\
	\tilde{p}_{-2} + \tilde{p}_0 +\tilde{p}_{2} +\tilde{p}_{4}&\text{if $\beta\in\{5,6\}$}.
	\end{cases}
\end{equation*}
For $\beta=1$ we have that
\begin{equation*}
	\begin{aligned}
	\tilde{\mathcal{G}}_1 (p_{k-2},\boldsymbol{\tilde{p}},\boldsymbol{\eta}) ={}&  
	\hat{\mathcal{V}}_{\alpha_{k-1}}(\tilde{p}_{-1},\tilde{p}_{-1}+\tilde{p}_{-2},\boldsymbol{\eta} )
	\hat{\mathcal{V}}_{\alpha_{k}}(\tilde{p}_{0}+p_{k-2},\tilde{p}_{-1},\boldsymbol{\eta} )
	\hat{\mathcal{V}}_{\alpha_{k+1}}(p_{k-2},\tilde{p}_{0}+p_{k-2},\boldsymbol{\eta} )
	\\
	\mathcal{P}_1 (p_{k-2},\boldsymbol{\tilde{p}},\boldsymbol{s}) = {}& e^{i  s_{k-1} \frac{1}{2} \tilde{p}_{-1}^2} 
	e^{i  s_{k} \frac{1}{2}( \tilde{p}_{0}+p_{k-2})^2}  e^{i  s_{k+1} \frac{1}{2}p_{k-2}^2}
	\end{aligned}
\end{equation*}
and for $\beta\geq2$ we have that
\begin{equation*}
	\begin{aligned}
	\tilde{\mathcal{G}}_\beta (p_{k-2},\boldsymbol{\tilde{p}},\boldsymbol{\eta}) &=
	\begin{cases}  
	 \hat{\mathcal{V}}_{\alpha_{k+\beta}}(p_{k-2},\tilde{p}_{\beta-1},\boldsymbol{\eta} )\tilde{\mathcal{G}}_{\beta-1} (\tilde{p}_{\beta-2},\boldsymbol{\tilde{p}},\boldsymbol{\eta}) & \text{$\beta$ even}
	\\
	\hat{\mathcal{V}}_{\alpha_{k+\beta}}(p_{k-2}, p_{k-2} + \tilde{p}_{\beta-1},\boldsymbol{\eta} )\tilde{\mathcal{G}}_{\beta-1} (p_{k-2}+ \tilde{p}_{\beta-1},\boldsymbol{\tilde{p}},\boldsymbol{\eta}) & \text{$\beta$ odd}
	\end{cases}
	\\
	\mathcal{P}_\beta (p_{k-2},\boldsymbol{\tilde{p}},\boldsymbol{s}) &=
	\begin{cases}  
	e^{i  s_{k+\beta} \frac{1}{2}p_{k-2}^2} \mathcal{P}_{\beta-1} (\tilde{p}_{\beta-2},\boldsymbol{\tilde{p}},\boldsymbol{s}) & \text{$\beta$ even}
	\\
	e^{i  s_{k+\beta} \frac{1}{2}p_{k-2}^2} \mathcal{P}_{\beta-1} (p_{k-2}+ \tilde{p}_{\beta-1},\boldsymbol{\tilde{p}},\boldsymbol{s}) & \text{$\beta$ odd}.
	\end{cases}
	\end{aligned}
\end{equation*}
Next we preform the change of variables $p_{k-2}\mapsto p_{k-2}-p_{k-1}$ and $p_m \mapsto p_m-p_{k-2}$ for all $m\in\{l,\dots,k-3\}$ with this change of variable and a relabelling we obtain that
\begin{equation*}
	\begin{aligned}
 	\MoveEqLeft   \mathcal{I}^{\mathrm{dob}}_2(\beta,k,s_{k+\beta+1};\alpha,\boldsymbol{x},\hbar) \varphi (x)
	=  \frac{1}{(2\pi\hbar)^d} \sum_{\alpha\in \N^{k+\beta}} (i\lambda)^{|\alpha|} \int_{\R_{+}^{|\alpha|}}\boldsymbol{1}_{[0,\hbar^{-1}t]}( \boldsymbol{s}_{1,k+\beta}^{+}+\hbar^{-1}s_{k+\beta+1}+ \boldsymbol{t}_{1,a_{k+\beta}}^{+})
	 \\
	 &\times   \int e^{ i   \langle  \hbar^{-1}x,p_{k-2} \rangle}       
	e^{i \hbar^{-1} s_{k+\beta+1} \frac{1}{2} p_{k-2}^2}   \hat{V}(-\tilde{p}_\beta ) 
	 \tilde{\mathcal{G}}_\beta (\tilde{p}_\beta+p_{k-2},\boldsymbol{\tilde{p}},\boldsymbol{\eta}) \mathcal{P}_\beta (\tilde{p}_\beta+p_{k-2},\boldsymbol{\tilde{p}},\boldsymbol{s}) 
	\\
	&\times  e^{i  s_{k-2} \frac{1}{2} (\tilde{p}_{-2}+ \tilde{p}_{-1})^2}   \hat{\mathcal{V}}_{\alpha_{k-2}}(\tilde{p}_{-2}+ \tilde{p}_{-1},p_{k-3}+\tilde{p}_\beta+ \pi_{\beta}(\boldsymbol{\tilde{p}}),\boldsymbol{\eta} )   e^{i  s_{k-3} \frac{1}{2} (p_{k-3}+\tilde{p}_\beta+ \pi_{\beta}(\boldsymbol{\tilde{p}}))^2}
	 \prod_{i=1}^{a_{k+\beta}} e^{i  t_{i} \frac{1}{2} \eta_{i}^2}  
	\\
	&\times   \hat{\mathcal{V}}_{\alpha_{k-3}}(p_{k-3}+\tilde{p}_\beta+ \pi_{\beta}(\boldsymbol{\tilde{p}}),p_{k-4}+\tilde{p}_\beta,\boldsymbol{\eta} ) 
	e^{i ( \hbar^{-1}t- \boldsymbol{s}_{1,k+\beta+1}^{+}- \boldsymbol{t}_{1,a_{k+\beta}}^{+}) \frac{1}{2} p_0^2}
	 \prod_{m=1}^{k-2}  e^{ -i   \langle  \hbar^{-1/d}x_{m},p_{m}-p_{m-1} \rangle}   
	 \\
	& \times
	\prod_{m=1}^{k-4}  
	e^{i  s_{m} \frac{1}{2}( p_{m}+\tilde{\pi}_m(\tilde{p}_\beta))^2}   \hat{\mathcal{V}}_{\alpha_m}(p_m+\tilde{\pi}_m(\tilde{p}_\beta),p_{m-1}+\tilde{\pi}_{m-1}(\tilde{p}_\beta),\boldsymbol{\eta} ) 
	  \hat{\varphi}(\tfrac{p_0}{\hbar}) \, d\boldsymbol{\eta}d\boldsymbol{p}   d\boldsymbol{t}d\boldsymbol{s},
	\end{aligned}
\end{equation*}
where we have used the notation
\begin{equation*}
	\tilde{\pi}_m(\tilde{p}_\beta) = 
	\begin{cases}
		\tilde{p}_\beta &\text{if $m\in\{l,\dots,k-4\}$} 
		\\
		0 &\text{if $m\notin\{l,\dots,k-4\}$} .
	\end{cases}
\end{equation*}
Using this form we can argue as in the proofs of Lemma~\ref{LE:Exp_ran_phases} and Lemma~\ref{expansion_aver_bound_Mainterm} and obtain that
\begin{equation*}
	\begin{aligned}
 	\mathbb{E}&\big[ \mathcal{T}(\beta,k,\sigma,\alpha,\tilde{\alpha};\hbar)\big] =   \frac{(2\pi)^{(2k-2-n)d}\hbar (\rho\hbar)^{2k-4-n}}{(2\pi\hbar)^d}  \int_{\R_{+}^{|\alpha|+|\tilde{\alpha}|}} f(\boldsymbol{s},\boldsymbol{\tilde{s}}) \int \Lambda_n(\boldsymbol{p},\boldsymbol{q},\sigma)    \hat{\varphi}(\tfrac{p_0}{\hbar}) \overline{\hat{\varphi}(\tfrac{q_0}{\hbar})} 
	\\
	\times & \delta(\tfrac{t}{\hbar}- \boldsymbol{s}_{0,k+\beta}^{+}-\tilde{s}_{k+\beta+1}- \boldsymbol{t}_{1,a_{k+\beta}}^{+})  \delta(\tfrac{t}{\hbar}- \boldsymbol{\tilde{s}}_{1,k+\beta+1}^{+}- \boldsymbol{\tilde{t}}_{1,a_{k+\beta}}^{+})  \mathcal{P}_\beta (p_{k-2},\boldsymbol{\tilde{p}},\boldsymbol{s})   \mathcal{P}_\beta (q_{k-2},\boldsymbol{\tilde{q}},\boldsymbol{\tilde{s}})  
	\\
	\times & e^{i  s_{k-2} \frac{1}{2} (\tilde{p}_{-2}+ \tilde{p}_{-1})^2}  e^{-i  \tilde{s}_{k-2} \frac{1}{2} (\tilde{q}_{-2}+ \tilde{q}_{-1})^2}     e^{i  s_{k-3} \frac{1}{2} (p_{k-3}+\tilde{p}_\beta+ \pi_{\beta}(\boldsymbol{\tilde{p}}))^2} e^{-i  \tilde{s}_{k-3} \frac{1}{2} (q_{k-3}+\tilde{q}_\beta+ \pi_{\beta}(\boldsymbol{\tilde{q}}))^2} \prod_{i=1}^{a_{k+\beta}} e^{i  t_{i} \frac{1}{2} \eta_{i}^2}  
	\\
	\times &  
	\prod_{i=1}^{\tilde{a}_{k+\beta}} e^{-i  \tilde{t}_{i} \frac{1}{2} \xi_{i}^2}  
	\prod_{m=0}^{k-4}  e^{i  s_{m} \frac{1}{2}( p_{m}+\tilde{\pi}_m(\tilde{p}_\beta))^2} e^{-i  \tilde{s}_{m} \frac{1}{2}( q_{m}+\tilde{\pi}_m(\tilde{q}_\beta))^2}  \mathcal{G}(\boldsymbol{p}, \boldsymbol{\tilde{p}},\boldsymbol{\tilde{q}},\boldsymbol{\eta}, \boldsymbol{\xi},\sigma ) d\boldsymbol{\eta} d\boldsymbol{\xi}d\boldsymbol{p}d\boldsymbol{q} d\boldsymbol{\tilde{p}}  d\boldsymbol{\tilde{q}}  d\boldsymbol{t} d\boldsymbol{\tilde{t}}  d\boldsymbol{s} d\boldsymbol{\tilde{s}},
	\end{aligned}
\end{equation*}
where we have preformed the change of variables $\tilde{s}_{k+\beta+1}\mapsto\hbar^{-1} \tilde{s}_{k+\beta+1}$. The function $f$ is given by
\begin{equation*}
	f(\boldsymbol{s},\boldsymbol{\tilde{s}}) = 
	\begin{cases}
	\boldsymbol{1}_{[0,\hbar^{-1}t]}\big( \sum_{i=1}^{k-4} s_i \big) \boldsymbol{1}_{[0,\hbar^{-1}t]}\big( \sum_{i=1,i \notin \sigma}^{k-3} \tilde{s}_i +  \tilde{s}_{k+\beta+1}  \big) & \text{if $\sigma_n\leq l$}
	\\
	\boldsymbol{1}_{[0,\hbar^{-1}t]}\big( \sum_{i=1,i\neq l}^{k-4} s_i  \big) \boldsymbol{1}_{[0,\hbar^{-1}t]}\big( \sum_{i=1,i \notin \sigma}^{k-3} \tilde{s}_i   +  \tilde{s}_{k+\beta+1}  \big) & \text{if $\sigma_n>l$}.
	\end{cases}
\end{equation*}
The function $\mathcal{G}$ is given by
\begin{equation*}
	\begin{aligned}
	\MoveEqLeft \mathcal{G}(\boldsymbol{p}, \boldsymbol{\tilde{p}},\boldsymbol{\tilde{p}},\boldsymbol{\eta}, \boldsymbol{\xi},\sigma ) 
	\\
	={}&  
	   \tilde{\mathcal{G}}_\beta (\tilde{p}_\beta+p_{k-2},\boldsymbol{\tilde{p}},\boldsymbol{\eta}) 
	     \hat{\mathcal{V}}_{\alpha_{k-2}}(\tilde{p}_{-2}+ \tilde{p}_{-1},p_{k-3}+\tilde{p}_\beta+ \pi_{\beta}(\boldsymbol{\tilde{p}}),\boldsymbol{\eta} )   \hat{\mathcal{V}}_{\alpha_{k-3}}(p_{k-3}+\tilde{p}_\beta+ \pi_{\beta}(\boldsymbol{\tilde{p}}),p_{k-4},\boldsymbol{\eta} )
	     \\
	     \times& \overline{  \tilde{\mathcal{G}}_\beta (\tilde{q}_\beta+q_{k-2},\boldsymbol{\tilde{q}},\boldsymbol{\xi}) }   \overline{  \hat{\mathcal{V}}_{\tilde{\alpha}_{k-2}}(\tilde{q}_{-2}+ \tilde{q}_{-1},q_{k-3}+\tilde{q}_\beta+ \pi_{\beta}(\boldsymbol{\tilde{q}}),\boldsymbol{\xi} )}  
	    \overline{ \hat{\mathcal{V}}_{\tilde{\alpha}_{k-3}}(q_{k-3}+\tilde{q}_\beta+ \pi_{\beta}(\boldsymbol{\tilde{q}}),q_{k-4},\boldsymbol{\xi} ) }
	   \\
	   \times&
	 \prod_{m=1}^{k-4}    \hat{\mathcal{V}}_{\alpha_m}(p_m+\tilde{\pi}_m(\tilde{p}_\beta),p_{m-1}+\tilde{\pi}_{m-1}(\tilde{p}_\beta),\boldsymbol{\eta} ) \overline{ \hat{\mathcal{V}}_{\tilde{\alpha}_m}(q_m+\tilde{\pi}_m(\tilde{q}_\beta),q_{m-1}+\tilde{\pi}_{m-1}(\tilde{q}_\beta),\boldsymbol{\xi} ) }.
	\end{aligned}
\end{equation*}
  The function $\Lambda_n(\boldsymbol{p},\boldsymbol{q},\sigma)$ is given by
  \begin{equation*}
  	\Lambda_n(\boldsymbol{p},\boldsymbol{q},\sigma) = \delta(p_{k-2}-q_{k-2})  \prod_{i=1}^n \delta(p_{\sigma_{i-1}}- q_{\sigma{i-1}} - p_{\sigma_n}+ q_{\sigma_n} ))   
	\prod_{i=1}^{n+1} \prod_{m=\sigma_{i-1}+1}^{\sigma_{i}-1} \delta(p_m-p_{\sigma_{i-1}}) \delta(q_m-q_{\sigma^1_{i-1}}).
  \end{equation*}
From applying the ``usually'' arguments we obtain the bound
\begin{equation}\label{EQ:expansion_aver_bound_3.recol_22_3.2_5}
	\begin{aligned}
 	\MoveEqLeft  
	\big| \Aver{\mathcal{T}(\beta,k,\sigma,\alpha,\tilde{\alpha};\hbar)}\big|
	 \leq \hbar^3    \frac{C^{|\alpha|+|\tilde{\alpha}|+2} \rho^2 (\rho t)^{2k-n-6} }{ (k-5)!(k-2-n)!}  
	  |\log(\tfrac{\hbar}{t})|^{n+19}  \norm{\varphi}_{\mathcal{H}^{5d+5}_\hbar(\R^d)}^2 \norm{\hat{V}}_{1,\infty,6d+6}^{|\alpha|+|\tilde{\alpha}|+2},
	\end{aligned}
\end{equation}
where we have used that $\beta\leq 6$. We remark that this is a ``combined''  bound, that will be true for all cases. Some of the different configurations will have better bounds, but we have choosen to use a single bound as it will suffices. Combining \eqref{EQ:expansion_aver_bound_3.recol_22_3.2_0},  \eqref{EQ:expansion_aver_bound_3.recol_22_3.2_1},  \eqref{EQ:expansion_aver_bound_3.recol_22_3.2_5} and arguing as above we obtain the estimate
\begin{equation}
\Aver*{ \big\lVert  \sum_{k=5}^{k_0} \mathcal{I}^{\mathrm{dob}}_2(\beta,k;\hbar) \varphi \big\rVert_{L^2(\R^d)}^2} \leq 
 \hbar \tilde{\lambda}^\beta  \norm{\varphi}_{\mathcal{H}^{5d+5}_\hbar(\R^d)}^2 k_0^{10} C^{k_0} \hbar  |\log(\tfrac{\hbar}{t})|^{k_0+19},
\end{equation}
where  $\tilde{\lambda}$ is given by \eqref{EQ_tilde_lambda} and satisfies that $\tilde{\lambda}<1$. We now turn to the case $\beta>6$ again for notational convenience we will assume $\beta$ to be even and in the end of the proof mention what we will need to change for $\beta$ odd. From the definition of the operator we have that
\begin{equation*}
	\begin{aligned}
 	\MoveEqLeft  \mathcal{\tilde{I}}^{\mathrm{dob}}_2(\beta,k,\boldsymbol{x},s_{k+\beta+1};\hbar,l) \varphi (x)
	=  \frac{1}{(2\pi\hbar)^d} \sum_{\alpha\in \N^{k+\beta}} (i\lambda)^{|\alpha|} \int_{\R_{+}^{|\alpha|}} \boldsymbol{1}_{[0,\hbar^{-1}t]}( \boldsymbol{s}_{1,k+\beta}^{+}+ \hbar^{-1}s_{k+\beta+1}+ \boldsymbol{t}_{1,a_{k+\beta}}^{+})
	\\
	\times& \int
	  e^{ i   \langle  \hbar^{-1}x,p_{k+\beta+1} \rangle} e^{ -i   \langle  \hbar^{-1/d}x_{l},p_{k+\beta+1}-p_{k+\beta} \rangle} e^{i \hbar^{-1} s_{k+\beta+1} \frac{1}{2} p_{k+\beta+1}^2} \hat{V}(p_{k+\beta+1}-p_{k+\beta}) e^{ -i   \langle  \hbar^{-1/d}x_{k-3},p_{\beta+k_1}-p_{k-4} \rangle} 
	 \\
	 \times&    
	e^{ -i   \langle  \hbar^{-1/d}(x_{k-2}-x_{k-3}),\sum_{m=0}^{\beta/2+1} p_{2m+k-2}-p_{2m+k-3} \rangle}   \prod_{m=1}^{k-4}  e^{ -i   \langle  \hbar^{-1/d}x_{m},p_{m}-p_{m-1} \rangle}  \prod_{m=1}^{k+\beta}  
	e^{i  s_{m} \frac{1}{2} p_{m}^2}   \hat{\mathcal{V}}_{\alpha_m}(p_m,p_{m-1},\boldsymbol{\eta} ) 
	\\
	\times&
	 \prod_{i=1}^{a_{k+\beta}} e^{i  t_{i} \frac{1}{2} \eta_{i}^2}       
	 e^{i ( \hbar^{-1}t- \boldsymbol{s}_{1,k}^{+}- \boldsymbol{t}_{1,a_{k+\beta}}^{+}) \frac{1}{2} p_0^2} \hat{\varphi}(\tfrac{p_0}{\hbar}) \, d\boldsymbol{\eta}d\boldsymbol{p}  d\boldsymbol{t}d\boldsymbol{s}.
	\end{aligned}
\end{equation*}
We start by preforming the change of variables $p_{k+\beta}\mapsto p_{k+\beta} - p_{k+\beta+1}$ and $p_{m}\mapsto p_{m} - p_{k+\beta}$ for all $m\in\{l,\dots,k+\beta-1\}$. After a relabelling, where we relabel $p_{k+\beta}$ as $\tilde{\eta}$ and $p_{k+\beta+1}$ as $p_{k+\beta}$ we get that
\begin{equation*}
	\begin{aligned}
 	\MoveEqLeft  \mathcal{\tilde{I}}^{\mathrm{dob}}_2(\beta,k,\boldsymbol{x},s_{k+\beta+1};\hbar,l) \varphi (x)
	= \sum_{\alpha\in \N^{k+\beta}}  \frac{(i\lambda)^{|\alpha|}}{(2\pi\hbar)^d}  \int_{\R_{+}^{|\alpha|}} \boldsymbol{1}_{[0,\hbar^{-1}t]}( \boldsymbol{s}_{1,k+\beta}^{+}+ \hbar^{-1}s_{k+\beta+1}+ \boldsymbol{t}_{1,a_{k+\beta}}^{+})
	\\
	\times&  \int e^{ i   \langle  \hbar^{-1}x,p_{k+\beta} \rangle}  e^{i \hbar^{-1} s_{k+\beta+1} \frac{1}{2} p_{k+\beta}^2} \hat{V}(-\tilde{\eta})   
	e^{ -i   \langle  \hbar^{-1/d}(x_{k-2}-x_{k-3}),\sum_{m=0}^{\beta/2+1} p_{2m+k-2}-p_{2m+k-3} \rangle}   
	\\
	\times&     
	\prod_{i=1}^{a_{k+\beta}} e^{i  t_{i} \frac{1}{2} \eta_{i}^2}  
	 \prod_{m=1}^{k+\beta}  
	e^{i  s_{m} \frac{1}{2} (p_{m}+\pi_m(\tilde{\eta}))^2}   \hat{\mathcal{V}}_{\alpha_m}(p_m+ \pi_m(\tilde{\eta}),p_{m-1} + \pi_{m-1}(\tilde{\eta}),\boldsymbol{\eta} )  e^{ -i   \langle  \hbar^{-1/d}x_{k-3},p_{\beta+k}-p_{k-4} \rangle}  
	\\
	\times& \prod_{m=1}^{k-4} \Big\{ e^{ -i   \langle  \hbar^{-1/d}x_{m},p_{m}-p_{m-1} \rangle} \Big\} e^{i ( \hbar^{-1}t-\boldsymbol{s}_{1,k+\beta}^{+}- \hbar^{-1}s_{k+\beta+1}- \boldsymbol{t}_{1,a_{k+\beta}}^{+}) \frac{1}{2} p_0^2} \hat{\varphi}(\tfrac{p_0}{\hbar}) \, d\boldsymbol{\eta}d\boldsymbol{p}  d\boldsymbol{t}d\boldsymbol{s},
	\end{aligned}
\end{equation*}
where the function $\pi_m(\tilde{\eta})$ is given in \eqref{EQ:expansion_aver_bound_3.recol_22_3.2_5}.
 As in the proof of Lemma~\ref{expansion_aver_bound_Mainterm_rec_trun22_special_case_1} we again have to distinguish the three cases  $k-2\notin\sigma$, $k-3,k-2\in\sigma$ and finally the case $k-2\in\sigma$ and $k-3\notin\sigma$. We start with the first case $k-2\notin\sigma$ and again we will assume $\sigma_n=k-3$, this is only done to simplify notation and the other case is estimated completely analogous. After taking the average we start by preforming the change of variables $x_{k-2}\mapsto x_{k-2}-x_{k-3}$ and   $\tilde{x}_{k-2}\mapsto \tilde{x}_{k-2}-\tilde{x}_{k-3}$. Then  by arguing as in the proof of Lemma~\ref{LE:Exp_ran_phases} and Lemma~\ref{expansion_aver_bound_Mainterm} we have that   
\begin{equation}\label{EQ:expansion_aver_bound_3.recol_22_3.2_11}
	\begin{aligned}
 	\MoveEqLeft  
	\Aver{ \mathcal{T}(\beta,k,\sigma,\alpha,\tilde{\alpha};\hbar)}
	=  \frac{(2\pi)^{(2k-2-n)d}\hbar(\rho\hbar)^{2k-4-n}}{(2\pi\hbar)^d} \int_{\R_{+}^{|\alpha|+|\tilde{\alpha}|+2}} f(\boldsymbol{s},\boldsymbol{\tilde{s}}) \int    \mathcal{G}(\boldsymbol{p}, \boldsymbol{\tilde{p}},\boldsymbol{\tilde{q}},\boldsymbol{\eta}, \boldsymbol{\xi},  \tilde{\eta},\tilde{\xi},\sigma )  
	 \\
	 &\times   \delta(\tfrac{t}{\hbar}- \boldsymbol{\tilde{s}}_{0,k+\beta+1}^{+}- \boldsymbol{\tilde{t}}_{1,a_{k+\beta}}^{+})   \delta(\tfrac{t}{\hbar}- \boldsymbol{s}_{0,k+\beta}^{+}-\tilde{s}_{k+\beta+1}- \boldsymbol{t}_{1,a_{k+\beta}}^{+})  \prod_{i=1}^{a_{k+\beta}} e^{i  t_{i} \frac{1}{2} \eta_{i}^2} \prod_{i=1}^{\tilde{a}_{k+\beta}}  e^{-i  \tilde{t}_{i} \frac{1}{2} \xi_{i}^2}
	 \\
	 &\times e^{ -i   \langle  x, p_{n}-\tilde{p}_{\beta-1} +\sum_{m=0}^{\beta/2} \tilde{p}_{2m-2}-\tilde{p}_{2m-3} \rangle}   e^{ i   \langle  \tilde{x}, p_{n}-\tilde{q}_{\beta-1} +\sum_{m=0}^{\beta/2} \tilde{q}_{2m-2}-\tilde{q}_{2m-3} \rangle}
	     e^{i  s_{k-3} \frac{1}{2} \tilde{p}_{-3}^2}  e^{-i  \tilde{s}_{k-3} \frac{1}{2} \tilde{q}_{-3}^2}    
	\\
	&\times    
	    e^{i  s_{k+\beta} \frac{1}{2} (p_{n}+\tilde{\eta})^2}   e^{ - i \tilde{s}_{k+\beta} \frac{1}{2} (p_{n}+\tilde{\xi})^2}   \prod_{m=-2}^{\beta-1}e^{i  s_{m+k} \frac{1}{2} (\tilde{p}_{m}+ \tilde{\eta})^2} e^{-i  \tilde{s}_{m+k} \frac{1}{2} (\tilde{q}_{m}+\tilde{\xi})^2}   
	  \\
	  &\times  \prod_{i=0}^{n-1} \prod_{m=\sigma_i}^{\sigma_{i+1}-1} e^{i  s_{m} \frac{1}{2} (p_{i}+\pi_m(\tilde{\eta}))^2}  e^{-i\tilde{s}_{m} \frac{1}{2} (p_{i}+\pi_m(\tilde{\xi}))^2}  |\hat{\varphi}(\tfrac{p_0}{\hbar})|^2
	  \, d\boldsymbol{\eta} d\boldsymbol{\xi} dx d\tilde{x}d\tilde{\eta}d\tilde{\xi} d\boldsymbol{p} d\boldsymbol{\tilde{p}} d\boldsymbol{\tilde{q}}   d\boldsymbol{t}  d\boldsymbol{\tilde{t}}  d\boldsymbol{\tilde{s}}d\boldsymbol{s},
	\end{aligned}
\end{equation}
where the function $f$ is given by
\begin{equation*}
	f(\boldsymbol{s},\boldsymbol{\tilde{s}}) = \boldsymbol{1}_{[0,\hbar^{-1}t]}\big( \sum_{i=1}^{k-4} s_i  \big) \boldsymbol{1}_{[0,\hbar^{-1}t]}\big( \sum_{i=1,i \notin \sigma}^{k-4} \tilde{s}_i + \tilde{s}_{k+\beta+1}  \big)
\end{equation*}
and the function $\mathcal{G}$ is given by
\begin{equation*}
	\begin{aligned}
	& \mathcal{G}(\boldsymbol{p}, \boldsymbol{\tilde{p}},\boldsymbol{\tilde{q}},\boldsymbol{\eta}, \boldsymbol{\xi},  \tilde{\eta},\tilde{\xi},\sigma )  
	= \overline{\hat{V}(\tilde{\eta})}\hat{V}(\tilde{\xi}) 
	  \hat{\mathcal{V}}_{\alpha_{\beta+k}}(p_{n}+\tilde{\eta}, \tilde{p}_{\beta-1}+\tilde{\eta},\boldsymbol{\eta} )  
	 \overline{\hat{\mathcal{V}}_{\tilde{\alpha}_{\beta+k}}(p_{n}+\tilde{\xi}, \tilde{q}_{\beta-1}+\tilde{\xi},\boldsymbol{\xi} )} 
	\\
	&\times \Big\{\prod_{m=-2}^{\beta-1}  
	 \hat{\mathcal{V}}_{\alpha_{m+k}}(\tilde{p}_m+\tilde{\eta},\tilde{p}_{m-1}+\tilde{\eta},\boldsymbol{\eta} ) \overline{ \hat{\mathcal{V}}_{\tilde{\alpha}_{m+k}}(\tilde{q}_m+\tilde{\xi},\tilde{q}_{m-1}+\tilde{\xi},\boldsymbol{\xi} )}  \Big\}
	  \hat{\mathcal{V}}_{\alpha_{k-3}}(\tilde{p}_{-3}+\tilde{\eta},p_{n-1}+\tilde{\eta},\boldsymbol{\eta} )   
	  \\
	  &\times \overline{\hat{\mathcal{V}}_{\tilde{\alpha}_{k-3}}(\tilde{q}_{-3}+\tilde{\xi},p_{n-1}+\tilde{\xi},\boldsymbol{\eta} )  }
	 \prod_{i=1}^{n-1} \hat{\mathcal{V}}_{\alpha_{\sigma_i}}(p_i+\pi_{\sigma_i}(\tilde{\eta}),p_{i-1}\pi_{\sigma_i-1}(\tilde{\eta}),\boldsymbol{\eta} )
	 \overline{ \hat{\mathcal{V}}_{\tilde{\alpha}_{\sigma_i}}(p_i+\pi_{\sigma_i}(\tilde{\xi}),p_{i-1}+\pi_{\sigma_i-1}(\tilde{\xi}),\boldsymbol{\xi} )}   
	 \\
	 &\times \prod_{m=\sigma_i+1}^{\sigma_{i+1}-1}  
	   \hat{\mathcal{V}}_{\alpha_m}(p_i+\pi_m(\tilde{\eta}),p_{i}+\pi_m(\tilde{\eta}),\boldsymbol{\eta} )\overline{ \hat{\mathcal{V}}_{\tilde{\alpha}_m}(p_i+\pi_m(\tilde{\xi}),p_{i}+\pi_m(\tilde{\xi}),\boldsymbol{\xi} )}
	   \prod_{m=1}^{\sigma_1-1}  
	  \hat{\mathcal{V}}_{\alpha_m}(p_0,p_{0},\boldsymbol{\eta} )  \overline{  \hat{\mathcal{V}}_{\tilde{\alpha}_m}(p_0,p_{0},\boldsymbol{\xi} )}.  
	\end{aligned}
\end{equation*}
We again use the function $\zeta$ as in the previous proofs and write the delta functions as Fourier transforms of the constant function $1$. We the obtain that
\begin{equation*}
	\begin{aligned}
 	\MoveEqLeft  
	\Aver{ \mathcal{T}(\beta,k,\sigma,\alpha,\tilde{\alpha};\hbar)}
	=   \frac{(2\pi)^{(2k-2-n)d}\hbar(\rho\hbar)^{2k-4-n} e^{2\hbar^{-1}t\zeta}}{(2\pi)^2(2\pi\hbar)^d} \int_{\R_{+}^{|\alpha|+|\tilde{\alpha}|+2}} f(\boldsymbol{s},\boldsymbol{\tilde{s}}) \int  |\hat{\varphi}(\tfrac{p_0}{\hbar})|^2  
	 \\
	 \times&  e^{ -i   \langle  x, p_{n}-\tilde{p}_{\beta-1} +\sum_{m=0}^{\beta/2} \tilde{p}_{2m-2}-\tilde{p}_{2m-3} \rangle} 
	  e^{ i   \langle  \tilde{x}, p_{n}-\tilde{q}_{\beta-1} +\sum_{m=0}^{\beta/2} \tilde{q}_{2m-2}-\tilde{q}_{2m-3} \rangle}
	e^{i  s_{k-3} (\frac{1}{2} (\tilde{p}_{-3}+\tilde{\eta})^2+\nu+i\zeta)} 
	\\ 
	\times& e^{-i  \tilde{s}_{k-3}( \frac{1}{2}( \tilde{q}_{-3}+\tilde{\xi})^2-\tilde{\nu}-i\zeta)}    
	\prod_{i=1}^{a_{k+\beta}} e^{i  t_{i}( \frac{1}{2} \eta_{i}^2 +\nu+i\zeta)} \prod_{i=1}^{\tilde{a}_{k+\beta}}  e^{-i  \tilde{t}_{i}( \frac{1}{2} \xi_{i}^2-\tilde{\nu}-i\zeta)}    
	   \prod_{m=-2}^{\beta-1} \Big\{ e^{i  s_{m+k}( \frac{1}{2} (\tilde{p}_{m}+ \tilde{\eta})^2 +\nu+i\zeta)}
	   \\
	   \times&  e^{-i  \tilde{s}_{m+k} (\frac{1}{2} (\tilde{q}_{m}+\tilde{\xi})^2-\tilde{\nu}-i\zeta)} \Big\}
	  e^{i  s_{k+\beta}( \frac{1}{2} (p_{n}+\tilde{\eta})^2 +\nu+i\zeta)}   e^{ - i \tilde{s}_{k_1+\beta}( \frac{1}{2} (p_{n}+\tilde{\xi})^2-\tilde{\nu}-i\zeta)} \mathcal{G}(\boldsymbol{p}, \boldsymbol{\tilde{p}},\boldsymbol{\tilde{q}},\boldsymbol{\eta}, \boldsymbol{\xi},  \tilde{\eta},\tilde{\xi},\sigma )
	  \\
	  \times& \prod_{i=0}^{n-1} \prod_{m=\sigma_i}^{\sigma_{i+1}-1} e^{i  s_{m} (\frac{1}{2} (p_{i}+\pi_m(\tilde{\eta}))^2 +\nu+i\zeta)}  e^{-i\tilde{s}_{m}( \frac{1}{2} (p_{i}+\pi_m(\tilde{\xi}))^2-\tilde{\nu}-i\zeta)}
	  \, d\boldsymbol{\eta} d\boldsymbol{\xi} d\nu d\tilde{\nu} dx d\tilde{x}d\tilde{\eta}d\tilde{\xi} d\boldsymbol{p} d\boldsymbol{\tilde{p}} d\boldsymbol{\tilde{q}}   d\boldsymbol{t}  d\boldsymbol{\tilde{t}}  d\boldsymbol{\tilde{s}}d\boldsymbol{s}.
	\end{aligned}
\end{equation*}
 We now evaluate the integrals in $s_0$,$s_{k+\beta}$,$\tilde{s}_{k+\beta}$, $\tilde{s}_0$, $\tilde{s}_{\sigma_i}$ for $i\in\{1,\dots,n-1\}$ and $s_m$,$\tilde{s}_m$ for $m\in\{k-3,\dots,k_1+\beta-1\}$. From here we do the usual argument by dividing into different cases depending on the size of $t$ and $\tilde{t}$ and do integration by parts for all for all cases, where the varables is ``large''. For the variables $\tilde{p}_m$, $\tilde{q}_m$ with $m\in\{-2,\dots,\beta-1\}$ we use Lemma~\ref{LE:Fourier_trans_resolvent} and integration by parts as in the proof of Lemma~\ref{expansion_aver_bound_Mainterm_rec_trun22_special_case_1}. 
 Combining these arguments we get that
\begin{equation}\label{EQ:expansion_aver_bound_3.recol_22_3.2_15}
	\begin{aligned}
 	\MoveEqLeft  
	\big| \Aver{ \mathcal{T}(\beta,k,\sigma,\alpha,\tilde{\alpha};\hbar,l)}\big|
	 \leq  \frac{C^{|\alpha|+|\tilde{\alpha}|}\hbar (\rho\hbar)^{2k-4-n} }{(2\pi\hbar)^d}   \int_{\R_{+}^{2k-n-7}} f(\boldsymbol{s},\boldsymbol{\tilde{s}})  d\boldsymbol{s}   d\boldsymbol{\tilde{s}}    \sup_{\boldsymbol{\epsilon},\boldsymbol{\gamma},\boldsymbol{\delta},\boldsymbol{\varepsilon}}  \int  |\hat{\varphi}(\tfrac{p_0}{\hbar}) |^2
	  \\
	  &\times  \frac{ \langle \tilde{\eta}\rangle^{-d-1}}{ |\frac{1}{2} (p_{n_1}+\tilde{\xi})^2-\tilde{\nu}-i\zeta|}\frac{ \langle \tilde{\xi}\rangle^{-d-1}}{|\frac{1}{2} (p_{n_1}+\tilde{\eta})^2+\nu +i\zeta|}    \frac{ \langle p_0\rangle^{2(d+1)}}{\langle \tilde{\nu}\rangle |\frac{1}{2} p_0^2-\tilde{\nu}-i\zeta|}\frac{1}{\langle \nu\rangle |\frac{1}{2} p_0^2+\nu +i\zeta|}
	  \\
	  &\times  
	  \frac{ \langle\nu\rangle \langle \tilde{p}_{-3}\rangle^{-d-1} } {|\frac{1}{2} (\tilde{p}_{-3}+\tilde{\eta})^2+\nu +i\zeta|}\frac{\langle\tilde{\nu}\rangle \langle \tilde{q}_{-3}\rangle^{-d-1} }{|\frac{1}{2} (\tilde{q}_{-3}+\tilde{\xi})^2-\tilde{\nu} - i\zeta|}  \prod_{\substack{m=1\\m\neq n_1}}^{n} \frac{\langle p_m-p_{m-1} \rangle^{-d-1}}{ |\frac{1}{2} (p_{m}+\pi_{\sigma_m}(\tilde{\xi}))^2-\tilde{\nu}-i\zeta|} 
	  \\
	  &\times
	   \Big| \partial_{\boldsymbol{\eta}}^{\boldsymbol{\epsilon}} \partial_{\boldsymbol{\xi}}^{\boldsymbol{\gamma}}  \partial_{\boldsymbol{\tilde{p}}}^{\boldsymbol{\delta}} \partial_{\boldsymbol{\tilde{q}}}^{\boldsymbol{\varepsilon}} \mathcal{G}(\boldsymbol{p}, \boldsymbol{\tilde{p}},\boldsymbol{\tilde{p}},\boldsymbol{\eta}, \boldsymbol{\xi},\tilde{\eta},\tilde{\xi},\sigma )  \Big|
	    \prod_{m=-2}^{\beta-1}   \frac{  1}{ \langle \tilde{z}_m \rangle^{d+1} |\tilde{z}_m-	(-1)^m\tilde{x}| \langle z_m \rangle^{d+1}  |z_m- (-1)^mx|} 
	 \\
	&\times 
	 \langle \tilde{\eta}\rangle^{d+1}  \langle \tilde{\xi}\rangle^{d+1}\langle \tilde{q}_{-3} -p_{n-1} \rangle^{d+1} \langle \tilde{p}_{-3} -p_{n-1} \rangle^{d+1} \prod_{m=1}^{n-1} \langle p_m-p_{m-1}\rangle^{3(d+1)}
	\, dx d\tilde{x}d\tilde{\eta}d\tilde{\xi}d\boldsymbol{z} d\boldsymbol{\tilde{z}} d\boldsymbol{\tilde{p}} d\boldsymbol{\tilde{q}} d\boldsymbol{p}d\boldsymbol{\eta}d\boldsymbol{\xi},
	\end{aligned}
\end{equation}
where $C$ is a constant only depending on the dimension and we have used the notation as given in \eqref{EQ:expansion_aver_bound_3.recol_22_3.1_sup_notation}. Moreover, we have inserted the function and used the inequality
\begin{equation*}
	\begin{aligned}
	\MoveEqLeft \frac{\langle \tilde{\nu} \rangle \langle \nu \rangle \langle \tilde{\eta}\rangle^{d+1}  \langle \tilde{\xi}\rangle^{d+1} \langle \tilde{p}_{0}\rangle^{2d+2} \langle \tilde{q}_{-3} -p_{n-1} \rangle^{d+1} \langle \tilde{p}_{-3} -p_{n-1} \rangle^{d+1} \prod_{m=1}^{n-1} \langle p_m-p_{m-1}\rangle^{3(d+1)}  }{\langle \tilde{\nu} \rangle \langle \nu \rangle \langle \tilde{\eta}\rangle^{d+1}  \langle \tilde{\xi}\rangle^{d+1} \langle  \tilde{p}_{0}\rangle^{2d+2} \langle \tilde{q}_{-3} -p_{n-1} \rangle^{d+1} \langle \tilde{p}_{-3} -p_{n-1} \rangle^{d+1} \prod_{m=1}^{n-1} \langle p_m-p_{m-1}\rangle^{3(d+1)}}
	\\
	&\leq C \frac{\langle \tilde{\nu} \rangle \langle \nu \rangle \langle \tilde{\eta}\rangle^{d+1}  \langle \tilde{\xi}\rangle^{d+1}\langle \tilde{p}_{0}\rangle^{2d+2} \langle \tilde{q}_{-3} -p_{n-1} \rangle^{d+1} \langle \tilde{p}_{-3} -p_{n-1} \rangle^{d+1} \prod_{m=1}^{n-1}  \langle p_m-p_{m-1}\rangle^{3(d+1)}  }{\langle \tilde{\nu} \rangle \langle \nu \rangle \langle \tilde{\eta}\rangle^{d+1}  \langle \tilde{\xi}\rangle^{d+1} \langle \tilde{q}_{-3} \rangle^{d+1} \langle \tilde{p}_{-3}\rangle^{d+1} \prod_{m=1}^{n-1}  \langle p_m-p_{m-1}\rangle^{d+1}}.
	\end{aligned}
\end{equation*}
To estimate this we note that
  \begin{equation}\label{EQ:expansion_aver_bound_3.recol_22_3.2_16}
	\begin{aligned}
	\MoveEqLeft  \sup_{\boldsymbol{p},\tilde{\eta},\tilde{\xi}}   \int  \langle \tilde{\eta}\rangle^{d+1}  \langle \tilde{\xi}\rangle^{d+1}  \frac{ \Big| \partial_{\boldsymbol{\eta}}^{\boldsymbol{\epsilon}} \partial_{\boldsymbol{\xi}}^{\boldsymbol{\gamma}}  \partial_{\boldsymbol{\tilde{p}}}^{\boldsymbol{\delta}} \partial_{\boldsymbol{\tilde{q}}}^{\boldsymbol{\varepsilon}} \mathcal{G}(\boldsymbol{p}, \boldsymbol{\tilde{p}},\boldsymbol{\tilde{p}},\boldsymbol{\eta}, \boldsymbol{\xi},\tilde{\eta},\tilde{\xi},\sigma )  \Big|}{
	( \langle \tilde{q}_{-3} -p_{n-1} \rangle \langle \tilde{p}_{-3} -p_{n-1} \rangle)^{-d-1}} \prod_{m=1}^{n-1} \langle p_m-p_{m-1}\rangle^{3(d+1)}
	\, d\boldsymbol{\tilde{p}} d\boldsymbol{\tilde{q}} d\boldsymbol{\eta}d\boldsymbol{\xi} dp_{n}
	 \\
	 &\leq  C^{|\alpha|+|\tilde{\alpha}|}  \norm{\hat{V}}_{1,\infty,3d+3}^{|\alpha|+|\tilde{\alpha}|+2}. 
	\end{aligned}
\end{equation}
Moreover, we have that
\begin{equation}\label{EQ:expansion_aver_bound_3.recol_22_3.2_17}
	\int_{\R_{+}^{2k-n-7}} f(\boldsymbol{s},\boldsymbol{\tilde{s}})  d\boldsymbol{s}_{1,k-4}   d\boldsymbol{\tilde{s}}_{1,(k-4-n)_{+}+1}  \leq \frac{t^{2k-n-6}}{\hbar^{k-4+(k-4-n)_{+}+1}(k-4)!(k-3-n)!}.
\end{equation}
The special power of $\hbar$ is due to the form of $f$ since we always have dependence of at least one $\tilde{s}$. Moreover, we have multiplied with an additional $t$ to obtain a power not containing the term $(k-4-n)_{+}$.
By applying Lemma~\ref{LE:int_posistion} and Lemma~\ref{LE:resolvent_int_est} we get that
\begin{equation}\label{EQ:expansion_aver_bound_3.recol_22_3.2_18}
	\begin{aligned}
 	\MoveEqLeft  \int  |\hat{\varphi}(\tfrac{p_0}{\hbar}) |^2
	 \frac{ \langle \tilde{\eta}\rangle^{-d-1}}{ |\frac{1}{2} (p_{n}+\tilde{\xi})^2-\tilde{\nu}-i\zeta|}\frac{ \langle \tilde{\xi}\rangle^{-d-1}}{|\frac{1}{2} (p_{n}+\tilde{\eta)}^2+\nu +i\zeta|}    \frac{ \langle p_0\rangle^{2(d+1)}}{\langle \tilde{\nu}\rangle |\frac{1}{2} p_0^2-\tilde{\nu}-i\zeta|}\frac{1}{\langle \nu\rangle |\frac{1}{2} p_0^2+\nu +i\zeta|}
	  \\
	  &\times  
	  \frac{ \langle\nu\rangle \langle \tilde{p}_{-3}\rangle^{-d-1} } {|\frac{1}{2} (\tilde{p}_{-3}+\tilde{\eta})^2+\nu +i\zeta|}\frac{\langle\tilde{\nu}\rangle \langle \tilde{q}_{-3}\rangle^{-d-1} }{|\frac{1}{2} (\tilde{q}_{-3}+\tilde{\xi})^2-\tilde{\nu} - i\zeta|}  \prod_{m=1}^{n-1} \frac{\langle p_m-p_{m-1} \rangle^{-d-1}}{ |\frac{1}{2} (p_{m}+\pi_{\sigma_m}(\tilde{\xi}))^2-\tilde{\nu}-i\zeta|} 
	  \\
	  &\times
	    \prod_{m=-2}^{\beta-1}   \frac{  1}{ \langle \tilde{z}_m \rangle^{d+1} |\tilde{z}_m-	(-1)^m\tilde{x}| \langle z_m \rangle^{d+1}  |z_m- (-1)^mx|} 
	\, dx d\tilde{x}d\tilde{\eta}d\tilde{\xi}d\boldsymbol{z} d\boldsymbol{\tilde{z}}  d\boldsymbol{p} d\tilde{p}_{-3} d\tilde{q}_{-3}
	\\
	\leq{}&   (2\pi\hbar)^dC |\log(\tfrac{\hbar}{t})|^{n+5}  \norm{\varphi}_{\mathcal{H}^{d+1}_\hbar(\R^d)}^2 \int \frac{1}{\langle x \rangle^{\frac{\beta+2}{2}}\langle \tilde{x} \rangle^{\frac{\beta+2}{2}}}  \, dx d\tilde{x}
	\\
	\leq{}& (2\pi\hbar)^d C |\log(\tfrac{\hbar}{t})|^{n+5}  \norm{\varphi}_{\mathcal{H}^{d+1}_\hbar(\R^d)}^2, 
	\end{aligned}
\end{equation}
where we have used that $\beta>6$ and we have evaluated the integrals in $\tilde{\eta}$ and $\tilde{\xi}$ as the last ones. Combining the estimates in \eqref{EQ:expansion_aver_bound_3.recol_22_3.2_15}, \eqref{EQ:expansion_aver_bound_3.recol_22_3.2_16}, \eqref{EQ:expansion_aver_bound_3.recol_22_3.2_17} and \eqref{EQ:expansion_aver_bound_3.recol_22_3.2_18} we get that
\begin{equation}\label{EQ:expansion_aver_bound_3.recol_22_3.2_19}
	\begin{aligned}
 	\MoveEqLeft  
	\big| \Aver{ \mathcal{T}(\beta,k,\sigma,\alpha,\tilde{\alpha};\hbar,l)}\big|
	 \leq \hbar^3    \frac{C^{|\alpha|+|\tilde{\alpha}|} \rho^2 (\rho t)^{2k-n-6} }{ (k-4)!(k-3-n)!}  
	  |\log(\tfrac{\hbar}{t})|^{n+5}  \norm{\varphi}_{\mathcal{H}^{d+1}_\hbar(\R^d)}^2 \norm{\hat{V}}_{1,\infty,3d+3}^{|\alpha|+|\tilde{\alpha}|+2}.
	\end{aligned}
\end{equation}
Notice that for this case we have that $n\leq k-3$. We now turn to the case $k-2\in\sigma$ and $k-3\notin\sigma$ as in the proof of Lemma~\ref{expansion_aver_bound_Mainterm_rec_trun22_special_case_1} we here preform the change of variables $x_{k-3}\mapsto x_{k-3}-x_{k-2}$ and   $\tilde{x}_{k-3}\mapsto \tilde{x}_{k-3}-\tilde{x}_{k-2}$ after taking the average. From this we get the same expression as in \eqref{EQ:expansion_aver_bound_3.recol_22_3.2_11}, but where we in the phase function have inner products between $p_{n}$ and $x$, $\tilde{x}$ we now have inner products between $p_{n-1}$ and $x$, $\tilde{x}$. Then by an analogous argument to that we have done above we obtain that
\begin{equation}\label{EQ:expansion_aver_bound_3.recol_22_3.2_20}
	\begin{aligned}
 	\MoveEqLeft  
	\big| \Aver{ \mathcal{T}(\beta,k,\sigma,\alpha,\tilde{\alpha};\hbar,l)}\big|
	 \leq \hbar^3    \frac{C^{|\alpha|+|\tilde{\alpha}|} \rho^2 (\rho t)^{2k-n-6} }{ (k-4)!(k-3-n)!}  
	  |\log(\tfrac{\hbar}{t})|^{n+5}  \norm{\varphi}_{\mathcal{H}^{d+1}_\hbar(\R^d)}^2 \norm{\hat{V}}_{1,\infty,3d+3}^{|\alpha|+|\tilde{\alpha}|+2}.
	\end{aligned}
\end{equation}
What remains is the case $k-2,k-3\in \sigma$. For this case we also preform the change of variables $x_{k-2}\mapsto x_{k-2}-x_{k-3}$ and   $\tilde{x}_{k-2}\mapsto \tilde{x}_{k-2}-\tilde{x}_{k-3}$ after taking the average. However this time when we argue as in the proofs of  Lemma~\ref{LE:Exp_ran_phases} and Lemma~\ref{expansion_aver_bound_Mainterm} we obtain that
\begin{equation*}
	\begin{aligned}
 	\MoveEqLeft  
	\Aver{ \mathcal{T}(\beta,k,\sigma,\alpha,\tilde{\alpha};\hbar)}
	=  \frac{(2\pi)^{(2k-2-n)d}\hbar(\rho\hbar)^{2k-4-n}}{(2\pi\hbar)^d} \int_{\R_{+}^{|\alpha|+|\tilde{\alpha}|+2}} f(\boldsymbol{s},\boldsymbol{\tilde{s}}) \int  |\hat{\varphi}(\tfrac{p_0}{\hbar})|^2   \delta(\tfrac{t}{\hbar}- \boldsymbol{\tilde{s}}_{0,k+\beta+1}^{+}- \boldsymbol{\tilde{t}}_{1,a_{k+\beta}}^{+}) 
	 \\
	 \times&    \delta(\tfrac{t}{\hbar}- \boldsymbol{s}_{0,k+\beta}^{+}-\tilde{s}_{k+\beta+1}- \boldsymbol{t}_{1,a_{k+\beta}}^{+})   
	 e^{ -i   \langle  x, \tilde{q}_{\beta-1}-\tilde{p}_{\beta-1} +\sum_{m=0}^{\beta/2} \tilde{p}_{2m-2}-\tilde{p}_{2m-3}- \tilde{q}_{2m-2}+\tilde{q}_{2m-3} \rangle}
	 \prod_{i=1}^{a_{k+\beta}} e^{i  t_{i} \frac{1}{2} \eta_{i}^2}     e^{i  s_{k-3} \frac{1}{2} \tilde{p}_{-3}^2}
	 \\
	 \times&
	    e^{-i  \tilde{s}_{k-3} \frac{1}{2} \tilde{q}_{-3}^2}  e^{i  s_{k+\beta} \frac{1}{2} (p_{n-1}+\tilde{\eta})^2}   e^{ - i \tilde{s}_{k+\beta} \frac{1}{2} (p_{n-1}+\tilde{\xi})^2}   \prod_{i=1}^{\tilde{a}_{k+\beta}}  e^{-i  \tilde{t}_{i} \frac{1}{2} \xi_{i}^2} 
	 \prod_{m=-2}^{\beta-1}e^{i  s_{m+k} \frac{1}{2} (\tilde{p}_{m}+ \tilde{\eta})^2} e^{-i  \tilde{s}_{m+k} \frac{1}{2} (\tilde{q}_{m}+\tilde{\xi})^2}   
	 \\
	 \times &
	 \prod_{i=0}^{n-2} \prod_{m=\sigma_i}^{\sigma_{i+1}-1} e^{i  s_{m} \frac{1}{2} (p_{i}+\pi_m(\tilde{\eta}))^2}  e^{-i\tilde{s}_{m} \frac{1}{2} (p_{i}+\pi_m(\tilde{\xi}))^2}
	  \mathcal{G}(\boldsymbol{p}, \boldsymbol{\tilde{p}},\boldsymbol{\tilde{q}},\boldsymbol{\eta}, \boldsymbol{\xi},  \tilde{\eta},\tilde{\xi},\sigma )  
	  \, d\boldsymbol{\eta} d\boldsymbol{\xi} dx d\tilde{x}d\tilde{\eta}d\tilde{\xi} d\boldsymbol{p} d\boldsymbol{\tilde{p}} d\boldsymbol{\tilde{q}}   d\boldsymbol{t}  d\boldsymbol{\tilde{t}}  d\boldsymbol{\tilde{s}}d\boldsymbol{s},
	\end{aligned}
\end{equation*}
where the function $f$ is given by
\begin{equation*}
	f(\boldsymbol{s},\boldsymbol{\tilde{s}}) = \boldsymbol{1}_{[0,\hbar^{-1}t]}\big( \sum_{i=1,i\neq l}^{k-4} s_i  \big) \boldsymbol{1}_{[0,\hbar^{-1}t]}\big( \sum_{i=1,i \notin \sigma}^{k-4} \tilde{s}_i + \tilde{s}_{k+1+\beta}  \big),
\end{equation*}
and the function $\mathcal{G}$ is given by
\begin{equation*}
	\begin{aligned}
	& \mathcal{G}(\boldsymbol{p}, \boldsymbol{\tilde{p}},\boldsymbol{\tilde{q}},\boldsymbol{\eta}, \boldsymbol{\xi},  \tilde{\eta},\tilde{\xi},\sigma )  
	= \overline{\hat{V}(\tilde{\eta})}\hat{V}(\tilde{\xi}) 
	  \hat{\mathcal{V}}_{\alpha_{\beta+k}}(p_{n-1}+\tilde{\eta}, \tilde{p}_{\beta-1}+\tilde{\eta},\boldsymbol{\eta} )  
	 \overline{\hat{\mathcal{V}}_{\tilde{\alpha}_{\beta+k}}(p_{n-1}+\tilde{\xi}, \tilde{q}_{\beta-1}+\tilde{\xi},\boldsymbol{\xi} )} 
	\\
	&\times \Big\{\prod_{m=-2}^{\beta-1}  
	 \hat{\mathcal{V}}_{\alpha_{m+k}}(\tilde{p}_m+\tilde{\eta},\tilde{p}_{m-1}+\tilde{\eta},\boldsymbol{\eta} ) \overline{ \hat{\mathcal{V}}_{\tilde{\alpha}_{m+k}}(\tilde{q}_m+\tilde{\xi},\tilde{q}_{m-1}+\tilde{\xi},\boldsymbol{\xi} )}  \Big\}
	  \hat{\mathcal{V}}_{\alpha_{k-3}}(\tilde{p}_{-3}+\tilde{\eta},p_{n-2}+\tilde{\eta},\boldsymbol{\eta} )   
	  \\
	  &\times \overline{\hat{\mathcal{V}}_{\tilde{\alpha}_{k-3}}(\tilde{q}_{-3}+\tilde{\xi},p_{n-2}+\tilde{\xi},\boldsymbol{\eta} )  }
	 \prod_{i=1}^{n-2} \hat{\mathcal{V}}_{\alpha_{\sigma_i}}(p_i+\pi_{\sigma_i}(\tilde{\eta}),p_{i-1}\pi_{\sigma_i-1}(\tilde{\eta}),\boldsymbol{\eta} )
	 \overline{ \hat{\mathcal{V}}_{\tilde{\alpha}_{\sigma_i}}(p_i+\pi_{\sigma_i}(\tilde{\xi}),p_{i-1}+\pi_{\sigma_i-1}(\tilde{\xi}),\boldsymbol{\xi} )}   
	 \\
	 &\times \prod_{m=\sigma_i+1}^{\sigma_{i+1}-1}  
	   \hat{\mathcal{V}}_{\alpha_m}(p_i+\pi_m(\tilde{\eta}),p_{i}+\pi_m(\tilde{\eta}),\boldsymbol{\eta} )\overline{ \hat{\mathcal{V}}_{\tilde{\alpha}_m}(p_i+\pi_m(\tilde{\xi}),p_{i}+\pi_m(\tilde{\xi}),\boldsymbol{\xi} )}
	   \prod_{m=1}^{\sigma_1-1}  
	  \hat{\mathcal{V}}_{\alpha_m}(p_0,p_{0},\boldsymbol{\eta} )  \overline{  \hat{\mathcal{V}}_{\tilde{\alpha}_m}(p_0,p_{0},\boldsymbol{\xi} )}.  
	\end{aligned}
\end{equation*}
We can again argue as above and obtain that
\begin{equation}\label{EQ:expansion_aver_bound_3.recol_22_3.2_21}
	\begin{aligned}
 	\MoveEqLeft  
	\big| \Aver{ \mathcal{T}(\beta,k,\sigma,\alpha,\tilde{\alpha};\hbar,l)}\big|
	 \leq  \frac{C^{|\alpha|+|\tilde{\alpha}|}\hbar (\rho\hbar)^{2k-4-n} }{(2\pi\hbar)^d}   \int_{\R_{+}^{2k-n-8}} f(\boldsymbol{s},\boldsymbol{\tilde{s}})  d\boldsymbol{s}   d\boldsymbol{\tilde{s}}    \sup_{\boldsymbol{\epsilon},\boldsymbol{\gamma},\boldsymbol{\delta},\boldsymbol{\varepsilon}}  \int  |\hat{\varphi}(\tfrac{p_0}{\hbar}) |^2
	  \\
	  &\times  \frac{ \langle \tilde{\eta}+p_{n-1}\rangle^{-d-1}}{ |\frac{1}{2} (p_{n-1}+\tilde{\xi})^2-\tilde{\nu}-i\zeta|} 
	  \frac{   \langle \tilde{\eta}+p_{l^*}\rangle^{-d-1} \langle \tilde{\xi}\rangle^{-d-1}}{|\frac{1}{2} (p_{n-1}+\tilde{\eta})^2+\nu +i\zeta|}    
	  \frac{ \langle p_{n-1} -\tilde{p}_{\beta-1} \rangle^{-d-1}}{|\frac{1}{2} (p_{l^*}+\tilde{\eta})^2+\nu +i\zeta|}
	  \frac{ \langle p_0\rangle^{4(d+1)}}{\langle \tilde{\nu}\rangle |\frac{1}{2} p_0^2-\tilde{\nu}-i\zeta|}
	  \\
	  &\times  
	  \frac{1}{\langle \nu\rangle |\frac{1}{2} p_0^2+\nu +i\zeta|}
	  \frac{ \langle\nu\rangle \langle \tilde{p}_{-3}\rangle^{-d-1} } {|\frac{1}{2} (\tilde{p}_{-3}+\tilde{\eta})^2+\nu +i\zeta|}\frac{\langle\tilde{\nu}\rangle \langle \tilde{q}_{-3}\rangle^{-d-1} }{|\frac{1}{2} (\tilde{q}_{-3}+\tilde{\xi})^2-\tilde{\nu} - i\zeta|}  \prod_{m=1}^{n-1} \frac{\langle p_m-p_{m-1} \rangle^{-d-1}}{ |\frac{1}{2} (p_{m}+\pi_{\sigma_m}(\tilde{\xi}))^2-\tilde{\nu}-i\zeta|} 
	  \\
	  &\times
	   \Big| \partial_{\boldsymbol{\eta}}^{\boldsymbol{\epsilon}} \partial_{\boldsymbol{\xi}}^{\boldsymbol{\gamma}}  \partial_{\boldsymbol{\tilde{p}}}^{\boldsymbol{\delta}} \partial_{\boldsymbol{\tilde{q}}}^{\boldsymbol{\varepsilon}} \mathcal{G}(\boldsymbol{p}, \boldsymbol{\tilde{p}},\boldsymbol{\tilde{p}},\boldsymbol{\eta}, \boldsymbol{\xi},\tilde{\eta},\tilde{\xi},\sigma )  \Big|
	    \prod_{m=-2}^{\beta-1}   \frac{  \langle \tilde{z}_m \rangle^{-d-1} \langle z_m \rangle^{-d-1} }{  |\tilde{z}_m-(-1)^m\tilde{x}|   |z_m- (-1)^mx|}  \prod_{m=1}^{n-2} \langle p_m-p_{m-1}\rangle^{2d+2}
	 \\
	&\times 
	  \langle \tilde{\eta}\rangle^{2d+2} \langle \tilde{\xi}\rangle^{d+1}\langle \tilde{q}_{-3} -p_{n-2} \rangle^{d+1} \langle \tilde{p}_{-3} -p_{n-2} \rangle^{2d+2} \prod_{m=1}^{n-2} \langle p_m-p_{m-1}\rangle^{3(d+1)}
	  \\
	  &\times  \langle p_{n-1}-\tilde{p}_{\beta-1}\rangle^{d+1}   \prod_{m=-2}^{\beta-1} \langle \tilde{p}_m-\tilde{p}_{m-1}\rangle^{d+1}
	\, dx d\tilde{x}d\tilde{\eta}d\tilde{\xi}d\boldsymbol{z} d\boldsymbol{\tilde{z}} d\boldsymbol{\tilde{p}} d\boldsymbol{\tilde{q}} d\boldsymbol{p}d\boldsymbol{\eta}d\boldsymbol{\xi},
	\end{aligned}
\end{equation}
where $C$ is a constant only depending on the dimension and the index $l^{*}$ is the largest index such that $\sigma_{l^*}\leq l$ and we have used the notation as given in \eqref{EQ:expansion_aver_bound_3.recol_22_3.1_sup_notation}.  Compared to the above argument we have here also inserted a slightly different function but the principle is the same. 
To estimate this we note that
  \begin{equation}\label{EQ:expansion_aver_bound_3.recol_22_3.2_22}
	\begin{aligned}
	\MoveEqLeft  \sup_{\boldsymbol{p},\tilde{\eta},\tilde{\xi}}   \int  \Big| \partial_{\boldsymbol{\eta}}^{\boldsymbol{\epsilon}} \partial_{\boldsymbol{\xi}}^{\boldsymbol{\gamma}}  \partial_{\boldsymbol{\tilde{p}}}^{\boldsymbol{\delta}} \partial_{\boldsymbol{\tilde{q}}}^{\boldsymbol{\varepsilon}} \mathcal{G}(\boldsymbol{p}, \boldsymbol{\tilde{p}},\boldsymbol{\tilde{p}},\boldsymbol{\eta}, \boldsymbol{\xi},\tilde{\eta},\tilde{\xi},\sigma )  \Big|
	\langle \tilde{q}_{-3} -p_{n-2} \rangle^{d+1} \langle \tilde{p}_{-3} -p_{n-2} \rangle^{2d+2} \prod_{m=1 }^{n-2} \langle p_m-p_{m-1}\rangle^{3d+3}
	\\
	&\times \langle \tilde{\eta}\rangle^{2d+2}  \langle \tilde{\xi}\rangle^{2d+2}   \langle p_{n-1}-\tilde{p}_{\beta-1}\rangle^{d+1} \prod_{m=1}^{n-2} \langle p_m-p_{m-1}\rangle^{2d+2}  \prod_{m=-2}^{\beta-1} \langle \tilde{p}_m-\tilde{p}_{m-1}\rangle^{d+1} \, d\boldsymbol{\tilde{p}} d\boldsymbol{\tilde{q}} d\boldsymbol{\eta}d\boldsymbol{\xi}
	 \\
	 \leq{}&  C^{|\alpha|+|\tilde{\alpha}|}  \norm{\hat{V}}_{1,\infty,3d+3}^{|\alpha|+|\tilde{\alpha}|+2}. 
	\end{aligned}
\end{equation}
Note that we are not integrating in $\tilde{p}_{-3}$ and $\tilde{q}_{-3}$. Moreover we have that
\begin{equation}\label{EQ:expansion_aver_bound_3.recol_22_3.2_23}
	\int_{\R_{+}^{2k-n-8}} f(\boldsymbol{s},\boldsymbol{\tilde{s}})  d\boldsymbol{s}_{1,k-5}   d\boldsymbol{\tilde{s}}_{1,k-4-(n-2)_{+}+1}  \leq \frac{t^{2k-n-6}}{\hbar^{2k-8-(n-2)_{+}+1}(k-5)!(k-3-(n-2)_{+})!}.
\end{equation}
The special power of $\hbar$ is due to the form of $f$ since we always have dependence of at least one $\tilde{s}$. Moreover, we have multiplied with two additional $t$s to obtain a power not containing the term $(n-2)_{+}$.
By applying Lemma~\ref{LE:int_posistion}, Lemma~\ref{LE:resolvent_int_est} and Lemma~\ref{LE:est_res_combined} we get that
\begin{equation}\label{EQ:expansion_aver_bound_3.recol_22_3.2_24}
	\begin{aligned}
 	\MoveEqLeft  \int  |\hat{\varphi}(\tfrac{p_0}{\hbar}) |^2
	 \frac{ \langle \tilde{\eta}+p_{n-1}\rangle^{-d-1}}{ |\frac{1}{2} (p_{n-1}+\tilde{\xi})^2-\tilde{\nu}-i\zeta|} 
	  \frac{   \langle \tilde{\eta}+p_{l^*}\rangle^{-d-1} \langle \tilde{\xi}\rangle^{-d-1}}{|\frac{1}{2} (p_{n-1}+\tilde{\eta})^2+\nu +i\zeta|}    
	  \frac{ \langle p_{n-1} -\tilde{p}_{\beta-1} \rangle^{-d-1}}{|\frac{1}{2} (p_{l^*}+\tilde{\eta})^2+\nu +i\zeta|}
	  \frac{ \langle p_0\rangle^{4(d+1)}}{\langle \tilde{\nu}\rangle |\frac{1}{2} p_0^2-\tilde{\nu}-i\zeta|}
	  \\
	  &\times  
	  \frac{1}{\langle \nu\rangle |\frac{1}{2} p_0^2+\nu +i\zeta|}
	  \frac{ \langle\nu\rangle \langle \tilde{p}_{-3}\rangle^{-d-1} } {|\frac{1}{2} (\tilde{p}_{-3}+\tilde{\eta})^2+\nu +i\zeta|}\frac{\langle\tilde{\nu}\rangle \langle \tilde{q}_{-3}\rangle^{-d-1} }{|\frac{1}{2} (\tilde{q}_{-3}+\tilde{\xi})^2-\tilde{\nu} - i\zeta|}  \prod_{m=1}^{n-2} \frac{\langle p_m-p_{m-1} \rangle^{-d-1}}{ |\frac{1}{2} (p_{m}+\pi_{\sigma_m}(\tilde{\xi}))^2-\tilde{\nu}-i\zeta|} 
	\\
	&\times     
	    \prod_{m=-2}^{\beta-1}   \frac{  \langle \tilde{z}_m \rangle^{-d-1}\langle z_m \rangle^{-d-1}}{  |\tilde{z}_m-	(-1)^m\tilde{x}|   |z_m- (-1)^mx|} 
	\, dx d\tilde{x}d\tilde{\eta}d\tilde{\xi}d\boldsymbol{z} d\boldsymbol{\tilde{z}}  d\boldsymbol{p} d\tilde{p}_{-3} d\tilde{q}_{-3}
	\\
	\leq{}&   (2\pi\hbar)^dC |\log(\tfrac{\hbar}{t})|^{n+5}  \norm{\varphi}_{\mathcal{H}^{2d+2}_\hbar(\R^d)}^2 \int \frac{1}{\langle x \rangle^{\frac{\beta+2}{2}}\langle \tilde{x} \rangle^{\frac{\beta+2}{2}}}  \, dx d\tilde{x}
	\\
	\leq{}& (2\pi\hbar)^d C |\log(\tfrac{\hbar}{t})|^{n+5}  \norm{\varphi}_{\mathcal{H}^{2d+2}_\hbar(\R^d)}^2, 
	\end{aligned}
\end{equation}
where we have evaluated the integrals in the following order first $p_{-3}$, then $\tilde{\eta}$ using Lemma~\ref{LE:est_res_combined}. This produce the factor $|p_{n-1}-p_{l^*}|$ which we integrate by evaluating the integral in $p_{n-1}$. Next the integrals in $\tilde{q}_3$ and the remaining $p$ excluding $p_0$ is evaluated using Lemma~\ref{LE:resolvent_int_est}. Then the integral in $\tilde{\xi}$, $\nu$ and $\tilde{\nu}$ is evaluated using Lemma~\ref{LE:resolvent_int_est} and finaly the integral in $p_0$. Combining the estimates in \eqref{EQ:expansion_aver_bound_3.recol_22_3.2_21}, \eqref{EQ:expansion_aver_bound_3.recol_22_3.2_22}, \eqref{EQ:expansion_aver_bound_3.recol_22_3.2_23} and \eqref{EQ:expansion_aver_bound_3.recol_22_3.2_24} we obtain the estimate
\begin{equation}\label{EQ:expansion_aver_bound_3.recol_22_3.2_25}
	\begin{aligned}
 	\MoveEqLeft  
	\big| \Aver{ \mathcal{T}(\beta,k,\sigma,\alpha,\tilde{\alpha};\hbar,l)}\big|
	 \leq \hbar^3    \frac{C^{|\alpha|+|\tilde{\alpha}|} \rho^2 (\rho t)^{2k-n-6} }{(k-5)!(k-3-(n-2)_{+})!}  
	  |\log(\tfrac{\hbar}{t})|^{n+5}  \norm{\varphi}_{\mathcal{H}^{2d+2}_\hbar(\R^d)}^2 \norm{\hat{V}}_{1,\infty,3d+3}^{|\alpha|+|\tilde{\alpha}|+2}.
	\end{aligned}
\end{equation}
Combining the estimates in \eqref{EQ:expansion_aver_bound_3.recol_22_3.2_0}, \eqref{EQ:expansion_aver_bound_3.recol_22_3.2_1}, \eqref{EQ:expansion_aver_bound_3.recol_22_3.2_19}, \eqref{EQ:expansion_aver_bound_3.recol_22_3.2_20} and \eqref{EQ:expansion_aver_bound_3.recol_22_3.2_25} we obtain the bound
\begin{equation}\label{EQ:expansion_aver_bound_3.recol_22_3.2_26}
	\begin{aligned}
	\MoveEqLeft \big\lVert \sum_{k=5}^{k_0}  \mathcal{I}^{\mathrm{dob}}_2(\beta,k;\hbar) \varphi \big\rVert_{L^2(\R^d)}^2 
	 \leq \hbar \norm{\varphi}_{\mathcal{H}^{2d+2}_\hbar(\R^d)}^2 \sum_{\alpha,\tilde{\alpha}\in\N^{k+\beta}} (C_d\lambda\norm{\hat{V}}_{1,\infty,3d+3})^{|\alpha|+|\tilde{\alpha}|+2} k_0^4  
	 \\
	 &\times  \sum_{k_1=5}^{k_0} \sum_{k_2=0}^{k_0+2-k_1} \sum_{l=1}^{k-4} \Big[ Ck^3  \rho (\rho t)^{2k-n-5}
	  |\log(\tfrac{\hbar}{t})|^{k+3} + \sum_{n=0}^{k-3} \binom{k-2}{n}^2 n! \frac{ \rho (\rho t)^{2k-n-5} }{ (k-4)!(k-3-n)!}  
	  |\log(\tfrac{\hbar}{t})|^{n+5}   \Big]
	  \\
	  \leq{}&\tilde{\lambda}^{\beta} \norm{\varphi}_{\mathcal{H}^{2d+2}_\hbar(\R^d)}^2 \hbar C^{k_0} k_0^{10}  |\log(\tfrac{\hbar}{t})|^{k_0+3},
	\end{aligned}
\end{equation}
where $\tilde{\lambda}$ is given by \eqref{EQ_tilde_lambda}. This concludes the proof for the case $\beta$ even for $\beta$ odd we get a slightly different form of \eqref{EQ:expansion_aver_bound_3.recol_22_3.2_1} similar to the one stated in the proof of Lemma~\ref{expansion_aver_bound_Mainterm_rec_trun22_special_case_1} but otherwise the proofs are analogous.
This concludes the proof as the the bound obtained for $\beta>6$ is ``better'' than the bound stated in the Lemma.  
\end{proof}
 \begin{lemma}\label{expansion_aver_bound_Mainterm_rec_trun22_special_case_3}
Assume we are in the setting of Lemma~\ref{expansion_aver_bound_Mainterm_rec_trun22_5}. Then for any $\beta\in\N$ we have that
\begin{equation*}
	\begin{aligned}
	 \mathbb{E}\Big[\big\lVert  \sum_{k=4}^{k_0}   \mathcal{I}^{\mathrm{dob}}_3(\beta,k;\hbar)  \varphi\big\rVert_{L^2(\R^d)}^2 \Big]
	\leq  \tau_0^{-3} \tilde{\lambda}^\beta  \norm{\varphi}_{\mathcal{H}^{5d+5}_\hbar(\R^d)}^2 k_0^{5} C^{k_0} |\log(\tfrac{\hbar}{t})|^{k_0+20}.
	\end{aligned}
\end{equation*}
where $\tilde{\lambda}<1$.  
\end{lemma}
\begin{proof}
The proof is very similar to the proof of Lemma~\ref{expansion_aver_bound_Mainterm_rec_trun22_special_case_1}. Firstly we note that
\begin{equation}\label{EQ:expansion_aver_bound_3.recol_22_3.3_0}
\big\lVert  \sum_{k=4}^{k_0} \mathcal{I}^{\mathrm{dob}}_3(\beta,k;\hbar) \varphi \big\rVert_{L^2(\R^d)}^2 \leq \frac{k_0^2 t \lambda^2}{\tau_0 \hbar^2} \sum_{k=5}^{k_0}   \int_{0}^{t_0}  \big\lVert \tilde{\mathcal{I}}^{\mathrm{dob}}_3(\beta,k,s_{k+\beta+1};\hbar) \varphi \big\rVert_{L^2(\R^d)}^2 \, ds_{k+\beta+1},
\end{equation}
where $t_0=\frac{t}{\tau_0}$ and
\begin{equation*}
	\begin{aligned}	
	 \tilde{\mathcal{I}}^{\mathrm{dob}}_3(\beta,k,s_{k+\beta+1};\hbar)
	= {} &   \sum_{\boldsymbol{x}\in\mathcal{X}_{\neq}^{k-1}}  \sum_{\alpha\in\N^{k}}   \int_{[0,t_0]^{\beta+k}} \boldsymbol{1}_{[s_{k+\beta+1},t]}(s_{k+\beta})
	 V^{s_{k+\beta+1}}_{\hbar, x_{k-1}}  
	  \tilde{\Theta}_{\beta}^{\mathrm{dob}}(\boldsymbol{s}_{\beta+k,k};V,\hbar)
	  \\
	  &\times  \prod_{m=1}^{k}\Theta_{\alpha_m}(s_{m-1},{s}_{m},x_{\iota(m)};V,\hbar)   U_{\hbar,0}(-t)   \, d\boldsymbol{s}_{k,1}.
	 \end{aligned}
\end{equation*}
Applying Lemma~\ref{LE:crossing_dom_ladder} we obtain that 
\begin{equation}\label{EQ:expansion_aver_bound_3.recol_22_3.3_1}
	\begin{aligned}
	\MoveEqLeft \mathbb{E}\Big[ \int_{0}^{t_0}  \big\lVert \tilde{\mathcal{I}}^{\mathrm{dob}}_3(\beta,k,s_{k+\beta+1};\hbar) \varphi \big\rVert_{L^2(\R^d)}^2 \, ds_{k+\beta+1}\Big]
	\\
	 &\leq 2 \sum_{n=0}^{k-1} \binom{k-1}{n} n! \sum_{\alpha,\tilde{\alpha}\in\N^{k+\beta}} \lambda^{|\alpha|+|\tilde{\alpha}|} \sum_{\sigma\in\mathcal{A}(k-2,n)} \big|\Aver{ \mathcal{T}(\beta,k,\sigma,\alpha,\tilde{\alpha};\hbar)}\big| ,
	\end{aligned}
\end{equation}
where
\begin{equation*}
	\begin{aligned}
	\mathcal{T}(\beta,k,\sigma,\alpha,\tilde{\alpha};\hbar)
	= {}&
	 \sum_{(\boldsymbol{x},\boldsymbol{\tilde{x}})\in \mathcal{X}_{\neq}^{2k-4}}  \prod_{i=1}^n \frac{\delta(x_{\sigma_i}- \tilde{x}_{\sigma_{i}})}{\rho}  
	 \\
	 &\times \int_{0}^{t_0}  \int_{\R^d} \mathcal{I}^{\mathrm{dob}}_3(\beta,k,s_{k+\beta+1};\alpha,\boldsymbol{x},\hbar) \varphi (x) \overline{ \mathcal{I}^{\mathrm{dob}}_3(\beta,k,s_{k+\beta+1};\tilde{\alpha},\boldsymbol{\tilde{x}},\hbar) \varphi (x)}  \,dx ds_{k+\beta+1}. 
	\end{aligned}
\end{equation*}
Before we proceed we will divide into the two cases $\beta\leq 6$ and $\beta>6$. We will start with the first case. From the definition of the operator we have that
\begin{equation*}
	\begin{aligned}
 	\MoveEqLeft   \mathcal{I}^{\mathrm{dob}}_3(\beta,k,s_{k+\beta+1};\alpha,\boldsymbol{x},\hbar) \varphi (x)
	=  \frac{1}{(2\pi\hbar)^d} \sum_{\alpha\in \N^{k+\beta}} (i\lambda)^{|\alpha|} \int_{\R_{+}^{|\alpha|}}  \boldsymbol{1}_{[0,\hbar^{-1}t_0]}( \boldsymbol{s}_{1,k+\beta}^{+}+\hbar^{-1}s_{k+\beta+1}+ \boldsymbol{t}_{1,a_{k+\beta}}^{+})
	 \\
	 \times& \int_{\R^{(|\alpha|+1)d}}  e^{ i   \langle  \hbar^{-1}x,p_{k+\beta} \rangle}  
	   e^{ -i   \langle  \hbar^{-1/d}x_{l},p_{k+\beta+1}-p_{k+\beta} \rangle}      
	e^{i \hbar^{-1} s_{k+\beta+1} \frac{1}{2} p_{k+\beta+1}^2}   \hat{V}(p_{k+\beta+1}-p_{k+\beta}) 
	 \prod_{i=1}^{a_{k+\beta}} e^{i  t_{i} \frac{1}{2} \eta_{i}^2}   
	 \\
	 \times&    
	 \prod_{m=-3}^\beta e^{ -i   \langle  \hbar^{-1/d}x_{\tilde{\iota}(k)}, p_{k+m}-p_{k+m-1} \rangle}
	\prod_{m=1}^{k-4}  e^{ -i   \langle  \hbar^{-1/d}x_{m},p_{m}-p_{m-1} \rangle}      
	 \prod_{m=1}^{k+\beta}  
	e^{i  s_{m} \frac{1}{2} p_{m}^2}   \hat{\mathcal{V}}_{\alpha_m}(p_m,p_{m-1},\boldsymbol{\eta} ) 
	\\
	\times & e^{i ( \hbar^{-1}t_0- \boldsymbol{s}_{1,k+\beta}^{+}- \boldsymbol{t}_{1,a_{k+\beta}}^{+}) \frac{1}{2} p_0^2} \hat{\varphi}(\tfrac{p_0}{\hbar}) \, d\boldsymbol{\eta}d\boldsymbol{p}   d\boldsymbol{t} d\boldsymbol{s}.
	\end{aligned}
\end{equation*}
As in the proof of Lemma~\ref{expansion_aver_bound_Mainterm_rec_trun22_special_case_1} we preform the change of variables 
\begin{equation*}
p_{k+m}\mapsto p_{k+m}-p_{k+m+1} \quad\text{and}\quad  p_{k-3}\mapsto p_{k-3}-p_{k+m}\quad  \text{for all }
\begin{cases}
	m\in\{-2,0\} &\text{if $\beta\in\{1,2\}$}
	\\
	m\in\{-2,0,2\} &\text{if $\beta\in\{3,4\}$}
	\\
	m\in\{-2,0,2,4\} &\text{if $\beta\in\{5,6\}$}.
\end{cases}
\end{equation*}
After this change of variables and a relabelling we obtain that
\begin{equation*}
	\begin{aligned}
 	\MoveEqLeft   \mathcal{I}^{\mathrm{dob}}_3(\beta,k,s_{k+\beta+1};\alpha,\boldsymbol{x},\hbar) \varphi (x)
	=  \frac{1}{(2\pi\hbar)^d} \sum_{\alpha\in \N^{k+\beta}} (i\lambda)^{|\alpha|} \int_{\R_{+}^{|\alpha|}} \boldsymbol{1}_{[0,\hbar^{-1}t_0]}( \boldsymbol{s}_{1,k+\beta}^{+}+\hbar^{-1}s_{k+\beta+1}+ \boldsymbol{t}_{1,a_{k+\beta}}^{+})
	 \\
	 &\times  \int e^{ i   \langle  \hbar^{-1}x,p_{k-1} \rangle}  
	  e^{ -i   \langle  \hbar^{-1/d}x_{k-1},p_{k-1}-p_{k-2} \rangle}      
	e^{i \hbar^{-1} s_{k+\beta+1} \frac{1}{2} p_{k-1}^2}  \hat{V}(p_{k-1}-p_{k-2}) 
	 \prod_{i=1}^{a_{k+\beta}} e^{i  t_{i} \frac{1}{2} \eta_{i}^2}  
	\\
	&\times  \tilde{\mathcal{G}}_\beta (p_{k-2},\boldsymbol{\tilde{p}},\boldsymbol{\eta}) \mathcal{P}_\beta (p_{k-2},\boldsymbol{\tilde{p}},\boldsymbol{s})  e^{i  s_{k-2} \frac{1}{2} (\tilde{p}_{-2}+ \tilde{p}_{-1})^2}   \hat{\mathcal{V}}_{\alpha_{k-2}}(\tilde{p}_{-2}+ \tilde{p}_{-1},p_{k-3}+ \pi_{\beta}(\boldsymbol{\tilde{p}}),\boldsymbol{\eta} )   e^{i  s_{k-3} \frac{1}{2} (p_{k-3}+ \pi_{\beta}(\boldsymbol{\tilde{p}}))^2}
	\\
	&\times   \hat{\mathcal{V}}_{\alpha_{k-3}}(p_{k-3}+ \pi_{\beta}(\boldsymbol{\tilde{p}}),p_{k-4},\boldsymbol{\eta} ) 
	 \prod_{m=1}^{k-2}  e^{ -i   \langle  \hbar^{-1/d}x_{m},p_{m}-p_{m-1} \rangle}   \prod_{m=1}^{k-4}  
	e^{i  s_{m} \frac{1}{2} p_{m}^2}   \hat{\mathcal{V}}_{\alpha_m}(p_m,p_{m-1},\boldsymbol{\eta} ) 
	 \\
	& \times e^{i ( \hbar^{-1}t_0- \boldsymbol{s}_{1,k+\beta}^{+}- \boldsymbol{t}_{1,a_{k+\beta}}^{+}) \frac{1}{2} p_0^2} \hat{\varphi}(\tfrac{p_0}{\hbar}) \, d\boldsymbol{\eta}d\boldsymbol{p}d\boldsymbol{\tilde{p}}   d\boldsymbol{t} d\boldsymbol{s},
	\end{aligned}
\end{equation*}
where
\begin{equation*}
	\pi_{\beta}(\boldsymbol{\tilde{p}}) =
	\begin{cases}
	\tilde{p}_{-2} + \tilde{p}_0 &\text{if $\beta\in\{1,2\}$}
	\\
	\tilde{p}_{-2} + \tilde{p}_0 +\tilde{p}_{2} &\text{if $\beta\in\{3,4\}$}
	\\
	\tilde{p}_{-2} + \tilde{p}_0 +\tilde{p}_{2} +\tilde{p}_{4}&\text{if $\beta\in\{5,6\}$}.
	\end{cases}
\end{equation*}
For $\beta=1$ we have that
\begin{equation*}
	\begin{aligned}
	\tilde{\mathcal{G}}_1 (p_{k-2},\boldsymbol{\tilde{p}},\boldsymbol{\eta}) ={}&  
	\hat{\mathcal{V}}_{\alpha_{k-1}}(\tilde{p}_{-1},\tilde{p}_{-1}+\tilde{p}_{-2},\boldsymbol{\eta} )
	\hat{\mathcal{V}}_{\alpha_{k}}(\tilde{p}_{0}+p_{k-2},\tilde{p}_{-1},\boldsymbol{\eta} )
	\hat{\mathcal{V}}_{\alpha_{k+1}}(p_{k-2},\tilde{p}_{0}+p_{k-2},\boldsymbol{\eta} )
	\\
	\mathcal{P}_1 (p_{k-2},\boldsymbol{\tilde{p}},\boldsymbol{s}) = {}& e^{i  s_{k-1} \frac{1}{2} \tilde{p}_{-1}^2} 
	e^{i  s_{k} \frac{1}{2}( \tilde{p}_{0}+p_{k-2})^2}  e^{i  s_{k+1} \frac{1}{2}p_{k-2}^2}
	\end{aligned}
\end{equation*}
and for $\beta\geq2$ we have that
\begin{equation*}
	\begin{aligned}
	\tilde{\mathcal{G}}_\beta (p_{k-2},\boldsymbol{\tilde{p}},\boldsymbol{\eta}) &=
	\begin{cases}  
	 \hat{\mathcal{V}}_{\alpha_{k+\beta}}(p_{k-2},\tilde{p}_{\beta-1},\boldsymbol{\eta} )\tilde{\mathcal{G}}_{\beta-1} (\tilde{p}_{\beta-2},\boldsymbol{\tilde{p}},\boldsymbol{\eta}) & \text{$\beta$ even}
	\\
	\hat{\mathcal{V}}_{\alpha_{k+\beta}}(p_{k-2}, p_{k-2} + \tilde{p}_{\beta-1},\boldsymbol{\eta} )\tilde{\mathcal{G}}_{\beta-1} (p_{k-2}+ \tilde{p}_{\beta-1},\boldsymbol{\tilde{p}},\boldsymbol{\eta}) & \text{$\beta$ odd}
	\end{cases}
	\\
	\mathcal{P}_\beta (p_{k-2},\boldsymbol{\tilde{p}},\boldsymbol{s}) &=
	\begin{cases}  
	e^{i  s_{k+\beta} \frac{1}{2}p_{k-2}^2} \mathcal{P}_{\beta-1} (\tilde{p}_{\beta-2},\boldsymbol{\tilde{p}},\boldsymbol{s}) & \text{$\beta$ even}
	\\
	e^{i  s_{k+\beta} \frac{1}{2}p_{k-2}^2} \mathcal{P}_{\beta-1} (p_{k-2}+ \tilde{p}_{\beta-1},\boldsymbol{\tilde{p}},\boldsymbol{s}) & \text{$\beta$ odd}.
	\end{cases}
	\end{aligned}
\end{equation*}
As in the previous proofs we can now argue as in the proofs of  of Lemma~\ref{LE:Exp_ran_phases} and Lemma~\ref{expansion_aver_bound_Mainterm} and obtain that
\begin{equation*}
	\begin{aligned}
 	\mathbb{E}&\big[ \mathcal{T}(\beta,k,\sigma,\alpha,\tilde{\alpha};\hbar)\big] =   \frac{(2\pi)^{(2k-2-n)d}\hbar (\rho\hbar)^{2k-2-n}}{(2\pi\hbar)^d}  \int_{\R_{+}^{|\alpha|+|\tilde{\alpha}|}} f(\boldsymbol{s},\boldsymbol{\tilde{s}}) \int \Lambda_n(\boldsymbol{p},\boldsymbol{q},\sigma)    \hat{\varphi}(\tfrac{p_0}{\hbar}) \overline{\hat{\varphi}(\tfrac{q_0}{\hbar})} 
	\\
	\times & \delta(\tfrac{t_0}{\hbar}- \boldsymbol{s}_{0,k+\beta}^{+}-\tilde{s}_{k+\beta+1}- \boldsymbol{t}_{1,a_{k+\beta}}^{+})  \delta(\tfrac{t_0}{\hbar}- \boldsymbol{\tilde{s}}_{1,k+\beta+1}^{+}- \boldsymbol{\tilde{t}}_{1,a_{k+\beta}}^{+})  \mathcal{P}_\beta (p_{k-2},\boldsymbol{\tilde{p}},\boldsymbol{s})   \mathcal{P}_\beta (q_{k-2},\boldsymbol{\tilde{q}},\boldsymbol{\tilde{s}})  
	\\
	\times & e^{i  s_{k-2} \frac{1}{2} (\tilde{p}_{-2}+ \tilde{p}_{-1})^2}  e^{-i  \tilde{s}_{k-2} \frac{1}{2} (\tilde{q}_{-2}+ \tilde{q}_{-1})^2}     e^{i  s_{k-3} \frac{1}{2} (p_{k-3}+\tilde{p}_\beta+ \pi_{\beta}(\boldsymbol{\tilde{p}}))^2} e^{-i  \tilde{s}_{k-3} \frac{1}{2} (q_{k-3}+\tilde{q}_\beta+ \pi_{\beta}(\boldsymbol{\tilde{q}}))^2} \prod_{i=1}^{a_{k+\beta}} e^{i  t_{i} \frac{1}{2} \eta_{i}^2}  
	\\
	\times &  
	\prod_{i=1}^{\tilde{a}_{k+\beta}} e^{-i  \tilde{t}_{i} \frac{1}{2} \xi_{i}^2}  
	\prod_{m=0}^{k_1-4}  e^{i  s_{m} \frac{1}{2}( p_{m}+\tilde{\pi}_m(\tilde{p}_\beta))^2} e^{-i  \tilde{s}_{m} \frac{1}{2}( q_{m}+\tilde{\pi}_m(\tilde{q}_\beta))^2}  \mathcal{G}(\boldsymbol{p}, \boldsymbol{\tilde{p}},\boldsymbol{\tilde{q}},\boldsymbol{\eta}, \boldsymbol{\xi},\sigma ) d\boldsymbol{\eta} d\boldsymbol{\xi}d\boldsymbol{p}d\boldsymbol{q} d\boldsymbol{\tilde{p}}  d\boldsymbol{\tilde{q}}  d\boldsymbol{t} d\boldsymbol{\tilde{t}}  d\boldsymbol{s} d\boldsymbol{\tilde{s}},
	\end{aligned}
\end{equation*}
where we have preformed the change of variables $\tilde{s}_{k+\beta+1}\mapsto\hbar^{-1} \tilde{s}_{k+\beta+1}$. The function $f$ is given by
\begin{equation*}
	f(\boldsymbol{s},\boldsymbol{\tilde{s}}) = 
	\begin{cases}
	\boldsymbol{1}_{[0,\hbar^{-1}t]}\big( \sum_{i=1}^{k-4} s_i + s_{k+\beta} \big) \boldsymbol{1}_{[0,\hbar^{-1}t]}\big( \sum_{i=1,i \notin \sigma}^{k-3} \tilde{s}_i +  \tilde{s}_{k+\beta} + \tilde{s}_{k+\beta+1}  \big) & \text{if $k-2\notin\sigma $}
	\\
	\boldsymbol{1}_{[0,\hbar^{-1}t]}\big( \sum_{i=1}^{k-4} s_i + s_{k+\beta}  \big) \boldsymbol{1}_{[0,\hbar^{-1}t]}\big( \sum_{i=1,i \notin \sigma}^{k-3} \tilde{s}_i   +  \tilde{s}_{k+\beta+1}  \big) & \text{if $k-2\in\sigma$}.
	\end{cases}
\end{equation*}
The function $\mathcal{G}$ is given by
\begin{equation*}
	\begin{aligned}
	\MoveEqLeft \mathcal{G}(\boldsymbol{p}, \boldsymbol{\tilde{p}},\boldsymbol{\tilde{p}},\boldsymbol{\eta}, \boldsymbol{\xi},\sigma ) =  
	   \hat{V}(p_{k-1}-p_{k-2}) \overline{ \hat{V}(q_{k-1}-q_{k-2})} \tilde{\mathcal{G}}_\beta (\tilde{p}_\beta+p_{k-2},\boldsymbol{\tilde{p}},\boldsymbol{\eta}) \overline{  \tilde{\mathcal{G}}_\beta (\tilde{q}_\beta+q_{k-2},\boldsymbol{\tilde{q}},\boldsymbol{\xi}) }  
	     \\
	     \times&
	     \hat{\mathcal{V}}_{\alpha_{k-2}}(\tilde{p}_{-2}+ \tilde{p}_{-1},p_{k-3}+\tilde{p}_\beta+ \pi_{\beta}(\boldsymbol{\tilde{p}}),\boldsymbol{\eta} )
	     \overline{  \hat{\mathcal{V}}_{\tilde{\alpha}_{k-2}}(\tilde{q}_{-2}+ \tilde{q}_{-1},q_{k-3}+\tilde{q}_\beta+ \pi_{\beta}(\boldsymbol{\tilde{q}}),\boldsymbol{\xi} )}
	   \\
	   \times&\hat{\mathcal{V}}_{\alpha_{k-3}}(p_{k-3}+\tilde{p}_\beta+ \pi_{\beta}(\boldsymbol{\tilde{p}}),p_{k-4},\boldsymbol{\eta} )  
	    \overline{ \hat{\mathcal{V}}_{\tilde{\alpha}_{k-3}}(q_{k-3}+\tilde{q}_\beta+ \pi_{\beta}(\boldsymbol{\tilde{q}}),q_{k-4},\boldsymbol{\xi} ) }
	   \\
	   \times&
	 \prod_{m=1}^{k-4}    \hat{\mathcal{V}}_{\alpha_m}(p_m+\tilde{\pi}_m(\tilde{p}_\beta),p_{m-1}+\tilde{\pi}_{m-1}(\tilde{p}_\beta),\boldsymbol{\eta} ) \overline{ \hat{\mathcal{V}}_{\tilde{\alpha}_m}(q_m+\tilde{\pi}_m(\tilde{q}_\beta),q_{m-1}+\tilde{\pi}_{m-1}(\tilde{q}_\beta),\boldsymbol{\xi} ) }.
	\end{aligned}
\end{equation*}
  The function $\Lambda_n(\boldsymbol{p},\boldsymbol{q},\sigma)$ is given by
  \begin{equation*}
  	\Lambda_n(\boldsymbol{p},\boldsymbol{q},\sigma) = \delta(p_{k-1}-q_{k-1})  \prod_{i=1}^n \delta(p_{\sigma_{i-1}}- q_{\sigma_{i-1}} - p_{\sigma_n}+ q_{\sigma_n} ))   
	\prod_{i=1}^{n+1} \prod_{m=\sigma_{i-1}+1}^{\sigma_{i}-1} \delta(p_m-p_{\sigma_{i-1}}) \delta(q_m-q_{\sigma^1_{i-1}}).
  \end{equation*}
  To estimate these averages we can use the same techniques as in the previous proofs. However we will split into the two different cases $\sigma_n<k-1$ and $\sigma_n=k-1$. We start with the case where $\sigma_n=k-1$. Here we obtain the bound
  \begin{equation*}	
  \begin{aligned}
 	\MoveEqLeft  
	\big| \Aver{\mathcal{T}(\beta,k,\sigma,\alpha,\tilde{\alpha};\hbar)}\big|
	 \leq \hbar^2    \frac{C^{|\alpha|+|\tilde{\alpha}|+2} \rho (\rho t_0)^{2k-n-3} }{ (k-3)!(k-n)!}  
	  |\log(\tfrac{\hbar}{t})|^{n+20}  \norm{\varphi}_{\mathcal{H}^{5d+5}_\hbar(\R^d)}^2 \norm{\hat{V}}_{1,\infty,6d+6}^{|\alpha|+|\tilde{\alpha}|+2},
	\end{aligned}
\end{equation*}
  where we again have used that $\beta\leq6$. Recalling that $t_0=\frac{t}{\tau_0}$ and observing that for all possible values of $k$ and $n$ we have that $2k-n-3\geq 2$ we get that  
    \begin{equation}\label{EQ:expansion_aver_bound_3.recol_22_3.3_5}
	\begin{aligned}
 	\MoveEqLeft  
	\big| \Aver{\mathcal{T}(\beta,k,\sigma,\alpha,\tilde{\alpha};\hbar)}\big|
	 \leq \frac{\hbar^2}{\tau_0^2}    \frac{C^{|\alpha|+|\tilde{\alpha}|+2} \rho (\rho t)^{2k-n-3} }{ (k-3)!(k-n)!}  
	  |\log(\tfrac{\hbar}{t})|^{n+20}  \norm{\varphi}_{\mathcal{H}^{5d+5}_\hbar(\R^d)}^2 \norm{\hat{V}}_{1,\infty,6d+6}^{|\alpha|+|\tilde{\alpha}|+2}.
	\end{aligned}
\end{equation}
  Note that for any possible value of $k$ and $n$ we always have the integrals in $s_{k+\beta}$ and $\tilde{s}_{k+\beta+1}$.\footnote{This obersevation is crucial when we want to generalise this to the case where we consider time division. Since we can here have cases where we only have that it is the variables $s_{k+\beta}$, $\tilde{s}_{k+\beta}$ and $\tilde{s}_{k+\beta+1}$ that lives in a small time interval. } 
  For the case where $\sigma_n<k-1$ we argue as in the previous proofs and obtain the bound
  \begin{equation}\label{EQ:expansion_aver_bound_3.recol_22_3.3_6}	
  \begin{aligned}
 	\MoveEqLeft  
	\big| \Aver{\mathcal{T}(\beta,k,\sigma,\alpha,\tilde{\alpha};\hbar)}\big|
	 \leq \hbar^3    \frac{C^{|\alpha|+|\tilde{\alpha}|+2} \rho^2 (\rho t_0)^{2k-n-4} }{ (k-3)!(k-1-n)!}  
	  |\log(\tfrac{\hbar}{t})|^{n+20}  \norm{\varphi}_{\mathcal{H}^{5d+5}_\hbar(\R^d)}^2 \norm{\hat{V}}_{1,\infty,6d+6}^{|\alpha|+|\tilde{\alpha}|+2}.
	\end{aligned}
\end{equation}
From combining \eqref{EQ:expansion_aver_bound_3.recol_22_3.3_5} and \eqref{EQ:expansion_aver_bound_3.recol_22_3.3_6} we obtain for any configuration $\sigma$ in this case the bound
    \begin{equation}\label{EQ:expansion_aver_bound_3.recol_22_3.3_7}
	\begin{aligned}
 	\MoveEqLeft  
	\big| \Aver{\mathcal{T}(\beta,k,\sigma,\alpha,\tilde{\alpha};\hbar)}\big|
	 \leq \frac{\hbar^2}{\tau_0^2}    \frac{C^{|\alpha|+|\tilde{\alpha}|+2} \rho^2 (\rho t)^{2k-n-4} }{ (k-3)!(k-1-n)!}  
	  |\log(\tfrac{\hbar}{t})|^{n+20}  \norm{\varphi}_{\mathcal{H}^{5d+5}_\hbar(\R^d)}^2 \norm{\hat{V}}_{1,\infty,6d+6}^{|\alpha|+|\tilde{\alpha}|+2}.
	\end{aligned}
\end{equation}
Then from combining \eqref{EQ:expansion_aver_bound_3.recol_22_3.3_0}, \eqref{EQ:expansion_aver_bound_3.recol_22_3.3_1} and \eqref{EQ:expansion_aver_bound_3.recol_22_3.3_7} we obtain that
\begin{equation*}
	\begin{aligned}
	 \mathbb{E}\Big[\big\lVert  \sum_{k=4}^{k_0}   \mathcal{I}^{\mathrm{dob}}_3(\beta,k;\hbar)  \varphi\big\rVert_{L^2(\R^d)}^2 \Big]
	\leq  \tau_0^{-3} \tilde{\lambda}^\beta  \norm{\varphi}_{\mathcal{H}^{5d+5}_\hbar(\R^d)}^2 k_0^{5} C^{k_0}  |\log(\tfrac{\hbar}{t})|^{k_0+20}.
	\end{aligned}
\end{equation*}
This concludes the proof for the case of $\beta\leq 6$. We now turn to the case $\beta>6$ again for notational convenience we will assume $\beta$ to be even and in the end of the proof mention what we will need to change for $\beta$ odd. From the definition of the operator we have that
\begin{equation}\label{EQ:expansion_aver_bound_3.recol_22_3.3_8}
	\begin{aligned}
 	\MoveEqLeft  \mathcal{\tilde{I}}^{\mathrm{dob}}_3(\beta,k,\boldsymbol{x},s_{k+\beta+1};\hbar) \varphi (x)
	=  \frac{1}{(2\pi\hbar)^d} \sum_{\alpha\in \N^{k+\beta}} (i\lambda)^{|\alpha|} \int_{\R_{+}^{|\alpha|}} \boldsymbol{1}_{[0,\hbar^{-1}t_0]}( \boldsymbol{s}_{1,k+\beta}^{+}+ \hbar^{-1}s_{k+\beta+1}+ \boldsymbol{t}_{1,a_{k+\beta}}^{+})
	\\
	\times& \int_{\R^{(|\alpha|+1)d}}
	  e^{ i   \langle  \hbar^{-1}x,p_{k+\beta+1} \rangle} e^{ -i   \langle  \hbar^{-1/d}x_{k-1},p_{k+\beta+1}-p_{k+\beta} \rangle} e^{i \hbar^{-1} s_{k+\beta+1} \frac{1}{2} p_{k+\beta+1}^2} \hat{V}(p_{k+\beta+1}-p_{k+\beta})
	 \\
	 \times&    
	e^{ -i   \langle  \hbar^{-1/d}(x_{k-2}-x_{k-3}),\sum_{m=0}^{\beta/2+1} p_{2m+k-2}-p_{2m+k-3} \rangle} e^{ -i   \langle  \hbar^{-1/d}x_{k-3},p_{\beta+k}-p_{k-4} \rangle}   \prod_{m=1}^{k-4}  e^{ -i   \langle  \hbar^{-1/d}x_{m},p_{m}-p_{m-1} \rangle}
	\\
	\times&
	 \prod_{i=1}^{a_{k+\beta}} e^{i  t_{i} \frac{1}{2} \eta_{i}^2}       
	  \Big\{\prod_{m=1}^{k+\beta}  
	e^{i  s_{m} \frac{1}{2} p_{m}^2}   \hat{\mathcal{V}}_{\alpha_m}(p_m,p_{m-1},\boldsymbol{\eta} )  \Big\}e^{i ( \hbar^{-1}t- \boldsymbol{s}_{1,k}^{+}- \boldsymbol{t}_{1,a_{k+\beta}}^{+}) \frac{1}{2} p_0^2} \hat{\varphi}(\tfrac{p_0}{\hbar}) \, d\boldsymbol{\eta}d\boldsymbol{p}  d\boldsymbol{t}d\boldsymbol{s}.
	\end{aligned}
\end{equation}
As in the proof of Lemma~\ref{expansion_aver_bound_Mainterm_rec_trun22_special_case_1} we again have to distinguish the three cases  $k-2\notin\sigma$, $k-3,k-2\in\sigma$ and finally the case $k-2\in\sigma$ and $k-3\notin\sigma$. We start with the first case $k-2\notin\sigma$. After taking the average we star by preforming the change of variables $x_{k-2}\mapsto x_{k-2}-x_{k-3}$ and   $\tilde{x}_{k-2}\mapsto \tilde{x}_{k-2}-\tilde{x}_{k-3}$. Then  by arguing as in the proof of Lemma~\ref{LE:Exp_ran_phases} and Lemma~\ref{expansion_aver_bound_Mainterm} we have that   
\begin{equation}\label{EQ:expansion_aver_bound_3.recol_22_3.3_9}
	\begin{aligned}
 	\MoveEqLeft  
	\Aver{ \mathcal{T}(\beta,k,\sigma,\alpha,\tilde{\alpha};\hbar)}
	=  \frac{(2\pi)^{(2k-2-n)d}\hbar(\rho\hbar)^{2k-2-n}}{(2\pi\hbar)^d} \int_{\R_{+}^{|\alpha|+|\tilde{\alpha}|+2}} f(\boldsymbol{s},\boldsymbol{\tilde{s}}) \int  \Lambda_n(\boldsymbol{p},\boldsymbol{q},\sigma)    \mathcal{G}(\boldsymbol{p}, \boldsymbol{\tilde{p}},\boldsymbol{\tilde{q}},\boldsymbol{\eta}, \boldsymbol{\xi},  \tilde{\eta},\tilde{\xi},\sigma )  
	 \\
	 &\times   \delta(\tfrac{t}{\hbar}- \boldsymbol{\tilde{s}}_{0,k+\beta+1}^{+}- \boldsymbol{\tilde{t}}_{1,a_{k+\beta}}^{+})   \delta(\tfrac{t}{\hbar}- \boldsymbol{s}_{0,k+\beta}^{+}-\tilde{s}_{k+\beta+1}- \boldsymbol{t}_{1,a_{k+\beta}}^{+})  \prod_{i=1}^{a_{k+\beta}} e^{i  t_{i} \frac{1}{2} \eta_{i}^2} \prod_{i=1}^{\tilde{a}_{k+\beta}}  e^{-i  \tilde{t}_{i} \frac{1}{2} \xi_{i}^2}
	 \\
	 &\times e^{ -i   \langle  x, p_{k-3}-\tilde{p}_{\beta-1} +\sum_{m=0}^{\beta/2} \tilde{p}_{2m-2}-\tilde{p}_{2m-3} \rangle}   e^{ i   \langle  \tilde{x}, q_{k-3}-\tilde{q}_{\beta-1} +\sum_{m=0}^{\beta/2} \tilde{q}_{2m-2}-\tilde{q}_{2m-3} \rangle}  e^{i  s_{k+\beta} \frac{1}{2} p_{k-3}^2}   e^{ - i \tilde{s}_{k+\beta} \frac{1}{2} q_{k-3}^2}
	\\
	&\times    
	      \prod_{m=-3}^{\beta-1}e^{i  s_{m+k} \frac{1}{2} \tilde{p}_{m}^2} e^{-i  \tilde{s}_{m+k} \frac{1}{2} \tilde{q}_{m}^2}   
	 \prod_{m=0}^{k-4} e^{i  s_{m} \frac{1}{2} p_{m}^2}  e^{-i\tilde{s}_{m} \frac{1}{2}q_{m}^2}  |\hat{\varphi}(\tfrac{p_0}{\hbar})|^2
	  \, d\boldsymbol{\eta} d\boldsymbol{\xi} dx d\tilde{x}d\tilde{\eta}d\tilde{\xi} d\boldsymbol{p} d\boldsymbol{\tilde{p}} d\boldsymbol{\tilde{q}}   d\boldsymbol{t}  d\boldsymbol{\tilde{t}}  d\boldsymbol{\tilde{s}}d\boldsymbol{s},
	\end{aligned}
\end{equation}
where the function $f$ is given by
\begin{equation}\label{EQ:expansion_aver_bound_3.recol_22_3.3_10}
	f(\boldsymbol{s},\boldsymbol{\tilde{s}}) =
	\begin{cases}
	 \boldsymbol{1}_{[0,\hbar^{-1}t]}\big( \sum_{i=1}^{k-4} s_i +s_{k+\beta}  \big) \boldsymbol{1}_{[0,\hbar^{-1}t]}\big( \sum_{i=1,i \notin \sigma}^{k-4} \tilde{s}_i +\tilde{s}_{k+\beta} + \tilde{s}_{k+\beta+1}  \big) &\text{if $k-3\notin\sigma$}
	 \\
	  \boldsymbol{1}_{[0,\hbar^{-1}t]}\big( \sum_{i=1}^{k-4} s_i +s_{k+\beta}  \big) \boldsymbol{1}_{[0,\hbar^{-1}t]}\big( \sum_{i=1,i \notin \sigma}^{k-4} \tilde{s}_i + \tilde{s}_{k+\beta+1}  \big) &\text{if $k-3\in\sigma$},
	 \end{cases}
\end{equation}
and the function $\mathcal{G}$ is given by
\begin{equation}\label{EQ:expansion_aver_bound_3.recol_22_3.3_11}
	\begin{aligned}
	& \mathcal{G}(\boldsymbol{p}, \boldsymbol{\tilde{p}},\boldsymbol{\tilde{q}},\boldsymbol{\eta}, \boldsymbol{\xi},  \tilde{\eta},\tilde{\xi},\sigma )  
	=\hat{V}(p_{k-2}-p_{k-3})  \overline{\hat{V}(q_{k-2}-q_{k-3})} 
	  \hat{\mathcal{V}}_{\alpha_{\beta+k}}(p_{k-3}, \tilde{p}_{\beta-1},\boldsymbol{\eta} )  
	 \overline{\hat{\mathcal{V}}_{\tilde{\alpha}_{\beta+k}}(q_{k-3}, \tilde{q}_{\beta-1},\boldsymbol{\xi} )} 
	\\
	&\times\prod_{m=-2}^{\beta-1}   \Big\{
	 \hat{\mathcal{V}}_{\alpha_{m+k}}(\tilde{p}_m,\tilde{p}_{m-1},\boldsymbol{\eta} ) \overline{ \hat{\mathcal{V}}_{\tilde{\alpha}_{m+k}}(\tilde{q}_m,\tilde{q}_{m-1},\boldsymbol{\xi} )}  \Big\}
	  \hat{\mathcal{V}}_{\alpha_{k-3}}(\tilde{p}_{-3},p_{k-4},\boldsymbol{\eta} )   
	  \overline{\hat{\mathcal{V}}_{\tilde{\alpha}_{k-3}}(\tilde{q}_{-3},q_{k-4},\boldsymbol{\eta} )  }
	 \\
	 &\times \prod_{m=1}^{k-3}  
	   \hat{\mathcal{V}}_{\alpha_m}(p_m,p_{m-1},\boldsymbol{\eta} )\overline{ \hat{\mathcal{V}}_{\tilde{\alpha}_m}(q_m,q_{m-1},\boldsymbol{\xi} )}.  
	\end{aligned}
\end{equation}
and the function $\Lambda_n$ will, due to our relabelling, depend on $\sigma$ in the following way, where we stress that we have assumed $k-2\notin\sigma$. If we have that $\sigma_n\leq k-3$ we get that
  \begin{equation}\label{EQ:expansion_aver_bound_3.recol_22_3.3_11.1}
  	\begin{aligned}
  	\Lambda_n(\boldsymbol{p},\boldsymbol{q},\sigma) 
	={}& \delta(p_{k-2}-q_{k-2})  \prod_{i=1}^n \delta(p_{\sigma_{i-1}}- q_{\sigma_{i-1}} - p_{\sigma_n}+ q_{\sigma_n} )   
	\prod_{i=1}^{n+1} \prod_{m=\sigma_{i-1}+1}^{\sigma_{i}-1} \delta(p_m-p_{\sigma_{i-1}}) \delta(q_m-q_{\sigma_{i-1}})
	\\
	&\times \prod_{m=\sigma_{n}+1}^{k-2} \delta(p_m-p_{\sigma_{n}}) \delta(q_m-q_{\sigma_{n}}).
	\end{aligned}
  \end{equation}
  In the case where $\sigma_n=k-1$ we get that
    \begin{equation}\label{EQ:expansion_aver_bound_3.recol_22_3.3_11.2}
  	\begin{aligned}
  	 \Lambda_n(\boldsymbol{p},\boldsymbol{q},\sigma) 
	={}&
	\delta(p_{k-2}-q_{k-2})  \prod_{i=1}^n \delta(p_{\sigma_{i-1}}- q_{\sigma_{i-1}} )   
	\prod_{i=1}^{n-1} \prod_{m=\sigma_{i-1}+1}^{\sigma_{i}-1} \delta(p_m-p_{\sigma_{i-1}}) \delta(q_m-q_{\sigma_{i-1}}). 
	\end{aligned}
  \end{equation}
For any of the configuration $\sigma$ we can argue as in the proof of Lemma~\ref{expansion_aver_bound_Mainterm_rec_trun22_special_case_1} and obtain the bound
\begin{equation}\label{EQ:expansion_aver_bound_3.recol_22_3.3_12}
	\begin{aligned}
 	\MoveEqLeft  
	\big| \Aver{ \mathcal{T}(\beta,k,\sigma,\alpha,\tilde{\alpha};\hbar)}\big|
	 \leq \hbar^3    \frac{C^{|\alpha|+|\tilde{\alpha}|+2} \rho^3 (\rho t)^{2k-n-5} }{ (k-3)!(k-2-n)!}  
	  |\log(\tfrac{\hbar}{t})|^{n+4}  \norm{\varphi}_{\mathcal{H}^{d+1}_\hbar(\R^d)}^2 \norm{\hat{V}}_{1,\infty,3d+3}^{|\alpha|+|\tilde{\alpha}|+2},
	\end{aligned}
\end{equation}
where we have used that $t_0<t $. For the case where we have $k-2\in\sigma$ and $k-3\notin\sigma$ we switch the roles of $x_{k-2}$ and $x_{k-3}$ in the above arguments as done in the proofs of Lemma~\ref{expansion_aver_bound_Mainterm_rec_trun22_special_case_1} and obtain for any $\sigma$ the bound
\begin{equation}\label{EQ:expansion_aver_bound_3.recol_22_3.3_13}
	\begin{aligned}
 	\MoveEqLeft  
	\big| \Aver{ \mathcal{T}(\beta,k,\sigma,\alpha,\tilde{\alpha};\hbar)}\big|
	 \leq \hbar^3    \frac{C^{|\alpha|+|\tilde{\alpha}|+2} \rho^3 (\rho t)^{2k-n-5} }{ (k-3)!(k-2-n)!}  
	  |\log(\tfrac{\hbar}{t})|^{n+4}  \norm{\varphi}_{\mathcal{H}^{d+1}_\hbar(\R^d)}^2 \norm{\hat{V}}_{1,\infty,3d+3}^{|\alpha|+|\tilde{\alpha}|+2}.
	\end{aligned}
\end{equation}
We now turn to the case where we have $k-2,k-3\in\sigma$.  After taking the average we star by preforming the change of variables $x_{k-2}\mapsto x_{k-2}-x_{k-3}$ and   $\tilde{x}_{k-2}\mapsto \tilde{x}_{k-2}-\tilde{x}_{k-3}$. Then  by arguing as in the proof of Lemma~\ref{LE:Exp_ran_phases} and Lemma~\ref{expansion_aver_bound_Mainterm} we have that   
\begin{equation*}
	\begin{aligned}
 	\MoveEqLeft  
	\Aver{ \mathcal{T}(\beta,k,\sigma,\alpha,\tilde{\alpha};\hbar)}
	=  \frac{(2\pi)^{(2k-2-n)d}\hbar(\rho\hbar)^{2k-2-n}}{(2\pi\hbar)^d} \int_{\R_{+}^{|\alpha|+|\tilde{\alpha}|+2}} f(\boldsymbol{s},\boldsymbol{\tilde{s}}) \int  \Lambda_n(\boldsymbol{p},\boldsymbol{q},\sigma)    \mathcal{G}(\boldsymbol{p}, \boldsymbol{\tilde{p}},\boldsymbol{\tilde{q}},\boldsymbol{\eta}, \boldsymbol{\xi},  \tilde{\eta},\tilde{\xi},\sigma )  
	 \\
	 &\times   \delta(\tfrac{t}{\hbar}- \boldsymbol{\tilde{s}}_{0,k+\beta+1}^{+}- \boldsymbol{\tilde{t}}_{1,a_{k+\beta}}^{+})   \delta(\tfrac{t}{\hbar}- \boldsymbol{s}_{0,k+\beta}^{+}-\tilde{s}_{k+\beta+1}- \boldsymbol{t}_{1,a_{k+\beta}}^{+})  \prod_{i=1}^{a_{k+\beta}} e^{i  t_{i} \frac{1}{2} \eta_{i}^2} \prod_{i=1}^{\tilde{a}_{k+\beta}}  e^{-i  \tilde{t}_{i} \frac{1}{2} \xi_{i}^2}
	 \\
	 &\times e^{ -i   \langle  x, p_{k-3}- q_{k-3}+\tilde{q}_{\beta-1}-\tilde{p}_{\beta-1} +\sum_{m=0}^{\beta/2} \tilde{p}_{2m-2}-\tilde{p}_{2m-3} -  \tilde{q}_{2m-2}+\tilde{q}_{2m-3} \rangle}     e^{i  s_{k+\beta} \frac{1}{2} p_{k-3}^2}   e^{ - i \tilde{s}_{k+\beta} \frac{1}{2} q_{k-3}^2}
	\\
	&\times    
	      \prod_{m=-3}^{\beta-1}e^{i  s_{m+k} \frac{1}{2} \tilde{p}_{m}^2} e^{-i  \tilde{s}_{m+k} \frac{1}{2} \tilde{q}_{m}^2}   
	 \prod_{m=0}^{k-4} e^{i  s_{m} \frac{1}{2} p_{m}^2}  e^{-i\tilde{s}_{m} \frac{1}{2}q_{m}^2}  |\hat{\varphi}(\tfrac{p_0}{\hbar})|^2
	  \, d\boldsymbol{\eta} d\boldsymbol{\xi} dx d\tilde{x}d\tilde{\eta}d\tilde{\xi} d\boldsymbol{p} d\boldsymbol{\tilde{p}} d\boldsymbol{\tilde{q}}   d\boldsymbol{t}  d\boldsymbol{\tilde{t}}  d\boldsymbol{\tilde{s}}d\boldsymbol{s},
	\end{aligned}
\end{equation*}
where the function $f$ is given by \eqref{EQ:expansion_aver_bound_3.recol_22_3.3_10}, the function $\mathcal{G}$ is given by \eqref{EQ:expansion_aver_bound_3.recol_22_3.3_11}, if $\sigma_n=k-2$ the function $\Lambda_n$ is given by \eqref{EQ:expansion_aver_bound_3.recol_22_3.3_11.1} and if $\sigma_n=k-1$ the function $\Lambda_n$ is given by \eqref{EQ:expansion_aver_bound_3.recol_22_3.3_11.2}. The estimates will be different for the two cases $\sigma_n=k-2$ and $\sigma_n=k-1$. For both cases we can argue as in the proof of Lemma~\ref{expansion_aver_bound_Mainterm_rec_trun22_special_case_1} to obtain the bounds. For the case where  $\sigma_n=k-2$ we obtain the bound 
\begin{equation}\label{EQ:expansion_aver_bound_3.recol_22_3.3_15}
	\begin{aligned}
 	\MoveEqLeft  
	\big| \Aver{ \mathcal{T}(\beta,k,\sigma,\alpha,\tilde{\alpha};\hbar)}\big|
	 \leq \hbar^3    \frac{C^{|\alpha|+|\tilde{\alpha}|+2} \rho^2  (\rho t)^{2k-n-4} }{ (k-3)!(k-1-n)!}  
	  |\log(\tfrac{\hbar}{t})|^{n+4}  \norm{\varphi}_{\mathcal{H}^{d+1}_\hbar(\R^d)}^2 \norm{\hat{V}}_{1,\infty,3d+3}^{|\alpha|+|\tilde{\alpha}|+2},
	\end{aligned}
\end{equation}
where we again have used that $t<t_0$. For the case where $\sigma_n=k-1$ we get that estimate 
\begin{equation}\label{EQ:expansion_aver_bound_3.recol_22_3.3_16}
	\begin{aligned}
	\big| \Aver{ \mathcal{T}(\beta,k,\sigma,\alpha,\tilde{\alpha};\hbar)}\big|
	& \leq \hbar^2    \frac{C^{|\alpha|+|\tilde{\alpha}|+2} \rho (\rho t_0)^{2k-n-3} }{ (k-3)!(k-n)!}  
	  |\log(\tfrac{\hbar}{t_0})|^{n+4}  \norm{\varphi}_{\mathcal{H}^{d+1}_\hbar(\R^d)}^2 \norm{\hat{V}}_{1,\infty,3d+3}^{|\alpha|+|\tilde{\alpha}|+2}
	   \\
	   &\leq  \frac{\hbar^2}{\tau_0^2}    \frac{C^{|\alpha|+|\tilde{\alpha}|+2} \rho (\rho t)^{2k-n-3} }{ (k-3)!(k-n)!}  
	  |\log(\tfrac{\hbar}{t})|^{n+4}  \norm{\varphi}_{\mathcal{H}^{d+1}_\hbar(\R^d)}^2 \norm{\hat{V}}_{1,\infty,3d+3}^{|\alpha|+|\tilde{\alpha}|+2},
	\end{aligned}
\end{equation}
where we have used that for any possible $k$ and $n$ we have that $2k-n-3\geq2$ and $t_0=\frac{t}{\tau_0}$. Again it is the integrals in the variables $s_{k+\beta}$ and $\tilde{s}_{k+\beta+1}$ that ensures we get the factor $\tau_0^{-2}$.
Now from combining the estimates in \eqref{EQ:expansion_aver_bound_3.recol_22_3.3_0}, \eqref{EQ:expansion_aver_bound_3.recol_22_3.3_1}, \eqref{EQ:expansion_aver_bound_3.recol_22_3.3_12}, \eqref{EQ:expansion_aver_bound_3.recol_22_3.3_13}, \eqref{EQ:expansion_aver_bound_3.recol_22_3.3_15} \eqref{EQ:expansion_aver_bound_3.recol_22_3.3_16} and by arguing as in the previous proofs we get that
\begin{equation*}
	\begin{aligned}
	 \mathbb{E}\Big[\big\lVert  \sum_{k=4}^{k_0}   \mathcal{I}^{\mathrm{dob}}_3(\beta,k;\hbar)  \varphi\big\rVert_{L^2(\R^d)}^2 \Big]
	\leq  \tau_0^{-3} \tilde{\lambda}^\beta  \norm{\varphi}_{\mathcal{H}^{d+1}_\hbar(\R^d)}^2 k_0^{5} C^{k_0} |\log(\tfrac{\hbar}{t})|^{k_0+4}
	\end{aligned}
\end{equation*}
for all $\beta>6$ and even. For $\beta>6$ and odd we obtain the same estimate by an analogous argument. The main difference is in the starting expression of the function $ \mathcal{\tilde{I}}^{\mathrm{dob}}_3(\beta,k,\boldsymbol{x},s_{k+\beta+1};\hbar) \varphi (x)$ as described in the proof of  Lemma~\ref{expansion_aver_bound_Mainterm_rec_trun22_special_case_1}. This concludes the proof.
\end{proof}
\begin{lemma}\label{expansion_aver_bound_Mainterm_rec_trun22_6}
Assume we are in the setting of Definition~\ref{def_recol_reminder}. Let $\varphi \in \mathcal{H}^{5d+5}_\hbar(\R^d)$. Then 
\begin{equation*}
	\begin{aligned}
	\MoveEqLeft \mathbb{E}\Big[\big\lVert  \sum_{k_1=1}^{k_0} \sum_{k_2=1}^{k_0}   \boldsymbol{1}_{\{k_1+k_2\geq4\}} 
	\\
	&\times 
	  \sum_{(\boldsymbol{x}_1,\boldsymbol{x}_2)\in \mathcal{X}_{\neq}^{k_1+k_2-2}}  \mathcal{E}_{2,1}^{\mathrm{rec}}(k_1,k_2,\boldsymbol{x}_1,\iota,\tfrac{t}{\tau_0};\hbar) \mathcal{I}_{2,1}^{\mathrm{rec}}(k_2,\boldsymbol{x}_2,\iota,\tfrac{\tau-1}{\tau_0}t;\hbar)\varphi\big\rVert_{L^2(\R^d)}^2 \Big]
	\\
	&\leq C \hbar k_0^{16} C^{k_0}   |\log(\tfrac{\hbar}{t})|^{k_0+11} \norm{\varphi}_{\mathcal{H}^{5d+5}_\hbar(\R^d)}^2 +C \tau_0^{-3}  \norm{\varphi}_{\mathcal{H}^{5d+5}_\hbar(\R^d)}^2 k_0^{16} C^{k_0}   |\log(\tfrac{\hbar}{t})|^{k_0+22} ,
	\end{aligned}
\end{equation*}
and
\begin{equation*}
	\begin{aligned}
	\MoveEqLeft \mathbb{E}\Big[\big\lVert  \sum_{k_1=1}^{k_0} \sum_{k_2=1}^{k_0}  \boldsymbol{1}_{\{k_1+k_2\geq4\}} 
	\\
	&\times 
	  \sum_{(\boldsymbol{x}_1,\boldsymbol{x}_2)\in \mathcal{X}_{\neq}^{k_1+k_2-2}} \mathcal{E}_{2,j}^{\mathrm{rec}}(k_1+1,k_2,(x_{2,k_2},\boldsymbol{x}_1),\iota,\tfrac{t}{\tau_0};\hbar) \mathcal{I}_{2,j}^{\mathrm{rec}}(k_2,\boldsymbol{x}_2,\iota,\tfrac{\tau-1}{\tau_0}t;\hbar)\varphi\big\rVert_{L^2(\R^d)}^2 \Big]
	\\
	&\leq C \hbar k_0^{16} C^{k_0}   |\log(\tfrac{\hbar}{t})|^{k_0+11} \norm{\varphi}_{\mathcal{H}^{5d+5}_\hbar(\R^d)}^2 +C \tau_0^{-3}  \norm{\varphi}_{\mathcal{H}^{5d+5}_\hbar(\R^d)}^2 k_0^{16} C^{k_0}   |\log(\tfrac{\hbar}{t})|^{k_0+22},
	\end{aligned}
\end{equation*}
where the constant $C$ depends on the single site potential $V$ and the coupling constant $\lambda$. In particular we have that the function is in  $L^2(\R^d)$ $\Pro$-almost surely. 
\end{lemma}
\begin{proof}
The proof is done by combining the arguments from the proof of Lemma~\ref{expansion_aver_bound_Mainterm_2}, Lemma~\ref{expansion_aver_bound_Mainterm_3} and Lemma~\ref{expansion_aver_bound_Mainterm_rec_trun22_5}. Again we will need to continue expand certain terms. For the terms where we have to continue the expansion we also have to prove ``new'' versions of  
 Lemma~\ref{expansion_aver_bound_Mainterm_rec_trun22_special_case_1}, Lemma~\ref{expansion_aver_bound_Mainterm_rec_trun22_special_case_2}, and  Lemma~\ref{expansion_aver_bound_Mainterm_rec_trun22_special_case_3}. This is done by combining the arguments in the original proof with the arguments in Lemma~\ref{expansion_aver_bound_Mainterm_2} and Lemma~\ref{expansion_aver_bound_Mainterm_3}.
\end{proof}
\section{Regularisation of main terms}\label{Sec:reg_main_op}
In this section we will prove some results that enable us to ``regularise'' the main terms in our expansion. The regularisation can be thought of as going to complex valued energies, which of cause does not make physical sense. But this will ensure that the resolvents become bounded operators. The main thing is that these bounds will be uniform in $\hbar$.
We will from here and in the following need to choose $k_0$ and $\tau_0$ depending on $\hbar$ in the statement of some of our Lemmas. To be precise we will choose them as
\begin{equation*}
	k_0 = \frac{|\log(\hbar)|}{10 \log(|\log(\hbar)|)} \qquad\text{and}\qquad \tau_0 = \hbar^{-\frac{1}{3}}.
\end{equation*}
With this choice of $k_0$ we have that
\begin{equation}
	\lim_{\hbar\rightarrow 0 } \hbar^{\frac13} C^{k_0} |\log(\hbar)|^{k_0+50} =0,
\end{equation}
where $C$ is some constant. The choice of $\tau_0$ ensure the factor $\hbar^{\frac13}$, when needed. We will use the convention that if we sum up to either $k_0$ or $\tau_0$ we mean the sum over all natural numbers less than or equal to $k_0$ or $\tau_0$ respectively. 
\begin{lemma}\label{regularise_expansion}
Assume we are in the same setting as in Definition~\ref{functions_for_exp_def}, let $\varphi \in L^2(\R^d)$ and $\varepsilon>0$ be given. Moreover let
 \begin{equation*}
	k_0 = \frac{|\log(\hbar)|}{10 \log(|\log(\hbar)|)} .
\end{equation*} 
Then there exists $\hbar_0,\gamma_0>0$ such that 
\begin{equation*}
\sup_{\hbar\in(0,\hbar_0]} \sum_{k=1}^{k_0}\mathbb{E} \bigg[\Big\lVert  \sum_{\boldsymbol{x}\in\mathcal{X}_{\neq}^{k}}  \mathcal{I}_{0,0}(k,\boldsymbol{x},t;\hbar)\varphi-\mathcal{I}_{0,0}^\gamma(k,\boldsymbol{x},t;\hbar)\varphi \Big\rVert_{L^2(\R^d)}^2	\bigg]^\frac{1}{2}  \leq \varepsilon,
\end{equation*}
for all  $\gamma\in[0,\gamma_0]$, where $\gamma_0$ and $\hbar_0$ are independent of each other. 
\end{lemma}
\begin{proof}
By definition,
\begin{equation*}
	\begin{aligned}
 	\MoveEqLeft  \sum_{\boldsymbol{x}\in\mathcal{X}_{\neq}^{k}}  \mathcal{I}_{0,0}(k,\boldsymbol{x},t;\hbar)\varphi(x) -\mathcal{I}_{0,0}^\gamma(k,\boldsymbol{x},t;\hbar)\varphi (x) 
	\\
	= {}& \frac{1}{(2\pi\hbar)^d\hbar^{k}} \sum_{\alpha \in \N^k} (i\lambda)^{|\alpha|} \sum_{\boldsymbol{x} \in \mathcal{X}_{\neq}^k} \int_{[0,t]_{\leq}^k} \int_{\R^{(k+1)d}}
	 e^{ i   \langle  \hbar^{-1}x,p_{k} \rangle} \prod_{m=1}^k  e^{ -i   \langle  \hbar^{-1/d}x_m,p_{m}-p_{m-1} \rangle}   e^{i  s_{m} \frac{1}{2}\hbar^{-1} (p_{m}^2-p_{m-1}^2)}
	\\
	&\times \Big\{\prod_{m=1}^k \Psi_{\alpha_m}(p_{m},p_{m-1}) - \prod_{m=1}^k \Psi_{\alpha_m}^\gamma(p_{m},p_{m-1})   \, \Big\}  e^{i  t \frac{1}{2}\hbar^{-1} p_0^2} \hat{\varphi}(\tfrac{p_0}{\hbar}) \,d\boldsymbol{p} d\boldsymbol{s},
	\end{aligned}
\end{equation*}
where we have omitted the dependece on $\hbar^{-1}(s_{m-1}-s_{m})$ and  $V$ from all $\Psi_{\alpha_m}$ and $\Psi_{\alpha_m}^\gamma$. For the difference of the products we will use the identity
\begin{equation}\label{def_of_prod}
	\prod_{m=1}^k  A_m - \prod_{m=1}^k B_m = \sum_{m=1}^k A_1\cdots A_{m-1}(A_m-B_m)B_{m+1}\cdots B_k,
\end{equation}
with the convention $A_0=B_{k+1}=1$. We then have that
\begin{equation}\label{splitting_fuct_cont_1}
	\begin{aligned}
 	\sum_{\boldsymbol{x}\in\mathcal{X}_{\neq}^{k}}  \mathcal{I}_{0,0}(k,\boldsymbol{x},t;\hbar)\varphi(x) -\mathcal{I}_{0,0}^\gamma(k,\boldsymbol{x},t;\hbar)\varphi (x) =    \sum_{\nu=1}^k \mathcal{I}_\nu^\gamma(k,t;\hbar)\varphi(x),
	\end{aligned}
\end{equation}
where
\begin{equation*}
	\begin{aligned}
 	\mathcal{I}_\nu^\gamma(k,t;\hbar)\varphi(x) = {}& \frac{1}{(2\pi\hbar)^d\hbar^{k}} \sum_{\alpha \in \N^k} (i\lambda)^{|\alpha|} \sum_{\boldsymbol{x} \in \mathcal{X}_{\neq}^k} \int_{[0,t]_{\leq}^k} \int
	 e^{ i   \langle  \hbar^{-1}x,p_{k} \rangle} \Big\{\prod_{m=1}^k  e^{ -i   \langle  \hbar^{-1/d}x_m,p_{m}-p_{m-1} \rangle}   
	\\
	&\times e^{i  s_{m} \frac{1}{2}\hbar^{-1} (p_{m}^2-p_{m-1}^2)}\tilde{\Psi}_{\alpha_m}^{\gamma,\nu}(p_{m},p_{m-1})   \, \Big\}  e^{i  t \frac{1}{2}\hbar^{-1} p_0^2} \hat{\varphi}(\tfrac{p_0}{\hbar}) \,d\boldsymbol{p} d\boldsymbol{s},
	\end{aligned}
\end{equation*}
with the notation
\begin{equation*}
	\tilde{\Psi}_{\alpha_m}^{\gamma,\nu}(p_{m},p_{m-1})  = \begin{cases}
	\Psi_{\alpha_m}(p_{m},p_{m-1}) &\text{if $m<\nu$}
	\\
	\Psi_{\alpha_m}(p_{m},p_{m-1}) -\Psi_{\alpha_m}^\gamma(p_{m},p_{m-1}) &\text{if $m=\nu$}
	\\
	\Psi_{\alpha_m}^\gamma(p_{m},p_{m-1}) &\text{if $m>\nu$}.
	\end{cases}
\end{equation*}
Note that for the cases where $\alpha_\nu=1$ we have that $\tilde{\Psi}_{\alpha_\nu}^{\gamma,\nu}(p_{\nu},p_{\nu-1}) =0$ for all $(p_{\nu},p_{\nu-1})$. Hence in the following we will only consider the cases where $\alpha_\nu>1$. Using \eqref{splitting_fuct_cont_1}  and the triangle equality for $L^2(\R^d\times\Omega,dx\otimes\mathbb{P})$-norm we see that
\begin{equation}\label{EQ:regularise_expansion_1}
\mathbb{E} \bigg[\Big\lVert  \sum_{\boldsymbol{x}\in\mathcal{X}_{\neq}^{k}}  \mathcal{I}_{0,0}(k,\boldsymbol{x},t;\hbar)\varphi-\mathcal{I}_{0,0}^\gamma(k,\boldsymbol{x},t;\hbar)\varphi \Big\rVert_{L^2(\R^d)}^2	\bigg]^\frac{1}{2}    \leq  \sum_{\nu=1}^k \mathbb{E} \Big[\big\lVert  \mathcal{I}_\nu^\gamma(k,t;\hbar)\varphi \big\rVert_{L^2(\R^d)}^2	\Big]^\frac{1}{2}  .
\end{equation}
Since the form of the function $\mathcal{I}_\nu^\gamma(k,\alpha;\hbar)\varphi(x)$ is essentially the same as $\sum_{\boldsymbol{x}\in\mathcal{X}_{\neq}^{k}}  \mathcal{I}_{0,0}(k,\boldsymbol{x},t;\hbar)\varphi(x)$ we will use an argument similar to that used in Lemma~\ref{expansion_aver_bound_Mainterm}. We again expand the $L^2$-norm and get
\begin{equation}\label{EQ:regularise_expansion_2}
	\big\lVert  \mathcal{I}_\nu^\gamma(k,t;\hbar)\varphi \big\rVert_{L^2(\R^d)}^2
	=\sum_{n=0}^k \sum_{\sigma^1,\sigma^2\in\mathcal{A}(k,n)} \sum_{\kappa\in\mathcal{S}_n}\sum_{\alpha,\tilde{\alpha}\in\N^k} \mathcal{T}(n,\alpha,\tilde{\alpha},\sigma^1,\sigma^2,\kappa),
\end{equation}
where the numbers $\mathcal{T}(n,\alpha,\tilde{\alpha},\sigma^1,\sigma^2,\kappa)$ are given by
\begin{equation*}
	\mathcal{T}(n,\alpha,\tilde{\alpha},\sigma^1,\sigma^2,\kappa)
	= \sum_{(\boldsymbol{x},\boldsymbol{\tilde{x}})\in \mathcal{X}_{\neq}^{2k}}  \prod_{i=1}^n \rho^{-1}\delta(x_{\sigma_i^1}- \tilde{x}_{\sigma_{\kappa(i)^2}})  \int_{\R^d} \mathcal{I}_\nu^\gamma (k,\boldsymbol{x},\alpha,t;\hbar) \varphi(x)\overline{ \mathcal{I}_\nu^\gamma(k,\boldsymbol{\tilde{x}},\tilde{\alpha},t;\hbar) \varphi(x)}  \,dx,
\end{equation*}
where again $\mathcal{I}_\nu^\gamma (k,\boldsymbol{x},\alpha,t;\hbar) \varphi(x)$ is defined as $ \mathcal{I}_\nu^\gamma(k,t;\hbar)\varphi (x)$ for a fixed $\boldsymbol{x}$ and $\alpha$. Again we consider the direct and crossing terms separately. For the direct terms we get from arguing as in the proof of Lemma~\ref{expansion_aver_bound_Mainterm} that
   \begin{equation}\label{EQ:regularise_expansion_3}
   	\begin{aligned}
	 |\Aver{\mathcal{T}(n,\alpha,\tilde{\alpha},\sigma^1,\sigma^2,\mathrm{id})}|
	 \leq {}&  C^{|\alpha|+|\tilde{\alpha}|}  \norm{\varphi}_{L^2(\R^d)}^2 \left(\lambda \norm{\hat{V}}_{1,\infty,3d+3}\right)^{|\alpha|+|\tilde{\alpha}|}  \frac{ (\rho t)^{2k-n}}{k!(k-n)!}  
	 \\
	 &\times  \int_{\R_{+}^{a_k}}   \frac{1-e^{-\gamma\frac{1}{2}\boldsymbol{t}_{1,\alpha_\nu-1}^{+}}}{\prod_{i=1}^{a_k}\max(1, | t_i |)^\frac{d}{2}}  d\boldsymbol{t} 
	  \int_{\R_{+}^{\tilde{a}_k}} \frac{1-e^{-\gamma\frac{1}{2}\boldsymbol{t}_{1,\tilde{\alpha}_\nu-1}^{+}}}{  \prod_{i=1}^{\tilde{a}_k}  \max(1, | t_i |)^\frac{d}{2}}  d\boldsymbol{t}.
	\end{aligned}
\end{equation}
The difference to the estimate obtained in Lemma~\ref{expansion_aver_bound_Mainterm} is the two times integrals we have to estimate differently in this case.  The two factors $1-e^{-\gamma\frac{1}{2}\boldsymbol{t}_{1,\alpha_\nu-1}^{+}}$ and $1-e^{-\gamma\frac{1}{2}\boldsymbol{t}_{1,\tilde{\alpha}_\nu-1}^{+}}$ is due to the definition of the functions $\tilde{\Psi}_{\alpha_m}^{\gamma,\nu}$. We have estimated all other exponentials with a $\gamma$ by $1$.  
Using the identity in \eqref{def_of_prod} and the fact that the exponentials are bounded by one, we get
  \begin{equation*}
	\begin{aligned}
	 (1-e^{-\gamma\frac{1}{2}\boldsymbol{t}_{1,\alpha_\nu-1}^{+}})(1-e^{-\gamma\frac{1}{2}\boldsymbol{t}_{1,\tilde{\alpha}_\nu-1}^{+}}) \leq \sum_{l_1=1}^{\alpha_\nu-1} \sum_{l_2=1}^{\tilde{\alpha}_\nu-1} (1-e^{-\gamma\frac{1}{2}t_{l_1}})(1- e^{-\gamma\frac{1}{2}t_{l_2+\alpha_\nu}}).
	\end{aligned}
\end{equation*} 
Using the inequality we have just obtained we have that 
\begin{equation}\label{EQ:regularise_expansion_4}
	\begin{aligned}
	\MoveEqLeft \int_{\R_{+}^{a_k}}   \frac{1-e^{-\gamma\frac{1}{2}\boldsymbol{t}_{1,\alpha_\nu-1}^{+}}}{\prod_{i=1}^{a_k}\max(1, | t_i |)^\frac{d}{2}}  d\boldsymbol{t} 
	  \int_{\R_{+}^{\tilde{a}_k}} \frac{1-e^{-\gamma\frac{1}{2}\boldsymbol{t}_{1,\tilde{\alpha}_\nu-1}^{+}}}{  \prod_{i=1}^{\tilde{a}_k}  \max(1, | t_i |)^\frac{d}{2}}  d\boldsymbol{t} 
	\\
	&\leq (\alpha_\nu-1)(\tilde{\alpha}_\nu-1)  \left( \frac{d}{d-2}\right)^{a_k+\tilde{a}_k-2} \left( \int_{\R_{+}} \frac{(1-e^{-\gamma\frac{1}{2}t})}{\max(1, | t |)^\frac{d}{2}}  dt \right)^2 \leq (\alpha_\nu-1)(\tilde{\alpha}_\nu-1)  C^{a_k+\tilde{a}_k-2}  \gamma^{\frac13}, 
	\end{aligned}
\end{equation}
where we in the last inequality have used Lemma~\ref{LE:gamma_reg_time_int}. With these estimates we get from combining  \eqref{EQ:regularise_expansion_3}, \eqref{EQ:regularise_expansion_4} that
   \begin{equation}\label{EQ:regularise_expansion_5}
   	\begin{aligned}
	 |\Aver{\mathcal{T}(n,\alpha,\tilde{\alpha},\sigma^1,\sigma^2,\mathrm{id})}|
	 \leq (\alpha_\nu-1)(\tilde{\alpha}_\nu-1)  \left(\lambda C_d\norm{\hat{V}}_{1,\infty,3d+3}\right)^{|\alpha|+|\tilde{\alpha}|}   \gamma^{\frac{1}{3}}  \norm{\varphi}_{L^2(\R^d)}^2   \frac{ (\rho t)^{2k-n}}{k!(k-n)!}  
	\end{aligned}
\end{equation} 
    for all $\gamma\in[0,\gamma_0]$. For the crossing terms we get the estimate
       \begin{equation}\label{EQ:regularise_expansion_6}
   	\begin{aligned}
	|\Aver{\mathcal{T}(n,\alpha,\tilde{\alpha},\sigma^1,\sigma^2,\kappa)}|
	 \leq      \left(C_d \lambda\norm{\hat{V}}_{1,\infty,3d+3}\right)^{|\alpha|+|\tilde{\alpha}|}     \frac{\rho(\rho t)^{2k-n-1} \hbar |\log(\tfrac{\hbar}{t})|^{n+3}}{(k-1)!(k-n)!} \norm{\varphi}_{\mathcal{H}^{2d+2}_\hbar(\R^d)}^2,
	\end{aligned}
\end{equation}
This estimate is obtained complete analogous to the estimate obtained in the proof of Lemma~\ref{expansion_aver_bound_Mainterm}.   
Combining the estimates in \eqref{EQ:regularise_expansion_2} \eqref{EQ:regularise_expansion_5} and \eqref{EQ:regularise_expansion_6} we get
\begin{equation}\label{EQ:regularise_expansion_7}
	\begin{aligned}
	\MoveEqLeft \mathbb{E}\Big[\big\lVert  \mathcal{I}_\nu^\gamma(k,t;\hbar)\varphi \big\rVert_{L^2(\R^d)}^2\Big]
	\leq \sum_{\alpha,\tilde{\alpha}\in\N^k} \norm{\varphi}_{\mathcal{H}^{2d+2}_\hbar(\R^d)}^2  \left(C_d\lambda\norm{\hat{V}}_{1,\infty,3d+3}\right)^{|\alpha|+|\tilde{\alpha}|}   
	\\
	&\times  
	\sum_{n=0}^k  \Big[  (\alpha_\nu-1)(\tilde{\alpha}_\nu-1) \gamma^{\frac13} \frac{ (\rho t)^{2k-n}}{k!(k-n)!} 	+ \hbar |\log(\tfrac{\hbar}{t})|^{n+3} \frac{\rho(\rho t)^{2k-n-1} n!}{(k-1)!(k-n)!} \Big],
	\end{aligned}
\end{equation}
where we have used that the number of elements in $\mathcal{A}(k,n)$ is bounded by $2^k$ and that the number of elements in $\mathcal{S}_n$ is $n!$. Evaluating the sums in $\alpha$ and $\tilde{\alpha}$ and combing the estimates in \eqref{EQ:regularise_expansion_1} and \eqref{EQ:regularise_expansion_7} we obtain that
\begin{equation*}
	\begin{aligned}
	\MoveEqLeft \sum_{k=1}^{k_0}\mathbb{E} \bigg[\Big\lVert  \sum_{\boldsymbol{x}\in\mathcal{X}_{\neq}^{k}}  \mathcal{I}_{0,0}(k,\boldsymbol{x},t;\hbar)\varphi-\mathcal{I}_{0,0}^\gamma(k,\boldsymbol{x},t;\hbar)\varphi \Big\rVert_{L^2(\R^d)}^2	\bigg]^\frac{1}{2}    
	\\
	&\leq C \norm{\varphi}_{\mathcal{H}^{2d+2}_\hbar(\R^d)} \sum_{k=1}^{k_0}  \sum_{\nu=1}^k \Big[  \sum_{n=0}^k \gamma^{\frac13} \frac{ (\rho t)^{2k-n}}{k!(k-n)!} 	+ \hbar |\log(\tfrac{\hbar}{t})|^{n+3} \frac{\rho(\rho t)^{2k-n-1} n!}{(k-1)!(k-n)!} \Big]^{\frac{1}{2}}
	\\
	&\leq C \norm{\varphi}_{\mathcal{H}^{2d+2}_\hbar(\R^d)} \sum_{k=1}^{k_0}  \sum_{\nu=1}^k \Big[  \gamma^{\frac13} \frac{ (\rho t)^{k} e^{\rho t}}{k!} 	+  C \hbar |\log(\tfrac{\hbar}{t})|^{k_0+3} \Big]^{\frac{1}{2}},
	\end{aligned}
\end{equation*}
where the constant $C$ depends on the potential and the coupling constant but is independent of $\gamma$ and $\hbar$. Using the subadditivity of the square root and that $\sum_{k=1}^{\infty} \frac{k (\sqrt{\rho t})^{k} }{\sqrt{k!}}$ is finite we get that
\begin{equation}\label{EQ:regularise_expansion_8}
	\begin{aligned}
	\MoveEqLeft \sum_{k=1}^{k_0}\mathbb{E} \bigg[\Big\lVert  \sum_{\boldsymbol{x}\in\mathcal{X}_{\neq}^{k}}  \mathcal{I}_{0,0}(k,\boldsymbol{x},t;\hbar)\varphi-\mathcal{I}_{0,0}^\gamma(k,\boldsymbol{x},t;\hbar)\varphi \Big\rVert_{L^2(\R^d)}^2	\bigg]^\frac{1}{2}    
\leq C \norm{\varphi}_{\mathcal{H}^{2d+2}_\hbar(\R^d)} \Big( \gamma^{\frac16} +\sqrt{ \hbar k_0^2 |\log(\tfrac{\hbar}{t})|^{k_0+3}} \Big). 
	\end{aligned}
\end{equation}
Since we have that $\sup_{\hbar\in I}\norm{\varphi}_{\mathcal{H}^{2d+2}_\hbar(\R^d)}$ is uniformly bounded and $\hbar k_0^2 |\log(\tfrac{\hbar}{t})|^{k_0+3}\rightarrow 0$ as $\hbar\rightarrow 0$ due to our assumptions on $k_0$ we get the existence of $\gamma_0$ and $\hbar_0$ depending on $\varepsilon$ from \eqref{EQ:regularise_expansion_8} such that
\begin{equation*}
\sup_{\hbar\in(0,\hbar_0]} \sum_{k=1}^{k_0}\mathbb{E} \bigg[\Big\lVert  \sum_{\boldsymbol{x}\in\mathcal{X}_{\neq}^{k}}  \mathcal{I}_{0,0}(k,\boldsymbol{x},t;\hbar)\varphi-\mathcal{I}_{0,0}^\gamma(k,\boldsymbol{x},t;\hbar)\varphi \Big\rVert_{L^2(\R^d)}^2	\bigg]^\frac{1}{2}  \leq \varepsilon,
\end{equation*}
for all  $\gamma\in[0,\gamma_0]$, where $\gamma_0$ and $\hbar_0$ are independent of each other. This concludes the proof.
\end{proof}
In the following we will use the operator $\mathcal{I}_{\infty}^\gamma(k,\boldsymbol{x},t;\hbar)$, which is defined analogous to $\mathcal{I}_{0,0}^\gamma(k,\boldsymbol{x},t;\hbar)$ but with $\Psi_{\alpha_m}^\gamma(p_{m},p_{m-1},\hbar^{-1}(s_{m-1}-s_m);V)$ replaced by $\Psi_{\alpha_m}^\gamma(p_{m},p_{m-1},\infty;V)$.
\begin{lemma}\label{regularise_expansion_in_time}
Assume we are in the same setting as in Definition~\ref{functions_for_exp_def}, let $\varphi \in L^2(\R^d)$ and let $\gamma>0$. Moreover let
 \begin{equation*}
	k_0 = \frac{|\log(\hbar)|}{10 \log(|\log(\hbar)|)} .
\end{equation*} 
Then 
\begin{equation*}
\lim_{\hbar\rightarrow 0}\sum_{k=1}^{k_0}     \mathbb{E} \bigg[\Big\lVert \sum_{\boldsymbol{x}\in\mathcal{X}^{k}_{\neq}}  \mathcal{I}_{0,0}^\gamma(k,\boldsymbol{x},t;\hbar) \varphi-\mathcal{I}_{\infty}^\gamma(k,\boldsymbol{x},t;\hbar)\varphi \Big\rVert_{L^2(\R^d)}^2	\bigg]^\frac{1}{2}  = 0.
\end{equation*}
\end{lemma}
\begin{proof}
The proof will be similar to the proof of Lemma~\ref{regularise_expansion}. Firstly by using the identity in \eqref{def_of_prod} we have that  
\begin{equation*}
	\begin{aligned}
 	\sum_{\boldsymbol{x}\in\mathcal{X}^{k}_{\neq}} \mathcal{I}_{0,0}^\gamma(k,\boldsymbol{x},t;\hbar) \varphi(x)-\mathcal{I}_{\infty}^\gamma(k,\boldsymbol{x},t;\hbar)\varphi(x)  
	=   \sum_{\nu=1}^k   \mathcal{I}_{\infty,\nu}^\gamma(k,t;\hbar)\varphi(x),  
	\end{aligned}
\end{equation*}
where
\begin{equation*}
	\begin{aligned}
 	\mathcal{I}_\nu^\gamma(k,t;\hbar)\varphi(x) = {}& \frac{1}{(2\pi\hbar)^d\hbar^{k}} \sum_{\alpha \in \N^k} (i\lambda)^{|\alpha|} \sum_{\boldsymbol{x} \in \mathcal{X}_{\neq}^k} \int_{[0,t]_{\leq}^k} \int
	 e^{ i   \langle  \hbar^{-1}x,p_{k} \rangle} \Big\{\prod_{m=1}^k  e^{ -i   \langle  \hbar^{-1/d}x_m,p_{m}-p_{m-1} \rangle}   
	\\
	&\times e^{i  s_{m} \frac{1}{2}\hbar^{-1} (p_{m}^2-p_{m-1}^2)}\tilde{\Psi}_{\alpha_m,\infty}^{\gamma,\nu}(p_{m},p_{m-1})   \, \Big\}  e^{i  t \frac{1}{2}\hbar^{-1} p_0^2} \hat{\varphi}(\tfrac{p_0}{\hbar}) \,d\boldsymbol{p} d\boldsymbol{s},
	\end{aligned}
\end{equation*}
with the notation
\begin{equation*}
	\begin{aligned}
	 \tilde{\Psi}_{\alpha_m,\infty}^{\gamma,\nu}(p_{m},p_{m-1})  
	= \begin{cases}
	\Psi_{\alpha_m}(p_{m},p_{m-1},\hbar^{-1}(s_{m-1}-s_m);V) &\text{if $m<\nu$}
	\\
	\Psi_{\alpha_m}(p_{m},p_{m-1},\hbar^{-1}(s_{m-1}-s_m);V) -\Psi_{\alpha_m}^\gamma(p_{m},p_{m-1},\infty;V) &\text{if $m=\nu$}
	\\
	\Psi_{\alpha_m}^\gamma(p_{m},p_{m-1},\infty;V) &\text{if $m>\nu$}.
	\end{cases}
	\end{aligned}
\end{equation*}
Note again that for the cases where $\alpha_\nu=1$ we have that $ \tilde{\Psi}_{\alpha_{\nu},\infty}^{\gamma,\nu}(p_{\nu},p_{\nu-1}) =0$ for all $(p_{\nu},p_{\nu-1})$. Hence in the following we will only consider the cases where $\alpha_\nu>1$. 
The triangle inequality for $L^2(\R^d\times\Omega,dx\otimes\mathbb{P})$-norm gives us that 
\begin{equation}\label{EQ:regularise_expansion_in_time_1}
 \mathbb{E} \bigg[\Big\lVert \sum_{\boldsymbol{x}\in\mathcal{X}^{k}_{\neq}}  \mathcal{I}_{0,0}^\gamma(k,\boldsymbol{x},t;\hbar) \varphi-\mathcal{I}_{\infty}^\gamma(k,\boldsymbol{x},t;\hbar)\varphi \Big\rVert_{L^2(\R^d)}^2	\bigg]^\frac{1}{2} \leq  \sum_{\nu=1}^k \mathbb{E} \Big[\big\lVert   \mathcal{I}_{\infty,\nu}^\gamma(k,t;\hbar)\varphi \big\rVert_{L^2(\R^d)}^2	\Big]^\frac{1}{2}  ,
\end{equation}
As in the previous proof we will use an argument similar to that used in Lemma~\ref{expansion_aver_bound_Mainterm}. We again expand the $L^2$-norm and get
\begin{equation}\label{EQ:regularise_expansion_in_time_2}
	\big\lVert  \mathcal{I}_{\infty,\nu}^\gamma(k,t;\hbar)\varphi \big\rVert_{L^2(\R^d)}^2
	=\sum_{n=0}^k \sum_{\sigma^1,\sigma^2\in\mathcal{A}(k,n)} \sum_{\kappa\in\mathcal{S}_n}\sum_{\alpha,\tilde{\alpha}\in\N^k} \mathcal{T}(n,\alpha,\tilde{\alpha},\sigma^1,\sigma^2,\kappa),
\end{equation}
where the numbers $\mathcal{T}(n,\alpha,\tilde{\alpha},\sigma^1,\sigma^2,\kappa)$ are given by
\begin{equation*}
	\begin{aligned}
	\MoveEqLeft \mathcal{T}(n,\alpha,\tilde{\alpha},\sigma^1,\sigma^2,\kappa)
	\\
	&= \sum_{(\boldsymbol{x},\boldsymbol{\tilde{x}})\in \mathcal{X}_{\neq}^{2k}} \prod_{i=1}^n \rho^{-1}\delta(x_{\sigma_i^1}- \tilde{x}_{\sigma_{\kappa(i)^2}})  \int_{\R^d} \mathcal{I}_{\infty,\nu}^\gamma (k,\boldsymbol{x},\alpha,t;\hbar) \varphi(x)\overline{ \mathcal{I}_{\infty,\nu}^\gamma(k,\boldsymbol{\tilde{x}},\tilde{\alpha},t;\hbar) \varphi(x)}  \,dx,
	\end{aligned}
\end{equation*}
where again $\mathcal{I}_{\infty,\nu}^\gamma (k,\boldsymbol{x},\alpha,t;\hbar) \varphi(x)$ is defined as $ \mathcal{I}_{\infty,\nu}^\gamma(k,t;\hbar)\varphi (x)$ for a fixed $\boldsymbol{x}$ and $\alpha$. Again we consider the direct and crossing terms separately. 
For the direct terms we get from arguing as in the proof of Lemma~\ref{expansion_aver_bound_Mainterm} that
   \begin{equation*}
   	\begin{aligned}
	\MoveEqLeft  \Aver{\mathcal{T}(n,\alpha,\tilde{\alpha},\sigma^1,\sigma^2,\mathrm{id})}
	 \leq    \norm{\varphi}_{L^2(\R^d)}^2 \left(C\lambda \norm{\hat{V}}_{1,\infty,3d+3}\right)^{|\alpha|+|\tilde{\alpha}|}  \rho^{2k-n}
	  \int_{[0,t]_{\leq}^k} \int_{[0,t]_{\leq}^k}  \prod_{i=1}^{n} \tfrac{1}{\hbar \max(1,\hbar^{-1} | s_{\sigma_i^1}- \tilde{s}_{\sigma_{i}^2} |)^\frac{d}{2}}
	\\
	\times&   \int_{\R_{+}^{a_k+\tilde{a}_k}}    \tfrac{\boldsymbol{1}_{[\hbar^{-1}(s_{\nu-1}-s_\nu),\infty)}(\boldsymbol{t}_{1,\alpha_\nu-1}^{+}) e^{-\gamma\frac{1}{2}\boldsymbol{t}_{1,\alpha_\nu-1}^{+}}}{\prod_{i=1}^{a_k}\max(1, | t_i |)^\frac{d}{2} }   \tfrac{  \boldsymbol{1}_{[\hbar^{-1}(\tilde{s}_{\nu-1}-\tilde{s}_\nu),\infty)}(\boldsymbol{\tilde{t}}_{1,\alpha_\nu-1}^{+}) e^{-\gamma\frac{1}{2}\boldsymbol{\tilde{t}}_{1,\alpha_\nu-1}^{+}}}{\prod_{i=1}^{\tilde{a}_k}\max(1, | \tilde{t}_i |)^\frac{d}{2}}  \, d\boldsymbol{t}   d\boldsymbol{s}d\tilde{\boldsymbol{s}},
	\end{aligned}
\end{equation*}
where we her have not estimate any time integrals yet. 
We now bound $\boldsymbol{1}_{[\hbar^{-1}(\tilde{s}_{\nu-1}-\tilde{s}_\nu),\infty)}(\boldsymbol{\tilde{t}}_{1,\alpha_\nu-1}^{+}) e^{-\gamma\frac{1}{2}\boldsymbol{\tilde{t}}_{1,\alpha_\nu-1}^{+}}$ by $1$ and evaluate all integrals in $\tilde{s}$ and $\tilde{t}$. We then get
   \begin{equation}\label{EQ:regularise_expansion_in_time_3}
   	\begin{aligned}
	\Aver{\mathcal{T}(n,\alpha,\tilde{\alpha},\sigma^1,\sigma^2,\mathrm{id})} \leq {}&  C^{|\alpha|+|\tilde{\alpha}|}   \norm{\varphi}_{L^2(\R^d)}^2 \left(\lambda \norm{\hat{V}}_{1,\infty,3d+3}\right)^{|\alpha|+|\tilde{\alpha}|}    \frac{ \rho^{2k-n} t^{k-n}}{(k-n)!}  
	 \\
	 &\times  \int_{[0,t]_{\leq}^k}   \int_{\R_{+}^{a_k}}    \frac{\boldsymbol{1}_{[\hbar^{-1}(s_{\nu-1}-s_\nu),\infty)}(\boldsymbol{t}_{1,\alpha_\nu-1}^{+}) e^{-\gamma\frac{1}{2}\boldsymbol{t}_{1,\alpha_\nu-1}^{+}}}{\prod_{i=1}^{a_k}\max(1, | t_i |)^\frac{d}{2} }
	  \, d\boldsymbol{t}  d\boldsymbol{s},
	\end{aligned}
\end{equation}
where we have absorbed all constants from the integration only depending on dimension in to the constant $C^{|\alpha|+|\tilde{\alpha}|} $.
For the remaining integrals in time we will divide them into two according to if $\gamma\hbar^{-1}(s_{\nu-1}-s_\nu)$  is smaller than one or greater than or equal to one. In the latter case we have that
  \begin{equation*}
	\begin{aligned}
	\boldsymbol{1}_{[\hbar^{-1}(s_{\nu-1}-s_\nu),\infty)}(\boldsymbol{t}_{1,\alpha_\nu}^{+})  e^{-\gamma\frac{1}{2}\boldsymbol{t}_{1,\alpha_\nu}^{+}} \leq \frac{C}{|\gamma \hbar^{-1}(s_{\nu-1}-s_\nu)|^d}
	\end{aligned}
\end{equation*} 
for all $\boldsymbol{t}_{1,\alpha_\nu}^{+}$ and in the first case we can use the trivial bound one. This yields 
  \begin{equation}\label{EQ:regularise_expansion_in_time_4}
	\begin{aligned}
	\MoveEqLeft   \int_{[0,t]_{\leq}^k}   \int_{\R_{+}^{a_k}}    \frac{\boldsymbol{1}_{[\hbar^{-1}(s_{\nu-1}-s_\nu),\infty)}(\boldsymbol{t}_{1,\alpha_\nu-1}^{+}) e^{-\gamma\frac{1}{2}\boldsymbol{t}_{1,\alpha_\nu-1}^{+}}}{\prod_{i=1}^{a_k}\max(1, | t_i |)^\frac{d}{2} }
	  \, d\boldsymbol{t}  d\boldsymbol{s}
	  \\
	&\leq  \left(\frac{d}{d-2}\right)^{a_k} \int_{[0,t]_{\leq}^k}  \frac{C}{\max(1,|\gamma \hbar^{-1}(s_{\nu-1}-s_\nu)|^d)} \,d\boldsymbol{s}_{1,k}
	\leq \frac{\hbar t^{k-1} d}{\gamma(d-1)(k-1)! }  \left(\frac{d}{d-2}\right)^{a_k}.
	\end{aligned}
\end{equation} 
Combining \eqref{EQ:regularise_expansion_in_time_3} and \eqref{EQ:regularise_expansion_in_time_4} we obtain the bound
   \begin{equation}\label{EQ:regularise_expansion_in_time_5}
   	\begin{aligned}
	|\Aver{\mathcal{T}(n,\alpha,\tilde{\alpha},\sigma^1,\sigma^2,\mathrm{id})}| \leq {}&   \norm{\varphi}_{L^2(\R^d)}^2 \left(C_d\lambda \norm{\hat{V}}_{1,\infty,3d+3}\right)^{|\alpha|+|\tilde{\alpha}|}    \frac{\hbar \rho (\rho t)^{2k-n-1}}{\gamma (k-n)!(k-1)!},  
	\end{aligned}
\end{equation}
where we again have absorbed all constants from the integration only depending on dimension in to the constant $C^{|\alpha|+|\tilde{\alpha}|} $. For the crossing terms we get the estimate
       \begin{equation}\label{EQ:regularise_expansion_in_time_6}
   	\begin{aligned}
	|\Aver{\mathcal{T}(n,\alpha,\tilde{\alpha},\sigma^1,\sigma^2,\kappa)}|
	 \leq      \left(C_d \lambda\norm{\hat{V}}_{1,\infty,3d+3}\right)^{|\alpha|+|\tilde{\alpha}|}     \frac{\rho(\rho t)^{2k-n-1} \hbar |\log(\tfrac{\hbar}{t})|^{n+3}}{(k-1)!(k-n)!} \norm{\varphi}_{\mathcal{H}^{2d+2}_\hbar(\R^d)}^2.
	\end{aligned}
\end{equation}
This estimate is obtained complete analogous to the estimate obtained in the proof of Lemma~\ref{expansion_aver_bound_Mainterm}. Combining the estimates in \eqref{EQ:regularise_expansion_in_time_2} \eqref{EQ:regularise_expansion_in_time_5} and \eqref{EQ:regularise_expansion_in_time_6} we get
\begin{equation}\label{EQ:regularise_expansion_in_time_7}
	\begin{aligned}
	\MoveEqLeft \mathbb{E}\Big[\big\lVert  \mathcal{I}_\nu^\gamma(k,t;\hbar)\varphi \big\rVert_{L^2(\R^d)}^2\Big]
	\leq \sum_{\alpha,\tilde{\alpha}\in\N^k} \norm{\varphi}_{\mathcal{H}^{2d+2}_\hbar(\R^d)}^2  \left(C_d\lambda\norm{\hat{V}}_{1,\infty,3d+3}\right)^{|\alpha|+|\tilde{\alpha}|}   
	\\
	&\times  
	\sum_{n=0}^k  \Big[  \frac{ \hbar \rho(\rho t)^{2k-n-1}}{\gamma (k-1)!(k-n)!} 	+ \hbar |\log(\tfrac{\hbar}{t})|^{n+3} \frac{\rho(\rho t)^{2k-n-1} n!}{(k-1)!(k-n)!} \Big],
	\end{aligned}
\end{equation}
where we have used that the number of elements in $\mathcal{A}(k,n)$ is bounded by $2^k$ and that the number of elements in $\mathcal{S}_n$ is $n!$. By arguing as in the proof of Lemma~\ref{regularise_expansion} we get from \eqref{EQ:regularise_expansion_in_time_1} and  \eqref{EQ:regularise_expansion_in_time_7} that
\begin{equation}\label{EQ:regularise_expansion_in_time_8}
	\begin{aligned}
	\MoveEqLeft\mathbb{E} \bigg[\Big\lVert \sum_{\boldsymbol{x}\in\mathcal{X}^{k}_{\neq}}  \mathcal{I}_{0,0}^\gamma(k,\boldsymbol{x},t;\hbar) \varphi-\mathcal{I}_{\infty}^\gamma(k,\boldsymbol{x},t;\hbar)\varphi \Big\rVert_{L^2(\R^d)}^2	\bigg]^\frac{1}{2}    
\leq C \norm{\varphi}_{\mathcal{H}^{2d+2}_\hbar(\R^d)} \Big( \sqrt{\tfrac{\hbar}{\gamma}} +\sqrt{ \hbar k_0^2 |\log(\tfrac{\hbar}{t})|^{k_0+3}} \Big). 
	\end{aligned}
\end{equation}
Since we have that $\sup_{\hbar\in I}\norm{\varphi}_{\mathcal{H}^{2d+2}_\hbar(\R^d)}$ and our assumption on $k_0$ we get from the estimate in \eqref{EQ:regularise_expansion_in_time_8} that 
\begin{equation*}
\lim_{\hbar\rightarrow 0}\sum_{k=1}^{k_0}     \mathbb{E} \bigg[\Big\lVert \sum_{\boldsymbol{x}\in\mathcal{X}^{k}_{\neq}}  \mathcal{I}_{0,0}^\gamma(k,\boldsymbol{x},t;\hbar) \varphi-\mathcal{I}_{\infty}^\gamma(k,\boldsymbol{x},t;\hbar)\varphi \Big\rVert_{L^2(\R^d)}^2	\bigg]^\frac{1}{2}  = 0.
\end{equation*}
As desired. This concludes the proof.
\end{proof}
\section{Convergence of error terms}\label{Sec:con_error_term}
\begin{lemma}\label{LE:Convergence_of_duhamel_expansion} 
Assume we are in the setting of Lemma~\ref{duhamel_expansion_lemma} and that $d=3$. Let $\{\varphi_\hbar\}_{\hbar\in I}$ be a uniformly bounded family in $\mathcal{H}_\hbar^{5d+5}$. Moreover let 
\begin{equation*}
	k_0 = \frac{|\log(\hbar)|}{10 \log(|\log(\hbar)|)} \qquad\text{and}\qquad \tau_0 = \hbar^{-\frac{1}{3}}.
\end{equation*}
Then we have that
\begin{equation*}
	\begin{aligned}
	\lim_{\hbar\rightarrow 0} \mathbb{E} \big[  \norm{ \mathcal{R}(\tau_0,k_0:\hbar) \varphi_\hbar}_{L^2(\R^d)}^2\big] = 0.
	\end{aligned}
\end{equation*}
\end{lemma}
\begin{proof}
Firstly recal that 
\begin{equation*}
	\begin{aligned}
	 \mathcal{R}(\tau_0,k_0:\hbar) = \sum_{\tau=1}^{\tau_0} U_{\hbar,\lambda}(-\tfrac{\tau_0-\tau}{\tau_0}t) \big[  \mathcal{R}_\tau^{\mathrm{rec}}(k_0;\hbar)+ \sum_{i=0}^2 \sum_{j=\min(1,i)}^i \mathcal{R}_{\tau,i,j}^{k_0}(N;\hbar)  \big]. 
	\end{aligned}
\end{equation*}
We get from Jensens inequality, the unitarity of $U_{\hbar,\lambda}(t)$ and linearity of the average that
\begin{equation}\label{EQ:Convergence_of_duhamel_expansion_1}
	\begin{aligned}
	\mathbb{E} \big[  \norm{ \mathcal{R}(\tau_0,k_0:\hbar) \varphi_\hbar}_{L^2(\R^d)}^2\big] \leq  \sum_{\tau=1}^{\tau_0} \tau_0 \Big\{\mathbb{E} \big[  \norm{  \mathcal{R}_\tau^{\mathrm{rec}}(k_0;\hbar) }_{L^2(\R^d)}^2\big] + 4\sum_{i=0}^2 \sum_{j=\min(1,i)}^i \mathbb{E} \big[  \norm{\mathcal{R}_{\tau,i,j}^{k_0}(N;\hbar) }_{L^2(\R^d)}^2\big]   \Big\}.
	\end{aligned}
\end{equation}
We will consider each average separately. We will start with the averages of $ \norm{  \mathcal{R}_\tau^{\mathrm{rec}}(k_0;\hbar) }_{L^2(\R^d)}^2$. By definition we have for $\tau=1$ that
\begin{equation*}
	\mathcal{R}_1^{\mathrm{rec}}(k_0;\hbar)  = \sum_{j=1}^2
	  \sum_{k_1=6-j}^{k_0} \sum_{\iota\in\mathcal{Q}_{k_1,2,j}}   \sum_{\boldsymbol{x}_1\in \mathcal{X}_{\neq}^{k_1-2}}  \mathcal{E}_{2,j}^{\mathrm{rec}}(k_1,0,\boldsymbol{x}_1,\iota,\tfrac{t}{\tau_0};\hbar).
\end{equation*}
Recalling that we in Section~\ref{Sec:tech_est_truncated_recol_2} divided $\mathcal{Q}_{k_1+k_2,2,2}$ into three disjoint sets we get that
\begin{equation}\label{EQ:Convergence_of_duhamel_expansion_2}
	\begin{aligned}
	\mathbb{E} \big[  \norm{  \mathcal{R}_1^{\mathrm{rec}}(k_0;\hbar) }_{L^2(\R^d)}^2\big] \leq {}& 4 \mathbb{E}\Big[\big\lVert \sum_{k_1=5}^{k_0} \sum_{\iota\in\mathcal{Q}_{k_1,2,1}}   \sum_{\boldsymbol{x}_1\in \mathcal{X}_{\neq}^{k_1-2}}  \mathcal{E}_{2,1}^{\mathrm{rec}}(k_1,0,\boldsymbol{x}_1,\iota,\tfrac{t}{\tau_0};\hbar)\varphi\big\rVert_{L^2(\R^d)}^2 \Big]
	\\
	&+ 4\sum_{j=1}^3 \mathbb{E}\Big[\big\lVert \sum_{k_1=4}^{k_0} \sum_{\iota\in\mathcal{Q}_{k_1,2,2}^j}   \sum_{\boldsymbol{x}_1\in \mathcal{X}_{\neq}^{k_1-2}}  \mathcal{E}_{2,1}^{\mathrm{rec}}(k_1,0,\boldsymbol{x}_1,\iota,\tfrac{t}{\tau_0};\hbar).\varphi\big\rVert_{L^2(\R^d)}^2 \Big].
	\end{aligned}
\end{equation}
From applying \Cref{expansion_aver_bound_Mainterm_rec_trun21_1,expansion_aver_bound_Mainterm_rec_trun22_1,expansion_aver_bound_Mainterm_rec_trun22_3,expansion_aver_bound_Mainterm_rec_trun22_5} we obtain that
\begin{equation}\label{EQ:Convergence_of_duhamel_expansion_3}
	\begin{aligned}
	\mathbb{E} \big[  \norm{  \mathcal{R}_1^{\mathrm{rec}}(k_0;\hbar) }_{L^2(\R^d)}^2\big] \leq {}& C k_0^{19} C^{k_0}  |\log(\tfrac{\hbar}{\tau_0})|^{k_0+7}  \norm{\varphi}_{\mathcal{H}^{2d+2}_\hbar(\R^d)}^2  \hbar  +C \tau_0^{-3}  \norm{\varphi}_{\mathcal{H}^{5d+5}_\hbar(\R^d)}^2 k_0^{10} C^{k_0} |\log(\tfrac{\hbar}{t})|^{k_0+20}, 
	\end{aligned}
\end{equation}
where we have combined some of the estimates obtained in the different lemmas by multiplying one of the constants with 10. 
For $\tau\in\{2,\dots,\tau_0\}$ we have that
\begin{equation*}
	\begin{aligned}
	\MoveEqLeft  \mathcal{R}_\tau^{\mathrm{rec}}(k_0;\hbar) 
	\\
	={}&\sum_{j=1}^2 \sum_{k_1=1}^{k_0} \sum_{k_2=0}^{k_0} \sum_{\iota\in\mathcal{Q}_{k_1+k_2,2,j}}  \boldsymbol{1}_{\{k_1+k_2\geq6-j\}}   \Big\{  \sum_{(\boldsymbol{x}_1,\boldsymbol{x}_2)\in \mathcal{X}_{\neq}^{k_1+k_2-2}}  \mathcal{E}_{2,j}^{\mathrm{rec}}(k_1,k_2,\boldsymbol{x}_1,\iota,\tfrac{t}{\tau_0};\hbar) \mathcal{I}_{2,j}^{\mathrm{rec}}(k_2,\boldsymbol{x}_2,\iota,\tfrac{\tau-1}{\tau_0}t;\hbar)
	 \\
	 &+
	  \sum_{(\boldsymbol{x}_1,\boldsymbol{x}_2)\in \mathcal{X}_{\neq}^{k_1+k_2-2}}  \mathcal{E}_{2,j}^{\mathrm{rec}}(k_1+1,k_2,(x_{2,k_2},\boldsymbol{x}_1),\iota,\tfrac{t}{\tau_0};\hbar) \mathcal{I}_{2,j}^{\mathrm{rec}}(k_2,\boldsymbol{x}_2,\iota,\tfrac{\tau-1}{\tau_0}t;\hbar) \Big\}.
	\end{aligned}
\end{equation*}
Hence if we apply the ``same'' inequality as in \eqref{EQ:Convergence_of_duhamel_expansion_2} and apply \Cref{expansion_aver_bound_Mainterm_rec_trun21_2,expansion_aver_bound_Mainterm_rec_trun22_2,expansion_aver_bound_Mainterm_rec_trun22_4,expansion_aver_bound_Mainterm_rec_trun22_6} we obtain for all $\tau\geq2$ that
\begin{equation}\label{EQ:Convergence_of_duhamel_expansion_4}
	\begin{aligned}
	\mathbb{E} \big[  \norm{  \mathcal{R}_\tau^{\mathrm{rec}}(k_0;\hbar) }_{L^2(\R^d)}^2\big] \leq {}& C k_0^{21} C^{k_0}  |\log(\tfrac{\hbar}{\tau_0})|^{k_0+9}  \norm{\varphi}_{\mathcal{H}^{3d+3}_\hbar(\R^d)}^2  \hbar  +C \tau_0^{-3}  \norm{\varphi}_{\mathcal{H}^{5d+5}_\hbar(\R^d)}^2 k_0^{16} C^{k_0}   |\log(\tfrac{\hbar}{t})|^{k_0+22}.
	\end{aligned}
\end{equation}
Combining the two estimates in \eqref{EQ:Convergence_of_duhamel_expansion_3} and \eqref{EQ:Convergence_of_duhamel_expansion_4} we obtain that
\begin{equation*}
	\begin{aligned}
	 \MoveEqLeft \sum_{\tau=1}^{\tau_0} \tau_0 \mathbb{E} \big[  \norm{  \mathcal{R}_\tau^{\mathrm{rec}}(k_0;\hbar) }_{L^2(\R^d)}^2\big] 
	 \\
	 &\leq C \tau_0^2 k_0^{21} C^{k_0}  |\log(\tfrac{\hbar}{\tau_0})|^{k_0+9}  \norm{\varphi}_{\mathcal{H}^{3d+3}_\hbar(\R^d)}^2  \hbar  +C \tau_0^{-1}  \norm{\varphi}_{\mathcal{H}^{5d+5}_\hbar(\R^d)}^2 k_0^{16} C^{k_0}   |\log(\tfrac{\hbar}{t})|^{k_0+22}. 
	\end{aligned}
\end{equation*}
With this estimate and our assumptions on $k_0$ and $\tau_0$ it follows that
\begin{equation}\label{EQ:Convergence_of_duhamel_expansion_5}
	\begin{aligned}
	 \lim_{\hbar\rightarrow 0}\sum_{\tau=1}^{\tau_0} \tau_0 \mathbb{E} \big[  \norm{  \mathcal{R}_\tau^{\mathrm{rec}}(k_0;\hbar) }_{L^2(\R^d)}^2\big] =0.
	\end{aligned}
\end{equation}
We now turn to the error terms where we have seen $k_0+1$ different collision centers. Here we have by definition for all $(i,j)\in\{(0,0), (1,1), (2,1),(2,2)\}$ that
\begin{equation*}
	\begin{aligned}
	\mathcal{R}_{1,i,j}^{k_0}(\hbar) =
	  \sum_{\iota\in\mathcal{Q}_{k_0+1,i,j}}\sum_{\boldsymbol{x}_1\in \mathcal{X}_{\neq}^{k_0+1-i}}  \mathcal{E}_{i,j}^{k_0}(\boldsymbol{x}_1,\iota,\tfrac{t}{\tau_0};\hbar) .
	 \end{aligned}
\end{equation*}
Applying \Cref{LE:Truncated_Without_re_bound_Mainterm_1,expansion_aver_bound_Mainterm_rec_trun,expansion_aver_bound_Mainterm_rec_trun2j} We obtain the estimate 
\begin{equation}\label{EQ:Convergence_of_duhamel_expansion_6}
	\begin{aligned}
	\MoveEqLeft  4\sum_{i=0}^2 \sum_{j=\min(1,i)}^i \mathbb{E} \big[  \norm{\mathcal{R}_{1,i,j}^{k_0}(N;\hbar) }_{L^2(\R^d)}^2\big]   
	 \\
	 &\leq 16\Big[ \frac{k_0^{13} C^{k_0}}{ \tau_0^{k_0-2}\hbar (k_0-2)!} \norm{\varphi}_{L^{2}(\R^d)}^2 +   \frac{k_0^{16} C^{k_0}}{ \tau_0^{k_0-4}} \hbar  |\log(\tfrac{\hbar}{t})|^{k_0+7} \norm{\varphi}_{\mathcal{H}^{3d+3}_\hbar(\R^d)}^2\Big],
	\end{aligned}
\end{equation}
where we have estimate all terms with the ``worst'' estimate. For $\tau\in\{2,\dots,\tau_0\}$ and all $(i,j)\in\{(0,0), (1,1), (2,1),(2,2)\}$ we have by definition that 
\begin{equation*}
	\begin{aligned}
\MoveEqLeft \mathcal{R}_{\tau,i,j}^{k_0}(N;\hbar) 
	 =  \sum_{\iota\in\mathcal{Q}_{k_0+1,i,j}}   \sum_{\boldsymbol{x}\in \mathcal{X}_{\neq}^{k_0+1-i}}  \mathcal{E}_{i,j}^{k_0}(\boldsymbol{x},\iota,\tfrac{t}{\tau_0};\hbar) U_{\hbar,0}(-\tfrac{\tau-1}{\tau_0}t;\hbar) 
	 \\
	 &+ \sum_{k_2=1}^{k_0} \sum_{\iota\in\mathcal{Q}_{k_0+k_2+1,i,j}}   \sum_{(\boldsymbol{x}_1,\boldsymbol{x}_2)\in \mathcal{X}_{\neq}^{k_0+k_2+1-i}}  \mathcal{E}_{i,j}^{k_0}(\boldsymbol{x}_1,\iota,\tfrac{t}{\tau_0};\hbar) \mathcal{I}_{i,j}(k_2,\boldsymbol{x}_2,\iota,\tfrac{\tau-1}{\tau_0}t;\hbar) 
	 \\
	 &+  \sum_{k_2=1}^{k_0} \sum_{\iota\in\mathcal{Q}_{k_0+k_2+1,i,j}}   \sum_{(\boldsymbol{x}_1,\boldsymbol{x}_2)\in \mathcal{X}_{\neq}^{k_0+k_2+1-i}}  \mathcal{E}_{i,j}^{k_0+1}((x_{2,k_2},\boldsymbol{x}_1),\iota,\tfrac{t}{\tau_0};\hbar) \mathcal{I}_{i,j}(k_2,\boldsymbol{x}_2,\iota,\tfrac{\tau-1}{\tau_0}t;\hbar) 
	 \\
	 &+ \sum_{k_1=1}^{k_0} \sum_{k_2=k_0-k_1+1}^{k_0} \sum_{\iota\in\mathcal{Q}_{k_1+k_2,i,j}} 
	  \sum_{(\boldsymbol{x}_1,\boldsymbol{x}_2)\in \mathcal{X}_{\neq}^{k_1+k_2-i}}  \mathcal{I}_{i,j}(k_1,\boldsymbol{x}_1,\iota,\tfrac{t}{\tau_0};\hbar)  \mathcal{I}_{i,j}(k_2,\boldsymbol{x}_2,\iota,\tfrac{\tau-1}{\tau_0}t;\hbar)
	  \\
	  &+  \sum_{k_1=1}^{k_0} \sum_{k_2=k_0-k_1+1}^{k_0} \sum_{\iota\in\mathcal{Q}_{k_1+k_2,i,j}} 
	  \sum_{(\boldsymbol{x}_1,\boldsymbol{x}_2)\in \mathcal{X}_{\neq}^{k_1+k_2-i}}  \mathcal{I}_{i,j}(k_1+1,(x_{2,k_2},\boldsymbol{x}_1),\iota,\tfrac{t}{\tau_0};\hbar)  \mathcal{I}_{i,j}(k_2,\boldsymbol{x}_2,\iota,\tfrac{\tau-1}{\tau_0}t;\hbar). 
	 \end{aligned}
\end{equation*}
Here we will apply the following results: \Cref{expansion_aver_bound_Mainterm_2,expansion_aver_bound_Mainterm_3,LE:Truncated_Without_re_bound_Mainterm_1,LE:Truncated_Without_re_bound_Mainterm_2,expansion_aver_bound_Mainterm_rec_2,expansion_aver_bound_Mainterm_rec_3,expansion_aver_bound_Mainterm_rec_trun,expansion_aver_bound_Mainterm_rec_trun_2,expansion_aver_bound_Mainterm_rec_21_2,expansion_aver_bound_Mainterm_rec_22_2,expansion_aver_bound_Mainterm_rec_trun2j,expansion_aver_bound_Mainterm_rec_trun_j_2}. This gives us the estimate
\begin{equation}\label{EQ:Convergence_of_duhamel_expansion_7}
	\begin{aligned}
	\MoveEqLeft  4\sum_{i=0}^2 \sum_{j=\min(1,i)}^i \mathbb{E} \big[  \norm{\mathcal{R}_{1,i,j}^{k_0}(N;\hbar) }_{L^2(\R^d)}^2\big]   
	 \\
	 &\leq 50\Big[ \frac{k_0^{15} C^{k_0}}{ \tau_0^{k_0-2}\hbar (k_0-2)!} \norm{\varphi}_{L^{2}(\R^d)}^2 +   \frac{k_0^{17} C^{k_0}}{ \tau_0^{k_0-4}} \hbar  |\log(\tfrac{\hbar}{t})|^{k_0+7} \norm{\varphi}_{\mathcal{H}^{3d+3}_\hbar(\R^d)}^2\Big].
	\end{aligned}
\end{equation}
Combining the estimates in \eqref{EQ:Convergence_of_duhamel_expansion_6} and \eqref{EQ:Convergence_of_duhamel_expansion_7} we obtain that
\begin{equation*}
	\begin{aligned}
	\MoveEqLeft  \sum_{\tau=1}^{\tau_0} 4 \tau_0 \sum_{i=0}^2 \sum_{j=\min(1,i)}^i \mathbb{E} \big[  \norm{\mathcal{R}_{\tau,i,j}^{k_0}(N;\hbar) }_{L^2(\R^d)}^2\big]  
	 \\
	 &\leq  50\Big[ \frac{k_0^{15} C^{k_0}}{ \tau_0^{k_0-4}\hbar (k_0-2)!} \norm{\varphi}_{L^{2}(\R^d)}^2 +   \frac{k_0^{17} C^{k_0}}{ \tau_0^{k_0-6}} \hbar  |\log(\tfrac{\hbar}{t})|^{k_0+7} \norm{\varphi}_{\mathcal{H}^{3d+3}_\hbar(\R^d)}^2\Big].
	\end{aligned}
\end{equation*}
With this estimate and our assumptions on $k_0$ and $\tau_0$ it follows that
\begin{equation}\label{EQ:Convergence_of_duhamel_expansion_8}
	\begin{aligned}
	 \lim_{\hbar\rightarrow 0} \sum_{\tau=1}^{\tau_0} 4 \tau_0 \sum_{i=0}^2 \sum_{j=\min(1,i)}^i \mathbb{E} \big[  \norm{\mathcal{R}_{\tau,i,j}^{k_0}(N;\hbar) }_{L^2(\R^d)}^2\big]   =0.
	\end{aligned}
\end{equation}
Finally by combing \eqref{EQ:Convergence_of_duhamel_expansion_1}, \eqref{EQ:Convergence_of_duhamel_expansion_5} and \eqref{EQ:Convergence_of_duhamel_expansion_8} we get that
\begin{equation*}
	\begin{aligned}
	 \lim_{\hbar\rightarrow 0} \mathbb{E} \big[  \norm{ \mathcal{R}(\tau_0,k_0:\hbar) \varphi_\hbar}_{L^2(\R^d)}^2\big]=0.
	\end{aligned}
\end{equation*}
This concludes the proof.
\end{proof}
 \begin{lemma}\label{LE:unf_norm_bound_without_rec} 
Assume we are in the setting of Lemma~\ref{duhamel_expansion_lemma}, let $\{\varphi_\hbar\}_{\hbar\in I}$ be a uniformly bounded family in $\mathcal{H}_\hbar^{5d+5}(\R^d)$. Moreover let 
\begin{equation*}
	k_0 = \frac{|\log(\hbar)|}{10 \log(|\log(\hbar)|)}.
\end{equation*}
Then we have that
\begin{equation*}
	\begin{aligned}
	\sup_{\hbar\in I}   \sum_{k=0}^{k_0}     \mathbb{E}\Big[\big\lVert  \sum_{\boldsymbol{x}\in\mathcal{X}_{\neq}^{k}}  \mathcal{I}_{0,0}(k,\boldsymbol{x},\iota,t;\hbar) \varphi_\hbar\big\rVert_{L^2(\R^d)}^2 \Big]^{\frac{1}{2}} \leq C,
	\end{aligned}
\end{equation*}
where $ \mathcal{I}_{0,0}(0,\boldsymbol{x},\iota,t;\hbar) = U_{\hbar,0}(-t)$.
\end{lemma}
\begin{proof}
Firstly note that by assumption and unitarity of $U_{\hbar,0}(-t)$ we have that
\begin{equation}\label{EQ:unf_norm_bound_without_rec_1}
	\begin{aligned}
	\sup_{\hbar\in I}      \mathbb{E}\Big[\big\lVert   \mathcal{I}_{0,0}(0,\boldsymbol{x},\iota,t;\hbar) \varphi_\hbar\big\rVert_{L^2(\R^d)}^2 \Big]^{\frac{1}{2}} \leq \sup_{\hbar\in I} \norm{\varphi_\hbar}_{L^2(\R^d)}\leq C.
	\end{aligned}
\end{equation}
 For $k\geq1$ we have from Lemma~\ref{expansion_aver_bound_Mainterm} the estimate
\begin{equation*}
	\begin{aligned}
	\MoveEqLeft \mathbb{E}\Big[\big\lVert \sum_{\boldsymbol{x}\in\mathcal{X}_{\neq}^{k}}\mathcal{I}_{0,0}(k,\boldsymbol{x},t;\hbar) \varphi\big\rVert_{L^2(\R^d)}^2\Big] \leq   \frac{ C^k}{k!} \norm{\varphi}_{L^2(\R^d)}^2  
	+ \hbar k C^k |\log(\tfrac{\hbar}{t})|^{k+3}  \norm{\varphi}_{\mathcal{H}^{2d+2}_\hbar(\R^d)}^2,
	\end{aligned}
\end{equation*}
where $C$ is independent of $\hbar$.
Applying this and the subadditivity of the square-root we get that
\begin{equation*}
	\begin{aligned}
	  \sum_{k=1}^{k_0}     \mathbb{E}\Big[\big\lVert  \sum_{\boldsymbol{x}\in\mathcal{X}_{\neq}^{k}}  \mathcal{I}_{0,0}(k,\boldsymbol{x},\iota,t;\hbar) \varphi_\hbar\big\rVert_{L^2(\R^d)}^2 \Big]^{\frac{1}{2}} \leq  \sum_{k=1}^{k_0} \sqrt{ \frac{ C^k}{k!}} \norm{\varphi}_{L^2(\R^d)}  
	+ \sqrt{\hbar k C^k |\log(\tfrac{\hbar}{t})|^{k+3} } \norm{\varphi}_{\mathcal{H}^{2d+2}_\hbar(\R^d)}.
	\end{aligned}
\end{equation*}
Since we have that  
\begin{equation*}
	\begin{aligned}
	  \sum_{k=1}^{\infty} \sqrt{ \frac{ C^k}{k!}}<\infty\quad\text{and}\quad\lim_{\hbar\rightarrow0} \hbar k_0^3 C^{k_0} |\log(\tfrac{\hbar}{t})|^{k_0+3} =0,
	\end{aligned}
\end{equation*}
due to our assumption on $k_0$, we get that
\begin{equation}\label{EQ:unf_norm_bound_without_rec_2}
	\sup_{\hbar\in I}   \sum_{k=1}^{k_0}     \mathbb{E}\Big[\big\lVert  \sum_{\boldsymbol{x}\in\mathcal{X}_{\neq}^{k}}  \mathcal{I}_{0,0}(k,\boldsymbol{x},\iota,t;\hbar) \varphi_\hbar\big\rVert_{L^2(\R^d)}^2 \Big]^{\frac{1}{2}} \leq C.
\end{equation}
From \eqref{EQ:unf_norm_bound_without_rec_1} and \eqref{EQ:unf_norm_bound_without_rec_2} we get the desired result.
\end{proof}
 \begin{lemma}\label{LE:norm_conv_with_op_1_rec} 
Assume we are in the setting of Lemma~\ref{duhamel_expansion_lemma}, let $\{\varphi_\hbar\}_{\hbar\in I}$ be a uniformly bounded family in $\mathcal{H}_\hbar^{5d+5}(\R^d)$. Moreover let 
\begin{equation*}
	k_0 = \frac{|\log(\hbar)|}{10 \log(|\log(\hbar)|)}.
\end{equation*}
Then for $(i,j)\in\{(1,1), (2,1),(2,2)\}$ we have that
\begin{equation*}
	\begin{aligned}
	\lim_{\hbar\rightarrow 0}   \sum_{k=k_{ij}}^{k_0} \sum_{\iota\in\mathcal{Q}_{k,i,j}}     \mathbb{E}\Big[\big\lVert  \sum_{\boldsymbol{x}\in\mathcal{X}_{\neq}^{k-i}}  \mathcal{I}_{i,j}(k,\boldsymbol{x},\iota,t;\hbar) \varphi_\hbar\big\rVert_{L^2(\R^d)}^2 \Big]^{\frac{1}{2}}= 0,
	\end{aligned}
\end{equation*}
where $k_{11}=3$, $k_{21}=5$, $k_{22}=4$.
\end{lemma}
\begin{proof}
We will only prove the case $(i,j)=(1,1)$ as the other cases are proved analogously.
As previous we will start by ``writting'' up the $L^2$-norm and divide into cases depending on the number of collisions. However, before this we will insert ``$1$'' in the following way $1-e^{-\hbar^{-\frac{4}{7}\frac{1}{2}(s_{m_1}-s_{m_2})}}+e^{-\hbar^{-\frac{4}{7}\frac{1}{2}(s_{m_1}-s_{m_2})}}$ , where $m_1$ and $m_2$ are the numbers associated to $\iota$. The numbers $s_{m_1}$ and $s_{m_2}$ are the time variables we preform integration in from the definition of  the operator $\mathcal{I}_{1,1}(k,\boldsymbol{x},t;\hbar)$. To be precise we will make the following splitting
\begin{equation*}
	\begin{aligned}
	\mathcal{I}_{1,1}(k,\boldsymbol{x},\iota,t;\hbar) = \mathcal{I}_{1,1}^1(k,\boldsymbol{x},\iota,t;\hbar)  + \mathcal{I}_{1,1}^2(k,\boldsymbol{x},\iota,t;\hbar) ,
	  \end{aligned}
\end{equation*}
where
\begin{equation*}
	\begin{aligned}
	\mathcal{I}_{1,1}^1(k,\boldsymbol{x},\iota,t;\hbar) =  \sum_{\alpha\in \N^k} (i\lambda)^{|\alpha|} \int_{[0,t]_{\leq}^k} \big(1-e^{-\hbar^{-\frac{4}{7}}\frac{1}{2}(s_{m_1}-s_{m_2})}\big)  \prod_{m=1}^k \Theta_{\alpha_m}(s_{m-1},{s}_{m},x_{\iota(m)};V,\hbar)\, d\boldsymbol{s}_{k,1}U_{\hbar,0}(-t)
	  \end{aligned}
\end{equation*}
and 
\begin{equation*}
	\begin{aligned}
	\mathcal{I}_{1,1}^2(k,\boldsymbol{x},\iota,t;\hbar) =  \sum_{\alpha\in \N^k} (i\lambda)^{|\alpha|} \int_{[0,t]_{\leq}^k} e^{-\hbar^{-\frac{4}{7}}\frac{1}{2}(s_{m_1}-s_{m_2})}  \prod_{m=1}^k \Theta_{\alpha_m}(s_{m-1},{s}_{m},x_{\iota(m)};V,\hbar)\, d\boldsymbol{s}_{k,1}U_{\hbar,0}(-t).
	  \end{aligned}
\end{equation*}
Applying the triangle inequality for the $L^2(\R^d\times \Omega, dx\otimes \Pro)$-norm we get
\begin{equation}\label{EQ:norm_conv_with_op_1_rec_0}
	\begin{aligned}
	    &\sum_{k=3}^{k_0} \sum_{\iota\in\mathcal{Q}_{k,1,1}}  \mathbb{E}\Big[\big\lVert  \sum_{\boldsymbol{x}\in\mathcal{X}_{\neq}^{k-1}}  \mathcal{I}_{1,1}(k,\boldsymbol{x},t;\hbar) \varphi_\hbar\big\rVert_{L^2(\R^d)}^2 \Big]^\frac{1}{2} 
	   \\
	   &\leq  \sum_{k=3}^{k_0} \sum_{\iota\in\mathcal{Q}_{k,1,1}}   \mathbb{E}\Big[\big\lVert  \sum_{\boldsymbol{x}\in\mathcal{X}_{\neq}^{k-1}}  \mathcal{I}_{1,1}^1(k,\boldsymbol{x},t;\hbar) \varphi_\hbar\big\rVert_{L^2(\R^d)}^2 \Big]^\frac{1}{2} +  \mathbb{E}\Big[\big\lVert  \sum_{\boldsymbol{x}\in\mathcal{X}_{\neq}^{k-1}}  \mathcal{I}_{1,1}^2(k,\boldsymbol{x},t;\hbar) \varphi_\hbar\big\rVert_{L^2(\R^d)}^2 \Big]^\frac{1}{2}.
	\end{aligned}
\end{equation}
As in the proof of Lemma~\ref{expansion_aver_bound_Mainterm_rec_trun21_1} we will then apply Lemma~\ref{LE:crossing_dom_ladder} such that we only have to consider ladder terms for both cases. Hence we obtain for $j$ equal $1$ or $2$ and for each fixed $\kappa$ and $\iota$ that
\begin{equation}\label{EQ:norm_conv_with_op_1_rec_1}
	\begin{aligned}
	   \mathbb{E}\Big[\big\lVert  \sum_{\boldsymbol{x}\in\mathcal{X}_{\neq}^{k-1}}  \mathcal{I}_{1,1}^j(k,\boldsymbol{x},t;\hbar) \varphi_\hbar\big\rVert_{L^2(\R^d)}^2 \Big]
	    \leq2 \sum_{\alpha,\tilde{\alpha}\in \N^k}\lambda^{|\alpha|+|\tilde{\alpha}|}   \sum_{n=0}^{k-1} \binom{k-1}{n} n!  \sum_{\sigma\in\mathcal{A}(k-1,n)} \big|\Aver{ \mathcal{T}^j(n,\sigma,\alpha,\tilde{\alpha};\hbar)}\big|,
	\end{aligned}
\end{equation}
where
\begin{equation*}
	\begin{aligned}
	 \mathcal{T}^j(n,\sigma,\alpha,\tilde{\alpha};\hbar) = \sum_{(\boldsymbol{x},\tilde{\boldsymbol{x}})\in \mathcal{X}_{\neq}^{2k-2}}    \prod_{i=1}^n \frac{\delta(x_{\sigma_i}- \tilde{x}_{\sigma_{i}})}{\rho}
	  \int_{\R^d}  \mathcal{I}_{1,1}^j(k,\boldsymbol{x},\alpha,t;\hbar)\varphi_\hbar(x)\overline{ \mathcal{I}_{1,1}^j(k,\boldsymbol{\tilde{x}},\tilde{\alpha},t;\hbar) \varphi_\hbar(x)}  \,dx,
	\end{aligned}
\end{equation*}
where as usual $ \mathcal{I}_{1,1}^j(k,\boldsymbol{x},\alpha,t;\hbar)$ is defined as $\mathcal{I}_{1,1}^j(k,\boldsymbol{x},t;\hbar)$ for fixed $\alpha$. We start with the case $j=1$ for a fixed $\kappa$ and $\iota$. By definition of the operators we have that
\begin{equation*}
	\begin{aligned}
 	\MoveEqLeft  \mathcal{I}_{1,1}^1(k,\boldsymbol{x},\alpha,\iota,t;\hbar) \varphi(x) 
	=  \frac{1}{(2\pi\hbar)^{d}} \int_{\R_{+}^{j+1}} \int \delta(\tfrac{t}{\hbar}- \boldsymbol{s}_{0,k}^{+}+ \boldsymbol{t}_{1,a_k}^{+})
	 \big(1-\prod_{m=m_1}^{m_2-1} e^{-\hbar^{\frac{3}{7}}\frac{1}{2} s_m }\prod_{j=a_{m-1}+1}^{a_m}  e^{-\hbar^{\frac{3}{7}}\frac{1}{2}  t_j }\big)
	\\
	\times & e^{ i  \hbar^{-1}  \langle x,p_{k} \rangle}   \prod_{i=1}^{a_k} e^{i  t_{i} \frac{1}{2} \eta_{i}^2} \Big\{\prod_{m=1}^k e^{i  s_{m} \frac{1}{2} p_{m}^2}
	  e^{ -i   \langle  \hbar^{-1/d}x_{\iota(m)},p_{m}-p_{m-1} \rangle}     \hat{\mathcal{V}}_{\alpha_m}(p_m,p_{m-1},\boldsymbol{\eta} )  \Big\}    e^{is_0 \frac{1}{2} p_0^2} \hat{\varphi}(\tfrac{p_0}{\hbar}) \, d\boldsymbol{\eta}d\boldsymbol{p} dy   d\boldsymbol{t} d\boldsymbol{s},
	\end{aligned}
\end{equation*}
where we have done the change of variables in the time variables as described in Observation~\ref{obs_form_I_op_kernel}. Under this change of variables we have that the difference $s_{m_1}-s_{m_2}$ is mapped to $ \sum_{m=m_1}^{m_2-1} s_m$.
We now  preform the change of variables $p_{m_2-1}\mapsto p_{m_2-1} - p_{m_2}$ and $p_{m}\mapsto p_{m} - p_{m_2-1}$ for all $m\in\{m_1,\dots,m_2-2\}$, where $m_1$ and $m_2$ are the numbers associated to the map $\iota$. This change of variables and a relabelling yields the form
\begin{equation}\label{EQ:norm_conv_with_op_1_rec_4}
	\begin{aligned}
 	\MoveEqLeft \mathcal{I}_{1,1}^1(k,\boldsymbol{x},\alpha,\iota,t;\hbar) \varphi(x) 
	=  \frac{1}{(2\pi\hbar)^d} \int_{\R_{+}^{j+1}} \int \delta(\tfrac{t}{\hbar}- \boldsymbol{s}_{0,k}^{+}+ \boldsymbol{t}_{1,a_k}^{+})
	 \big(1-\prod_{m=m_1}^{m_2-1} e^{-\hbar^{\frac{3}{7}}\frac{1}{2} s_m }\prod_{j=a_{m-1}+1}^{a_m}  e^{-\hbar^{\frac{3}{7}}\frac{1}{2}  t_j }\big)
	 \\
	 \times&
	 e^{ i   \langle  \hbar^{-1}x,p_{k-1} \rangle}   e^{i  s_{m_2} \frac{1}{2} p_{m_2-1}^2} 
	   \hat{\mathcal{V}}_{\alpha_{m_2}}(p_{m_2-1},\tilde{p}+p_{m_2-1},\boldsymbol{\eta} )  \prod_{i=1}^{a_k} e^{i  t_{i} \frac{1}{2} \eta_{i}^2}  
	  \prod_{m=1}^{k-1}  e^{ -i   \langle  \hbar^{-1/d}x_{m},p_{m}-p_{m-1}\rangle} 
	   \\
	   \times& 
	    \Big\{\prod_{m=1}^{k-1}  
	   e^{i  s_{\iota^{*}(m)} \frac{1}{2} (p_{m}+\pi_{m}^1(\tilde{p}))^2}   \hat{\mathcal{V}}_{\alpha_{\iota^{*}(m)}}(p_m+\pi_{m}^1(\tilde{p}),p_{m-1}+\pi_{m}^2(\tilde{p}),\boldsymbol{\eta} )  \Big\}   e^{i  s_0 \frac{1}{2} p_0^2} \hat{\varphi}(\tfrac{p_0}{\hbar}) \, d\boldsymbol{\eta}d\boldsymbol{p}d\tilde{p}   d\boldsymbol{t} d\boldsymbol{s},
	\end{aligned}
\end{equation}
where $\iota^{*}$ is the function associated to $\iota$ and
\begin{equation*}
	\pi_{m}^1(\tilde{p} ) = \begin{cases}
	\tilde{p} & \text{if $m\in\{ m_1,\dots,m_2-1  \}$}
	\\
	0 &\text{otherwise}.
	\end{cases}
	\quad\text{and}\quad 
	\pi_{m}^2(\tilde{p} ) = \begin{cases}
	\tilde{p} & \text{if $m\in\{ m_1+1,\dots,m_2-1  \}$}
	\\
	0 &\text{otherwise}.
	\end{cases}
\end{equation*}
When taking the averages of $ \mathcal{T}^1(n,\sigma,\alpha,\tilde{\alpha};\hbar)$ we will consider the following two different cases separately. The cases are given by if there is at least one $\sigma_i$ such that $m_1 <\sigma_i <m_2$ or not. For the case where such an $\sigma_i$ we can argue as in the proof of Lemma~\ref{expansion_aver_bound_Mainterm_rec} end obtain the bound
 \begin{equation}\label{EQ:norm_conv_with_op_1_rec_5}
   	\begin{aligned}
	 \Aver{\mathcal{T}^1(n,\sigma,\alpha,\tilde{\alpha};\hbar) }
	 \leq {}&   C_d^{a_k+\tilde{a}_k}  \left(\norm{\hat{V}}_{1,\infty,5d+5}\right)^{|\alpha|+|\tilde{\alpha}|} 
	  \hbar |\log(\zeta)|^{n+4} \frac{\rho(\rho t)^{2k-n-3} }{(k-2)!(k-n-1)!} \norm{\varphi}_{\mathcal{H}^{2d+2}_\hbar(\R^d)}^2.
	\end{aligned}
\end{equation}
where in addition to the arguments previous used in the proof of Lemma~\ref{expansion_aver_bound_Mainterm_rec}  also have used the trivial estimate
\begin{equation*}
	\begin{aligned}
	 1-\prod_{m=m_1}^{m_2-1} e^{-\hbar^{\frac{3}{7}}\frac{1}{2} s_m }\prod_{j=a_{m-1}+1}^{a_m}  e^{-\hbar^{\frac{3}{7}}\frac{1}{2}  t_j }
	 \leq 1.
	 \end{aligned}
\end{equation*}
For the other case we have there exists an $i_1$ such that $\sigma_{i_1}< m_1 <m_2\leq \sigma_{i_1+1}$. When we take the average of $ \mathcal{T}^1(n,\sigma,\alpha,\tilde{\alpha};\hbar)$ we apply Lemma~\ref{LE:Exp_ran_phases} and integrate all variables in the delta functions to obtain that
\begin{equation*}
	\begin{aligned}
 	\MoveEqLeft \Aver{\mathcal{T}^1(n,\sigma,\alpha,\tilde{\alpha};\hbar)}
	=  \frac{(\rho\hbar(2\pi)^d)^{2k-2-n}}{(2\pi\hbar)^d} \int_{\R_{+}^{2j+2}} \int \delta(\tfrac{t}{\hbar}- \boldsymbol{s}_{0,k}^{+}+ \boldsymbol{t}_{1,a_k}^{+}) \delta(\tfrac{t}{\hbar}- \boldsymbol{\tilde{s}}_{0,k}^{+}+ \boldsymbol{\tilde{t}}_{1,a_k}^{+})  
	\\
	\times & \big(1-\prod_{m=m_1}^{m_2-1} e^{-\hbar^{\frac{3}{7}}\frac{1}{2} s_m }\prod_{j=a_{m-1}+1}^{a_m}  e^{-\hbar^{\frac{3}{7}}\frac{1}{2}  t_j }\big)  \big(1-\prod_{m=m_1}^{m_2-1} e^{-\hbar^{\frac{3}{7}}\frac{1}{2} \tilde{s}_m }\prod_{j=\tilde{a}_{m-1}+1}^{\tilde{a}_m}  e^{-\hbar^{\frac{3}{7}}\frac{1}{2}  \tilde{t}_j }\big)   e^{i  s_{m_2} \frac{1}{2} p_{i_1}^2}     e^{-i  \tilde{s}_{m_2} \frac{1}{2} p_{i_1}^2}  
	 \\
	 \times&    \prod_{i=1}^n  e^{i  s_{\iota^{*}(\sigma_i)} \frac{1}{2} (p_{i}+\pi_{\sigma_i}^1(\tilde{p}))^2}  e^{ -i  \tilde{s}_{\iota^{*}(\sigma_i)} \frac{1}{2} (p_{i}+\pi_{\sigma_i}^1(\tilde{q}))^2}   \prod_{i=0}^n \prod_{m=\sigma_i+1}^{\sigma_{i+1}-1}  e^{i  s_{\iota^{*}(m)} \frac{1}{2} (p_{i}+\pi_{m}^1(\tilde{p}))^2}  e^{-i  \tilde{s}_{\iota^{*}(m)} \frac{1}{2} (p_{i}+\pi_{m}^1(\tilde{q}))^2}
	 \\
	 \times& 
	  e^{i  s_0 \frac{1}{2} p_0^2} e^{-i  \tilde{s}_0 \frac{1}{2} p_0^2} \prod_{i=1}^{a_k} e^{i  t_{i} \frac{1}{2} \eta_{i}^2}    \prod_{i=1}^{\tilde{a}_k} e^{-i  \tilde{t}_{i} \frac{1}{2} \xi_{i}^2} \mathcal{G}(\boldsymbol{p},\tilde{p},\tilde{q},\boldsymbol{\eta},\boldsymbol{\xi})   |\hat{\varphi}(\tfrac{p_0}{\hbar})|^2 \, d\boldsymbol{\eta} d\boldsymbol{\xi}d\boldsymbol{p}d\tilde{p}   d\boldsymbol{t} d\boldsymbol{s} d\boldsymbol{\tilde{t}} d\boldsymbol{\tilde{s}},
	\end{aligned}
\end{equation*}
where the function $\mathcal{G}(\boldsymbol{p},\tilde{p},\tilde{q},\boldsymbol{\eta},\boldsymbol{\xi})$ is given by
\begin{equation*}
	\begin{aligned}
 	\mathcal{G}(\boldsymbol{p},\tilde{p},\tilde{q},\boldsymbol{\eta},\boldsymbol{\xi}) ={}&  \hat{\mathcal{V}}_{\alpha_{m_2}}(p_{i_1},\tilde{p}+p_{i_1},\boldsymbol{\eta} )   \overline{\hat{\mathcal{V}}_{\tilde{\alpha}_{m_2}}(p_{i_1},\tilde{q}+p_{i_1},\boldsymbol{\xi} ) } 
	  \\
	   &\times \prod_{i=1}^n  \hat{\mathcal{V}}_{\alpha_{\iota^{*}(\sigma_i)}}(p_i+\pi_{\sigma_i}^1(\tilde{p}),p_{\sigma_i-1}+\pi_{\sigma_i}^2(\tilde{p}),\boldsymbol{\eta} ) 
	    \overline{ \hat{\mathcal{V}}_{\tilde{\alpha}_{\iota^{*}(\sigma_i)}}(p_i+\pi_{\sigma_i}^1(\tilde{q}),p_{\sigma_i-1}+\pi_{\sigma_i}^2(\tilde{q}),\boldsymbol{\xi} )}  
	 \\
	 &\times \prod_{i=0}^n \prod_{m=\sigma_i+1}^{\sigma_{i+1}-1} \hat{\mathcal{V}}_{\alpha_{\iota^{*}(m)}}(p_i+\pi_{m}^1(\tilde{p}),p_{i}+\pi_{m}^2(\tilde{p}),\boldsymbol{\eta} ) \overline{ \hat{\mathcal{V}}_{\tilde{\alpha}_{\iota^{*}(m)}}(p_i+\pi_{m}^1(\tilde{q}),p_{i}+\pi_{m}^2(\tilde{q}),\boldsymbol{\xi} )}. 
	\end{aligned}
\end{equation*}
Due to the definition of $\pi_m$ and our assumptions on $\sigma$ we have that $\pi_m$ is different from zero only in the places where it is added to $p_{i_1}$. Hence we preform the following change of variables $\tilde{p}\mapsto \tilde{p}-p_{i_1}$ and $\tilde{q}\mapsto \tilde{q}-p_{i_1}$. With this change of variables we obtain the following expression
\begin{equation}\label{EQ:norm_conv_with_op_1_rec_6}
	\begin{aligned}
 	\MoveEqLeft \Aver{\mathcal{T}^1(n,\sigma,\alpha,\tilde{\alpha};\hbar)}
	=  \frac{(\rho\hbar(2\pi)^d)^{2k-2-n}}{(2\pi\hbar)^d} \int_{\R_{+}^{2j}} \int  \boldsymbol{1}_{[0,\hbar^{-1}t]}( \boldsymbol{s}_{1,k}^{+}+ \boldsymbol{t}_{1,a_k}^{+})  \boldsymbol{1}_{[0,\hbar^{-1}t]}( \boldsymbol{\tilde{s}}_{1,k}^{+}+ \boldsymbol{\tilde{t}}_{1,a_k}^{+})    
	\\
	\times & \big(1-\prod_{m=m_1}^{m_2-1} e^{-\hbar^{\frac{3}{7}}\frac{1}{2} s_m }\prod_{j=\tilde{a}_{m-1}+1}^{\tilde{a}_m}  e^{-\hbar^{\frac{3}{7}}\frac{1}{2}  t_j }\big)  \big(1-\prod_{m=m_1}^{m_2-1} e^{-\hbar^{\frac{3}{7}}\frac{1}{2} \tilde{s}_m }\prod_{j=a_{m-1}+1}^{a_m}  e^{-\hbar^{\frac{3}{7}}\frac{1}{2}  \tilde{t}_j }\big)   e^{i  s_{m_2} \frac{1}{2} p_{i_1}^2}     e^{-i  \tilde{s}_{m_2} \frac{1}{2} p_{i_1}^2}  
	 \\
	 \times&    \prod_{i=1}^n  e^{i  (s_{\iota^{*}(\sigma_i)} - \tilde{s}_{\iota^{*}(\sigma_i)})  \frac{1}{2} p_{i}^2}  
	 \prod_{i=0,i\neq i_1}^n \prod_{m=\sigma_i+1}^{\sigma_{i+1}-1}  e^{i  (s_{\iota^{*}(m)}-\tilde{s}_{\iota^{*}(m)}) \frac{1}{2} p_{i}^2}
	      \prod_{m=\sigma_{i_1}+1}^{m_1-1}  e^{i  (s_{\iota^{*}(m)}-\tilde{s}_{\iota^{*}(m)}) \frac{1}{2} p_{i_1}^2}   
      	\\
	\times&
	\prod_{m=m_1}^{m_2-1}  e^{i  s_{\iota^{*}(m)} \frac{1}{2} \tilde{p}^2}  e^{-i  \tilde{s}_{\iota^{*}(m)} \frac{1}{2} \tilde{q}^2}   \prod_{m=m_2}^{\sigma_{i_1+1}-1}  e^{i  (s_{\iota^{*}(m)}-\tilde{s}_{\iota^{*}(m)}) \frac{1}{2} p_{i_1}^2} \prod_{i=1}^{a_k} e^{i  t_{i} \frac{1}{2} \eta_{i}^2}    \prod_{i=1}^{\tilde{a}_k} e^{-i  \tilde{t}_{i} \frac{1}{2} \xi_{i}^2}
	 \\
	 \times& 
	  e^{i  (\hbar^{-1}t- \boldsymbol{s}_{1,k}^{+}- \boldsymbol{t}_{1,a_k}^{+}) \frac{1}{2} p_0^2} e^{-i  (\hbar^{-1}t- \boldsymbol{\tilde{s}}_{1,k}^{+}- \boldsymbol{\tilde{t}}_{1,a_k}^{+}) \frac{1}{2} p_0^2}  \mathcal{G}(\boldsymbol{p},\tilde{p}-p_{i_1},\tilde{q}-p_{i_1},\boldsymbol{\eta},\boldsymbol{\xi})   |\hat{\varphi}(\tfrac{p_0}{\hbar})|^2 \, d\boldsymbol{\eta} d\boldsymbol{\xi}d\boldsymbol{p}d\tilde{p}   d\boldsymbol{t} d\boldsymbol{s} d\boldsymbol{\tilde{t}} d\boldsymbol{\tilde{s}},
	\end{aligned}
\end{equation}
where we also have evaluated the integral in $s_0$ and $\tilde{s}_0$. As in the proof of Lemma~\ref{expansion_aver_bound_Mainterm} we will now apply the ``Fourier method'' to estimate the integrals. So we start by divide into $2^{a_k+\tilde{a}_k}$ cases depending on the size of  the $t$'s and $\tilde{t}$'s. Then for each of this cases we divide into a further $2^{n+2}$ cases depending on the size of linear combinations of the $s$ and $\tilde{s}$ variables.   Preforming this argument we obtain the estimate
\begin{equation}\label{EQ:norm_conv_with_op_1_rec_7}
 	 |\Aver{\mathcal{T}^1(n,\sigma,\alpha,\tilde{\alpha};\hbar)}|
	\leq C^{a_k + \tilde{a}_k + n+2} (\rho\hbar)^{2k-2-n} \norm{\varphi_\hbar}_{L^2(\R^d)}^2 \norm{\hat{V}}_{1,\infty,3d+3}^{|\alpha| + |\tilde{\alpha}|} \mathcal{I}(t,\hbar),
\end{equation}
where
\begin{equation*}
	\begin{aligned}
	\mathcal{I}(t,\hbar) ={}& \int_{\R_{+}^{2j}}
	\frac{\boldsymbol{1}_{[0,\hbar^{-1}t]}( \boldsymbol{s}_{1,k}^{+}+ \boldsymbol{t}_{1,a_k}^{+})  }{\max(1, | \sum_{m=m_1}^{m_2-1} s_{\iota^{*}(m)} | )^\frac{d}{2}}  
	\frac{ \boldsymbol{1}_{[0,\hbar^{-1}t]}( \boldsymbol{\tilde{s}}_{1,k}^{+}+ \boldsymbol{\tilde{t}}_{1,a_k}^{+})}{\max(1, | \sum_{m=m_1}^{m_2-1} \tilde{s}_{\iota^{*}(m)} | )^\frac{d}{2}} 
	\prod_{i=1}^{a_k} \frac{1}{\max(1, | t_i |)^\frac{d}{2}} 
	\prod_{i=1}^{\tilde{a}_k} \frac{1}{\max(1, | \tilde{t}_i |)^\frac{d}{2}} 
	\\
	&\times \frac{\big(1-\prod_{m=m_1}^{m_2-1} e^{-\hbar^{\frac{3}{7}}\frac{1}{2} s_m }\prod_{j=a_{m-1}+1}^{a_m}  e^{-\hbar^{\frac{3}{7}}\frac{1}{2}  t_j }\big)  \big(1-\prod_{m=m_1}^{m_2-1} e^{-\hbar^{\frac{3}{7}}\frac{1}{2} \tilde{s}_m }\prod_{j=\tilde{a}_{m-1}+1}^{\tilde{a}_m}  e^{-\hbar^{\frac{3}{7}}\frac{1}{2}  \tilde{t}_j }\big) }{\max(1, | \sum_{m=\sigma_{i_1}}^{m_1-1} s_{\iota^{*}(\sigma_i)} - \tilde{s}_{\iota^{*}(\sigma_i)} + \sum_{m=m_2}^{\sigma_{i_1+1}-1} s_{\iota^{*}(\sigma_i)} - \tilde{s}_{\iota^{*}(\sigma_i)} +s_{m_3}-\tilde{s}_{m_2} | )^\frac{d}{2}}  
	\\
	&\times   \prod_{i=1,i\neq i_1}^{n} \frac{1}{\max(1, | \sum_{m=\sigma_i}^{\sigma_{i+1}-1} s_{\iota^{*}(m)} - \tilde{s}_{\iota^{*}(m)} |)^\frac{d}{2}} 
	    d\boldsymbol{t} d\boldsymbol{s} d\boldsymbol{\tilde{t}} d\boldsymbol{\tilde{s}}.
	\end{aligned}
\end{equation*}
To estimate the time integrals we  first observe that
\begin{equation*}
	\begin{aligned}
 	\MoveEqLeft \boldsymbol{1}_{[0,\hbar^{-1}t]}( \boldsymbol{s}_{1,k}^{+}+ \boldsymbol{t}_{1,a_k}^{+})  \boldsymbol{1}_{[0,\hbar^{-1}t]}( \boldsymbol{\tilde{s}}_{1,k}^{+}+ \boldsymbol{\tilde{t}}_{1,a_k}^{+}) 
	\\
	\leq{}&  \boldsymbol{1}_{[0,\hbar^{-1}t]}( \boldsymbol{s}_{1,m_1-1}^{+} + \boldsymbol{s}_{m_2,k}^{+})  \boldsymbol{1}_{[0,\hbar^{-1}t]}( \boldsymbol{s}_{m_1,m_2-1}^{+} )  \boldsymbol{1}_{[0,\hbar^{-1}t]}(  \sum_{m=1,m\notin\iota^*(\sigma)}^{k}\tilde{s}_{m}- \boldsymbol{\tilde{s}}_{m_1,m_2-1}^{+} ) \boldsymbol{1}_{[0,\hbar^{-1}t]}( \boldsymbol{\tilde{s}}_{m_1,m_2-1}^{+} ).
	\end{aligned}
\end{equation*}
After applying this inequality we evaluate all integrals in $s$ and $\tilde{s}$ except $s_{m_1},\dots,s_{m_2-1}$ and $\tilde{s}_{m_1},\dots,\tilde{s}_{m_2-1}$. Moreover we also use that $\iota^*(m)=m$ for all $m\leq m_2-1$. This gives the bound
\begin{equation*}
	\begin{aligned}
	\mathcal{I}(t,\hbar) 
	\leq{}& 
	\frac{C(\hbar^{-1}t)^{2k-2(m2-m_1)-n}}{(k-m_2+m_1)! (k-m_2+m_1-n)!} 
	\\
	&\times \int_{\R_{+}^{a_k+m_2-m_1}}
	\boldsymbol{1}_{[0,\hbar^{-1}t]}( \boldsymbol{s}_{m_1,m_2-1}^{+} )
	 \tfrac{ \big(1-\prod_{m=m_1}^{m_2-1} e^{-\hbar^{\frac{3}{7}}\frac{1}{2} s_m }\prod_{j=a_{m-1}+1}^{a_m}  e^{-\hbar^{\frac{3}{7}}\frac{1}{2}  t_j }\big)}   {\max(1, | \sum_{m=m_1}^{m_2-1} s_{m} | )^\frac{d}{2}} 
	 \prod_{i=1}^{a_k} \tfrac{1}{\max(1, | t_i |)^\frac{d}{2}}    d\boldsymbol{t} d\boldsymbol{s}
	\\
	&\times
	 \int_{\R_{+}^{\tilde{a}_k +m_2-m_1}}
	 \boldsymbol{1}_{[0,\hbar^{-1}t]}( \boldsymbol{\tilde{s}}_{m_1,m_2-1}^{+} )
	\tfrac{\big(1-\prod_{m=m_1}^{m_2-1} e^{-\hbar^{\frac{3}{7}}\frac{1}{2} \tilde{s}_m }\prod_{j=\tilde{a}_{m-1}+1}^{\tilde{a}_m}  e^{-\hbar^{\frac{3}{7}}\frac{1}{2}  \tilde{t}_j }\big) }{\max(1, | \sum_{m=m_1}^{m_2-1} \tilde{s}_{m} | )^\frac{d}{2}} 
	\prod_{i=1}^{\tilde{a}_k} \tfrac{1}{\max(1, | \tilde{t}_i |)^\frac{d}{2}} 
	   d\boldsymbol{\tilde{t}} d\boldsymbol{\tilde{s}}.
	\end{aligned}
\end{equation*}
Using the identity for difference given in \eqref{def_of_prod} of products and that $e^{-t}\leq1$ for all $t\geq0$ we get that
\begin{equation*}
	\begin{aligned}
	\big(1-\prod_{m=m_1}^{m_2-1} e^{-\hbar^{\frac{3}{7}}\frac{1}{2} s_m }\prod_{j=a_{m-1}+1}^{a_m}  e^{-\hbar^{\frac{3}{7}}\frac{1}{2}  t_j }\big) \leq \sum_{m=m_1}^{m_2-1} 1- e^{-\hbar^{\frac{3}{7}}\frac{1}{2} s_m } + \sum_{m=m_1}^{m_2-1}\sum_{j=a_{m-1}+1}^{a_m} 1-e^{-\hbar^{\frac{3}{7}}\frac{1}{2}  t_j }.
	\end{aligned}
\end{equation*}
Applying this inequality and Lemma~\ref{LE:gamma_reg_time_int} we get that
\begin{equation*}
	\begin{aligned}
	\MoveEqLeft \int_{\R_{+}^{a_k+m_2-m_1}}
	\boldsymbol{1}_{[0,\hbar^{-1}t]}( \boldsymbol{s}_{m_1,m_2-1}^{+} )
	 \frac{ \big(1-\prod_{m=m_1}^{m_2-1} e^{-\hbar^{\frac{3}{7}}\frac{1}{2} s_m }\prod_{j=a_{m-1}+1}^{a_m}  e^{-\hbar^{\frac{3}{7}}\frac{1}{2}  t_j }\big)}   {\max(1, | \sum_{m=m_1}^{m_2-1} s_{\iota^{*}(m)} | )^\frac{d}{2}} 
	 \prod_{i=1}^{a_k} \frac{1}{\max(1, | t_i |)^\frac{d}{2}}    d\boldsymbol{t} d\boldsymbol{s}
	 \\
	 \leq {}& \sum_{m=m_1}^{m_2-1} \int_{\R_{+}^{a_k+m_2-m_1}} \boldsymbol{1}_{[0,\hbar^{-1}t]}( \boldsymbol{s}_{m_1,m_2-1}^{+} -s_m)
	 \frac{ 1- e^{-\hbar^{\frac{3}{7}}\frac{1}{2} s_m }}   {\max(1, | s_m | )^\frac{d}{2}} 
	 \prod_{i=1}^{a_k} \frac{1}{\max(1, | t_i |)^\frac{d}{2}}    d\boldsymbol{t} d\boldsymbol{s}
	 \\
	 &+  \sum_{m=m_1}^{m_2-1}\sum_{j=a_{m-1}+1}^{a_m} \int_{\R_{+}^{a_k+m_2-m_1}}
	 \boldsymbol{1}_{[0,\hbar^{-1}t]}( \boldsymbol{s}_{m_1+1,m_2-1}^{+} )
	 \frac{ 1-e^{-\hbar^{\frac{3}{7}}\frac{1}{2}  t_j }}   {\max(1, |  s_{m_1} | )^\frac{d}{2}} 
	 \prod_{i=1}^{a_k} \frac{1}{\max(1, | t_i |)^\frac{d}{2}}    d\boldsymbol{t} d\boldsymbol{s}
	 \\
	 \leq{}& C^{|\alpha|} \hbar^{\frac17}\frac{(\hbar^{-1}t)^{m_2-m_1-1}}{(m_2-m_1-1)!}\sum_{m=1}^k \alpha_m,
	\end{aligned}
\end{equation*}
where we have used that we at most can have $|\alpha|$ terms and that for any $m\in\{m_1,\dots,m_2-1\}$ we have that
\begin{equation*}
	\frac{1}{\max(1, | \sum_{m=m_1}^{m_2-1} s_{\iota^{*}(m)} | )^\frac{d}{2}}  \leq\frac{1} {\max(1, | s_m| )^\frac{d}{2}} .
\end{equation*}
We get an analogous estimate for the integral over the remaining $\tilde{s}$ and $\tilde{t}$. Using these estimates we obtain that
\begin{equation}\label{EQ:norm_conv_with_op_1_rec_8}
	\begin{aligned}
	\mathcal{I}(t,\hbar) 
	\leq{}&C^{|\alpha|+|\tilde{\alpha}|} \hbar^{\frac27} 
	\frac{(\hbar^{-1}t)^{2k-2-n}}{(k-m_2+m_1)! (k-m_2+m_1-n)!(m_2-m_1-1)!(m_2-m_1-1)!}  \sum_{m=1}^k \alpha_m \sum_{m=1}^k \tilde{\alpha}_m
	\\
	\leq{}&  C^{|\alpha|+|\tilde{\alpha}|} \hbar^{\frac27} 
	\frac{ 2^{2k-2-n}(\hbar^{-1}t)^{2k-2-n}}{(k-1)! (k-1-n)!}  \sum_{m=1}^k \alpha_m \sum_{m=1}^k \tilde{\alpha}_m.
	\end{aligned}
\end{equation}
From combining \eqref{EQ:norm_conv_with_op_1_rec_7} and \eqref{EQ:norm_conv_with_op_1_rec_8} we obtain that
\begin{equation}\label{EQ:norm_conv_with_op_1_rec_9}
 	 |\Aver{\mathcal{T}^1(n,\sigma,\alpha,\tilde{\alpha};\hbar)}|
	\leq (C_d\norm{\hat{V}}_{1,\infty,3d+3})^{|\alpha|+|\tilde{\alpha}|} \hbar^{\frac27}  \norm{\varphi_\hbar}_{L^2(\R^d)}^2   \frac{ (2\rho t)^{2k-2-n}}{(k-1)! (k-1-n)!} \sum_{m=1}^k \alpha_m \sum_{m=1}^k \tilde{\alpha}_m.
\end{equation}
From \eqref{EQ:norm_conv_with_op_1_rec_1}, \eqref{EQ:norm_conv_with_op_1_rec_5} and \eqref{EQ:norm_conv_with_op_1_rec_9} we obtain the estimate
\begin{equation*}
	\begin{aligned}
	   \MoveEqLeft \mathbb{E}\Big[\big\lVert  \sum_{\boldsymbol{x}\in\mathcal{X}_{\neq}^{k-1}}  \mathcal{I}_{1,1}^1(k,\boldsymbol{x},t;\hbar) \varphi_\hbar\big\rVert_{L^2(\R^d)}^2 \Big]
	    \leq 2 \sum_{\alpha,\tilde{\alpha}\in \N^k}\lambda^{|\alpha|+|\tilde{\alpha}|}   \sum_{n=0}^{k-1} \binom{k-1}{n} n!  
	    \\
	    &\times  \sum_{\sigma\in\mathcal{A}(k-1,n)}\Big\{ C_d^{a_k+\tilde{a}_k}  \left(\norm{\hat{V}}_{1,\infty,5d+5}\right)^{|\alpha|+|\tilde{\alpha}|} 
	  \hbar |\log(\zeta)|^{n+4} \frac{\rho(\rho t)^{2k-n-3} }{(k-2)!(k-n-1)!} \norm{\varphi_\hbar}_{\mathcal{H}^{2d+2}_\hbar(\R^d)}^2
	  \\
	  &+(C_d\norm{\hat{V}}_{1,\infty,3d+3})^{|\alpha|+|\tilde{\alpha}|} \hbar^{\frac27}  \norm{\varphi_\hbar}_{L^2(\R^d)}^2   \frac{ (2\rho t)^{2k-2-n}}{(k-1)! (k-1-n)!} \sum_{m=1}^k \alpha_m \sum_{m=1}^k \tilde{\alpha}_m\Big\}.
	\end{aligned}
\end{equation*}
We can estimate all these sums as in the previous proofs and obtain that
\begin{equation}\label{EQ:norm_conv_with_op_1_rec_10}
	\begin{aligned}
	   \MoveEqLeft \mathbb{E}\Big[\big\lVert  \sum_{\boldsymbol{x}\in\mathcal{X}_{\neq}^{k-1}}  \mathcal{I}_{1,1}^1(k,\boldsymbol{x},t;\hbar) \varphi_\hbar\big\rVert_{L^2(\R^d)}^2 \Big]
	    \leq (C^{k}k  \hbar |\log(\hbar)|^{k+4}  + C^{k}k^2 \hbar^{\frac27}) \sup_{\hbar\in\mathcal{I}} \norm{\varphi}_{\mathcal{H}^{2d+2}_\hbar(\R^d)}^2.
	\end{aligned}
\end{equation}
We now turn to the case where $j=2$. Here we will again consider the two different cases are given by if there is at least one $\sigma_i$ such that $m_1 <\sigma_i <m_2$ or not. For the case where such an $\sigma_i$ we can again argue as in the proof of Lemma~\ref{expansion_aver_bound_Mainterm_rec} end obtain the bound
 \begin{equation}\label{EQ:norm_conv_with_op_1_rec_11}
   	\begin{aligned}
	 \Aver{\mathcal{T}^2(n,\sigma,\alpha,\tilde{\alpha};\hbar) }
	 \leq {}&   C_d^{a_k+\tilde{a}_k}  \left(\norm{\hat{V}}_{1,\infty,5d+5}\right)^{|\alpha|+|\tilde{\alpha}|} 
	  \hbar |\log(\tfrac{\hbar}{t})|^{n+4} \frac{\rho(\rho t)^{2k-n-3} }{(k-2)!(k-n-1)!} \norm{\varphi}_{\mathcal{H}^{2d+2}_\hbar(\R^d)}^2.
	\end{aligned}
\end{equation}
For the other case we again have there exists an $i_1$ such that $\sigma_{i_1}< m_1 <m_2\leq \sigma_{i_1+1}$. Observe that for this to be true then $n\leq k-1-m_2+m_1$. For the previous estimates this observation was not crucial but it will be now. We write up the expression for $\mathcal{I}_{1,1}^1(k,\boldsymbol{x},\alpha,\iota,t;\hbar) \varphi(x)$ and preform the change of variables $p_{m_2-1}\mapsto p_{m_2-1} - p_{m_2}$ and $p_{m}\mapsto p_{m} - p_{m_2-1}$ for all $m\in\{m_1,\dots,m_2-2\}$, where $m_1$ and $m_2$ are the numbers associated to the map $\iota$. Moreover, we also do a relabelling by changing $p_{m_2-1}$ into $\tilde{p}$ and $p_m$ into $p_{m-1}$ for all $m\in\{m_2,\dots,k\}$. This yields the form
\begin{equation*}
	\begin{aligned}
 	\MoveEqLeft \mathcal{I}_{1,1}^2(n,\boldsymbol{x},\alpha,\iota,t;\hbar) \varphi(x) 
	=  \frac{1}{(2\pi\hbar)^d\hbar^{k}} \int_{[0,t]_{\leq}^k} \int
	 e^{ i   \langle  \hbar^{-1}x,p_{k-1} \rangle}  
	 e^{i  (s_{m_1}-s_{m_2}) \frac{1}{2}\hbar^{-1} \tilde{p}^2 } e^{-\hbar^{-\frac{4}{7}}\frac{1}{2}(s_{m_1}-s_{m_2})} 
	 \\
	 \times& 
	 e^{i \hbar^{-1}  \langle \tilde{p},s_{m_1} p_{m_1}-s_{m_2}p_{m_2-1}  \rangle}
	 e^{i \hbar^{-1}\sum_{i=m_1+1}^{m_2-1} s_i \langle \tilde{p}, p_i-p_{i-1}  \rangle}  \prod_{m=1}^{k-1}  e^{ -i   \langle  \hbar^{-1/d}x_{m},p_{m}-p_{m-1}\rangle} e^{i  s_{\iota^{*}(m)} \frac{1}{2}\hbar^{-1} (p_{m}^2-p_{m-1}^2)} 
	 \\
	 \times&  \prod_{m=1}^{k-1}  \Psi_{\alpha_{\iota^{*}(m)}}(p_{m} +\pi_{m}^1(\tilde{p})  ,p_{m-1}+ \pi_{m}^2(\tilde{p}) ,\hbar^{-1}(s_{\iota^{*}(m)-1}-s_{\iota^{*}(m)}); V) 
	  \\
	  \times&   \Psi_{\alpha_{m_2}}(p_{m_2-1}  ,\tilde{p}+p_{m_2-1} ,\hbar^{-1}(s_{m_2-1}-s_{m_2}); V)  e^{i  t \frac{1}{2}\hbar^{-1} p_0^2} \hat{\varphi}(\tfrac{p_0}{\hbar})  \,d\boldsymbol{p} d\tilde{p} d\boldsymbol{s}.
	\end{aligned}
\end{equation*}
Note the we here have used the form obtained in the proof for the Duhamel expansion. We can again use this to write an expression for $\mathcal{T}^2(n,\sigma,\alpha,\tilde{\alpha};\hbar)$. This expression will be similar to the one considered in the proof of Lemma~\ref{expansion_aver_bound_Mainterm}. Hence we can preform an argument similar to the one used there. Where the main difference will be in how we handle the integrals in the time variables $s$ and $\tilde{s}$. The main ideas is to first divide into cases depending on the size of the $t$'s and $\tilde{t}$'s. Then for each of these terms we use Lemma~\ref{app_quadratic_integral_tech_est} in the variables $p_1,\dots,p_n$. The main difference to the argument used in Lemma~\ref{expansion_aver_bound_Mainterm} is in the estimate done in \eqref{EQ:exp_aver_b_Mainterm_5}. For the current case we need to estimate the following expression
  \begin{equation}\label{EQ:norm_conv_with_op_1_rec_12}
	\begin{aligned}
	 \int_{[0,t]_{\leq}^k}\int_{[0,t]_{\leq}^{k}}   \tfrac{e^{-\hbar^{-\frac{4}{7}}\frac{1}{2}(s_{m_1}-s_{m_2})} e^{-\hbar^{-\frac{4}{7}}\frac{1}{2}(\tilde{s}_{m_1}-\tilde{s}_{m_2})} }{\max(1, \hbar^{-1}| s_{\iota^*(\sigma_n)}- \tilde{s}_{\iota^*(\sigma_n)} +l^1_n(\boldsymbol{t}, \boldsymbol{\tilde{t}}) |)^\frac{d}{2}} 
	\prod_{i=1}^{n-1} \tfrac{1}{\max(1,\hbar^{-1} | s_{\iota^*(\sigma_i)}- \tilde{s}_{\iota^*(\sigma_i)}-s_{\iota^*(\sigma_{i+1})}+ \tilde{s}_{\iota^*(\sigma_{i+1})} +l^1_i(\boldsymbol{t}, \boldsymbol{\tilde{t}}) |)^\frac{d}{2}}  \, d\boldsymbol{\tilde{s}} d\boldsymbol{s},
	\end{aligned}
\end{equation} 
where we due to our assumptions have that ${\iota^*(\sigma_i)}\notin\{m_1,\dots,m_2\}$ for all $i$ and $l^1_i(\boldsymbol{t}, \boldsymbol{\tilde{t}})$ are a linear function of the the variables $\boldsymbol{t}$ and $\boldsymbol{\tilde{t}}$. Moreover the the linear function $l^1_{i_1}(\boldsymbol{t}, \boldsymbol{\tilde{t}})$ will also depend on $s_m$ and $\tilde{s}_m$ for $m\in\{m_1,\dots,m_2\}$. In order to estimate \eqref{EQ:norm_conv_with_op_1_rec_12} we firstly observe that
  \begin{equation*}
	\begin{aligned}
	 \boldsymbol{1}_{[0,t]_{\leq}^k}(s_k,\dots,s_1)  \leq \boldsymbol{1}_{[0,t]_{\leq}^{k-m_2+m_1-1}}(s_k,\dots,s_{m_2+1},s_{m_1-1},\dots,s_1) \boldsymbol{1}_{[0,t]_{\leq}^{m_2-m_1+1}}(s_{m_2},\dots,s_{m_1}). 
	\end{aligned}
\end{equation*} 
We use the same estimate for $\boldsymbol{1}_{[0,t]_{\leq}^k}(\tilde{s}_k,\dots,\tilde{s}_1)$ but here we also exclude $\tilde{s}_{\iota^*(\sigma_i)}$ for all $i\in\{1,\dots,n\}$. Applying these observations/estimates and integrating all variables except $s_m$ and $\tilde{s}_m$ for $m\in\{m_1,\dots,m_2\}$ we get that
  \begin{equation}\label{EQ:norm_conv_with_op_1_rec_14}
	\begin{aligned}
	\MoveEqLeft  \int_{[0,t]_{\leq}^k}\int_{[0,t]_{\leq}^{k}}   \tfrac{e^{-\hbar^{-\frac{4}{7}}\frac{1}{2}(s_{m_1}-s_{m_2})} e^{-\hbar^{-\frac{4}{7}}\frac{1}{2}(\tilde{s}_{m_1}-\tilde{s}_{m_2})} }{\max(1, \hbar^{-1}| s_{\iota^*(\sigma_n)}- \tilde{s}_{\iota^*(\sigma_n)} +l^1_n(\boldsymbol{t}, \boldsymbol{\tilde{t}}) |)^\frac{d}{2}} 
	\prod_{i=1}^{n-1} \tfrac{1}{\max(1,\hbar^{-1} | s_{\iota^*(\sigma_i)}- \tilde{s}_{\iota^*(\sigma_i)}-s_{\iota^*(\sigma_{i+1})}+ \tilde{s}_{\iota^*(\sigma_{i+1})} +l^1_i(\boldsymbol{t}, \boldsymbol{\tilde{t}}) |)^\frac{d}{2}}  \, d\boldsymbol{\tilde{s}} d\boldsymbol{s}
	\\
	\leq{}&C^n \tfrac{\hbar^n t^{2k -2m_2+2m_1-2-n}}{(k-m_2+m_1-1)!(k -m_2+m_1-1-n)!} \Big(\int_{[0,t]_{\leq}^{m_2-m_1+1}}   e^{-\hbar^{-\frac{4}{7}}\frac{1}{2}(s_{m_1}-s_{m_2})} \, d\boldsymbol{s}\Big)^2.
	\end{aligned}
\end{equation} 
In order to estimate the last integral we start by integrating all variables $s_i$ with $m_1<i<m_2$, observe that there is at least one such index. This gives us
  \begin{equation*}
	\begin{aligned}
	\int_{[0,t]_{\leq}^{m_2-m_1+1}}   e^{-\hbar^{-\frac{4}{7}}\frac{1}{2}(s_{m_1}-s_{m_2})} \, d\boldsymbol{s} 
	=\tfrac{1}{(m_2-m_1-1)!} \int_{\R^{2}}  \boldsymbol{1}_{[0,t]_{\leq}^{2}}(s_{m_2}, s_{m_1}) (s_{m_1}-s_{m_2})^{m_2-m_1-1} e^{-\hbar^{\frac{4}{7}}\frac{1}{2}s_{m_1}} \, d\boldsymbol{s}.
	\end{aligned}
\end{equation*} 
 We then preform the change of variables $s_{m_1}\mapsto \hbar^{-\frac{4}{7}}(s_{m_1} -s_{m_2})$. With this change of variables we get
  \begin{equation}\label{EQ:norm_conv_with_op_1_rec_15}
	\begin{aligned}
	\int_{[0,t]_{\leq}^{m_2-m_1+1}}   e^{-\hbar^{-\frac{4}{7}}\frac{1}{2}(s_{m_1}-s_{m_2})} \, d\boldsymbol{s}  
	&=  \tfrac{\hbar^{\frac{4}{7} (m_2-m_1)}}{(m_2-m_1-1)!} \int_{\R^{2}}  \boldsymbol{1}_{[0,t]_{\leq}^{2}}(s_{m_2}, \hbar^{\frac{4}{7}}s_{m_1} +s_{m_2}) s_{m_1}^{m_2-m_1-1} e^{-\frac{1}{2}s_{m_1}} \, d\boldsymbol{s}
	\\
	&\leq  \tfrac{\hbar^{\frac{4}{7} (m_2-m_1)} t }{(m_2-m_1-1)!} \int_{0}^\infty  s_{m_1}^{m_2-m_1-1} e^{-\frac{1}{2}s_{m_1}} \, ds_{m_1} = \tfrac{\hbar^{\frac{4}{7} (m_2-m_1)} t }{2^{m_2-m_1-1}}.
	\end{aligned}
\end{equation} 
Combining the estimates in \eqref{EQ:norm_conv_with_op_1_rec_14} and \eqref{EQ:norm_conv_with_op_1_rec_15} we obtain that
  \begin{equation}\label{EQ:norm_conv_with_op_1_rec_16}
	\begin{aligned}
	\MoveEqLeft  \int_{[0,t]_{\leq}^k}\int_{[0,t]_{\leq}^{k}}   \tfrac{e^{-\hbar^{-\frac{4}{7}}\frac{1}{2}(s_{m_1}-s_{m_2})} e^{-\hbar^{-\frac{4}{7}}\frac{1}{2}(\tilde{s}_{m_1}-\tilde{s}_{m_2})} }{\max(1, \hbar^{-1}| s_{\iota^*(\sigma_n)}- \tilde{s}_{\iota^*(\sigma_n)} +l^1_n(\boldsymbol{t}, \boldsymbol{\tilde{t}}) |)^\frac{d}{2}} 
	\prod_{i=1}^{n-1} \tfrac{1}{\max(1,\hbar^{-1} | s_{\iota^*(\sigma_i)}- \tilde{s}_{\iota^*(\sigma_i)}-s_{\iota^*(\sigma_{i+1})}+ \tilde{s}_{\iota^*(\sigma_{i+1})} +l^1_i(\boldsymbol{t}, \boldsymbol{\tilde{t}}) |)^\frac{d}{2}}  \, d\boldsymbol{\tilde{s}} d\boldsymbol{s}
	\\
	\leq{}&C^n \tfrac{\hbar^{n+\frac{4}{7} (m_2-m_1)} t^{2k -2m_2+2m_1-n}}{(k-m_2+m_1-1)!(k -m_2+m_1-1-n)! 4^{m_2-m_1-1}}. 
	\end{aligned}
\end{equation} 
Using the estimate obtained in \eqref{EQ:norm_conv_with_op_1_rec_16} and by arguing as in the proof of Lemma~\ref{expansion_aver_bound_Mainterm} we obtain that
   \begin{equation}\label{EQ:norm_conv_with_op_1_rec_17}
   	\begin{aligned}
	 \Aver{\mathcal{T}(n,\alpha,\tilde{\alpha},\sigma^1,\sigma^2,\mathrm{id})}
	 \leq{}&  C_d^{a_k+\tilde{a}_k+n}  \norm{\varphi}_{L^2(\R^d)}^2 \left(\norm{\hat{V}}_{1,\infty,3d+3}\right)^{|\alpha|+|\tilde{\alpha}|}  
	 \\
	 &\times \frac{ \rho^{2k-n} \hbar^{\frac{8}{7} (m_2-m_1)} t^{2k -2m_2+2m_1-n} }{(k-m_2+m_1-1)!(k -m_2+m_1-1-n)! 4^{m_2-m_1-1}}.
	\end{aligned}
\end{equation}
Combining the estimates in \eqref{EQ:norm_conv_with_op_1_rec_1},  \eqref{EQ:norm_conv_with_op_1_rec_11}, \eqref{EQ:norm_conv_with_op_1_rec_16} and recalling that $n\leq k-1-m_2+m_1$ for the cases where there exists $i_1$ such that $\sigma_{i_1}< m_1 <m_2\leq \sigma_{i_1+1}$ we obtain that
\begin{equation*}
	\begin{aligned}
	   \MoveEqLeft \mathbb{E}\Big[\big\lVert  \sum_{\boldsymbol{x}\in\mathcal{X}_{\neq}^{k-1}}  \mathcal{I}_{1,1}^1(k,\boldsymbol{x},t;\hbar) \varphi_\hbar\big\rVert_{L^2(\R^d)}^2 \Big]
	    \leq 2 \sum_{\alpha,\tilde{\alpha}\in \N^k}\lambda^{|\alpha|+|\tilde{\alpha}|}   \sum_{n=0}^{k-1} \binom{k-1}{n} n!  
	    \\
	    &\times  \sum_{\sigma\in\mathcal{A}(k-1,n)}\Big\{ C_d^{a_k+\tilde{a}_k}  \left(\norm{\hat{V}}_{1,\infty,5d+5}\right)^{|\alpha|+|\tilde{\alpha}|} 
	  \hbar |\log(\zeta)|^{n+4} \frac{\rho(\rho t)^{2k-n-3} }{(k-2)!(k-n-1)!} \norm{\varphi_\hbar}_{\mathcal{H}^{2d+2}_\hbar(\R^d)}^2
	  \\
	  &+ (C_d\norm{\hat{V}}_{1,\infty,3d+3})^{|\alpha|+|\tilde{\alpha}|}  \norm{\varphi_\hbar}_{L^2(\R^d)}^2   \frac{ \boldsymbol{1}_{\{n\leq  k-1-m_2+m_1\}} \rho^{2k-n} \hbar^{\frac{8}{7} (m_2-m_1)} t^{2k -2m_2+2m_1-n} }{(k-m_2+m_1-1)!(k -m_2+m_1-1-n)! 4^{m_2-m_1-1}}\Big\}.
	\end{aligned}
\end{equation*}
Using that for the second term we have that $n\leq k-1-m_2+m_1$ we obtain the estimate
\begin{equation}\label{EQ:norm_conv_with_op_1_rec_18}
	\begin{aligned}
	   \MoveEqLeft \mathbb{E}\Big[\big\lVert  \sum_{\boldsymbol{x}\in\mathcal{X}_{\neq}^{k-1}}  \mathcal{I}_{1,1}^1(k,\boldsymbol{x},t;\hbar) \varphi_\hbar\big\rVert_{L^2(\R^d)}^2 \Big]
	    \leq (C^{k}k  \hbar |\log(\hbar)|^{k+4}  + C^{k}k^2 \hbar^{\frac{16}{7}}) \sup_{\hbar\in\mathcal{I}} \norm{\varphi}_{\mathcal{H}^{2d+2}_\hbar(\R^d)}^2,
	\end{aligned}
\end{equation}
where this is found using the same arguments as previously. From \eqref{EQ:norm_conv_with_op_1_rec_0}, \eqref{EQ:norm_conv_with_op_1_rec_10} and \eqref{EQ:norm_conv_with_op_1_rec_18} we obtain the estimate
\begin{equation}\label{EQ:norm_conv_with_op_1_rec_19}
	\begin{aligned}
	   \MoveEqLeft  \sum_{k=3}^{k_0} \sum_{\iota\in\mathcal{Q}_{k,1,1}}  \mathbb{E}\Big[\big\lVert  \sum_{\boldsymbol{x}\in\mathcal{X}_{\neq}^{k-1}}  \mathcal{I}_{1,1}(k,\boldsymbol{x},t;\hbar) \varphi_\hbar\big\rVert_{L^2(\R^d)}^2 \Big]^\frac{1}{2} 
	   \\
	   \leq{}&  \sum_{k=3}^{k_0} \sum_{\iota\in\mathcal{Q}_{k,1,1}} \sqrt{ (C^{k}k  \hbar |\log(\hbar)|^{k+4}  + C^{k}k^2 \hbar^{\frac27}) \sup_{\hbar\in\mathcal{I}} \norm{\varphi}_{\mathcal{H}^{2d+2}_\hbar(\R^d)}^2} 
	   \\
	   &+  \sum_{k=3}^{k_0} \sum_{\iota\in\mathcal{Q}_{k,1,1}} \sqrt{ (C^{k}k  \hbar |\log(\hbar)|^{k+4}  + C^{k}k^2 \hbar^{\frac{16}{7}}) \sup_{\hbar\in\mathcal{I}} \norm{\varphi}_{\mathcal{H}^{2d+2}_\hbar(\R^d)}^2}
	   \\
	   \leq{}& \sqrt{ \sup_{\hbar\in\mathcal{I}} \norm{\varphi}_{\mathcal{H}^{2d+2}_\hbar(\R^d)}^2} \Big\{ \sqrt{ (C^{k_0}k_0^7  \hbar |\log(\hbar)|^{k_0+4}  + C^{k_0}k_0^9 \hbar^{\frac27}) }  +  \sqrt{ (C^{k_0}k_0^7  \hbar |\log(\hbar)|^{k_0+4}  + C^{k_0}k_0^9 \hbar^{\frac{16}{7}})  }\Big\}.
	\end{aligned}
\end{equation}
Our assumptions on $k_0$ and the estimate obtained in \eqref{EQ:norm_conv_with_op_1_rec_19} give us that
\begin{equation*}
	\begin{aligned}
	\lim_{\hbar\rightarrow 0}   \sum_{k=3}^{k_0} \sum_{\iota\in\mathcal{Q}_{k,1,1}}     \mathbb{E}\Big[\big\lVert  \sum_{\boldsymbol{x}\in\mathcal{X}_{\neq}^{k-1}}  \mathcal{I}_{1,1}(k,\boldsymbol{x},t;\hbar) \varphi_\hbar\big\rVert_{L^2(\R^d)}^2 \Big]^{\frac{1}{2}}= 0.
	\end{aligned}
\end{equation*}
As desired. This concludes the proof.
\end{proof}
\section{Semiclassical limits for terms without recollisions}\label{Sec:limit_main_terms}
Before we start considering the term wise limits and the full limits we will need the following bounds for the function $ \Psi_{m}^\gamma(p_m,p_0,\infty;V) $. In this Lemma it is important that $\gamma>0$.
\begin{lemma}\label{est_der_Psi_gamma}
	Let $m$  in $\N$, $\gamma\in(0,1]$ and $V\in\mathcal{S}(\R^d)$ be given and consider the function  $\Psi_{m}^\gamma(p_m,p_0,\infty ; V)$ defined in Definition~\ref{def_reg_operator}. The function is smooth in the variables $p_m$ and $p_0$, and for all all $\alpha$, $\beta$, $\delta$ and $\epsilon$ in $\N^d_0$ there exists a constant $C$ such that
\begin{equation*}
	\begin{aligned}
	\sup_{q,p_0,p_m \in \R^d} |  \partial_{p_m}^\alpha \partial_{p_0}^\beta \partial_{q}^\delta\Psi_{m}^\gamma(p_m+q,p_0+q,\infty ; V)|  &\leq C,
	\\
	\sup_{q\in \R^d} \int_{\R^d} | p_0^{\epsilon}  \partial_{p_m}^\alpha \partial_{p_0}^\beta \partial_{q}^\delta\Psi_{m}^\gamma(p_m+q,p_0+q,\infty ; V)| \, dp_0  &\leq C\sum_{\epsilon_1 \leq \epsilon} |p_m^{\epsilon_1}|,
	\\
	\sup_{p_0\in \R^d}  | p_0^{\alpha}   \partial_{p_0}^\beta \Psi_{m}^\gamma(p_m,p_0,\infty ; V)|  & \leq C \sum_{\alpha_1 \leq \alpha} |p_m^{\alpha_1}| 
	\\
	 \int_{\R^d} | p_0^{\alpha}   \partial_{p_0}^\beta \Psi_{m}^\gamma(p_m,p_0,\infty ; V)| \, dp_0  &\leq C \sum_{\alpha_1 \leq \alpha} |p_m^{\alpha_1}|.
	 \end{aligned}
\end{equation*}
In all estimates the role of $p_0$ and $p_m$ can be switched.
\end{lemma}
\begin{proof}
Note that in the case $m=1$ we have by definition that
\begin{equation*}
	\Psi_{1}^\gamma(p_1+q,p_0+q,\infty ; V) = \hat{V}(p_1-p_0).
\end{equation*}
Thus the results follow directly from our assumptions on $V$ in this case. We now assume $m\geq2$.
First we observe that via the change of variables $p_i \mapsto p_i-p_0-q$ for all $i \in \{1,\dots,m-1\}$ we can write the function $\Psi_{m}^\gamma(p_m,p_0,\infty ; V)$ as
\begin{equation*}
	\begin{aligned}
	 \MoveEqLeft \Psi_{m}^\gamma(p_m+q,p_0+q,\infty ; V)
	 \\
	={}& \frac{1}{(2\pi)^{dm}}   \int_{\R_{+}^{m-1}}  \int    \hat{V}(p_{1})\prod_{i=2}^{m-1} \hat{V}(p_{i}-p_{i-1})
	 \hat{V}(p_m-p_{m-1}-p_0) \prod_{i=1}^{m-1}  e^{i t_i \frac{1}{2} (p_{i}^2+2\langle p_i,p_{0}+q\rangle+i\gamma)}\, d\boldsymbol{p}_{1,m-1}  d\boldsymbol{t}.
	\end{aligned}
\end{equation*}
The derivatives are then given by
\begin{equation*}
	\begin{aligned}
	\MoveEqLeft \partial_{p_m}^\alpha \partial_{p_0}^\beta\partial_{q}^\delta \Psi_{m}^\gamma(p_m+q,p_0+q,\infty ; V)
	\\
	={}& \frac{1}{(2\pi)^{dm}}   \int_{\R_{+}^{m-1}}  \int     \hat{V}(p_{1})\prod_{i=2}^{m-1} \hat{V}(p_{i}-p_{i-1})  \prod_{i=1}^{m-1}  e^{i t_i \frac{1}{2} (p_{i}^2+i\gamma)}
  	\\
	&\times 
	 \partial_{p_0}^\beta \big[  \partial_{p_m}^\alpha \hat{V}(p_m-p_{m-1}-p_0) \partial_{q}^\delta   e^{i \langle\sum_{i=1}^{m-1} t_i p_i,p_{0}+q\rangle} \big] \, d\boldsymbol{p}_{1,m-1}  d\boldsymbol{t}
	 \\
	 ={}& \frac{1}{(2\pi)^{dm}} \sum_{\beta_1\leq \beta} \binom{\beta}{\beta_1}i^{|\beta_1|}  \int_{\R_{+}^{m-1}}  \int   \hat{V}(p_{1})\prod_{i=2}^{m-1} \hat{V}(p_{i}-p_{i-1})  \prod_{i=1}^{m-1}  e^{i t_i \frac{1}{2} (p_{i}^2+i\gamma)}
  	\\
	&\times 
	 \big[  \partial_1^{\alpha+\beta-\beta_1} \hat{V}](p_m-p_{m-1}-p_0)  \Big(i\sum_{i=1}^{m-1} t_i p_i \Big)^{\beta_1+\delta}    e^{i \langle\sum_{i=1}^{m-1} t_i p_i,p_{0}+q\rangle}  \, d\boldsymbol{p}_{1,m-1}  d\boldsymbol{t}.
	\end{aligned}
\end{equation*}
If we just consider the sum $ (\sum_{i=1}^{m-1} t_i p_i )^{\beta_1}$, then we have that
\begin{equation*}
	\begin{aligned}
	 \Big(\sum_{i=1}^{m-1} t_i p_i \Big)^{\beta_1+\delta} = \sum_{\epsilon_1+\cdot+\epsilon_{m-1}=\beta_1+\delta} \prod_{i=1}^{m-1} \binom{(\beta_{1}+\delta)_i}{\epsilon_{1,i}\cdots\epsilon_{m-1,i}} t_i^{|\epsilon_i|} p_i^{\epsilon_i},
	\end{aligned}
\end{equation*}
where the notation $(\beta_{1}+\delta)_i$ is the $i$'th coordinate of $\beta_1+\delta$ and  $\epsilon_{j,i}$ is the $i$'th coordinate of $\epsilon_j$.  With this notation we get the following bound
\begin{equation}\label{psi_est_2_1}
	\begin{aligned}
	\MoveEqLeft |\partial_{p_m}^\alpha \partial_{p_0}^\beta\partial_{q}^\delta \Psi_{m}^\gamma(p_m+q,p_0+q,s ; V)|
	\\
	 \leq {}& \frac{1}{(2\pi)^{dm}}   \sum_{\beta_1\leq \beta} \binom{\beta}{\beta_1}  \sum_{\epsilon_1+\cdot+\epsilon_{m-1}=\beta_1+\delta} \int_{\R_{+}^{m-1}} \int  \Big| \hat{V}(p_{1}) \prod_{i=2}^{m-1} \hat{V}(p_{i}-p_{i-1}) 
	 \\
	 &\times   \prod_{i=1}^{m-1 }  \binom{(\beta_{1}+\delta)_i}{\epsilon_{1,i}\cdots\epsilon_{m-1,i}} t_i^{|\epsilon_i|} e^{- t_i \frac{1}{2} \gamma} p_i^{\epsilon_i}    \big[  \partial_1^{\alpha+\beta-\beta_1} \hat{V}](p_m-p_{m-1}-p_0)  \Big|  \, d\boldsymbol{p}_{1,m-1}  d\boldsymbol{t}.
	\end{aligned}
\end{equation}
Since we have that $\hat{V}(p_{1}) \prod_{i=2}^{m-1} \hat{V}(p_{i}-p_{i-1})$ is Schwartz in $(p_1,\dots,p_m-1)$ and $t_i^{|\epsilon_i|} e^{- t_i \frac{1}{2} \gamma} $ is integrable for all values of $|\epsilon_i|$ over the positive half-axis we get that
\begin{equation*}
	\sup_{q,p_0,p_m \in \R^d} |  \partial_{p_m}^\alpha \partial_{p_0}^\beta \partial_{q}^\delta\Psi_{m}^\gamma(p_m+q,p_0+q,\infty ; V)|  \leq C.
\end{equation*}
where $C$ depends on $\gamma$. From \eqref{psi_est_2_1} we do also immediately get that
\begin{equation*}
	\sup_{q\in \R^d} \int_{\R^d} | p_0^{\epsilon}  \partial_{p_m}^\alpha \partial_{p_0}^\beta \partial_{q}^\delta\Psi_{m}^\gamma(p_m+q,p_0+q,\infty ; V)| \, dp_0  \leq C \sum_{\epsilon_1 \leq \epsilon} |p_m^{\epsilon_1}|.
\end{equation*}
The last two bounds follow from analogous arguments. This concludes the proof.
\end{proof}
For the next lemmas we will use the following notation.
\begin{notation}\label{notation_op_main_lim}
Assume we are in the same setting as in Definition~\ref{functions_for_exp_def} and with the notation defined in Section~\ref{Sec:reg_main_op}. That is we have $\mathcal{I}_{\infty}^\gamma(k,\boldsymbol{x},t;\hbar)$, which is defined analogous to $\mathcal{I}_{0,0}^\gamma(k,\boldsymbol{x},t;\hbar)$ but with $\Psi_{\alpha_m}^\gamma(p_{m},p_{m-1},\hbar^{-1}(s_{m-1}-s_m);V)$ replaced by $\Psi_{\alpha_m}^\gamma(p_{m},p_{m-1},\infty;V)$.
 Then we will use the following notation/identification 
\begin{equation*}
	\begin{aligned}
		 \sum_{\boldsymbol{x} \in \mathcal{X}_{\neq}^k}  \mathcal{I}_{\infty}^\gamma(k,\boldsymbol{x},t;\hbar) = \sum_{\alpha\in \N^k} (i\lambda)^{|\alpha|} \mathcal{I}^\gamma_\infty(k,\alpha;\hbar).
	  \end{aligned}
\end{equation*}
That is for each $\alpha\in\N^k$ we have that
\begin{equation*}
	\begin{aligned}
	\mathcal{I}_\infty^\gamma(k,\alpha;\hbar) =  \sum_{\boldsymbol{x} \in \mathcal{X}_{\neq}^k} \int_{[0,t]_{\leq}^k}   \prod_{m=1}^k \tilde{\Theta}_{\alpha_m}^\gamma(s_{m-1},{s}_{m},x_m;V,\hbar)\, d\boldsymbol{s}_{k,1}U_{\hbar,0}(-t),
	  \end{aligned}
\end{equation*}
where the kernels of $\tilde{\Theta}_{\alpha_m}^\gamma(s_{m-1},{s}_{m},x_m;V,\hbar)$ is given by
\begin{equation*}
	\begin{aligned}
	 (p_{m},p_{0}) \mapsto{} \frac{1}{\hbar} e^{-i\hbar^{-1/d}\langle z,p_{m}-p_{0} \rangle} 
  e^{i s_{m}\hbar^{-1} \frac{1}{2} (p_{m}^2-p_{0}^2)} \Psi_{\alpha_m}^\gamma(p_{m},p_{0},\infty; V),
  	\end{aligned}
\end{equation*}
\end{notation}
\begin{lemma}\label{L.con_term_onesided_1}
Assume we are in the same setting as in Definition~\ref{functions_for_exp_def} with the notation introduced in Notation~\ref{notation_op_main_lim} and let  $\{\varphi_\hbar\}_{\hbar\in I}$ be a uniform semiclassical family in $\mathcal{H}^{5d+5}_\hbar(\R^d)$ with Wigner measure $\mu_0$. Let $k\in\N$ and $\alpha \in \N^k$.  Then for any $a \in \mathcal{S}(\R^{2d})$ we have that
\begin{equation*}
	\begin{aligned}
	\MoveEqLeft \lim_{\hbar\rightarrow0} \Aver{\langle \OpW(a)\mathcal{I}_\infty^\gamma(k,\alpha;\hbar)\varphi_\hbar,U_{\hbar,0}(-t) \varphi_\hbar \rangle }
	=  \frac{((2\pi)^d\rho t)^{k}}{k!} \int  a( x- tq,q) \prod_{m=1}^k  
	 \Psi_{\alpha_m}^\gamma(q,q,\infty;V)   \, d\mu_0(x,q).
	\end{aligned}
\end{equation*}
\end{lemma}
\begin{proof}
From the definition of the operators we have that the inner product is given by
\begin{equation*}
	\begin{aligned}
	\MoveEqLeft \langle \OpW(a)\mathcal{I}_\infty^\gamma(k,\alpha;\hbar)\varphi_\hbar,U_{\hbar,0}(-t) \varphi_\hbar \rangle
	= \frac{1}{(2\pi\hbar)^{3d}\hbar^{k}} \sum_{\boldsymbol{x} \in \mathcal{X}_{\neq}^k}  \int_{[0,t]_{\leq}^k} \int e^{ i \hbar^{-1}  \langle  x-y,q_{1} \rangle} a(\tfrac{x+y}{2},q_{1})  
	\\
	\times& e^{- i   \hbar^{-1}\langle   y_0 ,p_{0} \rangle} 
	e^{ i   \langle  \hbar^{-1}y,p_{k} \rangle} \prod_{m=1}^k  e^{ -i   \langle  \hbar^{-1/d}x_m,p_{m}-p_{m-1} \rangle}  
	  \Psi_{\alpha_m}^\gamma(p_{m},p_{m-1},\infty;V)  e^{i  s_{m} \frac{1}{2}\hbar^{-1} (p_{m}^2-p_{m-1}^2)} 
	  \\
	  \times&
	e^{- i   \hbar^{-1}\langle  x - x_0 ,q_{0} \rangle}   e^{i  t \frac{1}{2}\hbar^{-1} p_0^2} \varphi_\hbar(y_0) e^{-i  t \frac{1}{2}\hbar^{-1} q_0^2} \overline{\varphi_\hbar(x_0)}  \, dx_0 dy_0  dx dy d\boldsymbol{p}  d\boldsymbol{q}d\boldsymbol{s}.
	\end{aligned}
\end{equation*}
When taking the average we get from Campbell's theorem that
\begin{equation*}
	\begin{aligned}
	\MoveEqLeft \Aver{ \langle \OpW(a)\mathcal{I}_\infty^\gamma(k,\alpha;\hbar)\varphi_\hbar,U_{\hbar,0}(-t) \varphi_\hbar \rangle }
	= \frac{{((2\pi)^d\rho t)^{k}}}{k!(2\pi\hbar)^{3d}}   \int e^{ i \hbar^{-1}  \langle  x-y,q_{1} \rangle} a(\tfrac{x+y}{2},q_{1})  e^{- i   \hbar^{-1}\langle   y_0 -y - \frac{1}{2}tp_0 ,p_{0} \rangle} 
	\\
	&\times e^{ i   \hbar^{-1}\langle  x_0 - x -\frac{1}{2}tq_0 ,q_{0} \rangle}  \prod_{m=1}^k  
	 \Psi_{\alpha_m}^\gamma(p_{0},p_{0},\infty;V)   \varphi_\hbar(y_0)  \overline{\varphi_\hbar(x_0)}  \, dx_0 dy_0  dx dy d{p}_{0}  d\boldsymbol{q}, 
	\end{aligned}
\end{equation*}
where we have integrated all delta functions arising from the integrals in the positions of the Poisson point process and evaluated all integrals in $\boldsymbol{s}$. We now preform the change of variables $(x,y)\mapsto (x-y,\frac{1}{2}(x+y))$. This gives us
\begin{equation*}
	\begin{aligned}
	\MoveEqLeft \Aver{ \langle \OpW(a)\mathcal{I}_\infty^\gamma(k,\alpha;\hbar)\varphi_\hbar,U_{\hbar,0}(-t) \varphi_\hbar \rangle }
	= \frac{{((2\pi)^d\rho t)^{k}}}{k!(2\pi\hbar)^{3d}}   \int e^{ i \hbar^{-1}  \langle  x,q_{1} -\frac{1}{2}(p_0+q_0) \rangle} a(y,q_{1})  e^{- i   \hbar^{-1}\langle   y_0 -y - \frac{1}{2}tp_0 ,p_{0} \rangle} 
	\\
	&\times e^{ i   \hbar^{-1}\langle  x_0 -y -\frac{1}{2}tq_0 ,q_{0} \rangle}  \prod_{m=1}^k  
	 \Psi_{\alpha_m}^\gamma(p_{0},p_{0},\infty;V)   \varphi_\hbar(y_0)  \overline{\varphi_\hbar(x_0)}  \, dx_0 dy_0  dx dy d{p}_{0}  d\boldsymbol{q}.
	\end{aligned}
\end{equation*}
The integral in $q_1$ and then the integral in $x$ is an inverse Fourier transform and a Fourier transform respectively. Combing this with the change of variables $p_0\mapsto \hbar^{-1}(p_0-q)$ and then $y\mapsto  y - y_0 + \frac{1}{2}tp_0+tq  $, we have
\begin{equation}\label{con_term_1.2.0}
	\begin{aligned}
	\MoveEqLeft \Aver{ \langle \OpW(a)\mathcal{I}_\infty^\gamma(k,\alpha;\hbar)\varphi_\hbar,U_{\hbar,0}(-t) \varphi_\hbar \rangle }
	= \frac{{((2\pi)^d\rho t)^{k}}}{ k! (2\pi\hbar)^{d}(2\pi)^{d}}   \int e^{ i   \hbar^{-1}\langle  x_0 -y_0  ,q \rangle} a(y + y_0 - \tfrac{1}{2}t\hbar p_0-tq,\tfrac{1}{2} \hbar p_0+q)
	\\
	&\times  e^{ i  \langle   y   ,p_{0}  \rangle} 
	 \prod_{m=1}^k  
	 \Psi_{\alpha_m}^\gamma(\hbar p_{0}+q,\hbar p_{0}+q,\infty;V)   \varphi_\hbar(y_0)  \overline{\varphi_\hbar(x_0)}  \, dx_0 dy_0   dy d{p}_{0}  dq.
	\end{aligned}
\end{equation}
To calculate the integral in $p_0$ and $y$ we perform a first order Taylor expansion of the function
\begin{equation*}
	p_0 \mapsto a(y + y_0 - \tfrac{1}{2}t\hbar p_0-tq,\tfrac{1}{2} \hbar p_0+q)  \Psi_{\alpha_m}^\gamma(\hbar p_{0}+q,\hbar p_{0}+q,\infty;V).
\end{equation*}
This produces an error of order $\hbar$. To see this let 
\begin{equation}\label{con_term_1.2.1}
	f(y,y_0,\hbar p_0, q)=a(y + y_0 - \tfrac{1}{2}t\hbar p_0-tq,\tfrac{1}{2} \hbar p_0+q)  \Psi_{\alpha_m}^\gamma(\hbar p_{0}+q,\hbar p_{0}+q,\infty;V).
\end{equation}
Observe that for this function we have integrability in the variables $y$ and $q$ also after taking a number of derivatives. This is due to our assumption that $a\in \mathcal{S}(\R^{2d})$ and Lemma~\ref{est_der_Psi_gamma}. Then by a first order Taylor expansion in $\hbar p_0$ around zero we get that
\begin{equation*}
	f(y,y_0,\hbar p_0, q)=f(y,y_0,0, q) +\sum_{\abs{\alpha}=1}  (\hbar p_0)^\alpha \int_0^1  \partial_\eta^{\alpha} f(y,y_0,\eta + s\hbar  p_0, q))\big|_{\eta=0} \, ds.
\end{equation*}
Using this we have
\begin{equation}\label{con_term_1.2.2}
	\begin{aligned}
	\MoveEqLeft     \int e^{ i   \hbar^{-1}\langle  x_0 -y_0  ,q \rangle} e^{ i  \langle   y   ,p_{0}  \rangle}  f(y,y_0,\hbar p_0, q)   \varphi_\hbar(y_0)  \overline{\varphi_\hbar(x_0)}  \, dx_0 dy_0   dy d{p}_{0}  dq
	\\
	={}&   \int e^{ i   \hbar^{-1}\langle  x_0 -y_0  ,q \rangle} e^{ i  \langle   y   ,p_{0}  \rangle}  f(y,y_0,0, q)   \varphi_\hbar(y_0)  \overline{\varphi_\hbar(x_0)}  \, dx_0 dy_0   dy d{p}_{0}  dq
	\\
	&+\sum_{\abs{\alpha}=1} \int\int_0^1 e^{ i   \hbar^{-1}\langle  x_0 -y_0  ,q \rangle} e^{ i  \langle   y   ,p_{0}  \rangle}   (\hbar p_0)^\alpha   \partial_\eta^{\alpha} f(y,y_0,\eta + s\hbar  p_0, q))\big|_{\eta=0}   
	 \varphi_\hbar(y_0)  \overline{\varphi_\hbar(x_0)}  \, \, ds dx_0 dy_0   dy d{p}_{0}  dq,
	\end{aligned}
\end{equation}
where we have omitted the prefactor $ (2\pi\hbar)^{-d}$. In order to estimate the error terms we define the two linear operators
\begin{equation*}
	L_q= \frac{1 - i\langle \nabla_y , \hbar^{-1}(x_0-y_0) \rangle}{1 + |\hbar^{-1}(x_0-y_0)|^2 } \quad\text{and}\quad  L_{y} = \frac{1 - i\langle \nabla_y , p_0 \rangle}{1 + p_0^2 }.
\end{equation*}
These operators acts as the identity on $e^{ i   \hbar^{-1}\langle  x_0 -y_0  ,q \rangle}$ and  $e^{ i  \langle   y   ,p_{0}  \rangle} $ respectively.  Inserting each of these $d+1$ times, and integrating by parts, we obtain for each $\alpha$ the estimate
\begin{equation}\label{con_term_1.2.3}
	\begin{aligned}
	\MoveEqLeft \frac{1}{ (2\pi\hbar)^{d}} \Big|  \int \int_0^1 e^{ i   \hbar^{-1}\langle  x_0 -y_0  ,q \rangle} e^{ i  \langle   y   ,p_{0}  \rangle}   (\hbar p_0)^\alpha   \partial_\eta^{\alpha} f(y,y_0,\eta + s\hbar  p_0, q))\big|_{\eta=0}   
	\varphi_\hbar(y_0)  \overline{\varphi_\hbar(x_0)}  \, \, ds dx_0 dy_0   dy d{p}_{0}  dq \Big|
	\\
	\leq{}& \frac{\hbar}{ (2\pi\hbar)^{d}}  \Big| \sum_{|{\tilde{\alpha}} | \leq d+1 } \binom{d+1}{({\tilde{\alpha}},d+1-|{\tilde{\alpha}}|)}  \sum_{|\epsilon | \leq d+1 } \binom{d+1}{(\epsilon,d+1-|\epsilon|)} 
	\int \frac{ |(\hbar^{-1}x_0-\hbar^{-1}y_0)^\epsilon| }{(1 + |\hbar^{-1}(x_0-y_0)|^2)^{d+1} } \frac{|p_0^{\tilde{\alpha}}| }{(1 + p_0^2)^{d+1} }    	
	\\
	&\phantom{\int_{\R^{5d}}}{} \times (|\varphi_\hbar(y_0)|^2 + |{\varphi_\hbar(x_0)} |^2)
 \sup_{\eta\in\R^d}| \partial_y^{{\tilde{\alpha}}+\alpha} \partial_q^\epsilon\partial_\eta^{\alpha} f(y,y_0,\eta, q))   |  \,  dx_0 dy_0   dy d{p}_{0}  dq \Big|
 	\\
	\leq{}& \hbar C,
	\end{aligned}
\end{equation}
where we have used the inequality $2| \varphi_\hbar(x_0)   \varphi_\hbar(y_0)| \leq | \varphi_\hbar(x_0) |^2+ | \varphi_\hbar(y_0)|^2$. Combining the estimates in  \cref{con_term_1.2.0,con_term_1.2.1,con_term_1.2.2,con_term_1.2.3} yields
\begin{equation}\label{con_term_1.2}
	\begin{aligned}
	\MoveEqLeft \Aver{ \langle \OpW(a)\mathcal{I}_\infty^\gamma(k,\alpha;\hbar)\varphi_\hbar,U_{\hbar,0}(-t) \varphi_\hbar \rangle }
	\\
	={}& \frac{{((2\pi)^d\rho t)^{k}}}{ k! (2\pi\hbar)^{d}}   \int e^{ i   \hbar^{-1}\langle  x_0 -y_0  ,q \rangle} a( y_0 -tq, q)    \prod_{m=1}^k  
	 \Psi_{\alpha_m}^\gamma(q,q,\infty;V)   \varphi_\hbar(y_0)  \overline{\varphi_\hbar(x_0)}  \, dy_0   dq dx_0 
	 +\mathcal{O}(\hbar).
	\end{aligned}
\end{equation}
From our assumptions on our states and \eqref{con_term_1.2} we get that
\begin{equation*}
	\begin{aligned}
	\MoveEqLeft \lim_{\hbar\rightarrow0} \Aver{\langle \OpW(a)\mathcal{I}_\infty^\gamma(k,\alpha;\hbar)\varphi_\hbar,U_{\hbar,0}(-t) \varphi_\hbar \rangle }
	=  \frac{((2\pi)^d\rho t)^{k}}{k!} \int  a( x- tq,q) \prod_{m=1}^k  
	 \Psi_{\alpha_m}^\gamma(q,q,\infty;V)   \, d\mu_0(x,q).
	\end{aligned}
\end{equation*}
This concludes the proof.
\end{proof}
Next we want to establish the term wise convergence for the terms, where we are not just acting with the free propagator. This will be done in a number of steps. In order to limit notation we will define the following setting.
\begin{setting}\label{setting_conv_twosided}
Assume we are in setting of Definition~\ref{functions_for_exp_def} with the notation introduced in Notation~\ref{notation_op_main_lim}. Let  $\{\varphi_\hbar\}_{\hbar\in I}$ be a uniform semiclassical family in $\mathcal{H}^{5d+5}_\hbar(\R^d)$ with Wigner measure $\mu_0$. Let $k\in\N$, $\alpha \in \N^k$, $r \in\N$ and ${\tilde{\alpha}} \in \N^r$. Then for any $a \in \mathcal{S}(\R^{2d})$ consider the expression
\begin{equation*}
	\begin{aligned}
	 \langle \OpW(a)\mathcal{I}_\infty^\gamma(k,\alpha;\hbar)\varphi_\hbar,\mathcal{I}_\infty^\gamma(r,{\tilde{\alpha}};\hbar)\varphi_\hbar \rangle 
	= \sum_{n=0}^{\min(k,r)} \sum_{\sigma^1 \in \mathcal{A}(k,n)} \sum_{\sigma^2\in\mathcal{A}(r,n)} \sum_{\kappa\in\mathcal{S}_n} \mathcal{J}(a,\varphi_\hbar,n,\sigma^1,\sigma^2,\kappa,\hbar), 
		\end{aligned}
\end{equation*}
	where $\mathcal{J}(a,\varphi_\hbar,n,\sigma^1,\sigma^2,\tau,\hbar)$ is defined by
\begin{equation*}
	\begin{aligned}
	\MoveEqLeft  \mathcal{J}(a,\varphi_\hbar,n,\sigma^1,\sigma^2,\kappa,\hbar) = \sum_{(\boldsymbol{x},\boldsymbol{\tilde{x}})\in \mathcal{X}_{\neq}^{k+r}}\prod_{i=1}^n \rho^{-n}\delta(x_{\sigma_i^1}- \tilde{x}_{\sigma_{\kappa(i)}^2})   \int_{[0,t]_{\leq}^k} \int_{[0,t]_{\leq}^r} 
	\\
	 &\times \langle \OpW(a) \prod_{m=1}^k \tilde{\Theta}_{\alpha_m}^\gamma(s_{{m-1}},{s}_{m},x_m;V,\hbar) \varphi_\hbar,\prod_{m=1}^r \tilde{\Theta}_{{\tilde{\alpha}}_m}^\gamma(s_{{m-1}},{s}_{m},x_m;V,\hbar) \varphi_\hbar \rangle  \, d\boldsymbol{\tilde{s}}d\boldsymbol{s},
	\end{aligned}
\end{equation*}
where $\tilde{\Theta}$ is the operators from Notation~\ref{notation_op_main_lim}.
Finally for $\varepsilon\geq0$ we define $\mathcal{J}_\varepsilon(a,\varphi_\hbar,n,\sigma^1,\sigma^2,\kappa,\hbar)$ by
\begin{equation*}
	\begin{aligned}
	\MoveEqLeft  \mathcal{J}_\varepsilon(a,\varphi_\hbar,n,\sigma^1,\sigma^2,\kappa,\hbar)=\sum_{(\boldsymbol{x},\boldsymbol{\tilde{x}})\in \mathcal{X}_{\neq}^{k+r}} \prod_{i=1}^n \delta(x_{\sigma_i^1}- \tilde{x}_{\sigma_{\kappa(i)}^2})   \int_{[0,t]_{\leq}^k} \int_{[0,t]_{\leq}^r} \prod_{i=1}^n e^{-(\varepsilon\hbar^{-1}(\tilde{s}_{\sigma^2_{\kappa(i)}}-s_{\sigma^1_{i}}))^2}
	\\
	 &\times \langle \OpW(a) \prod_{m=1}^k \tilde{\Theta}_{\alpha_m}^\gamma(s_{{m-1}},{s}_{m},x_m;V,\hbar) \varphi_\hbar,\prod_{m=1}^r \tilde{\Theta}_{{\tilde{\alpha}}_m}^\gamma(s_{{m-1}},{s}_{m},x_m;V,\hbar) \varphi_\hbar \rangle  \, d\boldsymbol{\tilde{s}}d\boldsymbol{s}.
	\end{aligned}
\end{equation*}
\end{setting}
\begin{lemma}\label{lemma_gaus_reg_2}
Assume we are in Setting~\ref{setting_conv_twosided} and assume that $\kappa$ is different from the identity. Then we have that 
\begin{equation*}
	\begin{aligned}
	|\Aver{  \mathcal{J}(a,\varphi_\hbar,n,\sigma^1,\sigma^2,\kappa,\hbar)}| \leq C_a  \frac{\hbar\rho (\rho t)^{k+r-n-1}}{(k-1)!(r-n)!} (C  \norm{\hat{V}}_{1,\infty,5d+5})^{|\alpha|+|{\tilde{\alpha}}|} |\log(\tfrac{\hbar}{t})|^{n+3} \norm{\varphi_\hbar}^{2}_{\mathcal{H}_\hbar^{2d+2}(\R^d)},
	\end{aligned}
\end{equation*}
where $C_a$ depends on the symbol $a$ and $C$ only depend on the dimension.
\end{lemma}
\begin{proof}
From the definition of the number $\mathcal{J}(a,\varphi_\hbar,n,\sigma^1,\sigma^2,\kappa,\hbar)$ and by applying the version of Lemma~\ref{LE:Exp_ran_phases} mentioned in Remark~\ref{RE:LE:Exp_ran_phases} we get that 
\begin{equation*}
	\begin{aligned}
	\MoveEqLeft \Aver{ \mathcal{J}(a,\varphi_\hbar,n,\sigma^1,\sigma^2,\kappa,\hbar)} 
	=\frac{((2\pi)^d\rho)^{k+r-n}}{(2\pi\hbar)^{3d}\hbar^{n}}  \int_{[0,t]_{\leq}^k}\int_{[0,t]_{\leq}^r}  \int e^{ i \hbar^{-1}  \langle  x-y,q_{n+1} \rangle}  e^{ i  \hbar^{-1} \langle  y,p_{n} \rangle} e^{- i \hbar^{-1}  \langle  x,q_{n} \rangle}  
	\\
	\times&    e^{ -i \hbar^{-1}  \langle  y_0,p_{0} \rangle} e^{ i  \hbar^{-1} \langle x_0,q_{0} \rangle} 
	 \prod_{i=1}^n \delta(p_{i}- q_{\kappa(i)} - p_{0}+ q_{0} - \tilde{l}^{\kappa}_{i}(\boldsymbol{q})) 
	 \Psi_{\alpha_{\sigma^1_{i}}}^\gamma(p_{i},p_{i-1})  e^{i  s_{\sigma^1_{i}} \frac{1}{2}\hbar^{-1} (p_{i}^2-p_{i-1}^2)} 
	\\
	\times&  \overline{ \Psi_{{\tilde{\alpha}}_{\sigma^2_{i}}}^\gamma(q_{i},q_{i-1}) } e^{-i  \tilde{s}_{\sigma^2_{i}} \frac{1}{2}\hbar^{-1} (q_{i}^2-q_{i-1}^2)}
	 \prod_{i=1}^{n+1}  \prod_{m=\sigma^1_{i-1}+1}^{\sigma^1_{i}-1} \Psi_{\alpha_m}^\gamma(p_{i-1},p_{i-1})  \prod_{m=\sigma^2_{i-1}+1}^{\sigma^2_{i}-1} \overline{ \Psi_{{\tilde{\alpha}}_m}^\gamma(q_{i-1},q_{i-1}) }
	 \\
	 \times& 
	 a(\tfrac{x+y}{2},q_{n+1})   e^{i  t \frac{1}{2}\hbar^{-1} p_0^2} \varphi_\hbar(y_0) e^{-i  t \frac{1}{2}\hbar^{-1} q_0^2} \overline{\varphi_\hbar(x_0)}  \,d\boldsymbol{x} d\boldsymbol{y}d\boldsymbol{p}  d\boldsymbol{q} d\boldsymbol{\tilde{s}} d\boldsymbol{s},
	\end{aligned}
\end{equation*}
  where we have used the notation $\tilde{l}^{\kappa}_{i}(\boldsymbol{q})= \sum_{j=1}^{i} ( q_{{\kappa(j-1)}}-q_{{\kappa(j)-1}} )$, and
the convention $\kappa(0)=0$. We are here not explicit writing the dependence on $\infty$ and $V$ for the $\Psi$-functions. We have also evaluated some of the integrals in the position variables and momentum variables and done a relabelling of some variables. We now do the change of variables 
\begin{equation*}
	\begin{aligned}
	s_k \mapsto s_k \quad\text{and}\quad s_m\mapsto s_m-s_{m+1} \quad\text{for all $m\in\{1,,\dots,k-1\}$}
	\\
	\tilde{s}_r \mapsto \tilde{s}_r \quad\text{and}\quad \tilde{s}_m\mapsto \tilde{s}_m-\tilde{s}_{m+1} \quad\text{for all $m\in\{1,,\dots,r-1\}$}.
	\end{aligned}
\end{equation*}
Furthermore, we also make the change of variables $(x,y)\mapsto (x-y,\frac{1}{2}(x+y))$. This gives us
\begin{equation*}
	\begin{aligned}
	\MoveEqLeft \Aver{ \mathcal{J}(a,\varphi_\hbar,n,\sigma^1,\sigma^2,\kappa,\hbar)} 
	=\frac{((2\pi)^d\rho)^{k+r-n}}{(2\pi\hbar)^{2d}\hbar^{n}}  \int_{\R_{+}^k}\int_{\R_{+}^r}  \int \boldsymbol{1}_{[0,t]}( \boldsymbol{s}_{1,k}^{+})  \boldsymbol{1}_{[0,t]}( \boldsymbol{\tilde{s}}_{1,r}^{+})   e^{ i  \hbar^{-1} \langle  y,p_{n} - q_n \rangle} 
	\\
	\times&    e^{ -i \hbar^{-1}  \langle  y_0,p_{0} \rangle} e^{ i  \hbar^{-1} \langle x_0,q_{0} \rangle} 
	 \prod_{i=1}^n  \delta(p_{i}- q_{\kappa(i)} - p_{0}+ q_{0} - \tilde{l}^{\kappa}_{i}(\boldsymbol{q})) 
	 \Psi_{\alpha_{\sigma^1_{i}}}^\gamma(p_{i},p_{i-1})  e^{i  \boldsymbol{s}_{\sigma^1_{i},\sigma^1_{i+1}-1}^{+} \frac{1}{2}\hbar^{-1}p_{i}^2} 
	\\
	\times&  \overline{ \Psi_{{\tilde{\alpha}}_{\sigma^2_{i}}}^\gamma(q_{i},q_{i-1}) } e^{-i  \boldsymbol{\tilde{s}}_{\sigma^2_{i},\sigma^2_{i+1}-1}^{+} \frac{1}{2}\hbar^{-1} q_{i}^2}
	 \prod_{i=1}^{n+1}  \prod_{m=\sigma^1_{i-1}+1}^{\sigma^1_{i}-1} \Psi_{\alpha_m}^\gamma(p_{i-1},p_{i-1})   \prod_{m=\sigma^2_{i-1}+1}^{\sigma^2_{i}-1} \overline{ \Psi_{{\tilde{\alpha}}_m}^\gamma(q_{i-1},q_{i-1}) }
	 \\
	 \times& 
	 a(y,\tfrac{p_n +q_n}{2})   e^{i  (t- \boldsymbol{s}_{1,k}^{+}) \frac{1}{2}\hbar^{-1} p_0^2} \varphi_\hbar(y_0) e^{-i  (t- \boldsymbol{\tilde{s}}_{1,r}^{+})\frac{1}{2}\hbar^{-1} q_0^2} \overline{\varphi_\hbar(x_0)}  \,d\boldsymbol{x} d\boldsymbol{y}d\boldsymbol{p}  d\boldsymbol{q} d\boldsymbol{\tilde{s}} d\boldsymbol{s},
	\end{aligned}
\end{equation*}
where we after the change of variables have evaluated the integrals in $x$ and $q_{n+1}$.
We now evaluate the integrals in $x_0$, $y_0$ and all integrals in all $p$'s except $p_0$. This gives us
\begin{equation*}
	\begin{aligned}
	\MoveEqLeft \Aver{ \mathcal{J}(a,\varphi_\hbar,n,\sigma^1,\sigma^2,\kappa,\hbar)} 
	=\frac{((2\pi)^d\hbar\rho)^{k+r-n}}{(2\pi\hbar)^{2d}}  \int_{\R_{+}^k}\int_{\R_{+}^r}  \int \boldsymbol{1}_{[0,\hbar^{-1}t]}( \boldsymbol{s}_{1,k}^{+})  \boldsymbol{1}_{[0,\hbar^{-1}t]}( \boldsymbol{\tilde{s}}_{1,r}^{+})   e^{ i  \hbar^{-1} \langle  y,p_{0} - q_0 \rangle} 
	\\
	\times& 
	 e^{i  (\hbar^{-1}t- \boldsymbol{s}_{1,k}^{+}) \frac{1}{2} p_0^2}  e^{-i  (\hbar^{-1}t- \boldsymbol{\tilde{s}}_{1,r}^{+})\frac{1}{2} q_0^2}   \mathcal{G}(p_0,\boldsymbol{q},\sigma^1,\sigma^2) \prod_{i=1}^n  e^{i  \boldsymbol{s}_{\sigma^1_{i},\sigma^1_{i+1}-1}^{+} \frac{1}{2}(q_{\kappa(i)} + p_{0}- q_{0} + \tilde{l}^{\kappa}_{i}(\boldsymbol{q}))^2}  
	 e^{-i  \boldsymbol{\tilde{s}}_{\sigma^2_{i},\sigma^2_{i+1}-1}^{+} \frac{1}{2} q_{i}^2}  
	 \\
	 \times& 
	 a(y,\tfrac{p_0 -q_0}{2}+q_n)     \hat{\varphi}_\hbar(\tfrac{p_0}{\hbar}) \overline{\hat{\varphi}_\hbar(\tfrac{q_0}{\hbar})}  \, dydp_0  d\boldsymbol{q} d\boldsymbol{\tilde{s}} d\boldsymbol{s},
	\end{aligned}
\end{equation*}
where we have done the change of variables $s_m\mapsto\hbar^{-1}s_m$ and $\tilde{s}_m\mapsto\hbar^{-1}\tilde{s}_m$ for all $m$, we have introduced the notation
\begin{equation*}
	\begin{aligned}
	\MoveEqLeft \mathcal{G}(p_0,\boldsymbol{q},\sigma^1,\sigma^2) 
	= 
	 \prod_{i=1}^n  \Psi_{\alpha_{\sigma^1_{i}}}^\gamma(q_{\kappa(i)} + p_{0}- q_{0} + \tilde{l}^{\kappa}_{i}(\boldsymbol{q}), q_{\kappa(i)-1} + p_{0}- q_{0} + \tilde{l}^{\kappa}_{i}(\boldsymbol{q}))  
	\overline{ \Psi_{{\tilde{\alpha}}_{\sigma^2_{i}}}^\gamma(q_{i},q_{i-1}) }    
	 \\
	 \times&
	   \prod_{i=1}^{n+1}  \prod_{m=\sigma^1_{i-1}+1}^{\sigma^1_{i}-1} \Psi_{\alpha_m}^\gamma(q_{\kappa(i)-1} + p_{0}- q_{0} + \tilde{l}^{\kappa}_{i}(\boldsymbol{q}),q_{\kappa(i)-1} + p_{0}- q_{0} + \tilde{l}^{\kappa}_{i}(\boldsymbol{q}))    \prod_{m=\sigma^2_{i-1}+1}^{\sigma^2_{i}-1} \overline{ \Psi_{{\tilde{\alpha}}_m}^\gamma(q_{i-1},q_{i-1}) },
	\end{aligned}
\end{equation*}
and used that
\begin{equation*}
    q_{\kappa(n)} + \tilde{l}^{\kappa}_{n}(\boldsymbol{q}) = q_n \quad\text{and}\quad  q_{\kappa(i-1)}  + \tilde{l}^{\kappa}_{i-1}(\boldsymbol{q}) = q_{\kappa(i)-1}  + l^{\kappa}_{i}(\boldsymbol{q}).
    \end{equation*}
As in previous proofs we will let $i^{*}$ denote the smallest index such that $\kappa(i^{*})\neq i^{*}$ and we define the function $f(\boldsymbol{s},\boldsymbol{\tilde{s}})$ to be  
\begin{equation*}
	f(\boldsymbol{s},\boldsymbol{\tilde{s}}) = \boldsymbol{1}_{[0,\hbar^{-1}t]}\big( \sum_{i=1,i\neq i^{*}}^{k} s_i  \big) \boldsymbol{1}_{[0,\hbar^{-1}t]}\big( \sum_{i=1,i \notin \sigma}^{r} \tilde{s}_i  \big).
\end{equation*}
Using this function and introducing the two variables $s_0$ and $\tilde{s}_0$ we get that
\begin{equation*}
	\begin{aligned}
	\MoveEqLeft |\Aver{ \mathcal{J}(a,\varphi_\hbar,n,\sigma^1,\sigma^2,\kappa,\hbar)} |
	=\frac{((2\pi)^d\hbar\rho)^{k+r-n}}{(2\pi\hbar)^{2d}}\Big|  \int_{\R_{+}^k}\int_{\R_{+}^r}  \int \delta(\tfrac{t}{\hbar}- \boldsymbol{s}_{0,k}^{+})   \delta(\tfrac{t}{\hbar}-\boldsymbol{\tilde{s}}_{0,r}^{+}) f(\boldsymbol{s},\boldsymbol{\tilde{s}})
	\\
	\times& 
	 e^{ i  \hbar^{-1} \langle  y,p_{0} - q_0 \rangle} e^{i  s_0 \frac{1}{2} p_0^2}  e^{-i  \tilde{s}_0\frac{1}{2} q_0^2}    \mathcal{G}(p_0,\boldsymbol{q},\sigma^1,\sigma^2) \prod_{i=1}^n  e^{i  \boldsymbol{s}_{\sigma^1_{i},\sigma^1_{i+1}-1}^{+} \frac{1}{2}(q_{\kappa(i)} + p_{0}- q_{0} + \tilde{l}^{\kappa}_{i}(\boldsymbol{q}))^2}  
	 e^{-i  \boldsymbol{\tilde{s}}_{\sigma^2_{i},\sigma^2_{i+1}-1}^{+} \frac{1}{2} q_{i}^2}  
	 \\
	 \times& 
	 \langle \nabla_y \rangle^{6d+6} a(y,\tfrac{p_0 -q_0}{2}+q_n)    \langle \hbar^{-1}( p_{0} - q_0) \rangle^{-6d-6}      \hat{\varphi}_\hbar(\tfrac{p_0}{\hbar}) \overline{\hat{\varphi}_\hbar(\tfrac{q_0}{\hbar})}  \,dydp_0  d\boldsymbol{q} d\boldsymbol{\tilde{s}} d\boldsymbol{s}\Big|,
	\end{aligned}
\end{equation*}
where we have used integration by parts in $y$ to obtain the factor $\langle \hbar^{-1}( p_{0} - q_0) \rangle^{-6d-6}$. We again use the function $\zeta$ as in the previous proofs and write the delta functions as Fourier transforms of the constant function $1$. We then obtain that
\begin{equation*}
	\begin{aligned}
	\MoveEqLeft |\Aver{ \mathcal{J}(a,\varphi_\hbar,n,\sigma^1,\sigma^2,\kappa,\hbar)} |
	\leq \frac{((2\pi)^d\hbar\rho)^{k+r-n}}{(2\pi)^2(2\pi\hbar)^{2d}}\Big|  \int_{\R_{+}^k}\int_{\R_{+}^r}  \int  f(\boldsymbol{s},\boldsymbol{\tilde{s}})  e^{ i  \hbar^{-1} \langle  y,p_{0} - q_0 \rangle}  \mathcal{G}(p_0,\boldsymbol{q},\sigma^1,\sigma^2)
	\\
	\times& 
	 e^{i  s_0 (\frac{1}{2} p_0^2+\nu +i\zeta)}  e^{-i  \tilde{s}_0(\frac{1}{2} q_0^2 -\tilde{\nu} - i\zeta)}    \prod_{i=1}^n  e^{i  \boldsymbol{s}_{\sigma^1_{i},\sigma^1_{i+1}-1}^{+} (\frac{1}{2}(q_{\kappa(i)} + p_{0}- q_{0} + \tilde{l}^{\kappa}_{i}(\boldsymbol{q}))^2+\nu +i\zeta)}  
	 e^{-i  \boldsymbol{\tilde{s}}_{\sigma^2_{i},\sigma^2_{i+1}-1}^{+}( \frac{1}{2} q_{i}^2 -\tilde{\nu} - i\zeta)}  
	 \\
	 \times& 
	 \langle \nabla_y \rangle^{6d+6} a(y,\tfrac{p_0 -q_0}{2}+q_n)    \langle \hbar^{-1}( p_{0} - q_0) \rangle^{-6d-6}      \hat{\varphi}_\hbar(\tfrac{p_0}{\hbar}) \overline{\hat{\varphi}_\hbar(\tfrac{q_0}{\hbar})}  \,dydp_0 d\nu d\tilde{\nu}  d\boldsymbol{q} d\boldsymbol{\tilde{s}} d\boldsymbol{s}\Big|,
	\end{aligned}
\end{equation*}
Next we evaluate all integrals in $s$ and $\tilde{s}$ that $f$ does not depend on. We then get the estimate
\begin{equation*}
	\begin{aligned}
	\MoveEqLeft |\Aver{ \mathcal{J}(a,\varphi_\hbar,n,\sigma^1,\sigma^2,\kappa,\hbar)} |
	\leq \frac{((2\pi)^d\hbar\rho)^{k+r-n}}{(2\pi)^2(2\pi\hbar)^{2d}}  \int_{\R_{+}^{k+r-n-1}} \int  f(\boldsymbol{s},\boldsymbol{\tilde{s}})   \frac{|\mathcal{G}(p_0,\boldsymbol{q},\sigma^1,\sigma^2)| }{  \langle \hbar^{-1}( p_{0} - q_0) \rangle^{6d+6}}
	\\
	\times& 
	\frac{1}{| \frac{1}{2} p_0^2+\nu +i\zeta|} \frac{1}{\frac{1}{2} q_0^2 -\tilde{\nu} - i\zeta}  \frac{1}{|\frac{1}{2}(q_{\kappa(i^{*})} + p_{0}- q_{0} + q_{i^{*}-1} -q_{\kappa(i^{*})-1})^2+\nu +i\zeta|}  
	 \prod_{i=1}^n  \frac{1}{|(\frac{1}{2} q_{i}^2 -\tilde{\nu} - i\zeta)|}
	 \\
	 \times& 
	| \langle \nabla_y \rangle^{6d+6} a(y,\tfrac{p_0 -q_0}{2}+q_n) |      | \hat{\varphi}_\hbar(\tfrac{p_0}{\hbar})| |\hat{\varphi}_\hbar(\tfrac{q_0}{\hbar})|  \,dydp_0 d\nu d\tilde{\nu}  d\boldsymbol{q} d\boldsymbol{\tilde{s}} d\boldsymbol{s}.
	\end{aligned}
\end{equation*}
From here we will mostly use the standard tricks used in previous proofs. In this case we use the estimates 
\begin{equation*}
	\begin{aligned}
	\MoveEqLeft \frac{ \langle q_{\kappa(i^{*})}\rangle^{d+1}  \langle q_{\kappa(i^{*})-1}\rangle^{d+1} \langle q_{\kappa(i^{*})+1}\rangle^{d+1} \langle q_{\kappa(i^{*})} - q_{\kappa(i^{*})-1} +p_{0} - q_0 + q_{i^{*}-1} \rangle^{d+1} }
	{
	 \prod_{m=1}^{\kappa(i^{*})-2} \langle q_m-q_{m-1}\rangle^{-d-1} \prod_{m=\kappa(i^{*})+2}^{n} \langle q_m-q_{m-1}\rangle^{-d-1}
	 }
	\\
	&\leq C\langle q_{0}\rangle^{4d+4} \langle p_{0}-q_0 \rangle^{d+1} \langle q_{\kappa(i^{*})}- q_{\kappa(i^{*})-1} \rangle^{d+1} \prod_{m=1}^{n}  \langle q_m-q_{m-1}\rangle^{5d+5}  
	\end{aligned}
\end{equation*}
and 
  \begin{equation*}
	\begin{aligned}
	\MoveEqLeft  \sup_{\boldsymbol{q},p_0}  \big|  \langle q_{\kappa(i^{*})}- q_{\kappa(i^{*})-1} \rangle^{d+1} \prod_{m=1}^{n}  \langle q_m-q_{m-1}\rangle^{5d+5}  \mathcal{G}(p_0,\boldsymbol{q},\sigma^1,\sigma^2)
	 \big|
	\leq  C^{|\alpha|+|{\tilde{\alpha}}|}  \norm{\hat{V}}_{1,\infty,5d+5}^{|\alpha|+|{\tilde{\alpha}}|}. 
	\end{aligned}
\end{equation*}
Moreover we have that
\begin{equation*}
	\int_{\R_{+}^{k+r-n-1}} f(\boldsymbol{s},\boldsymbol{\tilde{s}})  d\boldsymbol{s}_{1,k-1}   d\boldsymbol{\tilde{s}}_{1,r-n}  \leq \frac{t^{k+r-n-1}}{\hbar^{k+r-n-1}(k-1)!(r-n)!}.
\end{equation*}
With these estimates we obtain that
\begin{equation*}
	\begin{aligned}
	\MoveEqLeft |\Aver{ \mathcal{J}(a,\varphi_\hbar,n,\sigma^1,\sigma^2,\kappa,\hbar)} |
	\leq \frac{\hbar\rho (\rho t)^{k+r-n-1}}{ (2\pi\hbar)^{2d} (k-1)!(r-n)!} C^{|\alpha|+|{\tilde{\alpha}}|}  \norm{\hat{V}}_{1,\infty,5d+5}^{|\alpha|+|{\tilde{\alpha}}|} \int  \frac{\langle q_{0}\rangle^{4d+4} \langle p_{0} - q_0 \rangle}{  \langle \hbar^{-1}( p_{0} - q_0) \rangle^{6d+6}}
	\\
	\times& \frac{1}{ \langle\nu\rangle | \frac{1}{2} p_0^2+\nu +i\zeta|} \frac{1}{\langle\tilde{\nu}\rangle|\frac{1}{2} q_0^2 -\tilde{\nu} - i\zeta|} \prod_{i=\kappa(i^{*})-1}^{\kappa(i^{*})+1} \frac{1}{  \langle q_i\rangle^{d+1} |(\frac{1}{2} q_{i}^2 -\tilde{\nu} - i\zeta)|}    \prod_{i=1}^{\kappa(i^{*})-2}  \frac{\langle q_i-q_{i-1}\rangle^{-d-1}}{ |(\frac{1}{2} q_{i}^2 -\tilde{\nu} - i\zeta)|}
	\\
	\times& \frac{ \langle\nu\rangle\langle\tilde{\nu}\rangle  \langle q_{\kappa(i^{*})} - q_{\kappa(i^{*})-1} +p_{0} - q_0 + q_{i^{*}-1} \rangle^{-d-1} }{|\frac{1}{2}(q_{\kappa(i^{*})} + p_{0}- q_{0} + q_{i^{*}-1} -q_{\kappa(i^{*})-1})^2+\nu +i\zeta|} | \langle \nabla_y \rangle^{6d+6} a(y,\tfrac{p_0 -q_0}{2}+q_n) |   \prod_{i=\kappa(i^{*})+2}^{n}\frac{\langle q_i-q_{i-1}\rangle^{-d-1}}{ |(\frac{1}{2} q_{i}^2 -\tilde{\nu} - i\zeta)|} 
	\\
	 \times& 
	   ( | \hat{\varphi}_\hbar(\tfrac{p_0}{\hbar})|^2 + |\hat{\varphi}_\hbar(\tfrac{q_0}{\hbar})|^2)  \,dydp_0 d\nu d\tilde{\nu}  d\boldsymbol{q} d\boldsymbol{\tilde{s}} d\boldsymbol{s},
	\end{aligned}
\end{equation*}
where we also have used that $  | \hat{\varphi}_\hbar(\tfrac{p_0}{\hbar})\hat{\varphi}_\hbar(\tfrac{q_0}{\hbar})|\leq | \hat{\varphi}_\hbar(\tfrac{p_0}{\hbar})|^2 + |\hat{\varphi}_\hbar(\tfrac{q_0}{\hbar})|^2$.
By applying Lemma~\ref{LE:int_posistion} and Lemma~\ref{LE:resolvent_int_est} we obtain that
\begin{equation*}
	\begin{aligned}
	\MoveEqLeft |\Aver{ \mathcal{J}(a,\varphi_\hbar,n,\sigma^1,\sigma^2,\kappa,\hbar)} |
	\leq \frac{\hbar\rho (\rho t)^{k+r-n-1}}{ (2\pi\hbar)^{2d} (k-1)!(r-n)!} C^{|\alpha|+|{\tilde{\alpha}}|}  \norm{\hat{V}}_{1,\infty,5d+5}^{|\alpha|+|{\tilde{\alpha}}|} |\log(\tfrac{\hbar}{t})|^{n+3}
	\\
	 \times& 
	\sup_{q_n\in\R^d}  \int   | \langle \nabla_y \rangle^{6d+6} a(y,q_n) |    \,dy \int  \frac{\langle q_{0}\rangle^{4d+4} \langle p_{0} - q_0 \rangle}{  \langle \hbar^{-1}( p_{0} - q_0) \rangle^{6d+6}}  ( | \hat{\varphi}_\hbar(\tfrac{p_0}{\hbar})|^2 + |\hat{\varphi}_\hbar(\tfrac{q_0}{\hbar})|^2)  \,dp_0 dq_0,
	\end{aligned}
\end{equation*}
where we have absorbed all constants only depending on dimension into $C^{|\alpha|+|{\tilde{\alpha}}|}$. For the remaining integrals we have that the integral over the symbol $a$ is a finite constant by assumption. For the other two integrals we have that
\begin{equation*}
	\begin{aligned}
	\MoveEqLeft 
	 \int  \frac{\langle q_{0}\rangle^{4d+4} \langle p_{0} - q_0 \rangle}{  \langle \hbar^{-1}( p_{0} - q_0) \rangle^{6d+6}}  ( | \hat{\varphi}_\hbar(\tfrac{p_0}{\hbar})|^2 + |\hat{\varphi}_\hbar(\tfrac{q_0}{\hbar})|^2)  \,dp_0 dq_0 
	 \\
	 \leq{}&   \int  \frac{ \langle p_{0} - q_0 \rangle^{5d+5}}{  \langle \hbar^{-1}( p_{0} - q_0) \rangle^{6d+6}}   \langle p_{0}\rangle^{4d+4} | \hat{\varphi}_\hbar(\tfrac{p_0}{\hbar})|^2  \,dp_0 dq_0
	 +
	  \int  \frac{ \langle p_{0} - q_0 \rangle}{  \langle \hbar^{-1}( p_{0} - q_0) \rangle^{6d+6}} \langle q_{0}\rangle^{4d+4} |\hat{\varphi}_\hbar(\tfrac{q_0}{\hbar})|^2  \,dp_0 dq_0
	  \\
	  \leq {}& \hbar^{2d} C \norm{\varphi_\hbar}^{2}_{\mathcal{H}_\hbar^{2d+2}(\R^d)}.
	\end{aligned}
\end{equation*}
Combining these estimates we get that
\begin{equation*}
	\begin{aligned}
	|\Aver{  \mathcal{J}(a,\varphi_\hbar,n,\sigma^1,\sigma^2,\kappa,\hbar)}| \leq  C_a \frac{\hbar\rho (\rho t)^{k+r-n-1}}{(k-1)!(r-n)!} (C  \norm{\hat{V}}_{1,\infty,5d+5})^{|\alpha|+|{\tilde{\alpha}}|} |\log(\tfrac{\hbar}{t})|^{n+3} \norm{\varphi_\hbar}^{2}_{\mathcal{H}_\hbar^{2d+2}(\R^d)},
	\end{aligned}
\end{equation*}
where $C_a$ only depend on the symbol $a$ and $C$ only on the dimension. This concludes the proof.
\end{proof}
 \begin{lemma}\label{lemma_gaus_reg}
Assume we are in Setting~\ref{setting_conv_twosided} and let $\delta>0$. Then there exists $\varepsilon_0>0$ independent of $\hbar$ such that
\begin{equation*}
	\begin{aligned}
	\sup_{\hbar\in I}|\Aver{ \mathcal{J}(a,\varphi_\hbar,n,\sigma^1,\sigma^2,\mathrm{id},\hbar)-  \mathcal{J}_\varepsilon(a,\varphi_\hbar,n,\sigma^1,\sigma^2,\mathrm{id},\hbar)}| \leq  C \frac{ n t^{r+k-n}}{k!(r-n)!}  (C\norm{\hat{V}}_{1,\infty,2d+2})^{|\alpha|+|{\tilde{\alpha}}|} \delta, 
	\end{aligned}
\end{equation*}
for all $\varepsilon \in (0,\varepsilon_0]$, where the constant depend on $\gamma$.
\end{lemma}
\begin{proof}
The proof will be very similar to that of Lemma~\ref{expansion_aver_bound_Mainterm}. By definition,
\begin{equation*}
	\begin{aligned}
	\MoveEqLeft \mathcal{J}(a,\varphi_\hbar,n,\sigma^1,\sigma^2,\mathrm{id},\hbar)-  \mathcal{J}_\varepsilon(a,\varphi_\hbar,n,\sigma^1,\sigma^2,\mathrm{id},\hbar)
	= \sum_{(\boldsymbol{x},\boldsymbol{\tilde{x}})\in \mathcal{X}_{\neq}^{k+r}}  \frac{ \prod_{i=1}^n \delta(x_{\sigma_i^1}- \tilde{x}_{\sigma_{i}^2}) }{\rho^n (2\pi\hbar)^{3d}\hbar^{k+r}}   \int_{[0,t]_{\leq}^k}\int_{[0,t]_{\leq}^r}  
	\\
	\times & \int  \Big( 1- \prod_{i=1}^n e^{-(\varepsilon\hbar^{-1}(\tilde{s}_{\sigma^2_{i}}-s_{\sigma^1_{i}}))^2} \Big)    e^{ i \hbar^{-1}  \langle  x-y,q_{r+1} \rangle} a(\tfrac{x+y}{2},q_{r+1}) e^{ i  \hbar^{-1} \langle  y,p_{k} \rangle}  e^{- i \hbar^{-1}  \langle  x,q_{r} \rangle} 
	e^{ -i \hbar^{-1}  \langle  y_0 - t\frac{1}{2}p_0,p_{0} \rangle} 
	\\
	\times&  \varphi_\hbar(y_0) e^{ i  \hbar^{-1} \langle  x_0 - t\frac{1}{2}q_0,q_{0} \rangle}   \overline{\varphi_\hbar(x_0)}
	\prod_{m=1}^k  \Psi_{\alpha_m}^\gamma(p_{m},p_{m-1},\infty;V) e^{ -i   \langle  \hbar^{-1/d}x_m,p_{m}-p_{m-1} \rangle} e^{i  s_{m} \frac{1}{2}\hbar^{-1} (p_{m}^2-p_{m-1}^2)}  
	\\
	\times &  
	\prod_{m=1}^r 
	\overline{ \Psi_{\beta_m}^\gamma(q_{m},q_{m-1},\infty;V) } e^{ i   \langle  \hbar^{-1/d}\tilde{x}_m,q_{m}-q_{m-1} \rangle}   e^{-i  \tilde{s}_{m} \frac{1}{2}\hbar^{-1} (q_{m}^2-q_{m-1}^2)} 
	     \,dx dy dx_0 dy_0 d\boldsymbol{p} d\boldsymbol{q}   d\boldsymbol{\tilde{s}} d\boldsymbol{s}.
	\end{aligned}
\end{equation*}
We then make the change of variables $(x,y)\mapsto (x-y,\frac{1}{2}(x+y))$ and applying the version of Lemma~\ref{LE:Exp_ran_phases} mentioned in Remark~\ref{RE:LE:Exp_ran_phases}. With this we obtain that
\begin{equation*}
	\begin{aligned}
	\MoveEqLeft \Aver{\mathcal{J}(a,\varphi_\hbar,n,\sigma^1,\sigma^2,\mathrm{id},\hbar)-  \mathcal{J}_\varepsilon(a,\varphi_\hbar,n,\sigma^1,\sigma^2,\mathrm{id},\hbar)}
	= \frac{((2\pi)^d\rho)^{k+r-n}}{(2\pi\hbar)^{3d}\hbar^{n}}     \int_{[0,t]_{\leq}^k}\int_{[0,t]_{\leq}^r}  
	\\
	\times & \int  \Big( 1- \prod_{i=1}^n e^{-(\varepsilon\hbar^{-1}(\tilde{s}_{\sigma^2_{i}}-s_{\sigma^1_{i}}))^2} \Big)    e^{ i \hbar^{-1}  \langle  x,q_{n+1} \rangle} a(y,q_{n+1}) e^{- i \hbar^{-1}  \langle  y+\frac{1}{2}x,q_n \rangle}     e^{ i  \hbar^{-1} \langle  y-\frac{1}{2}x,{p}_n \rangle}
	e^{ i \hbar^{-1}    t\frac{1}{2}p_0^2} 
	\\
	\times&  \hat{\varphi}_\hbar(p_0) e^{ -i  \hbar^{-1}   t\frac{1}{2}q_0^2}   \overline{\hat{\varphi}_\hbar(q_0)}
	 \prod_{i=1}^n \Big\{ \delta(p_{i}- q_{i} - p_{0}+ q_{0}) 
	 \Psi_{\alpha_{\sigma^1_{i}}}^\gamma(p_{i},p_{i-1})  e^{i  s_{\sigma^1_{i}} \frac{1}{2}\hbar^{-1} (p_{i}^2-p_{i-1}^2)}  \overline{ \Psi_{{\tilde{\alpha}}_{\sigma^2_{i}}}^\gamma(q_{i},q_{i-1}) }
	\\
	\times&  e^{-i  \tilde{s}_{\sigma^2_{i}} \frac{1}{2}\hbar^{-1} (q_{i}^2-q_{i-1}^2)}\Big\}
	 \prod_{i=1}^{n+1}  \prod_{m=\sigma^1_{i-1}+1}^{\sigma^1_{i}-1} \Psi_{\alpha_m}^\gamma(p_{i-1},p_{i-1})  \prod_{m=\sigma^2_{i-1}+1}^{\sigma^2_{i}-1} \overline{ \Psi_{{\tilde{\alpha}}_m}^\gamma(q_{i-1},q_{i-1}) }
	     \,dx dy dx_0 dy_0 d\boldsymbol{p} d\boldsymbol{q}   d\boldsymbol{\tilde{s}} d\boldsymbol{s},
	\end{aligned}
\end{equation*}
where we have also evaluated the integrals in $x_0$ and $y_0$.  We now evaluate the integrals in $p_1,\dots,p_n,q_{n+1}$ and $x$. This gives us that
\begin{equation*}
	\begin{aligned}
	\MoveEqLeft \Aver{\mathcal{J}(a,\varphi_\hbar,n,\sigma^1,\sigma^2,\mathrm{id},\hbar)-  \mathcal{J}_\varepsilon(a,\varphi_\hbar,n,\sigma^1,\sigma^2,\mathrm{id},\hbar)}
	= \frac{((2\pi)^d\rho)^{k+r-n}}{(2\pi\hbar)^{2d}\hbar^{n}}     \int_{[0,t]_{\leq}^k}\int_{[0,t]_{\leq}^r}   \int  a(y,\tfrac{p_0-q_0}{2}+q_n)
	\\
	\times &  \Big( 1- \prod_{i=1}^n e^{-(\varepsilon\hbar^{-1}(\tilde{s}_{\sigma^2_{i}}-s_{\sigma^1_{i}}))^2} \Big)           
	 \prod_{i=1}^n \Big\{ 
	 e^{i  s_{\sigma^1_{i}} \frac{1}{2}\hbar^{-1} ((p_0-q_0+q_i)^2-(p_0-q_0+q_{i-1})^2)} 
	 \Psi_{\alpha_{\sigma^1_{i}}}^\gamma(p_0-q_0+q_i,p_0-q_0+q_{i-1})   
	\\
	\times&  e^{-i  \tilde{s}_{\sigma^2_{i}} \frac{1}{2}\hbar^{-1} (q_{i}^2-q_{i-1}^2)}  \overline{\Psi_{{\tilde{\alpha}}_{\sigma^2_{i}}}^\gamma(q_{i},q_{i-1}) } \Big\}
	 \prod_{i=1}^{n+1} \Big\{ \prod_{m=\sigma^1_{i-1}+1}^{\sigma^1_{i}-1} \Psi_{\alpha_m}^\gamma(q_{i-1}+p_0-q_0,q_{i-1}+p_0-q_0)  
	 \\
	 \times& \prod_{m=\sigma^2_{i-1}+1}^{\sigma^2_{i}-1} \overline{ \Psi_{{\tilde{\alpha}}_m}^\gamma(q_{i-1},q_{i-1}) }\Big\}
	   e^{ i  \hbar^{-1} \langle  y,p_0-q_0 \rangle}
	e^{ i \hbar^{-1}    t\frac{1}{2}p_0^2} \hat{\varphi}_\hbar(p_0) e^{ -i  \hbar^{-1}   t\frac{1}{2}q_0^2}\overline{\hat{\varphi}_\hbar(q_0)}
	     \,dx dy dx_0 dy_0 d\boldsymbol{p} d\boldsymbol{q}   d\boldsymbol{\tilde{s}} d\boldsymbol{s}.
	\end{aligned}
\end{equation*}
Recalling the expression for the $\Psi$ functions from \eqref{obs_con_form_psi} we can rewrite these functions as follows
 \begin{equation*}
	\begin{aligned}
	 \Psi_{m}^\gamma(p_m,p_0,\infty ; V)
	&=    \int_{\R_{+}^{m-1}}  \int  
  \hat{\mathcal{V}}_{m}(p_m,p_0,\boldsymbol{\eta}_{1,m-1})   
  	e^{- \frac{1}{2} \gamma\boldsymbol{t}_{1,m-1}^{+}}
	    \prod_{i=1}^{m-1} e^{i \frac{1}{2} t_{i}(\eta_i^2-p_0^2)} \, d\boldsymbol{\eta}_{1,m-1}  d\boldsymbol{t}
	    \\
	    &= \int_{\R_{+}^{m-1}}  \int  
  \tilde{\mathcal{V}}_{m}(p_m-p_0,\boldsymbol{\eta}_{1,m-1})   
	    \prod_{i=1}^{m-1} e^{i \frac{1}{2} t_{i}(\eta_i^2+i\gamma)}  e^{i t_{i}\langle \eta_i,p_0\rangle} \, d\boldsymbol{\eta}_{1,m-1}  d\boldsymbol{t},
	\end{aligned}
\end{equation*}
where the functions  $\tilde{\mathcal{V}}_{m}(p_m-p_0,\boldsymbol{\eta}_{1,m-1})$ is defined by
\begin{equation*}
	\tilde{\mathcal{V}}_{m}(p_{m}-p_0,\boldsymbol{\eta}_{1,m-1}) = \frac{1}{(2\pi)^{dm}}  \begin{cases}
	   \hat{V}(p_{m}-p_{0}) & \text{if $m=1$,} 
	 \\
	  \hat{V}(p_{m}-p_{0}-\eta_{m-1}) \hat{V}(\eta_{1} ) \prod_{i=2}^{m-1}  \hat{V}(\eta_{i}-\eta_{i-1}) & \text{if $m>1$.} 
	  \end{cases}
\end{equation*}
Using this expression we get that
\begin{equation*}
	\begin{aligned}
	\MoveEqLeft \Aver{\mathcal{J}(a,\varphi_\hbar,n,\sigma^1,\sigma^2,\mathrm{id},\hbar)-  \mathcal{J}_\varepsilon(a,\varphi_\hbar,n,\sigma^1,\sigma^2,\mathrm{id},\hbar)}
	= \frac{((2\pi)^d\rho)^{k+r-n}}{(2\pi\hbar)^{2d}\hbar^{n}}     \int_{[0,t]_{\leq}^k}\int_{[0,t]_{\leq}^r}  \int_{\R^{a_k+\tilde{a}_r}_{+}} 
	\\
	\times &\int   \Big( 1- \prod_{i=1}^n e^{-(\varepsilon\hbar^{-1}(\tilde{s}_{\sigma^2_{i}}-s_{\sigma^1_{i}}))^2} \Big)           
	 \prod_{i=1}^n e^{-i  (\tilde{s}_{\sigma^2_{i}}-s_{\sigma^1_{i}}) \frac{1}{2}\hbar^{-1} (q_{i}^2-q_{i-1}^2)} e^{i\langle q_i, l_i(\hbar,\boldsymbol{s},\boldsymbol{\tilde{s}},q_0,p_0,\boldsymbol{t},\boldsymbol{\tilde{t}},\boldsymbol{\eta},\boldsymbol{\xi})\rangle}
	  \prod_{i=1}^{a_k} e^{i \frac{1}{2} t_{i}(\eta_i^2+i\gamma)} 
	  \\
	  \times&\prod_{i=1}^{\tilde{a}_r} e^{-i \frac{1}{2} t_{i}(\xi_i^2-i\gamma)} 
	 \prod_{i=1}^n 
	 \tilde{\mathcal{V}}_{\alpha_{\sigma^1_{i}}}(q_i-q_{i-1},\boldsymbol{\eta})   
	 \overline{ \tilde{\mathcal{V}}_{{\tilde{\alpha}}_{\sigma^2_{i}}}(q_{i}-q_{i-1},\boldsymbol{\xi}) }
	 \prod_{i=1}^{n+1}  \prod_{m=\sigma^1_{i-1}+1}^{\sigma^1_{i}-1}  \tilde{\mathcal{V}}_{\alpha_m}(0,\boldsymbol{\eta})  
	  \prod_{m=\sigma^2_{i-1}+1}^{\sigma^2_{i}-1} \overline{  \tilde{\mathcal{V}}_{{\tilde{\alpha}}_m}(0,\boldsymbol{\xi}) }
	  \\
	  \times&
	e^{ i  \hbar^{-1} \langle  y,p_0-q_0 \rangle} a(y,\tfrac{p_0-q_0}{2}+q_n)
	e^{ i \hbar^{-1}    t\frac{1}{2}p_0^2} \hat{\varphi}_\hbar(p_0) e^{ -i  \hbar^{-1}   t\frac{1}{2}q_0^2}\overline{\hat{\varphi}_\hbar(q_0)}
	     \,dx dy dx_0 dy_0 dp_0 d\boldsymbol{q}d\boldsymbol{t}d\boldsymbol{\tilde{t}}   d\boldsymbol{\tilde{s}} d\boldsymbol{s},
	\end{aligned}
\end{equation*}
where  the functions $l_i(\hbar,\boldsymbol{s},\boldsymbol{\tilde{s}},q_0,p_0,\boldsymbol{t},\boldsymbol{\tilde{t}},\boldsymbol{\eta},\boldsymbol{\xi})$ are linear functions. As in previous proofs we will now use Lemma~\ref{app_quadratic_integral_tech_est} for the variables $q_1,\dots,q_n$. Moreover, we use integration by parts in the variable  $y$ to get the factor $\sqrt{1+|\hbar^{-1}(q_0-p_0)|^2}^{-d-1}$. We then obtain the estimate
\begin{equation}\label{lemma_gaus_reg_2.1}
	\begin{aligned}
	\MoveEqLeft |\Aver{ \mathcal{J}(a,\varphi_\hbar,n,\sigma^1,\sigma^2,\mathrm{id},\hbar)-  \mathcal{J}_\varepsilon(a,\varphi_\hbar,n,\sigma^1,\sigma^2,\mathrm{id},\hbar)}|\leq   \frac{C\rho^{k+r-n}}{\hbar^{2d}}  \norm{\hat{V}}_{1,\infty,2d+2}^{|\alpha|+|\tilde{\alpha}|}  \int_{\R^{a_k}_{+}}   e^{-\frac{1}{2} \boldsymbol{t}_{1,a_k}^{+}\gamma}\,d\boldsymbol{t}
	\\
	\times& \sup_{|\beta|\leq d+1}\int\frac{ | \partial_y^{\beta}a(y,q)  \hat{\varphi}_\hbar(\tfrac{p_0}{\hbar})   {\hat{\varphi}_\hbar(\tfrac{q_0}{\hbar})}|}{ (1+|\hbar^{-1}(q_0-p_0)|^2)^{\frac{d+1}{2}} } dydp_0dq_0dq \int_{\R^{\tilde{a}_r}_{+}} e^{- \frac{1}{2} \boldsymbol{\tilde{t}}_{1,\tilde{a}_r}^{+}\gamma} \,d\boldsymbol{\tilde{t}}   \int_{[0,t]_{\leq}^r} \int_{[0,t]_{\leq}^k}  \tfrac{1}{\hbar \max(1,\hbar^{-1} |  \tilde{s}_{\sigma_{n}^2} -s_{\sigma_n^1} |)^\frac{d}{2}}
	\\
	\times&  \Big( 1- \prod_{i=1}^n e^{-(\varepsilon\hbar^{-1}(\tilde{s}_{\sigma^2_{\kappa(i)}}-s_{\sigma^1_{i}}))^2} \Big)   \prod_{i=1}^{n-1} \tfrac{1}{\hbar \max(1,\hbar^{-1} |  \tilde{s}_{\sigma_{i}^2} -s_{\sigma_i^1} -( \tilde{s}_{\sigma_{i+1}^2}-s_{\sigma_{i+1}^1}) |)^\frac{d}{2}} d\boldsymbol{s}d\tilde{\boldsymbol{s}},
	\end{aligned}
\end{equation}
 Firstly we observe that using the inequality $2| \hat{\varphi}_\hbar(\tfrac{p_0}{\hbar})   {\hat{\varphi}_\hbar(\tfrac{q_0}{\hbar})}| \leq | \hat{\varphi}_\hbar(\tfrac{p_0}{\hbar}) |^2+ | {\hat{\varphi}_\hbar(\tfrac{q_0}{\hbar})}|^2$ we get the estimate
\begin{equation}\label{lemma_gaus_reg_2.2}
	\begin{aligned}
	 \int\frac{ | \partial_y^{\beta}a(y,q)  \hat{\varphi}_\hbar(\tfrac{p_0}{\hbar})   {\hat{\varphi}_\hbar(\tfrac{q_0}{\hbar})}|}{ (1+|\hbar^{-1}(q_0-p_0)|^2)^{\frac{d+1}{2}} } dydp_0dq_0dq
	\leq C \hbar^{2d} \norm{\varphi_\hbar}_{L^2(\R^d)}^2 \norm{\partial_x^{\beta}a(x,p)}_{L^1(\R^d\times\R^d)}.
	\end{aligned}
\end{equation}
Since $\gamma>0$ we get that
\begin{equation}\label{lemma_gaus_reg_2.22}
	\begin{aligned}
	\int_{\R^{a_k}_{+}}   e^{-\frac{1}{2} \boldsymbol{t}_{1,a_k}^{+}\gamma}\,d\boldsymbol{t}
 \int_{\R^{\tilde{a}_r}_{+}} e^{- \frac{1}{2} \boldsymbol{\tilde{t}}_{1,\tilde{a}_r}^{+}\gamma} \,d\boldsymbol{\tilde{t}}  \leq C_\gamma.
 	\end{aligned}
\end{equation}
Preforming the change of variables $\tilde{s}_{\sigma^2_{i}}\mapsto \hbar^{-1}(\tilde{s}_{\sigma^2_{i}}-s_{\sigma^1_{i}})$ for all $i$ we get that 
\begin{equation}\label{lemma_gaus_reg_2.3}
	\begin{aligned}
	\MoveEqLeft   \int_{[0,t]_{\leq}^r} \int_{[0,t]_{\leq}^k}  \tfrac{ 1- \prod_{i=1}^n e^{-(\varepsilon\hbar^{-1}(\tilde{s}_{\sigma^2_{\kappa(i)}}-s_{\sigma^1_{i}}))^2}}{\hbar \max(1,\hbar^{-1} |  \tilde{s}_{\sigma_{n}^2} -s_{\sigma_n^1} |)^\frac{d}{2}}   \prod_{i=1}^{n-1} \tfrac{1}{\hbar \max(1,\hbar^{-1} |  \tilde{s}_{\sigma_{i}^2} -s_{\sigma_i^1} -( \tilde{s}_{\sigma_{i+1}^2}-s_{\sigma_{i+1}^1}) |)^\frac{d}{2}} d\boldsymbol{s}d\tilde{\boldsymbol{s}}
	\\
	&=\frac{t^{r+k-n}}{k!(r-n)!}  \int_{\R^n} \frac{ 1- \prod_{i=1}^n e^{-(\varepsilon s_i)^2}}{ \max(1, | s_n |)^\frac{d}{2}}   \prod_{i=1}^{n-1} \frac{1}{\max(1, |  s_i-s_{i+1} |)^\frac{d}{2}} d\boldsymbol{s}.
	\end{aligned}
\end{equation}
Since we have that
\begin{equation}
	 \int_{\R^n} \frac{ 1}{ \max(1, | s_n |)^\frac{d}{2}}   \prod_{i=1}^{n-1} \frac{1}{\max(1, |  s_i-s_{i+1} |)^\frac{d}{2}} d\boldsymbol{s} <\infty .
\end{equation}
It follows by dominated convergence that we can find $\varepsilon_0$ such that
\begin{equation}\label{lemma_gaus_reg_2.4}
	\begin{aligned}
	 \int_{\R^n} \frac{ 1- \prod_{i=1}^n e^{-(\varepsilon s_i)^2}}{ \max(1, | s_n |)^\frac{d}{2}}   \prod_{i=1}^{n-1} \frac{1}{\max(1, |  s_i-s_{i+1} |)^\frac{d}{2}} d\boldsymbol{s} \leq C \delta
	\end{aligned}
\end{equation}
for all $\varepsilon\in(0,\varepsilon_0]$. Combining the estimates in \eqref{lemma_gaus_reg_2.1}, \eqref{lemma_gaus_reg_2.2}, \eqref{lemma_gaus_reg_2.22} , \eqref{lemma_gaus_reg_2.3} and \eqref{lemma_gaus_reg_2.4} we arrive at the estimate
\begin{equation*}
	\begin{aligned}
	|\Aver{ \mathcal{J}(a,\varphi_\hbar,n,\sigma^1,\sigma^2,\mathrm{id},\hbar)-  \mathcal{J}_\varepsilon(a,\varphi_\hbar,n,\sigma^1,\sigma^2,\mathrm{id},\hbar)}| \leq C \frac{ n t^{r+k-n}}{k!(r-n)!}  (C\norm{\hat{V}}_{1,\infty,2d+2})^{|\alpha|+|\tilde{\alpha}|} \delta  
	\end{aligned}
\end{equation*}
for all $\varepsilon\in(0,\varepsilon_0]$, where the constant $C$ is independent of $\hbar$ but depend on $\gamma$.
\end{proof}
 In the introduction we presented the Weyl quantisation of a symbol. We will in the following lemma also need the ``$0$-quantisation'' of a symbol. For $a\in \mathcal{S}(\R^{2d})$ this quantisation is given by    
\begin{equation*}
	 \mathrm{Op}_{0,\hbar}(a)\psi(x) = \frac{1}{(2\pi\hbar)^d} \int_{\R^{2d}} e^{i\hbar^{-1}\langle x-y,p \rangle} a(y,p) \psi(y) \, dydp 
\end{equation*}
with $\psi\in\mathcal{S}(\R^d)$. The two types of quantisation are related, in particular we have that
\begin{equation}\label{obs_con_quantisation}
	\norm{\OpW(a)- \mathrm{Op}_{0,\hbar}(a)}_{\mathrm{op}} =\mathcal{O}(\hbar),
\end{equation}
 for all $a\in\mathcal{S}(\R^{2d})$; for more details see \cite{MR897108,MR2952218}. This in particular implies that if we have convergence of our semiclassical measures for the Weyl-quantisation we also have it for the $0$-quantisation, and vice versa. This we have already used in the proof of Lemma~\ref{L.con_term_onesided_1}.
\begin{lemma}\label{lemma_gaus_reg_3}
Assume we are in Setting~\ref{setting_conv_twosided} and assume that $\kappa$ is the identity. Then we have that 
\begin{equation}\label{composit_sym_epsilon}
	\begin{aligned}
	\Aver{ \mathcal{J}_\varepsilon(a,\varphi_\hbar,n,\sigma^1,\sigma^2,\kappa,\hbar)} = \langle \mathrm{Op}_{0,\hbar}(a\circ\Phi_{\gamma,\varepsilon}^t) \varphi_\hbar, \varphi_\hbar \rangle  + \mathcal{O}_\varepsilon(\sqrt{\hbar}),
	\end{aligned}
\end{equation}
where the symbol $a\circ\Phi_{\gamma,\varepsilon}^t$ is Schwartz class and given by
\begin{equation*}
	\begin{aligned}
	a\circ\Phi_{\gamma,\varepsilon}^t(x,q_0)={}& \frac{((2\pi)^d\rho)^{k+r-n} \pi^{\frac{n}{2}}}{\varepsilon^n}   \int_{[0,t]_{\leq}^n}  \int_{\R^{nd}}    e^{ i   \hbar^{-1}\langle  x_0 -y_0,q_{0} \rangle} 
	  \prod_{i=1}^n   e^{ \frac{1}{4}\varepsilon^{-2} \frac{1}{2}(q_i^2-q_{i-1}^2)}   
	\\
	&\times 
	    \prod_{i=1}^n \Psi_{\alpha_{\sigma^1_{i}}}^\gamma(q_{i} ,q_{i-1})   \overline{ \Psi_{{\tilde{\alpha}}_{\sigma^2_{i}}}^\gamma(q_{i},q_{i-1}) }  \prod_{i=1}^{n+1}   \prod_{m=\sigma^2_{i-1}+1}^{\sigma^2_{i}-1} \overline{ \Psi_{{\tilde{\alpha}}_m}^\gamma(q_{i-1},q_{i-1}) }
	 \\
	&\times \prod_{m=\sigma^1_{i-1}+1}^{\sigma^1_{i}-1} \Psi_{\alpha_m}^\gamma(q_{i-1} ,q_{i-1} )    \prod_{i=1}^{n+1} \frac{(s_{i-1}-s_i)^{\sigma_i^1-\sigma_{i-1}^1-1}}{(\sigma_i^1-\sigma_{i-1}^1-1)!}  \frac{(s_{i-1}-s_i )^{\sigma_i^2-\sigma_{i-1}^2-1}}{(\sigma_i^2-\sigma_{i-1}^2-1)!}
	\\
	&\times a(x+ y_0 - tq_0 -\sum_{i=1}^n s_{i}(q_{i} - q_{i-1}) , q_n)     \,  d\boldsymbol{q}_{1,n} d\boldsymbol{s}_{n,1}.
	\end{aligned}
\end{equation*}
\end{lemma}
We observe that from this lemma and \eqref{obs_con_quantisation} we get that
\begin{equation}
	\lim_{\hbar\rightarrow0}\Aver{ \mathcal{J}_\varepsilon(a,\varphi_\hbar,n,\sigma^1,\sigma^2,\kappa,\hbar)}  = \int_{\R^{2d}} a\circ\Phi^t_{\gamma,\varepsilon}(y,q_0) \,d\mu_0(y,q_0).
\end{equation}
\begin{proof}
From the definition of the number $\mathcal{J}_\varepsilon(a,\varphi_\hbar,n,\sigma^1,\sigma^2,\kappa,\hbar)$ and by applying Campbell's theorem we get that 
\begin{equation}
	\begin{aligned}
	\MoveEqLeft \Aver{ \mathcal{J}_\varepsilon(a,\varphi_\hbar,n,\sigma^1,\sigma^2,\kappa,\hbar)} 
	=\frac{((2\pi)^d\rho)^{k+r-n}}{(2\pi\hbar)^{3d}\hbar^{2n}}  \int_{[0,t]_{\leq}^k}\int_{[0,t]_{\leq}^r}  \int e^{ i \hbar^{-1}  \langle  x-y,q_{n+1} \rangle} \prod_{i=1}^n e^{-(\varepsilon\hbar^{-1}(\tilde{s}_{\sigma^2_{i}}-s_{\sigma^1_{i}}))^2}    e^{ i  \hbar^{-1} \langle  y,p_{n} \rangle} 
	\\
	&\times  e^{- i \hbar^{-1}  \langle  x,q_{n} \rangle}   e^{ -i \hbar^{-1}  \langle  y_0,p_{0} \rangle} e^{ i  \hbar^{-1} \langle x_0,q_{0} \rangle} 
	 \prod_{i=1}^n \Big\{ e^{ -i   \langle  \hbar^{-1/d}x_{i},p_{i}-p_{i-1} -  q_{i}+q_{i-1}  \rangle} 
	  \Psi_{\alpha_{\sigma^1_{i}}}^\gamma(p_{i},p_{i-1})  e^{i  s_{\sigma^1_{i}} \frac{1}{2}\hbar^{-1} (p_{i}^2-p_{i-1}^2)} 
	  \\
	  &\times \overline{ \Psi_{\beta_{\sigma^2_{i}}}^\gamma(q_{i},q_{i-1}) } e^{-i  \tilde{s}_{\sigma^2_{i}} \frac{1}{2}\hbar^{-1} (q_{i}^2-q_{i-1}^2)}\Big\} \prod_{i=1}^{n+1}  \prod_{m=\sigma^1_{i-1}+1}^{\sigma^1_{i}-1} \Psi_{\alpha_m}^\gamma(p_{i-1},p_{i-1})  \prod_{m=\sigma^2_{i-1}+1}^{\sigma^2_{i}-1} \overline{ \Psi_{\beta_m}^\gamma(q_{i-1},q_{i-1}) }
	\\
	&\times a(\tfrac{x+y}{2},q_{n+1})   e^{i  t \frac{1}{2}\hbar^{-1} p_0^2} \varphi_\hbar(y_0) e^{-i  t \frac{1}{2}\hbar^{-1} q_0^2} \overline{\varphi_\hbar(x_0)}  \, dy d\boldsymbol{x} dy_0 d\boldsymbol{p}  d\boldsymbol{q} d\boldsymbol{\tilde{s}} d\boldsymbol{s},
	\end{aligned}
\end{equation}
where we are using the convention $\sigma^1_{n+1}=k+1$ and $\sigma^2_{\kappa(n+1)}=r+1$. We are here not explicit writing the dependence on $\infty$ and $V$ for the $\Psi$-functions. We have also evaluated some of the integrals in the position variables and momentum variables and furthermore a relabelling of some variables. To see how the matching of the remaining momentum variables is, it will be convenient to do the change of variables $x_1\mapsto \hbar^{-\frac{1}{d}}x_1$ and $x_i\mapsto \hbar^{-\frac{1}{d}}(x_i - x_{i-1})$ for $i  \in \{2,\dots,n\}$. Furthermore we also make the change of variables $(x,y)\mapsto (x-y,\frac{1}{2}(x+y))$. This gives us
\begin{equation}
	\begin{aligned}
	\MoveEqLeft  \Aver{ \mathcal{J}_\varepsilon(a,\varphi_\hbar,n,\sigma^1,\sigma^2,\kappa,\hbar)} 
	=  \frac{((2\pi)^d\rho)^{k+r-n}}{(2\pi\hbar)^{3d}\hbar^n}  \int_{[0,t]_{\leq}^k}\int_{[0,t]_{\leq}^r}  \int\prod_{i=1}^n e^{-(\varepsilon\hbar^{-1}(\tilde{s}_{\sigma^2_{i}}-s_{\sigma^1_{i}}))^2}   e^{ i  \hbar^{-1} \langle  x,q_{n+1} - \frac{1}{2}(p_n +q_n) \rangle}
	\\
	&\times e^{ i \hbar^{-1}  \langle  y,p_{n} - q_n \rangle}   e^{ -i \hbar^{-1}  \langle  y_0,p_{0} \rangle} e^{ i   \hbar^{-1}\langle  x_0,q_{0} \rangle} 
	 \prod_{i=1}^n  \Big\{e^{ -i   \langle  x_{i},p_{i-1} -q_{i-1}-p_n+ q_{n}   \rangle} 
	 \Psi_{\alpha_{\sigma^1_{i}}}^\gamma(p_{i},p_{i-1})  e^{i  s_{\sigma^1_{i}} \frac{1}{2}\hbar^{-1} (p_{i}^2-p_{i-1}^2)} 
	 \\
	 &\times \overline{ \Psi_{\beta_{\sigma^2_{i}}}^\gamma(q_{i},q_{i-1}) } e^{-i  \tilde{s}_{\sigma^2_{i}} \frac{1}{2}\hbar^{-1} (q_{i}^2-q_{i-1}^2)}\Big\}
	  \prod_{i=1}^{n+1}  \prod_{m=\sigma^1_{i-1}+1}^{\sigma^1_{i}-1} \Psi_{\alpha_m}^\gamma(p_{i-1},p_{i-1})  \prod_{m=\sigma^2_{i-1}+1}^{\sigma^2_{i}-1} \overline{ \Psi_{\beta_m}^\gamma(q_{i-1},q_{i-1}) }
	\\
	&\times a(y,q_{n+1})   e^{i  t \frac{1}{2}\hbar^{-1} p_0^2} \varphi_\hbar(y_0) e^{-i  t \frac{1}{2}\hbar^{-1} q_0^2} \overline{\varphi_\hbar(x_0)}  \, dy d\boldsymbol{x}dy_0 d\boldsymbol{p} d\boldsymbol{q} d\boldsymbol{\tilde{s}} d\boldsymbol{s}.
	\end{aligned}
\end{equation}
We now evaluate all integrals in all $x$'es except $x_0$, in all $p$'s except $p_n$ and $q_{n+1}$. This gives us
\begin{equation}
	\begin{aligned}
	\MoveEqLeft  \Aver{ \mathcal{J}_\varepsilon(a,\varphi_\hbar,n,\sigma^1,\sigma^2,\kappa,\hbar)}
	= \frac{((2\pi)^d\rho)^{k+r-n}}{(2\pi\hbar)^{2d}\hbar^n}  \int_{[0,t]_{\leq}^k}\int_{[0,t]_{\leq}^r}  \int  e^{ i \hbar^{-1}  \langle  y-y_0,p_{n} - q_n \rangle} e^{ i   \hbar^{-1}\langle  x_0 -y_0,q_{0} \rangle}  
	\\
	\times&  
	 \prod_{i=1}^n \Big\{ e^{-(\varepsilon\hbar^{-1}(\tilde{s}_{\sigma^2_{i}}-s_{\sigma^1_{i}}))^2}  \Psi_{\alpha_{\sigma^1_{i}}}^\gamma(q_{i}+p_n- q_{n} ,q_{i-1}+p_n- q_{n} )
	  e^{i  s_{\sigma^1_{i}} \frac{1}{2}\hbar^{-1} ((q_{i}+p_n- q_{n})^2-(q_{i-1}+p_n- q_{n} )^2)} 
	\\
	\times &  \overline{ \Psi_{\beta_{\sigma^2_{i}}}^\gamma(q_{i},q_{i-1}) } e^{-i  \tilde{s}_{\sigma^2_{i}} \frac{1}{2}\hbar^{-1} (q_{i}^2-q_{i-1}^2)} \Big\} e^{i  t \frac{1}{2}\hbar^{-1} (p_n+q_0-q_n)^2}   a(y, \tfrac{1}{2}(p_n +q_n))   \varphi_\hbar(y_0) e^{-i  t \frac{1}{2}\hbar^{-1} q_0^2} \overline{\varphi_\hbar(x_0)}  
	\\
	\times & \prod_{i=1}^{n+1}  \prod_{m=\sigma^2_{i-1}+1}^{\sigma^2_{i}-1} \overline{ \Psi_{\beta_m}^\gamma(q_{i-1},q_{i-1}) }  \prod_{m=\sigma^1_{i-1}+1}^{\sigma^1_{i}-1} \Psi_{\alpha_m}^\gamma(q_{i-1}+p_n- q_{n} ,q_{i-1}+p_n- q_{n})   \,dx_0 dy  dy_0 dp_{n} d\boldsymbol{q} d\boldsymbol{\tilde{s}}d\boldsymbol{s} .
	\end{aligned}
\end{equation}
By rewriting the phase functions and making the change of variables $p_n\mapsto p_n-q_{n}$ and $y\mapsto y- y_0 +\tfrac{1}{2}t(p_{n}+q_{n} - q_n)$  we get the  expression
\begin{equation}
	\begin{aligned}
	\MoveEqLeft  \Aver{ \mathcal{J}_\varepsilon(a,\varphi_\hbar,n,\sigma^1,\sigma^2,\kappa,\hbar)}
	= \frac{((2\pi)^d\rho)^{k+r-n}}{(2\pi\hbar)^{2d}\hbar^n}  \int_{[0,t]_{\leq}^k}\int_{[0,t]_{\leq}^r}  \int e^{ i \hbar^{-1}  \langle  y +tq_0 +\sum_{i=1}^n s_{\sigma^1_{i}}(q_{i} - q_{i-1}) ,p_{n}+q_{n} - q_n \rangle}   
	\\
	&\times  \prod_{i=1}^n e^{-(\varepsilon\hbar^{-1}(\tilde{s}_{\sigma^2_{i}}-s_{\sigma^1_{i}}))^2}   
	e^{-i ( \tilde{s}_{\sigma^2_{i}}-s_{\sigma^1_{i}}) \frac{1}{2}\hbar^{-1} (q_{i}^2-q_{i-1}^2)}   \Psi_{\alpha_{\sigma^1_{i}}}^\gamma(q_{i}+p_n ,q_{i-1}+p_n )  \overline{ \Psi_{\beta_{\sigma^2_{i}}}^\gamma(q_{i},q_{i-1}) } 
	\\
	&\times 
	 e^{ i   \hbar^{-1}\langle  x_0 -y_0,q_{0} \rangle}\prod_{i=1}^{n+1}   \prod_{m=\sigma^2_{i-1}+1}^{\sigma^2_{i}-1} \overline{ \Psi_{\beta_m}^\gamma(q_{i-1},q_{i-1}) }
	 \prod_{m=\sigma^1_{i-1}+1}^{\sigma^1_{i}-1} \Psi_{\alpha_m}^\gamma(q_{i-1}+p_n,q_{i-1}+p_n ) 
	\\
	&\times   a(y+ y_0 -\tfrac{1}{2}t(p_{n}+q_{n} - q_n), \tfrac{1}{2}(p_n + q_{n} +q_n))   \varphi_\hbar(y_0)  \overline{\varphi_\hbar(x_0)}   \,dx_0 dy  dy_0 dp_{n}  d\boldsymbol{q} d\boldsymbol{\tilde{s}} d\boldsymbol{s},
	\end{aligned}
\end{equation}
We will just consider the integrals in the time variables. Our integrals have the  form
\begin{equation}
	\begin{aligned}
	\MoveEqLeft   \int_{[0,t]_{\leq}^k}\int_{[0,t]_{\leq}^r} f(s_{\sigma_1^1},\dots,s_{\sigma_n^1}) g(\tilde{s}_{\sigma_2^1},\dots,{s}_{\sigma_n^2})d\boldsymbol{\tilde{s}} d\boldsymbol{s}
	\\
	={}&  \int_{\R^{2n}} \boldsymbol{1}_{[0,t]_{\leq}^n}(s_n,\dots,s_1) \boldsymbol{1}_{[0,t]_{\leq}^n}(\tilde{s}_n,\dots,\tilde{s}_1) f(s_{1},\dots,s_{n}) g(\tilde{s}_{1},\dots,{s}_{n}) 
	\\
	&\times \prod_{i=1}^{n+1} \frac{(s_{i-1}-s_i)^{\sigma_i^1-\sigma_{i-1}^1-1}}{(\sigma_i^1-\sigma_{i-1}^1-1)!}  \frac{(\tilde{s}_{i-1}-\tilde{s}_i)^{\sigma_i^2-\sigma_{i-1}^2-1}}{(\sigma_i^2-\sigma_{i-1}^2-1)!}\,d\boldsymbol{\tilde{s}} d\boldsymbol{s},
	\end{aligned}
\end{equation}
We then have the expression
\begin{equation}
	\begin{aligned}
	\MoveEqLeft  \Aver{ \mathcal{J}_\varepsilon(a,\varphi_\hbar,n,\sigma^1,\sigma^2,\kappa,\hbar)}
	= \frac{((2\pi)^d\rho)^{k+r-n}}{(2\pi\hbar)^{2d}\hbar^n}  \int_{\R^{2n}}  \int  e^{ i \hbar^{-1}  \langle  y +tq_0 +\sum_{i=1}^n s_{i}(q_{i} - q_{i-1}) ,p_{n} \rangle} e^{ i   \hbar^{-1}\langle  x_0 -y_0,q_{0} \rangle} 
	\\
	&\times   \prod_{i=1}^n e^{-(\varepsilon\hbar^{-1}(\tilde{s}_{i}-s_{i}))^2}   e^{-i  \hbar^{-1}(\tilde{s}_{i}-s_{i}) \frac{1}{2}(q_i^2-q_{i-1}^2)}   \Psi_{\alpha_{\sigma^1_{i}}}^\gamma(q_{i}+p_n ,q_{i-1}+p_n)   \overline{ \Psi_{{\tilde{\alpha}}_{\sigma^2_{i}}}^\gamma(q_{i},q_{i-1}) }  
	\\
	&\times 
	     \prod_{i=1}^{n+1}   \prod_{m=\sigma^2_{i-1}+1}^{\sigma^2_{i}-1} \overline{ \Psi_{{\tilde{\alpha}}_m}^\gamma(q_{i-1},q_{i-1}) }
	\prod_{m=\sigma^1_{i-1}+1}^{\sigma^1_{i}-1} \Psi_{\alpha_m}^\gamma(q_{i-1}+p_n ,q_{i-1}+p_n ) 
	 \\
	&\times \boldsymbol{1}_{[0,t]_{\leq}^n}(\tilde{s}_n,\dots,\tilde{s}_1 )  \boldsymbol{1}_{[0,t]_{\leq}^n}(s_n,\dots,s_1)    \prod_{i=1}^{n+1} \frac{(s_{i-1}-s_i)^{\sigma_i^1-\sigma_{i-1}^1-1}}{(\sigma_i^1-\sigma_{i-1}^1-1)!}  \frac{(\ \tilde{s}_{i-1}- \tilde{s}_i )^{\sigma_i^2-\sigma_{i-1}^2-1}}{(\sigma_i^2-\sigma_{i-1}^2-1)!}
	\\
	&\times a(y+ y_0 -\tfrac{1}{2}tp_{n}, \tfrac{1}{2}p_n  +q_n)  \varphi_\hbar(y_0)  \overline{\varphi_\hbar(x_0)}    \,dx_0 dy  dy_0 dp_{n}  d\boldsymbol{q} d\boldsymbol{s}d\boldsymbol{\tilde{s}}.
	\end{aligned}
\end{equation}
Here we make the change of variables $p_n\mapsto \hbar^{-1}p_n$. As in the proof of Lemma~\ref{L.con_term_onesided_1} we preform a first order Taylor expansion of the function
\begin{equation}
	\begin{aligned}
	p_n \mapsto  \prod_{i=1}^n & \Psi_{\alpha_{\sigma^1_{i}}}^\gamma(q_{i}+\hbar p_n ,q_{i-1}+\hbar p_n)   \prod_{i=1}^{n+1}  \prod_{m=\sigma^1_{i-1}+1}^{\sigma^1_{i}-1} \Psi_{\alpha_m}^\gamma(q_{i-1}+\hbar p_n ,q_{i-1}+\hbar p_n ) 
	\\
	&\times a(y+ y_0 -\tfrac{1}{2}t\hbar p_{n}, \tfrac{1}{2}\hbar p_n  +q_n). 
	\end{aligned}
\end{equation}
This will, with a similar argument to that in Lemma~\ref{L.con_term_onesided_1}, produce an error of order $\hbar$. Hence 
\begin{equation}\label{gaus_reg_est_3.0}
	\begin{aligned}
	\MoveEqLeft  \Aver{ \mathcal{J}_\varepsilon(a,\varphi_\hbar,n,\sigma^1,\sigma^2,\kappa,\hbar)}
	= \frac{((2\pi)^d\rho)^{k+r-n}}{(2\pi\hbar)^{d}\hbar^n}  \int_{\R^{2n}}  \int_{\R^{(n+3)d}}    e^{ i   \hbar^{-1}\langle  x_0 -y_0,q_{0} \rangle} 
	  \prod_{i=1}^n \Big\{ e^{-(\varepsilon\hbar^{-1}(\tilde{s}_{i}-s_{i}))^2}   
	  \\
	  &\times e^{-i  \hbar^{-1}(\tilde{s}_{i}-s_{i}) \frac{1}{2}(q_i^2-q_{i-1}^2)}   
	 \Psi_{\alpha_{\sigma^1_{i}}}^\gamma(q_{i} ,q_{i-1})   \overline{ \Psi_{{\tilde{\alpha}}_{\sigma^2_{i}}}^\gamma(q_{i},q_{i-1}) } \Big\}
	 \boldsymbol{1}_{[0,t]_{\leq}^n}(\tilde{s}_n,\dots,\tilde{s}_1 )  \boldsymbol{1}_{[0,t]_{\leq}^n}(s_n,\dots,s_1)  
	 \\
	 &\times \prod_{i=1}^{n+1} \frac{(s_{i-1}-s_i)^{\sigma_i^1-\sigma_{i-1}^1-1}}{(\sigma_i^1-\sigma_{i-1}^1-1)!}  \frac{(\ \tilde{s}_{i-1}- \tilde{s}_i )^{\sigma_i^2-\sigma_{i-1}^2-1}}{(\sigma_i^2-\sigma_{i-1}^2-1)!}  \prod_{m=\sigma^2_{i-1}+1}^{\sigma^2_{i}-1} \overline{ \Psi_{{\tilde{\alpha}}_m}^\gamma(q_{i-1},q_{i-1}) }\prod_{m=\sigma^1_{i-1}+1}^{\sigma^1_{i}-1} \Psi_{\alpha_m}^\gamma(q_{i-1} ,q_{i-1} ) 
	\\
	&\times  a( y_0 - tq_0 -\sum_{i=1}^n s_{i}(q_{i} - q_{i-1}) , q_n)  \varphi_\hbar(y_0)  \overline{\varphi_\hbar(x_0)}    \,dx_0   dy_0  d\boldsymbol{q} d\boldsymbol{s}d\boldsymbol{\tilde{s}} + \mathcal{O}(\hbar).
	\end{aligned}
\end{equation}
Before we proceed we will remove the function $ \boldsymbol{1}_{[0,t]_{\leq}^n}(\tilde{s}_n,\dots,\tilde{s}_1 ) $ from our expression.  To this end we will by $\tilde{ \mathcal{J}}_\varepsilon(a,\varphi,n,\sigma^1,\sigma^2,\kappa,\hbar)$ denote the expression with this function removed. By the change of variables $\tilde{s}_n\mapsto \tilde{s}_n $, $\tilde{s}_i\mapsto \tilde{s}_i - \tilde{s}_{i-1} $ for all $i\in\{1,\dots,n-1\}$ and $s_n\mapsto s_n $, $s_i\mapsto s_i - s_{i-1} $ for all $i\in\{1,\dots,n-1\}$ and afterwards  arguing as in previous lemmas we obtain the bound
\begin{equation}\label{gaus_reg_est_3.1}
	\begin{aligned}
	\MoveEqLeft | \Aver{ \mathcal{J}_\varepsilon(a,\varphi_\hbar,n,\sigma^1,\sigma^2,\kappa,\hbar) - \tilde{ \mathcal{J}}_\varepsilon(a,\varphi_\hbar,n,\sigma^1,\sigma^2,\kappa,\hbar)}|
	\\
	\leq{}& C\hbar^{-n}  \int_{\R^{2n}}  \boldsymbol{1}_{[t,\infty)}(\boldsymbol{\tilde{s}}_{1,n}^+ )  \boldsymbol{1}_{[0,t]}(\boldsymbol{s}_{1,n}^{+})   \prod_{i=1}^n e^{-(\varepsilon\hbar^{-1}(\boldsymbol{\tilde{s}}_{i,n}^+-\boldsymbol{s}_{i,n}^+))^2}  d\boldsymbol{s}d\boldsymbol{\tilde{s}} + \mathcal{O}(\hbar),
	\end{aligned}
\end{equation}
where the constant $C$ is independent of $\hbar$. We now change variables $\tilde{s}_i \mapsto \hbar^{-1}(\tilde{s}_i-s_i) $ for all $i$, and note that
\begin{equation}\label{gaus_reg_est_3.2}
	\begin{aligned}
	\MoveEqLeft   
	  \int_{\R^{2n}}  \boldsymbol{1}_{[t,\infty)}(\hbar \boldsymbol{\tilde{s}}_{1,n}^+ +\boldsymbol{s}_{1,n}^{+} )  \boldsymbol{1}_{[0,t]}(\boldsymbol{s}_{1,n}^{+})   \prod_{i=1}^n e^{-(\varepsilon\boldsymbol{\tilde{s}}_{i,n}^+)^2}  d\boldsymbol{s}_{1,n}d\boldsymbol{\tilde{s}}_{1,n} 
	 \\
	 &\leq \sqrt{\hbar}  \int_{\R^{2n}}  \boldsymbol{1}_{[t,\infty)}(\hbar \boldsymbol{\tilde{s}}_{1,n}^+ +\boldsymbol{s}_{1,n}^{+} )  \boldsymbol{1}_{[0,t]}(\boldsymbol{s}_{1,n}^{+}) \frac{\sqrt{ \boldsymbol{\tilde{s}}_{1,n}^+ }}{\sqrt{t-\boldsymbol{s}_{1,n}^{+}}}  \prod_{i=1}^n e^{-(\varepsilon\boldsymbol{\tilde{s}}_{i,n}^+)^2}  d\boldsymbol{s}_{1,n}d\boldsymbol{\tilde{s}}_{1,n} 
	 \\
	  &\leq \sqrt{\frac{\hbar}{\varepsilon^{n+1}}} C  \int_{[0,t]^{n}}  \boldsymbol{1}_{[0,t]_{\leq}^n}(s_n,\dots,s_2)   \frac{1}{\sqrt{t-s_1}}  d\boldsymbol{s}_{1,n}
	  \leq \sqrt{\frac{\hbar}{\varepsilon^{n+1}}} C.
	\end{aligned}
\end{equation}
Combining \eqref{gaus_reg_est_3.0}, \eqref{gaus_reg_est_3.1} and \eqref{gaus_reg_est_3.2} and the change of variables $\tilde{s}_i\mapsto\hbar^{-1}( \tilde{s}_i - s_i)$ for all $i\in\{1,\dots,n\}$ we obtain that
\begin{equation} \label{gaus_reg_est_3.3}
	\begin{aligned}
	\MoveEqLeft  \Aver{ \mathcal{J}_\varepsilon(a,\varphi_\hbar,n,\sigma^1,\sigma^2,\kappa,\hbar)}
	= \frac{((2\pi)^d\rho)^{k+r-n} \pi^{\frac{n}{2}}}{(2\pi\hbar)^{d} \varepsilon^n}  \int_{\R^{n}}  \int    e^{ i   \hbar^{-1}\langle  x_0 -y_0,q_{0} \rangle} 
	  \prod_{i=1}^n \Big\{  e^{ \frac{1}{4}\varepsilon^{-2} \frac{1}{2}(q_i^2-q_{i-1}^2)}   
	 \Psi_{\alpha_{\sigma^1_{i}}}^\gamma(q_{i} ,q_{i-1})  
	 \\
	 &\times  \overline{ \Psi_{{\tilde{\alpha}}_{\sigma^2_{i}}}^\gamma(q_{i},q_{i-1}) } \Big\}  \boldsymbol{1}_{[0,t]_{\leq}^n}(s_n,\dots,s_1)    \prod_{i=1}^{n+1} \Big\{ \frac{(s_{i-1}-s_i)^{\sigma_i^1-\sigma_{i-1}^1-1}}{(\sigma_i^1-\sigma_{i-1}^1-1)!}  \frac{(s_{i-1}-s_i )^{\sigma_i^2-\sigma_{i-1}^2-1}}{(\sigma_i^2-\sigma_{i-1}^2-1)!}   
	 \\
	&\times  \prod_{m=\sigma^2_{i-1}+1}^{\sigma^2_{i}-1} \overline{ \Psi_{{\tilde{\alpha}}_m}^\gamma(q_{i-1},q_{i-1}) }
	\prod_{m=\sigma^1_{i-1}+1}^{\sigma^1_{i}-1} \Psi_{\alpha_m}^\gamma(q_{i-1} ,q_{i-1} ) \Big\}
	 a( y_0 - tq_0 -\sum_{i=1}^n s_{i}(q_{i} - q_{i-1}) , q_n)  
	 \\
	 &\times \varphi_\hbar(y_0)  \overline{\varphi_\hbar(x_0)}    \,dx_0   dy_0  d\boldsymbol{q} d\boldsymbol{s} + \mathcal{O}_\varepsilon(\sqrt{\hbar}),
	\end{aligned}
\end{equation}
where we have also evaluated the integrals in $\tilde{s}$  (these are Fourier transforms of the  functions $e^{-(\varepsilon\tilde{s}_{i})^2} $). We observe that the above expression can be written in the form
\begin{equation}
	\begin{aligned}
	\Aver{ \mathcal{J}_\varepsilon(a,\varphi_\hbar,n,\sigma^1,\sigma^2,\kappa,\hbar)} = \langle \mathrm{Op}_{0,\hbar}(a\circ\Phi_{\gamma,\varepsilon}^t) \varphi_\hbar, \varphi_\hbar \rangle  + \mathcal{O}_\varepsilon(\sqrt{\hbar}),
	\end{aligned}
\end{equation}
where the symbol $a\circ\Phi_{\gamma,\varepsilon}^t(x,q_0)$ can be seen directly in \eqref{gaus_reg_est_3.3}. Moreover it follows from Lemma~\ref{est_der_Psi_gamma} that this symbol is Schwartz class.
\end{proof}
\begin{lemma}\label{LE:con_term_twosided}
Assume we are in the same setting as in Definition~\ref{functions_for_exp_def} and let $\{\varphi_\hbar\}_{\hbar\in I}$ be a uniform semiclassical family in $\mathcal{H}^{d+5}_\hbar(\R^d)$ with Wigner measure $\mu_0$. Moreover, let 
\begin{equation*}
	k_0 = \frac{|\log(\hbar)|}{10 \log(|\log(\hbar)|)}.
\end{equation*} 
Then for any $a \in \mathcal{S}(\R^{2d})$ we have that
\begin{equation*}
	\begin{aligned}
	\MoveEqLeft \lim_{\hbar\rightarrow0} \sum_{k,r=1}^{k_0} \sum_{\alpha\in\N^k} \sum_{{\tilde{\alpha}}\in\N^r} (i\lambda)^{|\alpha|}(-i\lambda)^{|{\tilde{\alpha}}|} \Aver{  \langle \OpW(a)\mathcal{I}_\infty^\gamma(k,\alpha;\hbar)\varphi_\hbar,\mathcal{I}_\infty^\gamma(r,{\tilde{\alpha}};\hbar)\varphi_\hbar \rangle 
}
	\\
	&=  \sum_{k,r=1}^{\infty} \sum_{\alpha\in\N^k} \sum_{{\tilde{\alpha}}\in\N^r} \sum_{n=0}^{\min(k,r)} \sum_{\sigma^1 \in\mathcal{A}(k,n)}  \sum_{\sigma^2\in\mathcal{A}(r,n)} (i\lambda)^{|\alpha|}(-i\lambda)^{|{\tilde{\alpha}}|} \int a\circ\Phi^t_{\gamma,0}(y,q_0) \,d\mu_0(y,q_0),
		\end{aligned}
\end{equation*}
where the symbol $a\circ\Phi^t_{\gamma,0}(y,q_0)$ is given by
\begin{equation}\label{composit_sym}
	\begin{aligned}
	 \MoveEqLeft a\circ\Phi^t_{\gamma,0}(y,q_0)
	 =  ((2\pi)^d\rho)^{(k+r-n)} (2\pi)^n  \int_{[0,t]_{\leq}^n}  \int   
	a(y- tq_0 -\sum_{i=1}^{n}s_{i}  (q_{i}-q_{i-1}) ,q_{n}) \prod_{i=1}^{n} \delta( \tfrac{1}{2} (q_{i}^2-q_{i-1}^2))  
	\\
	&\times  \prod_{i=1}^{n}    \Psi_{\alpha_{\sigma_i^1}}^\gamma(q_{i} ,q_{i-1},\infty;V )  \overline{ \Psi_{{\tilde{\alpha}}_{{\sigma_{i}^2}}}^\gamma(q_{i},q_{i-1} ,\infty;V)  } 
	   \prod_{i=1}^{n+1}\prod_{m=\sigma^2_{i-1}+1}^{\sigma^2_{i}-1} \overline{ \Psi_{{\tilde{\alpha}}_m}^\gamma(q_{i-1},q_{i-1},\infty;V) }
	\\
	&\times  
	       \prod_{i=1}^{n+1}\prod_{m=\sigma^1_{i-1}+1}^{\sigma^1_{i}-1} \Psi_{\alpha_m}^\gamma(q_{i-1} ,q_{i-1},\infty;V ) 
	 \prod_{i=1}^{n+1} \frac{(s_{i-1}-s_i)^{\sigma_i^1-\sigma_{i-1}^1+\sigma_i^2-\sigma_{i-1}^2-2}}{(\sigma_i^1-\sigma_{i-1}^1-1)!(\sigma_i^2-\sigma_{i-1}^2-1)!}\,  d\boldsymbol{q}_{1,n}d\boldsymbol{s}_{n,1}.
	\end{aligned}
\end{equation}
\end{lemma}
\begin{proof}
As in Setting~\ref{setting_conv_twosided} we have for all $k$, $r$, $\alpha\in\N^k$ and ${\tilde{\alpha}}\in\N^r$ that
\begin{equation*}
	\begin{aligned}
	 \langle \OpW(a)\mathcal{I}_\infty^\gamma(k,\alpha;\hbar)\varphi_\hbar,\mathcal{I}_\infty^\gamma(r,{\tilde{\alpha}};\hbar)\varphi_\hbar \rangle 
	= \sum_{n=0}^{\min(k,r)} \sum_{\sigma^1 \in \mathcal{A}(k,n)} \sum_{\sigma^2\in\mathcal{A}(r,n)} \sum_{\kappa\in\mathcal{S}_n} \mathcal{J}(a,\varphi_\hbar,n,\sigma^1,\sigma^2,\kappa,\hbar), 
		\end{aligned}
\end{equation*}
where $\mathcal{J}(a,\varphi_\hbar,n,\sigma^1,\sigma^2,\tau,\hbar)$ is defined in Setting~\ref{setting_conv_twosided}. This gives us that
\begin{equation}\label{EQ:con_term_twosided_1}
	\begin{aligned}
	 \MoveEqLeft \sum_{k,r=1}^{k_0} \sum_{\alpha\in\N^k} \sum_{{\tilde{\alpha}}\in\N^r} (i\lambda)^{|\alpha|}(-i\lambda)^{|{\tilde{\alpha}}|} \Aver{  \langle \OpW(a)\mathcal{I}_\infty^\gamma(k,\alpha;\hbar)\varphi_\hbar,\mathcal{I}_\infty^\gamma(r,{\tilde{\alpha}};\hbar)\varphi_\hbar \rangle }
	 \\
	 ={}& \sum_{k,r=1}^{k_0} \sum_{\alpha\in\N^k} \sum_{{\tilde{\alpha}}\in\N^r} \sum_{n=0}^{\min(k,r)} \sum_{\sigma^1 \in \mathcal{A}(k,n)} \sum_{\sigma^2\in\mathcal{A}(r,n)} (i\lambda)^{|\alpha|}(-i\lambda)^{|{\tilde{\alpha}}|} \Aver{ \mathcal{J}(a,\varphi_\hbar,n,\sigma^1,\sigma^2,\mathrm{id},\hbar)}
	 \\
	 &+ \sum_{k,r=1}^{k_0} \sum_{\alpha\in\N^k} \sum_{{\tilde{\alpha}}\in\N^r}  \sum_{n=0}^{\min(k,r)} \sum_{\sigma^1 \in \mathcal{A}(k,n)} \sum_{\sigma^2\in\mathcal{A}(r,n)} \sum_{\kappa\in\mathcal{S}_n\setminus\{\mathrm{id}\}} (i\lambda)^{|\alpha|}(-i\lambda)^{|{\tilde{\alpha}}|} \Aver{ \mathcal{J}(a,\varphi_\hbar,n,\sigma^1,\sigma^2,\kappa,\hbar)}.
	\end{aligned}
\end{equation}
Using Lemma~\ref{lemma_gaus_reg_2}
we have that 
\begin{equation*}
	\begin{aligned}
	 \MoveEqLeft \Big|  \sum_{k,r=1}^{k_0} \sum_{\alpha\in\N^k} \sum_{{\tilde{\alpha}}\in\N^r}  \sum_{n=0}^{\min(k,r)} \sum_{\sigma^1 \in \mathcal{A}(k,n)} \sum_{\sigma^2\in\mathcal{A}(r,n)} \sum_{\kappa\in\mathcal{S}_n\setminus\{\mathrm{id}\}} (i\lambda)^{|\alpha|}(-i\lambda)^{|{\tilde{\alpha}}|} \Aver{ \mathcal{J}(a,\varphi_\hbar,n,\sigma^1,\sigma^2,\kappa,\hbar)}\Big|
	 \\
	 \leq{}&  \sum_{k,r=1}^{k_0} \sum_{\alpha\in\N^k} \sum_{{\tilde{\alpha}}\in\N^r}  \sum_{n=0}^{\min(k,r)} 
	C_a 2^{k+r} \frac{\hbar\rho (\rho t)^{k+r-n-1}n!}{(k-1)!(r-n)!} (\lambda C  \norm{\hat{V}}_{1,\infty,5d+5})^{|\alpha|+|{\tilde{\alpha}}|} |\log(\tfrac{\hbar}{t})|^{n+3} \norm{\varphi_\hbar}^{2}_{\mathcal{H}_\hbar^{2d+2}(\R^d)},
	\end{aligned}
\end{equation*}
where $C_a$ depends on the symbol $a$ and $C$ only depend on the dimension. We have moreover used that $\mathcal{S}_n$ contains $n!$ elements and that the number of elements in $\mathcal{A}(k,n)$ is bounded by $2^k$ and the number of elements in $\mathcal{A}(r,n)$ is bounded by $2^r$. From arguing as  in the previous proofs we obtain the bound
\begin{equation}\label{EQ:con_term_twosided_2}
	\begin{aligned}
	 \MoveEqLeft \Big|  \sum_{k,r=1}^{k_0} \sum_{\alpha\in\N^k} \sum_{{\tilde{\alpha}}\in\N^r}  \sum_{n=0}^{\min(k,r)} \sum_{\sigma^1 \in \mathcal{A}(k,n)} \sum_{\sigma^2\in\mathcal{A}(r,n)} \sum_{\kappa\in\mathcal{S}_n\setminus\{\mathrm{id}\}} (i\lambda)^{|\alpha|}(-i\lambda)^{|{\tilde{\alpha}}|} \Aver{ \mathcal{J}(a,\varphi_\hbar,n,\sigma^1,\sigma^2,\kappa,\hbar)}\Big|
	 \\
	 \leq{}& C  C^{k_0} k_0^3 \hbar| \log(\tfrac{\hbar}{t})|^{k_0+3}  \sup_{\hbar\in I} \norm{\varphi_\hbar}^{2}_{\mathcal{H}_\hbar^{2d+2}(\R^d)}.
	\end{aligned}
\end{equation}
Our assumptions on $k_0$ and the estimate in \eqref{EQ:con_term_twosided_2} implies that
\begin{equation}\label{EQ:con_term_twosided_3}
	 \lim_{\hbar\rightarrow 0} \Big|  \sum_{k,r=1}^{k_0} \sum_{\alpha\in\N^k} \sum_{{\tilde{\alpha}}\in\N^r}  \sum_{n=0}^{\min(k,r)} \sum_{\sigma^1 \in \mathcal{A}(k,n)} \sum_{\sigma^2\in\mathcal{A}(r,n)} \sum_{\kappa\in\mathcal{S}_n\setminus\{\mathrm{id}\}} (i\lambda)^{|\alpha|}(-i\lambda)^{|{\tilde{\alpha}}|} \Aver{ \mathcal{J}(a,\varphi_\hbar,n,\sigma^1,\sigma^2,\kappa,\hbar)}\Big|
	=0.
\end{equation}
From arguing as in the proof of Lemma~\ref{lemma_gaus_reg} we can obtain the bound
\begin{equation}
	|\Aver{ \mathcal{J}(a,\varphi_\hbar,n,\sigma^1,\sigma^2,\mathrm{id},\hbar)}| \leq  C \frac{ t^{r+k-n}}{k!(r-n)!}  (C\norm{\hat{V}}_{1,\infty,2d+2})^{|\alpha|+|{\tilde{\alpha}}|}  \norm{\varphi_\hbar}^2_{L^2(\R^d)}.
\end{equation}
Applying this estimate gives us that
\begin{equation}\label{EQ:con_term_twosided_4}
	\begin{aligned}
	 \MoveEqLeft \sup_{\hbar \in I} \sum_{k,r=1}^{k_0} \sum_{\alpha\in\N^k} \sum_{{\tilde{\alpha}}\in\N^r} \sum_{n=0}^{\min(k,r)} \sum_{\sigma^1 \in \mathcal{A}(k,n)} \sum_{\sigma^2\in\mathcal{A}(r,n)} | (i\lambda)^{|\alpha|}(-i\lambda)^{|{\tilde{\alpha}}|} \Aver{ \mathcal{J}(a,\varphi_\hbar,n,\sigma^1,\sigma^2,\mathrm{id},\hbar)}|
	 \\
	 &\leq \sup_{\hbar \in I}  \norm{\varphi_\hbar}^2_{L^2(\R^d)}   \sum_{n=0}^{k_0} \sum_{k,r=n}^{k_0} C^{k+r} \frac{ t^{r+k+n}}{k!(r-n)!}    \leq \sup_{\hbar \in I}  \norm{\varphi_\hbar}^2_{L^2(\R^d)}   \sum_{n=0}^{\infty} \sum_{k,r=0}^{\infty} \frac{ C^{r+k+n}}{k!n!r!}     <\infty.
	\end{aligned}
\end{equation}
From this estimate we have that the sums are absolutely convergent. Hence using dominated convergence we get from \eqref{EQ:con_term_twosided_1}, \eqref{EQ:con_term_twosided_3} and \eqref{EQ:con_term_twosided_4} that
\begin{equation}\label{EQ:con_term_twosided_5}
	\begin{aligned}
	 \MoveEqLeft \lim_{\hbar\rightarrow 0} \sum_{k,r=1}^{k_0} \sum_{\alpha\in\N^k} \sum_{{\tilde{\alpha}}\in\N^r} (i\lambda)^{|\alpha|}(-i\lambda)^{|{\tilde{\alpha}}|} \Aver{  \langle \OpW(a)\mathcal{I}_\infty^\gamma(k,\alpha;\hbar)\varphi_\hbar,\mathcal{I}_\infty^\gamma(r,{\tilde{\alpha}};\hbar)\varphi_\hbar \rangle }
	 \\
	 ={}& \sum_{k,r=1}^{\infty} \sum_{\alpha\in\N^k} \sum_{{\tilde{\alpha}}\in\N^r} \sum_{n=0}^{\min(k,r)} \sum_{\sigma^1 \in \mathcal{A}(k,n)} \sum_{\sigma^2\in\mathcal{A}(r,n)} (i\lambda)^{|\alpha|}(-i\lambda)^{|{\tilde{\alpha}}|} \lim_{\hbar\rightarrow 0}  \Aver{ \mathcal{J}(a,\varphi_\hbar,n,\sigma^1,\sigma^2,\mathrm{id},\hbar)}.
	\end{aligned}
\end{equation}
 Let $\delta>0$ then from Lemma~\ref{lemma_gaus_reg} we have that there exsist $\varepsilon_0$ such that for all $\varepsilon<\varepsilon_0$ we have that
 \begin{equation}\label{EQ:con_term_twosided_6}
	\begin{aligned}
	\sup_{\hbar\in I}|\Aver{ \mathcal{J}(a,\varphi_\hbar,n,\sigma^1,\sigma^2,\mathrm{id},\hbar)-  \mathcal{J}_\varepsilon(a,\varphi_\hbar,n,\sigma^1,\sigma^2,\mathrm{id},\hbar)}| \leq  C\delta. 
	\end{aligned}
\end{equation}
Furthermore, from Lemma~\ref{lemma_gaus_reg_3} we get that 
\begin{equation}\label{EQ:con_term_twosided_7}
	\begin{aligned}
	\Aver{ \mathcal{J}_\varepsilon(a,\varphi_\hbar,n,\sigma^1,\sigma^2,\kappa,\hbar)} = \langle \mathrm{Op}_{0,\hbar}(a\circ\Phi_{\gamma,\varepsilon}^t) \varphi_\hbar, \varphi_\hbar \rangle  + C_\varepsilon\sqrt{\hbar},
	\end{aligned}
\end{equation}
where the constant $C_\varepsilon$ is independent of $\hbar$ but will increase polynomial as $\varepsilon$ goes to zero. Moreover from Lemma~\ref{lemma_gaus_reg_3} we also have that
\begin{equation}\label{EQ:con_term_twosided_8}
	\begin{aligned}
	\lim_{\hbar\rightarrow0}|\langle \mathrm{Op}_{0,\hbar}(a\circ\Phi_{\gamma,\varepsilon}^t) \varphi_\hbar, \varphi_\hbar \rangle - \int a\circ\Phi^t_{\gamma,\varepsilon}(y,q_0) \,d\mu_0(y,q_0)| =0.
	\end{aligned}
\end{equation}  
We now compare the integrals of the symbols $a\circ\Phi^t_{\gamma,\varepsilon}$ and $a\circ\Phi^t_{\gamma,0}$, defined in \eqref{composit_sym_epsilon} and \eqref{composit_sym} respectively, against the measure $\mu_0$. Here we see that 
\begin{equation*}
	\begin{aligned}
	\MoveEqLeft \int\big( a\circ\Phi^t_{\gamma,\varepsilon}(y,q_0) - a\circ\Phi^t_{\gamma,0}(y,q_0) \big) \,d\mu_0(y,q_0)
	=  ((2\pi)^d\rho)^{(k+r-n)} (2\pi)^n   
	 \\
	 &\times  \int \int_{[0,t]_{\leq}^n}  \int   a(y- tq_0 -\sum_{i=1}^{n}s_{i}  (q_{i}-q_{i-1}) ,q_{n})  
	 \prod_{i=1}^{n} \Big\{   \Psi_{\alpha_{\sigma_i^1}}^\gamma(q_{i} ,q_{i-1},\infty;V )  \overline{ \Psi_{{\tilde{\alpha}}_{{\sigma_{i}^2}}}^\gamma(q_{i},q_{i-1} ,\infty;V)  } \Big\} 
	\\
	&\times \Big\{ \frac{1}{(4\pi\varepsilon^2)^{\frac{n}{2}}}  \prod_{i=1}^n   e^{ \frac{1}{4}\varepsilon^{-2} \frac{1}{2}(q_i^2-q_{i-1}^2)}- \prod_{i=1}^{n} \delta( \tfrac{1}{2} (q_{i}^2-q_{i-1}^2))  \Big\} 
	   \prod_{i=1}^{n+1}\Big\{\prod_{m=\sigma^2_{i-1}+1}^{\sigma^2_{i}-1} \overline{ \Psi_{{\tilde{\alpha}}_m}^\gamma(q_{i-1},q_{i-1},\infty;V) } 
	\\
	&\times 
	\prod_{m=\sigma^1_{i-1}+1}^{\sigma^1_{i}-1} \Psi_{\alpha_m}^\gamma(q_{i-1} ,q_{i-1},\infty;V ) \Big\}
	 \prod_{i=1}^{n+1} \frac{(s_{i-1}-s_i)^{\sigma_i^1-\sigma_{i-1}^1+\sigma_i^2-\sigma_{i-1}^2-2}}{(\sigma_i^1-\sigma_{i-1}^1-1)!(\sigma_i^2-\sigma_{i-1}^2-1)!}\,  d\boldsymbol{q}d\boldsymbol{s} d\mu_0(y,q_0).
	\end{aligned}
\end{equation*}
We know that in a distributional sense we have that
\begin{equation*}
	 \frac{1}{(4\pi\varepsilon^2)^{\frac{n}{2}}}  \prod_{i=1}^n   e^{ \frac{1}{4}\varepsilon^{-2} \frac{1}{2}(q_i^2-q_{i-1}^2)}- \prod_{i=1}^{n} \delta( \tfrac{1}{2} (q_{i}^2-q_{i-1}^2)) \rightarrow 0 \quad\text{as $\varepsilon \rightarrow 0$}.
\end{equation*}
In particular this implies that
\begin{equation}\label{EQ:con_term_twosided_9}
	\begin{aligned}
	\lim_{\varepsilon\rightarrow0}\big| \int\big( a\circ\Phi^t_{\gamma,\varepsilon}(y,q_0) - a\circ\Phi^t_{\gamma,0}(y,q_0) \big) \,d\mu_0(y,q_0)\big|=0.
	\end{aligned}
\end{equation}
To apply all these estimates we note that
\begin{equation}\label{EQ:con_term_twosided_10}
	\begin{aligned}
		\MoveEqLeft \big| \Aver{ \mathcal{J}(a,\varphi_\hbar,n,\sigma^1,\sigma^2,\mathrm{id},\hbar)} - \int a\circ\Phi^t_{\gamma,0}(y,q_0) \,d\mu_0(y,q_0)\big|
		\\
		\leq {}&   |\Aver{ \mathcal{J}(a,\varphi_\hbar,n,\sigma^1,\sigma^2,\mathrm{id},\hbar)-  \mathcal{J}_\varepsilon(a,\varphi_\hbar,n,\sigma^1,\sigma^2,\mathrm{id},\hbar)}| 
		\\
		&+ | \Aver{ \mathcal{J}_\varepsilon(a,\varphi_\hbar,n,\sigma^1,\sigma^2,\kappa,\hbar)} - \langle \mathrm{Op}_{0,\hbar}(a\circ\Phi_{\gamma,\varepsilon}^t) \varphi_\hbar, \varphi_\hbar \rangle |
		\\
		&+|\langle \mathrm{Op}_{0,\hbar}(a\circ\Phi_{\gamma,\varepsilon}^t) \varphi_\hbar, \varphi_\hbar \rangle - \int a\circ\Phi^t_{\gamma,\varepsilon}(y,q_0) \,d\mu_0(y,q_0)| 
		\\
		&+ \big| \int \big( a\circ\Phi^t_{\gamma,\varepsilon}(y,q_0) - a\circ\Phi^t_{\gamma,0}(y,q_0) \big) \,d\mu_0(y,q_0)\big|.
	\end{aligned}
\end{equation}
By first letting $\hbar\rightarrow0$ and the $\varepsilon\rightarrow0$ it follows from \eqref{EQ:con_term_twosided_6}, \eqref{EQ:con_term_twosided_7}, \eqref{EQ:con_term_twosided_8} \eqref{EQ:con_term_twosided_9} and \eqref{EQ:con_term_twosided_10} that 
\begin{equation}\label{EQ:con_term_twosided_11}
	\begin{aligned}
		\lim_{\hbar\rightarrow0} \big| \Aver{ \mathcal{J}(a,\varphi_\hbar,n,\sigma^1,\sigma^2,\mathrm{id},\hbar)} - \int a\circ\Phi^t_{\gamma,0}(y,q_0) \,d\mu_0(y,q_0)\big| =0.
	\end{aligned}
\end{equation}
Lastly by \eqref{EQ:con_term_twosided_5} and \eqref{EQ:con_term_twosided_11} the desired result follows and this concludes the proof.
\end{proof}
\begin{lemma}\label{LE:main_term_reg_main_thm}
Assume we are in setting of Definition~\ref{functions_for_exp_def}. Let  $\{\varphi_\hbar\}_{\hbar\in I}$ be a uniform semiclassical family in $\mathcal{H}^{5d+5}_\hbar(\R^d)$ with Wigner measure $\mu_0$. 
Let
\begin{equation*}
	k_0 = \frac{|\log(\hbar)|}{10 \log(|\log(\hbar)|)}.
\end{equation*} 
 Then, for $\gamma>0$, 
\begin{equation*}
	 \lim_{\hbar\rightarrow 0} \sum_{k,r=0}^{k_0} \mathbb{E}\Big[ \sum_{\boldsymbol{x}\in\mathcal{X}_{\neq}^{k}}\sum_{\boldsymbol{\tilde{x}}\in\mathcal{X}_{\neq}^{r}}  \langle \OpW(a)  \mathcal{I}_{0,0}^\gamma(k,\boldsymbol{x},t;\hbar)\varphi_\hbar ,  \mathcal{I}_{0,0}^\gamma(r,\boldsymbol{\tilde{x}},t;\hbar)\varphi_\hbar \rangle\Big]
 = \int f^\gamma(t,y,p) \, d\mu_0(y,q_0),
\end{equation*}
where we have used the convention $\mathcal{I}_{0,0}(0,\boldsymbol{\tilde{x}},t;\hbar)=U_{\hbar,0}(-t) $ and the notation that
\begin{equation*}
	\begin{aligned}
	f^\gamma(t,y,q_0) ={}&   \sum_{n=0}^\infty  \frac{((2\pi)^{d+1}\rho)^{n}}{(2\pi)^{d}}    \int_{[0,t]_{\leq}^n}  \int   
	    a(y- tq_0 -\sum_{i=1}^{n}s_{i}  (q_{i}-q_{i-1}) ,q_{n})  |\hat{\varphi}_1( q_0)|^2
	    \\
	    &\times   \prod_{i=1}^{n} \delta( \tfrac{1}{2} (q_{i}^2-q_{i-1}^2))| \hat{T}^\gamma(q_i,q_{i-1})|^2
	  \prod_{i=1}^{n+1}  e^{-(2\pi)^{d}\rho(s_{i-1}-s_i) 2\im( \hat{T}^\gamma(q,q))}  \,  d\boldsymbol{q}d\boldsymbol{s}. 
	\end{aligned}
\end{equation*}
\end{lemma}
\begin{proof}
From Lemma~\ref{regularise_expansion_in_time} we have that
\begin{equation*}
\lim_{\hbar\rightarrow 0}\sum_{k=1}^{k_0}     \mathbb{E} \bigg[\Big\lVert \sum_{\boldsymbol{x}\in\mathcal{X}^{k}_{\neq}}  \mathcal{I}_{0,0}^\gamma(k,\boldsymbol{x},t;\hbar) \varphi-\mathcal{I}_{\infty}^\gamma(k,\boldsymbol{x},t;\hbar)\varphi \Big\rVert_{L^2(\R^d)}^2	\bigg]^\frac{1}{2}  = 0.
\end{equation*}
This gives us that
\begin{equation}\label{EQ:main_term_reg_main_thm_1}
	\begin{aligned}
	 \lim_{\hbar\rightarrow 0} \sum_{k,r=0}^{k_0} \mathbb{E}\Big[ \sum_{\boldsymbol{x}\in\mathcal{X}_{\neq}^{k}}\sum_{\boldsymbol{\tilde{x}}\in\mathcal{X}_{\neq}^{r}}  \langle& \OpW(a)  \mathcal{I}_{0,0}^\gamma(k,\boldsymbol{x},t;\hbar)\varphi_\hbar ,  \mathcal{I}_{0,0}^\gamma(r,\boldsymbol{\tilde{x}},t;\hbar)\varphi_\hbar \rangle 
	 \\
	 &-  \langle \OpW(a)  \mathcal{I}_{\infty}^\gamma(k,\boldsymbol{x},t;\hbar)\varphi_\hbar ,  \mathcal{I}_{\infty}^\gamma(r,\boldsymbol{\tilde{x}},t;\hbar)\varphi_\hbar \rangle\Big]
 =0.
 	\end{aligned}
\end{equation}
Using the definition of the operators and the notation introduced in Notation~\ref{notation_op_main_lim} we have that
\begin{equation*}
	\begin{aligned}
	\MoveEqLeft  \sum_{k,r=0}^{k_0} \mathbb{E}\Big[ \sum_{\boldsymbol{x}\in\mathcal{X}_{\neq}^{k}}\sum_{\boldsymbol{\tilde{x}}\in\mathcal{X}_{\neq}^{r}}  \langle \OpW(a)  \mathcal{I}_{\infty}^\gamma(k,\boldsymbol{x},t;\hbar)\varphi_\hbar ,  \mathcal{I}_{\infty}^\gamma(r,\boldsymbol{\tilde{x}},t;\hbar)\varphi_\hbar \rangle\Big]
	 \\
	={}&\sum_{k,r=1}^{k_0} \sum_{\alpha\in\N^k}  \sum_{{\tilde{\alpha}}\in\N^r}     (-i\lambda)^{|\alpha|}    (i\lambda)^{|{\tilde{\alpha}}|} \Aver{ \langle \OpW(a) \mathcal{I}_{\infty}^\gamma(k,\alpha;\hbar)\varphi_\hbar , \mathcal{I}_\infty^\gamma(r,{\tilde{\alpha}};\hbar)\varphi_\hbar  \rangle} .
	\end{aligned}
\end{equation*}
By arguing as in the proof of Lemma~\ref{LE:con_term_twosided} we see that the sum is absolute convergent for all $\gamma\geq0$ uniformly in $\hbar$. Hence we obtain that
\begin{equation}\label{EQ:main_term_reg_main_thm_2}
	\begin{aligned}
	\MoveEqLeft  \lim_{\hbar\rightarrow0} \sum_{k,r=0}^{k_0} \mathbb{E}\Big[ \sum_{\boldsymbol{x}\in\mathcal{X}_{\neq}^{k}}\sum_{\boldsymbol{\tilde{x}}\in\mathcal{X}_{\neq}^{r}}  \langle \OpW(a)  \mathcal{I}_{\infty}^\gamma(k,\boldsymbol{x},t;\hbar)\varphi_\hbar ,  \mathcal{I}_{\infty}^\gamma(r,\boldsymbol{\tilde{x}},t;\hbar)\varphi_\hbar \rangle\Big]
	 \\
	={}&\sum_{k,r=0}^{\infty} \sum_{\alpha\in\N^k}  \sum_{{\tilde{\alpha}}\in\N^r}     (-i\lambda)^{|\alpha|}    (i\lambda)^{|{\tilde{\alpha}}|} \lim_{\hbar\rightarrow0} \Aver{ \langle \OpW(a) \mathcal{I}_{\infty}^\gamma(k,\alpha;\hbar)\varphi_\hbar , \mathcal{I}_\infty^\gamma(r,{\tilde{\alpha}};\hbar)\varphi_\hbar  \rangle} .
	\end{aligned}
\end{equation}
We now just need to find the term-wise limits. For the case $k=r=0$ we have, by Egorov's theorem, that
\begin{equation*}
	\begin{aligned}
	\langle \OpW(a) U_{\hbar,0}(-t)\varphi_\hbar , U_{\hbar,0}(-t)\varphi_\hbar  \rangle 
	= \frac{1}{(2\pi \hbar)^d} \int e^{i\hbar^{-1}\langle x-y, q \rangle} a(\tfrac{x+y}{2}-tq,q) \varphi_\hbar(y) \overline{\varphi_\hbar(x)} \, dy  dq dx,
	\end{aligned}
\end{equation*}
and hence
\begin{equation}\label{p_con_main_reg_1}
	\begin{aligned}
	\MoveEqLeft\lim_{\hbar\rightarrow 0}  \langle \OpW(a) U_{\hbar,0}(-t)\varphi_\hbar , U_{\hbar,0}(-t)\varphi_\hbar  \rangle =
	 \int  a(y-tq,q)  \,d\mu_0(y,q).
	\end{aligned}
\end{equation}
Lemma~\ref{L.con_term_onesided_1} gives us that
\begin{equation}\label{p_con_main_reg_2}
	\begin{aligned}
	\MoveEqLeft \sum_{k=1}^\infty \sum_{\alpha\in\N^k}     (i\lambda)^{|\alpha|}  \lim_{\hbar\rightarrow0} \Aver{ \langle \OpW(a) \mathcal{I}_\infty^\gamma(k,\alpha;\hbar)\varphi_\hbar , U_{\hbar,0}(-t)\varphi_\hbar  \rangle }
	\\
	&=  \sum_{k=1}^\infty  \sum_{\alpha\in\N^k}  (i\lambda)^{|\alpha|}    \frac{((2\pi)^d\rho t)^{k}}{k!} \int  a( y- tq,q) \prod_{m=1}^k  
	 \Psi_{\alpha_m}^\gamma(q,q,\infty;V)   \, d\mu_0(y,q)
	 \\
	 &=  \sum_{k=1}^\infty     \frac{(i (2\pi)^d\rho)^{k} t^k}{k!} \int a( y- tq,q) (\hat{T}^\gamma(q,q))^k    \, d\mu_0(y,q),
	\end{aligned}
\end{equation}
where we have used the identity \eqref{born_series}, which is
\begin{equation}\label{born_series_2}
	\hat{T}^\gamma(p,p_0) =-i  \sum_{n=1}^\infty (i\lambda)^{n} \Psi_{n}^\gamma(p,p_0,\infty ;V),
\end{equation}
from Remark~\ref{born_series_convergence_remark}. Analogously we get that
\begin{equation}\label{p_con_main_reg_3}
	\begin{aligned}
	\MoveEqLeft   \sum_{k=1}^\infty \sum_{\alpha\in\N^k} \lim_{\hbar\rightarrow0}     (-i\lambda)^{|\alpha|} \Aver{ \langle \OpW(a) U_{\hbar,0}(-t)\varphi_\hbar, \mathcal{I}_\infty^\gamma(k,\alpha;\hbar)\varphi_\hbar   \rangle}
	\\
	 &=  \sum_{k=1}^\infty     \frac{(-i (2\pi)^d\rho)^{k} t^k}{k!} \int  a( y- tq,q) (\overline{\hat{T}^\gamma(q,q)})^k      \, d\mu_0(y,q).
	\end{aligned}
\end{equation}
For the final term we get from Lemma~\ref{LE:con_term_twosided} that
\begin{equation}\label{split_sums_final}
	\begin{aligned}
	\MoveEqLeft \sum_{k=1}^\infty\sum_{\alpha\in\N^k}  \sum_{r=1}^\infty \sum_{{\tilde{\alpha}}\in\N^r}     (-i\lambda)^{|\alpha|}    (i\lambda)^{|{\tilde{\alpha}}|}\lim_{\hbar\rightarrow0}\Aver{ \langle \OpW(a) \mathcal{I}_\infty^\gamma(k,\alpha;\hbar)\varphi_\hbar , \mathcal{I}_\infty^\gamma(r,{\tilde{\alpha}};\hbar)\varphi_\hbar  \rangle} 
	\\
	={}& \sum_{k,r=1}^\infty   \sum_{\alpha\in\N^k}  \sum_{{\tilde{\alpha}}\in\N^r}   \sum_{n=0}^{\min(k,r)} \sum_{\sigma^1 \in\mathcal{A}(k,n)}  \sum_{\sigma^2\in\mathcal{A}(r,n)}    (-i\lambda)^{|\alpha|}    (i\lambda)^{|{\tilde{\alpha}}|}  \int a\circ\Phi^t_{\gamma,0}(y,q_0) \,d\mu_0(y,q_0),
		\end{aligned}
\end{equation}
where the symbol $a\circ\Phi^t_{\gamma,0}(y,q_0)$ is given by
\begin{equation*}
	\begin{aligned}
	 a\circ\Phi^t_{\gamma,0}(y,q_0)
	 ={}&  ((2\pi)^d\rho)^{(k+r-n)} (2\pi)^n  \int_{[0,t]_{\leq}^n}  \int   
	a(y- tq_0 -\sum_{i=1}^{n}s_{i}  (q_{i}-q_{i-1}) ,q_{n}) \prod_{i=1}^{n} \delta( \tfrac{1}{2} (q_{i}^2-q_{i-1}^2))  
	\\
	&\times  \prod_{i=1}^{n}    \Psi_{\alpha_{\sigma_i^1}}^\gamma(q_{i} ,q_{i-1},\infty;V )  \overline{ \Psi_{{\tilde{\alpha}}_{{\sigma_{i}^2}}}^\gamma(q_{i},q_{i-1} ,\infty;V)  } 
	   \prod_{i=1}^{n+1}\prod_{m=\sigma^2_{i-1}+1}^{\sigma^2_{i}-1} \overline{ \Psi_{{\tilde{\alpha}}_m}^\gamma(q_{i-1},q_{i-1},\infty;V) }
	\\
	&\times  
	       \prod_{i=1}^{n+1}\prod_{m=\sigma^1_{i-1}+1}^{\sigma^1_{i}-1} \Psi_{\alpha_m}^\gamma(q_{i-1} ,q_{i-1},\infty;V ) 
	 \prod_{i=1}^{n+1} \frac{(s_{i-1}-s_i)^{\sigma_i^1-\sigma_{i-1}^1+\sigma_i^2-\sigma_{i-1}^2-2}}{(\sigma_i^1-\sigma_{i-1}^1-1)!(\sigma_i^2-\sigma_{i-1}^2-1)!}\,  d\boldsymbol{q}d\boldsymbol{s}.
	 \end{aligned}
\end{equation*}
We split the sums in \eqref{split_sums_final} into terms with $n=0$ and the rest.  More precisely we write 
\begin{equation*}
	\begin{aligned}
 	 \sum_{k,r=1}^\infty   \sum_{\alpha\in\N^k}  \sum_{{\tilde{\alpha}}\in\N^r}   \sum_{n=0}^{\min(k,r)} \sum_{\sigma^1 \in\mathcal{A}(k,n)}  \sum_{\sigma^2\in\mathcal{A}(r,n)}
	= \sum_{k,r=1}^\infty    \sum_{\alpha\in\N^k} \sum_{{\tilde{\alpha}}\in\N^r}   +  \sum_{n=1}^\infty \sum_{k,r=n}^\infty     \sum_{\alpha\in\N^k} \sum_{{\tilde{\alpha}}\in\N^r}     \sum_{\sigma^1 \in\mathcal{A}(k,n)}  \sum_{\sigma^2\in\mathcal{A}(r,n)}.
	\end{aligned}
\end{equation*}
Using this splitting and reordering of the sums and the identity for the Born series \eqref{born_series_2} we obtain from \eqref{p_con_main_reg_1}, \eqref{p_con_main_reg_2}, \eqref{p_con_main_reg_3} and \eqref{split_sums_final}  that 
\begin{equation}\label{p_con_main_reg_4}
	\begin{aligned}
	\MoveEqLeft \sum_{k,r=0}^{\infty} \sum_{\alpha\in\N^k}  \sum_{{\tilde{\alpha}}\in\N^r}     (-i\lambda)^{|\alpha|}    (i\lambda)^{|{\tilde{\alpha}}|} \lim_{\hbar\rightarrow0} \Aver{ \langle \OpW(a) \mathcal{I}_{\infty}^\gamma(k,\alpha;\hbar)\varphi_\hbar , \mathcal{I}_\infty^\gamma(r,{\tilde{\alpha}};\hbar)\varphi_\hbar  \rangle}
	\\
	={}& \sum_{k,r=0}^\infty       i^k    (-i)^r  \frac{((2\pi)^d\rho t)^{(k+r)}}{k!r!}  \int   a(y- tq  ,q) (\hat{T}^\gamma(q,q))^k  (\overline{\hat{T}^\gamma(q,q)})^r    \,  d\mu_0(y,q)
	 \\
	 &+ \sum_{n=1}^\infty \sum_{k,r=n}^\infty    \sum_{\sigma^1 \in\mathcal{A}(k,n)}  \sum_{\sigma^2\in\mathcal{A}(r,n)} ((2\pi)^d\rho)^{(k+r-n)} (2\pi)^{n} 
	\int  \int_{[0,t]_{\leq}^n}  \int    \prod_{i=1}^{n} \delta( \tfrac{1}{2} (q_{i}^2-q_{i-1}^2))
	\\
	 &\times 
	 \prod_{i=1}^{n+1} \frac{(i(s_{i-1}-s_i)\hat{T}^\gamma(q,q))^{\sigma_i^1-\sigma_{i-1}^1-1}  (-i(s_{i-1}-s_i)\overline{\hat{T}^\gamma(q,q)})^{\sigma_i^2-\sigma_{i-1}^2-1}}{(\sigma_i^1-\sigma_{i-1}^1-1)!(\sigma_i^2-\sigma_{i-1}^2-1)!} 
	 \\
	&\times  
	    a(y- tq_0 -\sum_{i=1}^{n}s_{i}  (q_{i}-q_{i-1}) ,q_{n})   \hat{\varphi}_1(q_0) \overline{\hat{\varphi}_2(q_0)}   \prod_{i=1}^{n} |\hat{T}^\gamma(q_i,q_{i-1})|^2  \,  d\boldsymbol{q}d\boldsymbol{s} d\mu_0(y,q_0).
	\end{aligned}
\end{equation}
If we just consider the sums over  $k$, $r$, $\sigma^1$ and $\sigma^2$ we see that
\begin{equation}\label{p_con_main_reg_5}
	\begin{aligned}
	\MoveEqLeft  \sum_{k=n}^\infty  \sum_{r=n}^\infty    \sum_{\sigma^1 \in\mathcal{A}(k,n)}  \sum_{\sigma^2\in\mathcal{A}(r,n)} ((2\pi)^d\rho)^{(k+r-n)} 
	\\
	&\times \prod_{i=1}^{n+1} \frac{(i(s_{i-1}-s_i)\hat{T}^\gamma(q,q))^{\sigma_i^1-\sigma_{i-1}^1-1}  (-i(s_{i-1}-s_i)\overline{\hat{T}^\gamma(q,q)})^{\sigma_i^2-\sigma_{i-1}^2-1}}{(\sigma_i^1-\sigma_{i-1}^1-1)!(\sigma_i^2-\sigma_{i-1}^2-1)!}  
	\\
	={}&((2\pi)^d\rho)^{n}  \sum_{k=n}^\infty     \sum_{\sigma^1 \in\mathcal{A}(k,n)}   \prod_{i=1}^{n+1} \frac{(i (2\pi)^d\rho (s_{i-1}-s_i)\hat{T}^\gamma(q,q))^{\sigma_i^1-\sigma_{i-1}^1-1} }{(\sigma_i^1-\sigma_{i-1}^1-1)!} 
	\\
	&\times  \sum_{r=n}^\infty    \sum_{\sigma^2\in\mathcal{A}(r,n)}    \prod_{i=1}^{n+1}   \frac{  (-i (2\pi)^d\rho (s_{i-1}-s_i)\overline{\hat{T}^\gamma(q,q)})^{\sigma_i^2-\sigma_{i-1}^2-1}}{(\sigma_i^2-\sigma_{i-1}^2-1)!}  
	\\
	={}&((2\pi)^d\rho)^{n}  \prod_{i=1}^{n+1}   \sum_{k_i=0}^\infty    \frac{(i (2\pi)^d\rho (s_{i-1}-s_i)\hat{T}^\gamma(q,q))^{k_i} }{k_i!}  \sum_{r_i=0}^\infty     \frac{  (-i (2\pi)^d\rho (s_{i-1}-s_i)\overline{\hat{T}^\gamma(q,q)})^{r_i}}{r_i!} 
	\\
	={}&((2\pi)^d\rho)^{n}  \prod_{i=1}^{n+1}   e^{(2\pi)^{d}\rho(s_{i-1}-s_i) (i \hat{T}^\gamma(q,q) - i\overline{\hat{T}^\gamma(q,q)})} = ((2\pi)^d\rho)^{n}  \prod_{i=1}^{n+1}  e^{-(2\pi)^{d}\rho(s_{i-1}-s_i) 2\im( \hat{T}^\gamma(q,q))}, 
	\end{aligned}
\end{equation}
where we in the last equality have used the ``optical theorem'' \eqref{EQ:Optical_theorem}. Combining \eqref{EQ:main_term_reg_main_thm_1}, \eqref{EQ:main_term_reg_main_thm_2},  \eqref{p_con_main_reg_4} and \eqref{p_con_main_reg_5}  we obtain the desired result and this concludes the proof.
\end{proof}
\section{Proof of Theorem~\ref{main_thm2}}\label{sec:proof_main}
We are considering the  inner product
\begin{equation}
	\langle U_{\hbar,\lambda}(t)\OpW(a) U_{\hbar,\lambda}(-t)  \varphi_\hbar, \varphi_\hbar\rangle=\langle \OpW(a) U_{\hbar,\lambda}(-t)  \varphi_\hbar, U_{\hbar,\lambda}(-t)\varphi_\hbar\rangle,
\end{equation}
where $\{\varphi_\hbar\}_{\hbar\in I}$ is a uniform semiclassical family in $\mathcal{H}_\hbar^{5d+5}(\R^d)$ with Wigner measure $\mu_0$. From  Proposition~\ref{duhamel_expansion_lemma} we have for any $\tau_0\in\N$ and $k_0\in\N$ that
\begin{equation}
	 U_{\hbar,\lambda}(-t) 
	= U_{\hbar,0}(-t) +  \sum_{i=0}^2 \sum_{j=\min(1,i)}^i   \sum_{k=k_{ij}}^{k_0} \sum_{\iota\in\mathcal{Q}_{k,i,j}}   \sum_{\boldsymbol{x}\in\mathcal{X}_{\neq}^{k-i}}      \mathcal{I}_{i,j}(k,\boldsymbol{x},t;\hbar) + \mathcal{R}(\tau_0,k_0:\hbar),
\end{equation}
where $k_{00}=1$, $k_{11}=3$, $k_{21}=5$, $k_{22}=4$. The operator $ \mathcal{R}(\tau_0,k_0:\hbar) $ is defined in the proposition and
the operators $\mathcal{I}_{i,j}(k,\boldsymbol{x},\alpha,t;\hbar) $ are defined in Definition~\ref{functions_for_exp_def} and \ref{functions_for_exp_rec_def}. By applying Proposition~\ref{duhamel_expansion_lemma} twice with
\begin{equation}\label{EQ:def_k_tau_main_proof}
	k_0 = \frac{|\log(\hbar)|}{10 \log(|\log(\hbar)|)} \qquad\text{and}\qquad \tau_0 = \hbar^{-\frac{1}{3}}.
\end{equation}
we get the expansion
\begin{equation}\label{main_evo_expansion}
	\begin{aligned}
	 \MoveEqLeft \langle \OpW(a) U_{\hbar,\lambda}(-t) \varphi_\hbar, U_{\hbar,\lambda}(-t)\varphi_\hbar\rangle
	\\
	&=  \mathcal{M}(k_0,\hbar) 
	 +\mathcal{E}_1(k_0;\hbar)+ \mathcal{E}_2(k_0,\tau_0;\hbar) + \mathcal{E}_3(k_0;\hbar) + \mathcal{E}_4(k_0,\tau_0;\hbar).
	\end{aligned}
\end{equation}
where the main term $\mathcal{M}(k_0,\hbar)$ is given by
\begin{equation}\label{main_term_expansion}
	\begin{aligned}
	\mathcal{M}(k_0,\hbar)
	={}&  \sum_{k,r=1}^{k_0}  \sum_{\boldsymbol{x}\in\mathcal{X}_{\neq}^{k}}\sum_{\boldsymbol{\tilde{x}}\in\mathcal{X}_{\neq}^{r}}  \langle \OpW(a)  \mathcal{I}_{0,0}(k,\boldsymbol{x},t;\hbar)\varphi_\hbar ,  \mathcal{I}_{0,0}(r,\boldsymbol{\tilde{x}},t;\hbar)\varphi_\hbar \rangle,
	\end{aligned}
\end{equation}
where we have used the convention $\mathcal{I}_{0,0}(0,\boldsymbol{\tilde{x}},t;\hbar)=U_{\hbar,0}(-t) $. The error terms are given by
\begin{equation}\label{error_main_thm_1}
	\begin{aligned}
	\mathcal{E}_1(k_0;\hbar)=  \sum_{i=1}^2 \sum_{j=\min(1,i)}^i   \sum_{r=k_{ij}}^{k_0} \sum_{\iota\in\mathcal{Q}_{k,i,j}}   \langle  \OpW(a) U_{\hbar,\lambda}(-t) \varphi_\hbar, \sum_{\boldsymbol{x}\in\mathcal{X}_{\neq}^{k-i}}   \mathcal{I}_{i,j}(r,\boldsymbol{x},t;\hbar) \varphi_\hbar\rangle,
	\end{aligned}
\end{equation}
\begin{equation}\label{error_main_thm_2}
	\begin{aligned}
	  \mathcal{E}_2(k_0,\tau_0;\hbar)= \langle \OpW(a) U_{\hbar,\lambda}(-t) \varphi_\hbar,  \mathcal{R}(\tau_0,k_0:\hbar)\varphi_\hbar\rangle,
	\end{aligned}
\end{equation}
\begin{equation}\label{error_main_thm_3}
	\begin{aligned}
	\mathcal{E}_3(k_0;\hbar)
	  ={}& \sum_{i=1}^2 \sum_{j=\min(1,i)}^i   \sum_{k=k_{ij}}^{k_0}    \sum_{r=0}^{k_0} \sum_{\iota\in\mathcal{Q}_{k,i,j}}   \sum_{\boldsymbol{x}\in\mathcal{X}_{\neq}^{k-i}}  \sum_{\boldsymbol{\tilde{x}}\in\mathcal{X}_{\neq}^{r}}   \langle\OpW(a)    \mathcal{I}_{i,j}(k,\boldsymbol{x},t;\hbar) \varphi_\hbar ,   \mathcal{I}_{0,0}(r,\boldsymbol{\tilde{x}},t;\hbar)\varphi_\hbar   \rangle, 
	\end{aligned}
\end{equation}
and
\begin{equation}\label{error_main_thm_4}
	\begin{aligned}
	\mathcal{E}_4(k_0,\tau_0;\hbar)
	  ={}&  \sum_{r=0}^{k_0}  \sum_{\boldsymbol{\tilde{x}}\in\mathcal{X}_{\neq}^{r}}   \langle\OpW(a)   \mathcal{R}(\tau_0,k_0:\hbar) \varphi_\hbar ,   \mathcal{I}_{0,0}(r,\boldsymbol{\tilde{x}},t;\hbar)\varphi_\hbar   \rangle. 
	\end{aligned}
\end{equation}
We will consider the asymptotics for each term separately. We start with  $\mathcal{E}_1(k_0;\hbar)$. 
By using the definition of  $\mathcal{E}_1(k_0;\hbar)$, linearity of the integral and Cauchy-Schwarz for the inner product in $L^2(\R^d)$ we get the  estimate
\begin{equation}\label{EQ:error_main_thm_1.1}
	\begin{aligned}
	 \Aver{| \mathcal{E}_1(k_0;\hbar)|} \leq {}& \norm{\OpW(a)}_{\mathrm{op}} \norm{\varphi_\hbar}_{L^2(\R^d)}  \sum_{i=1}^2 \sum_{j=\min(1,i)}^i   \sum_{r=k_{ij}}^{k_0} \sum_{\iota\in\mathcal{Q}_{k,i,j}}  \mathbb{E}\Big[\big\lVert  \sum_{\boldsymbol{x}\in\mathcal{X}_{\neq}^{k-i}}  \mathcal{I}_{i,j}(k,\boldsymbol{x},\iota,t;\hbar) \varphi_\hbar\big\rVert_{L^2(\R^d)} \Big]
	 \\
	 \leq {}& \norm{\OpW(a)}_{\mathrm{op}} \norm{\varphi_\hbar}_{L^2(\R^d)}  \sum_{i=1}^2 \sum_{j=\min(1,i)}^i   \sum_{r=k_{ij}}^{k_0} \sum_{\iota\in\mathcal{Q}_{k,i,j}}  \mathbb{E}\Big[\big\lVert  \sum_{\boldsymbol{x}\in\mathcal{X}_{\neq}^{k-i}}  \mathcal{I}_{i,j}(k,\boldsymbol{x},\iota,t;\hbar) \varphi_\hbar\big\rVert_{L^2(\R^d)}^2 \Big]^{\frac{1}{2}},
	\end{aligned}
\end{equation}
where we in the last inequality have used Jensen's inequality. From Lemma~\ref{LE:norm_conv_with_op_1_rec} we get that
\begin{equation}\label{EQ:error_main_thm_1.2}
	\lim_{\hbar\rightarrow0} \sum_{r=k_{ij}}^{k_0} \sum_{\iota\in\mathcal{Q}_{k,i,j}}  \mathbb{E}\Big[\big\lVert  \sum_{\boldsymbol{x}\in\mathcal{X}_{\neq}^{k-i}}  \mathcal{I}_{i,j}(k,\boldsymbol{x},\iota,t;\hbar) \varphi_\hbar\big\rVert_{L^2(\R^d)}^2 \Big]^{\frac{1}{2}} =0.
\end{equation}
for all $(i,j)\in\{(1,1),(2,1),(2,3)\}$. Since we have that both $\norm{\OpW(a)}_{\mathrm{op}}$ and  $\norm{\varphi_\hbar}_{L^2(\R^d)} $ are uniformly bounded for $\hbar\in I$ it follows by combining \eqref{EQ:error_main_thm_1.1} and \eqref{EQ:error_main_thm_1.2} that
\begin{equation}\label{EQ:error_main_thm_1.3}
	\lim_{\hbar\rightarrow0} \Aver{|\mathcal{E}_1(k_0;\hbar) |} = 0.
\end{equation}
For the second term we get from the definition of  $\mathcal{E}_2(k_0,\tau_0;\hbar)$, linearity of the integral, Cauchy-Schwarz for the inner product in $L^2(\R^d)$ and Jensen's inequality we get the  estimate
\begin{equation}\label{EQ:error_main_thm_1.4}
	 \Aver{| \mathcal{E}_2(k_0,\tau_0;\hbar)|} \leq \norm{\OpW(a)}_{\mathrm{op}} \norm{\varphi_\hbar}_{L^2(\R^d)}  \mathbb{E}\Big[\big\lVert \mathcal{R}(\tau_0,k_0:\hbar)\varphi_\hbar\big\rVert_{L^2(\R^d)}^2 \Big]^{\frac12}.
\end{equation}
Lemma~\ref{LE:Convergence_of_duhamel_expansion} gives us that 
\begin{equation}\label{EQ:error_main_thm_1.5}
	 \lim_{\hbar\rightarrow0}  \mathbb{E}\Big[\big\lVert \mathcal{R}(\tau_0,k_0:\hbar)\varphi_\hbar\big\rVert_{L^2(\R^d)}^2 \Big]^{\frac12} =0.
\end{equation}
Hence  from combining \eqref{EQ:error_main_thm_1.4} and \eqref{EQ:error_main_thm_1.5} it follows that
\begin{equation}\label{EQ:error_main_thm_1.6}
	 \lim_{\hbar\rightarrow0}\Aver{| \mathcal{E}_2(k_0,\tau_0;\hbar)|} =0,
\end{equation}
 since $\norm{\OpW(a)}_{\mathrm{op}}$ and  $\norm{\varphi_\hbar}_{L^2(\R^d)} $ are uniformly bounded for $\hbar\in I$. To obtain that
\begin{equation}\label{EQ:error_main_thm_1.7}
	\lim_{\hbar\rightarrow0} \Aver{|\mathcal{E}_3(k_0;\hbar) |} + \Aver{|\mathcal{E}_4(k_0,\tau_0;\hbar) |} = 0,
\end{equation}
we can argue as above. Except when we use our assumption to say that $\norm{\varphi_\hbar}_{L^2(\R^d)} $ is uniformly bounded for $\hbar\in I$ we now need to use Lemma~\ref{LE:unf_norm_bound_without_rec} to get the uniform bound
\begin{equation*}
	\begin{aligned}
	\sup_{\hbar\in I}   \sum_{r=0}^{k_0}     \mathbb{E}\Big[\big\lVert  \sum_{\boldsymbol{x}\in\mathcal{X}_{\neq}^{r}}  \mathcal{I}_{0,0}(r,\boldsymbol{x},\iota,t;\hbar) \varphi_\hbar\big\rVert_{L^2(\R^d)}^2 \Big]^{\frac{1}{2}} \leq C,
	\end{aligned}
\end{equation*}
where we use our assumptions on $k_0$. We now turn to the main term $\mathcal{M}(k_0,\hbar)$.  We will in the following let $\mathcal{M}^\gamma(k_0,\hbar)$ be defined as in \eqref{main_term_expansion} but with $\mathcal{I}_{0,0}(k,\boldsymbol{x},t;\hbar)$ and $\mathcal{I}_{0,0}(r,\boldsymbol{\tilde{x}},t;\hbar)$ replaced by $\mathcal{I}_{0,0}^\gamma(k,\boldsymbol{x},t;\hbar)$ and $\mathcal{I}_{0,0}^\gamma(r,\boldsymbol{\tilde{x}},t;\hbar)$ respectively. Then by Lemma~\ref{regularise_expansion} it follows that
\begin{equation}\label{EQ:error_main_thm_1.8}
	\lim_{\gamma\rightarrow0}\sup_{\hbar\in I}\Aver{|\mathcal{M}(\infty,\hbar) -\mathcal{M}^\gamma(\infty,\hbar)| }  =0.
\end{equation}
Moreover for $\gamma\geq0$ we let
\begin{equation*}
	\begin{aligned}
	f^\gamma(t,y,q_0) ={}&   \sum_{n=0}^\infty  \frac{((2\pi)^{d+1}\rho)^{n}}{(2\pi)^{d}}    \int_{[0,t]_{\leq}^n}  \int   
	    a(y- tq_0 -\sum_{i=1}^{n}s_{i}  (q_{i}-q_{i-1}) ,q_{n})  |\hat{\varphi}_1( q_0)|^2
	    \\
	    &\times   \prod_{i=1}^{n} \delta( \tfrac{1}{2} (q_{i}^2-q_{i-1}^2))| \hat{T}^\gamma(q_i,q_{i-1})|^2
	  \prod_{i=1}^{n+1}  e^{-(2\pi)^{d}\rho(s_{i-1}-s_i) 2\im( \hat{T}^\gamma(q,q))}  \,  d\boldsymbol{q}d\boldsymbol{s} 
	\end{aligned}
\end{equation*}
and we will use the convention $f(t,y,q_0)=f^0(t,y,q_0)$. Note that by definition we have that 
\begin{equation}\label{EQ:error_main_thm_1.9} 
	\lim_{\gamma\rightarrow0}f^\gamma(t,y,q_0)=f(t,y,q_0).
\end{equation} 
Then using \eqref{main_evo_expansion} we get that
\begin{equation}\label{main_est_proof_main}
	\begin{aligned}
	 \MoveEqLeft | \mathbb{E}	 [\langle A_\hbar(t) \varphi_{\hbar}, \varphi_{\hbar} \rangle -  \int f(t,x,p) \, d\mu_0(x,p) ] |
	\\
	\leq {}&\Aver{|\mathcal{M}(\infty,\hbar) -\mathcal{M}^\gamma(\infty,\hbar)| } + \Aver{|\mathcal{E}_1(k_0;\hbar) |} + \Aver{|\mathcal{E}_2(k_0,\tau_0;\hbar) |} + \Aver{|\mathcal{E}_3(k_0;\hbar) |} + \Aver{|\mathcal{E}_4(k_0,\tau_0;\hbar) |}
	\\
	&+ \Aver{|\mathcal{M}^\gamma(\infty,\hbar) - \int f^\gamma(t,x,p) \, d\mu_0(x,p) |} + |  \int f^\gamma(t,x,p) \, d\mu_0(x,p) -  \int f(t,x,p) \, d\mu_0(x,p)  |.
	\end{aligned}
\end{equation}
\Cref{EQ:error_main_thm_1.3,EQ:error_main_thm_1.6,EQ:error_main_thm_1.7,EQ:error_main_thm_1.8,EQ:error_main_thm_1.9,LE:main_term_reg_main_thm} gives us that by first taking $\hbar$ to zero and then taking $\gamma$ to zero the bound in \eqref{main_est_proof_main} will go to zero. This implies  the desired convergence and concludes the proof.  
\section{Proof of Theorem~\ref{Main_thm_gene}}\label{sec:proof_main_gen}
Before we give a proof of Theorem~\ref{Main_thm_gene} we will make the following observations regarding the solution to the linear Boltzmann equation when the collision kernel is as in Theorem~\ref{main_thm2}.
\begin{observation}\label{Sec13_obs}
Firstly recall that from Remark~\ref{born_series_convergence_remark} we have that
\begin{equation}
	\hat{T}(p,p_0) =-i  \sum_{n=1}^\infty (i\lambda)^{n} \Psi_{n}(p,p_0,\infty ; V).
\end{equation}
Applying Lemma~\ref{psi_m_gamma_est} and under the condition $\lambda C_d  \norm{ \hat{V}}_{1,\infty,2n+2}<1$ we have that
\begin{equation}\label{Sec13_obs_est_1}
	\begin{aligned}
	\sup_{p\in\R^d} \int_{\R^d} |\hat{T}(p,p_0) | \, dp_0 \leq \sup_{p\in\R^d}   \sum_{n=1}^\infty \lambda^{n}  \int_{\R^d}| \Psi_{n}(p,p_0,\infty ; V)|  \, dp_0 \leq  \sum_{n=1}^\infty  (\lambda C_d  \norm{ \hat{V}}_{1,\infty,2n+2})^n = C
	\\
	\sup_{p,p_0\in\R^d}  |\hat{T}(p,p_0) |  \leq \sup_{p,p_0 \in\R^d}   \sum_{n=1}^\infty \lambda^{n} | \Psi_{n}(p,p_0,\infty ; V)|  \leq  \sum_{n=1}^\infty  (\lambda C_d  \norm{ \hat{V}}_{1,\infty,2n+2})^n = C,
	\end{aligned}
\end{equation}
where $C$ is some constant. Secondly recall that the solution of the linear Boltzmann equation \eqref{intial_val_pro_boltzmanneqq} can be expressed as the collision series
 \begin{equation}\label{Sec13_obs_est_2}
 	 f(t,x,p) = \sum_{n=0}^\infty  f^{(n)}(t,x,p),
 \end{equation}
where
\begin{equation*}
	\begin{aligned}
	f^{(n)}(t,x,q_0)
	=   \int_{[0,t]_{\leq}^n}  \int_{\R^{nd}}   & \prod_{i=1}^{n} \Sigma(q_i,q_{i-1})  
	  \prod_{i=1}^{n+1}  e^{-(s_{i-1}-s_i) \Sigma_{\text{\rm tot}}(q)}  
	  \\  
	  &\times a(x- tq_0 -\sum_{i=1}^{n}s_{i}  (q_{i}-q_{i-1}) ,q_{n}) \,  d\boldsymbol{q}_{1,n}d\boldsymbol{s}_{n,1},
	  \end{aligned}
\end{equation*}
where we assume the initial data $a$ is Schwartz class and  
\begin{equation*}
	\Sigma(p,q) = (2\pi)^{d+1} \rho\,  \delta( \tfrac{1}{2} p^2-\tfrac{1}{2}q^2)| \hat{T}(p,q)|^2 .
\end{equation*}
Applying the two estimates from \eqref{Sec13_obs_est_1} and using that $ \Sigma_{\text{\rm tot}}(q)$ is positive for all $q$ it follows that
\begin{equation*}
	\begin{aligned}
	|f^{(n)}(t,x,q_0)| \leq \frac{(tC^2)^n}{n!} \sup_{x\in\R^d} \int |a(x,q_n)| \,dq_n. 
	  \end{aligned}
\end{equation*}
Combining this estimate with \eqref{Sec13_obs_est_2} we obtain that
 \begin{equation}
 	\sup_{x,p\in\R^d}| f(t,x,p) |\leq C 
 \end{equation}
for some constant $C$. Moreover we also have that $f(t,x,p)$ is continuous in $x$ and $p$.
\end{observation}
\begin{proof}[Proof of Theorem~\ref{Main_thm_gene}] Since Theorem~\ref{main_thm1} and Theorem~\ref{main_thm2} are equivalent we will only prove that the assertion of Theorem~\ref{main_thm2} is still valid. Hence  
let $A_\hbar(t)= U_{\hbar,\lambda}(t)\, A_\hbar\, U_{\hbar,\lambda}(-t)$ be the Heisenberg evolution of $A_\hbar=\OpW(a)$, where $a\in\mathcal{S}(\R^{2d})$. Then by applying Cauchy-Schwartz we get that
  \begin{equation*}
  	\begin{aligned}
	| \langle A_\hbar(t) \varphi_{\hbar}^n,  \varphi_{\hbar}^n\rangle - \langle A_\hbar(t) \varphi_{\hbar},  \varphi_{\hbar}\rangle | &=  | \langle A_\hbar(t) (\varphi_{\hbar}^n-\varphi_\hbar),  \varphi_{\hbar}^n\rangle - \langle A_\hbar(t) \varphi_{\hbar},  \varphi_\hbar^n-\varphi_{\hbar}\rangle |
	\\
	&\leq \norm{\OpW(a))}_{\mathrm{op}} \norm{\varphi_{\hbar}^n-\varphi_\hbar}_{L^2(\R^d)} ( \norm{\varphi_{\hbar}^n}_{L^2(\R^d)} +  \norm{\varphi_\hbar}_{L^2(\R^d)}).
	\end{aligned}
  \end{equation*}
By assumption this inequality implies that
  \begin{equation}\label{EQ_gene_1}
  	\begin{aligned}
	\lim_{n\rightarrow0}\sup_{\hbar\in I}| \langle A_\hbar(t) \varphi_{\hbar}^n,  \varphi_{\hbar}^n\rangle - \langle A_\hbar(t) \varphi_{\hbar},  \varphi_{\hbar}\rangle | = 0.
	\end{aligned}
  \end{equation}
We have that for each $n$ the families $\{\varphi_\hbar^n\}_{\hbar\in I}$ will satisfies the assumptions of Theorem~\ref{main_thm2}. Hence we get  the convergence
  \begin{equation}\label{EQ_gene_2}
 \mathbb{E}	\langle A_\hbar(t) \varphi_{\hbar}^n,  \varphi_{\hbar}^n\rangle \to  \int f(t,x,p) \, d\mu_{0}^n(x,p) \qquad\text{as $\hbar\rightarrow0$}.
  \end{equation}
  We note that
    \begin{equation}\label{EQ_gene_3}
    \begin{aligned}
	\MoveEqLeft | \mathbb{E}	\langle A_\hbar(t) \varphi_{\hbar},  \varphi_{\hbar}\rangle - \int f(t,x,p) \, d\mu_{0}(x,p)|
	\\
	\leq {}& | \mathbb{E}[	\langle A_\hbar(t) \varphi_{\hbar},  \varphi_{\hbar}\rangle - \langle A_\hbar(t) \varphi_{\hbar}^n,  \varphi_{\hbar}^n\rangle ] | + |\mathbb{E}	\langle A_\hbar(t) \varphi_{\hbar}^n,  \varphi_{\hbar}^n\rangle- \int f(t,x,p) \, d\mu_{0}^n(x,p)|  
	\\
	&+ |  \int f(t,x,p) \, d\mu_{0}^n(x,p) -  \int f(t,x,p) \, d\mu_{0}(x,p)|.
 \end{aligned}
   \end{equation}
 It follows form \eqref{EQ_gene_1} that the first term in this upper bound can be made small and that the second term can be made small follows from \eqref{EQ_gene_2}. Finally since the sequence $\{\mu_0^n\}$ converges weakly to $\mu_0$ and $f(t,x,p)$ is a bounded continuous function in $x$ and $p$ as observed in Observation~\ref{Sec13_obs} it follows that the third term in the upper bound in \eqref{EQ_gene_3} can be made small. Hence we have obtained the desired convergence and this ends the proof.
\end{proof}

We end this section by giving two examples of well approximated states. These are two particular cases of  Lagrangian states. These two types of states are
\begin{equation}
\varphi_{1,\hbar}(x) = w(x)e^{i\hbar^{-1}\langle x,p_0\rangle} \qquad\text{and}\qquad \varphi_{2,\hbar}(x)=\hbar^{-d/2} w(\hbar^{-1} (x-x_0)),
\end{equation}
where we will only assume that $w\in L^2(\R^d)$. The associated Wigner measures are
\begin{equation}
\mu_{0}^1(x,p) = |w(x)|^2  \delta(p-p_0) \,dx\, dp \qquad\text{and}\qquad \mu_{0}^2(x,p)=\delta(x-x_0)  |\hat w(p)|^2 \,dx\, dp.
\end{equation}
To find the approximating states we let  for all $\varepsilon>0$ $\psi_\varepsilon$ be a standard mollifier. That is $\psi_\varepsilon(x)=\varepsilon^{-d} \psi(\varepsilon^{-1}x)$, where $\psi\in C^\infty_0(\R^d)$, positive with support contained in the unit ball and integrating to 1. We set
\begin{equation*}
	w_\varepsilon(x) = \psi_\varepsilon* w(x) = \int  \psi_\varepsilon(x-y) w(y)\,dy.
\end{equation*}
It is standard results that $w_\varepsilon\in C^\infty(\R^d)$ and that
\begin{equation}
	\norm{w_\varepsilon}_{L^2(\R^d)} \leq C \norm{w}_{L^2(\R^d)} \qquad\text{and}\qquad \lim_{\varepsilon\rightarrow0}\norm{w_\varepsilon-w}_{L^2(\R^d)} =0.
\end{equation}
Moreover for each fixed $\varepsilon>0$ we have that $w_\varepsilon\in \mathcal{H}^{5d+5}_1(\R^d)$. This implies that for each fixed $\varepsilon$ we will have that the Lagrangian states
\begin{equation}
\varphi_{1,\hbar}^\varepsilon(x) = w_\varepsilon(x)e^{i\hbar^{-1}\langle x,p_0\rangle} \qquad\text{and}\qquad \varphi_{2,\hbar}^\varepsilon(x)=\hbar^{-d/2} w_\varepsilon(\hbar^{-1} (x-x_0))
\end{equation}
will satisfies the assumptions of Theorem~\ref{main_thm2}. Moreover, we have in particular that
\begin{equation}
	\lim_{n\rightarrow \infty} \sup_{\hbar\in (0,1]} \norm{\varphi_\hbar-\varphi_\hbar^n}_{L^2(\R^d)} = 0,
\end{equation}
where we have chosen $\varepsilon=\frac{1}{n}$. What remains is to verify that the measures will converge weakly.
The two associated Wigner measures $\mu_{0,\varepsilon}^1$ and $\mu_{0,\varepsilon}^2$ are given by
\begin{equation}
\mu_{0,\varepsilon}^1(x,p) = |w_\varepsilon(x)|^2  \delta(p-p_0) \,dx\, dp \qquad\text{and}\qquad \mu_{0,\varepsilon}^2(x,p)=\delta(x-x_0)  |\hat w_\varepsilon(p)|^2 \,dx\, dp.
\end{equation}
With this note that for any continuous and bounded function $f$ we get that
  \begin{equation}
  	\begin{aligned}
 	\MoveEqLeft  \big| \int f(x,p) \, d\mu_{0,\varepsilon}^1(x,p) -  \int f(x,p) \, d\mu_{0}^1(x,p) \big| 
  	\\
  	&\leq  \big| \int f(x,p_0)(| w_\varepsilon(x)|^2 - |w(x)|^2 ) \, dx  \big| \leq C (\norm{w_\varepsilon}_{L^2(\R^d)}^2 - \norm{w}_{L^2(\R^d)}^2).
  	\end{aligned}
  \end{equation}
This inequality implies that $\{\mu_{0,\varepsilon}^1\}$ converges weakly to $\mu_{0}^1$. Before we consider the measure $\mu_{0,\varepsilon}^2(x,p)$ we note that
\begin{equation}
	\hat w_\varepsilon(p) = \widehat{\psi_\varepsilon * w} (p) = \hat{\psi}(\varepsilon p) \hat{w}(p) \to  \hat{w}(p) \quad\text{as $\varepsilon\rightarrow0$},
\end{equation}
where the convergence result follows since $\psi$ integrates to one. With this in mind we see that
  \begin{equation}
  	\begin{aligned}
 	\MoveEqLeft  \big| \int f(t,x,p) \, d\mu_{0,\varepsilon}^2(x,p) -  \int f(t,x,p) \, d\mu_{0}^2(x,p) \big| 
  	\\
  	&\leq  \big| \int f(t,x_0,p)(| \hat{w}_\varepsilon(p)|^2 - |\hat{w}(p)|^2 ) \, dp  \big| \leq C \int |(\hat{\psi}(\varepsilon x)|^2-1) | \hat{w}(p)|^2\, dp. 
  	\end{aligned}
  \end{equation}
Since we have now constructed these two approximating sequences it follows form Theorem~\ref{Main_thm_gene} that the statement of Theorem~\ref{main_thm2} is indeed valid for states of the form 
\begin{equation}
\varphi_{1,\hbar}(x) = w(x)e^{i\hbar^{-1}\langle x,p_0\rangle} \qquad\text{and}\qquad \varphi_{2,\hbar}(x)=\hbar^{-d/2} w(\hbar^{-1} (x-x_0)),
\end{equation}
where we only assume $w\in L^2(\R^d)$.
\appendix
\section{Proofs of integral estimates}
\begin{proof}[Proof of Lemma~\ref{LE:int_posistion}]
Recall that we have $y\in\R^d$ given and want to establish the estimate
 \begin{equation}
 	\int_{\R^d}  \frac{1}{\langle x \rangle^{d+1} |x-y|} \, dx \leq \frac{C}{\langle y \rangle^{\frac{1}{2}}}.
 \end{equation}
We split into the two cases if $|y|>4$ or $|y|\leq4$. For the latter case we have that 
  \begin{equation}\label{EQ:int_posistion_1}
 	\int_{\R^d}  \frac{1}{\langle x \rangle^{d+1} |x-y|} \, dx \leq C. 
 \end{equation}
This estimate is true for a suitable choice of constant since the integral is finite for any value of $y$. Hence we are left with the case where $|y|>4$ for this case we split the integral in the following way
  \begin{equation}\label{EQ:int_posistion_2}
 	\int_{\R^d}  \frac{1}{\langle x \rangle^{d+1} |x-y|} \, dx = \int_{B(y,\sqrt{|y|})}  \frac{1}{\langle x \rangle^{d+1} |x-y|} \, dx  + \int_{ B(y,\sqrt{|y|})^c}  \frac{1}{\langle x \rangle^{d+1} |x-y|} \, dx.
 \end{equation}
 For the second integral in the right hand side of \eqref{EQ:int_posistion_2} we have that
\begin{equation}\label{EQ:int_posistion_3}
 	\int_{ B(y,\sqrt{|y|})^c}  \frac{1}{\langle x \rangle^{d+1} |x-y|} \, dx \leq \frac{1}{\sqrt{|y|}} \int_{ B(y,\sqrt{|y|})^c}  \frac{1}{\langle x \rangle^{d+1}} \,dx \leq \frac{C}{\sqrt{|y|}}.
 \end{equation}
 For the first integral of the left hand side of \eqref{EQ:int_posistion_2} we have by switching to spherical coordinates that
   \begin{equation}\label{EQ:int_posistion_4}
   	\begin{aligned}
 	 \int_{B(y,\sqrt{y})}  \frac{1}{\langle x \rangle^{d+1} |x-y|} \, dx&  =  \int_{B(0,\sqrt{y})}  \frac{1}{\langle x+y \rangle^{d+1} |x|} 
	 \\
	 &= \int_{0}^{\sqrt{y}} \int_{-1}^1 \frac{r^{d-2} (1-z^2)^{d-3}}{(1+r^2 +|y|^2 +2|y|rz )^{\frac{d+1}{2}} } \, dz dr.
	 \end{aligned}
 \end{equation}
 We observe that for any positive $r$ and $z\in[-1,1]$ we have that $r^2 +|y|^2 +2|y|rz\geq (|y|-r)^2$. Moreover, on the domain of integration and since $|y|>4$ we have that $|y|-r\geq \sqrt{|y|}$. Using this we obtain that
    \begin{equation}\label{EQ:int_posistion_5}
   	\begin{aligned}
 	 \int_{0}^{\sqrt{y}} \int_{-1}^1 \frac{r^{d-2} (1-z^2)^{d-3}}{(1+r^2 +|y|^2 +2|y|rz )^{\frac{d+1}{2}} } \, dz dr &\leq 2  \int_{0}^{\sqrt{y}} \frac{r^{d-2} }{(1+|y| )^{\frac{d+1}{2}} } \, dr 
	 \\
	 &\leq 2 \frac{|y|^{\frac{d}{2}} }{(1+|y| )^{\frac{d+1}{2}} } \leq  C \frac{1 }{(1+|y| )^{\frac{1}{2}} }
	 \end{aligned}
 \end{equation}
 From combining \eqref{EQ:int_posistion_1}, \eqref{EQ:int_posistion_2}, \eqref{EQ:int_posistion_3} and \eqref{EQ:int_posistion_5} we obtain the desired result.
\end{proof}
\begin{proof}[Proof of estimate \eqref{LE:resolvent_int_est_EQ_1}]
Recall that $\zeta\in(0,\frac{1}{2})$ and we want to prove
 \begin{equation}\label{Pf:resolvent_int_est_EQ_1}
 	\sup_{p\in\R^d}\int_\R  \frac{1}{\langle \nu \rangle |\frac{1}{2} p^2+\nu +i\zeta|} \, d\nu \leq C| \log(\zeta)|.
 \end{equation}
We observe that
\begin{equation}
	\begin{aligned}
	\int_\R  \frac{1}{\langle \nu \rangle |\frac{1}{2} p^2+\nu +i\zeta|} \, d\nu &= \int_{\R}  \frac{1}{\sqrt{1+ \nu^2} \sqrt{(\frac{1}{2} p^2+\nu)^2 +\zeta^2}} \, d\nu 
	\\
	&\leq 2\int_{0}^\infty  \frac{1}{\sqrt{1+ \nu^2} \sqrt{\nu^2 +\zeta^2}} \, d\nu 
	\\
	&\leq C \int_{0}^\infty  \frac{1}{(1+ \nu) (\nu +\zeta)} \, d\nu  
	\\
	&= 2C \int_0^\infty \frac{1}{\nu+\zeta}-\frac{1}{\nu+1} \,d\nu
	\\
	&= 2C [\log(\nu+\zeta) - \log(\nu+1)]_{\nu=0}^{\infty} = 2C |\log(\zeta)|. 
	\end{aligned}
\end{equation}	
This gives us the estimate in \eqref{Pf:resolvent_int_est_EQ_1}. 
\end{proof}
\begin{proof}[Proof of equation \eqref{LE:resolvent_int_est_EQ_2}]
Recall that  $\zeta\in(0,\frac{1}{2})$ and we want to prove the estimate
  \begin{equation}\label{Pf:resolvent_int_est_EQ_2.1}
 	\sup_{\nu\in\R}\int_{\R^d}  \frac{\langle \nu\rangle }{\langle p \rangle^{d+1} |\frac{1}{2} p^2+\nu +i\zeta|} \, dp \leq C| \log(\zeta)|.
 \end{equation}
We first do a change of variables to spherical coordinates this gives us
 \begin{equation}
 	\int_{\R^d}  \frac{\langle \nu\rangle }{\langle p \rangle^{d+1} |\frac{1}{2} p^2+\nu +i\zeta|} \, dp = C \int_0^\infty  \frac{\sqrt{ 1+\nu^2} r^{d-1} }{(1+r^2)^{\frac{d+1}{2}} \sqrt{(\frac{1}{2} r^2+\nu)^2 +\zeta^2}} \, dr,
 \end{equation}
where $C$ only depend on the dimension. Set $u=r^2$ by a change of variables we get
 \begin{equation}
 	\begin{aligned}
 	 \int_0^\infty  \frac{\sqrt{ 1+\nu^2} r^{d-1} }{(1+r^2)^{\frac{d+1}{2}} \sqrt{(\frac{1}{2} r^2+\nu)^2 +\zeta^2}} \, dr &= C  	 \int_0^\infty  \frac{\sqrt{ 1+\nu^2} u^{\frac{d-2}{2}} }{(1+u)^{\frac{d+1}{2}} \sqrt{(\frac{1}{2} u+\nu)^2 +\zeta^2}} \, du 
	 \\
	 &\leq C  	 \int_0^\infty  \frac{1+|\nu|  }{(1+u)^{\frac{3}{2}} \sqrt{(\frac{1}{2} u+\nu)^2 +\zeta^2}} \, du
	 \end{aligned}
 \end{equation}
Using the inequality $\sqrt{a^2+b^2}\geq \frac{1}{\sqrt{2}}(|a| +|b|)$ we get that
 \begin{equation}
 	\begin{aligned}
  	 \int_0^\infty  \frac{1+|\nu|  }{(1+u)^{\frac{3}{2}} \sqrt{(\frac{1}{2} u+\nu)^2 +\zeta^2}} \, du \leq C \int_0^\infty  \frac{1+|\nu|  }{(1+u)^{\frac{3}{2}}( |\frac{1}{2} u+\nu| +\zeta)} \, du
	 \end{aligned}
 \end{equation}
 We now divide into the two cases $\nu\geq0$ and $\nu<0$. For the first case we have that
  \begin{equation}
 	\begin{aligned}
  	 \int_0^\infty  \frac{1+|\nu|  }{(1+u)^{\frac{3}{2}}( |\frac{1}{2} u+\nu| +\zeta)} \, du &= \int_0^\infty  \frac{1+\nu  }{(1+u)^{\frac{3}{2}}( \frac{1}{2} u+\nu +\zeta)} \, du
	 \\
	 &=\int_0^\infty  \frac{1 }{(1+u)^{\frac{3}{2}}( \frac{1}{2} u+\nu +\zeta)} \, du + \int_0^\infty  \frac{\nu  }{(1+u)^{\frac{3}{2}}( \frac{1}{2} u+\nu +\zeta)} \, du
	 \\
	 &\leq \int_0^\infty  \frac{1 }{(1+u)( \frac{1}{2} u +\zeta)} \, du + \int_0^\infty  \frac{1  }{(1+u)^{\frac{3}{2}}} \, du
	 \\
	 &\leq C |\log(\zeta)|+C
	 \end{aligned}
 \end{equation}
 In the case $\nu<0$ we will assume $\nu>0$ and write $-\nu$ in the equation. That is we consider the integral
   \begin{equation}
 	\begin{aligned}
  	\int_0^\infty  \frac{1+\nu  }{(1+u)^{\frac{3}{2}}( |\frac{1}{2} u-\nu| +\zeta)} \, du
	 \end{aligned}
 \end{equation}
 First we consider the case where $\nu\in(0,1]$. For this case we have that
    \begin{equation}
 	\begin{aligned}
  	\MoveEqLeft \int_0^\infty  \frac{1+\nu  }{(1+u)^{\frac{3}{2}}( |\frac{1}{2} u-\nu| +\zeta)} \, du 
	\\
	&\leq \int_{2\nu}^\infty  \frac{2 }{(1+u)^{\frac{3}{2}}( \frac{1}{2} u-\nu +\zeta)} \, du +  \int_0^{2\nu}  \frac{2 }{(1+u)^{\frac{3}{2}}( \nu-\frac{1}{2} u +\zeta)} \, du
	\\
	&= \int_{0}^\infty  \frac{2 }{(1+u+2\nu)^{\frac{3}{2}}( \frac{1}{2} u +\zeta)} \, du +  \int_0^{2\nu}  \frac{2 }{(1+2\nu-u)^{\frac{3}{2}}( \frac{1}{2} u +\zeta)} \, du
	\\
	&\leq  \int_{0}^\infty  \frac{2 }{(1+u)( \frac{1}{2} u +\zeta)} \, du +  \int_0^{2}  \frac{2 }{ \frac{1}{2} u +\zeta} \, du \leq C |\log(\zeta)|+C
	 \end{aligned}
 \end{equation}
 Lastly we assume $\nu>1$. Here we do the change of variables $u\mapsto \nu^{-1}u$. This gives us
    \begin{equation}
 	\begin{aligned}
  	\int_0^\infty  \frac{1+\nu  }{(1+u)^{\frac{3}{2}}( |\frac{1}{2} u-\nu| +\zeta)} \, du = \frac{1+\nu}{\nu^{3/2}} \int_0^\infty  \frac{1  }{(\frac{1}{\nu}+u)^{\frac{3}{2}}( |\frac{1}{2} u-1| +\frac{\zeta}{\nu})} \, du 
	 \end{aligned}
 \end{equation}
 To estimate this integral we split the integral up in 4 parts and estimate each separately. We have that 
     \begin{equation}
 	\begin{aligned}
  	\MoveEqLeft  \frac{1+\nu}{\nu^{3/2}}  \int_0^1  \frac{1  }{(\frac{1}{\nu}+u)^{\frac{3}{2}}( |\frac{1}{2} u-1| +\frac{\zeta}{\nu})} \, du 
	\\
	&\leq \frac{1+\nu}{\nu^{3/2}} \frac{1}{\frac{1}{2}+\frac{\zeta}{\nu}} \int_{0}^1  \frac{1  }{(\frac{1}{\nu}+u)^{\frac{3}{2}}} \,du = \frac{4(1+\nu)}{(\nu+2\zeta)\sqrt{\nu}} \big[ \sqrt{\nu} -\frac{1}{\sqrt{\nu^{-1}+1}} \big] \leq \frac{8(1+\nu)}{\nu+2\zeta} \leq \frac{16}{1+2\zeta} 
	 \end{aligned}
 \end{equation}
      \begin{equation}
 	\begin{aligned}
  	\MoveEqLeft   \frac{1+\nu}{\nu^{3/2}} \int_1^2  \frac{1  }{(\frac{1}{\nu}+u)^{\frac{3}{2}}( |\frac{1}{2} u-1| +\frac{\zeta}{\nu})} \, du
	 \leq  \frac{1+\nu}{\nu^{3/2}}  \frac{1  }{(\frac{1}{\nu}+1)^{\frac{3}{2}}}  \int_1^2  \frac{1  }{ 1- \frac{1}{2} u +\frac{\zeta}{\nu}} \, du 
	 \\
	 &= \frac{2  }{(\nu+1)^{\frac{1}{2}}}  \big[ \log(\tfrac{1}{2}+\frac{\zeta}{\nu}) -\log(\frac{\zeta}{\nu}) \big] \leq C |\log(\zeta)|+C
	 \end{aligned}
 \end{equation}
       \begin{equation}
 	\begin{aligned}
  	 \MoveEqLeft \frac{1+\nu}{\nu^{3/2}} \int_2^3  \frac{1  }{(\frac{1}{\nu}+u)^{\frac{3}{2}}( |\frac{1}{2} u-1| +\frac{\zeta}{\nu})} \, du
	 \leq  \frac{1+\nu}{\nu^{3/2}} \frac{1  }{(\frac{1}{\nu}+2)^{\frac{3}{2}}}  \int_2^3  \frac{1  }{  \frac{1}{2} u -1+\frac{\zeta}{\nu}} \, du 
	 \\
	 &= \frac{2  }{(2\nu+1)^{\frac{1}{2}}}  \big[ \log(\tfrac{1}{2}+\frac{\zeta}{\nu}) -\log(\frac{\zeta}{\nu}) \big]   \leq C |\log(\zeta)|+C
	 \end{aligned}
 \end{equation}
        \begin{equation}
 	\begin{aligned}
  	 \MoveEqLeft  \frac{1+\nu}{\nu^{3/2}} \int_3^\infty  \frac{1  }{(\frac{1}{\nu}+u)^{\frac{3}{2}}( |\frac{1}{2} u-1| +\frac{\zeta}{\nu})} \, du
	 \\
	 & \leq \frac{1}{\frac{1}{2}+\frac{\zeta}{\nu}} \int_{3}^\infty  \frac{1  }{(\frac{1}{\nu}+u)^{\frac{3}{2}}} \,du = \frac{4(1+\nu)}{(\nu+2\zeta)\sqrt{\nu}}  \frac{1}{\sqrt{\nu^{-1}+3}} \leq \frac{4(1+\nu)}{(\nu+2\zeta)} \leq \frac{8}{(1+2\zeta)}
	 \end{aligned}
 \end{equation}
 Combining all of these estimates we get that \eqref{Pf:resolvent_int_est_EQ_2.1} is valid.
\end{proof}
\begin{proof}[Proof of estimate \eqref{LE:resolvent_int_est_EQ_3}]
 Again we have that  $\zeta\in(0,\frac{1}{2})$. Here we consider the integral
  \begin{equation}
	\sup_{q\in\R^d,\nu\in\R} \int_{\R^d}  \frac{1}{| p- q |}  \frac{\langle \nu\rangle}{ \langle p \rangle^{d+1}|\frac{1}{2} p^2+\nu+i\zeta|}
	\, dp 
\end{equation}
 We start by changing to spherical coordinates, where we measure the angular variable against the fixed vector $q$. This gives us the following expression ofr the integral.
   \begin{equation}
   \begin{aligned}
	\MoveEqLeft \int_{\R^d}  \frac{1}{| p- q |}  \frac{\langle \nu\rangle}{ \langle p \rangle^{d+1}|\frac{1}{2} p^2+\nu+i\zeta|}
	\, dp 
	\\
	&= c \int_0^\infty \int_{-1}^1  \frac{1}{\sqrt{r^2 - 2r|q|\theta +|q|^2}}  \frac{\langle \nu\rangle r^{d-1}}{ (1+r^2)^{\frac{d+1}{2}} |\frac{1}{2} r^2+\nu+i\zeta|} \,d\theta dr
	\\
	&=  c \int_0^\infty \frac{ \langle \nu\rangle r^{d-1}}{ (1+r^2)^{\frac{d+1}{2}} |\frac{1}{2} r^2+\nu+i\zeta|} \int_{-1}^1  \frac{1}{\sqrt{r^2 - 2r|q|\theta +|q|^2}}   \,d\theta dr
	\end{aligned}
\end{equation}
 Evaluating the $\theta$ integral yields
   \begin{equation}
   \begin{aligned}
	\int_{-1}^1  \frac{1}{\sqrt{r^2 - 2r|q|\theta +|q|^2}}   \,d\theta  &= -\frac{1}{r |q| } [\sqrt{r^2 - 2r|q|\theta +|q|^2}]_{\theta=-1}^{\theta=1} 
	= \frac{r + |q|-|r-|q|| }{r |q| }. 
	\end{aligned}
\end{equation}
 Using this we split the integral over $r$ into two parts according to if $r\geq|q|$ or $r<|q|$. This gives us that
    \begin{equation}
   \begin{aligned}
	\MoveEqLeft \int_{\R^d}  \frac{1}{| p- q |}  \frac{\langle \nu\rangle}{ \langle p \rangle^{d+1}|\frac{1}{2} p^2+\nu+i\zeta|}
	\, dp 
	\\
	&\leq c \int_0^{|q|} \frac{ \langle \nu\rangle r^{d-1}}{ |q| (1+r^2)^{\frac{d+1}{2}} |\frac{1}{2} r^2+\nu+i\zeta|} dr+  c \int_{|q|}^\infty \frac{\langle \nu\rangle r^{d-2}}{ (1+r^2)^{\frac{d+1}{2}} |\frac{1}{2} r^2+\nu+i\zeta|}dr
	\\
	\\
	&\leq  c \int_{0}^\infty \frac{\langle \nu\rangle r^{d-2}}{ (1+r^2)^{\frac{d+1}{2}} |\frac{1}{2} r^2+\nu+i\zeta|}dr \leq C|\log(\zeta)|,
	\end{aligned}
\end{equation}
 where the last inequality is obtained similar to the estimates obtained in the two previous proofs. This gives us the desired estimate.
\end{proof}
\begin{proof}[Proof of estimate \eqref{LE:resolvent_int_est_EQ_4}]
We now consider the fourth estimate of Lemma~\ref{LE:resolvent_int_est}. That is the estimate
  \begin{equation}\label{Pf:resolvent_int_est_EQ_4}
 	\sup_{q\in\R^d,\nu\in\R}\int_{\R^d}  \frac{1 }{\langle p-q \rangle^{d+1} |\frac{1}{2} p^2+\nu +i\zeta|} \, dp \leq C| \log(\zeta)|,
 \end{equation}
where again $\zeta\in(0,\frac{1}{2})$. For the case where $|q|=0$ we get the estimate from \eqref{LE:resolvent_int_est_EQ_2} so we may assume $|q|>0$. We will consider the three cases $\nu\geq0$ $\nu\in(0,-1]$ and $\nu<-1$. For $\nu\geq0$ we have that
   \begin{equation}
   \begin{aligned}
 	\MoveEqLeft \int_{\R^d}  \frac{1}{\langle p-q \rangle^{d+1} |\frac{1}{2} p^2+\nu +i\zeta|} \, dp 
	\\
	&\leq \int_{\{|p|\leq1\}}  \frac{1}{ |\frac{1}{2} p^2 +i\zeta|} \, dp +   \frac{1 }{\frac{1}{2} +\zeta}  \int_{\{|p|>1\}}  \frac{1 }{\langle p-q \rangle^{d+1}} \, dp
	\\
	&\leq  C \int_{0}^1  \frac{r^{\frac{d-2}{2}}}{ \frac{1}{2} r +\zeta} \, dr +   \frac{C }{\frac{1}{2} +\zeta} \leq C [\log(\tfrac{1}{2} + \zeta) - \log( \zeta)] +   \frac{C }{\frac{1}{2} +\zeta}   \leq C|\log( \zeta)| +C 
	\end{aligned}
 \end{equation}
 For the case when $\nu\in(0,-1]$ we will write $-\nu$  in he equation and think of $\nu$ being in $(0,1]$. We have here that 
   \begin{equation}
   \begin{aligned}
 	\MoveEqLeft \int_{\R^d}  \frac{1}{\langle p-q \rangle^{d+1} |\frac{1}{2} p^2-\nu +i\zeta|} \, dp 
	\\
	&\leq \int_{\{|p|\leq2\}}  \frac{1}{ |\frac{1}{2} p^2 -\nu +i\zeta|} \, dp +   \frac{1 }{1 +\zeta}  \int_{\{|p|>2\}}  \frac{1 }{\langle p-q \rangle^{d+1}} \, dp
	\\
	&\leq  C \int_{0}^{4}  \frac{r^{\frac{d-2}{2}}}{ |\frac{1}{2} r - \nu| +\zeta} \, dr +   C \leq  C \int_{0}^{2\nu}  \frac{1}{ \nu- \frac{1}{2} r  +\zeta} \, dr +  C \int_{2\nu}^{4}  \frac{1}{ \frac{1}{2} r - \nu +\zeta} \, dr +   C 
	\\
	&= C[\log(\nu+\zeta) - \log(\zeta) + \log(2-\nu+\zeta) - \log(\zeta) ] +C \leq C|\log( \zeta)| +C. 
	\end{aligned}
 \end{equation}
 What remains is the case where $\nu<-1$. So We are considering the integral
 \begin{equation}
   \begin{aligned}
 	\int_{\R^d}  \frac{1}{\langle p-q \rangle^{d+1} |\frac{1}{2} p^2-\nu +i\zeta|} \, dp,
	\end{aligned}
 \end{equation}
 where we again think of $\nu>1$. We will here consider the two cases where either $\nu\leq |q|^4+2$ or $\nu> |q|^4+2$. We start with the second case and we let $\tilde{\nu}$ be the number $\tilde{\nu}=\nu- |q|^4-2$. So when $\nu$ is large is it either because $\tilde{\nu}$ is large or $|q|$. We define the set  
 $A_\nu = \{p\in\R^d \, |\, \sqrt{2\nu-2} \leq |p| \leq \sqrt{2\nu+2} \}$. With this set we have that
     \begin{equation}
   \begin{aligned}
 	\MoveEqLeft \int_{\R^d}  \frac{1}{\langle p-q \rangle^{d+1} |\frac{1}{2} p^2-\nu +i\zeta|} \, dp
	\\
	&= \int_{A_\nu}  \frac{1}{\langle p-q \rangle^{d+1} |\frac{1}{2} p^2-\nu +i\zeta|} \, dp+\int_{A_\nu^{c}}  \frac{1}{\langle p-q \rangle^{d+1} |\frac{1}{2} p^2-\nu +i\zeta|} \, dp.
	\end{aligned}
 \end{equation}
 For the last integral we have that  
 \begin{equation}
   \begin{aligned}
 	 \int_{A_\nu^{c}}  \frac{1}{\langle p-q \rangle^{d+1} |\frac{1}{2} p^2-\nu +i\zeta|} \, dp \leq \frac{C}{1+\zeta}  \int_{A_\nu^{c}}  \frac{1}{\langle p-q \rangle^{d+1}} \leq \frac{C}{1+\zeta}.
	\end{aligned}
 \end{equation}
  For the second integral we note that for $p\in A_\nu$ we have that 
  \begin{equation}
  	|p-q| \geq \min_{p\in A_\nu}(|p|) -|q| = \sqrt{2\nu-2} - |q| = \sqrt{2\tilde{\nu}+2|q|^4+2} -|q|>0
  \end{equation}
  With this observation we get that
   \begin{equation}
   \begin{aligned}
 	\MoveEqLeft  \int_{A_\nu}  \frac{1}{\langle p-q \rangle^{d+1} |\frac{1}{2} p^2-\nu +i\zeta|} \, dp \leq \frac{1}{\langle  \sqrt{2\tilde{\nu}+2|q|^4+2} -|q| \rangle^{d+1}}  \int_{A_\nu}  \frac{1}{ |\frac{1}{2} p^2-\nu | +\zeta} \, dp
	 \\ 
	 &  =\frac{C}{\langle  \sqrt{2\tilde{\nu}+2|q|^4+2} -|q| \rangle^{d+1}}  \int_{\sqrt{2\nu-2}}^{\sqrt{2\nu+2}}  \frac{r^{d-1}}{ |\frac{1}{2} r^2-\nu | +\zeta} \, dr
	 \\ 
	 &  =\frac{C}{\langle  \sqrt{2\tilde{\nu}+2|q|^4+2} -|q| \rangle^{d+1}}  \int_{\nu-1}^{\nu+1}  \frac{r^{\frac{d-2}{2}}}{ | r-\nu | +\zeta} \, dr
	 \\
	 &\leq \frac{C (\tilde{\nu} +|q|^4 +2)^{\frac{d-2}{2}}}{\langle  \sqrt{2\tilde{\nu}+2|q|^4+2} -|q| \rangle^{d+1}} \Big[  \int_{\nu-1}^{\nu}  \frac{1}{ \nu -r +\zeta} \, dr +  \int_{\nu}^{\nu+1}  \frac{1}{  r-\nu  +\zeta} \, dr\Big] 
	 \\
	 &\leq C \Big[ \log(1+\zeta) -\log(\zeta) \Big]  \leq C|\log(\zeta)| +C
	\end{aligned}
 \end{equation}
  What remains is the case where $\nu\leq |q|^4+2$. Here we will further divide into the two cases $|q|\leq1$ and $|q|>1$. For the first case we divide the integral into one part over the ball of radius $5$ centred at zero and the rest. Then each of these integrals is estimated as above. In the case $|q|>1$ we start with the change of variables $p\mapsto p-q$. This gives us
   \begin{equation}
   \begin{aligned}
 	\MoveEqLeft \int_{\R^d}  \frac{1}{\langle p-q \rangle^{d+1} |\frac{1}{2} p^2-\nu +i\zeta|} \, dp = \int_{\R^d}  \frac{1}{\langle p \rangle^{d+1} |\frac{1}{2} (p+q)^2-\nu +i\zeta|} \, dp
	\\
	& = C  \int_{0}^\infty \int_{-1}^{1}  \frac{r^{d-1}}{\langle r \rangle^{d+1} |\frac{1}{2} r^2+ \frac{1}{2} |q|^2 + r|q|\theta  -\nu +i\zeta|} \, d\theta dr,
	\end{aligned}
 \end{equation}
  where we in the last equality have change to spherical coordinates, where the angular variables is measured against the fixed vector $q$. We now do the change of variables $r\mapsto r^2$  and $\theta\mapsto \sqrt{r}|q|\theta$  This gives us
     \begin{equation}
   \begin{aligned}
 	\MoveEqLeft    \int_{0}^\infty \int_{-1}^{1}  \frac{r^{d-1}}{\langle r \rangle^{d+1} |\frac{1}{2} r^2+ \frac{1}{2} |q|^2 + r|q|\theta  -\nu +i\zeta|} \, d\theta dr
	\\
	&= C  \int_{0}^\infty \int_{-\sqrt{r}|q|}^{\sqrt{r}|q|}  \frac{r^{\frac{d-3}{2}}}{|q|\langle \sqrt{r} \rangle^{d+1} |\frac{1}{2} r+ \frac{1}{2} |q|^2 + \theta  -\nu +i\zeta|} \, d\theta dr
	\\
	&\leq C  \int_{0}^\infty \int_{-\sqrt{r}|q|}^{\sqrt{r}|q|}  \frac{1}{|q|(1+r)^2 |\frac{1}{2} r+ \frac{1}{2} |q|^2 + \theta  -\nu +i\zeta|} \, d\theta dr
	\end{aligned}
 \end{equation}
  Next we do the change of variables $\theta \mapsto \frac{1}{2} r+ \frac{1}{2} |q|^2 + \theta  -\nu$. With this change of variables the bounds on our integral in $\theta$ becomes $\frac{1}{2} r+ \frac{1}{2} |q|^2 + \sqrt{r}|q|  -\nu$ and $\frac{1}{2} r+ \frac{1}{2} |q|^2 - \sqrt{r}|q|  -\nu$. Let $g(r) = \frac{1}{2} r+ \frac{1}{2} |q|^2 + \sqrt{r}|q|  + |q|^4+2$. Then we have that
  \begin{equation}
  	\frac{1}{2} r+ \frac{1}{2} |q|^2 + \sqrt{r}|q|  -\nu \leq g(r) \quad\text{and}\quad  \frac{1}{2} r+ \frac{1}{2} |q|^2 - \sqrt{r}|q|  -\nu > - g(r).
  \end{equation}
  Hence we get that
       \begin{equation}
   \begin{aligned}
 	\MoveEqLeft     \int_{0}^\infty \int_{-\sqrt{r}|q|}^{\sqrt{r}|q|}  \frac{1}{|q|(1+r)^2 |\frac{1}{2} r+ \frac{1}{2} |q|^2 + \theta  -\nu +i\zeta|} \, d\theta dr 
	\\
	&\leq  C \int_{0}^\infty  \frac{1}{|q|(1+r)^2}  \int_{-g(r)}^{g(r)}  \frac{1}{ | \theta| + \zeta} \, d\theta dr =  C \int_{0}^\infty  \frac{1}{|q|(1+r)^2}  \int_{0}^{g(r)}  \frac{1}{ \theta + \zeta} \, d\theta dr  
	\\
	& = C \int_{0}^\infty  \frac{1}{|q|(1+r)^2}  \big[ \log(g(r)+\zeta) - \log(\zeta) \big] dr 
	\\
	&\leq C|\log(\zeta)| +C \frac{\log(g(r)+\zeta)}{|q| \sqrt{(1+r)}}  \int_{0}^\infty  \frac{1}{(1+r)^{\frac{3}{2}}}  dr  \leq C|\log(\zeta)| +C,
	\end{aligned}
 \end{equation}
  where we have used the definition of $g(r)$ to conclude the the fraction is bounded recalling that $|q|>1$. Combing the estimates from above we get the desired conclusion. This ends the proof.
  \end{proof}
\begin{proof}[Proof of Lemma~\ref{LE:est_res_combined}]
Recall that $\zeta\in(0,1/2)$ and we want to prove the following estimate
 \begin{equation}\label{Pf:est_res_combined_eq_1}
 	\sup_{\nu,\tilde{\nu}\in\R}\int_{\R^d}   \frac{\langle \nu\rangle \langle \tilde{\nu} \rangle }{\langle p \rangle^{d+1} \langle p-q \rangle^{d+1} |\frac{1}{2} p^2+\nu +i\zeta|  |\frac{1}{2} (p-q)^2+\tilde{\nu} +i\zeta|} \, dp \leq \frac{ C\log(\zeta)^2}{|q| + \zeta},
 \end{equation}
First we assume that $|q|>0$. Then we do a change of variables to spherical coordinates, where the angular momentum is measure against the fixed vector $q$. This gives us
 \begin{equation}
 	\begin{aligned}
 	\MoveEqLeft \int_{\R^d}   \frac{\langle \nu\rangle \langle \tilde{\nu} \rangle }{\langle p \rangle^{d+1} \langle p-q \rangle^{d+1} |\frac{1}{2} p^2+\nu +i\zeta|  |\frac{1}{2} (p-q)^2+\tilde{\nu} +i\zeta|} \, dp
	\\
	={}& \int_{0}^\infty \int_{-1}^1  \frac{r^{d-1}\langle \nu\rangle \langle \tilde{\nu} \rangle }{\langle r \rangle^{d+1} \langle \sqrt{r^2 + |q|^2 - 2r|q|\theta} \rangle^{d+1} |\frac{1}{2} r^2+\nu +i\zeta|  |\frac{1}{2} (r^2 + |q|^2 - 2r|q|\theta)+\tilde{\nu} +i\zeta|} \, d\theta dr
	\\
	\leq{}& \frac{1}{|q|} \int_{0}^\infty \frac{r^{d-2}\langle \nu\rangle }{\langle r \rangle^{d+1}  |\frac{1}{2} r^2+\nu +i\zeta| dr  }  \int_{\R}  \frac{ \langle \tilde{\nu} \rangle }{ \langle \sqrt{\theta} \rangle^{d+1}|\theta+\tilde{\nu} +i\zeta|} \, d\theta,
	\end{aligned}
 \end{equation}
where we in the last inequality first have extended the integration domain to the whole line for $\theta$ and then done the change of variables $\theta\mapsto \frac{1}{2}r^2 + \frac{1}{2}|q|^2 - r|q|\theta$. Applying Lemma~\ref{LE:resolvent_int_est} we get that
 \begin{equation}
 	\begin{aligned}
 	\int_{\R^d}   \frac{\langle \nu\rangle \langle \tilde{\nu} \rangle }{\langle p \rangle^{d+1} \langle p-q \rangle^{d+1} |\frac{1}{2} p^2+\nu +i\zeta|  |\frac{1}{2} (p-q)^2+\tilde{\nu} +i\zeta|} \, dp \leq \frac{C\log(\zeta)^2}{|q|}.
	\end{aligned}
 \end{equation}
For any $q\in\R^d$ we have that
 \begin{equation}
 	\begin{aligned}
 	\MoveEqLeft \int_{\R^d}   \frac{\langle \nu\rangle \langle \tilde{\nu} \rangle }{\langle p \rangle^{d+1} \langle p-q \rangle^{d+1} |\frac{1}{2} p^2+\nu +i\zeta|  |\frac{1}{2} (p-q)^2+\tilde{\nu} +i\zeta|} \, dp 
	\\
	&\leq \frac{C}{\zeta}  \int_{\R^d}   \frac{\langle \nu\rangle  }{\langle p \rangle^{d+1} |\frac{1}{2} p^2+\nu +i\zeta|  } \, dp  \leq  \frac{C |\log(\zeta)| }{\zeta},
	\end{aligned}
 \end{equation}
where we again have used Lemma~\ref{LE:resolvent_int_est}. Combining the two estimates we obtain the estimate in \eqref{Pf:est_res_combined_eq_1}. 
\end{proof}

 \bibliographystyle{plain}
\bibliography{Bib.bib}

\end{document}